%% file: paper.tex
\newif\iflongversion
\longversiontrue
\documentclass[a4paper,UKenglish,thm-restate]{lipics-v2021}

\usepackage{xcolor}
\usepackage{mdframed}
\usepackage{hyperref}
\usepackage{mathtools}
\usepackage{amsmath}

\usepackage{amsthm}
\usepackage{relsize}

\newcommand{\citeappendix}[2]{\iflongversion
Appendix~\ref{#1}\else
Appendix #2~\cite{DBLP:journals/corr/abs-2401-17199}\fi}
\newcommand{\citeappendixtwo}[4]{\iflongversion
Appendix~\ref{#1} and \ref{#2}\else
Appendix #3 and #4~\cite{DBLP:journals/corr/abs-2401-17199}\fi}

\usepackage{lineno}

\usepackage{ottalt}

  \input{mGL-inc}

  \renewottcommands[mGL]

\usepackage{ragged2e}

\input{math}

\usepackage{todonotes}

\nolinenumbers

\newcommand{\mGL}{\mathsf{mGL}}
\newcommand{\GS}{\mathsf{GS}}
\newcommand{\MS}{\mathsf{MS}}
\newcommand{\GST}{\mathsf{GT}}
\newcommand{\MST}{\mathsf{MT}}
\newcommand{\mGLL}{Mixed Graded/Linear}

\usepackage{resizegather}

\EventEditors{J\"{o}rg Endrullis and Sylvain Schmitz}
\EventNoEds{2}
\EventLongTitle{33rd EACSL Annual Conference on Computer Science Logic (CSL 2025)}
\EventShortTitle{CSL 2025}
\EventAcronym{CSL}
\EventYear{2025}
\EventDate{February 10--14, 2025}
\EventLocation{Amsterdam, Netherlands}
\EventLogo{}
\SeriesVolume{326}
\ArticleNo{32}

\iflongversion
\title{A Mixed Linear and Graded Logic:\\Proofs, Terms, and Models {\small{(with appendices)}}}
\hideLIPIcs 
\else
\title{A Mixed Linear and Graded Logic:\\Proofs, Terms, and Models}
\fi

\titlerunning{A Mixed Linear and Graded Logic: Proofs, Terms, and Models}

\Copyright{Victoria Vollmer, Danielle Marshall, Harley Eades III, Dominic Orchard}

\author{Victoria Vollmer}
  {School of Computing, University of Kent, UK}
  {v.vollmer@kent.ac.uk}{https://orcid.org/0009-0005-2082-559X}{}

\author{Danielle Marshall}
{School of Computing, University of Kent, UK. \url{https://starsandspira.ls/}}
{dm635@kent.ac.uk}{https://orcid.org/0000-0002-4284-3757}{}

\author{Harley Eades III}
  {Computer Science, Augusta University,
    USA. \url{http://metatheorem.org/}}
  {harley.eades@gmail.com}{https://orcid.org/0000-0001-8474-5971}{}

\author{Dominic Orchard}
  {School of Computing, University of Kent, UK
\and Department of Computer Science and Technology, University of
Cambridge, UK}
  {d.a.orchard@kent.ac.uk
\and dominic.orchard@cl.cam.ac.uk}{https://orcid.org/0000-0002-7058-7842}{}

\authorrunning{V. Vollmer, D. Marshall, H. Eades III and D. Orchard}

\ccsdesc[500]{Theory of computation~Logic}

\keywords{linear logic, graded modal logic, adjoint decomposition}

\iflongversion
\else
\relatedversiondetails[cite=DBLP:journals/corr/abs-2401-17199]{Full version
  with appendices}{https://arxiv.org/abs/2401.17199}
\fi

\begin{document}

\maketitle

\begin{abstract}
  \input{abstract}
\end{abstract}



\iflongversion
\vspace{1em}
\noindent\fbox{\begin{minipage}{\dimexpr\textwidth-2\fboxsep-2\fboxrule\relax}
  This is an extended version of the paper with appendices.
  The official version of the published paper can be found in
  the 33rd EACSL Annual Conference on Computer Science Logic (CSL 2025):
  \newline ~\url{https://doi.org/10.4230/LIPIcs.CSL.2025.32}
  \end{minipage}}
\fi

\section{Introduction}
\label{sec:introduction}
\input{intro-ottput}

\section{Mixed Graded/Linear Logic: Proofs and Terms}
\label{sec:mixed_graded/linear_logic}
\input{mgl-ottput}

\section{Model}
\label{sec:mGL-seq_denotational_model}
\input{mGL-seq-cat-semantic-ottput}

\section{Natural Deduction}
\label{subsec:mgl_nat_deduc}
\input{mGL-term-ottput}

\section{Discussion}
\label{sec:discussion}
\input{discussion-ottput}

\paragraph*{Acknowledgements}
 We thank all the anonymous reviewers of this, and previous versions, of this paper.
 We are also grateful for discussions with Peter Hanukaev and helpful
 comments from Paulo Torrens on a draft of this manuscript.
 This work was supported in part by the EPSRC
 grant EP/T013516/1 (\emph{Verifying Resource-like Data Use in Programs via
Types}). The second author received support through Schmidt Sciences, LLC.
\bibliography{refs}

\iflongversion
\onecolumn
\appendix
\input{appendix}
\fi

\end{document}

%% file: mGL-inc.tex
\newcommand{\mGLdrule}[4][]{{\displaystyle\frac{\begin{array}{l}#2\end{array}}{#3}\quad\mGLdrulename{#4}}}

\newcommand{\mGLpremise}[1]{ #1 \\}
\newenvironment{mGLdefnblock}[3][]{ \framebox{\mbox{#2}} \quad #3 \\[0pt]}{}

\newcommand{\mGLnt}[1]{\mathit{#1}}
\newcommand{\mGLmv}[1]{\mathit{#1}}
\newcommand{\mGLkw}[1]{\mathbf{#1}}
\newcommand{\mGLsym}[1]{#1}

\newcommand{\mGLdrulename}[1]{\textsc{#1}}

\usepackage{amsmath}
                 \usepackage{cmll}
                 \usepackage{stmaryrd}
                 \usepackage{relsize}
                 \usepackage{supertabular}

\newcommand{\mGLdruleGSTXXUnitRName}[0]{\mGLdrulename{GST\_UnitR}}
\newcommand{\mGLdruleGSTXXUnitR}[1]{\mGLdrule[#1]{%
}{
\emptyset  \odot  \emptyset  \vdash_{\mathsf{GS} }  \mathsf{j}  \mGLsym{:}  \mathsf{J}}{%
{\mGLdruleGSTXXUnitRName}{}%
}}

\newcommand{\mGLdruleGSTXXUnitLName}[0]{\mGLdrulename{GST\_UnitL}}
\newcommand{\mGLdruleGSTXXUnitL}[1]{\mGLdrule[#1]{%
\mGLpremise{ (  \delta_{{\mathrm{1}}}  \mGLsym{,}  \delta_{{\mathrm{2}}}  )   \odot   ( \Delta_{{\mathrm{1}}}  \mGLsym{,}  \Delta_{{\mathrm{2}}} )   \vdash_{\mathsf{GS} }  \mGLnt{t}  \mGLsym{:}  \mGLnt{X}}%
}{
 (  \delta_{{\mathrm{1}}}  \mGLsym{,}  \mGLnt{r}  \mGLsym{,}  \delta_{{\mathrm{2}}}  )   \odot   ( \Delta_{{\mathrm{1}}}  \mGLsym{,}  \mGLmv{x}  \mGLsym{:}  \mathsf{J}  \mGLsym{,}  \Delta_{{\mathrm{2}}} )   \vdash_{\mathsf{GS} }  \mathsf{let} \, \mathsf{j} \, \mGLsym{=}  \mGLmv{x} \, \mathsf{in} \, \mGLnt{t}  \mGLsym{:}  \mGLnt{X}}{%
{\mGLdruleGSTXXUnitLName}{}%
}}

\newcommand{\mGLdruleGSTXXTenRName}[0]{\mGLdrulename{GST\_TenR}}
\newcommand{\mGLdruleGSTXXTenR}[1]{\mGLdrule[#1]{%
\mGLpremise{\delta_{{\mathrm{1}}}  \odot  \Delta_{{\mathrm{1}}}  \vdash_{\mathsf{GS} }  \mGLnt{t_{{\mathrm{1}}}}  \mGLsym{:}  \mGLnt{X}}%
\mGLpremise{\delta_{{\mathrm{2}}}  \odot  \Delta_{{\mathrm{2}}}  \vdash_{\mathsf{GS} }  \mGLnt{t_{{\mathrm{2}}}}  \mGLsym{:}  \mGLnt{Y}}%
}{
 (  \delta_{{\mathrm{1}}}  \mGLsym{,}  \delta_{{\mathrm{2}}}  )   \odot   ( \Delta_{{\mathrm{1}}}  \mGLsym{,}  \Delta_{{\mathrm{2}}} )   \vdash_{\mathsf{GS} }  \mGLsym{(}  \mGLnt{t_{{\mathrm{1}}}}  \mGLsym{,}  \mGLnt{t_{{\mathrm{2}}}}  \mGLsym{)}  \mGLsym{:}  \mGLnt{X}  \boxtimes  \mGLnt{Y}}{%
{\mGLdruleGSTXXTenRName}{}%
}}

\newcommand{\mGLdruleGSTXXSubName}[0]{\mGLdrulename{GST\_Sub}}
\newcommand{\mGLdruleGSTXXSub}[1]{\mGLdrule[#1]{%
\mGLpremise{\delta_{{\mathrm{1}}}  \odot  \Delta  \vdash_{\mathsf{GS} }  \mGLnt{t}  \mGLsym{:}  \mGLnt{X}  \quad  \delta_{{\mathrm{1}}}  \leq  \delta_{{\mathrm{2}}}}%
}{
\delta_{{\mathrm{2}}}  \odot  \Delta  \vdash_{\mathsf{GS} }  \mGLnt{t}  \mGLsym{:}  \mGLnt{X}}{%
{\mGLdruleGSTXXSubName}{}%
}}

\newcommand{\mGLdruleGSTXXidName}[0]{\mGLdrulename{GST\_id}}
\newcommand{\mGLdruleGSTXXid}[1]{\mGLdrule[#1]{%
}{
1  \odot  \mGLmv{x}  \mGLsym{:}  \mGLnt{X}  \vdash_{\mathsf{GS} }  \mGLmv{x}  \mGLsym{:}  \mGLnt{X}}{%
{\mGLdruleGSTXXidName}{}%
}}

\newcommand{\mGLdruleGSTXXTenLName}[0]{\mGLdrulename{GST\_TenL}}
\newcommand{\mGLdruleGSTXXTenL}[1]{\mGLdrule[#1]{%
\mGLpremise{ (  \delta_{{\mathrm{1}}}  \mGLsym{,}  \mGLnt{r}  \mGLsym{,}  \mGLnt{r}  \mGLsym{,}  \delta_{{\mathrm{2}}}  )   \odot   ( \Delta_{{\mathrm{1}}}  \mGLsym{,}  \mGLmv{x}  \mGLsym{:}  \mGLnt{X}  \mGLsym{,}  \mGLmv{y}  \mGLsym{:}  \mGLnt{Y}  \mGLsym{,}  \Delta_{{\mathrm{2}}} )   \vdash_{\mathsf{GS} }  \mGLnt{t}  \mGLsym{:}  \mGLnt{Z}}%
}{
 (  \delta_{{\mathrm{1}}}  \mGLsym{,}  \mGLnt{r}  \mGLsym{,}  \delta_{{\mathrm{2}}}  )   \odot   ( \Delta_{{\mathrm{1}}}  \mGLsym{,}  \mGLmv{z}  \mGLsym{:}  \mGLnt{X}  \boxtimes  \mGLnt{Y}  \mGLsym{,}  \Delta_{{\mathrm{2}}} )   \vdash_{\mathsf{GS} }   \mathsf{let} \,( \mGLmv{x} , \mGLmv{y} ) =  \mGLmv{z} \, \mathsf{in} \, \mGLnt{t}   \mGLsym{:}  \mGLnt{Z}}{%
{\mGLdruleGSTXXTenLName}{}%
}}

\newcommand{\mGLdruleGSTXXLinRName}[0]{\mGLdrulename{GST\_LinR}}
\newcommand{\mGLdruleGSTXXLinR}[1]{\mGLdrule[#1]{%
\mGLpremise{\delta  \odot  \Delta  \mGLsym{;}  \emptyset  \vdash_{\mathsf{MS} }  \mGLnt{l}  \mGLsym{:}  \mGLnt{B}}%
}{
\delta  \odot  \Delta  \vdash_{\mathsf{GS} }  \mathsf{Lin} \, \mGLnt{l}  \mGLsym{:}  \mathsf{Lin} \, \mGLnt{B}}{%
{\mGLdruleGSTXXLinRName}{}%
}}

\newcommand{\mGLdruleGSTXXCutName}[0]{\mGLdrulename{GST\_Cut}}
\newcommand{\mGLdruleGSTXXCut}[1]{\mGLdrule[#1]{%
\mGLpremise{\delta_{{\mathrm{2}}}  \odot  \Delta_{{\mathrm{2}}}  \vdash_{\mathsf{GS} }  \mGLnt{t_{{\mathrm{1}}}}  \mGLsym{:}  \mGLnt{X}}%
\mGLpremise{ (  \delta_{{\mathrm{1}}}  \mGLsym{,}  \mGLnt{r}  \mGLsym{,}  \delta_{{\mathrm{3}}}  )   \odot   ( \Delta_{{\mathrm{1}}}  \mGLsym{,}  \mGLmv{x}  \mGLsym{:}  \mGLnt{X}  \mGLsym{,}  \Delta_{{\mathrm{3}}} )   \vdash_{\mathsf{GS} }  \mGLnt{t_{{\mathrm{2}}}}  \mGLsym{:}  \mGLnt{Y}}%
}{
 (  \delta_{{\mathrm{1}}}  \mGLsym{,}   \mGLnt{r}   *   \delta_{{\mathrm{2}}}   \mGLsym{,}  \delta_{{\mathrm{3}}}  )   \odot   ( \Delta_{{\mathrm{1}}}  \mGLsym{,}  \Delta_{{\mathrm{2}}}  \mGLsym{,}  \Delta_{{\mathrm{3}}} )   \vdash_{\mathsf{GS} }  \mGLsym{[}  \mGLnt{t_{{\mathrm{1}}}}  \mGLsym{/}  \mGLmv{x}  \mGLsym{]}  \mGLnt{t_{{\mathrm{2}}}}  \mGLsym{:}  \mGLnt{Y}}{%
{\mGLdruleGSTXXCutName}{}%
}}

\newcommand{\mGLdruleGSTXXWeakName}[0]{\mGLdrulename{GST\_Weak}}
\newcommand{\mGLdruleGSTXXWeak}[1]{\mGLdrule[#1]{%
\mGLpremise{ (  \delta_{{\mathrm{1}}}  \mGLsym{,}  \delta_{{\mathrm{2}}}  )   \odot   ( \Delta_{{\mathrm{1}}}  \mGLsym{,}  \Delta_{{\mathrm{2}}} )   \vdash_{\mathsf{GS} }  \mGLnt{t}  \mGLsym{:}  \mGLnt{Y}}%
}{
 (  \delta_{{\mathrm{1}}}  \mGLsym{,}  \mathsf{0}  \mGLsym{,}  \delta_{{\mathrm{2}}}  )   \odot   ( \Delta_{{\mathrm{1}}}  \mGLsym{,}  \mGLmv{x}  \mGLsym{:}  \mGLnt{X}  \mGLsym{,}  \Delta_{{\mathrm{2}}} )   \vdash_{\mathsf{GS} }  \mGLnt{t}  \mGLsym{:}  \mGLnt{Y}}{%
{\mGLdruleGSTXXWeakName}{}%
}}

\newcommand{\mGLdruleGSTXXContName}[0]{\mGLdrulename{GST\_Cont}}
\newcommand{\mGLdruleGSTXXCont}[1]{\mGLdrule[#1]{%
\mGLpremise{ (  \delta_{{\mathrm{1}}}  \mGLsym{,}  \mGLnt{r_{{\mathrm{1}}}}  \mGLsym{,}  \mGLnt{r_{{\mathrm{2}}}}  \mGLsym{,}  \delta_{{\mathrm{2}}}  )   \odot   ( \Delta_{{\mathrm{1}}}  \mGLsym{,}  \mGLmv{x}  \mGLsym{:}  \mGLnt{X}  \mGLsym{,}  \mGLmv{y}  \mGLsym{:}  \mGLnt{X}  \mGLsym{,}  \Delta_{{\mathrm{2}}} )   \vdash_{\mathsf{GS} }  \mGLnt{t}  \mGLsym{:}  \mGLnt{Y}}%
}{
 (  \delta_{{\mathrm{1}}}  \mGLsym{,}  \mGLnt{r_{{\mathrm{1}}}}  +  \mGLnt{r_{{\mathrm{2}}}}  \mGLsym{,}  \delta_{{\mathrm{2}}}  )   \odot   ( \Delta_{{\mathrm{1}}}  \mGLsym{,}  \mGLmv{x}  \mGLsym{:}  \mGLnt{X}  \mGLsym{,}  \Delta_{{\mathrm{2}}} )   \vdash_{\mathsf{GS} }  \mGLsym{[}  \mGLmv{x}  \mGLsym{/}  \mGLmv{y}  \mGLsym{]}  \mGLnt{t}  \mGLsym{:}  \mGLnt{Y}}{%
{\mGLdruleGSTXXContName}{}%
}}

\newcommand{\mGLdruleGSTXXExName}[0]{\mGLdrulename{GST\_Ex}}
\newcommand{\mGLdruleGSTXXEx}[1]{\mGLdrule[#1]{%
\mGLpremise{ (  \delta_{{\mathrm{1}}}  \mGLsym{,}  \mGLnt{r_{{\mathrm{1}}}}  \mGLsym{,}  \mGLnt{r_{{\mathrm{2}}}}  \mGLsym{,}  \delta_{{\mathrm{2}}}  )   \odot   ( \Delta_{{\mathrm{1}}}  \mGLsym{,}  \mGLmv{x}  \mGLsym{:}  \mGLnt{X}  \mGLsym{,}  \mGLmv{y}  \mGLsym{:}  \mGLnt{Y}  \mGLsym{,}  \Delta_{{\mathrm{2}}} )   \vdash_{\mathsf{GS} }  \mGLnt{t}  \mGLsym{:}  \mGLnt{Z}}%
}{
 (  \delta_{{\mathrm{1}}}  \mGLsym{,}  \mGLnt{r_{{\mathrm{2}}}}  \mGLsym{,}  \mGLnt{r_{{\mathrm{1}}}}  \mGLsym{,}  \delta_{{\mathrm{2}}}  )   \odot   ( \Delta_{{\mathrm{1}}}  \mGLsym{,}  \mGLmv{y}  \mGLsym{:}  \mGLnt{Y}  \mGLsym{,}  \mGLmv{x}  \mGLsym{:}  \mGLnt{X}  \mGLsym{,}  \Delta_{{\mathrm{2}}} )   \vdash_{\mathsf{GS} }  \mGLnt{t}  \mGLsym{:}  \mGLnt{Z}}{%
{\mGLdruleGSTXXExName}{}%
}}

\newcommand{\mGLdruleGSTXXMCutName}[0]{\mGLdrulename{GST\_MCut}}
\newcommand{\mGLdruleGSTXXMCut}[1]{\mGLdrule[#1]{%
\mGLpremise{\delta_{{\mathrm{2}}}  \odot  \Delta_{{\mathrm{2}}}  \vdash_{\mathsf{GS} }  \mGLnt{t_{{\mathrm{1}}}}  \mGLsym{:}  \mGLnt{X}}%
\mGLpremise{ (  \delta_{{\mathrm{1}}}  \mGLsym{,}  \delta  \mGLsym{,}  \delta_{{\mathrm{3}}}  )   \odot   ( \Delta_{{\mathrm{1}}}  \mGLsym{,}   \mGLmv{x} ^{ \mGLmv{n} } :  \mGLnt{X} ^{ \mGLmv{n} }   \mGLsym{,}  \Delta_{{\mathrm{3}}} )   \vdash_{\mathsf{GS} }  \mGLnt{t_{{\mathrm{2}}}}  \mGLsym{:}  \mGLnt{Y}}%
}{
 (  \delta_{{\mathrm{1}}}  \mGLsym{,}   (   \delta  \boxast [ { \delta_{{\mathrm{2}}} }^{ \mGLmv{n} } ]   )   \mGLsym{,}  \delta_{{\mathrm{3}}}  )   \odot   ( \Delta_{{\mathrm{1}}}  \mGLsym{,}  \Delta_{{\mathrm{2}}}  \mGLsym{,}  \Delta_{{\mathrm{3}}} )   \vdash_{\mathsf{GS} }  \mGLsym{[}   \mGLnt{t_{{\mathrm{1}}}} ,\ldots, \mGLnt{t_{{\mathrm{1}}}}   \mGLsym{/}  \mGLmv{x_{{\mathrm{1}}}}  \mGLsym{,} \, ... \, \mGLsym{,}  \mGLmv{x_{\mGLmv{n}}}  \mGLsym{]}  \mGLnt{t_{{\mathrm{2}}}}  \mGLsym{:}  \mGLnt{Y}}{%
{\mGLdruleGSTXXMCutName}{}%
}}

\newcommand{\mGLdruleMSTXXidName}[0]{\mGLdrulename{MST\_id}}
\newcommand{\mGLdruleMSTXXid}[1]{\mGLdrule[#1]{%
}{
\emptyset  \odot  \emptyset  \mGLsym{;}  \mGLmv{x}  \mGLsym{:}  \mGLnt{A}  \vdash_{\mathsf{MS} }  \mGLmv{x}  \mGLsym{:}  \mGLnt{A}}{%
{\mGLdruleMSTXXidName}{}%
}}

\newcommand{\mGLdruleMSTXXUnitRName}[0]{\mGLdrulename{MST\_UnitR}}
\newcommand{\mGLdruleMSTXXUnitR}[1]{\mGLdrule[#1]{%
}{
\emptyset  \odot  \emptyset  \mGLsym{;}  \emptyset  \vdash_{\mathsf{MS} }  \mathsf{i}  \mGLsym{:}  \mathsf{I}}{%
{\mGLdruleMSTXXUnitRName}{}%
}}

\newcommand{\mGLdruleMSTXXUnitLName}[0]{\mGLdrulename{MST\_UnitL}}
\newcommand{\mGLdruleMSTXXUnitL}[1]{\mGLdrule[#1]{%
\mGLpremise{\delta  \odot  \Delta  \mGLsym{;}   ( \Gamma_{{\mathrm{1}}}  \mGLsym{,}  \Gamma_{{\mathrm{2}}} )   \vdash_{\mathsf{MS} }  \mGLnt{l}  \mGLsym{:}  \mGLnt{A}}%
}{
\delta  \odot  \Delta  \mGLsym{;}   ( \Gamma_{{\mathrm{1}}}  \mGLsym{,}  \mGLmv{x}  \mGLsym{:}  \mathsf{I}  \mGLsym{,}  \Gamma_{{\mathrm{2}}} )   \vdash_{\mathsf{MS} }  \mathsf{let} \, \mathsf{i} \, \mGLsym{=}  \mGLmv{x} \, \mathsf{in} \, \mGLnt{l}  \mGLsym{:}  \mGLnt{A}}{%
{\mGLdruleMSTXXUnitLName}{}%
}}

\newcommand{\mGLdruleMSTXXGUnitLName}[0]{\mGLdrulename{MST\_GUnitL}}
\newcommand{\mGLdruleMSTXXGUnitL}[1]{\mGLdrule[#1]{%
\mGLpremise{ (  \delta_{{\mathrm{1}}}  \mGLsym{,}  \delta_{{\mathrm{2}}}  )   \odot   ( \Delta_{{\mathrm{1}}}  \mGLsym{,}  \Delta_{{\mathrm{2}}} )   \mGLsym{;}  \Gamma  \vdash_{\mathsf{MS} }  \mGLnt{l}  \mGLsym{:}  \mGLnt{A}}%
}{
 (  \delta_{{\mathrm{1}}}  \mGLsym{,}  \mGLnt{r}  \mGLsym{,}  \delta_{{\mathrm{2}}}  )   \odot   ( \Delta_{{\mathrm{1}}}  \mGLsym{,}  \mGLmv{z}  \mGLsym{:}  \mathsf{J}  \mGLsym{,}  \Delta_{{\mathrm{2}}} )   \mGLsym{;}  \Gamma  \vdash_{\mathsf{MS} }  \mathsf{let} \, \mathsf{j} \, \mGLsym{=}  \mGLmv{z} \, \mathsf{in} \, \mGLnt{l}  \mGLsym{:}  \mGLnt{A}}{%
{\mGLdruleMSTXXGUnitLName}{}%
}}

\newcommand{\mGLdruleMSTXXTenLName}[0]{\mGLdrulename{MST\_TenL}}
\newcommand{\mGLdruleMSTXXTenL}[1]{\mGLdrule[#1]{%
\mGLpremise{\delta  \odot  \Delta  \mGLsym{;}   ( \Gamma_{{\mathrm{1}}}  \mGLsym{,}  \mGLmv{x}  \mGLsym{:}  \mGLnt{A}  \mGLsym{,}  \mGLmv{y}  \mGLsym{:}  \mGLnt{B}  \mGLsym{,}  \Gamma_{{\mathrm{2}}} )   \vdash_{\mathsf{MS} }  \mGLnt{l}  \mGLsym{:}  \mGLnt{C}}%
}{
\delta  \odot  \Delta  \mGLsym{;}   ( \Gamma_{{\mathrm{1}}}  \mGLsym{,}  \mGLmv{z}  \mGLsym{:}  \mGLnt{A}  \otimes  \mGLnt{B}  \mGLsym{,}  \Gamma_{{\mathrm{2}}} )   \vdash_{\mathsf{MS} }  \mathsf{let} \, \mGLsym{(}  \mGLmv{x}  \mGLsym{,}  \mGLmv{y}  \mGLsym{)}  \mGLsym{=}  \mGLmv{z} \, \mathsf{in} \, \mGLnt{l}  \mGLsym{:}  \mGLnt{C}}{%
{\mGLdruleMSTXXTenLName}{}%
}}

\newcommand{\mGLdruleMSTXXTenRName}[0]{\mGLdrulename{MST\_TenR}}
\newcommand{\mGLdruleMSTXXTenR}[1]{\mGLdrule[#1]{%
\mGLpremise{\delta_{{\mathrm{1}}}  \odot  \Delta_{{\mathrm{1}}}  \mGLsym{;}  \Gamma_{{\mathrm{1}}}  \vdash_{\mathsf{MS} }  \mGLnt{l_{{\mathrm{1}}}}  \mGLsym{:}  \mGLnt{A}}%
\mGLpremise{\delta_{{\mathrm{2}}}  \odot  \Delta_{{\mathrm{2}}}  \mGLsym{;}  \Gamma_{{\mathrm{2}}}  \vdash_{\mathsf{MS} }  \mGLnt{l_{{\mathrm{2}}}}  \mGLsym{:}  \mGLnt{B}}%
}{
 (  \delta_{{\mathrm{1}}}  \mGLsym{,}  \delta_{{\mathrm{2}}}  )   \odot   ( \Delta_{{\mathrm{1}}}  \mGLsym{,}  \Delta_{{\mathrm{2}}} )   \mGLsym{;}   ( \Gamma_{{\mathrm{1}}}  \mGLsym{,}  \Gamma_{{\mathrm{2}}} )   \vdash_{\mathsf{MS} }  \mGLsym{(}  \mGLnt{l_{{\mathrm{1}}}}  \mGLsym{,}  \mGLnt{l_{{\mathrm{2}}}}  \mGLsym{)}  \mGLsym{:}  \mGLnt{A}  \otimes  \mGLnt{B}}{%
{\mGLdruleMSTXXTenRName}{}%
}}

\newcommand{\mGLdruleMSTXXImpRName}[0]{\mGLdrulename{MST\_ImpR}}
\newcommand{\mGLdruleMSTXXImpR}[1]{\mGLdrule[#1]{%
\mGLpremise{\delta  \odot  \Delta  \mGLsym{;}   ( \Gamma  \mGLsym{,}  \mGLmv{x}  \mGLsym{:}  \mGLnt{A} )   \vdash_{\mathsf{MS} }  \mGLnt{l}  \mGLsym{:}  \mGLnt{B}}%
}{
\delta  \odot  \Delta  \mGLsym{;}  \Gamma  \vdash_{\mathsf{MS} }  \lambda  \mGLmv{x}  \mGLsym{.}  \mGLnt{l}  \mGLsym{:}  \mGLnt{A}  \multimap  \mGLnt{B}}{%
{\mGLdruleMSTXXImpRName}{}%
}}

\newcommand{\mGLdruleMSTXXImpLName}[0]{\mGLdrulename{MST\_ImpL}}
\newcommand{\mGLdruleMSTXXImpL}[1]{\mGLdrule[#1]{%
\mGLpremise{\delta_{{\mathrm{2}}}  \odot  \Delta_{{\mathrm{2}}}  \mGLsym{;}  \Gamma_{{\mathrm{2}}}  \vdash_{\mathsf{MS} }  \mGLnt{l_{{\mathrm{2}}}}  \mGLsym{:}  \mGLnt{A}}%
\mGLpremise{\delta_{{\mathrm{1}}}  \odot  \Delta_{{\mathrm{1}}}  \mGLsym{;}   ( \Gamma_{{\mathrm{1}}}  \mGLsym{,}  \mGLmv{x}  \mGLsym{:}  \mGLnt{B}  \mGLsym{,}  \Gamma_{{\mathrm{3}}} )   \vdash_{\mathsf{MS} }  \mGLnt{l_{{\mathrm{1}}}}  \mGLsym{:}  \mGLnt{C}}%
}{
 (  \delta_{{\mathrm{1}}}  \mGLsym{,}  \delta_{{\mathrm{2}}}  )   \odot   ( \Delta_{{\mathrm{1}}}  \mGLsym{,}  \Delta_{{\mathrm{2}}} )   \mGLsym{;}   ( \Gamma_{{\mathrm{1}}}  \mGLsym{,}  \mGLmv{z}  \mGLsym{:}  \mGLnt{A}  \multimap  \mGLnt{B}  \mGLsym{,}  \Gamma_{{\mathrm{2}}}  \mGLsym{,}  \Gamma_{{\mathrm{3}}} )   \vdash_{\mathsf{MS} }  \mGLsym{[}  \mGLmv{z} \, \mGLnt{l_{{\mathrm{2}}}}  \mGLsym{/}  \mGLmv{x}  \mGLsym{]}  \mGLnt{l_{{\mathrm{1}}}}  \mGLsym{:}  \mGLnt{C}}{%
{\mGLdruleMSTXXImpLName}{}%
}}

\newcommand{\mGLdruleMSTXXSubName}[0]{\mGLdrulename{MST\_Sub}}
\newcommand{\mGLdruleMSTXXSub}[1]{\mGLdrule[#1]{%
\mGLpremise{\delta_{{\mathrm{1}}}  \odot  \Delta  \mGLsym{;}  \Gamma  \vdash_{\mathsf{MS} }  \mGLnt{l}  \mGLsym{:}  \mGLnt{B}  \quad  \delta_{{\mathrm{1}}}  \leq  \delta_{{\mathrm{2}}}}%
}{
\delta_{{\mathrm{2}}}  \odot  \Delta  \mGLsym{;}  \Gamma  \vdash_{\mathsf{MS} }  \mGLnt{l}  \mGLsym{:}  \mGLnt{B}}{%
{\mGLdruleMSTXXSubName}{}%
}}

\newcommand{\mGLdruleMSTXXGTenLName}[0]{\mGLdrulename{MST\_GTenL}}
\newcommand{\mGLdruleMSTXXGTenL}[1]{\mGLdrule[#1]{%
\mGLpremise{ (  \delta_{{\mathrm{1}}}  \mGLsym{,}  \mGLnt{r}  \mGLsym{,}  \mGLnt{r}  \mGLsym{,}  \delta_{{\mathrm{2}}}  )   \odot   ( \Delta_{{\mathrm{1}}}  \mGLsym{,}  \mGLmv{x}  \mGLsym{:}  \mGLnt{X}  \mGLsym{,}  \mGLmv{y}  \mGLsym{:}  \mGLnt{Y}  \mGLsym{,}  \Delta_{{\mathrm{2}}} )   \mGLsym{;}  \Gamma  \vdash_{\mathsf{MS} }  \mGLnt{l}  \mGLsym{:}  \mGLnt{A}}%
}{
 (  \delta_{{\mathrm{1}}}  \mGLsym{,}  \mGLnt{r}  \mGLsym{,}  \delta_{{\mathrm{2}}}  )   \odot   ( \Delta_{{\mathrm{1}}}  \mGLsym{,}  \mGLmv{z}  \mGLsym{:}  \mGLnt{X}  \boxtimes  \mGLnt{Y}  \mGLsym{,}  \Delta_{{\mathrm{2}}} )   \mGLsym{;}  \Gamma  \vdash_{\mathsf{MS} }  \mathsf{let} \, \mGLsym{(}  \mGLmv{x}  \mGLsym{,}  \mGLmv{y}  \mGLsym{)}  \mGLsym{=}  \mGLmv{z} \, \mathsf{in} \, \mGLnt{l}  \mGLsym{:}  \mGLnt{A}}{%
{\mGLdruleMSTXXGTenLName}{}%
}}

\newcommand{\mGLdruleMSTXXLinLName}[0]{\mGLdrulename{MST\_LinL}}
\newcommand{\mGLdruleMSTXXLinL}[1]{\mGLdrule[#1]{%
\mGLpremise{\delta  \odot  \Delta  \mGLsym{;}   ( \mGLmv{x}  \mGLsym{:}  \mGLnt{A}  \mGLsym{,}  \Gamma )   \vdash_{\mathsf{MS} }  \mGLnt{l}  \mGLsym{:}  \mGLnt{B}}%
}{
 (  \delta  \mGLsym{,}  1  )   \odot   ( \Delta  \mGLsym{,}  \mGLmv{z}  \mGLsym{:}  \mathsf{Lin} \, \mGLnt{A} )   \mGLsym{;}  \Gamma  \vdash_{\mathsf{MS} }  \mGLsym{[}   \mathsf{Unlin} \, \mGLmv{z}   \mGLsym{/}  \mGLmv{x}  \mGLsym{]}  \mGLnt{l}  \mGLsym{:}  \mGLnt{B}}{%
{\mGLdruleMSTXXLinLName}{}%
}}

\newcommand{\mGLdruleMSTXXGrdLName}[0]{\mGLdrulename{MST\_GrdL}}
\newcommand{\mGLdruleMSTXXGrdL}[1]{\mGLdrule[#1]{%
\mGLpremise{ (  \delta  \mGLsym{,}  \mGLnt{r}  )   \odot   ( \Delta  \mGLsym{,}  \mGLmv{x}  \mGLsym{:}  \mGLnt{X} )   \mGLsym{;}  \Gamma  \vdash_{\mathsf{MS} }  \mGLnt{l}  \mGLsym{:}  \mGLnt{C}}%
}{
\delta  \odot  \Delta  \mGLsym{;}   ( \mGLmv{z}  \mGLsym{:}   \mathsf{Grd} _{ \mGLnt{r} }\, \mGLnt{X}   \mGLsym{,}  \Gamma )   \vdash_{\mathsf{MS} }  \mathsf{let} \, \mathsf{Grd} \, \mGLnt{r} \, \mGLmv{x}  \mGLsym{=}  \mGLmv{z} \, \mathsf{in} \, \mGLnt{l}  \mGLsym{:}  \mGLnt{C}}{%
{\mGLdruleMSTXXGrdLName}{}%
}}

\newcommand{\mGLdruleMSTXXGrdRName}[0]{\mGLdrulename{MST\_GrdR}}
\newcommand{\mGLdruleMSTXXGrdR}[1]{\mGLdrule[#1]{%
\mGLpremise{\delta  \odot  \Delta  \vdash_{\mathsf{GS} }  \mGLnt{t}  \mGLsym{:}  \mGLnt{X}}%
}{
 \mGLnt{r}   *   \delta   \odot  \Delta  \mGLsym{;}  \emptyset  \vdash_{\mathsf{MS} }  \mathsf{Grd} \, \mGLnt{r} \, \mGLnt{t}  \mGLsym{:}   \mathsf{Grd} _{ \mGLnt{r} }\, \mGLnt{X} }{%
{\mGLdruleMSTXXGrdRName}{}%
}}

\newcommand{\mGLdruleMSTXXCutName}[0]{\mGLdrulename{MST\_Cut}}
\newcommand{\mGLdruleMSTXXCut}[1]{\mGLdrule[#1]{%
\mGLpremise{\delta_{{\mathrm{2}}}  \odot  \Delta_{{\mathrm{2}}}  \mGLsym{;}  \Gamma_{{\mathrm{2}}}  \vdash_{\mathsf{MS} }  \mGLnt{l_{{\mathrm{1}}}}  \mGLsym{:}  \mGLnt{A}}%
\mGLpremise{\delta_{{\mathrm{1}}}  \odot  \Delta_{{\mathrm{1}}}  \mGLsym{;}   ( \Gamma_{{\mathrm{1}}}  \mGLsym{,}  \mGLmv{x}  \mGLsym{:}  \mGLnt{A}  \mGLsym{,}  \Gamma_{{\mathrm{3}}} )   \vdash_{\mathsf{MS} }  \mGLnt{l_{{\mathrm{2}}}}  \mGLsym{:}  \mGLnt{B}}%
}{
 (  \delta_{{\mathrm{1}}}  \mGLsym{,}  \delta_{{\mathrm{2}}}  )   \odot   ( \Delta_{{\mathrm{1}}}  \mGLsym{,}  \Delta_{{\mathrm{2}}} )   \mGLsym{;}   ( \Gamma_{{\mathrm{1}}}  \mGLsym{,}  \Gamma_{{\mathrm{2}}}  \mGLsym{,}  \Gamma_{{\mathrm{3}}} )   \vdash_{\mathsf{MS} }  \mGLsym{[}  \mGLnt{l_{{\mathrm{1}}}}  \mGLsym{/}  \mGLmv{x}  \mGLsym{]}  \mGLnt{l_{{\mathrm{2}}}}  \mGLsym{:}  \mGLnt{B}}{%
{\mGLdruleMSTXXCutName}{}%
}}

\newcommand{\mGLdruleMSTXXGCutName}[0]{\mGLdrulename{MST\_GCut}}
\newcommand{\mGLdruleMSTXXGCut}[1]{\mGLdrule[#1]{%
\mGLpremise{\delta_{{\mathrm{2}}}  \odot  \Delta_{{\mathrm{2}}}  \vdash_{\mathsf{GS} }  \mGLnt{t}  \mGLsym{:}  \mGLnt{X}}%
\mGLpremise{ (  \delta_{{\mathrm{1}}}  \mGLsym{,}  \mGLnt{r}  \mGLsym{,}  \delta_{{\mathrm{3}}}  )   \odot   ( \Delta_{{\mathrm{1}}}  \mGLsym{,}  \mGLmv{x}  \mGLsym{:}  \mGLnt{X}  \mGLsym{,}  \Delta_{{\mathrm{3}}} )   \mGLsym{;}  \Gamma  \vdash_{\mathsf{MS} }  \mGLnt{l}  \mGLsym{:}  \mGLnt{B}}%
}{
 (  \delta_{{\mathrm{1}}}  \mGLsym{,}   \mGLnt{r}   *   \delta_{{\mathrm{2}}}   \mGLsym{,}  \delta_{{\mathrm{3}}}  )   \odot   ( \Delta_{{\mathrm{1}}}  \mGLsym{,}  \Delta_{{\mathrm{2}}}  \mGLsym{,}  \Delta_{{\mathrm{3}}} )   \mGLsym{;}  \Gamma  \vdash_{\mathsf{MS} }  \mGLsym{[}  \mGLnt{t}  \mGLsym{/}  \mGLmv{x}  \mGLsym{]}  \mGLnt{l}  \mGLsym{:}  \mGLnt{B}}{%
{\mGLdruleMSTXXGCutName}{}%
}}

\newcommand{\mGLdruleMSTXXMGCutName}[0]{\mGLdrulename{MST\_MGCut}}
\newcommand{\mGLdruleMSTXXMGCut}[1]{\mGLdrule[#1]{%
\mGLpremise{\delta_{{\mathrm{2}}}  \odot  \Delta_{{\mathrm{2}}}  \vdash_{\mathsf{GS} }  \mGLnt{t}  \mGLsym{:}  \mGLnt{X}}%
\mGLpremise{ (  \delta_{{\mathrm{1}}}  \mGLsym{,}  \delta  \mGLsym{,}  \delta_{{\mathrm{3}}}  )   \odot   ( \Delta_{{\mathrm{1}}}  \mGLsym{,}   \mGLnt{X} ^{ \mGLmv{n} }   \mGLsym{,}  \Delta_{{\mathrm{3}}} )   \mGLsym{;}  \Gamma  \vdash_{\mathsf{MS} }  \mGLnt{l}  \mGLsym{:}  \mGLnt{B}}%
}{
 (  \delta_{{\mathrm{1}}}  \mGLsym{,}   (   \delta  \boxast [ { \delta_{{\mathrm{2}}} }^{ \mGLmv{n} } ]   )   \mGLsym{,}  \delta_{{\mathrm{3}}}  )   \odot   ( \Delta_{{\mathrm{1}}}  \mGLsym{,}  \Delta_{{\mathrm{2}}}  \mGLsym{,}  \Delta_{{\mathrm{3}}} )   \mGLsym{;}  \Gamma  \vdash_{\mathsf{MS} }  \mGLsym{[}   \mGLnt{t} ,\ldots, \mGLnt{t}   \mGLsym{/}  \mGLmv{x_{{\mathrm{1}}}}  \mGLsym{,} \, ... \, \mGLsym{,}  \mGLmv{x_{\mGLmv{n}}}  \mGLsym{]}  \mGLnt{l}  \mGLsym{:}  \mGLnt{B}}{%
{\mGLdruleMSTXXMGCutName}{}%
}}

\newcommand{\mGLdruleMSTXXWeakName}[0]{\mGLdrulename{MST\_Weak}}
\newcommand{\mGLdruleMSTXXWeak}[1]{\mGLdrule[#1]{%
\mGLpremise{ (  \delta_{{\mathrm{1}}}  \mGLsym{,}  \delta_{{\mathrm{2}}}  )   \odot   ( \Delta_{{\mathrm{1}}}  \mGLsym{,}  \Delta_{{\mathrm{2}}} )   \mGLsym{;}  \Gamma  \vdash_{\mathsf{MS} }  \mGLnt{l}  \mGLsym{:}  \mGLnt{B}}%
}{
 (  \delta_{{\mathrm{1}}}  \mGLsym{,}  \mathsf{0}  \mGLsym{,}  \delta_{{\mathrm{2}}}  )   \odot   ( \Delta_{{\mathrm{1}}}  \mGLsym{,}  \mGLmv{x}  \mGLsym{:}  \mGLnt{X}  \mGLsym{,}  \Delta_{{\mathrm{2}}} )   \mGLsym{;}  \Gamma  \vdash_{\mathsf{MS} }  \mGLnt{l}  \mGLsym{:}  \mGLnt{B}}{%
{\mGLdruleMSTXXWeakName}{}%
}}

\newcommand{\mGLdruleMSTXXContName}[0]{\mGLdrulename{MST\_Cont}}
\newcommand{\mGLdruleMSTXXCont}[1]{\mGLdrule[#1]{%
\mGLpremise{ (  \delta_{{\mathrm{1}}}  \mGLsym{,}  \mGLnt{r_{{\mathrm{1}}}}  \mGLsym{,}  \mGLnt{r_{{\mathrm{2}}}}  \mGLsym{,}  \delta_{{\mathrm{2}}}  )   \odot   ( \Delta_{{\mathrm{1}}}  \mGLsym{,}  \mGLmv{x}  \mGLsym{:}  \mGLnt{X}  \mGLsym{,}  \mGLmv{y}  \mGLsym{:}  \mGLnt{X}  \mGLsym{,}  \Delta_{{\mathrm{2}}} )   \mGLsym{;}  \Gamma  \vdash_{\mathsf{MS} }  \mGLnt{l}  \mGLsym{:}  \mGLnt{B}}%
}{
 (  \delta_{{\mathrm{1}}}  \mGLsym{,}  \mGLnt{r_{{\mathrm{1}}}}  +  \mGLnt{r_{{\mathrm{2}}}}  \mGLsym{,}  \delta_{{\mathrm{2}}}  )   \odot   ( \Delta_{{\mathrm{1}}}  \mGLsym{,}  \mGLmv{x}  \mGLsym{:}  \mGLnt{X}  \mGLsym{,}  \Delta_{{\mathrm{2}}} )   \mGLsym{;}  \Gamma  \vdash_{\mathsf{MS} }  \mGLsym{[}  \mGLmv{x}  \mGLsym{/}  \mGLmv{y}  \mGLsym{]}  \mGLnt{l}  \mGLsym{:}  \mGLnt{B}}{%
{\mGLdruleMSTXXContName}{}%
}}

\newcommand{\mGLdruleMSTXXGExName}[0]{\mGLdrulename{MST\_GEx}}
\newcommand{\mGLdruleMSTXXGEx}[1]{\mGLdrule[#1]{%
\mGLpremise{ (  \delta_{{\mathrm{1}}}  \mGLsym{,}  \mGLnt{r_{{\mathrm{1}}}}  \mGLsym{,}  \mGLnt{r_{{\mathrm{2}}}}  \mGLsym{,}  \delta_{{\mathrm{2}}}  )   \odot   ( \Delta_{{\mathrm{1}}}  \mGLsym{,}  \mGLmv{x}  \mGLsym{:}  \mGLnt{X}  \mGLsym{,}  \mGLmv{y}  \mGLsym{:}  \mGLnt{Y}  \mGLsym{,}  \Delta_{{\mathrm{2}}} )   \mGLsym{;}  \Gamma  \vdash_{\mathsf{MS} }  \mGLnt{l}  \mGLsym{:}  \mGLnt{B}}%
}{
 (  \delta_{{\mathrm{1}}}  \mGLsym{,}  \mGLnt{r_{{\mathrm{2}}}}  \mGLsym{,}  \mGLnt{r_{{\mathrm{1}}}}  \mGLsym{,}  \delta_{{\mathrm{2}}}  )   \odot   ( \Delta_{{\mathrm{1}}}  \mGLsym{,}  \mGLmv{x}  \mGLsym{:}  \mGLnt{Y}  \mGLsym{,}  \mGLmv{y}  \mGLsym{:}  \mGLnt{X}  \mGLsym{,}  \Delta_{{\mathrm{2}}} )   \mGLsym{;}  \Gamma  \vdash_{\mathsf{MS} }  \mGLnt{l}  \mGLsym{:}  \mGLnt{B}}{%
{\mGLdruleMSTXXGExName}{}%
}}

\newcommand{\mGLdruleMSTXXExName}[0]{\mGLdrulename{MST\_Ex}}
\newcommand{\mGLdruleMSTXXEx}[1]{\mGLdrule[#1]{%
\mGLpremise{\delta  \odot  \Delta  \mGLsym{;}   ( \Gamma_{{\mathrm{1}}}  \mGLsym{,}  \mGLmv{x}  \mGLsym{:}  \mGLnt{A}  \mGLsym{,}  \mGLmv{y}  \mGLsym{:}  \mGLnt{B}  \mGLsym{,}  \Gamma_{{\mathrm{2}}} )   \vdash_{\mathsf{MS} }  \mGLnt{l}  \mGLsym{:}  \mGLnt{C}}%
}{
\delta  \odot  \Delta  \mGLsym{;}   ( \Gamma_{{\mathrm{1}}}  \mGLsym{,}  \mGLmv{y}  \mGLsym{:}  \mGLnt{B}  \mGLsym{,}  \mGLmv{x}  \mGLsym{:}  \mGLnt{A}  \mGLsym{,}  \Gamma_{{\mathrm{2}}} )   \vdash_{\mathsf{MS} }  \mGLnt{l}  \mGLsym{:}  \mGLnt{C}}{%
{\mGLdruleMSTXXExName}{}%
}}

\newcommand{\mGLdruleMSTXXGImpLName}[0]{\mGLdrulename{MST\_GImpL}}
\newcommand{\mGLdruleMSTXXGImpL}[1]{\mGLdrule[#1]{%
\mGLpremise{\delta_{{\mathrm{2}}}  \odot  \Delta_{{\mathrm{2}}}  \vdash_{\mathsf{GS} }  \mGLnt{t}  \mGLsym{:}  \mGLnt{X}}%
\mGLpremise{ (  \delta_{{\mathrm{1}}}  \mGLsym{,}  \delta_{{\mathrm{3}}}  )   \odot   ( \Delta_{{\mathrm{1}}}  \mGLsym{,}  \Delta_{{\mathrm{3}}} )   \mGLsym{;}   ( \Gamma_{{\mathrm{1}}}  \mGLsym{,}  \mGLmv{x}  \mGLsym{:}  \mGLnt{A}  \mGLsym{,}  \Gamma_{{\mathrm{2}}} )   \vdash_{\mathsf{MS} }  \mGLnt{l}  \mGLsym{:}  \mGLnt{B}}%
}{
 (  \delta_{{\mathrm{1}}}  \mGLsym{,}   \mGLnt{r}   *   \delta_{{\mathrm{2}}}   \mGLsym{,}  \delta_{{\mathrm{3}}}  )   \odot   ( \Delta_{{\mathrm{1}}}  \mGLsym{,}  \Delta_{{\mathrm{2}}}  \mGLsym{,}  \Delta_{{\mathrm{3}}} )   \mGLsym{;}   ( \Gamma_{{\mathrm{1}}}  \mGLsym{,}  \mGLmv{z}  \mGLsym{:}   \mathsf{Grd} _{ \mGLnt{r} }\, \mGLnt{X}   \multimap  \mGLnt{A}  \mGLsym{,}  \Gamma_{{\mathrm{2}}} )   \vdash_{\mathsf{MS} }  \mGLsym{[}  \mGLmv{z} \, \mGLsym{(}  \mathsf{Grd} \, \mGLnt{r} \, \mGLnt{t}  \mGLsym{)}  \mGLsym{/}  \mGLmv{x}  \mGLsym{]}  \mGLnt{l}  \mGLsym{:}  \mGLnt{B}}{%
{\mGLdruleMSTXXGImpLName}{}%
}}

\newcommand{\mGLdruleMSTXXGImpRName}[0]{\mGLdrulename{MST\_GImpR}}
\newcommand{\mGLdruleMSTXXGImpR}[1]{\mGLdrule[#1]{%
\mGLpremise{ (  \delta  \mGLsym{,}  \mGLnt{r}  )   \odot   ( \Delta  \mGLsym{,}  \mGLmv{x}  \mGLsym{:}  \mGLnt{X} )   \mGLsym{;}  \Gamma  \vdash_{\mathsf{MS} }  \mGLnt{l}  \mGLsym{:}  \mGLnt{A}}%
}{
\delta  \odot  \Delta  \mGLsym{;}  \Gamma  \vdash_{\mathsf{MS} }  \lambda  \mGLmv{y}  \mGLsym{.}  \mGLsym{(}  \mathsf{let} \, \mathsf{Grd} \, \mGLnt{r} \, \mGLmv{x}  \mGLsym{=}  \mGLmv{y} \, \mathsf{in} \, \mGLnt{l}  \mGLsym{)}  \mGLsym{:}   \mathsf{Grd} _{ \mGLnt{r} }\, \mGLnt{X}   \multimap  \mGLnt{A}}{%
{\mGLdruleMSTXXGImpRName}{}%
}}

\newcommand{\mGLdruleMSTXXBoxIName}[0]{\mGLdrulename{MST\_BoxI}}
\newcommand{\mGLdruleMSTXXBoxI}[1]{\mGLdrule[#1]{%
\mGLpremise{\delta  \odot  \Delta  \mGLsym{;}  \emptyset  \vdash_{\mathsf{MS} }  \mGLnt{l}  \mGLsym{:}  \mGLnt{A}}%
}{
 (  \mGLnt{r}  *  \delta  )   \odot  \Delta  \mGLsym{;}  \emptyset  \vdash_{\mathsf{MS} }  \mathsf{Grd} \, \mGLnt{r} \, \mGLsym{(}  \mathsf{Lin} \, \mGLnt{l}  \mGLsym{)}  \mGLsym{:}   \Box _{ \mGLnt{r} } \mGLnt{A} }{%
{\mGLdruleMSTXXBoxIName}{}%
}}

\newcommand{\mGLdruleMSTXXBoxEName}[0]{\mGLdrulename{MST\_BoxE}}
\newcommand{\mGLdruleMSTXXBoxE}[1]{\mGLdrule[#1]{%
\mGLpremise{\delta_{{\mathrm{2}}}  \odot  \Delta_{{\mathrm{2}}}  \mGLsym{;}  \Gamma_{{\mathrm{2}}}  \vdash_{\mathsf{MS} }  \mGLnt{l_{{\mathrm{1}}}}  \mGLsym{:}   \Box _{ \mGLnt{r} } \mGLnt{A} }%
\mGLpremise{ (  \delta_{{\mathrm{1}}}  \mGLsym{,}  \mGLnt{r}  \mGLsym{,}  \delta_{{\mathrm{3}}}  )   \odot   ( \Delta_{{\mathrm{1}}}  \mGLsym{,}  \mGLmv{x}  \mGLsym{:}  \mathsf{Lin} \, \mGLnt{A}  \mGLsym{,}  \Delta_{{\mathrm{3}}} )   \mGLsym{;}  \Gamma_{{\mathrm{1}}}  \vdash_{\mathsf{MS} }  \mGLnt{l_{{\mathrm{2}}}}  \mGLsym{:}  \mGLnt{B}}%
}{
 (  \delta_{{\mathrm{1}}}  \mGLsym{,}  \delta_{{\mathrm{2}}}  \mGLsym{,}  \delta_{{\mathrm{3}}}  )   \odot   ( \Delta_{{\mathrm{1}}}  \mGLsym{,}  \Delta_{{\mathrm{2}}}  \mGLsym{,}  \Delta_{{\mathrm{3}}} )   \mGLsym{;}   ( \Gamma_{{\mathrm{1}}}  \mGLsym{,}  \Gamma_{{\mathrm{2}}} )   \vdash_{\mathsf{MS} }  \mathsf{let} \, \mathsf{Grd} \, \mGLnt{r} \, \mGLmv{x}  \mGLsym{=}  \mGLnt{l_{{\mathrm{1}}}} \, \mathsf{in} \, \mGLnt{l_{{\mathrm{2}}}}  \mGLsym{:}  \mGLnt{B}}{%
{\mGLdruleMSTXXBoxEName}{}%
}}

\newcommand{\mGLdruleGTXXUnitIName}[0]{\mGLdrulename{GT\_UnitI}}

\newcommand{\mGLdruleGTXXUnitEName}[0]{\mGLdrulename{GT\_UnitE}}

\newcommand{\mGLdruleGTXXTenIName}[0]{\mGLdrulename{GT\_TenI}}

\newcommand{\mGLdruleGTXXSubName}[0]{\mGLdrulename{GT\_Sub}}

\newcommand{\mGLdruleGTXXIdName}[0]{\mGLdrulename{GT\_Id}}

\newcommand{\mGLdruleGTXXLinIName}[0]{\mGLdrulename{GT\_LinI}}
\newcommand{\mGLdruleGTXXLinI}[1]{\mGLdrule[#1]{%
\mGLpremise{\delta  \odot  \Delta  \mGLsym{;}  \emptyset  \vdash_{\mathsf{MT} }  \mGLnt{l}  \mGLsym{:}  \mGLnt{B}}%
}{
\delta  \odot  \Delta  \vdash_{\mathsf{GT} }  \mathsf{Lin} \, \mGLnt{l}  \mGLsym{:}  \mathsf{Lin} \, \mGLnt{B}}{%
{\mGLdruleGTXXLinIName}{}%
}}

\newcommand{\mGLdruleGTXXTenEName}[0]{\mGLdrulename{GT\_TenE}}

\newcommand{\mGLdruleGTXXWeakName}[0]{\mGLdrulename{GT\_Weak}}
\newcommand{\mGLdruleGTXXWeak}[1]{\mGLdrule[#1]{%
\mGLpremise{ (  \delta_{{\mathrm{1}}}  \mGLsym{,}  \delta_{{\mathrm{2}}}  )   \odot   ( \Delta_{{\mathrm{1}}}  \mGLsym{,}  \Delta_{{\mathrm{2}}} )   \vdash_{\mathsf{GT} }  \mGLnt{t}  \mGLsym{:}  \mGLnt{Y}}%
}{
 (  \delta_{{\mathrm{1}}}  \mGLsym{,}  \mathsf{0}  \mGLsym{,}  \delta_{{\mathrm{2}}}  )   \odot   ( \Delta_{{\mathrm{1}}}  \mGLsym{,}  \mGLmv{x}  \mGLsym{:}  \mGLnt{X}  \mGLsym{,}  \Delta_{{\mathrm{2}}} )   \vdash_{\mathsf{GT} }  \mGLnt{t}  \mGLsym{:}  \mGLnt{Y}}{%
{\mGLdruleGTXXWeakName}{}%
}}

\newcommand{\mGLdruleGTXXContName}[0]{\mGLdrulename{GT\_Cont}}
\newcommand{\mGLdruleGTXXCont}[1]{\mGLdrule[#1]{%
\mGLpremise{ (  \delta_{{\mathrm{1}}}  \mGLsym{,}  \mGLnt{r_{{\mathrm{1}}}}  \mGLsym{,}  \mGLnt{r_{{\mathrm{2}}}}  \mGLsym{,}  \delta_{{\mathrm{2}}}  )   \odot   ( \Delta_{{\mathrm{1}}}  \mGLsym{,}  \mGLmv{x}  \mGLsym{:}  \mGLnt{X}  \mGLsym{,}  \mGLmv{y}  \mGLsym{:}  \mGLnt{X}  \mGLsym{,}  \Delta_{{\mathrm{2}}} )   \vdash_{\mathsf{GT} }  \mGLnt{t}  \mGLsym{:}  \mGLnt{Y}}%
}{
 (  \delta_{{\mathrm{1}}}  \mGLsym{,}  \mGLnt{r_{{\mathrm{1}}}}  +  \mGLnt{r_{{\mathrm{2}}}}  \mGLsym{,}  \delta_{{\mathrm{2}}}  )   \odot   ( \Delta_{{\mathrm{1}}}  \mGLsym{,}  \mGLmv{x}  \mGLsym{:}  \mGLnt{X}  \mGLsym{,}  \Delta_{{\mathrm{2}}} )   \vdash_{\mathsf{GT} }  \mGLsym{[}  \mGLmv{x}  \mGLsym{/}  \mGLmv{y}  \mGLsym{]}  \mGLnt{t}  \mGLsym{:}  \mGLnt{Y}}{%
{\mGLdruleGTXXContName}{}%
}}

\newcommand{\mGLdruleGTXXExName}[0]{\mGLdrulename{GT\_Ex}}
\newcommand{\mGLdruleGTXXEx}[1]{\mGLdrule[#1]{%
\mGLpremise{ (  \delta_{{\mathrm{1}}}  \mGLsym{,}  \mGLnt{r_{{\mathrm{1}}}}  \mGLsym{,}  \mGLnt{r_{{\mathrm{2}}}}  \mGLsym{,}  \delta_{{\mathrm{2}}}  )   \odot   ( \Delta_{{\mathrm{1}}}  \mGLsym{,}  \mGLmv{x}  \mGLsym{:}  \mGLnt{X}  \mGLsym{,}  \mGLmv{y}  \mGLsym{:}  \mGLnt{Y}  \mGLsym{,}  \Delta_{{\mathrm{2}}} )   \vdash_{\mathsf{GT} }  \mGLnt{t}  \mGLsym{:}  \mGLnt{Z}}%
}{
 (  \delta_{{\mathrm{1}}}  \mGLsym{,}  \mGLnt{r_{{\mathrm{2}}}}  \mGLsym{,}  \mGLnt{r_{{\mathrm{1}}}}  \mGLsym{,}  \delta_{{\mathrm{2}}}  )   \odot   ( \Delta_{{\mathrm{1}}}  \mGLsym{,}  \mGLmv{y}  \mGLsym{:}  \mGLnt{Y}  \mGLsym{,}  \mGLmv{x}  \mGLsym{:}  \mGLnt{X}  \mGLsym{,}  \Delta_{{\mathrm{2}}} )   \vdash_{\mathsf{GT} }  \mGLnt{t}  \mGLsym{:}  \mGLnt{Z}}{%
{\mGLdruleGTXXExName}{}%
}}

\newcommand{\mGLdruleMTXXIdName}[0]{\mGLdrulename{MT\_Id}}

\newcommand{\mGLdruleMTXXUnitIName}[0]{\mGLdrulename{MT\_UnitI}}

\newcommand{\mGLdruleMTXXUnitEName}[0]{\mGLdrulename{MT\_UnitE}}

\newcommand{\mGLdruleMTXXGUnitEName}[0]{\mGLdrulename{MT\_GUnitE}}

\newcommand{\mGLdruleMTXXGTenEName}[0]{\mGLdrulename{MT\_GTenE}}

\newcommand{\mGLdruleMTXXTenIName}[0]{\mGLdrulename{MT\_TenI}}

\newcommand{\mGLdruleMTXXTenEName}[0]{\mGLdrulename{MT\_TenE}}

\newcommand{\mGLdruleMTXXImpIName}[0]{\mGLdrulename{MT\_ImpI}}

\newcommand{\mGLdruleMTXXImpEName}[0]{\mGLdrulename{MT\_ImpE}}

\newcommand{\mGLdruleMTXXGrdEName}[0]{\mGLdrulename{MT\_GrdE}}
\newcommand{\mGLdruleMTXXGrdE}[1]{\mGLdrule[#1]{%
\mGLpremise{\delta_{{\mathrm{2}}}  \odot  \Delta_{{\mathrm{2}}}  \mGLsym{;}  \Gamma_{{\mathrm{2}}}  \vdash_{\mathsf{MT} }  \mGLnt{l_{{\mathrm{1}}}}  \mGLsym{:}   \mathsf{Grd} _{ \mGLnt{r} }\, \mGLnt{X} }%
\mGLpremise{ (  \delta_{{\mathrm{1}}}  \mGLsym{,}  \mGLnt{r}  \mGLsym{,}  \delta_{{\mathrm{3}}}  )   \odot   ( \Delta_{{\mathrm{1}}}  \mGLsym{,}  \mGLmv{x}  \mGLsym{:}  \mGLnt{X}  \mGLsym{,}  \Delta_{{\mathrm{3}}} )   \mGLsym{;}  \Gamma_{{\mathrm{1}}}  \vdash_{\mathsf{MT} }  \mGLnt{l_{{\mathrm{2}}}}  \mGLsym{:}  \mGLnt{B}}%
}{
 (  \delta_{{\mathrm{1}}}  \mGLsym{,}  \delta_{{\mathrm{2}}}  \mGLsym{,}  \delta_{{\mathrm{3}}}  )   \odot   ( \Delta_{{\mathrm{1}}}  \mGLsym{,}  \Delta_{{\mathrm{2}}}  \mGLsym{,}  \Delta_{{\mathrm{3}}} )   \mGLsym{;}   ( \Gamma_{{\mathrm{1}}}  \mGLsym{,}  \Gamma_{{\mathrm{2}}} )   \vdash_{\mathsf{MT} }  \mathsf{let} \, \mathsf{Grd} \, \mGLnt{r} \, \mGLmv{x}  \mGLsym{=}  \mGLnt{l_{{\mathrm{1}}}} \, \mathsf{in} \, \mGLnt{l_{{\mathrm{2}}}}  \mGLsym{:}  \mGLnt{B}}{%
{\mGLdruleMTXXGrdEName}{}%
}}

\newcommand{\mGLdruleMTXXGSubName}[0]{\mGLdrulename{MT\_GSub}}

\newcommand{\mGLdruleMTXXGrdIName}[0]{\mGLdrulename{MT\_GrdI}}
\newcommand{\mGLdruleMTXXGrdI}[1]{\mGLdrule[#1]{%
\mGLpremise{\delta  \odot  \Delta  \vdash_{\mathsf{GT} }  \mGLnt{t}  \mGLsym{:}  \mGLnt{X}}%
}{
\mGLnt{r}  *  \delta  \odot  \Delta  \mGLsym{;}  \emptyset  \vdash_{\mathsf{MT} }  \mathsf{Grd} \, \mGLnt{r} \, \mGLnt{t}  \mGLsym{:}   \mathsf{Grd} _{ \mGLnt{r} }\, \mGLnt{X} }{%
{\mGLdruleMTXXGrdIName}{}%
}}

\newcommand{\mGLdruleMTXXLinEName}[0]{\mGLdrulename{MT\_LinE}}
\newcommand{\mGLdruleMTXXLinE}[1]{\mGLdrule[#1]{%
\mGLpremise{\delta  \odot  \Delta  \vdash_{\mathsf{GT} }  \mGLnt{t}  \mGLsym{:}  \mathsf{Lin} \, \mGLnt{A}}%
}{
\delta  \odot  \Delta  \mGLsym{;}  \emptyset  \vdash_{\mathsf{MT} }   \mathsf{Unlin} \, \mGLnt{t}   \mGLsym{:}  \mGLnt{A}}{%
{\mGLdruleMTXXLinEName}{}%
}}

\newcommand{\mGLdruleMTXXWeakName}[0]{\mGLdrulename{MT\_Weak}}

\newcommand{\mGLdruleMTXXContName}[0]{\mGLdrulename{MT\_Cont}}

\newcommand{\mGLdruleMTXXGExName}[0]{\mGLdrulename{MT\_GEx}}

\newcommand{\mGLdruleMTXXExName}[0]{\mGLdrulename{MT\_Ex}}



%% file: math.tex
\usepackage[barr]{xy}
\usepackage{xypic}
\usepackage{mathpartir}
\usepackage{relsize}

\let\mto\to
\renewcommand{\to}{\rightarrow}
\newcommand{\id}{\mathsf{id}}

\newcommand{\cat}[1]{\mathcal{#1}}

\newcommand{\op}[1]{{#1}^{\mathsf{op}}}

\newcommand{\interp}[1]{\llbracket #1 \rrbracket}
\newcommand{\interpMS}[1]{{\llbracket #1 \rrbracket}^\MS}
\newcommand{\interpGS}[1]{{\llbracket #1 \rrbracket}^\GS}

\newcommand{\func}[1]{\mathsf{#1}}

\newcommand{\Lin}{\mathsf{Lin}}

\newcommand{\Mny}{\mathsf{Mny}}

\newcommand{\ruleinterp}[1]{\Biggl\llbracket #1 \Biggr\rrbracket}

\makeatletter
\newcommand\dotover{\leavevmode\cleaders\hb@xt@ .22em{\hss $\cdot$\hss}\hfill\kern\z@}

\makeatother

\newcommand{\morph}[1]{\mto^{\text{\normalsize{$#1$}}}}

\newcommand{\catAux}[2]{[#1,#2]^\odot}
\newcommand{\catUnder}[2]{\catAux{\cat{#1}}{\cat{#2}}}


%% file: abstract.tex
Graded modal logics generalise standard modal logics
via families of modalities indexed by an algebraic structure
whose operations mediate between the different modalities.
The graded ``of-course'' modality $!_r$ captures how many times a proposition is used
and has an analogous interpretation to the of-course modality from linear
logic; the of-course modality from linear logic can be modelled
by a linear exponential comonad and graded of-course can be modelled
by a graded linear exponential comonad. Benton showed in his seminal
paper on Linear/Non-Linear logic that the of-course modality
can be split into two modalities connecting intuitionistic logic with
linear logic, forming a symmetric monoidal adjunction. Later, Fujii et al. demonstrated that every graded comonad can be decomposed into an adjunction
and a `strict action'. We give a similar result to Benton, leveraging Fujii et al.'s decomposition, showing that graded modalities can be
split into two modalities connecting a graded logic with a graded linear logic.
We propose a sequent calculus, its proof theory and categorical model,
and a natural deduction system
which we show is isomorphic to the sequent calculus system. Interestingly, our
system can also be understood as Linear/Non-Linear logic composed with
an action that adds the grading, further illuminating the shared principles
between linear logic and a class of graded modal logics.

%% file: intro-ottput.tex
Intuitionistic logic has a central role in the foundations of programming
language theory, providing a logical basis for type theories and type systems, and other
program reasoning principles. A significant amount of the
expressivity of proof systems for intuitionistic logic (both natural deduction and
sequent calculus forms) lies within the structure of the hypotheses---the
\emph{context}. Probing the foundations of this part of the logic has,
perhaps surprisingly, yielded the very fertile field of
\emph{substructural logics}~\cite{restall2002introduction} including
influential logics such as linear logic~\cite{Girard:1987} and its
variants, and the Lambek calculus~\cite{lambek1958mathematics}.

By restricting the manipulation of hypotheses in the context we typically arrive
at logics which align more closely with physical reality, where propositions are
instead `resources' that cannot necessarily be copied, discarded, or reordered.
Such restricted logics have been used to construct type systems for
safely manipulating values that should be treated in a resourceful way, such as file handlers, pointers
to mutable memory, or channels~\cite{walker2005substructural,wadler1990linear}.

However, such pervasive restrictions often hamper expressivity
and thus some substructural logics then seek to carefully control the reintroduction of structural rules.
For example, linear logic provides the $!$ modality (`of course') for
reintroducing weakening and contraction of propositions, which linear logic
 otherwise prohibits. However, this modality is coarse-grained: for those
propositions under the modality, it re-enables all the structural rules that
have been removed in linear logic.  \emph{Subexponentials} instead aim to be
more fine-grained, offering families of modalities capturing specific structural
rules~\cite{danos93kgc, kanovich2020soft}. The related notion of
\emph{grading}~\cite{Gaboardi:2016,Orchard:2019} gives an alternate view,
providing an indexed family of modalities whose indices are subject to an
algebra which accounts for any structural rules applied: structural rules are
`counted' by the algebra (whose operations mirror the shape of structural
rules). Bounded Linear Logic~\cite{girard1992bounded} is a special case where
the family of modalities $!_{n } A$ uses indices $n$ which are natural numbers
(or polynomial terms over naturals) counting the upper bound on usage of the
proposition $A$. Various systems generalise this approach to arbitrary semirings
to capture data-flow
properties~\cite{DBLP:journals/pacmpl/AbelB20,atkey2018syntax,DBLP:conf/csl/BreuvartP15,Gaboardi:2016,Ghica:2014,Katsumata:2018,DBLP:conf/birthday/McBride16,Orchard:2019,DBLP:conf/icalp/PetricekOM13,DBLP:conf/icfp/PetricekOM14, DBLP:conf/fscd/PimentelON21}.
Such graded systems annotate hypotheses/variables in the context with elements
of the semiring (`grades') denoting their usage, e.g., $x :_0 A \vdash t : B$ types a
term $t$ which does not use $x$ and $y :_{1 + 1} A \vdash t' : B$ types a term $t'$
in which $y$ is used in two different subterms once each, accounted for by the semiring
addition. A graded modality \emph{internalises} the semiring grade, causing a multiplication
to the grades of any captured dependencies when the graded modality is
introduced, e.g., $y :_{0*(1+1)} A \vdash \Box t' : \Box_0 B$.

We seek here to further
understand the underlying structure of graded modal logics by following
an `adjoint resolution' approach \`{a} la Benton's seminal
``A Mixed Linear and Non-Linear Logic: Proofs, Terms and Models'' at CSL 1994~\cite{Benton:1994}.
Benton showed that the exponential modality of linear logic (modelled by a
comonad) can be decomposed into an adjunction, defining a pair of `adjoint' logics (a linear logic and a
non-linear intuitionistic, or ``Cartesian'', logic) which embed into each other~\cite{Benton:1994}. This provides a beautiful reduction of
the core features of linear logic and its non-linearity modality.
\emph{Adjoint logic} applies the same idea but to
subexponentials~\cite{Pruiksma:2018,pruiksma2020message}.
We follow the same scheme, via the adjoint decomposition of graded
modalities which generalise linear logic's $!$ and which are traditionally modelled by graded exponential comonads~\cite{DBLP:conf/csl/BreuvartP15,Brunel:2014,Gaboardi:2016,Ghica:2014,Katsumata:2018,DBLP:conf/icfp/PetricekOM14}.
Whilst Benton's work has a pair of adjoint modalities mediating between the two sublogics, we have a
pair of a modality $\mathsf{Lin}$ and a \emph{graded modality} $\mathsf{Grd} _{ \mGLnt{r} }$. We give a categorical model, showing that these are captured by an LNL-like adjunction paired with a `strict action' for incorporating the grading, following the Fujii-Katsumata-Melliès adjoint decomposition of graded (co)monads~\cite{Fujii:2016b,Katsumata:2018}.
The result is a pair of logics which serve to explain and clarify
the relationship between linearity and grading. We call our system
\mGLL{} ($\mGL{}$) Logic.

This pair of logics also shines light on a relationship between two
styles of graded system in the literature: those which take
linear types as their basis augmented with a graded modality~\cite{Brunel:2014,Gaboardi:2016,Orchard:2019} versus those with no base
notion of linearity where grading is pervasive, tracking all substructurality~\cite{DBLP:journals/pacmpl/AbelB20,atkey2018syntax,DBLP:journals/pacmpl/BernardyBNJS18,DBLP:journals/pacmpl/ChoudhuryEEW21,DBLP:conf/birthday/McBride16,DBLP:conf/esop/MoonEO21,DBLP:conf/icfp/PetricekOM14}.
Our linear fragment is analogous to the former whilst our graded fragment is analogous
to the latter. The mutual embedding shows that these two styles of
graded logics have a similar relationship to the adjoint relationship of intuitionistic
logic and linear logic.

Aside from the internal motivation of better understanding the
relationship between grading and linearity, an external motivation for
this work is that it can provide a basis for flexible, safe programming with
resources. By separating out the linear fragment from
an intuitionistic graded fragment, one could avoid the strictures of
linearity for working with standard data types which need
not be linear, working only in the linear fragment for handling
resources like file handles. The mutual embedding would allow the programmer to
move smoothly between these two subcalculi, as seen
also in other adjunction-based calculi, e.g., for concurrent programming~\cite{DBLP:conf/fossacs/PfenningG15}.
The focus here however is on the core
theory rather than developing these applications yet.

Since our focus is on the relationship between grading and linearity,
we consider the semiring-graded modalities that generalise linear logic's $!$.
Other flavours of graded modality (e.g., graded
monads for capturing side effecting behaviour~\cite{DBLP:conf/popl/Katsumata14,DBLP:journals/corr/OrchardPM14}) are not
considered here.

\paragraph*{Roadmap}

Section~\ref{sec:mixed_graded/linear_logic} defines a pair of sequent calculi,
the mixed fragment $\MS{}$, which has both linear and graded assumptions, and the graded
fragment $\GS{}$, which has only graded assumptions and no function arrow. As described
above, these calculi have a mutual embedding via modalities between the two.
Section~\ref{sec:mGL-seq_denotational_model} considers the categorical model
of $\mGL{}$ leveraging recent work on the adjoint resolution of graded
comonads~\cite{Fujii:2016b,Katsumata:2018}. Section~\ref{subsec:mgl_nat_deduc}
provides the natural deduction formulation of the calculus, which is proved
equivalent to the sequent calculus version. Section~\ref{sec:discussion}
discusses how this work gives a view on the landscape of graded systems in the literature and considers other related work and future applications.


%% file: mgl-ottput.tex
\setlength\arraycolsep{0.3em}
\begin{figure}[b]
  \begin{gather*}
    \hspace{-1.9em}
    \begin{array}{rl}
    (\GS{}/\GST{}) \quad  t ::= &  \mGLmv{x} \mid \mathsf{j} \mid \mathsf{let} \, \mathsf{j} \, \mGLsym{=}  \mGLnt{t_{{\mathrm{1}}}} \, \mathsf{in} \, \mGLnt{t_{{\mathrm{2}}}}  \\
       \textit{graded} \qquad\quad\;  \mid &  \mGLsym{(}  \mGLnt{t_{{\mathrm{1}}}}  \mGLsym{,}  \mGLnt{t_{{\mathrm{2}}}}  \mGLsym{)} \mid \mathsf{let} \,( \mGLmv{x} , \mGLmv{y} ) =  \mGLnt{t_{{\mathrm{1}}}} \, \mathsf{in} \, \mGLnt{t_{{\mathrm{2}}}} \\
      \mid & \mathsf{Lin} \, \mGLnt{l} \\[4em]
    \end{array}
    \;
    \begin{array}{rl}
    (\MS{}/\MST{}) \quad  l  ::= & \mGLmv{x} \mid \mathsf{i} \mid \mathsf{let} \, \mathsf{i} \, \mGLsym{=}  \mGLnt{l_{{\mathrm{1}}}} \, \mathsf{in} \, \mGLnt{l_{{\mathrm{2}}}} \\
    \; \textit{linear} \qquad\quad\;\;  \mid & \mGLsym{(}  \mGLnt{l_{{\mathrm{1}}}}  \mGLsym{,}  \mGLnt{l_{{\mathrm{2}}}}  \mGLsym{)} \mid \mathsf{let} \, \mGLsym{(}  \mGLmv{x}  \mGLsym{,}  \mGLmv{y}  \mGLsym{)}  \mGLsym{=}  \mGLnt{l_{{\mathrm{1}}}} \, \mathsf{in} \, \mGLnt{l_{{\mathrm{2}}}} \\
          \mid & \lambda  \mGLmv{x}  \mGLsym{.}  \mGLnt{l} \mid \mGLnt{l_{{\mathrm{1}}}} \, \mGLnt{l_{{\mathrm{2}}}} \\
          \mid & \mathsf{Grd} \, \mGLnt{r} \, \mGLnt{t} \mid \mathsf{let} \, \mathsf{Grd} \, \mGLnt{r} \, \mGLmv{x}  \mGLsym{=}  \mGLnt{l_{{\mathrm{1}}}} \, \mathsf{in} \, \mGLnt{l_{{\mathrm{2}}}} \\
          \mid & \mathsf{Unlin} \, \mGLmv{z} \\
          \mid & \mathsf{let} \, \mathsf{j} \, \mGLsym{=}  \mGLmv{z} \, \mathsf{in} \, \mGLnt{l} \mid \mathsf{let} \, \mGLsym{(}  \mGLmv{x}  \mGLsym{,}  \mGLmv{y}  \mGLsym{)}  \mGLsym{=}  \mGLmv{z} \, \mathsf{in} \, \mGLnt{l}
    \end{array}
  \end{gather*}
  \justifying
  Variables are ranged over by $x$, $y$, $z$ in both fragments.

    Terms are mostly grouped above with introduction forms followed by elimination forms,
    though note that in the last two lines of syntax for $\mGLnt{l}$ there are additional eliminators: for the linear modality ($\mathsf{Unlin}$), for units $\mathsf{j}$, and for tensors coming from the graded context.
  \caption{Collected term syntax}
  \label{fig:collected-syntax}
\end{figure}

We present first a sequent calculus for \mGLL{} logic, which comes in the form of a
term assignment. Figure~\ref{fig:collected-syntax} collects the term syntax for reference; it will also be used in Section~\ref{subsec:mgl_nat_deduc} for the natural deduction formulation.
The syntax is explained with reference to its associated proof rules in the
next section.

Benton's approach has two proof systems~\cite{Benton:1994}: one system
of linear propositions (the L of LNL) with two contexts for linear and
non-linear propositions respectively, and one system of non-linear
propositions (the NL in LNL). We generalise this approach to the
graded setting by replacing the non-linear parts with \emph{graded} notions.
Thus, our system ($\mGL{}$) has two analogous proof systems: one of
linear propositions with two contexts for linear and graded propositions,
with judgments subscripted as $\vdash_{\mathsf{MS} }$ (for `Mixed (linear/graded) Sequent'),
and one system of graded propositions, with judgments
subscripted as $\vdash_{\mathsf{GS} }$ (`Graded Sequent').

The syntax of Benton's propositions is split into two, `\emph{conventional}' (i.e.,
Cartesian / non-linear) and \emph{linear}~\cite{Benton:1994}.
The syntax of our propositions is analogously split into two, \emph{graded} and \emph{linear}:
\[
\begin{array}{lll}
  \text{(Graded)} & \mGLnt{X}, \mGLnt{Y},\mGLnt{Z} ::= \mathsf{J} \mid \mGLnt{X}  \boxtimes  \mGLnt{Y} \mid \mathsf{Lin} \, \mGLnt{A}\\
  \text{(Linear)} & \mGLnt{A},\mGLnt{B},\mGLnt{C} ::= \mathsf{I} \mid \mGLnt{A}  \otimes  \mGLnt{B} \mid \mGLnt{A}  \multimap  \mGLnt{B} \mid \mathsf{Grd} _{ \mGLnt{r} }\, \mGLnt{X}\\
\end{array}
\]
where $\mathsf{J}$ and $\mathsf{I}$ are unit types, and $\boxtimes$ and
$\otimes$ are tensor (product) operators in their respective
domains.  In the case of the linear domain, the product is the
standard multiplicative conjunction.  The $\mathsf{Lin}$ modality
encapsulates a linear proposition as a graded proposition, and the
$\mathsf{Grd} _{ \mGLnt{r} }$ modality encapsulates a graded proposition (at
grade $r$, whose structure is defined below) as a linear proposition. Thus, the two logics
are interconnected by $\mathsf{Grd} _{ \mGLnt{r} }\, \mGLnt{X}$
and $\mathsf{Lin} \, \mGLnt{A}$. Using these two modalities we will later define graded modalities $\Box _{ \mGLnt{r} } \mGLnt{A}$ as $\mathsf{Grd} _{ \mGLnt{r} }\, \mGLsym{(}  \mathsf{Lin} \, \mGLnt{A}  \mGLsym{)}$, similarly to how the of-course modality $\oc  \mGLnt{A}$ can be
defined in LNL logic as the composition of two adjoint modalities.

\begin{definition}
  \label{def:resource-algebra}
Grades (ranged over by $r, s$) are drawn from a semiring parameterizing the system
  $\mGLsym{(}  \mathcal{R}  \mGLsym{,}  1  \mGLsym{,}  *  \mGLsym{,}  \mathsf{0}  \mGLsym{,}  +  \mGLsym{,}  \leq  \mGLsym{)}$ with preorder $\mGLsym{(}  \mathcal{R}  \mGLsym{,}  \leq  \mGLsym{)}$ such that both $*$ and $+$ are monotonic wrt $\leq$.
\end{definition}

\noindent
The semiring governs the structural rules: the additive part of the semiring is
involved in weakening and contraction, and the multiplicative part in usage and
composition. Various concrete examples of interesting semirings are given at the
end of Subsection~\ref{subsec:mgl_sequent_calculus}.

Section~\ref{subsec:mgl_nat_deduc} develops an equivalent
natural deduction formulation of $\mGL{}$. We then show that the natural
deduction and the sequent calculus are interderivable without modifying the term
witnessing a derivation. Thus, any semantic model of one is a model of the
other. We opt to focus on the sequent calculus form for now without loss
of generality.

\subsection{Sequent Calculus}
\label{subsec:mgl_sequent_calculus}
\input{mGL-seq-ottput}


%% file: mGL-seq-ottput.tex
\input{providenames.tex}

\renewcommand{\mGLdruleGSTXXUnitRName}{$\text{unit}^{\mathsf{J}}_R$}
\renewcommand{\mGLdruleGSTXXUnitLName}{$\text{unit}^{\mathsf{J}}_L$}
\renewcommand{\mGLdruleGSTXXTenRName}{$\boxtimes_R$}
\renewcommand{\mGLdruleGSTXXTenLName}{$\boxtimes_L$}
\renewcommand{\mGLdruleGSTXXSubName}{\textsc{sub$_{\mathsf{GS}}$}}
\renewcommand{\mGLdruleGSTXXidName}{\textsc{id$_{\mathsf{GS}}$}}
\renewcommand{\mGLdruleGSTXXLinRName}{$\mathsf{Lin}_R$}
\renewcommand{\mGLdruleGSTXXCutName}{\textsc{cut$_{\mathsf{GS}}$}}
\renewcommand{\mGLdruleGSTXXWeakName}{\textsc{weak$_{\mathsf{GS}}$}}
\renewcommand{\mGLdruleGSTXXContName}{\textsc{cont$_{\mathsf{GS}}$}}
\renewcommand{\mGLdruleGSTXXExName}{\textsc{ex$_{\mathsf{GS}}$}}

\renewcommand{\mGLdruleMSTXXUnitRName}{$\text{unit}^{\mathsf{I}}_R$}
\renewcommand{\mGLdruleMSTXXUnitLName}{$\text{unit}^{\mathsf{I}}_L$}
\renewcommand{\mGLdruleMSTXXGUnitLName}{$\text{unit}^{\mathsf{J}-\mathsf{MS}}_L$}
\renewcommand{\mGLdruleMSTXXTenRName}{$\otimes_R$}
\renewcommand{\mGLdruleMSTXXTenLName}{$\otimes_L$}
\renewcommand{\mGLdruleMSTXXGTenLName}{$\boxtimes_{L-\mathsf{MS}}$}
\renewcommand{\mGLdruleMSTXXSubName}{\textsc{sub$_{\mathsf{MS}}$}}
\renewcommand{\mGLdruleMSTXXidName}{\textsc{id$_{\mathsf{MS}}$}}
\renewcommand{\mGLdruleMSTXXLinLName}{$\mathsf{Lin}_L$}
\renewcommand{\mGLdruleMSTXXGrdLName}{$\mathsf{Grd}_L$}
\renewcommand{\mGLdruleMSTXXGrdRName}{$\mathsf{Grd}_R$}
\renewcommand{\mGLdruleMSTXXImpLName}{$\multimap_{L}$}
\renewcommand{\mGLdruleMSTXXImpRName}{$\multimap_{R}$}
\renewcommand{\mGLdruleMSTXXCutName}{\textsc{cut$_{\mathsf{MS}}$}}
\renewcommand{\mGLdruleMSTXXGCutName}{\textsc{gcut$_{\mathsf{MS}}$}}
\renewcommand{\mGLdruleMSTXXWeakName}{\textsc{weak$_{\mathsf{MS}}$}}
\renewcommand{\mGLdruleMSTXXContName}{\textsc{cont$_{\mathsf{MS}}$}}
\renewcommand{\mGLdruleMSTXXExName}{\textsc{ex$_{\mathsf{MS}}$}}
\renewcommand{\mGLdruleMSTXXGExName}{\textsc{gex$_{\mathsf{MS}}$}}
\renewcommand{\mGLdruleMSTXXBoxEName}{$\Box_E$}
\renewcommand{\mGLdruleMSTXXBoxIName}{$\Box_i$}
\renewcommand{\mGLdruleMSTXXGImpLName}{$\multimap_{GL}$}
\renewcommand{\mGLdruleMSTXXGImpRName}{$\multimap_{GR}$}
\renewcommand{\mGLdruleGSTXXMCutName}{\textsc{mcut}}
\renewcommand{\mGLdruleMSTXXMGCutName}{\textsc{gmcut}}

\renewcommand{\mGLdruleGTXXUnitIName}{$\text{unit}^{\mathsf{J}}_{I}$}
\renewcommand{\mGLdruleGTXXUnitEName}{$\text{unit}^{\mathsf{J}}_{E}$}
\renewcommand{\mGLdruleGTXXTenIName}{$\boxtimes_{I}$}
\renewcommand{\mGLdruleGTXXSubName}{\textsc{sub}}
\renewcommand{\mGLdruleGTXXIdName}{\textsc{id}}
\renewcommand{\mGLdruleGTXXLinIName}{$\mathsf{Lin}_{I}$}
\renewcommand{\mGLdruleGTXXTenEName}{$\boxtimes_{E}$}
\renewcommand{\mGLdruleGTXXWeakName}{\textsc{weak}}
\renewcommand{\mGLdruleGTXXContName}{\textsc{cont}}
\renewcommand{\mGLdruleGTXXExName}{\textsc{ex}}

\renewcommand{\mGLdruleMTXXIdName}{\textsc{id}}
\renewcommand{\mGLdruleMTXXUnitIName}{$\text{unit}^{\mathsf{I}}_{I}$}
\renewcommand{\mGLdruleMTXXUnitEName}{$\text{unit}^{\mathsf{I}}_{E}$}
\renewcommand{\mGLdruleMTXXGUnitEName}{$\text{unit}_{GE}$}
\renewcommand{\mGLdruleMTXXGTenEName}{$\boxtimes_{GE}$}
\renewcommand{\mGLdruleMTXXTenIName}{$\otimes_{I}$}
\renewcommand{\mGLdruleMTXXTenEName}{$\otimes_{E}$}
\renewcommand{\mGLdruleMTXXImpIName}{$\multimap_{I}$}
\renewcommand{\mGLdruleMTXXImpEName}{$\multimap_{E}$}
\renewcommand{\mGLdruleMTXXGrdEName}{$\mathsf{Grd}_{E}$}
\renewcommand{\mGLdruleMTXXGSubName}{\textsc{sub}}
\renewcommand{\mGLdruleMTXXGrdIName}{$\mathsf{Grd}_{I}$}
\renewcommand{\mGLdruleMTXXLinEName}{$\mathsf{Lin}_{E}$}
\renewcommand{\mGLdruleMTXXWeakName}{\textsc{weak}}
\renewcommand{\mGLdruleMTXXContName}{\textsc{cont}}
\renewcommand{\mGLdruleMTXXGExName}{\textsc{gex}}
\renewcommand{\mGLdruleMTXXExName}{\textsc{ex}}

\setlength{\arraycolsep}{0.1em}

We first define contexts used in the judgments:

\begin{definition}[Graded contexts]
  \label{def:graded-contexts}
  Suppose $\mGLsym{(}  \mathcal{R}  \mGLsym{,}  1  \mGLsym{,}  *  \mGLsym{,}  \mathsf{0}  \mGLsym{,}  +  \mGLsym{,}  \leq  \mGLsym{)}$ is a preordered semiring (Def.~\ref{def:resource-algebra}).  Then \emph{grade
  vectors} $\delta$ are sequences of $\mathcal{R}$, contexts $\Delta$ are sequences of
  graded formulas $\mGLnt{X}$, and contexts $\Gamma$ are sequences of linear
  formulas:
  \begin{align*}
    \delta := \emptyset \mid \delta  \mGLsym{,}  \mGLnt{r}
  \qquad\;\;
    \Delta := \emptyset \mid \Delta  \mGLsym{,}  \mGLmv{x}  \mGLsym{:}  \mGLnt{X}
  \qquad\;\;
    \Gamma := \emptyset \mid \Gamma  \mGLsym{,}  \mGLmv{x}  \mGLsym{:}  \mGLnt{A}
  \end{align*}
  The comma operator is overloaded for sequence concatenation, i.e., we can write
  $\delta_{{\mathrm{1}}}  \mGLsym{,}  \delta_{{\mathrm{2}}}$ and $\Delta_{{\mathrm{1}}}  \mGLsym{,}  \Delta_{{\mathrm{2}}}$, which further requires that
  $\Delta_{{\mathrm{1}}}$ and $\Delta_{{\mathrm{2}}}$ are disjoint contexts.

  A \emph{graded context} $\delta  \odot  \Delta$ is a pairing of a grade vector
  and a context defined as follows:
  \begin{align*}
  \begin{array}{cll}
    \emptyset  \odot  \emptyset = \emptyset
    \qquad\;\;
    (  \delta  \mGLsym{,}  \mGLnt{r}  )   \odot   ( \Delta  \mGLsym{,}  \mGLmv{x}  \mGLsym{:}  \mGLnt{X} ) = (\delta  \odot  \Delta),x:(\mGLnt{r}  \odot  \mGLnt{X})
  \end{array}
  \end{align*}
  where $\mGLnt{r}  \odot  \mGLnt{X}$ pairs a formula with a grade $r$ capturing (by the rules
  of the system) how the formula $X$ (named $x$) is used to form a judgment.

  We lift the operations of semirings to grade vectors,
  forming a semimodule, with the
  pointwise addition and scalar multiplication defined in a standard way:
  \begin{align*}
    \emptyset  +  \emptyset &= \emptyset & \mGLnt{r}  *  \emptyset &= \emptyset\\
    (  \delta_{{\mathrm{1}}}  \mGLsym{,}  \mGLnt{r_{{\mathrm{1}}}}  )   +   (  \delta_{{\mathrm{2}}}  \mGLsym{,}  \mGLnt{r_{{\mathrm{2}}}}  ) &= (  \delta_{{\mathrm{1}}}  +  \delta_{{\mathrm{2}}}  )   \mGLsym{,}   (  \mGLnt{r_{{\mathrm{1}}}}  +  \mGLnt{r_{{\mathrm{2}}}}  ) & \mGLnt{r}  *   (  \delta  \mGLsym{,}  \mGLnt{s}  ) &= (  \mGLnt{r}  *  \delta  )   \mGLsym{,}   (  \mGLnt{r}  *  \mGLnt{s}  )
  \end{align*}
  Addition of grade vectors requires the vectors to be of the same length.
\end{definition}
%

The judgment form for our fully graded logic $\delta  \odot  \Delta  \vdash_{\mathsf{GS} }  \mGLnt{t}  \mGLsym{:}  \mGLnt{X}$
captures a concluding proposition $X$ under the graded context
of assumptions $\delta  \odot  \Delta$. Mixed graded/linear logic judgments
$\delta  \odot  \Delta  \mGLsym{;}  \Gamma  \vdash_{\mathsf{MS} }  \mGLnt{l}  \mGLsym{:}  \mGLnt{A}$ are similar but also have a context $\Gamma$ of
linear assumptions which, being linear, do not have a corresponding grade vector.

The two judgments $\vdash_{\mathsf{GS} }$ and $\vdash_{\mathsf{MS} }$ (also called \emph{sub-logics} or \emph{fragments}) are
defined by mutual induction. We present conceptually related rules
from both systems side-by-side where possible, or one-after-the-other, in the
order $\mathsf{GS}$ then $\mathsf{MS}$.

The identity (axiom) rules are:
\begin{gather*}
  \begin{array}{lll}
  \mGLdruleGSTXXid{} \quad & \qquad & \quad \mGLdruleMSTXXid{}
\end{array}
\end{gather*}
The multiplicative identity $1$ is the `default' grade
for formulas in the graded logic $\GS{}$ (left), in the sense that we
can think of the right-hand side of the judgment as also
implicitly having grade $1$. The graded identity rule says that
a graded formula that is used must have the default grade.  For
example, in the natural number semiring $\mGLsym{(}  \mathbb{N}  \mGLsym{,}  1  \mGLsym{,}  *  \mGLsym{,}  0  \mGLsym{,}  +  \mGLsym{,}  \mGLsym{=}  \mGLsym{)}$ the
multiplicative identity $1 \in \mathbb{N}$ captures linear usage.
The mixed identity rule types linear assumption use, requiring just a singleton
linear context (forcing a lack of weakening). It also
requires that there are no graded
formulas in context---the graded context is empty $\emptyset$.

The `cut' rules are:
\begin{gather*}
\begin{array}{ll}
  \mGLdruleGSTXXCut{} &  \\[0.75em]
  \mGLdruleMSTXXCut{} & \mGLdruleMSTXXGCut{}
\end{array}
\end{gather*}
The \mGLdruleGSTXXCutName{} rule provides a cut through a graded proposition $X$
of grade $r$ in the receiving context (second premise). Thus, the resulting
term uses semiring multiplication (lifted to contexts, Def.~\ref{def:graded-contexts}) to capture sequential usage, scaling the grade vector
$\delta_{{\mathrm{2}}}$ of the cut term $\mGLnt{t_{{\mathrm{1}}}}$ by $\mGLnt{r}$.
The \mGLdruleMSTXXCutName{} rules provides a cut through a linear proposition $A$
and has no effect on the graded contexts. However, $\mathsf{MS}$ has
a further cut rule \mGLdruleMSTXXGCutName{} for graded propositions
in its graded context, incurring a scaling similarly to \mGLdruleGSTXXCutName{}.
This pattern occurs throughout: operations applied to the graded context
in $\mathsf{GS}$ have a sister rule in $\mathsf{MS}$ applying the same operation
in the $\mathsf{MS}$ graded context.

Both sub-logics have free use of exchange:
\begin{gather*}
\begin{array}{l}
  \mGLdruleGSTXXEx{} \\[0.75em]
\mGLdruleMSTXXEx{} \quad \mGLdruleMSTXXGEx{}
\end{array}
\end{gather*}
Exchanging graded propositions simultaneously exchanges their
 grades in the grade vector.

We can use weakening and contraction in the graded system and the mixed system
within the graded contexts, with the semiring's $0$ representing weakened
hypotheses and the grades of contracted hypotheses combined
via semiring addition $+$:
\begin{gather*}
  \hspace{-0.5em}
\begin{array}{ll}
  \mGLdruleGSTXXWeak{} & \mGLdruleGSTXXCont{} \\[0.75em]
  \mGLdruleMSTXXWeak{} & \mGLdruleMSTXXCont{}
\end{array}
\end{gather*}
The left and right rules for units for the graded
and mixed logics are akin to linear logic:
\begin{gather*}
  \begin{array}{ll}
  \mGLdruleGSTXXUnitL{} &
   \mGLdruleGSTXXUnitR{} \\[0.75em]
  \mGLdruleMSTXXUnitL{} & \mGLdruleMSTXXUnitR{} \\[0.75em]
  \mGLdruleMSTXXGUnitL{}
  \end{array}
\end{gather*}
Thus, in $\mathsf{GS}$, we can eliminate a graded unit $\mathsf{j}$
at an arbitrary grade $r$, whereas the linear unit $\mathsf{i}$ in
$\mathsf{MS}$ gets eliminated from the linear context.  The
additional left rule (\mGLdruleMSTXXGUnitLName{}) for
$\mathsf{MS}$ again similarly eliminates graded units $\mathsf{J}$ in the graded context.

Tensor products are then eliminated in each fragment as follows:
\begin{gather*}
\begin{array}{l}
    \mGLdruleGSTXXTenL{} \;\,
    \mGLdruleGSTXXTenR{} \\[0.75em]
  \mGLdruleMSTXXTenL{} \;\,
  \mGLdruleMSTXXTenR{} \\[0.75em]
  \mGLdruleMSTXXGTenL{}
\end{array}
\end{gather*}
The left rule for $\boxtimes$ eliminates from the graded context
at any grade $r$, where the components of the tensor product both inherit
this grade in the premise. Reading instead top-down,
 the graded tensor product requires that both components are graded with
the same grade; this is similar to linear products, where both components are
linear.

Note that Benton has two left rules for (non-linear) tensor products, in the `projection' style. We instead must use the pattern matching style for the
soundness of grading so that each component is bound to a variable with the same grade.

Only the mixed linear-graded system has implication, and only on linear
propositions, thus we have $\multimap$ in $\MS$ with left and right rules:
\begin{gather*}
  \begin{array}{c}
     \mGLdruleMSTXXImpL{}  \qquad \mGLdruleMSTXXImpR{}
  \end{array}
  \end{gather*}
In Lemma~\ref{lemma:graded_implication_in_mgl}, we recover a graded
implication through the modal operators of the system in the same way
that Melli\`es did for (ungraded) LNL logic~\cite{Mellies:2009}.

We now consider the modal operators $\mathsf{Lin}$ and $\mathsf{Grd} _{ \mGLnt{r} }$
which connect the two sub-logics.

The right rule for the $\mathsf{Lin}$ modality transports a linear formula from the
linear system $\MS$ into the graded system $\GS$ where it can be
reasoned with non-linearly as accounted for by grading. The
corresponding left rule is akin to
\emph{dereliction} from linear logic, enabling a linear assumption $x : A$ to
treated as a (renamed) graded assumption $z : \mathsf{Lin} \, \mGLnt{A}$ at grade $1$:
\begin{gather*}
\begin{array}{c}
  \mGLdruleMSTXXLinL{} \qquad \mGLdruleGSTXXLinR{}
\end{array}
\end{gather*}
The other modal operator $\mathsf{Grd}$, or rather the family
of modal operators $\mathsf{Grd}_r$, transports a graded
formula with its grade into the linear system where it can be reasoned
with linearly:
\begin{gather*}
  \begin{array}{c}
  \mGLdruleMSTXXGrdL{} \qquad \mGLdruleMSTXXGrdR{}
\end{array}
\end{gather*}
The right rule is akin to \emph{promotion} for $\mathsf{Grd} _{ \mGLnt{r} }$ where
we subsequently scale the graded context by the grade $r$. The left
rule `unboxes' a graded modality $\mathsf{Grd} _{ \mGLnt{r} }\, \mGLnt{X}$ providing access to the
$\mGLnt{X}$ formula `inside', graded at $\mGLnt{r}$.

Perhaps the most remarkable property of these modal operators is that
they decompose semiring-graded necessity modalities into $\Box _{ \mGLnt{r} } \mGLnt{A} = \mathsf{Grd} _{ \mGLnt{r} }\, \mGLsym{(}  \mathsf{Lin} \, \mGLnt{A}  \mGLsym{)}$~\cite{Katsumata:2018} within the mixed system.  In fact, their
introduction and elimination rules are derivable:
\begin{lemma}[$\mGL{}$ Graded Necessity Modality]
  \label{lemma:graded_modalities_in_mGL}
  The following are derivable:
  \begin{gather*}
    \hspace{-1.75em}
  \begin{array}{c}
    \mGLdruleMSTXXBoxE{} \,
    \mGLdruleMSTXXBoxI{}
  \end{array}
  \end{gather*}
\end{lemma}
\begin{proof}
  The elimination rule follows by applying \mGLdruleMSTXXGrdLName{}
  to the second premise and then applying cut with the first
  premise.  The introduction rule follows by
  \mGLdruleGSTXXLinRName{} then \mGLdruleMSTXXGrdRName{}.
\end{proof}
The previous lemma reveals a lot about the structure of existing graded type
systems. First, graded hypotheses, usually denoted with their grade $\mGLmv{x} :_{\mGLnt{r}} \mGLnt{A}$ in
the literature (e.g.~\cite{DBLP:journals/pacmpl/AbelB20,Orchard:2019}), are
graded hypotheses $\mGLnt{r}  \odot  \mathsf{Lin} \, \mGLnt{A}$ where the linear formula has been transported
into the graded system.  Second, the restriction to only graded variables in the
promotion rule is very explicit in the introduction rule
\mGLdruleMSTXXBoxIName{} above.  Third, the elimination rule (here
\mGLdruleMSTXXBoxEName) is really the left rule for $\mathsf{Grd}$ followed by a cut.
Thus, graded type systems, whilst typically of a natural deduction form,
incorporate a little of the flavour of sequent calculi in the rules for graded
modalities because of the integrated cut.

Using the modal operators we can derive a graded implication of the form $\mathsf{Grd} _{ \mGLnt{r} }\, \mGLnt{X}   \multimap  \mGLnt{A}$:
\begin{lemma}[$\mGL{}$ Graded Implication]
  \label{lemma:graded_implication_in_mgl}
  The following rules are derivable:
  \begin{gather*}
    \begin{array}{c}
      \mGLdruleMSTXXGImpL{}  \\[0.5em] \mGLdruleMSTXXGImpR{}
    \end{array}
    \end{gather*}
\end{lemma}
\begin{proof}
  The left rule follows by applying \mGLdruleMSTXXGrdRName{} to
  the first premise, and then \mGLdruleMSTXXImpLName{} using
  the second premise.  The right rule follows by applying
  \mGLdruleMSTXXGrdLName{} to the premise and then
  \mGLdruleMSTXXImpRName{}.
\end{proof}
\noindent
The final rules are the approximation rules
for grades:
\begin{gather*}
\begin{array}{l}
  \mGLdruleGSTXXSub{} \quad \mGLdruleMSTXXSub{}
\end{array}
\end{gather*}
Approximation allows for the abstraction of grades along an ordering. For
example, in the semiring $\mGLsym{(}  \mathbb{N}  \mGLsym{,}  1  \mGLsym{,}  *  \mGLsym{,}  0  \mGLsym{,}  +  \mGLsym{,}  \leq  \mGLsym{)}$ where the
order is the usual ordering on natural numbers, then a grade $\mGLnt{r}$ stands for ``at-most $\mGLnt{r}$'',
generalising the notion of `affine' usage tracking.  Disabling the ordering
by forcing it to be true only in the case of reflexive pairs (i.e., an equality
relation) results in exact usage tracking.  This demonstrates that having a notion of
approximation results in a more general usage tracking framework.
\begin{example}[Derivation]
  As an example derivation in the natural numbers semiring $\mGLsym{(}  \mathbb{N}  \mGLsym{,}  1  \mGLsym{,}  *  \mGLsym{,}  0  \mGLsym{,}  +  \mGLsym{,}  \leq  \mGLsym{)}$,
  for `affine' usage
  tracking, the following copies an assumption to make a pair in the graded fragment,
  then uses an approximation, and then transports the pair into the linear fragment,
  scaling its grades further:
  {\small{
      \begin{gather*}
        \vspace{-1em}
      \inferrule* [right=$\mGLdruleGSTXXTenRName{}$]
       {\inferrule* [flushleft,right=$\mGLdruleGSTXXidName{}$]
        {\,}{1  \odot  \mGLmv{x}  \mGLsym{:}  \mGLnt{X}  \vdash_{\mathsf{GS} }  \mGLmv{x}  \mGLsym{:}  \mGLnt{X}} \quad
        \inferrule* [right=$\mGLdruleGSTXXidName{}$]{\;}{1  \odot  \mGLmv{y}  \mGLsym{:}  \mGLnt{X}  \vdash_{\mathsf{GS} }  \mGLmv{y}  \mGLsym{:}  \mGLnt{X}}
       \hspace{-5em}}
    {
    \inferrule* [right=$\mGLdruleGSTXXContName{}$]
     {1  \mGLsym{,}  1  \odot  \mGLmv{x}  \mGLsym{:}  \mGLnt{X}  \mGLsym{,}  \mGLmv{y}  \mGLsym{:}  \mGLnt{X}  \vdash_{\mathsf{GS} }  \mGLsym{(}  \mGLmv{x}  \mGLsym{,}  \mGLmv{y}  \mGLsym{)}  \mGLsym{:}  \mGLnt{X}  \boxtimes  \mGLnt{X}}
    {
   \inferrule* [right=$\mGLdruleGSTXXSubName{}$]
       {2  \odot  \mGLmv{x}  \mGLsym{:}  \mGLnt{X}  \vdash_{\mathsf{GS} }  \mGLsym{[}  \mGLmv{x}  \mGLsym{/}  \mGLmv{y}  \mGLsym{]}  \mGLsym{(}  \mGLmv{x}  \mGLsym{,}  \mGLmv{y}  \mGLsym{)}  \mGLsym{:}  \mGLnt{X}  \boxtimes  \mGLnt{X} \quad 2 \leq 3}
       {\inferrule* [right=$\mGLdruleMSTXXGrdRName{}$]
         {\mGLsym{3}  \odot  \mGLmv{x}  \mGLsym{:}  \mGLnt{X}  \vdash_{\mathsf{GS} }  \mGLsym{(}  \mGLmv{x}  \mGLsym{,}  \mGLmv{x}  \mGLsym{)}  \mGLsym{:}  \mGLnt{X}  \boxtimes  \mGLnt{X}}
         {2  *  \mGLsym{3}  \odot  \mGLmv{x}  \mGLsym{:}  \mGLnt{X}  \mGLsym{;}  \emptyset  \vdash_{\mathsf{MS} }  \mathsf{Grd} \, 2 \, \mGLsym{(}  \mGLmv{x}  \mGLsym{,}  \mGLmv{x}  \mGLsym{)}  \mGLsym{:}   \mathsf{Grd} _{ 2 }\, \mGLsym{(}  \mGLnt{X}  \boxtimes  \mGLnt{X}  \mGLsym{)}}
       }
    }
    }
  \end{gather*}}}
\end{example}

\begin{example}[None-One-Tons~\cite{DBLP:conf/birthday/McBride16}] The
  semiring over $\{0, 1, \omega \}$ with
  $0 \leq \omega$ and $1 \leq \omega$, where $+$ and $*$ are saturating at $\omega$,
  can be used to distinguish between
  linear and various non-linear uses: assigning $\mGLnt{r}  \mGLsym{=}  1$ to linear usage,
  $\mGLnt{r}  \mGLsym{=}  0$ to non-usage (when a resource is discarded), and $\mGLnt{r} = \omega $
  to arbitrary usage.
\end{example}

Note, however, that even with the above semiring we are unable to exactly
represent the exponential modality $\oc$ from linear logic via some particular grade $r$
within the graded logic. This is because no matter which grade we
choose, we are able to ``push'' the grade into the tensor product using this
graded tensor elimination (\mGLdruleGSTXXTenLName{}), allowing
derivation of $\mathsf{Grd} _{ \mGLnt{r} }\, \mGLsym{(}  \mGLnt{X}  \boxtimes  \mGLnt{Y}  \mGLsym{)}   \multimap  \mGLsym{(}   \mathsf{Grd} _{ \mGLnt{r} }\, \mGLnt{X}   \otimes   \mathsf{Grd} _{ \mGLnt{r} }\, \mGLnt{Y}   \mGLsym{)}$, and yet in linear logic it
is not possible to derive $\oc (A \otimes B) \multimap \oc A \otimes \oc B$.
Therefore, our logic cannot reduce to Benton's LNL logic simply by taking the
Cartesian (trivial) semiring, as one might expect at first glance. This quality
is typical of \emph{graded base} systems, so reconciling these with linear logic
requires some additional structure on the semiring~\cite{hughes:lirmm-03271465}
(though this is not the focus here).

On the other hand, notice that we have another way to represent graded products: as
linear products wrapped in the derived graded modality, or $\Box _{ \mGLnt{r} } \mGLsym{(}  \mGLnt{A}  \otimes  \mGLnt{B}  \mGLsym{)}$.
Importantly, here it is not possible to `push' the grade `through' the tensor as
we can for the graded product; we cannot derive $\mGLsym{(}   \Box _{ \mGLnt{r} } \mGLnt{A}   \mGLsym{)}  \otimes  \mGLsym{(}   \Box _{ \mGLnt{r} } \mGLnt{B}   \mGLsym{)}$. This representation of graded products thus has behaviour more typical of a
\emph{linear base} graded type system, with our combined logic again giving us a
clearer understanding of the relationship between these contrasting styles.

\begin{example}[Security levels]
Information-Flow Control properties can be tracked by instantiating
the semiring with a lattice of security levels~\cite{Gaboardi:2016}, e.g.,
with $(\{\mathsf{Lo} \leq \mathsf{Hi}\}, \mathsf{Lo}, \wedge,
\mathsf{Hi}, \vee)$ where $\mathsf{Hi}$-graded inputs are treated as
irrelevant: we cannot depend on any high-security inputs
when building a low-security graded output $\mathsf{Grd}_{\mathsf{Lo}} A$.
\end{example}

\begin{example}[Sensitivity]
The real number semiring $(\mathbb{R}, 1, *, 0, +, \leq)$ can be
leveraged to capture a notion of numerical sensitivity in
programs/logic~\cite{DBLP:conf/popl/GaboardiHHNP13,DBLP:conf/icfp/DAntoniGAHP13}, where
a program is $k$-sensitive (for $k \in \mathbb{R}$) in a variable if a change
$\epsilon$ in its inputs to $x$ produces at most a change of $k\epsilon$
in the output of the program. This instantiation of the system tracks
sensitivities as grades where additional dependent-type-based
mechanisms are needed to lift program values into the types, e.g., $\mathsf{scale}
: (k : \mathbb{R}) \rightarrow \mathsf{Grd}_k \mathbb{R} \rightarrow \mathbb{R}$.
\end{example}

\subsection{Metatheory}
\label{subsubsec:metatheory_of_mgl}

\noindent
$\mGL{}$ enjoys a rich metatheory.  First, it satisfies cut elimination,
for which we give the full proof.  The proof of cut reduction requires
a generalization of the graded cut rules to graded \emph{multicut} rules in
order to accommodate the structural rule of graded contraction.\footnote{Whilst
cut reduction can be proved for intuitionistic sequent
calculus without multicut~\cite{von2001proof}, we use the standard
multicut approach as it relates well to the categorical models developed later.}

Thus, throughout the cut elimination proof we use the following graded
multicut rules:
\begin{gather*}
\begin{array}{lll}
  \mGLdruleGSTXXMCut{} \\[0.5em] \mGLdruleMSTXXMGCut{}
\end{array}
\end{gather*}
Both rules compute the contraction of the $n$ hypotheses
involved in the multicut on the cut-formula $\mGLnt{X}$.  To do this we use \emph{row-vector
matrix multiplication}.  We denote the matrix consisting of
$\mGLmv{n}$-copies of the row vector $\delta_{{\mathrm{2}}}$ by $[  \delta_{{\mathrm{2}}} ^{  \mGLmv{n}  } ]$. Then row-vector multiplication is:
\[
\delta  \boxast [ { \delta_{{\mathrm{2}}} }^{ \mGLmv{n} } ] = \bigoplus ^{ \mGLmv{n} }_{ \mGLmv{k}  = 1}   (  \delta  \mGLsym{(}  \mGLmv{k}  \mGLsym{)}  *  \delta_{{\mathrm{2}}}  )
\]
where the $*$ on the right is the scalar multiplication derived
from the semiring, $\bigoplus$ is the pointwise addition of
vectors, and where $\delta  \mGLsym{(}  \mGLmv{k}  \mGLsym{)}$ is the $k$-th element of the vector
$\delta$. This computes the usages of the hypotheses in
$\Delta_{{\mathrm{2}}}$ as the multiplication of a matrix of size
$1 \times n$ with a
matrix of size $n \times |\Delta_{{\mathrm{2}}}|$ to yield a matrix of size $1 \times |\Delta_{{\mathrm{2}}}|$.

\sloppy We now proceed with the proof of cut elimination.  The \emph{rank}
$\mathsf{Rank}  (  \mGLnt{X}  )$ and $\mathsf{Rank}  (  \mGLnt{A}  )$ of a formula is the height of the input
formula's syntax tree where constants are of rank $0$.  The
\emph{cut rank} $\mathsf{CutRank} \, \mGLsym{(}  \Pi  \mGLsym{)}$ of a derivation $\Pi$ of some
judgment is defined to be one more than the maximum rank of the cut
formula's in $\Pi$, and $0$ if $\Pi$ is cut free.  The
\emph{depth} $\mathsf{Depth}  (  \Pi  )$ of a derivation $\Pi$ is the length
of the longest path in the proof tree of $\Pi$, and hence, the
depth of an axiom is $0$.  We prove cut-elimination without term
annotations on the rules, in keeping with traditional proofs.

\begin{restatable}[Cut Reduction for $\mGL{}$]{lemma}{cutElimMGL}
  \label{lemma:cut_reduction_for_mgl}
  \begin{enumerate}
  \item[]
  \item (Graded) If
    $\Pi_{{\mathrm{1}}}$ is a proof of $\delta_{{\mathrm{2}}}  \odot  \Delta_{{\mathrm{2}}}  \vdash_{\mathsf{GS} }  \mGLnt{X}$ and
    $\Pi_{{\mathrm{2}}}$ is a proof of $(  \delta_{{\mathrm{1}}}  \mGLsym{,}  \delta  \mGLsym{,}  \delta_{{\mathrm{3}}}  )   \odot   ( \Delta_{{\mathrm{1}}}  \mGLsym{,}   \mGLnt{X} ^{ \mGLmv{n} }   \mGLsym{,}  \Delta_{{\mathrm{3}}} )   \vdash_{\mathsf{GS} }  \mGLnt{Y}$ with
    $\mathsf{CutRank} \, \mGLsym{(}  \Pi_{{\mathrm{1}}}  \mGLsym{)}  \mGLsym{,}  \mathsf{CutRank} \, \mGLsym{(}  \Pi_{{\mathrm{2}}}  \mGLsym{)} \, \leq \,  \mathsf{Rank}  (  \mGLnt{X}  )$, then
    there exists a proof $\Pi$ of
    $(   \delta_{{\mathrm{1}}}  \mGLsym{,}  \delta  \boxast [ { \delta_{{\mathrm{2}}} }^{ \mGLmv{n} } ]   \mGLsym{,}  \delta_{{\mathrm{3}}}  )   \odot   ( \Delta_{{\mathrm{1}}}  \mGLsym{,}  \Delta_{{\mathrm{2}}}  \mGLsym{,}  \Delta_{{\mathrm{3}}} )   \vdash_{\mathsf{GS} }  \mGLnt{Y}$ with $\mathsf{CutRank} \, \mGLsym{(}  \Pi  \mGLsym{)} \, \leq \,  \mathsf{Rank}  (  \mGLnt{X}  )$.

  \item (Graded/Mixed) If
    $\Pi_{{\mathrm{1}}}$ is a proof of $\delta_{{\mathrm{2}}}  \odot  \Delta_{{\mathrm{2}}}  \vdash_{\mathsf{GS} }  \mGLnt{X}$ and
    $\Pi_{{\mathrm{2}}}$ is a proof of $(  \delta_{{\mathrm{1}}}  \mGLsym{,}  \delta  \mGLsym{,}  \delta_{{\mathrm{3}}}  )   \odot   ( \Delta_{{\mathrm{1}}}  \mGLsym{,}   \mGLnt{X} ^{ \mGLmv{n} }   \mGLsym{,}  \Delta_{{\mathrm{3}}} )   \mGLsym{;}  \Gamma  \vdash_{\mathsf{MS} }  \mGLnt{B}$ with
    $\mathsf{CutRank} \, \mGLsym{(}  \Pi_{{\mathrm{1}}}  \mGLsym{)}  \mGLsym{,}  \mathsf{CutRank} \, \mGLsym{(}  \Pi_{{\mathrm{2}}}  \mGLsym{)} \, \leq \,  \mathsf{Rank}  (  \mGLnt{X}  )$, then
    there exists a proof $\Pi$ of
    $(   \delta_{{\mathrm{1}}}  \mGLsym{,}  \delta  \boxast [ { \delta_{{\mathrm{2}}} }^{ \mGLmv{n} } ]   \mGLsym{,}  \delta_{{\mathrm{3}}}  )   \odot   ( \Delta_{{\mathrm{1}}}  \mGLsym{,}  \Delta_{{\mathrm{2}}}  \mGLsym{,}  \Delta_{{\mathrm{3}}} )   \mGLsym{;}  \Gamma  \vdash_{\mathsf{MS} }  \mGLnt{B}$ with $\mathsf{CutRank} \, \mGLsym{(}  \Pi  \mGLsym{)} \, \leq \,  \mathsf{Rank}  (  \mGLnt{X}  )$.

  \item (Mixed) If
    $\Pi_{{\mathrm{1}}}$ is a proof of $\delta_{{\mathrm{2}}}  \odot  \Delta_{{\mathrm{2}}}  \mGLsym{;}  \Gamma_{{\mathrm{2}}}  \vdash_{\mathsf{MS} }  \mGLnt{A}$ and
    $\Pi_{{\mathrm{2}}}$ is a proof of $\delta_{{\mathrm{1}}}  \odot  \Delta_{{\mathrm{1}}}  \mGLsym{;}   ( \Gamma_{{\mathrm{1}}}  \mGLsym{,}  \mGLnt{A}  \mGLsym{,}  \Gamma_{{\mathrm{3}}} )   \vdash_{\mathsf{MS} }  \mGLnt{B}$ with
    $\mathsf{CutRank} \, \mGLsym{(}  \Pi_{{\mathrm{1}}}  \mGLsym{)}  \mGLsym{,}  \mathsf{CutRank} \, \mGLsym{(}  \Pi_{{\mathrm{2}}}  \mGLsym{)} \, \leq \,  \mathsf{Rank}  (  \mGLnt{A}  )$, then
    there exists a proof $\Pi$ of
    $(  \delta_{{\mathrm{1}}}  \mGLsym{,}  \delta_{{\mathrm{2}}}  )   \odot   ( \Delta_{{\mathrm{1}}}  \mGLsym{,}  \Delta_{{\mathrm{2}}} )   \mGLsym{;}   ( \Gamma_{{\mathrm{1}}}  \mGLsym{,}  \Gamma_{{\mathrm{2}}}  \mGLsym{,}  \Gamma_{{\mathrm{3}}} )   \vdash_{\mathsf{MS} }  \mGLnt{B}$ with $\mathsf{CutRank} \, \mGLsym{(}  \Pi  \mGLsym{)} \, \leq \,  \mathsf{Rank}  (  \mGLnt{A}  )$.
  \end{enumerate}
\end{restatable}
\begin{proof}
  By mutual induction on $\mathsf{Depth}  (  \Pi_{{\mathrm{1}}}  )   +   \mathsf{Depth}  (  \Pi_{{\mathrm{2}}}  )$
  (see \citeappendix{subsec:cut_reduction_of_mgl}{C.1} for proof).
\end{proof}

\begin{restatable}[Decreasing Order of $\mGL{}$]{lemma}{decreasingMGL}
  \label{lemma:decreasing_order_of_mgl}
  If $\Pi$ is a proof of $\delta  \odot  \Delta  \vdash_{\mathsf{GS} }  \mGLnt{X}$ or $\delta  \odot  \Delta  \mGLsym{;}  \Gamma  \vdash_{\mathsf{MS} }  \mGLnt{A}$
  with $\mathsf{CutRank} \, \mGLsym{(}  \Pi  \mGLsym{)} \, \mGLsym{>} \, 0$, then there is a proof $\Pi'$ of
  $\delta'  \odot  \Delta  \vdash_{\mathsf{GS} }  \mGLnt{X}$ or $\delta'  \odot  \Delta  \mGLsym{;}  \Gamma  \vdash_{\mathsf{MS} }  \mGLnt{A}$ with
  $\delta  \leq  \delta'$ and $\mathsf{CutRank} \, \mGLsym{(}  \Pi'  \mGLsym{)} \, \mGLsym{<} \, \mathsf{CutRank} \, \mGLsym{(}  \Pi  \mGLsym{)}$.
\end{restatable}
\begin{proof}
  By induction on $\mathsf{Depth}  (  \Pi  )$ (see
  \citeappendix{subsec:decreasing_order_of_mgl}{C.2} for proof).
\end{proof}

\begin{restatable}[Cut Elimination of $\mGL{}$]{theorem}{cutelim}
  \label{theorem:cut_elimination_of_mgl}
  If $\Pi$ is a proof of $\delta  \odot  \Delta  \vdash_{\mathsf{GS} }  \mGLnt{X}$ or $\delta  \odot  \Delta  \mGLsym{;}  \Gamma  \vdash_{\mathsf{MS} }  \mGLnt{A}$
  with $\mathsf{CutRank} \, \mGLsym{(}  \Pi  \mGLsym{)} \, \mGLsym{>} \, 0$, then there is an algorithm which yields
  a cut-free proof $\Pi'$ of $\delta  \odot  \Delta  \vdash_{\mathsf{GS} }  \mGLnt{X}$ or
  $\delta  \odot  \Delta  \mGLsym{;}  \Gamma  \vdash_{\mathsf{MS} }  \mGLnt{A}$ respectively.
\end{restatable}
\begin{proof}
  Follows immediately by induction on
  $\mathsf{CutRank} \, \mGLsym{(}  \Pi  \mGLsym{)}$ and the previous lemma.
\end{proof}

\begin{restatable}[Subformula property]{lemma}{subformula}
\label{lemma:subformula}
\begin{enumerate}
  \item[]
  \item
    (Graded) Every formula occurring in a cut-free proof $\Pi$ of a judgment, $\delta  \odot  \Delta  \vdash_{\mathsf{GS} }  \mGLnt{X}$,
    consists of subformulas of the formulas occurring in $\delta  \odot  \Delta  \vdash_{\mathsf{GS} }  \mGLnt{X}$.
  \item
    (Mixed) Every formula occurring in a cut-free proof $\Pi$ of a judgment, $\delta  \odot  \Delta  \mGLsym{;}  \Gamma  \vdash_{\mathsf{MS} }  \mGLnt{A}$,
    consists of subformulas of the formulas occurring in $\delta  \odot  \Delta  \mGLsym{;}  \Gamma  \vdash_{\mathsf{MS} }  \mGLnt{A}$.
\end{enumerate}
\end{restatable}
\begin{proof}
  By induction on $\Pi$ (See
  \citeappendix{subsec:proof_of_subformaul_property}{C.3}
  for proof).
\end{proof}

Lastly, we define an equational theory for $\mGL{}$:

\begin{definition}[Equational theory $\equiv$, subset]
  An equational theory on derivations accounts for
  equalities between proofs of the same sequent
  arising from the graded structure (where the terms
  are the same but the structure of the proof tree
  differs), as well as cut elimination,
  i.e., in $\mathsf{GS}$, if cut elimination on
  derivation $\Pi_{{\mathrm{1}}}$ of $\delta  \odot  \Delta  \vdash_{\mathsf{GS} }  \mGLnt{t}  \mGLsym{:}  \mGLnt{X}$
  yields the cut-free derivation of $\Pi_{{\mathrm{2}}}$ for $\delta  \odot  \Delta  \vdash_{\mathsf{GS} }  \mGLnt{t'}  \mGLsym{:}  \mGLnt{X}$
  then the equational theory has $\Pi_{{\mathrm{1}}} \equiv \Pi_{{\mathrm{2}}}$, and similarly for $\MS{}$.

  As a
  sample of two equations from the $\mathsf{GS}$ fragment, the following shows an
  equation leveraging the commutativity of contraction, and another on
  the interaction between weakening and contraction leveraging the left-unit
  of semiring addition:
  \input{mGL-seq-term-eq-theory-fragment-ottput}
\end{definition}
\citeappendix{sec:full-eq-theory}{A.1} gives the full definition of the equational theory.

\begin{remark}[``$\beta\eta$-equalities'' and ``Triangle identities'' via cut reduction]
  One might wonder where $\beta$-equalities are in the above
  equational theory, e.g., that $(\lambda x . l) l'$ in
  $\mathsf{MS}$ is equal to the cut $[l'/x] l$. Such $\beta$-equalities are provided by the cut elimination procedure, which
  reduces away interacting pairs of right and left formulas (the
  principal vs. principal cases).

  Similarly, $\eta$-equalities are equivalent to the identity
  expansion part of cut elimination procedure (where the cut of
  an identity axiom is transformed into an interacting
  left and right pair, with identity axioms expanded towards the leaves).

The internal derivations for the graded equivalent of the `triangle
identities' (that one usually has associated with an adjunction)
are also handled in the cut elimination procedure. The main feature
of the derivations for both identities is that after one step the left and right rules for the
modal operators match up. This leads to consecutive principal vs. principal
cases where rules for the interacting left and right pairs in the two subproofs are
removed by the reduction step.
\end{remark}


%% file: providenames.tex
\providecommand{\mGLdruleGSTXXUnitRName}{}
\providecommand{\mGLdruleGSTXXUnitLName}{}
\providecommand{\mGLdruleGSTXXTenRName}{}
\providecommand{\mGLdruleGSTXXTenLName}{}
\providecommand{\mGLdruleGSTXXSubName}{}
\providecommand{\mGLdruleGSTXXidName}{}
\providecommand{\mGLdruleGSTXXLinRName}{}
\providecommand{\mGLdruleGSTXXCutName}{}
\providecommand{\mGLdruleGSTXXWeakName}{}
\providecommand{\mGLdruleGSTXXContName}{}
\providecommand{\mGLdruleGSTXXExName}{}

\providecommand{\mGLdruleMSTXXUnitRName}{}
\providecommand{\mGLdruleMSTXXUnitLName}{}
\providecommand{\mGLdruleMSTXXGUnitLName}{}
\providecommand{\mGLdruleMSTXXTenRName}{}
\providecommand{\mGLdruleMSTXXTenLName}{}
\providecommand{\mGLdruleMSTXXGTenLName}{}
\providecommand{\mGLdruleMSTXXSubName}{}
\providecommand{\mGLdruleMSTXXidName}{}
\providecommand{\mGLdruleMSTXXLinLName}{}
\providecommand{\mGLdruleMSTXXGrdLName}{}
\providecommand{\mGLdruleMSTXXGrdRName}{}
\providecommand{\mGLdruleMSTXXImpLName}{}
\providecommand{\mGLdruleMSTXXImpRName}{}
\providecommand{\mGLdruleMSTXXCutName}{}
\providecommand{\mGLdruleMSTXXGCutName}{}
\providecommand{\mGLdruleMSTXXWeakName}{}
\providecommand{\mGLdruleMSTXXContName}{}
\providecommand{\mGLdruleMSTXXExName}{}
\providecommand{\mGLdruleMSTXXGExName}{}
\providecommand{\mGLdruleMSTXXBoxEName}{}
\providecommand{\mGLdruleMSTXXBoxIName}{}
\providecommand{\mGLdruleMSTXXGImpLName}{}
\providecommand{\mGLdruleMSTXXGImpRName}{}
\providecommand{\mGLdruleGSTXXMCutName}{}
\providecommand{\mGLdruleMSTXXMGCutName}{}

\providecommand{\mGLdruleGTXXUnitIName}{}
\providecommand{\mGLdruleGTXXUnitEName}{}
\providecommand{\mGLdruleGTXXTenIName}{}
\providecommand{\mGLdruleGTXXSubName}{}
\providecommand{\mGLdruleGTXXIdName}{}
\providecommand{\mGLdruleGTXXLinIName}{}
\providecommand{\mGLdruleGTXXTenEName}{}
\providecommand{\mGLdruleGTXXWeakName}{}
\providecommand{\mGLdruleGTXXContName}{}
\providecommand{\mGLdruleGTXXExName}{}

\providecommand{\mGLdruleMTXXIdName}{}
\providecommand{\mGLdruleMTXXUnitIName}{}
\providecommand{\mGLdruleMTXXUnitEName}{}
\providecommand{\mGLdruleMTXXGUnitEName}{}
\providecommand{\mGLdruleMTXXGTenEName}{}
\providecommand{\mGLdruleMTXXTenIName}{}
\providecommand{\mGLdruleMTXXTenEName}{}
\providecommand{\mGLdruleMTXXImpIName}{}
\providecommand{\mGLdruleMTXXImpEName}{}
\providecommand{\mGLdruleMTXXGrdEName}{}
\providecommand{\mGLdruleMTXXGSubName}{}
\providecommand{\mGLdruleMTXXGrdIName}{}
\providecommand{\mGLdruleMTXXLinEName}{}
\providecommand{\mGLdruleMTXXWeakName}{}
\providecommand{\mGLdruleMTXXContName}{}
\providecommand{\mGLdruleMTXXGExName}{}
\providecommand{\mGLdruleMTXXExName}{}

%% file: mGL-seq-term-eq-theory-fragment-ottput.tex
\begin{gather*}
  \setlength\arraycolsep{0.1em}
  \hspace{-2em}
  \begin{array}{rcl}
    \multirow{3}{*}{
 $ \inferrule*[right=$\mGLdruleGSTXXContName{}$]
  {
\inferrule*[right=$\mGLdruleGSTXXWeakName{}$]
           {\delta_{{\mathrm{1}}}  \mGLsym{,}  \mGLnt{r}  \mGLsym{,}  \delta_{{\mathrm{2}}}  \odot  \Delta_{{\mathrm{1}}}  \mGLsym{,}  \mGLnt{X}  \mGLsym{,}  \Delta_{{\mathrm{2}}}  \vdash_{\mathsf{GS} }  \mGLnt{t}  \mGLsym{:}  \mGLnt{Y}}
           {\delta_{{\mathrm{1}}}  \mGLsym{,}  0  \mGLsym{,}  \mGLnt{r}  \mGLsym{,}  \delta_{{\mathrm{2}}}  \odot  \Delta_{{\mathrm{1}}}  \mGLsym{,}  \mGLnt{X}  \mGLsym{,}  \mGLnt{X}  \mGLsym{,}  \Delta_{{\mathrm{2}}}  \vdash_{\mathsf{GS} }  \mGLnt{t}  \mGLsym{:}  \mGLnt{Y}}
  }
  {\delta_{{\mathrm{1}}}  \mGLsym{,}  0  +  \mGLnt{r}  \mGLsym{,}  \delta_{{\mathrm{2}}}  \odot  \Delta_{{\mathrm{1}}}  \mGLsym{,}  \mGLnt{X}  \mGLsym{,}  \Delta_{{\mathrm{2}}}  \vdash_{\mathsf{GS} }  \mGLnt{t}  \mGLsym{:}  \mGLnt{Y}}$
    }
    & \\
    & \equiv\;\; & \delta_{{\mathrm{1}}}  \mGLsym{,}  \mGLnt{r}  \mGLsym{,}  \delta_{{\mathrm{2}}}  \odot  \Delta_{{\mathrm{1}}}  \mGLsym{,}  \mGLnt{X}  \mGLsym{,}  \Delta_{{\mathrm{2}}}  \vdash_{\mathsf{GS} }  \mGLnt{t}  \mGLsym{:}  \mGLnt{Y} 
    \\
   &  & \hspace{13em} \textnormal{\large{(\textsc{contr-unitL})}} \\ \\
    %
    \multirow{3}{*}{$
\inferrule*[right=$\mGLdruleGSTXXExName{}$]
  {\delta_{{\mathrm{1}}}  \mGLsym{,}  \mGLnt{r}  \mGLsym{,}  \mGLnt{s}  \mGLsym{,}  \delta_{{\mathrm{2}}}  \odot  \Delta_{{\mathrm{1}}}  \mGLsym{,}  \mGLnt{X}  \mGLsym{,}  \mGLnt{X}  \mGLsym{,}  \Delta_{{\mathrm{2}}}  \vdash_{\mathsf{GS} }  \mGLnt{t}  \mGLsym{:}  \mGLnt{Y}}
  {\inferrule*[right=$\mGLdruleGSTXXContName{}$]
    {
    \delta_{{\mathrm{1}}}  \mGLsym{,}  \mGLnt{s}  \mGLsym{,}  \mGLnt{r}  \mGLsym{,}  \delta_{{\mathrm{2}}}  \odot  \Delta_{{\mathrm{1}}}  \mGLsym{,}  \mGLnt{X}  \mGLsym{,}  \mGLnt{X}  \mGLsym{,}  \Delta_{{\mathrm{2}}}  \vdash_{\mathsf{GS} }  \mGLnt{t}  \mGLsym{:}  \mGLnt{Y}
    }
    {
     \delta_{{\mathrm{1}}}  \mGLsym{,}  \mGLnt{s}  +  \mGLnt{r}  \mGLsym{,}  \delta_{{\mathrm{2}}}  \odot  \Delta_{{\mathrm{1}}}  \mGLsym{,}  \mGLnt{X}  \mGLsym{,}  \Delta_{{\mathrm{2}}}  \vdash_{\mathsf{GS} }  \mGLnt{t}  \mGLsym{:}  \mGLnt{Y}
    }
  }$}
  &  &  \multirow{3}{*}{$
 \inferrule*[right=$\mGLdruleGSTXXContName{}$]
  {
    \delta_{{\mathrm{1}}}  \mGLsym{,}  \mGLnt{r}  \mGLsym{,}  \mGLnt{s}  \mGLsym{,}  \delta_{{\mathrm{2}}}  \odot  \Delta_{{\mathrm{1}}}  \mGLsym{,}  \mGLnt{X}  \mGLsym{,}  \mGLnt{X}  \mGLsym{,}  \Delta_{{\mathrm{2}}}  \vdash_{\mathsf{GS} }  \mGLnt{t}  \mGLsym{:}  \mGLnt{Y}
  }
  {
   \delta_{{\mathrm{1}}}  \mGLsym{,}  \mGLnt{r}  +  \mGLnt{s}  \mGLsym{,}  \delta_{{\mathrm{2}}}  \odot  \Delta_{{\mathrm{1}}}  \mGLsym{,}  \mGLnt{X}  \mGLsym{,}  \mGLnt{X}  \mGLsym{,}  \Delta_{{\mathrm{2}}}  \vdash_{\mathsf{GS} }  \mGLnt{t}  \mGLsym{:}  \mGLnt{Y}
  }$} \\
    & \equiv\;\; & \\
    & &
  \end{array}
\tag{\textsc{contr-sym}}
\end{gather*}

%% file: mGL-seq-cat-semantic-ottput.tex
\newcommand{\diag}[0]{\mathsf{diag}}
\newcommand{\ixcat}{\cat{I}}
\newcommand{\basecat}{\cat{C}}
\newcommand{\functorCategory}[2]{[#1, #2]}
\newcommand{\Endo}[1]{\functorCategory{#1}{#1}}

\noindent
We detail a denotational model for $\mGL{}$ which is based on an
adjoint decomposition of graded comonads. We introduce key definitions
as needed.

A \emph{graded comonad} can be
summarised as a colax monoidal functor
$\Box : \ixcat{} \mto\ \Endo{\basecat{}}$
where $\ixcat{}$ is a preordered monoid
$(\ixcat, 1, *, \leq)$
treated as a monoidal category
and $\Endo{\cat{C}}$ is the category of endofunctors on $\cat{C}$~\cite{DBLP:conf/icalp/PetricekOM13,DBLP:conf/icfp/PetricekOM14}.
Colax monoidality of $\Box$ means that the laws of a monoidal functor become 2-cells, providing the graded comonad operations:
  \begin{align*}
    \hspace{-1em}\begin{array}{c}
{\xymatrix@C=2.7em@R=1.8em{
1 \ar[d]_-{1} \ar@/^1pc/[dr]^{\mathsf{Id}} & \\
\ixcat{} \ar[r]_-{\Box}
\ar@{}[ur]^(.2){}="a"^(.65){}="b" \ar@{=>}@<-0.7ex>^{\varepsilon}
"a";"b"
 & \Endo{\basecat}
}}
\quad 
{\xymatrix@C=2.7em@R=1.8em{
\ixcat{} \times \ixcat{}
\ar[d]_-{*}
\ar[r]^-{\Box \times \Box}
& \Endo{\basecat} \times \Endo{\basecat}
\ar[d]^-{\circ}
\\
\ixcat{} \ar[r]_-{\Box}
\ar@{}[ur]^(.2){}="a"^(.85){}="b"
\ar@{=>}@<-0.8ex>^{\delta}"a";"b"
 & \Endo{\basecat}
}}
\end{array}
\end{align*}
which are thus natural transformations
$\varepsilon_{A} : \Box _{ 1 } \mGLnt{A} \rightarrow \mGLnt{A}$.
and
$\delta_{r, s, A} : \Box _{ \mGLsym{(}  \mGLnt{r}  *  \mGLnt{s}  \mGLsym{)} } \mGLnt{A} \rightarrow \Box _{ \mGLnt{r} } \mGLsym{(}   \Box _{ \mGLnt{s} } \mGLnt{A}   \mGLsym{)}$.

Fujii et al.~\cite{Fujii:2016b} gave a formal theory for graded monads,
which can be easily dualised to graded comonads,
showing that in an analogous way to an ordinary comonad, every graded
comonad can be decomposed into an adjunction
$\func{Mny} \dashv \func{Lin} : \cat{M} \mto \cat{C} $
and (key to \emph{graded} comonads) a monoidal action
$\odot :  \mathcal{R} \times \cat{C} \mto \cat{C}$, and thus vice
versa:
\begin{lemma}(Resolution of a graded comonad~\cite{Katsumata:2018,Fujii:2016b})
  \label{lemma:adjoint-decomposition}
An adjunction $\func{L} \dashv \func{R} : \cat{M} \mto \cat{C} $ and a strict monoidal action $\odot :  \mathcal{R} \times \cat{C} \mto
\cat{C}$ together induce a graded comonad over the family of endofunctors defined by
$\Box_r = \func{L}(r \odot (\func{R} -)) : \cat{M} \mto \cat{M}$.
\end{lemma}
Along with some additional structure relating to substructurality
(see below), this result provides a model of $\mGL{}$ with $\cat{C}$ providing
a model for $\mathsf{GS}$ derivations, $\cat{M}$ providing a model
for $\mathsf{MS}$ derivations, the type constructor $\func{Grd}_r$ transporting
from $\mathsf{GS}$ to $\mathsf{MS}$ modelled
by $\func{L} (r \odot -) : \cat{C} \rightarrow \cat{M}$,
and type constructor $\func{Lin}$ transporting $\mathsf{MS}$
to $\mathsf{GS}$ modelled by $\func{R} : \cat{M} \rightarrow \cat{C}$.

However we need additional structure for the (sub)structural
behaviour of our logic. In the literature on graded modal type theories,
graded comonads are extended to \emph{graded exponential comonads} (sometimes called
 \emph{graded linear exponential comonads}~\cite{Katsumata:2018})
defined as a colax monoidal functor $\Box : \mathcal{R} \mto\ \Endo{\cat{M}}$
where $\mathcal{R}$ is a preordered semiring $\mGLsym{(}  \mathcal{R}  \mGLsym{,}  1  \mGLsym{,}  *  \mGLsym{,}  \mathsf{0}  \mGLsym{,}  +  \mGLsym{,}  \leq  \mGLsym{)}$
(viewed as a category), $\Endo{\cat{M}}$ is the
category of symmetric lax monoidal endofunctors on a symmetric
monoidal category $\cat{M}$, and $\Box$ has additional
symmetric lax monoidal structure for the additional
monoidality of $\mathcal{R}$ and $\Endo{\cat{M}}$~\cite{Gaboardi:2016}.
This additional structure provides natural transformations $w_A : \Box _{ \mathsf{0} } \mGLnt{A} \mto 1$ and $c_{r,s,A} : \Box _{ \mGLsym{(}  \mGLnt{r}  +  \mGLnt{s}  \mGLsym{)} } \mGLnt{A} \mto \mGLsym{(}   \Box _{ \mGLnt{r} } \mGLnt{A}   \mGLsym{)}  \otimes  \mGLsym{(}   \Box _{ \mGLnt{s} } \mGLnt{A}   \mGLsym{)}$ capturing (graded) weakening and contraction, subject
to comonoidal coherence conditions.
This additional structure can be induced by the adjoint decomposition
given an \emph{exponential action}:

\begin{definition}[Exponential action]
  \label{def:strict-action}
  Given a preordered semiring $\mGLsym{(}  \mathcal{R}  \mGLsym{,}  1  \mGLsym{,}  *  \mGLsym{,}  \mathsf{0}  \mGLsym{,}  +  \mGLsym{,}  \leq  \mGLsym{)}$ and a
  symmetric monoidal category $(\cat{C}, \mathsf{J}, \boxtimes)$, we
  say that a bifunctor
  $\odot : \mathcal{R} \times \mGLnt{C} \mto \cat{C}$ is
  \begin{enumerate}

  \item a \emph{strict action} (a \emph{strict graded comonad}), if it
    satisfies the following equalities:
    \begin{align*}
    \begin{array}{rrll}
      \varepsilon_X :\  & 1  \odot  \mGLnt{X} & = & \mGLnt{X}\\
      \delta_{X,r,s} :\ & \mGLsym{(}  \mGLnt{r}  *  \mGLnt{s}  \mGLsym{)}  \odot  \mGLnt{X} & = & \mGLnt{r}  \odot  \mGLsym{(}  \mGLnt{s}  \odot  \mGLnt{X}  \mGLsym{)}
    \end{array}
    \end{align*}
    Note that we treat these equalities as strict natural transformations
    named $\varepsilon$ and $\delta$;

    \item \emph{symmetric lax monoidal} in the second argument if it has:
    \begin{gather*}
      \begin{array}{rrll}
        \qquad
        m_{\mathsf{J},r} :\ &  \mathsf{J}       &  \rightarrow & \mGLnt{r}  \odot  \mathsf{J}  \\
        m_{\boxtimes,r,X,Y} :\ & (\mGLnt{r}  \odot  \mGLnt{X}) \boxtimes (\mGLnt{r}  \odot  \mGLnt{Y}) & \rightarrow & \mGLnt{r}  \odot  \mGLsym{(}  \mGLnt{X}  \boxtimes  \mGLnt{Y}  \mGLsym{)}
    \end{array}
    \end{gather*}
    where $m_{\mathsf{J}}$ is the unit of $m_{\boxtimes}$ and $m_{\boxtimes}$ is associative
    and commutative up to isomorphism;

  \item \emph{symmetric colax monoidal} between $(\mathcal{R}, 0, +, \leq)$ and
    $(\cat{C}, \mathsf{J}, \boxtimes)$ in the first argument if it has natural transformations:
    \[
    \begin{array}{rrll}
      \mathsf{weak}_X : & 0  \odot  \mGLnt{X} & \rightarrow & \mathsf{J} \\
      \mathsf{contr}_{r,s,X} : & \mGLsym{(}  \mGLnt{r}  +  \mGLnt{s}  \mGLsym{)}  \odot  \mGLnt{X} & \rightarrow & (\mGLnt{r}  \odot  \mGLnt{X}) \boxtimes (\mGLnt{s}  \odot  \mGLnt{X})
    \end{array}
    \]
    where $\mathsf{weak}$ is the unit of $\mathsf{contr}$, e.g.
    $\rho_{r \odot X} \circ (id \boxtimes \func{weak}_X) \circ \func{contr}_{r,0,X} = id$
    with right unitor $\rho$, and
    $\mathsf{contr}$ is associative and commutative, i.e.,
    that $\mathsf{contr}_{r,s,X} = c \circ \mathsf{contr}_{s,r,X}$.

    Furthermore,
    these natural transformations must be preserved by the strict
    action and monoidal structure as described by the standard additional equations
    in Figure~\ref{fig:equations-strict-action}.

    \begin{figure}[t]
      \begin{center}
      {\scalebox{0.97}{\begin{minipage}{1\linewidth}
            \begin{align*}
              \hspace{3em}
              \begin{array}{c}
      \xymatrix@C=2em{
        0  \odot  \mGLnt{X} \ar@{=}[r] \ar[ddr]_{\mathsf{weak}_{\mGLnt{X}}}
             & \mGLsym{(}  0  *  \mGLnt{s}  \mGLsym{)}  \odot  \mGLnt{X} \ar@{=}[d]^{\delta_{X,0,s}} \\
             & 0  \odot  \mGLsym{(}  \mGLnt{s}  \odot  \mGLnt{X}  \mGLsym{)} \ar[d]^{\mathsf{weak}_{\mGLnt{s}  \odot  \mGLnt{X}}} \\
             & \mathsf{J}
      }
      \xymatrix@C=2em{
        0  \odot  \mGLnt{X} \ar@{=}[r] \ar[d]_{\mathsf{weak}_{X}}
         & \mGLsym{(}  \mGLnt{s}  *  0  \mGLsym{)}  \odot  \mGLnt{X} \ar@{=}[d]^{\delta_{X,s,0}} \\
        \mathsf{J} \ar[dr]_{m_{\mathsf{J},s}}  & \mGLnt{s}  \odot  \mGLsym{(}  0  \odot  \mGLnt{X}  \mGLsym{)} \ar[d]^{s \odot \mathsf{weak}_{X}} \\
            & \mGLnt{s}  \odot  \mathsf{J}
      }
      \\
      \xymatrix@C=6em{
        \mGLsym{(}  \mGLnt{r}  *  \mGLsym{(}  \mGLnt{s_{{\mathrm{1}}}}  +  \mGLnt{s_{{\mathrm{2}}}}  \mGLsym{)}  \mGLsym{)}  \odot  \mGLnt{X} \ar@{=}[r]
                              \ar[d]_{\delta_{\mGLnt{r},\mGLnt{s_{{\mathrm{1}}}}  +  \mGLnt{s_{{\mathrm{2}}}},X}}
          & \mGLsym{(}  \mGLsym{(}  \mGLnt{r}  *  \mGLnt{s_{{\mathrm{1}}}}  \mGLsym{)}  +  \mGLsym{(}  \mGLnt{r}  *  \mGLnt{s_{{\mathrm{2}}}}  \mGLsym{)}  \mGLsym{)}  \odot  \mGLnt{X} \ar[d]^{\mathsf{contr}_{\mGLnt{r}  *  \mGLnt{s_{{\mathrm{1}}}}, \mGLnt{r}  *  \mGLnt{s_{{\mathrm{2}}}},X}} \\
        \mGLnt{r}  \odot  \mGLsym{(}  \mGLsym{(}  \mGLnt{s_{{\mathrm{1}}}}  +  \mGLnt{s_{{\mathrm{2}}}}  \mGLsym{)}  \odot  \mGLnt{X}  \mGLsym{)}
                              \ar[d]_{r \odot \mathsf{contr}_{\mGLnt{s_{{\mathrm{1}}}},\mGLnt{s_{{\mathrm{2}}}},X}}
          & \mGLsym{(}  \mGLnt{r}  *  \mGLnt{s_{{\mathrm{1}}}}  \mGLsym{)}  \odot  \mGLnt{X}   \boxtimes   \mGLsym{(}  \mGLnt{r}  *  \mGLnt{s_{{\mathrm{2}}}}  \mGLsym{)}  \odot  \mGLnt{X}
                              \ar[d]^{\delta_{r,\mGLnt{s_{{\mathrm{1}}}},X}\ \boxtimes\ \delta_{r,\mGLnt{s_{{\mathrm{2}}}},X}} \\
        \mGLnt{r}  \odot  \mGLsym{(}  \mGLsym{(}  \mGLnt{s_{{\mathrm{1}}}}  \odot  \mGLnt{X}  \mGLsym{)}  \boxtimes  \mGLsym{(}  \mGLnt{s_{{\mathrm{2}}}}  \odot  \mGLnt{X}  \mGLsym{)}  \mGLsym{)}
        & \ar[l]^{m_{\boxtimes, r, \mGLnt{s_{{\mathrm{1}}}}  \odot  \mGLnt{X}, \mGLnt{s_{{\mathrm{2}}}}  \odot  \mGLnt{X}}}
        \mGLnt{r}  \odot  \mGLsym{(}  \mGLnt{s_{{\mathrm{1}}}}  \odot  \mGLnt{X}  \mGLsym{)}  \boxtimes  \mGLnt{r}  \odot  \mGLsym{(}  \mGLnt{s_{{\mathrm{2}}}}  \odot  \mGLnt{X}  \mGLsym{)}
      }
      \\\\
      \xymatrix{
         \mGLsym{(}  \mGLsym{(}  \mGLnt{s_{{\mathrm{1}}}}  +  \mGLnt{s_{{\mathrm{2}}}}  \mGLsym{)}  *  \mGLnt{r}  \mGLsym{)}  \odot  \mGLnt{X} \ar@{=}[r]
                              \ar[d]_{\delta_{\mGLnt{s_{{\mathrm{1}}}}  +  \mGLnt{s_{{\mathrm{2}}}},\mGLnt{r},X}}
          & \mGLsym{(}  \mGLsym{(}  \mGLnt{s_{{\mathrm{1}}}}  *  \mGLnt{r}  \mGLsym{)}  +  \mGLsym{(}  \mGLnt{s_{{\mathrm{2}}}}  *  \mGLnt{r}  \mGLsym{)}  \mGLsym{)}  \odot  \mGLnt{X} \ar[d]^{\mathsf{contr}_{\mGLnt{s_{{\mathrm{1}}}}  *  \mGLnt{r}, \mGLnt{s_{{\mathrm{2}}}}  *  \mGLnt{r},X}} \\
        \mGLsym{(}  \mGLnt{s_{{\mathrm{1}}}}  +  \mGLnt{s_{{\mathrm{2}}}}  \mGLsym{)}  \odot  \mGLsym{(}  \mGLnt{r}  \odot  \mGLnt{X}  \mGLsym{)}
                              \ar[d]_{\mathsf{contr}_{\mGLnt{s_{{\mathrm{1}}}},\mGLnt{s_{{\mathrm{2}}}},{\mGLnt{r}  \odot  \mGLnt{X}}}}
          & \mGLsym{(}  \mGLnt{s_{{\mathrm{1}}}}  *  \mGLnt{r}  \mGLsym{)}  \odot  \mGLnt{X}   \boxtimes   \mGLsym{(}  \mGLnt{s_{{\mathrm{2}}}}  *  \mGLnt{r}  \mGLsym{)}  \odot  \mGLnt{X}
                              \ar[d]^{\delta_{\mGLnt{s_{{\mathrm{1}}}},r,X}\ \boxtimes\ \delta_{\mGLnt{s_{{\mathrm{2}}}},r,X}} \\
        \mGLsym{(}  \mGLnt{s_{{\mathrm{1}}}}  \odot  \mGLsym{(}  \mGLnt{r}  \odot  \mGLnt{X}  \mGLsym{)}  \mGLsym{)}  \boxtimes  \mGLsym{(}  \mGLnt{s_{{\mathrm{2}}}}  \odot  \mGLsym{(}  \mGLnt{r}  \odot  \mGLnt{X}  \mGLsym{)}  \mGLsym{)} \ar@{=}[r]
          & \mGLsym{(}  \mGLnt{s_{{\mathrm{1}}}}  \odot  \mGLsym{(}  \mGLnt{r}  \odot  \mGLnt{X}  \mGLsym{)}  \mGLsym{)}  \boxtimes  \mGLsym{(}  \mGLnt{s_{{\mathrm{2}}}}  \odot  \mGLsym{(}  \mGLnt{r}  \odot  \mGLnt{X}  \mGLsym{)}  \mGLsym{)}
      }
      \end{array}
    \end{align*}
      \end{minipage}}}
      \end{center}
      \caption{Further
        equations of a strict exponential action, interacting the colax symmetric
        monoidal structure, strict action, and (strict) monoidality.}
        \label{fig:equations-strict-action}
    \end{figure}

  \end{enumerate}
\end{definition}
 If we have all of the above properties then we refer to $\odot$ as an \emph{exponential action}. This terminology recalls the
\emph{exponential action} of Brunel et al.~\cite{Brunel:2014} which is
the same as the above but where strictness is instead
laxness in their definition. Our definition is also similar to linear exponential graded comonads (see e.g.,~\cite{Katsumata:2018,Gaboardi:2016}), but here the graded comonad is uncurried (in the form of an action) and has equalities for its natural transformations (strictness).

We define a \emph{strict exponential action} to be an exponential action as above but where the monoidal structure
$m_{\mathsf{J}}$ and $m_{\boxtimes}$ is also strict, where for clarity (in the appendix) we sometimes orient the equality as a morphism, where in the opposite direction we denote these morphisms by $n_{\mathsf{J},r}$ and $n_{\boxtimes,r,X,Y}$ respectively. Strictness of the monoidal structure
is needed for soundness of our model.

We now give the definition of the model of $\mGL{}$, where we now
use the opposite category $\op{\mathcal{R}}$ to capture the correct
polarity of the approximation rules.
\begin{definition}[\mGLL{} model]
  \label{def:mGLL-adjoint-model}
  Suppose $(\cat{C}, \mathsf{J}, \boxtimes)$ and $(\cat{M}, \mathsf{I}, \otimes)$
  are symmetric monoidal
  categories, where $\cat{M}$ is symmetric monoidal closed (with exponents $\multimap$),
  and $\mGLsym{(}  \mathcal{R}  \mGLsym{,}  1  \mGLsym{,}  *  \mGLsym{,}  \mathsf{0}  \mGLsym{,}  +  \mGLsym{,}  \leq  \mGLsym{)}$ is a preordered semiring.
  Then a \emph{\mGLL{} model} is a symmetric monoidal adjunction
  $\func{Mny} \dashv \func{Lin} : \cat{M} \mto \cat{C}$ along with an exponential action
  $\odot : \op{\mathcal{R}} \times \mathcal{C} \mto \mathcal{C}$.
\end{definition}
Thus an $\mGL{}$ model is essentially an LNL model with a strict action.
However, whilst Benton's LNL models are initially stated to require
that $\cat{M}$ is Cartesian closed, he goes on to show that
Cartesian properties are induced for the Eilenberg-Moore
category of $!$-coalgebras for a symmetric monoidal category~\cite{Benton:1994}.
In our setting, the Cartesian structure is not needed since the
$\MS$ logic is a mix of graded and linear logic,
rather than Cartesian and linear logic. That is, graded propositions
do not have arbitrary weakening and contraction, but instead these
structural rules are controlled by grades (and corresponding
underlying categorical structure~\cite{Gaboardi:2016,Katsumata:2018}). Therefore,
a symmetric monoidal closed $\cat{M}$ suffices.

From Definition~\ref{def:mGLL-adjoint-model}, we define our denotational
model of $\mGL{}$:
\begin{definition}[Interpretation of \mGLL{} Logic.]
  Given a \mGLL{} model (Def.~\ref{def:mGLL-adjoint-model}) (with $\func{Mny} \dashv \func{Lin} :
  \cat{M} \mto \cat{C} $ and $\odot : \op{\mathcal{R}} \times \mathcal{C} \mto
  \mathcal{C}$),  we interpret by two mutually
  defined interpretations $\interp{-}^{\GS}$ and $\interp{-}^{\MS{}}$
  on types and proofs (derivations):
  \begin{itemize}
  \item For every $\GS{}$ type $X$ there is an object $\interp{ \mGLnt{X} }^{\GS} \in \cat{C}$ and
    for every $\MS{}$ type $A$ there is an object $\interp{ \mGLnt{A} }^{\MS} \in \cat{M}$, mutually defined inductively as:
    \begin{align*}
      \begin{array}{rl}
      \interp{ \mathsf{J} }^{\GS} & = \mathsf{J} \\
      \interp{ \mGLnt{X}  \boxtimes  \mGLnt{Y} }^{\GS} & = \interp{ \mGLnt{X} }^{\GS} \boxtimes \interp{ \mGLnt{Y} }^{\GS} \\
      \interp{ \mathsf{Lin} \, \mGLnt{A} }^{\GS} & = \Lin \interp{ \mGLnt{A} }^\MS \\
      \quad & \quad
      \end{array}
      \quad & \quad
      \begin{array}{rl}
      \interp{ \mathsf{I} }^\MS & = \mathsf{I} \\
      \interp{ \mGLnt{A}  \otimes  \mGLnt{B} }^\MS & = \interp{ \mGLnt{A} }^\MS \otimes \interp{ \mGLnt{B} }^\MS \\
      \interp{ \mGLnt{A}  \multimap  \mGLnt{B} }^\MS & = \interp{ \mGLnt{A} }^\MS \multimap \interp{ \mGLnt{B} }^\MS \\
      \interp{  \mathsf{Grd} _{ \mGLnt{r} }\, \mGLnt{X}  }^\MS & = \Mny (r \odot \interp{ \mGLnt{X} }^{\GS} )
      \end{array}
    \end{align*}

   \item For every proof $\Pi$ of a $\GS{}$ sequent
    $(  \mGLnt{r_{{\mathrm{1}}}}  \mGLsym{,} \, ... \, \mGLsym{,}  \mGLnt{r_{\mGLmv{n}}}  )   \odot   ( \mGLmv{x_{{\mathrm{1}}}}  \mGLsym{:}  \mGLnt{X_{{\mathrm{1}}}}  \mGLsym{,} \, ... \, \mGLsym{,}  \mGLmv{x_{\mGLmv{n}}}  \mGLsym{:}  \mGLnt{X_{\mGLmv{n}}} )   \vdash_{\mathsf{GS} }  \mGLnt{t}  \mGLsym{:}  \mGLnt{X}$
    there is a morphism in the category $\cat{C}$:
    \[
      {\interp{ \Pi }^{\GS}} : (\mGLnt{r_{{\mathrm{1}}}}  \odot   \interp{ \mGLnt{X_{{\mathrm{1}}}} }^{\mathsf{GS} }) \boxtimes \ldots \boxtimes (\mGLnt{r_{\mGLmv{n}}}  \odot   \interp{ \mGLnt{X_{\mGLmv{n}}} }^{\mathsf{GS} }) \mto \interp{ \mGLnt{X} }
      \]
      (where an empty context is interpreted as $\emptyset^{\GS} = \mathsf{J}$).

  \item For every proof $\Pi$ of an $\MS{}$ sequent
    $(  \mGLnt{r_{{\mathrm{1}}}}  \mGLsym{,} \, ... \, \mGLsym{,}  \mGLnt{r_{\mGLmv{n}}}  )   \odot   ( \mGLmv{x_{{\mathrm{1}}}}  \mGLsym{:}  \mGLnt{X_{{\mathrm{1}}}}  \mGLsym{,} \, ... \, \mGLsym{,}  \mGLmv{x_{\mGLmv{n}}}  \mGLsym{:}  \mGLnt{X_{\mGLmv{n}}} )   \mGLsym{;}  \mGLmv{y_{{\mathrm{1}}}}  \mGLsym{:}  \mGLnt{A_{{\mathrm{1}}}}  \mGLsym{,} \, ... \, \mGLsym{,}  \mGLmv{y_{\mGLmv{m}}}  \mGLsym{:}  \mGLnt{A_{\mGLmv{m}}}  \vdash_{\mathsf{MS} }  \mGLnt{l}  \mGLsym{:}  \mGLnt{B}$
    there is a morphism in the category $\mathcal{M}$:
    \begin{gather*}
      {\interp{ \Pi }^{\MS}} :
      \Mny\mGLsym{(}  \mGLnt{r_{{\mathrm{1}}}}  \odot   \interp{ \mGLnt{X_{{\mathrm{1}}}} }^{\mathsf{GS} }   \mGLsym{)} \otimes \ldots \otimes \Mny \mGLsym{(}  \mGLnt{r_{\mGLmv{n}}}  \odot   \interp{ \mGLnt{X_{\mGLmv{n}}} }^{\mathsf{GS} }   \mGLsym{)} \otimes \interp{ \mGLnt{A_{{\mathrm{1}}}} }^{\MS} \otimes \ldots \otimes \interp{ \mGLnt{A_{\mGLmv{m}}} }^{\MS} \mto \interp{ \mGLnt{B} }^{\MS}
    \end{gather*}
    (where an empty $\MS{}$ context is interpreted as $\emptyset^{\MS} = \mathsf{I}$).
  \end{itemize}
  \citeappendix{subsec:model}{C.4} gives the full definition of the
  interpretation, including intermediate derivations from the
  $\mGL{}$ model.
\end{definition}
\noindent
Finally, we have our soundness and completeness theorems:

\begin{restatable}[Soundness of \mGLL{} Logic models]{theorem}{mGLLSoundTheorem}
  \label{theorem:soundness_of_mgll_logic}
  Suppose a mixed graded/linear model as above.
  Then for derivation
  $\Pi_{{\mathrm{1}}}$ of
  $\delta  \odot  \Delta  \vdash_{\mathsf{GS} }  \mGLnt{t_{{\mathrm{1}}}}  \mGLsym{:}  \mGLnt{X}$
  and derivation
  $\Pi_{{\mathrm{2}}}$ of
  $\delta  \odot  \Delta  \vdash_{\mathsf{GS} }  \mGLnt{t_{{\mathrm{2}}}}  \mGLsym{:}  \mGLnt{X}$
  then if $\Pi_{{\mathrm{1}}} \equiv \Pi_{{\mathrm{2}}}$
  then $\interp{\Pi_{{\mathrm{1}}}} = \interp{\Pi_{{\mathrm{2}}}}$.

  Similarly for $\Pi_{{\mathrm{1}}}$ of
  $\delta  \odot  \Delta  \mGLsym{;}  \Gamma  \vdash_{\mathsf{MS} }  \mGLnt{l_{{\mathrm{1}}}}  \mGLsym{:}  \mGLnt{A}$
  and derivation
  $\Pi_{{\mathrm{2}}}$ of
  $\delta  \odot  \Delta  \mGLsym{;}  \Gamma  \vdash_{\mathsf{MS} }  \mGLnt{l_{{\mathrm{2}}}}  \mGLsym{:}  \mGLnt{A}$
  then if $\Pi_{{\mathrm{1}}} \equiv \Pi_{{\mathrm{2}}}$
  then $\interp{\Pi_{{\mathrm{1}}}} = \interp{\Pi_{{\mathrm{2}}}}$.

\end{restatable}
\begin{proof}
  This proof holds by mutual induction.  For the details see \citeappendix{subsec:proof_of_soundness_of_mgll_logic}{C.5}.
\end{proof}

\begin{restatable}[Completeness of \mGLL{} Logic models]{theorem}{mGLLCompletenessTheorem}
  \label{theorem:completeness}
  For derivations
  $\Pi_{{\mathrm{1}}}$, $\Pi_{{\mathrm{2}}}$ (of either $\mathsf{GS}$ or $\mathsf{MS}$)
  if $\interp{ \Pi_{{\mathrm{1}}} } = \interp{ \Pi_{{\mathrm{2}}} }$ in all mixed graded/linear models, then
  $\Pi_{{\mathrm{1}}} \equiv \Pi_{{\mathrm{2}}}$.
\end{restatable}
\begin{proof}
  This is a standard proof,
  where we build a generic model based on the syntax and the equational theory.
  For the details see \citeappendix{subsec:proof_of_completeness_of_mgll_logic}{C.6}.
\end{proof}


%% file: mGL-term-ottput.tex
We now develop a natural deduction formulation of $\mGL{}$. Whilst sequent
calculus judgments were denoted $\vdash_{\mathsf{MS} }$ and $\vdash_{\mathsf{GS} }$, natural
deduction judgments are correspondingly $\vdash_{\mathsf{MT} }$ and $\vdash_{\mathsf{GT} }$.


%
The syntax for terms is identical to the sequent calculus, collected in Figure~\ref{fig:collected-syntax}.
%
\citeappendix{sec:term-assigment_for_the_mGL-nat-deduc}{B} gives the introduction
and elimination rules and structural rules for $\mGL{}$'s natural deduction
formulation. The unit constructors are $\mathsf{j}$ and $\mathsf{i}$. Tensor
products in both systems are denoted by pairs of terms with corresponding
let-expressions for eliminators.  The graded modal introduction form $\mathsf{Lin} \, \mGLnt{l}$
operates on mixed terms, dual to $\mathsf{Grd} \, \mGLnt{r} \, \mGLnt{t}$ which operates on graded
terms.  The mixed syntax includes abstraction $\lambda  \mGLmv{x}  \mGLsym{.}  \mGLnt{l}$ and
function application $\mGLnt{l_{{\mathrm{1}}}} \, \mGLnt{l_{{\mathrm{2}}}}$.

\noindent
The most interesting aspect is the rules for the modal
operators:
\begin{gather*}
\hspace{-2em}\begin{align*}
\begin{array}{ll}
  \mGLdruleGTXXLinI{} & \; \mGLdruleMTXXLinE{} \\[0.75em]
  \mGLdruleMTXXGrdI{} & \; \mGLdruleMTXXGrdE{}
\end{array}
\end{align*}
\end{gather*}
In the sequent calculus presented in
Section~\ref{sec:mixed_graded/linear_logic}, the right rule for $\mathsf{Lin}$ is in
the graded subsystem, but the left rule is in the mixed subsystem.  A similar
idea arises here, the introduction rule for $\mathsf{Lin}$ (rule
\mGLdruleGTXXLinIName{}) is in the graded subsystem and the elimination rule
(rule \mGLdruleMTXXLinEName{}) is in the mixed subsystem. Introducing $\mathsf{Grd} _{ \mGLnt{r} }$ formulas (rule \mGLdruleMTXXGrdIName{}) has the effect of scaling the input
grades by $\mGLnt{r}$. The elimination rule for $\mathsf{Grd} _{ \mGLnt{r} }$ (rule
\mGLdruleMTXXGrdEName{}) is a pattern match on the form of $\mGLnt{l_{{\mathrm{1}}}}$. Since
$\mathsf{Lin}$ and $\mathsf{Grd}$ are the decomposition of graded modalities
(Section~\ref{sec:mGL-seq_denotational_model}), the form of the elimination rule
for $\mathsf{Grd} _{ \mGLnt{r} }$ is defined in a way which resembles that of elimination rules
for graded modalities in other natural deduction-based type
systems~\cite{Orchard:2019}.

This formulation also has explicit graded structural rules:
\begin{gather*}
  \begin{align*}
  \begin{array}{l}
    \mGLdruleGTXXWeak{} \quad \mGLdruleGTXXCont{} \\[0.75em]
    \mGLdruleGTXXEx{}
  \end{array}
  \end{align*}
  \end{gather*}
In the transition from sequent calculus to natural deduction, left rules
transform into elimination rules, and as a result the additional graded
left rules in the mixed sequent calculus are no longer explicitly part of the
system, but can be derived.  We go on to prove that the sequent calculus
of Section~\ref{subsec:mgl_sequent_calculus} is equivalent to the natural
deduction system.

\input{mGL-nat-meta-ottput}


%% file: mGL-nat-meta-ottput.tex
We give two main results related to the natural deduction system; the first is
substitution for typing.
Note that this reuses the row-vector multiplication operation of
Section~\ref{subsubsec:metatheory_of_mgl}.
\begin{restatable}[Substitution for $\vdash_{\mathsf{GT} }$ and $\vdash_{\mathsf{MT} }$]{lemma}{substitutionLemma}
  \label{lemma:subsitution_for_gt_mt}
  The following hold by mutual induction:
  \begin{enumerate}
  \item (Graded) If $\delta_{{\mathrm{2}}}  \odot  \Delta_{{\mathrm{2}}}  \vdash_{\mathsf{GT} }  \mGLnt{t_{{\mathrm{1}}}}  \mGLsym{:}  \mGLnt{X}$ and
    $(  \delta_{{\mathrm{1}}}  \mGLsym{,}  \delta  \mGLsym{,}  \delta_{{\mathrm{3}}}  )   \odot   ( \Delta_{{\mathrm{1}}}  \mGLsym{,}   \mGLmv{x} ^{ \mGLmv{n} } :  \mGLnt{X} ^{ \mGLmv{n} }   \mGLsym{,}  \Delta_{{\mathrm{3}}} )   \vdash_{\mathsf{GT} }  \mGLnt{t_{{\mathrm{2}}}}  \mGLsym{:}  \mGLnt{Y}$, then
    $(   \delta_{{\mathrm{1}}}  \mGLsym{,}  \delta  \boxast [ { \delta_{{\mathrm{2}}} }^{ \mGLmv{n} } ]   \mGLsym{,}  \delta_{{\mathrm{3}}}  )   \odot   ( \Delta_{{\mathrm{1}}}  \mGLsym{,}  \Delta_{{\mathrm{2}}}  \mGLsym{,}  \Delta_{{\mathrm{3}}} )   \vdash_{\mathsf{GT} }  \mGLsym{[}   \mGLnt{t_{{\mathrm{1}}}} ,\ldots, \mGLnt{t_{{\mathrm{1}}}}   \mGLsym{/}  \mGLmv{x_{{\mathrm{1}}}}  \mGLsym{,} \, ... \, \mGLsym{,}  \mGLmv{x_{\mGLmv{n}}}  \mGLsym{]}  \mGLnt{t_{{\mathrm{2}}}}  \mGLsym{:}  \mGLnt{Y}$.

  \item (Graded/Mixed) If $\delta_{{\mathrm{2}}}  \odot  \Delta_{{\mathrm{2}}}  \vdash_{\mathsf{GT} }  \mGLnt{t}  \mGLsym{:}  \mGLnt{X}$ and
    $(  \delta_{{\mathrm{1}}}  \mGLsym{,}  \delta  \mGLsym{,}  \delta_{{\mathrm{3}}}  )   \odot   ( \Delta_{{\mathrm{1}}}  \mGLsym{,}   \mGLmv{x} ^{ \mGLmv{n} } :  \mGLnt{X} ^{ \mGLmv{n} }   \mGLsym{,}  \Delta_{{\mathrm{3}}} )   \mGLsym{;}  \Gamma  \vdash_{\mathsf{MT} }  \mGLnt{l}  \mGLsym{:}  \mGLnt{B}$, then
    $(   \delta_{{\mathrm{1}}}  \mGLsym{,}  \delta  \boxast [ { \delta_{{\mathrm{2}}} }^{ \mGLmv{n} } ]   \mGLsym{,}  \delta_{{\mathrm{3}}}  )   \odot   ( \Delta_{{\mathrm{1}}}  \mGLsym{,}  \Delta_{{\mathrm{2}}}  \mGLsym{,}  \Delta_{{\mathrm{3}}} )   \mGLsym{;}  \Gamma  \vdash_{\mathsf{MT} }  \mGLsym{[}   \mGLnt{t} ,\ldots, \mGLnt{t}   \mGLsym{/}  \mGLmv{x_{{\mathrm{1}}}}  \mGLsym{,} \, ... \, \mGLsym{,}  \mGLmv{x_{\mGLmv{n}}}  \mGLsym{]}  \mGLnt{l}  \mGLsym{:}  \mGLnt{B}$.

  \item (Mixed) If $\delta_{{\mathrm{2}}}  \odot  \Delta_{{\mathrm{2}}}  \mGLsym{;}  \Gamma_{{\mathrm{2}}}  \vdash_{\mathsf{MT} }  \mGLnt{l_{{\mathrm{1}}}}  \mGLsym{:}  \mGLnt{A}$ and
    $\delta_{{\mathrm{1}}}  \odot  \Delta_{{\mathrm{1}}}  \mGLsym{;}   ( \Gamma_{{\mathrm{1}}}  \mGLsym{,}  \mGLmv{x}  \mGLsym{:}  \mGLnt{A}  \mGLsym{,}  \Gamma_{{\mathrm{3}}} )   \vdash_{\mathsf{MT} }  \mGLnt{l_{{\mathrm{2}}}}  \mGLsym{:}  \mGLnt{B}$, then
    $(  \delta_{{\mathrm{1}}}  \mGLsym{,}  \delta_{{\mathrm{2}}}  )   \odot   ( \Delta_{{\mathrm{1}}}  \mGLsym{,}  \Delta_{{\mathrm{2}}} )   \mGLsym{;}   ( \Gamma_{{\mathrm{1}}}  \mGLsym{,}  \Gamma_{{\mathrm{2}}}  \mGLsym{,}  \Gamma_{{\mathrm{3}}} )   \vdash_{\mathsf{MT} }  \mGLsym{[}  \mGLnt{l_{{\mathrm{1}}}}  \mGLsym{/}  \mGLmv{x}  \mGLsym{]}  \mGLnt{l_{{\mathrm{2}}}}  \mGLsym{:}  \mGLnt{B}$.
  \end{enumerate}
\end{restatable}
\begin{proof}
  By mutual induction on the second assumed derivation (see
  \citeappendix{subsec:proof_of_substitution_for_mgll_logic}{C.7}).
\end{proof}
Since we have an explicit structural rule for contraction (above and
listed in \citeappendix{sec:term-assigment_for_the_mGL-nat-deduc}{B}), the
substitution lemma on the graded fragment is formalized as multi-substitution.
Otherwise, its proof is a fairly standard substitution proof for graded
systems (e.g., as in~\cite{Orchard:2019}). Lastly, the natural deduction system
is interderivable with the sequent calculus, which we establish such that the
term witnessing the derivations does not change between systems:
\begin{restatable}[Sequent calculus and natural deduction
interderivability]{theorem}{MSTimpliesMTgen}
  \label{lemma:mst-implies-mt-gen}
  $\delta  \odot  \Delta  \vdash_{\mathsf{GS} }  \mGLnt{t}  \mGLsym{:}  \mGLnt{X} \Leftrightarrow \delta  \odot  \Delta  \vdash_{\mathsf{GT} }  \mGLnt{t}  \mGLsym{:}  \mGLnt{X}$
  and $\delta  \odot  \Delta  \mGLsym{;}  \Gamma  \vdash_{\mathsf{MS} }  \mGLnt{l}  \mGLsym{:}  \mGLnt{A} \Leftrightarrow \delta  \odot  \Delta  \mGLsym{;}  \Gamma  \vdash_{\mathsf{MT} }  \mGLnt{l}  \mGLsym{:}  \mGLnt{A}$.
\end{restatable}
\begin{proof}
  By mutual induction on the assumed derivations (\citeappendixtwo{subsec:proof_of_mst_implies_mt}{subsec:proof_of_mt_implies_mst}{C.8}{C.9}).
  The sequent calculus to natural
  deduction direction requires the substitution lemma above.
\end{proof}
The implication of the previous result is that we only need a semantic model of
one of the two systems, and the other can be modelled using the same
interpretation of terms. We chose to model the sequent calculus form directly.

%% file: discussion-ottput.tex
\subsection{Relating linear base vs. graded base calculi}
\label{subsec:linear-vs-graded-base}

A major thread of graded type systems in the literature starts with a
linear logic base and then generalises the $!$ modality to a
semiring-graded modality atop a linear logic, e.g., the systems of Brunel et
al.~\cite{Brunel:2014}, Gaboardi et al.~\cite{Gaboardi:2016}, Orchard
et al.~\cite{Orchard:2019}, and others~\cite{gaboardi2013linear,hughes:lirmm-03271465}. Often these systems are
presented with a single context containing
both linear and graded propositions~\cite{Gaboardi:2016,Orchard:2019}. Overall,
these approaches have a common core which is isomorphic to the natural deduction $\MST$ fragment shown here
with the (natural deduction analogue of the) derived $\Box_r$ graded modality of Lemma~\ref{lemma:graded_modalities_in_mGL} as part of their
definition (i.e., not derived).
We refer to this style of graded type system as the \emph{linear base} style.

A contrasting approach has no base notion of linearity, but instead has pervasive
grading tracking substructurality, i.e., no linear assumptions, every assumption has a grade, and function arrows
come equipped with a grade describing the usage of their input in the function (e.g., written  $A \xrightarrow{r} B$). Such
systems include the coeffect calculi of Petricek et
al.~\cite{DBLP:conf/icalp/PetricekOM13,DBLP:conf/icfp/PetricekOM14}, the general
graded modal system of Bernardy et al.~\cite{DBLP:journals/pacmpl/AbelB20}, and several
others~\cite{atkey2018syntax,DBLP:journals/pacmpl/BernardyBNJS18,
DBLP:journals/pacmpl/ChoudhuryEEW21,
DBLP:conf/birthday/McBride16,DBLP:conf/esop/MoonEO21}. The $\GST{}$ fragment
of our system here corresponds to a common subset of these approaches: a subset
without function arrows and without a graded modality,
since there is no graded modality that lives in the $\GST{}$ side ($\Box_i$ is derived into $\MST{}$). Hughes et al. also develop
a program synthesis technique for graded base systems, where grades are used
to prune the search space~\cite{DBLP:conf/esop/HughesO24}; its synthesis calculus formulation
 resembles closely $\GS$.

Our work thus shows the relationship between the linear base and graded base style, namely that there is a mutual
embedding between these two approaches which generates the graded modality
 in the linear base (Lemma~\ref{lemma:graded_modalities_in_mGL}). Exploring this in more depth is further work. For instance, it is unclear
 what is needed to realise a graded comonadic modality in $\GST{}$ that arises from the embedding (or a different embedding), and how this could
 interact with a graded function arrow in $\GS$ or $\GST$. Pursuing this line of work would help to
 explain the relationship between the two dominant styles of graded system in the literature, which seem strongly related, and their relative expressive power.  Nonetheless,
 by following Benton's programme and giving it a graded rendering here, we can already see here the close connection between these two
 styles of graded system.

\subsection{Related work on adjoint logics}
\label{sec:related_work}

\input{related-work-ottput}

\subsection{Further work}
\paragraph*{Practical implementation to leverage linear/grading separation}
The separation of the mixed system ($\MS$/$\MST$) from
the purely graded fragment ($\GS$/$\GST$) (which acts more as a standard
intuitionistic system) can provide a basis for a programming
language design. In such a language, the restrictions of linearity could be
used only for handling data that needs to be linear, such as file
handles or channels. However, for data types which need not be linear,
e.g., primitive types like integers, characters, or structures
over them, the graded fragment could be used without having to confront
linearity constraints. The mutual embedding would allow the programmer to
move smoothly between these two subcalculi. Similar ideas
are discussed for the polarized extension of SILL for concurrent programming~\cite{DBLP:conf/fossacs/PfenningG15}. The implementation could borrow ideas from
the Granule programming language, which already provides a mature
and feature rich implementation of a linear-base style graded type system~\cite{Orchard:2019}. Instead, an $\mGL{}$-inspired implementation could be
based on the natural deduction term calculus with the modalities mediating
between the two judgments. Exploring this application, perhaps as an extension to Granule,
is future work.

\paragraph*{Other generalisations}

In LNL, the adjunction can be followed in the opposite
direction to derive a monad $? A = \mathsf{Lin} (\mathsf{Mny} A)$. However,
in $\mGL{}$ we do not get a graded monad by composing $\mathsf{Lin} \, \mGLsym{(}   \mathsf{Grd} _{ \mGLnt{r} }\, \mGLnt{X}   \mGLsym{)}$
since the adjoint resolution of a graded monad has a strict action
on the other side (on $\cal{M}$ in the model). Exploring a calculus
with a pair of actions to allow both graded monads and graded comonads is further work.

\paragraph*{Uniqueness typing}
Recent work has demonstrated that
\emph{uniqueness} is a closely related but distinct concept to linearity~\cite{DBLP:conf/esop/MarshallVO22}; uniqueness logic~\cite{DBLP:journals/tcs/Harrington06} is substructural in much
the same way as linear logic, but provides a monadic modality for enabling the
structural rules in contrast to linear logic's comonadic $\oc$ modality.
Building an adjoint model for uniqueness or a calculus which
unifies uniqueness and linearity~\cite{DBLP:conf/esop/MarshallVO22} would be
interesting further work, and this could potentially be extended to more recent
systems which develop graded notions of uniqueness~\cite{marshall2023functional}.

%% file: related-work-ottput.tex
Pruiksma et al. formalized a general way to add and remove structural
rules from a logic through adjunctions~\cite{Pruiksma:2018}. Their work is
similar to ours as it relates logics through adjoint decompositions based on modal
operators to control structural rules. Their formulation with ``modes of truth''
resembles our work with grades; however, modes of truth lack the algebraic properties
graded formulations depend on and instead have a very relational flavor.
Building on this work of Pruiksma et al., Jang et al. develop
a natural deduction formulation of adjoint logic~\cite{DBLP:conf/fscd/0001RPP24}.
They leverage this to give a functional language able to reason about
resource properties like strictness and erasure. Similar reasoning
can be developed ontop of our natural deduction formulation here, though
this is left as further work.

A question is whether grading can be unified with
the adjoint logic approach.  Eades and Orchard sketched a unification
based on generalising semiring operations to
relations rather than functions, with predicates classifying unit
values~\cite{DBLP:journals/corr/abs-2006-08854}.
Hanukaev et al. develop this idea further,
introducing a dependent type system based on a similar
structure as the logics here but using a generalised notion of
grading that combines the modes of adjoint
logic~\cite{10.1145/3609027.3609408}. They prove that their system is
well-formed syntactically, but do not introduce any semantic
model. Our logic $\mGL{}$ can be seen as an instantiation of their
system, but the categorical model given here could potentially be
generalised into a model of their system.

%% file: appendix.tex
\section{Collected rules for the $\mGL{}$ sequent calculus}
\label{sec:term-assigment_for_the_mGL-seq-calc}
\input{mGL-seq-term-ottput}

\section{Collected rules for the $\mGL{}$ natural deduction system}
\label{sec:term-assigment_for_the_mGL-nat-deduc}
\input{mGL-nat-term-ottput}

\section{Proofs of Properties of $\mGL{}$}
\label{sec:properties_of_mgl}

\subsection{Cut Reduction of $\mGL{}$}
\label{subsec:cut_reduction_of_mgl}
\input{cut-elim-mGL-ottput}

\subsection{Decreasing Order of $\mGL{}$}
\label{subsec:decreasing_order_of_mgl}
\input{decreasing-mGL-ottput}

\subsection{Proof of subformula property}
\label{subsec:proof_of_subformaul_property}
\input{subformula-ottput}

\subsection{Denotational model}
\label{subsec:model}
\input{mGL-interpretation-ottput}

\subsection{Proof of soundness of \mGLL{} Logic}
\label{subsec:proof_of_soundness_of_mgll_logic}
\input{mGL-soundness-theorem-ottput}
\input{mGL-soundness-theorem-eq-theory-ottput}

\subsection{Proof of completeness of \mGLL{} Logic}
\label{subsec:proof_of_completeness_of_mgll_logic}
\input{mGL-completeness-theorem-ottput}

\subsection{Proof of Substitution for \mGLL{} Logic}
\label{subsec:proof_of_substitution_for_mgll_logic}
\input{substitution-proof-ottput}

\subsection{Proof that MS/GS implies MT/GT}
\label{subsec:proof_of_mst_implies_mt}
\input{mst-implies-mt-ottput}

\subsection{Proof that MT/GT implies MS/GS}
\label{subsec:proof_of_mt_implies_mst}
\input{mt-implies-mst-ottput}



%% file: mGL-seq-term-ottput.tex
\begin{mdframed}
\drules[GST]{$\delta  \odot  \Delta  \vdash_{\mathsf{GS} }  \mGLnt{t}  \mGLsym{:}  \mGLnt{r}  \odot  \mGLnt{X}$}{Graded System}{
  id,UnitR,UnitL,TenR,TenL,LinR,Cut,Weak,Cont,Ex,Sub
}
\end{mdframed}

\begin{mdframed}
\drules[MST]{$\delta  \odot  \Delta  \mGLsym{;}  \Gamma  \vdash_{\mathsf{MS} }  \mGLnt{l}  \mGLsym{:}  \mGLnt{A}$}{Mixed System}{
  id,Sub,UnitR,UnitL
}
\end{mdframed}

\begin{mdframed}
\drules[MST]{$\delta  \odot  \Delta  \mGLsym{;}  \Gamma  \vdash_{\mathsf{MS} }  \mGLnt{l}  \mGLsym{:}  \mGLnt{A}$}{Mixed System (continued)}{
 ImpR,ImpL,TenL, TenR,GUnitL,GTenL,GrdR,LinL,GrdL,Cut,GCut
}
\end{mdframed}

\begin{mdframed}
\drules[MST]{$\delta  \odot  \Delta  \mGLsym{;}  \Gamma  \vdash_{\mathsf{MS} }  \mGLnt{l}  \mGLsym{:}  \mGLnt{A}$}{Mixed System (continued)}{
 Weak,Cont,Ex,GEx
}
\end{mdframed}

\subsection{Equational theory}
\label{sec:full-eq-theory}

If cut elimination on
  derivation $\Pi_{{\mathrm{1}}}$ of $\delta  \odot  \Delta  \vdash_{\mathsf{GS} }  \mGLnt{t}  \mGLsym{:}  \mGLnt{X}$
  yields the cut-free derivation of $\Pi_{{\mathrm{2}}}$ for $\delta  \odot  \Delta  \vdash_{\mathsf{GS} }  \mGLnt{t'}  \mGLsym{:}  \mGLnt{X}$
  then $\Pi_{{\mathrm{1}}} \equiv \Pi_{{\mathrm{2}}}$, and similarly for $\MS{}$.
  Below gives the remaining equalities.

Given an inequality $\mGLnt{r}  \leq  \mGLnt{s}$ then we write:
\begin{align*}
  \inferrule*[right=$\mGLdruleGSTXXSubName{}$]
  { \delta_{{\mathrm{1}}}  \mGLsym{,}  \mGLnt{s}  \mGLsym{,}  \delta_{{\mathrm{2}}}  \odot  \Delta_{{\mathrm{1}}}  \mGLsym{,}  \mGLnt{X}  \mGLsym{,}  \Delta_{{\mathrm{2}}}  \vdash_{\mathsf{GS} }  \mGLnt{t}  \mGLsym{:}  \mGLnt{Y} \quad \mGLnt{r}  \leq  \mGLnt{s}}
  { \delta_{{\mathrm{1}}}  \mGLsym{,}  \mGLnt{r}  \mGLsym{,}  \delta_{{\mathrm{2}}}  \odot  \Delta_{{\mathrm{1}}}  \mGLsym{,}  \mGLnt{X}  \mGLsym{,}  \Delta_{{\mathrm{2}}}  \vdash_{\mathsf{GS} }  \mGLnt{t}  \mGLsym{:}  \mGLnt{Y} }
\end{align*}
as syntactic sugar for:
\begin{align*}
  \inferrule*[right=$\mGLdruleGSTXXSubName{}$]
  { \delta_{{\mathrm{1}}}  \mGLsym{,}  \mGLnt{s}  \mGLsym{,}  \delta_{{\mathrm{2}}}  \odot  \Delta_{{\mathrm{1}}}  \mGLsym{,}  \mGLnt{X}  \mGLsym{,}  \Delta_{{\mathrm{2}}}  \vdash_{\mathsf{GS} }  \mGLnt{t}  \mGLsym{:}  \mGLnt{Y} \quad \dfrac{\mGLnt{r}  \leq  \mGLnt{s}}{(  \delta_{{\mathrm{1}}}  \mGLsym{,}  \mGLnt{r}  \mGLsym{,}  \delta_{{\mathrm{2}}}  )   \leq   (  \delta_{{\mathrm{1}}}  \mGLsym{,}  \mGLnt{s}  \mGLsym{,}  \delta_{{\mathrm{2}}}  )}\textsc{cong}  }
  { \delta_{{\mathrm{1}}}  \mGLsym{,}  \mGLnt{r}  \mGLsym{,}  \delta_{{\mathrm{2}}}  \odot  \Delta_{{\mathrm{1}}}  \mGLsym{,}  \mGLnt{X}  \mGLsym{,}  \Delta_{{\mathrm{2}}}  \vdash_{\mathsf{GS} }  \mGLnt{t}  \mGLsym{:}  \mGLnt{Y} }
\end{align*}


\begin{align*}
      \fbox{$\vdash_{\mathsf{GS}} \equiv$}\quad\hfill\textit{Graded system}%
\end{align*}

\input{mGL-seq-term-eq-theory-fragment-ottput}
  \begin{gather*}
    \inferrule*[right=$\mGLdruleGSTXXSubName{}$]
    {
      \inferrule*[right=\Phi( \delta )]
                 { \mathcal{G}[ {  \delta  } ]   \odot  \Delta  \vdash_{\mathsf{GS} }  \mGLnt{t}  \mGLsym{:}  \mGLnt{X} }
                 { \mathcal{G}'[ {  \delta  } ]   \odot  \Delta'  \vdash_{\mathsf{GS} }  \mGLnt{t'}  \mGLsym{:}  \mGLnt{X'} }
      \quad
      \delta  \leq  \delta'
    }
    {
      \mathcal{G}'[ {  \delta'  } ]   \odot  \Delta'  \vdash_{\mathsf{GS} }  \mGLnt{t'}  \mGLsym{:}  \mGLnt{X'}
    }
    \equiv
    \inferrule*[right=\Phi( \delta' )]
    {
      \inferrule*[right=$\mGLdruleGSTXXSubName{}$]
         { \mathcal{G}[ {  \delta  } ]   \odot  \Delta  \vdash_{\mathsf{GS} }  \mGLnt{t}  \mGLsym{:}  \mGLnt{X} \quad \delta  \leq  \delta' }
         { \mathcal{G}[ {  \delta'  } ]   \odot  \Delta  \vdash_{\mathsf{GS} }  \mGLnt{t}  \mGLsym{:}  \mGLnt{X} }
    }
    {
      \mathcal{G}'[ {  \delta'  } ]   \odot  \Delta'  \vdash_{\mathsf{GS} }  \mGLnt{t'}  \mGLsym{:}  \mGLnt{X'}
    }
    \tag{\textsc{sub-comm-conv}}
  \end{gather*}

    \begin{gather*}
  \inferrule*[right=$\mGLdruleGSTXXContName{}$]
  {
\inferrule*[right=$\mGLdruleGSTXXWeakName{}$]
           {\delta_{{\mathrm{1}}}  \mGLsym{,}  \mGLnt{r}  \mGLsym{,}  \delta_{{\mathrm{2}}}  \odot  \Delta_{{\mathrm{1}}}  \mGLsym{,}  \mGLnt{X}  \mGLsym{,}  \Delta_{{\mathrm{2}}}  \vdash_{\mathsf{GS} }  \mGLnt{t}  \mGLsym{:}  \mGLnt{Y}}
           {\delta_{{\mathrm{1}}}  \mGLsym{,}  \mGLnt{r}  \mGLsym{,}  0  \mGLsym{,}  \delta_{{\mathrm{2}}}  \odot  \Delta_{{\mathrm{1}}}  \mGLsym{,}  \mGLnt{X}  \mGLsym{,}  \mGLnt{X}  \mGLsym{,}  \Delta_{{\mathrm{2}}}  \vdash_{\mathsf{GS} }  \mGLnt{t}  \mGLsym{:}  \mGLnt{Y}}
  }
  {\delta_{{\mathrm{1}}}  \mGLsym{,}  \mGLnt{r}  +  0  \mGLsym{,}  \delta_{{\mathrm{2}}}  \odot  \Delta_{{\mathrm{1}}}  \mGLsym{,}  \mGLnt{X}  \mGLsym{,}  \Delta_{{\mathrm{2}}}  \vdash_{\mathsf{GS} }  \mGLnt{t}  \mGLsym{:}  \mGLnt{Y}}
\equiv\;\;
    \delta_{{\mathrm{1}}}  \mGLsym{,}  \mGLnt{r}  \mGLsym{,}  \delta_{{\mathrm{2}}}  \odot  \Delta_{{\mathrm{1}}}  \mGLsym{,}  \mGLnt{X}  \mGLsym{,}  \Delta_{{\mathrm{2}}}  \vdash_{\mathsf{GS} }  \mGLnt{t}  \mGLsym{:}  \mGLnt{Y}
\quad
\tag{\textsc{contr-unitR}}
  \end{gather*}

  \begin{gather*}
%
%
\inferrule*[right=$\mGLdruleGSTXXContName{}$]{
  \inferrule*[right=$\mGLdruleGSTXXContName{}$]
  {
    \delta_{{\mathrm{1}}}  \mGLsym{,}  \mGLnt{r_{{\mathrm{1}}}}  \mGLsym{,}  \mGLnt{r_{{\mathrm{2}}}}  \mGLsym{,}  \mGLnt{r_{{\mathrm{3}}}}  \mGLsym{,}  \delta_{{\mathrm{2}}}  \odot  \Delta_{{\mathrm{1}}}  \mGLsym{,}  \mGLnt{X}  \mGLsym{,}  \mGLnt{X}  \mGLsym{,}  \mGLnt{X}  \mGLsym{,}  \Delta_{{\mathrm{2}}}  \vdash_{\mathsf{GS} }  \mGLnt{t}  \mGLsym{:}  \mGLnt{Y}
  }
  {
   \delta_{{\mathrm{1}}}  \mGLsym{,}  \mGLnt{r_{{\mathrm{1}}}}  +  \mGLnt{r_{{\mathrm{2}}}}  \mGLsym{,}  \mGLnt{r_{{\mathrm{3}}}}  \mGLsym{,}  \delta_{{\mathrm{2}}}  \odot  \Delta_{{\mathrm{1}}}  \mGLsym{,}  \mGLnt{X}  \mGLsym{,}  \mGLnt{X}  \mGLsym{,}  \Delta_{{\mathrm{2}}}  \vdash_{\mathsf{GS} }  \mGLnt{t}  \mGLsym{:}  \mGLnt{Y}
  }
}
{
\delta_{{\mathrm{1}}}  \mGLsym{,}  \mGLsym{(}  \mGLnt{r_{{\mathrm{1}}}}  +  \mGLnt{r_{{\mathrm{2}}}}  \mGLsym{)}  +  \mGLnt{r_{{\mathrm{3}}}}  \mGLsym{,}  \delta_{{\mathrm{2}}}  \odot  \Delta_{{\mathrm{1}}}  \mGLsym{,}  \mGLnt{X}  \mGLsym{,}  \Delta_{{\mathrm{2}}}  \vdash_{\mathsf{GS} }  \mGLnt{t}  \mGLsym{:}  \mGLnt{Y}
}
\equiv
\inferrule*[right=$\mGLdruleGSTXXContName{}$]{
  \inferrule*[right=$\mGLdruleGSTXXContName{}$]
  {
    \delta_{{\mathrm{1}}}  \mGLsym{,}  \mGLnt{r_{{\mathrm{1}}}}  \mGLsym{,}  \mGLnt{r_{{\mathrm{2}}}}  \mGLsym{,}  \mGLnt{r_{{\mathrm{3}}}}  \mGLsym{,}  \delta_{{\mathrm{2}}}  \odot  \Delta_{{\mathrm{1}}}  \mGLsym{,}  \mGLnt{X}  \mGLsym{,}  \mGLnt{X}  \mGLsym{,}  \mGLnt{X}  \mGLsym{,}  \Delta_{{\mathrm{2}}}  \vdash_{\mathsf{GS} }  \mGLnt{t}  \mGLsym{:}  \mGLnt{Y}
  }
  {
   \delta_{{\mathrm{1}}}  \mGLsym{,}  \mGLnt{r_{{\mathrm{1}}}}  \mGLsym{,}  \mGLnt{r_{{\mathrm{2}}}}  +  \mGLnt{r_{{\mathrm{3}}}}  \mGLsym{,}  \delta_{{\mathrm{2}}}  \odot  \Delta_{{\mathrm{1}}}  \mGLsym{,}  \mGLnt{X}  \mGLsym{,}  \mGLnt{X}  \mGLsym{,}  \Delta_{{\mathrm{2}}}  \vdash_{\mathsf{GS} }  \mGLnt{t}  \mGLsym{:}  \mGLnt{Y}
  }
}
{
\delta_{{\mathrm{1}}}  \mGLsym{,}  \mGLnt{r_{{\mathrm{1}}}}  +  \mGLsym{(}  \mGLnt{r_{{\mathrm{2}}}}  +  \mGLnt{r_{{\mathrm{3}}}}  \mGLsym{)}  \mGLsym{,}  \delta_{{\mathrm{2}}}  \odot  \Delta_{{\mathrm{1}}}  \mGLsym{,}  \mGLnt{X}  \mGLsym{,}  \Delta_{{\mathrm{2}}}  \vdash_{\mathsf{GS} }  \mGLnt{t}  \mGLsym{:}  \mGLnt{Y}
}
\tag{\textsc{contr-assoc}}
  \end{gather*}

  \begin{gather*}
%
%
\inferrule*[right=$\mGLdruleGSTXXExName{}$]
           {
             \inferrule*[right=$\mGLdruleGSTXXExName{}$]
  { \delta_{{\mathrm{1}}}  \mGLsym{,}  \mGLnt{r}  \mGLsym{,}  \mGLnt{s}  \mGLsym{,}  \delta_{{\mathrm{2}}}  \odot  \Delta_{{\mathrm{1}}}  \mGLsym{,}  \mGLnt{X}  \mGLsym{,}  \mGLnt{Y}  \mGLsym{,}  \Delta_{{\mathrm{2}}}  \vdash_{\mathsf{GS} }  \mGLnt{t}  \mGLsym{:}  \mGLnt{Z} }
  { \delta_{{\mathrm{1}}}  \mGLsym{,}  \mGLnt{s}  \mGLsym{,}  \mGLnt{r}  \mGLsym{,}  \delta_{{\mathrm{2}}}  \odot  \Delta_{{\mathrm{1}}}  \mGLsym{,}  \mGLnt{Y}  \mGLsym{,}  \mGLnt{X}  \mGLsym{,}  \Delta_{{\mathrm{2}}}  \vdash_{\mathsf{GS} }  \mGLnt{t}  \mGLsym{:}  \mGLnt{Z} }
           }
           {
    \delta_{{\mathrm{1}}}  \mGLsym{,}  \mGLnt{r}  \mGLsym{,}  \mGLnt{s}  \mGLsym{,}  \delta_{{\mathrm{2}}}  \odot  \Delta_{{\mathrm{1}}}  \mGLsym{,}  \mGLnt{X}  \mGLsym{,}  \mGLnt{Y}  \mGLsym{,}  \Delta_{{\mathrm{2}}}  \vdash_{\mathsf{GS} }  \mGLnt{t}  \mGLsym{:}  \mGLnt{Z}
           }
\equiv
\;
\delta_{{\mathrm{1}}}  \mGLsym{,}  \mGLnt{r}  \mGLsym{,}  \mGLnt{s}  \mGLsym{,}  \delta_{{\mathrm{2}}}  \odot  \Delta_{{\mathrm{1}}}  \mGLsym{,}  \mGLnt{X}  \mGLsym{,}  \mGLnt{Y}  \mGLsym{,}  \Delta_{{\mathrm{2}}}  \vdash_{\mathsf{GS} }  \mGLnt{t}  \mGLsym{:}  \mGLnt{Z}
\tag{\textsc{ex-ex}}
  \end{gather*}

  \begin{gather*}
%
%
\inferrule*[right=$\mGLdruleGSTXXSubName$]
           {\delta_{{\mathrm{1}}}  \odot  \Delta  \vdash_{\mathsf{GS} }  \mGLnt{t}  \mGLsym{:}  \mGLnt{X} \quad \delta_{{\mathrm{1}}}  \leq  \delta_{{\mathrm{1}}}}
           {\delta_{{\mathrm{1}}}  \odot  \Delta  \vdash_{\mathsf{GS} }  \mGLnt{t}  \mGLsym{:}  \mGLnt{X}}
\equiv
\;
\delta_{{\mathrm{1}}}  \odot  \Delta  \vdash_{\mathsf{GS} }  \mGLnt{t}  \mGLsym{:}  \mGLnt{X}
\tag{\textsc{sub-refl}}
\\[2em]
%
%
\inferrule*[right=$\mGLdruleGSTXXSubName$]
{\inferrule*[right=$\mGLdruleGSTXXSubName$]
           {\delta_{{\mathrm{1}}}  \odot  \Delta  \vdash_{\mathsf{GS} }  \mGLnt{t}  \mGLsym{:}  \mGLnt{X} \quad \delta_{{\mathrm{1}}}  \leq  \delta_{{\mathrm{2}}}}
           {\delta_{{\mathrm{2}}}  \odot  \Delta  \vdash_{\mathsf{GS} }  \mGLnt{t}  \mGLsym{:}  \mGLnt{X}} \quad \delta_{{\mathrm{2}}}  \leq  \delta_{{\mathrm{3}}} }
{\delta_{{\mathrm{3}}}  \odot  \Delta  \vdash_{\mathsf{GS} }  \mGLnt{t}  \mGLsym{:}  \mGLnt{X}}
\equiv
\inferrule*[right=$\mGLdruleGSTXXSubName$]
           {\delta_{{\mathrm{1}}}  \odot  \Delta  \vdash_{\mathsf{GS} }  \mGLnt{t}  \mGLsym{:}  \mGLnt{X} \quad \dfrac{\delta_{{\mathrm{1}}}  \leq  \delta_{{\mathrm{2}}} \quad \delta_{{\mathrm{2}}}  \leq  \delta_{{\mathrm{3}}}}{\delta_{{\mathrm{1}}}  \leq  \delta_{{\mathrm{3}}}} }
           {\delta_{{\mathrm{3}}}  \odot  \Delta  \vdash_{\mathsf{GS} }  \mGLnt{t}  \mGLsym{:}  \mGLnt{X}}
\tag{\textsc{sub-trans}}
  \end{gather*}

  \begin{gather*}
%
%
\inferrule*[right=$\mGLdruleGSTXXContName$]
{
\inferrule*[right=$\mGLdruleGSTXXSubName$]
{
\inferrule*[right=$\mGLdruleGSTXXSubName$]
           {\delta_{{\mathrm{1}}}  \mGLsym{,}  \mGLnt{r_{{\mathrm{1}}}}  \mGLsym{,}  \mGLnt{s_{{\mathrm{1}}}}  \mGLsym{,}  \delta_{{\mathrm{2}}}  \odot  \Delta_{{\mathrm{1}}}  \mGLsym{,}  \mGLnt{X}  \mGLsym{,}  \mGLnt{X}  \mGLsym{,}  \Delta_{{\mathrm{2}}}  \vdash_{\mathsf{GS} }  \mGLnt{t}  \mGLsym{:}  \mGLnt{Y} \quad \mGLnt{r_{{\mathrm{1}}}}  \leq  \mGLnt{r_{{\mathrm{2}}}}}
           {\delta_{{\mathrm{1}}}  \mGLsym{,}  \mGLnt{r_{{\mathrm{2}}}}  \mGLsym{,}  \mGLnt{s_{{\mathrm{1}}}}  \mGLsym{,}  \delta_{{\mathrm{2}}}  \odot  \Delta_{{\mathrm{1}}}  \mGLsym{,}  \mGLnt{X}  \mGLsym{,}  \mGLnt{X}  \mGLsym{,}  \Delta_{{\mathrm{2}}}  \vdash_{\mathsf{GS} }  \mGLnt{t}  \mGLsym{:}  \mGLnt{Y}}
\quad
\mGLnt{s_{{\mathrm{1}}}}  \leq  \mGLnt{s_{{\mathrm{2}}}}
}
{\delta_{{\mathrm{1}}}  \mGLsym{,}  \mGLnt{r_{{\mathrm{2}}}}  \mGLsym{,}  \mGLnt{s_{{\mathrm{2}}}}  \mGLsym{,}  \delta_{{\mathrm{2}}}  \odot  \Delta_{{\mathrm{1}}}  \mGLsym{,}  \mGLnt{X}  \mGLsym{,}  \mGLnt{X}  \mGLsym{,}  \Delta_{{\mathrm{2}}}  \vdash_{\mathsf{GS} }  \mGLnt{t}  \mGLsym{:}  \mGLnt{Y}}
}
{\delta_{{\mathrm{1}}}  \mGLsym{,}  \mGLnt{r_{{\mathrm{2}}}}  +  \mGLnt{s_{{\mathrm{2}}}}  \mGLsym{,}  \delta_{{\mathrm{2}}}  \odot  \Delta_{{\mathrm{1}}}  \mGLsym{,}  \mGLnt{X}  \mGLsym{,}  \Delta_{{\mathrm{2}}}  \vdash_{\mathsf{GS} }  \mGLnt{t}  \mGLsym{:}  \mGLnt{Y} }
\equiv
\;
\inferrule*[right=$\mGLdruleGSTXXSubName$]
{
\inferrule*[right=$\mGLdruleGSTXXContName$]
{\delta_{{\mathrm{1}}}  \mGLsym{,}  \mGLnt{r_{{\mathrm{1}}}}  \mGLsym{,}  \mGLnt{s_{{\mathrm{1}}}}  \mGLsym{,}  \delta_{{\mathrm{2}}}  \odot  \Delta_{{\mathrm{1}}}  \mGLsym{,}  \mGLnt{X}  \mGLsym{,}  \mGLnt{X}  \mGLsym{,}  \Delta_{{\mathrm{2}}}  \vdash_{\mathsf{GS} }  \mGLnt{t}  \mGLsym{:}  \mGLnt{Y}}
{\delta_{{\mathrm{1}}}  \mGLsym{,}  \mGLnt{r_{{\mathrm{1}}}}  +  \mGLnt{s_{{\mathrm{1}}}}  \mGLsym{,}  \delta_{{\mathrm{2}}}  \odot  \Delta_{{\mathrm{1}}}  \mGLsym{,}  \mGLnt{X}  \mGLsym{,}  \Delta_{{\mathrm{2}}}  \vdash_{\mathsf{GS} }  \mGLnt{t}  \mGLsym{:}  \mGLnt{Y}}
\quad
\dfrac{\mGLnt{r_{{\mathrm{1}}}}  \leq  \mGLnt{r_{{\mathrm{2}}}} \quad \mGLnt{s_{{\mathrm{1}}}}  \leq  \mGLnt{s_{{\mathrm{2}}}}}{\mGLnt{r_{{\mathrm{1}}}}  +  \mGLnt{s_{{\mathrm{1}}}}  \leq  \mGLnt{r_{{\mathrm{2}}}}  +  \mGLnt{s_{{\mathrm{2}}}}}
}
{\delta_{{\mathrm{1}}}  \mGLsym{,}  \mGLnt{r_{{\mathrm{2}}}}  +  \mGLnt{s_{{\mathrm{2}}}}  \mGLsym{,}  \delta_{{\mathrm{2}}}  \odot  \Delta_{{\mathrm{1}}}  \mGLsym{,}  \mGLnt{X}  \mGLsym{,}  \Delta_{{\mathrm{2}}}  \vdash_{\mathsf{GS} }  \mGLnt{t}  \mGLsym{:}  \mGLnt{Y}}
\tag{\textsc{contr-mono}}
  \end{gather*}

  \begin{gather*}
%
%
\inferrule*[right=$\mGLdruleGSTXXSubName$]
{
\inferrule*[right=$\mGLdruleGSTXXUnitLName$]
{(  \delta_{{\mathrm{1}}}  \mGLsym{,}  \delta_{{\mathrm{2}}}  )   \odot   ( \Delta_{{\mathrm{1}}}  \mGLsym{,}  \Delta_{{\mathrm{2}}} )   \vdash_{\mathsf{GS} }  \mGLnt{t}  \mGLsym{:}  \mGLnt{X}}
{(  \delta_{{\mathrm{1}}}  \mGLsym{,}  \mGLnt{r}  \mGLsym{,}  \delta_{{\mathrm{2}}}  )   \odot   ( \Delta_{{\mathrm{1}}}  \mGLsym{,}  \mGLmv{x}  \mGLsym{:}  \mathsf{J}  \mGLsym{,}  \Delta_{{\mathrm{2}}} )   \vdash_{\mathsf{GS} }  \mathsf{let} \, \mathsf{j} \, \mGLsym{=}  \mGLmv{x} \, \mathsf{in} \, \mGLnt{t}  \mGLsym{:}  \mGLnt{Y}}
\quad
\mGLnt{r}  \leq  \mGLnt{s}
}
{(  \delta_{{\mathrm{1}}}  \mGLsym{,}  \mGLnt{s}  \mGLsym{,}  \delta_{{\mathrm{2}}}  )   \odot   ( \Delta_{{\mathrm{1}}}  \mGLsym{,}  \mGLmv{x}  \mGLsym{:}  \mathsf{J}  \mGLsym{,}  \Delta_{{\mathrm{2}}} )   \vdash_{\mathsf{GS} }  \mathsf{let} \, \mathsf{j} \, \mGLsym{=}  \mGLmv{x} \, \mathsf{in} \, \mGLnt{t}  \mGLsym{:}  \mGLnt{Y}}
\equiv
\inferrule*[right=$\mGLdruleGSTXXUnitLName$]
{(  \delta_{{\mathrm{1}}}  \mGLsym{,}  \delta_{{\mathrm{2}}}  )   \odot   ( \Delta_{{\mathrm{1}}}  \mGLsym{,}  \Delta_{{\mathrm{2}}} )   \vdash_{\mathsf{GS} }  \mGLnt{t}  \mGLsym{:}  \mGLnt{X}}
{(  \delta_{{\mathrm{1}}}  \mGLsym{,}  \mGLnt{s}  \mGLsym{,}  \delta_{{\mathrm{2}}}  )   \odot   ( \Delta_{{\mathrm{1}}}  \mGLsym{,}  \mGLmv{x}  \mGLsym{:}  \mathsf{J}  \mGLsym{,}  \Delta_{{\mathrm{2}}} )   \vdash_{\mathsf{GS} }  \mathsf{let} \, \mathsf{j} \, \mGLsym{=}  \mGLmv{x} \, \mathsf{in} \, \mGLnt{t}  \mGLsym{:}  \mGLnt{Y}}
\tag{\textsc{sub-unitL}}
  \end{gather*}

  \begin{gather*}
%
%
\inferrule*[right=$\mGLdruleGSTXXSubName$]
{
\inferrule*[right=$\mGLdruleGSTXXTenLName$]
{(  \delta_{{\mathrm{1}}}  \mGLsym{,}  \mGLnt{r}  \mGLsym{,}  \mGLnt{r}  \mGLsym{,}  \delta_{{\mathrm{2}}}  )   \odot   ( \Delta_{{\mathrm{1}}}  \mGLsym{,}  \mGLmv{x}  \mGLsym{:}  \mGLnt{X}  \mGLsym{,}  \mGLmv{y}  \mGLsym{:}  \mGLnt{Y}  \mGLsym{,}  \Delta_{{\mathrm{2}}} )   \vdash_{\mathsf{GS} }  \mGLnt{t}  \mGLsym{:}  \mGLnt{Z}}
{(  \delta_{{\mathrm{1}}}  \mGLsym{,}  \mGLnt{r}  \mGLsym{,}  \delta_{{\mathrm{2}}}  )   \odot   ( \Delta_{{\mathrm{1}}}  \mGLsym{,}  \mGLmv{z}  \mGLsym{:}  \mGLnt{X}  \boxtimes  \mGLnt{Y}  \mGLsym{,}  \Delta_{{\mathrm{2}}} )   \vdash_{\mathsf{GS} }   \mathsf{let} \,( \mGLmv{x} , \mGLmv{y} ) =  \mGLmv{z} \, \mathsf{in} \, \mGLnt{t}   \mGLsym{:}  \mGLnt{Z}}
\quad
\mGLnt{r}  \leq  \mGLnt{s}
}
{(  \delta_{{\mathrm{1}}}  \mGLsym{,}  \mGLnt{s}  \mGLsym{,}  \delta_{{\mathrm{2}}}  )   \odot   ( \Delta_{{\mathrm{1}}}  \mGLsym{,}  \mGLmv{z}  \mGLsym{:}  \mGLnt{X}  \boxtimes  \mGLnt{Y}  \mGLsym{,}  \Delta_{{\mathrm{2}}} )   \vdash_{\mathsf{GS} }   \mathsf{let} \,( \mGLmv{x} , \mGLmv{y} ) =  \mGLmv{z} \, \mathsf{in} \, \mGLnt{t}   \mGLsym{:}  \mGLnt{Z} }
\equiv
\inferrule*[right=$\mGLdruleGSTXXTenLName$]
{
\inferrule*[right=$\mGLdruleGSTXXSubName$]
{
  \inferrule*[right=$\mGLdruleGSTXXSubName$]
  {
    (  \delta_{{\mathrm{1}}}  \mGLsym{,}  \mGLnt{r}  \mGLsym{,}  \mGLnt{r}  \mGLsym{,}  \delta_{{\mathrm{2}}}  )   \odot   ( \Delta_{{\mathrm{1}}}  \mGLsym{,}  \mGLmv{x}  \mGLsym{:}  \mGLnt{X}  \mGLsym{,}  \mGLmv{y}  \mGLsym{:}  \mGLnt{Y}  \mGLsym{,}  \Delta_{{\mathrm{2}}} )   \vdash_{\mathsf{GS} }  \mGLnt{t}  \mGLsym{:}  \mGLnt{Z} \quad \mGLnt{r}  \leq  \mGLnt{s}
  }
  {(  \delta_{{\mathrm{1}}}  \mGLsym{,}  \mGLnt{s}  \mGLsym{,}  \mGLnt{r}  \mGLsym{,}  \delta_{{\mathrm{2}}}  )   \odot   ( \Delta_{{\mathrm{1}}}  \mGLsym{,}  \mGLmv{x}  \mGLsym{:}  \mGLnt{X}  \mGLsym{,}  \mGLmv{y}  \mGLsym{:}  \mGLnt{Y}  \mGLsym{,}  \Delta_{{\mathrm{2}}} )   \vdash_{\mathsf{GS} }  \mGLnt{t}  \mGLsym{:}  \mGLnt{Z}}
  \quad
  \mGLnt{r}  \leq  \mGLnt{s}
}
{(  \delta_{{\mathrm{1}}}  \mGLsym{,}  \mGLnt{s}  \mGLsym{,}  \mGLnt{s}  \mGLsym{,}  \delta_{{\mathrm{2}}}  )   \odot   ( \Delta_{{\mathrm{1}}}  \mGLsym{,}  \mGLmv{x}  \mGLsym{:}  \mGLnt{X}  \mGLsym{,}  \mGLmv{y}  \mGLsym{:}  \mGLnt{Y}  \mGLsym{,}  \Delta_{{\mathrm{2}}} )   \vdash_{\mathsf{GS} }  \mGLnt{t}  \mGLsym{:}  \mGLnt{Z}}
}
{(  \delta_{{\mathrm{1}}}  \mGLsym{,}  \mGLnt{s}  \mGLsym{,}  \delta_{{\mathrm{2}}}  )   \odot   ( \Delta_{{\mathrm{1}}}  \mGLsym{,}  \mGLmv{x}  \mGLsym{:}  \mathsf{J}  \mGLsym{,}  \Delta_{{\mathrm{2}}} )   \vdash_{\mathsf{GS} }  \mathsf{let} \, \mathsf{j} \, \mGLsym{=}  \mGLmv{x} \, \mathsf{in} \, \mGLnt{t}  \mGLsym{:}  \mGLnt{Y}}
\tag{\textsc{sub-tensorL}}
  \end{gather*}

    \begin{gather*}
%
%
\inferrule*[right=$\mGLdruleGSTXXCutName$]
 {
  \inferrule*[right=$\mGLdruleGSTXXSubName$]
             {\delta_{{\mathrm{2}}}  \odot  \Delta_{{\mathrm{2}}}  \vdash_{\mathsf{GS} }  \mGLnt{t_{{\mathrm{1}}}}  \mGLsym{:}  \mGLnt{X} \quad
               \delta_{{\mathrm{2}}}  \leq  \delta'_{{\mathrm{2}}}}
             {\delta'_{{\mathrm{2}}}  \odot  \Delta_{{\mathrm{2}}}  \vdash_{\mathsf{GS} }  \mGLnt{t_{{\mathrm{1}}}}  \mGLsym{:}  \mGLnt{X} }
               \\
  \inferrule*[right=$\mGLdruleGSTXXSubName$]
            {(  \delta_{{\mathrm{1}}}  \mGLsym{,}  \mGLnt{r}  \mGLsym{,}  \delta_{{\mathrm{3}}}  )   \odot   ( \Delta_{{\mathrm{1}}}  \mGLsym{,}  \mGLmv{x}  \mGLsym{:}  \mGLnt{X}  \mGLsym{,}  \Delta_{{\mathrm{3}}} )   \vdash_{\mathsf{GS} }  \mGLnt{t_{{\mathrm{2}}}}  \mGLsym{:}  \mGLnt{Y} \quad
              \mGLnt{r}  \leq  \mGLnt{r'}}
            {(  \delta_{{\mathrm{1}}}  \mGLsym{,}  \mGLnt{r'}  \mGLsym{,}  \delta_{{\mathrm{3}}}  )   \odot   ( \Delta_{{\mathrm{1}}}  \mGLsym{,}  \mGLmv{x}  \mGLsym{:}  \mGLnt{X}  \mGLsym{,}  \Delta_{{\mathrm{3}}} )   \vdash_{\mathsf{GS} }  \mGLnt{t_{{\mathrm{2}}}}  \mGLsym{:}  \mGLnt{Y}}
 }
 {(  \delta_{{\mathrm{1}}}  \mGLsym{,}   \mGLnt{r'}   *   \delta'_{{\mathrm{2}}}   \mGLsym{,}  \delta_{{\mathrm{3}}}  )   \odot   ( \Delta_{{\mathrm{1}}}  \mGLsym{,}  \Delta_{{\mathrm{2}}}  \mGLsym{,}  \Delta_{{\mathrm{3}}} )   \vdash_{\mathsf{GS} }  \mGLsym{[}  \mGLnt{t_{{\mathrm{1}}}}  \mGLsym{/}  \mGLmv{x}  \mGLsym{]}  \mGLnt{t_{{\mathrm{2}}}}  \mGLsym{:}  \mGLnt{Y} }
\equiv
\inferrule*[right=$\mGLdruleGSTXXSubName$]
{
\inferrule*[right=$\mGLdruleGSTXXCutName$]
 {\delta_{{\mathrm{2}}}  \odot  \Delta_{{\mathrm{2}}}  \vdash_{\mathsf{GS} }  \mGLnt{t_{{\mathrm{1}}}}  \mGLsym{:}  \mGLnt{X} \quad
  (  \delta_{{\mathrm{1}}}  \mGLsym{,}  \mGLnt{r}  \mGLsym{,}  \delta_{{\mathrm{3}}}  )   \odot   ( \Delta_{{\mathrm{1}}}  \mGLsym{,}  \mGLmv{x}  \mGLsym{:}  \mGLnt{X}  \mGLsym{,}  \Delta_{{\mathrm{3}}} )   \vdash_{\mathsf{GS} }  \mGLnt{t_{{\mathrm{2}}}}  \mGLsym{:}  \mGLnt{Y} }
 { (  \delta_{{\mathrm{1}}}  \mGLsym{,}   \mGLnt{r}   *   \delta_{{\mathrm{2}}}   \mGLsym{,}  \delta_{{\mathrm{3}}}  )   \odot   ( \Delta_{{\mathrm{1}}}  \mGLsym{,}  \Delta_{{\mathrm{2}}}  \mGLsym{,}  \Delta_{{\mathrm{3}}} )   \vdash_{\mathsf{GS} }  \mGLsym{[}  \mGLnt{t_{{\mathrm{1}}}}  \mGLsym{/}  \mGLmv{x}  \mGLsym{]}  \mGLnt{t_{{\mathrm{2}}}}  \mGLsym{:}  \mGLnt{Y} }
 \dfrac{\mGLnt{r}  \leq  \mGLnt{r'} \quad \delta_{{\mathrm{2}}}  \leq  \delta'_{{\mathrm{2}}}}
       {\mGLnt{r}  *  \delta_{{\mathrm{2}}}  \leq  \mGLnt{r'}  *  \delta'_{{\mathrm{2}}}}
}
{
(  \delta_{{\mathrm{1}}}  \mGLsym{,}   \mGLnt{r'}   *   \delta'_{{\mathrm{2}}}   \mGLsym{,}  \delta_{{\mathrm{3}}}  )   \odot   ( \Delta_{{\mathrm{1}}}  \mGLsym{,}  \Delta_{{\mathrm{2}}}  \mGLsym{,}  \Delta_{{\mathrm{3}}} )   \vdash_{\mathsf{GS} }  \mGLsym{[}  \mGLnt{t_{{\mathrm{1}}}}  \mGLsym{/}  \mGLmv{x}  \mGLsym{]}  \mGLnt{t_{{\mathrm{2}}}}  \mGLsym{:}  \mGLnt{Y}
}
\tag{\textsc{mult-mono}}
  \end{gather*}

  \begin{gather*}
    \fbox{$\vdash_{\MS} \equiv$}\quad\hfill\textit{Mixed system}%
    \\[1.5em]
    \inferrule*[right=$\mGLdruleMSTXXSubName{}$]
    {
      \inferrule*[right=\Phi( \delta )]
                 { \mathcal{G}[ {  \delta  } ]   \odot  \Delta  \mGLsym{;}  \Gamma  \vdash_{\mathsf{MS} }  \mGLnt{l}  \mGLsym{:}  \mGLnt{A} }
                 { \mathcal{G}'[ {  \delta  } ]   \odot  \Delta'  \mGLsym{;}  \Gamma  \vdash_{\mathsf{MS} }  \mGLnt{l'}  \mGLsym{:}  \mGLnt{A'} }
      \quad
      \delta  \leq  \delta'
    }
    {
      \mathcal{G}'[ {  \delta'  } ]   \odot  \Delta'  \mGLsym{;}  \Gamma'  \vdash_{\mathsf{MS} }  \mGLnt{l'}  \mGLsym{:}  \mGLnt{A'}
    }
    \equiv
    \inferrule*[right=\Phi( \delta' )]
    {
      \inferrule*[right=$\mGLdruleMSTXXSubName{}$]
         { \mathcal{G}[ {  \delta  } ]   \odot  \Delta  \mGLsym{;}  \Gamma  \vdash_{\mathsf{MS} }  \mGLnt{l}  \mGLsym{:}  \mGLnt{A} \quad \delta  \leq  \delta' }
         { \mathcal{G}[ {  \delta'  } ]   \odot  \Delta  \mGLsym{;}  \Gamma  \vdash_{\mathsf{MS} }  \mGLnt{l}  \mGLsym{:}  \mGLnt{A} }
    }
    {
      \mathcal{G}'[ {  \delta'  } ]   \odot  \Delta'  \mGLsym{;}  \Gamma'  \vdash_{\mathsf{MS} }  \mGLnt{l'}  \mGLsym{:}  \mGLnt{A'}
    }
    \tag{\textsc{sub-comm-conv}}
  \end{gather*}

  \begin{gather*}
  \inferrule*[right=$\mGLdruleMSTXXContName{}$]
  {
\inferrule*[right=$\mGLdruleMSTXXWeakName{}$]
           {\delta_{{\mathrm{1}}}  \mGLsym{,}  \mGLnt{r}  \mGLsym{,}  \delta_{{\mathrm{2}}}  \odot  \Delta_{{\mathrm{1}}}  \mGLsym{,}  \mGLnt{X}  \mGLsym{,}  \Delta_{{\mathrm{2}}}  \mGLsym{;}  \Gamma  \vdash_{\mathsf{MS} }  \mGLnt{l}  \mGLsym{:}  \mGLnt{B}}
           {\delta_{{\mathrm{1}}}  \mGLsym{,}  0  \mGLsym{,}  \mGLnt{r}  \mGLsym{,}  \delta_{{\mathrm{2}}}  \odot  \Delta_{{\mathrm{1}}}  \mGLsym{,}  \mGLnt{X}  \mGLsym{,}  \mGLnt{X}  \mGLsym{,}  \Delta_{{\mathrm{2}}}  \mGLsym{;}  \Gamma  \vdash_{\mathsf{MS} }  \mGLnt{l}  \mGLsym{:}  \mGLnt{B}}
  }
  {\delta_{{\mathrm{1}}}  \mGLsym{,}  0  +  \mGLnt{r}  \mGLsym{,}  \delta_{{\mathrm{2}}}  \odot  \Delta_{{\mathrm{1}}}  \mGLsym{,}  \mGLnt{X}  \mGLsym{,}  \Delta_{{\mathrm{2}}}  \mGLsym{;}  \Gamma  \vdash_{\mathsf{MS} }  \mGLnt{l}  \mGLsym{:}  \mGLnt{B}}
\equiv\;\;
    \delta_{{\mathrm{1}}}  \mGLsym{,}  \mGLnt{r}  \mGLsym{,}  \delta_{{\mathrm{2}}}  \odot  \Delta_{{\mathrm{1}}}  \mGLsym{,}  \mGLnt{X}  \mGLsym{,}  \Delta_{{\mathrm{2}}}  \mGLsym{;}  \Gamma  \vdash_{\mathsf{MS} }  \mGLnt{l}  \mGLsym{:}  \mGLnt{B}
\quad
\tag{\textsc{contr-unitL}}
  \end{gather*}

  \begin{gather*}
%
  \inferrule*[right=$\mGLdruleMSTXXContName{}$]
  {
\inferrule*[right=$\mGLdruleMSTXXWeakName{}$]
           {\delta_{{\mathrm{1}}}  \mGLsym{,}  \mGLnt{r}  \mGLsym{,}  \delta_{{\mathrm{2}}}  \odot  \Delta_{{\mathrm{1}}}  \mGLsym{,}  \mGLnt{X}  \mGLsym{,}  \Delta_{{\mathrm{2}}}  \mGLsym{;}  \Gamma  \vdash_{\mathsf{MS} }  \mGLnt{l}  \mGLsym{:}  \mGLnt{B}}
           {\delta_{{\mathrm{1}}}  \mGLsym{,}  \mGLnt{r}  \mGLsym{,}  0  \mGLsym{,}  \delta_{{\mathrm{2}}}  \odot  \Delta_{{\mathrm{1}}}  \mGLsym{,}  \mGLnt{X}  \mGLsym{,}  \mGLnt{X}  \mGLsym{,}  \Delta_{{\mathrm{2}}}  \mGLsym{;}  \Gamma  \vdash_{\mathsf{MS} }  \mGLnt{l}  \mGLsym{:}  \mGLnt{B}}
  }
  {\delta_{{\mathrm{1}}}  \mGLsym{,}  \mGLnt{r}  +  0  \mGLsym{,}  \delta_{{\mathrm{2}}}  \odot  \Delta_{{\mathrm{1}}}  \mGLsym{,}  \mGLnt{X}  \mGLsym{,}  \Delta_{{\mathrm{2}}}  \mGLsym{;}  \Gamma  \vdash_{\mathsf{MS} }  \mGLnt{l}  \mGLsym{:}  \mGLnt{B}}
\equiv\;\;
    \delta_{{\mathrm{1}}}  \mGLsym{,}  \mGLnt{r}  \mGLsym{,}  \delta_{{\mathrm{2}}}  \odot  \Delta_{{\mathrm{1}}}  \mGLsym{,}  \mGLnt{X}  \mGLsym{,}  \Delta_{{\mathrm{2}}}  \mGLsym{;}  \Gamma  \vdash_{\mathsf{MS} }  \mGLnt{l}  \mGLsym{:}  \mGLnt{B}
\quad
\tag{\textsc{contr-unitR}}
  \end{gather*}

  \begin{gather*}
%
%
  \inferrule*[right=$\mGLdruleMSTXXContName{}$]
  {
    \delta_{{\mathrm{1}}}  \mGLsym{,}  \mGLnt{r}  \mGLsym{,}  \mGLnt{s}  \mGLsym{,}  \delta_{{\mathrm{2}}}  \odot  \Delta_{{\mathrm{1}}}  \mGLsym{,}  \mGLnt{X}  \mGLsym{,}  \mGLnt{X}  \mGLsym{,}  \Delta_{{\mathrm{2}}}  \mGLsym{;}  \Gamma  \vdash_{\mathsf{MS} }  \mGLnt{l}  \mGLsym{:}  \mGLnt{B}
  }
  {
   \delta_{{\mathrm{1}}}  \mGLsym{,}  \mGLnt{r}  +  \mGLnt{s}  \mGLsym{,}  \delta_{{\mathrm{2}}}  \odot  \Delta_{{\mathrm{1}}}  \mGLsym{,}  \mGLnt{X}  \mGLsym{,}  \mGLnt{X}  \mGLsym{,}  \Delta_{{\mathrm{2}}}  \mGLsym{;}  \Gamma  \vdash_{\mathsf{MS} }  \mGLnt{l}  \mGLsym{:}  \mGLnt{B}
  }
  \equiv
\inferrule*[right=$\mGLdruleMSTXXExName{}$]
  {\delta_{{\mathrm{1}}}  \mGLsym{,}  \mGLnt{r}  \mGLsym{,}  \mGLnt{s}  \mGLsym{,}  \delta_{{\mathrm{2}}}  \odot  \Delta_{{\mathrm{1}}}  \mGLsym{,}  \mGLnt{X}  \mGLsym{,}  \mGLnt{X}  \mGLsym{,}  \Delta_{{\mathrm{2}}}  \mGLsym{;}  \Gamma  \vdash_{\mathsf{MS} }  \mGLnt{l}  \mGLsym{:}  \mGLnt{B}}
  {\inferrule*[right=$\mGLdruleMSTXXContName{}$]
    {
    \delta_{{\mathrm{1}}}  \mGLsym{,}  \mGLnt{s}  \mGLsym{,}  \mGLnt{r}  \mGLsym{,}  \delta_{{\mathrm{2}}}  \odot  \Delta_{{\mathrm{1}}}  \mGLsym{,}  \mGLnt{X}  \mGLsym{,}  \mGLnt{X}  \mGLsym{,}  \Delta_{{\mathrm{2}}}  \mGLsym{;}  \Gamma  \vdash_{\mathsf{MS} }  \mGLnt{l}  \mGLsym{:}  \mGLnt{B}
    }
    {
     \delta_{{\mathrm{1}}}  \mGLsym{,}  \mGLnt{s}  +  \mGLnt{r}  \mGLsym{,}  \delta_{{\mathrm{2}}}  \odot  \Delta_{{\mathrm{1}}}  \mGLsym{,}  \mGLnt{X}  \mGLsym{,}  \Delta_{{\mathrm{2}}}  \mGLsym{;}  \Gamma  \vdash_{\mathsf{MS} }  \mGLnt{l}  \mGLsym{:}  \mGLnt{B}
    }
  }
\tag{\textsc{contr-sym}}
  \end{gather*}

  \begin{gather*}
%
%
\inferrule*[right=$\mGLdruleMSTXXContName{}$]{
  \inferrule*[right=$\mGLdruleMSTXXContName{}$]
  {
    \delta_{{\mathrm{1}}}  \mGLsym{,}  \mGLnt{r_{{\mathrm{1}}}}  \mGLsym{,}  \mGLnt{r_{{\mathrm{2}}}}  \mGLsym{,}  \mGLnt{r_{{\mathrm{3}}}}  \mGLsym{,}  \delta_{{\mathrm{2}}}  \odot  \Delta_{{\mathrm{1}}}  \mGLsym{,}  \mGLnt{X}  \mGLsym{,}  \mGLnt{X}  \mGLsym{,}  \mGLnt{X}  \mGLsym{,}  \Delta_{{\mathrm{2}}}  \mGLsym{;}  \Gamma  \vdash_{\mathsf{MS} }  \mGLnt{l}  \mGLsym{:}  \mGLnt{B}
  }
  {
   \delta_{{\mathrm{1}}}  \mGLsym{,}  \mGLnt{r_{{\mathrm{1}}}}  +  \mGLnt{r_{{\mathrm{2}}}}  \mGLsym{,}  \mGLnt{r_{{\mathrm{3}}}}  \mGLsym{,}  \delta_{{\mathrm{2}}}  \odot  \Delta_{{\mathrm{1}}}  \mGLsym{,}  \mGLnt{X}  \mGLsym{,}  \mGLnt{X}  \mGLsym{,}  \Delta_{{\mathrm{2}}}  \mGLsym{;}  \Gamma  \vdash_{\mathsf{MS} }  \mGLnt{l}  \mGLsym{:}  \mGLnt{B}
  }
}
{
\delta_{{\mathrm{1}}}  \mGLsym{,}  \mGLsym{(}  \mGLnt{r_{{\mathrm{1}}}}  +  \mGLnt{r_{{\mathrm{2}}}}  \mGLsym{)}  +  \mGLnt{r_{{\mathrm{3}}}}  \mGLsym{,}  \delta_{{\mathrm{2}}}  \odot  \Delta_{{\mathrm{1}}}  \mGLsym{,}  \mGLnt{X}  \mGLsym{,}  \Delta_{{\mathrm{2}}}  \mGLsym{;}  \Gamma  \vdash_{\mathsf{MS} }  \mGLnt{l}  \mGLsym{:}  \mGLnt{B}
}
\equiv
\inferrule*[right=$\mGLdruleMSTXXContName{}$]{
  \inferrule*[right=$\mGLdruleMSTXXContName{}$]
  {
    \delta_{{\mathrm{1}}}  \mGLsym{,}  \mGLnt{r_{{\mathrm{1}}}}  \mGLsym{,}  \mGLnt{r_{{\mathrm{2}}}}  \mGLsym{,}  \mGLnt{r_{{\mathrm{3}}}}  \mGLsym{,}  \delta_{{\mathrm{2}}}  \odot  \Delta_{{\mathrm{1}}}  \mGLsym{,}  \mGLnt{X}  \mGLsym{,}  \mGLnt{X}  \mGLsym{,}  \mGLnt{X}  \mGLsym{,}  \Delta_{{\mathrm{2}}}  \mGLsym{;}  \Gamma  \vdash_{\mathsf{MS} }  \mGLnt{l}  \mGLsym{:}  \mGLnt{B}
  }
  {
   \delta_{{\mathrm{1}}}  \mGLsym{,}  \mGLnt{r_{{\mathrm{1}}}}  \mGLsym{,}  \mGLnt{r_{{\mathrm{2}}}}  +  \mGLnt{r_{{\mathrm{3}}}}  \mGLsym{,}  \delta_{{\mathrm{2}}}  \odot  \Delta_{{\mathrm{1}}}  \mGLsym{,}  \mGLnt{X}  \mGLsym{,}  \mGLnt{X}  \mGLsym{,}  \Delta_{{\mathrm{2}}}  \mGLsym{;}  \Gamma  \vdash_{\mathsf{MS} }  \mGLnt{l}  \mGLsym{:}  \mGLnt{B}
  }
}
{
\delta_{{\mathrm{1}}}  \mGLsym{,}  \mGLnt{r_{{\mathrm{1}}}}  +  \mGLsym{(}  \mGLnt{r_{{\mathrm{2}}}}  +  \mGLnt{r_{{\mathrm{3}}}}  \mGLsym{)}  \mGLsym{,}  \delta_{{\mathrm{2}}}  \odot  \Delta_{{\mathrm{1}}}  \mGLsym{,}  \mGLnt{X}  \mGLsym{,}  \Delta_{{\mathrm{2}}}  \mGLsym{;}  \Gamma  \vdash_{\mathsf{MS} }  \mGLnt{l}  \mGLsym{:}  \mGLnt{B}
}
\tag{\textsc{contr-assoc}}
  \end{gather*}

    \begin{gather*}
%
%
\inferrule*[right=$\mGLdruleMSTXXGExName{}$]
           {
             \inferrule*[right=$\mGLdruleMSTXXGExName{}$]
  { \delta_{{\mathrm{1}}}  \mGLsym{,}  \mGLnt{r}  \mGLsym{,}  \mGLnt{s}  \mGLsym{,}  \delta_{{\mathrm{2}}}  \odot  \Delta_{{\mathrm{1}}}  \mGLsym{,}  \mGLnt{X}  \mGLsym{,}  \mGLnt{Y}  \mGLsym{,}  \Delta_{{\mathrm{2}}}  \mGLsym{;}  \Gamma  \vdash_{\mathsf{MS} }  \mGLnt{l}  \mGLsym{:}  \mGLnt{B} }
  { \delta_{{\mathrm{1}}}  \mGLsym{,}  \mGLnt{s}  \mGLsym{,}  \mGLnt{r}  \mGLsym{,}  \delta_{{\mathrm{2}}}  \odot  \Delta_{{\mathrm{1}}}  \mGLsym{,}  \mGLnt{Y}  \mGLsym{,}  \mGLnt{X}  \mGLsym{,}  \Delta_{{\mathrm{2}}}  \mGLsym{;}  \Gamma  \vdash_{\mathsf{MS} }  \mGLnt{l}  \mGLsym{:}  \mGLnt{B} }
           }
           {
    \delta_{{\mathrm{1}}}  \mGLsym{,}  \mGLnt{r}  \mGLsym{,}  \mGLnt{s}  \mGLsym{,}  \delta_{{\mathrm{2}}}  \odot  \Delta_{{\mathrm{1}}}  \mGLsym{,}  \mGLnt{X}  \mGLsym{,}  \mGLnt{Y}  \mGLsym{,}  \Delta_{{\mathrm{2}}}  \mGLsym{;}  \Gamma  \vdash_{\mathsf{MS} }  \mGLnt{l}  \mGLsym{:}  \mGLnt{B}
           }
\equiv
\;
\delta_{{\mathrm{1}}}  \mGLsym{,}  \mGLnt{r}  \mGLsym{,}  \mGLnt{s}  \mGLsym{,}  \delta_{{\mathrm{2}}}  \odot  \Delta_{{\mathrm{1}}}  \mGLsym{,}  \mGLnt{X}  \mGLsym{,}  \mGLnt{Y}  \mGLsym{,}  \Delta_{{\mathrm{2}}}  \mGLsym{;}  \Gamma  \vdash_{\mathsf{MS} }  \mGLnt{l}  \mGLsym{:}  \mGLnt{B}
\tag{\textsc{gex-gex}}
  \end{gather*}

    \begin{gather*}
%
%
\inferrule*[right=$\mGLdruleMSTXXExName{}$]
           {
             \inferrule*[right=$\mGLdruleMSTXXExName{}$]
  { \delta  \odot  \Delta  \mGLsym{;}  \Gamma_{{\mathrm{1}}}  \mGLsym{,}  \mGLmv{x}  \mGLsym{:}  \mGLnt{A}  \mGLsym{,}  \mGLmv{y}  \mGLsym{:}  \mGLnt{B}  \mGLsym{,}  \Gamma_{{\mathrm{2}}}  \vdash_{\mathsf{MS} }  \mGLnt{l}  \mGLsym{:}  \mGLnt{B} }
  { \delta  \odot  \Delta  \mGLsym{;}  \Gamma_{{\mathrm{1}}}  \mGLsym{,}  \mGLmv{y}  \mGLsym{:}  \mGLnt{B}  \mGLsym{,}  \mGLmv{x}  \mGLsym{:}  \mGLnt{A}  \mGLsym{,}  \Gamma_{{\mathrm{2}}}  \vdash_{\mathsf{MS} }  \mGLnt{l}  \mGLsym{:}  \mGLnt{B} }
           }
           {
    \delta  \odot  \Delta  \mGLsym{;}  \Gamma_{{\mathrm{1}}}  \mGLsym{,}  \mGLmv{x}  \mGLsym{:}  \mGLnt{A}  \mGLsym{,}  \mGLmv{y}  \mGLsym{:}  \mGLnt{B}  \mGLsym{,}  \Gamma_{{\mathrm{2}}}  \vdash_{\mathsf{MS} }  \mGLnt{l}  \mGLsym{:}  \mGLnt{B}
           }
\equiv
\;
\delta  \odot  \Delta  \mGLsym{;}  \Gamma_{{\mathrm{1}}}  \mGLsym{,}  \mGLmv{x}  \mGLsym{:}  \mGLnt{A}  \mGLsym{,}  \mGLmv{y}  \mGLsym{:}  \mGLnt{B}  \mGLsym{,}  \Gamma_{{\mathrm{2}}}  \vdash_{\mathsf{MS} }  \mGLnt{l}  \mGLsym{:}  \mGLnt{B}
\tag{\textsc{ex-ex}}
  \end{gather*}

      \begin{gather*}
%
%
\inferrule*[right=$\mGLdruleMSTXXSubName$]
           {\delta_{{\mathrm{1}}}  \odot  \Delta  \mGLsym{;}  \Gamma  \vdash_{\mathsf{MS} }  \mGLnt{l}  \mGLsym{:}  \mGLnt{B} \quad \delta_{{\mathrm{1}}}  \leq  \delta_{{\mathrm{1}}}}
           {\delta_{{\mathrm{1}}}  \odot  \Delta  \mGLsym{;}  \Gamma  \vdash_{\mathsf{MS} }  \mGLnt{l}  \mGLsym{:}  \mGLnt{B}}
\equiv
\;
\delta_{{\mathrm{1}}}  \odot  \Delta  \mGLsym{;}  \Gamma  \vdash_{\mathsf{MS} }  \mGLnt{l}  \mGLsym{:}  \mGLnt{B}
\tag{\textsc{sub-refl}}
\\[2em]
%
%
\inferrule*[right=$\mGLdruleMSTXXSubName$]
{\inferrule*[right=$\mGLdruleMSTXXSubName$]
           {\delta_{{\mathrm{1}}}  \odot  \Delta  \mGLsym{;}  \Gamma  \vdash_{\mathsf{MS} }  \mGLnt{l}  \mGLsym{:}  \mGLnt{B} \quad \delta_{{\mathrm{1}}}  \leq  \delta_{{\mathrm{2}}}}
           {\delta_{{\mathrm{2}}}  \odot  \Delta  \mGLsym{;}  \Gamma  \vdash_{\mathsf{MS} }  \mGLnt{l}  \mGLsym{:}  \mGLnt{B}} \quad \delta_{{\mathrm{2}}}  \leq  \delta_{{\mathrm{3}}} }
{\delta_{{\mathrm{3}}}  \odot  \Delta  \mGLsym{;}  \Gamma  \vdash_{\mathsf{MS} }  \mGLnt{l}  \mGLsym{:}  \mGLnt{B}}
\equiv
\inferrule*[right=$\mGLdruleMSTXXSubName$]
           {\delta_{{\mathrm{1}}}  \odot  \Delta  \mGLsym{;}  \Gamma  \vdash_{\mathsf{MS} }  \mGLnt{l}  \mGLsym{:}  \mGLnt{B} \quad \dfrac{\delta_{{\mathrm{1}}}  \leq  \delta_{{\mathrm{2}}} \quad \delta_{{\mathrm{2}}}  \leq  \delta_{{\mathrm{3}}}}{\delta_{{\mathrm{1}}}  \leq  \delta_{{\mathrm{3}}}} }
           {\delta_{{\mathrm{3}}}  \odot  \Delta  \mGLsym{;}  \Gamma  \vdash_{\mathsf{MS} }  \mGLnt{l}  \mGLsym{:}  \mGLnt{B}}
\tag{\textsc{sub-trans}}
  \end{gather*}

  \begin{gather*}
%
%
\inferrule*[right=$\mGLdruleMSTXXContName$]
{
\inferrule*[right=$\mGLdruleMSTXXSubName$]
{
\inferrule*[right=$\mGLdruleMSTXXSubName$]
           {\delta_{{\mathrm{1}}}  \mGLsym{,}  \mGLnt{r_{{\mathrm{1}}}}  \mGLsym{,}  \mGLnt{s_{{\mathrm{1}}}}  \mGLsym{,}  \delta_{{\mathrm{2}}}  \odot  \Delta_{{\mathrm{1}}}  \mGLsym{,}  \mGLnt{X}  \mGLsym{,}  \mGLnt{X}  \mGLsym{,}  \Delta_{{\mathrm{2}}}  \mGLsym{;}  \Gamma  \vdash_{\mathsf{MS} }  \mGLnt{l}  \mGLsym{:}  \mGLnt{B} \quad \mGLnt{r_{{\mathrm{1}}}}  \leq  \mGLnt{r_{{\mathrm{2}}}}}
           {\delta_{{\mathrm{1}}}  \mGLsym{,}  \mGLnt{r_{{\mathrm{2}}}}  \mGLsym{,}  \mGLnt{s_{{\mathrm{1}}}}  \mGLsym{,}  \delta_{{\mathrm{2}}}  \odot  \Delta_{{\mathrm{1}}}  \mGLsym{,}  \mGLnt{X}  \mGLsym{,}  \mGLnt{X}  \mGLsym{,}  \Delta_{{\mathrm{2}}}  \mGLsym{;}  \Gamma  \vdash_{\mathsf{MS} }  \mGLnt{l}  \mGLsym{:}  \mGLnt{B}}
\quad
\mGLnt{s_{{\mathrm{1}}}}  \leq  \mGLnt{s_{{\mathrm{2}}}}
}
{\delta_{{\mathrm{1}}}  \mGLsym{,}  \mGLnt{r_{{\mathrm{2}}}}  \mGLsym{,}  \mGLnt{s_{{\mathrm{2}}}}  \mGLsym{,}  \delta_{{\mathrm{2}}}  \odot  \Delta_{{\mathrm{1}}}  \mGLsym{,}  \mGLnt{X}  \mGLsym{,}  \mGLnt{X}  \mGLsym{,}  \Delta_{{\mathrm{2}}}  \mGLsym{;}  \Gamma  \vdash_{\mathsf{MS} }  \mGLnt{l}  \mGLsym{:}  \mGLnt{B}}
}
{\delta_{{\mathrm{1}}}  \mGLsym{,}  \mGLnt{r_{{\mathrm{2}}}}  +  \mGLnt{s_{{\mathrm{2}}}}  \mGLsym{,}  \delta_{{\mathrm{2}}}  \odot  \Delta_{{\mathrm{1}}}  \mGLsym{,}  \mGLnt{X}  \mGLsym{,}  \Delta_{{\mathrm{2}}}  \mGLsym{;}  \Gamma  \vdash_{\mathsf{MS} }  \mGLnt{l}  \mGLsym{:}  \mGLnt{B} }
\equiv
\;
\inferrule*[right=$\mGLdruleMSTXXSubName$]
{
\inferrule*[right=$\mGLdruleMSTXXContName$]
{\delta_{{\mathrm{1}}}  \mGLsym{,}  \mGLnt{r_{{\mathrm{1}}}}  \mGLsym{,}  \mGLnt{s_{{\mathrm{1}}}}  \mGLsym{,}  \delta_{{\mathrm{2}}}  \odot  \Delta_{{\mathrm{1}}}  \mGLsym{,}  \mGLnt{X}  \mGLsym{,}  \mGLnt{X}  \mGLsym{,}  \Delta_{{\mathrm{2}}}  \mGLsym{;}  \Gamma  \vdash_{\mathsf{MS} }  \mGLnt{l}  \mGLsym{:}  \mGLnt{B}}
{\delta_{{\mathrm{1}}}  \mGLsym{,}  \mGLnt{r_{{\mathrm{1}}}}  +  \mGLnt{s_{{\mathrm{1}}}}  \mGLsym{,}  \delta_{{\mathrm{2}}}  \odot  \Delta_{{\mathrm{1}}}  \mGLsym{,}  \mGLnt{X}  \mGLsym{,}  \Delta_{{\mathrm{2}}}  \mGLsym{;}  \Gamma  \vdash_{\mathsf{MS} }  \mGLnt{l}  \mGLsym{:}  \mGLnt{B}}
\quad
\dfrac{\mGLnt{r_{{\mathrm{1}}}}  \leq  \mGLnt{r_{{\mathrm{2}}}} \quad \mGLnt{s_{{\mathrm{1}}}}  \leq  \mGLnt{s_{{\mathrm{2}}}}}{\mGLnt{r_{{\mathrm{1}}}}  +  \mGLnt{s_{{\mathrm{1}}}}  \leq  \mGLnt{r_{{\mathrm{2}}}}  +  \mGLnt{s_{{\mathrm{2}}}}}
}
{\delta_{{\mathrm{1}}}  \mGLsym{,}  \mGLnt{r_{{\mathrm{2}}}}  +  \mGLnt{s_{{\mathrm{2}}}}  \mGLsym{,}  \delta_{{\mathrm{2}}}  \odot  \Delta_{{\mathrm{1}}}  \mGLsym{,}  \mGLnt{X}  \mGLsym{,}  \Delta_{{\mathrm{2}}}  \mGLsym{;}  \Gamma  \vdash_{\mathsf{MS} }  \mGLnt{l}  \mGLsym{:}  \mGLnt{B}}
\tag{\textsc{contr-mono}}
  \end{gather*}

  \begin{gather*}
%
%
\inferrule*[right=$\mGLdruleMSTXXSubName$]
{
\inferrule*[right=$\mGLdruleMSTXXUnitLName$]
{(  \delta_{{\mathrm{1}}}  \mGLsym{,}  \delta_{{\mathrm{2}}}  )   \odot   ( \Delta_{{\mathrm{1}}}  \mGLsym{,}  \Delta_{{\mathrm{2}}} )   \mGLsym{;}  \Gamma  \vdash_{\mathsf{MS} }  \mGLnt{l}  \mGLsym{:}  \mGLnt{A}}
{(  \delta_{{\mathrm{1}}}  \mGLsym{,}  \mGLnt{r}  \mGLsym{,}  \delta_{{\mathrm{2}}}  )   \odot   ( \Delta_{{\mathrm{1}}}  \mGLsym{,}  \mGLmv{x}  \mGLsym{:}  \mathsf{J}  \mGLsym{,}  \Delta_{{\mathrm{2}}} )   \mGLsym{;}  \Gamma  \vdash_{\mathsf{MS} }  \mathsf{let} \, \mathsf{j} \, \mGLsym{=}  \mGLmv{x} \, \mathsf{in} \, \mGLnt{l}  \mGLsym{:}  \mGLnt{A}}
\quad
\mGLnt{r}  \leq  \mGLnt{s}
}
{(  \delta_{{\mathrm{1}}}  \mGLsym{,}  \mGLnt{s}  \mGLsym{,}  \delta_{{\mathrm{2}}}  )   \odot   ( \Delta_{{\mathrm{1}}}  \mGLsym{,}  \mGLmv{x}  \mGLsym{:}  \mathsf{J}  \mGLsym{,}  \Delta_{{\mathrm{2}}} )   \mGLsym{;}  \Gamma  \vdash_{\mathsf{MS} }  \mathsf{let} \, \mathsf{j} \, \mGLsym{=}  \mGLmv{x} \, \mathsf{in} \, \mGLnt{l}  \mGLsym{:}  \mGLnt{A}}
\equiv
\inferrule*[right=$\mGLdruleMSTXXUnitLName$]
{(  \delta_{{\mathrm{1}}}  \mGLsym{,}  \delta_{{\mathrm{2}}}  )   \odot   ( \Delta_{{\mathrm{1}}}  \mGLsym{,}  \Delta_{{\mathrm{2}}} )   \mGLsym{;}  \Gamma  \vdash_{\mathsf{MS} }  \mGLnt{l}  \mGLsym{:}  \mGLnt{A}}
{(  \delta_{{\mathrm{1}}}  \mGLsym{,}  \mGLnt{s}  \mGLsym{,}  \delta_{{\mathrm{2}}}  )   \odot   ( \Delta_{{\mathrm{1}}}  \mGLsym{,}  \mGLmv{x}  \mGLsym{:}  \mathsf{J}  \mGLsym{,}  \Delta_{{\mathrm{2}}} )   \mGLsym{;}  \Gamma  \vdash_{\mathsf{MS} }  \mathsf{let} \, \mathsf{j} \, \mGLsym{=}  \mGLmv{x} \, \mathsf{in} \, \mGLnt{l}  \mGLsym{:}  \mGLnt{A}}
\tag{\textsc{sub-unitL}}
  \end{gather*}

  \begin{gather*}
%
%
\inferrule*[right=$\mGLdruleMSTXXSubName$]
{
\inferrule*[right=$\mGLdruleMSTXXTenLName$]
{(  \delta_{{\mathrm{1}}}  \mGLsym{,}  \mGLnt{r}  \mGLsym{,}  \mGLnt{r}  \mGLsym{,}  \delta_{{\mathrm{2}}}  )   \odot   ( \Delta_{{\mathrm{1}}}  \mGLsym{,}  \mGLmv{x}  \mGLsym{:}  \mGLnt{X}  \mGLsym{,}  \mGLmv{y}  \mGLsym{:}  \mGLnt{Y}  \mGLsym{,}  \Delta_{{\mathrm{2}}} )   \mGLsym{;}  \Gamma  \vdash_{\mathsf{MS} }  \mGLnt{l}  \mGLsym{:}  \mGLnt{B}}
{(  \delta_{{\mathrm{1}}}  \mGLsym{,}  \mGLnt{r}  \mGLsym{,}  \delta_{{\mathrm{2}}}  )   \odot   ( \Delta_{{\mathrm{1}}}  \mGLsym{,}  \mGLmv{z}  \mGLsym{:}  \mGLnt{X}  \boxtimes  \mGLnt{Y}  \mGLsym{,}  \Delta_{{\mathrm{2}}} )   \mGLsym{;}  \Gamma  \vdash_{\mathsf{MS} }  \mathsf{let} \, \mGLsym{(}  \mGLmv{x}  \mGLsym{,}  \mGLmv{y}  \mGLsym{)}  \mGLsym{=}  \mGLmv{z} \, \mathsf{in} \, \mGLnt{l}  \mGLsym{:}  \mGLnt{B}}
\quad
\mGLnt{r}  \leq  \mGLnt{s}
}
{(  \delta_{{\mathrm{1}}}  \mGLsym{,}  \mGLnt{s}  \mGLsym{,}  \delta_{{\mathrm{2}}}  )   \odot   ( \Delta_{{\mathrm{1}}}  \mGLsym{,}  \mGLmv{z}  \mGLsym{:}  \mGLnt{X}  \boxtimes  \mGLnt{Y}  \mGLsym{,}  \Delta_{{\mathrm{2}}} )   \mGLsym{;}  \Gamma  \vdash_{\mathsf{MS} }  \mathsf{let} \, \mGLsym{(}  \mGLmv{x}  \mGLsym{,}  \mGLmv{y}  \mGLsym{)}  \mGLsym{=}  \mGLmv{z} \, \mathsf{in} \, \mGLnt{l}  \mGLsym{:}  \mGLnt{B} }
\equiv
\inferrule*[right=$\mGLdruleMSTXXTenLName$]
{
\inferrule*[right=$\mGLdruleMSTXXSubName$]
{
  \inferrule*[right=$\mGLdruleMSTXXSubName$]
  {
    (  \delta_{{\mathrm{1}}}  \mGLsym{,}  \mGLnt{r}  \mGLsym{,}  \mGLnt{r}  \mGLsym{,}  \delta_{{\mathrm{2}}}  )   \odot   ( \Delta_{{\mathrm{1}}}  \mGLsym{,}  \mGLmv{x}  \mGLsym{:}  \mGLnt{X}  \mGLsym{,}  \mGLmv{y}  \mGLsym{:}  \mGLnt{Y}  \mGLsym{,}  \Delta_{{\mathrm{2}}} )   \mGLsym{;}  \Gamma  \vdash_{\mathsf{MS} }  \mGLnt{l}  \mGLsym{:}  \mGLnt{B} \quad \mGLnt{r}  \leq  \mGLnt{s}
  }
  {(  \delta_{{\mathrm{1}}}  \mGLsym{,}  \mGLnt{s}  \mGLsym{,}  \mGLnt{r}  \mGLsym{,}  \delta_{{\mathrm{2}}}  )   \odot   ( \Delta_{{\mathrm{1}}}  \mGLsym{,}  \mGLmv{x}  \mGLsym{:}  \mGLnt{X}  \mGLsym{,}  \mGLmv{y}  \mGLsym{:}  \mGLnt{Y}  \mGLsym{,}  \Delta_{{\mathrm{2}}} )   \mGLsym{;}  \Gamma  \vdash_{\mathsf{MS} }  \mGLnt{l}  \mGLsym{:}  \mGLnt{B}}
  \quad
  \mGLnt{r}  \leq  \mGLnt{s}
}
{(  \delta_{{\mathrm{1}}}  \mGLsym{,}  \mGLnt{s}  \mGLsym{,}  \mGLnt{s}  \mGLsym{,}  \delta_{{\mathrm{2}}}  )   \odot   ( \Delta_{{\mathrm{1}}}  \mGLsym{,}  \mGLmv{x}  \mGLsym{:}  \mGLnt{X}  \mGLsym{,}  \mGLmv{y}  \mGLsym{:}  \mGLnt{Y}  \mGLsym{,}  \Delta_{{\mathrm{2}}} )   \mGLsym{;}  \Gamma  \vdash_{\mathsf{MS} }  \mGLnt{l}  \mGLsym{:}  \mGLnt{B}}
}
{(  \delta_{{\mathrm{1}}}  \mGLsym{,}  \mGLnt{s}  \mGLsym{,}  \delta_{{\mathrm{2}}}  )   \odot   ( \Delta_{{\mathrm{1}}}  \mGLsym{,}  \mGLmv{x}  \mGLsym{:}  \mathsf{J}  \mGLsym{,}  \Delta_{{\mathrm{2}}} )   \mGLsym{;}  \Gamma  \vdash_{\mathsf{MS} }  \mathsf{let} \, \mathsf{j} \, \mGLsym{=}  \mGLmv{x} \, \mathsf{in} \, \mGLnt{l}  \mGLsym{:}  \mGLnt{B}}
\tag{\textsc{sub-tensorL}}
  \end{gather*}

      \begin{gather*}
%
%
\inferrule*[right=$\mGLdruleMSTXXCutName$]
 {
  \inferrule*[right=$\mGLdruleGSTXXSubName$]
             {\delta_{{\mathrm{2}}}  \odot  \Delta_{{\mathrm{2}}}  \vdash_{\mathsf{GS} }  \mGLnt{t}  \mGLsym{:}  \mGLnt{X} \quad
               \delta_{{\mathrm{2}}}  \leq  \delta'_{{\mathrm{2}}}}
             {\delta'_{{\mathrm{2}}}  \odot  \Delta_{{\mathrm{2}}}  \vdash_{\mathsf{GS} }  \mGLnt{t}  \mGLsym{:}  \mGLnt{X}}
               \\
  \inferrule*[right=$\mGLdruleMSTXXSubName$]
            {(  \delta_{{\mathrm{1}}}  \mGLsym{,}  \mGLnt{r}  \mGLsym{,}  \delta_{{\mathrm{3}}}  )   \odot   ( \Delta_{{\mathrm{1}}}  \mGLsym{,}  \mGLmv{x}  \mGLsym{:}  \mGLnt{X}  \mGLsym{,}  \Delta_{{\mathrm{3}}} )   \mGLsym{;}  \Gamma  \vdash_{\mathsf{MS} }  \mGLnt{l}  \mGLsym{:}  \mGLnt{B} \quad
              \mGLnt{r}  \leq  \mGLnt{r'}}
            {(  \delta_{{\mathrm{1}}}  \mGLsym{,}  \mGLnt{r'}  \mGLsym{,}  \delta_{{\mathrm{3}}}  )   \odot   ( \Delta_{{\mathrm{1}}}  \mGLsym{,}  \mGLmv{x}  \mGLsym{:}  \mGLnt{X}  \mGLsym{,}  \Delta_{{\mathrm{3}}} )   \mGLsym{;}  \Gamma  \vdash_{\mathsf{MS} }  \mGLnt{l}  \mGLsym{:}  \mGLnt{B}}
 }
 {(  \delta_{{\mathrm{1}}}  \mGLsym{,}   \mGLnt{r'}   *   \delta'_{{\mathrm{2}}}   \mGLsym{,}  \delta_{{\mathrm{3}}}  )   \odot   ( \Delta_{{\mathrm{1}}}  \mGLsym{,}  \Delta_{{\mathrm{2}}}  \mGLsym{,}  \Delta_{{\mathrm{3}}} )   \mGLsym{;}  \Gamma  \vdash_{\mathsf{MS} }  \mGLsym{[}  \mGLnt{t}  \mGLsym{/}  \mGLmv{x}  \mGLsym{]}  \mGLnt{l}  \mGLsym{:}  \mGLnt{B} }
\equiv
\inferrule*[right=$\mGLdruleGSTXXSubName$]
{
\inferrule*[right=$\mGLdruleMSTXXCutName$]
 {\delta_{{\mathrm{2}}}  \odot  \Delta_{{\mathrm{2}}}  \vdash_{\mathsf{GS} }  \mGLnt{t}  \mGLsym{:}  \mGLnt{X} \quad
  (  \delta_{{\mathrm{1}}}  \mGLsym{,}  \mGLnt{r}  \mGLsym{,}  \delta_{{\mathrm{3}}}  )   \odot   ( \Delta_{{\mathrm{1}}}  \mGLsym{,}  \mGLmv{x}  \mGLsym{:}  \mGLnt{X}  \mGLsym{,}  \Delta_{{\mathrm{3}}} )   \mGLsym{;}  \Gamma  \vdash_{\mathsf{MS} }  \mGLnt{l}  \mGLsym{:}  \mGLnt{B} }
 { (  \delta_{{\mathrm{1}}}  \mGLsym{,}   \mGLnt{r}   *   \delta_{{\mathrm{2}}}   \mGLsym{,}  \delta_{{\mathrm{3}}}  )   \odot   ( \Delta_{{\mathrm{1}}}  \mGLsym{,}  \Delta_{{\mathrm{2}}}  \mGLsym{,}  \Delta_{{\mathrm{3}}} )   \mGLsym{;}  \Gamma  \vdash_{\mathsf{MS} }  \mGLsym{[}  \mGLnt{t}  \mGLsym{/}  \mGLmv{x}  \mGLsym{]}  \mGLnt{l}  \mGLsym{:}  \mGLnt{B} }
 \dfrac{\mGLnt{r}  \leq  \mGLnt{r'} \quad \delta_{{\mathrm{2}}}  \leq  \delta'_{{\mathrm{2}}}}
       {\mGLnt{r}  *  \delta_{{\mathrm{2}}}  \leq  \mGLnt{r'}  *  \delta'_{{\mathrm{2}}}}
}
{
(  \delta_{{\mathrm{1}}}  \mGLsym{,}   \mGLnt{r'}   *   \delta'_{{\mathrm{2}}}   \mGLsym{,}  \delta_{{\mathrm{3}}}  )   \odot   ( \Delta_{{\mathrm{1}}}  \mGLsym{,}  \Delta_{{\mathrm{2}}}  \mGLsym{,}  \Delta_{{\mathrm{3}}} )   \mGLsym{;}  \Gamma  \vdash_{\mathsf{MS} }  \mGLsym{[}  \mGLnt{t}  \mGLsym{/}  \mGLmv{x}  \mGLsym{]}  \mGLnt{l}  \mGLsym{:}  \mGLnt{B}
}
\tag{\textsc{mult-mono}}
  \end{gather*}

%% file: mGL-nat-term-ottput.tex
\begin{mdframed}
    \drules[GT]{$\delta  \odot  \Delta  \vdash_{\mathsf{GT} }  \mGLnt{t}  \mGLsym{:}  \mGLnt{X}$}{Graded System}{
      Id,UnitI,UnitE,TenI,TenE,LinI
    }
    \end{mdframed}

    \begin{mdframed}
    \drules[GT]{$\delta  \odot  \Delta  \vdash_{\mathsf{GT} }  \mGLnt{t}  \mGLsym{:}  \mGLnt{X}$}{Graded System (continued)}{
      Weak,Cont,Ex,Sub
    }
    \end{mdframed}

    \begin{mdframed}
    \drules[MT]{$\delta  \odot  \Delta  \mGLsym{;}  \Gamma  \vdash_{\mathsf{MT} }  \mGLnt{l}  \mGLsym{:}  \mGLnt{A}$}{Mixed System}{
      Id,GSub,UnitI,UnitE,TenI, TenE,ImpI,ImpE
    }
    \end{mdframed}

    \begin{mdframed}
    \drules[MT]{$\delta  \odot  \Delta  \mGLsym{;}  \Gamma  \vdash_{\mathsf{MT} }  \mGLnt{l}  \mGLsym{:}  \mGLnt{A}$}{Mixed System (continued)}{
     GrdI,LinE,GrdE,Weak,Cont,Ex,GEx
    }
\end{mdframed}

%% file: cut-elim-mGL-ottput.tex
Cut reduction done termless for simplicity.

\begin{lemma}[Cut Reduction GS]
      If $\Pi_{{\mathrm{1}}}$ is a proof of $\delta_{{\mathrm{2}}}  \odot  \Delta_{{\mathrm{2}}}  \vdash_{\mathsf{GS} }  \mGLnt{X}$ 
      and $\Pi_{{\mathrm{2}}}$ is a proof of $(  \delta_{{\mathrm{1}}}  \mGLsym{,}  \delta  \mGLsym{,}  \delta_{{\mathrm{3}}}  )   \odot   ( \Delta_{{\mathrm{1}}}  \mGLsym{,}   \mGLnt{X} ^{ \mGLmv{n} }   \mGLsym{,}  \Delta_{{\mathrm{3}}} )   \vdash_{\mathsf{GS} }  \mGLnt{Y}$ 
      with $\mathsf{CutRank} \, \mGLsym{(}  \Pi_{{\mathrm{1}}}  \mGLsym{)} \, \leq \,  \mathsf{Rank}  (  \mGLnt{X}  )$
      and $\mathsf{CutRank} \, \mGLsym{(}  \Pi_{{\mathrm{2}}}  \mGLsym{)} \, \leq \,  \mathsf{Rank}  (  \mGLnt{X}  )$,
      then there exists a proof $\Pi$ of the sequent $(  \delta_{{\mathrm{1}}}  \mGLsym{,}   (   \delta  \boxast [ { \delta_{{\mathrm{2}}} }^{ \mGLmv{n} } ]   )   \mGLsym{,}  \delta_{{\mathrm{3}}}  )   \odot   ( \Delta_{{\mathrm{1}}}  \mGLsym{,}  \Delta_{{\mathrm{2}}}  \mGLsym{,}  \Delta_{{\mathrm{3}}} )   \vdash_{\mathsf{GS} }  \mGLnt{Y}$ 
      with $\mathsf{CutRank} \, \mGLsym{(}  \Pi  \mGLsym{)} \, \leq \,  \mathsf{Rank}  (  \mGLnt{X}  )$.

    \end{lemma}

\begin{proof}
        This is by induction on $\mathsf{Depth}  (  \Pi_{{\mathrm{1}}}  )   +   \mathsf{Depth}  (  \Pi_{{\mathrm{2}}}  )$.
  \begin{enumerate}

  \item \textbf{Commuting Conversions}
    \begin{enumerate}
      \item \textbf{left-side:} Suppose we have

      \[
        \inferrule* [flushleft,right=,left=$\Pi_{{\mathrm{1}}} :$] {
          \pi_1
        }{\delta_{{\mathrm{3}}}  \odot  \Delta_{{\mathrm{3}}}  \vdash_{\mathsf{GS} }  \mGLnt{X}}
       \]
       \[
        \inferrule* [flushleft,right=$\mGLdruleGSTXXCutName{}$,left=$\Pi_{{\mathrm{2}}} :$] {
          \inferrule* [flushleft,right=,left=$\Pi_{{\mathrm{3}}} :$] {
            \pi_3
          }{(  \delta_{{\mathrm{2}}}  \mGLsym{,}  \gamma_{{\mathrm{3}}}  \mGLsym{,}  \delta_{{\mathrm{4}}}  )   \odot   ( \Delta_{{\mathrm{2}}}  \mGLsym{,}   \mGLnt{X} ^{ \mGLmv{n} }   \mGLsym{,}  \Delta_{{\mathrm{4}}} )   \vdash_{\mathsf{GS} }  \mGLnt{Y}}\\
          \inferrule* [flushleft,right=,left=$\Pi_{{\mathrm{4}}} :$] {
            \pi_4
          }{(  \delta_{{\mathrm{1}}}  \mGLsym{,}  \delta  \mGLsym{,}  \delta_{{\mathrm{5}}}  )   \odot   ( \Delta_{{\mathrm{1}}}  \mGLsym{,}   \mGLnt{Y} ^{ \mGLmv{m} }   \mGLsym{,}  \Delta_{{\mathrm{5}}} )   \vdash_{\mathsf{GS} }  \mGLnt{Z}}
        }{(  \delta_{{\mathrm{1}}}  \mGLsym{,}   (   \delta  \boxast [ {  (  \delta_{{\mathrm{2}}}  )  }^{ \mGLmv{m} } ]   )   \mGLsym{,}   (   \delta  \boxast [ {  (  \gamma_{{\mathrm{3}}}  )  }^{ \mGLmv{m} } ]   )   \mGLsym{,}   (   \delta  \boxast [ {  (  \delta_{{\mathrm{4}}}  )  }^{ \mGLmv{m} } ]   )   \mGLsym{,}  \delta_{{\mathrm{5}}}  )   \odot   ( \Delta_{{\mathrm{1}}}  \mGLsym{,}  \Delta_{{\mathrm{2}}}  \mGLsym{,}   \mGLnt{X} ^{ \mGLmv{n} }   \mGLsym{,}  \Delta_{{\mathrm{4}}}  \mGLsym{,}  \Delta_{{\mathrm{5}}} )   \vdash_{\mathsf{GS} }  \mGLnt{Z}}
      \]

      We know:
      \[
      \begin{array}{lll}
        \mathsf{Depth}  (  \Pi_{{\mathrm{1}}}  )   +   \mathsf{Depth}  (  \Pi_{{\mathrm{3}}}  )   \, \mGLsym{<} \,  \mathsf{Depth}  (  \Pi_{{\mathrm{1}}}  )   +   \mathsf{Depth}  (  \Pi_{{\mathrm{2}}}  )\\
        \mathsf{CutRank} \, \mGLsym{(}  \Pi_{{\mathrm{3}}}  \mGLsym{)} \, \leq \, \mGLkw{Max} \, \mGLsym{(}   \mathsf{CutRank} \, \mGLsym{(}  \Pi_{{\mathrm{3}}}  \mGLsym{)}  \mGLsym{,}  \mathsf{CutRank} \, \mGLsym{(}  \Pi_{{\mathrm{4}}}  \mGLsym{)}  \mGLsym{,}   \mathsf{Rank}  (  \mGLnt{Y}  )   + 1   \mGLsym{)} \, \leq \,  \mathsf{Rank}  (  \mGLnt{X}  )
      \end{array}
      \]

      and so applying the induction hypothesis
      to $\Pi_{{\mathrm{1}}}$ and $\Pi_{{\mathrm{3}}}$
      implies that there is a proof $\Pi'$ of
      $(  \delta_{{\mathrm{2}}}  \mGLsym{,}  \gamma  \mGLsym{,}  \delta_{{\mathrm{4}}}  )   \odot   ( \Delta_{{\mathrm{2}}}  \mGLsym{,}  \Delta_{{\mathrm{3}}}  \mGLsym{,}  \Delta_{{\mathrm{4}}} )   \vdash_{\mathsf{GS} }  \mGLnt{Y}$ with
      $\mathsf{CutRank} \, \mGLsym{(}  \Pi'  \mGLsym{)} \, \leq \,  \mathsf{Rank}  (  \mGLnt{X}  )$
      and $(   \gamma_{{\mathrm{3}}}  \boxast [ {  (  \delta_{{\mathrm{3}}}  )  }^{ \mGLmv{n} } ]   )   \mGLsym{=}  \gamma$.
      Thus, we construct the following proof $\Pi$:
      \[
      \inferrule* [flushleft,right=$\mGLdruleGSTXXCutName{}$,left=$\Pi :$] {
        \inferrule* [flushleft,left=$\Pi' : $] {
          \pi
        }{(  \delta_{{\mathrm{2}}}  \mGLsym{,}  \gamma  \mGLsym{,}  \delta_{{\mathrm{4}}}  )   \odot   ( \Delta_{{\mathrm{2}}}  \mGLsym{,}  \Delta_{{\mathrm{3}}}  \mGLsym{,}  \Delta_{{\mathrm{4}}} )   \vdash_{\mathsf{GS} }  \mGLnt{Y}}\\
        \inferrule* [flushleft,right=,left=$\Pi_{{\mathrm{4}}} :$] {
          \pi_4
        }{(  \delta_{{\mathrm{1}}}  \mGLsym{,}  \delta  \mGLsym{,}  \delta_{{\mathrm{5}}}  )   \odot   ( \Delta_{{\mathrm{1}}}  \mGLsym{,}   \mGLnt{Y} ^{ \mGLmv{n} }   \mGLsym{,}  \Delta_{{\mathrm{5}}} )   \vdash_{\mathsf{GS} }  \mGLnt{Z}}
      }{(  \delta_{{\mathrm{1}}}  \mGLsym{,}   (   \delta  \boxast [ {  (  \delta_{{\mathrm{2}}}  )  }^{ \mGLmv{n} } ]   )   \mGLsym{,}   (   \delta  \boxast [ { \gamma }^{ \mGLmv{n} } ]   )   \mGLsym{,}   (   \delta  \boxast [ {  (  \delta_{{\mathrm{4}}}  )  }^{ \mGLmv{n} } ]   )   \mGLsym{,}  \delta_{{\mathrm{5}}}  )   \odot   ( \Delta_{{\mathrm{1}}}  \mGLsym{,}  \Delta_{{\mathrm{2}}}  \mGLsym{,}  \Delta_{{\mathrm{3}}}  \mGLsym{,}  \Delta_{{\mathrm{4}}}  \mGLsym{,}  \Delta_{{\mathrm{5}}} )   \vdash_{\mathsf{GS} }  \mGLnt{Z}}
      \]
      Given the above we know:
      \[
        \begin{array}{lll}
          (   \gamma_{{\mathrm{3}}}  \boxast [ {  (  \delta_{{\mathrm{3}}}  )  }^{ \mGLmv{n} } ]   )   \mGLsym{=}  \gamma\\
          (   \gamma_{{\mathrm{3}}}  \boxast [ {  (  \delta_{{\mathrm{3}}}  )  }^{ \mGLmv{n} } ]   )   \mGLsym{(}  \mGLmv{j}  \mGLsym{)}  \mGLsym{=}  \gamma  \mGLsym{(}  \mGLmv{j}  \mGLsym{)} \text{  for any } j, 1<= j <= n\\
          \delta  \mGLsym{(}  \mGLmv{k}  \mGLsym{)}  \mGLsym{=}  \delta  \mGLsym{(}  \mGLmv{k}  \mGLsym{)} \text{  for any } k, 1<= k <= m
        \end{array}
        \]
      From this and the monotonicity properties of the operators in the resource
      algebra it follows that:
      \[
        \begin{array}{lll}
          (   \bigoplus ^{ \mGLmv{m} }_{ \mGLmv{k}  = 1}   (  \delta  \mGLsym{(}  \mGLmv{k}  \mGLsym{)}  *   (   \gamma_{{\mathrm{3}}}  \boxast [ {  (  \delta_{{\mathrm{3}}}  )  }^{ \mGLmv{n} } ]   )   )    )   \mGLsym{=}   (   \bigoplus ^{ \mGLmv{m} }_{ \mGLmv{k}  = 1}   (  \delta  \mGLsym{(}  \mGLmv{k}  \mGLsym{)}  *  \gamma  \mGLsym{(}  \mGLmv{j}  \mGLsym{)}  )    ) \\
          \delta  \boxast [ {  (   \gamma_{{\mathrm{3}}}  \boxast [ {  (  \delta_{{\mathrm{3}}}  )  }^{ \mGLmv{n} } ]   )  }^{ \mGLmv{m} } ]   \mGLsym{=}   \delta  \boxast [ { \gamma }^{ \mGLmv{m} } ]\\
          (    (   \delta  \boxast [ {  (  \gamma_{{\mathrm{3}}}  )  }^{ \mGLmv{m} } ]   )   \boxast [ {  (  \delta_{{\mathrm{3}}}  )  }^{ \mGLmv{n} } ]   )   \mGLsym{=}   (   \delta  \boxast [ { \gamma }^{ \mGLmv{m} } ]   )
        \end{array}
        \]
        and
      \[
      \begin{array}{lll}
        \mathsf{CutRank} \, \mGLsym{(}  \Pi'  \mGLsym{)} \, \leq \,  \mathsf{Rank}  (  \mGLnt{X}  )\\
        \mathsf{CutRank} \, \mGLsym{(}  \Pi_{{\mathrm{2}}}  \mGLsym{)} \, \mGLsym{=} \, \mGLkw{Max} \, \mGLsym{(}   \mathsf{CutRank} \, \mGLsym{(}  \Pi_{{\mathrm{3}}}  \mGLsym{)}  \mGLsym{,}  \mathsf{CutRank} \, \mGLsym{(}  \Pi_{{\mathrm{4}}}  \mGLsym{)}  \mGLsym{,}   \mathsf{Rank}  (  \mGLnt{Y}  )   + 1   \mGLsym{)} \, \leq \,  \mathsf{Rank}  (  \mGLnt{X}  )\\
      \end{array}
      \]
      This implies:
      \[
      \begin{array}{lll}
        \mathsf{CutRank} \, \mGLsym{(}  \Pi_{{\mathrm{4}}}  \mGLsym{)} \, \leq \,  \mathsf{Rank}  (  \mGLnt{X}  )\\
        \mathsf{Rank}  (  \mGLnt{Y}  )   + 1  \, \leq \,  \mathsf{Rank}  (  \mGLnt{X}  )
      \end{array}
      \]
      Thus, we obtain our result
      \[
        \mathsf{CutRank} \, \mGLsym{(}  \Pi  \mGLsym{)} \, \mGLsym{=} \, \mGLkw{Max} \, \mGLsym{(}   \mathsf{CutRank} \, \mGLsym{(}  \Pi'  \mGLsym{)}  \mGLsym{,}  \mathsf{CutRank} \, \mGLsym{(}  \Pi_{{\mathrm{4}}}  \mGLsym{)}  \mGLsym{,}   \mathsf{Rank}  (  \mGLnt{Y}  )   + 1   \mGLsym{)} \, \leq \,  \mathsf{Rank}  (  \mGLnt{X}  )
      \]

\item \textbf{cut vs. right-side cut (left case):} Suppose we have:

    \[
    \inferrule* [flushleft,right=,left=$\Pi_{{\mathrm{1}}} :$] {
      \pi_1
    }{\delta_{{\mathrm{2}}}  \odot  \Delta_{{\mathrm{2}}}  \vdash_{\mathsf{GS} }  \mGLnt{X}}
    \]
    \[
    \inferrule* [flushleft,right=$\mGLdruleGSTXXCutName{}$,left=$\Pi_{{\mathrm{2}}} :$] {
      \inferrule* [flushleft,right=,left=$\Pi_{{\mathrm{3}}} :$] {
        \pi_3
      }{\delta_{{\mathrm{4}}}  \odot  \Delta_{{\mathrm{4}}}  \vdash_{\mathsf{GS} }  \mGLnt{Y}}\\
      \inferrule* [flushleft,right=,left=$\Pi_{{\mathrm{4}}} :$] {
        \pi_4
      }{(  \delta_{{\mathrm{1}}}  \mGLsym{,}  \delta  \mGLsym{,}  \delta_{{\mathrm{4}}}  \mGLsym{,}  \delta'  \mGLsym{,}  \delta_{{\mathrm{5}}}  )   \odot   ( \Delta_{{\mathrm{1}}}  \mGLsym{,}   \mGLnt{X} ^{ \mGLmv{n} }   \mGLsym{,}  \Delta_{{\mathrm{3}}}  \mGLsym{,}   \mGLnt{Y} ^{ \mGLmv{m} }   \mGLsym{,}  \Delta_{{\mathrm{5}}} )   \vdash_{\mathsf{GS} }  \mGLnt{Z}}
    }{(  \delta_{{\mathrm{1}}}  \mGLsym{,}  \delta  \mGLsym{,}  \delta_{{\mathrm{3}}}  \mGLsym{,}   (   \delta'  \boxast [ { \delta_{{\mathrm{4}}} }^{ \mGLmv{m} } ]   )   \mGLsym{,}  \delta_{{\mathrm{5}}}  )   \odot   ( \Delta_{{\mathrm{1}}}  \mGLsym{,}   \mGLnt{X} ^{ \mGLmv{n} }   \mGLsym{,}  \Delta_{{\mathrm{3}}}  \mGLsym{,}  \Delta_{{\mathrm{4}}}  \mGLsym{,}  \Delta_{{\mathrm{5}}} )   \vdash_{\mathsf{GS} }  \mGLnt{Z}}
    \]
We know:
\[
\begin{array}{lll}
  \mathsf{Depth}  (  \Pi_{{\mathrm{1}}}  )   +   \mathsf{Depth}  (  \Pi_{{\mathrm{4}}}  )   \, \mGLsym{<} \,  \mathsf{Depth}  (  \Pi_{{\mathrm{1}}}  )   +   \mathsf{Depth}  (  \Pi_{{\mathrm{2}}}  )\\
  \mathsf{CutRank} \, \mGLsym{(}  \Pi_{{\mathrm{4}}}  \mGLsym{)} \, \leq \, \mGLkw{Max} \, \mGLsym{(}   \mathsf{CutRank} \, \mGLsym{(}  \Pi_{{\mathrm{3}}}  \mGLsym{)}  \mGLsym{,}  \mathsf{CutRank} \, \mGLsym{(}  \Pi_{{\mathrm{4}}}  \mGLsym{)}  \mGLsym{,}   \mathsf{Rank}  (  \mGLnt{Y}  )   + 1   \mGLsym{)} \, \leq \,  \mathsf{Rank}  (  \mGLnt{X}  )
\end{array}
\]
and so applying the induction hypothesis
to $\Pi_{{\mathrm{1}}}$ and $\Pi_{{\mathrm{4}}}$
implies that there is a proof $\Pi'$ of
$(  \delta_{{\mathrm{1}}}  \mGLsym{,}  \gamma  \mGLsym{,}  \delta_{{\mathrm{4}}}  \mGLsym{,}  \delta'  \mGLsym{,}  \delta_{{\mathrm{5}}}  )   \odot   ( \Delta_{{\mathrm{1}}}  \mGLsym{,}  \Delta_{{\mathrm{2}}}  \mGLsym{,}  \Delta_{{\mathrm{4}}}  \mGLsym{,}   \mGLnt{Y} ^{ \mGLmv{m} }   \mGLsym{,}  \Delta_{{\mathrm{5}}} )   \vdash_{\mathsf{GS} }  \mGLnt{Z}$ with
$\mathsf{CutRank} \, \mGLsym{(}  \Pi'  \mGLsym{)} \, \leq \,  \mathsf{Rank}  (  \mGLnt{X}  )$
and $(   \delta  \boxast [ {  (  \delta_{{\mathrm{2}}}  )  }^{ \mGLmv{n} } ]   )   \mGLsym{=}  \gamma$.
Thus, we construct the following proof $\Pi$:

\[
\inferrule* [flushleft,right=$\mGLdruleGSTXXCutName{}$,left=$\Pi :$] {
  \inferrule* [flushleft,left=$\Pi_{{\mathrm{3}}} : $] {
    \pi_3
  }{\delta_{{\mathrm{4}}}  \odot  \Delta_{{\mathrm{4}}}  \vdash_{\mathsf{GS} }  \mGLnt{Y}}\\
  \inferrule* [flushleft,right=,left=$\Pi' :$] {
    \pi'
  }{(  \delta_{{\mathrm{1}}}  \mGLsym{,}  \gamma  \mGLsym{,}  \delta_{{\mathrm{4}}}  \mGLsym{,}  \delta'  \mGLsym{,}  \delta_{{\mathrm{5}}}  )   \odot   ( \Delta_{{\mathrm{1}}}  \mGLsym{,}  \Delta_{{\mathrm{2}}}  \mGLsym{,}  \Delta_{{\mathrm{4}}}  \mGLsym{,}   \mGLnt{Y} ^{ \mGLmv{n} }   \mGLsym{,}  \Delta_{{\mathrm{5}}} )   \vdash_{\mathsf{GS} }  \mGLnt{Z}}
}{(  \delta_{{\mathrm{1}}}  \mGLsym{,}  \gamma  \mGLsym{,}  \delta_{{\mathrm{3}}}  \mGLsym{,}   (   \delta'  \boxast [ { \delta_{{\mathrm{4}}} }^{ \mGLmv{m} } ]   )   \mGLsym{,}  \delta_{{\mathrm{5}}}  )   \odot   ( \Delta_{{\mathrm{1}}}  \mGLsym{,}  \Delta_{{\mathrm{2}}}  \mGLsym{,}  \Delta_{{\mathrm{3}}}  \mGLsym{,}  \Delta_{{\mathrm{4}}}  \mGLsym{,}  \Delta_{{\mathrm{5}}} )   \vdash_{\mathsf{GS} }  \mGLnt{Z}}
\]
Given the above we know:
\[
\begin{array}{lll}
  (   \delta  \boxast [ {  (  \delta_{{\mathrm{2}}}  )  }^{ \mGLmv{n} } ]   )   \mGLsym{=}  \gamma\\
  \mathsf{CutRank} \, \mGLsym{(}  \Pi'  \mGLsym{)} \, \leq \,  \mathsf{Rank}  (  \mGLnt{X}  )\\
  \mathsf{CutRank} \, \mGLsym{(}  \Pi_{{\mathrm{2}}}  \mGLsym{)} \, \mGLsym{=} \, \mGLkw{Max} \, \mGLsym{(}   \mathsf{CutRank} \, \mGLsym{(}  \Pi_{{\mathrm{3}}}  \mGLsym{)}  \mGLsym{,}  \mathsf{CutRank} \, \mGLsym{(}  \Pi_{{\mathrm{4}}}  \mGLsym{)}  \mGLsym{,}   \mathsf{Rank}  (  \mGLnt{Y}  )   + 1   \mGLsym{)} \, \leq \,  \mathsf{Rank}  (  \mGLnt{X}  )\\
\end{array}
\]
This implies:
\[
\begin{array}{lll}
  \mathsf{CutRank} \, \mGLsym{(}  \Pi_{{\mathrm{3}}}  \mGLsym{)} \, \leq \,  \mathsf{Rank}  (  \mGLnt{X}  )\\
  \mathsf{Rank}  (  \mGLnt{Y}  )   + 1  \, \leq \,  \mathsf{Rank}  (  \mGLnt{X}  )
\end{array}
\]
Thus, we obtain our result:
\[
  \mathsf{CutRank} \, \mGLsym{(}  \Pi  \mGLsym{)} \, \mGLsym{=} \, \mGLkw{Max} \, \mGLsym{(}   \mathsf{CutRank} \, \mGLsym{(}  \Pi_{{\mathrm{3}}}  \mGLsym{)}  \mGLsym{,}  \mathsf{CutRank} \, \mGLsym{(}  \Pi'  \mGLsym{)}  \mGLsym{,}   \mathsf{Rank}  (  \mGLnt{Y}  )   + 1   \mGLsym{)} \, \leq \,  \mathsf{Rank}  (  \mGLnt{X}  )
  \]

  \item \textbf{cut vs. right-side cut (right case):} Suppose we have:
  \[
        \inferrule* [flushleft,right=,left=$\Pi_{{\mathrm{1}}} :$] {
          \pi_1
        }{\delta_{{\mathrm{2}}}  \odot  \Delta_{{\mathrm{2}}}  \vdash_{\mathsf{GS} }  \mGLnt{X}}
  \]

  \[
        \inferrule* [flushleft,right=$\mGLdruleGSTXXCutName{}$,left=$\Pi_{{\mathrm{2}}} :$] {
          \inferrule* [flushleft,right=,left=$\Pi_{{\mathrm{3}}} :$] {
            \pi_3
          }{\delta_{{\mathrm{3}}}  \odot  \Delta_{{\mathrm{3}}}  \vdash_{\mathsf{GS} }  \mGLnt{Y}}\\
          \inferrule* [flushleft,right=,left=$\Pi_{{\mathrm{4}}} :$] {
            \pi_4
          }{(  \delta_{{\mathrm{1}}}  \mGLsym{,}  \mGLnt{r_{{\mathrm{1}}}}  \mGLsym{,}  \delta_{{\mathrm{4}}}  \mGLsym{,}  \mGLnt{r_{{\mathrm{2}}}}  \mGLsym{,}  \delta_{{\mathrm{5}}}  )   \odot   ( \Delta_{{\mathrm{1}}}  \mGLsym{,}   \mGLnt{Y} ^{ \mGLmv{n} }   \mGLsym{,}  \Delta_{{\mathrm{4}}}  \mGLsym{,}   \mGLnt{X} ^{ \mGLmv{n} }   \mGLsym{,}  \Delta_{{\mathrm{5}}} )   \vdash_{\mathsf{GS} }  \mGLnt{Z}}
        }{(  \delta_{{\mathrm{4}}}  \mGLsym{,}  \mGLnt{r_{{\mathrm{1}}}}  \mGLsym{,}  \delta_{{\mathrm{1}}}  \mGLsym{,}  \delta_{{\mathrm{3}}}  \mGLsym{,}  \delta_{{\mathrm{5}}}  )   \odot   ( \Delta_{{\mathrm{4}}}  \mGLsym{,}   \mGLnt{X} ^{ \mGLmv{n} }   \mGLsym{,}  \Delta_{{\mathrm{1}}}  \mGLsym{,}  \Delta_{{\mathrm{3}}}  \mGLsym{,}  \Delta_{{\mathrm{5}}} )   \vdash_{\mathsf{GS} }  \mGLnt{Z}}
  \]
  Similar to the previous case.

    \end{enumerate}
    \item \textbf{$\eta$-Expansions}
      \begin{enumerate}
        \item \textbf{Tensor:}
        \[
            \begin{array}{lll}
              \inferrule* [flushleft,right=$\mGLdruleGSTXXidName{}$,left=$\Pi_{{\mathrm{1}}} :$] {
                \,
              }{1  \odot  \mGLnt{X}  \boxtimes  \mGLnt{Y}  \vdash_{\mathsf{GS} }  \mGLnt{X}  \boxtimes  \mGLnt{Y}}
              & \quad &
              \inferrule* [flushleft,right=$\mGLdruleGSTXXidName{}$,left=$\Pi_{{\mathrm{2}}} :$] {
                \,
              }{1  \odot  \mGLnt{X}  \boxtimes  \mGLnt{Y}  \vdash_{\mathsf{GS} }  \mGLnt{X}  \boxtimes  \mGLnt{Y}}
            \end{array}
            \]
            Since $\mathsf{Depth}  (  \Pi_{{\mathrm{1}}}  )   +   \mathsf{Depth}  (  \Pi_{{\mathrm{2}}}  )   \, \mGLsym{=} \, 0$ we find ourselves in a base
            case.   $\Pi_{{\mathrm{1}}}$ and $\Pi_{{\mathrm{2}}}$ are cut-free, so we must produce a
            cut-free proof, $\Pi$ with $1  \boxast [ { 1 }^{ \mGLmv{n} } ]   \mGLsym{=}  1$,
            which we do as follows:
            \[
            \inferrule* [flushleft,right=$\mGLdruleGSTXXTenLName{}$,] {
              \inferrule* [flushleft,right=$\mGLdruleGSTXXTenRName{}$,] {
                \inferrule* [flushleft,right=$\mGLdruleGSTXXidName{}$,] {
                  \,
                }{1  \odot  \mGLnt{X}  \vdash_{\mathsf{GS} }  \mGLnt{X}}\\
                \inferrule* [flushleft,right=$\mGLdruleGSTXXidName{}$,] {
                  \,
                }{1  \odot  \mGLnt{Y}  \vdash_{\mathsf{GS} }  \mGLnt{Y}}
              }{(  1  \mGLsym{,}  1  )   \odot   ( \mGLnt{X}  \mGLsym{,}  \mGLnt{Y} )   \vdash_{\mathsf{GS} }  \mGLnt{X}  \boxtimes  \mGLnt{Y}}
            }{1  \odot  \mGLnt{X}  \boxtimes  \mGLnt{Y}  \vdash_{\mathsf{GS} }  \mGLnt{X}  \boxtimes  \mGLnt{Y}}
            \]
        \item \textbf{Tensor Unit:}
        \[
            \begin{array}{lll}
              \inferrule* [flushleft,right=$\mGLdruleGSTXXidName{}$,left=$\Pi_{{\mathrm{1}}} :$] {
                \,
              }{1  \odot  \mathsf{J}  \vdash_{\mathsf{GS} }  \mathsf{J}}
              & \quad &
              \inferrule* [flushleft,right=$\mGLdruleGSTXXidName{}$,left=$\Pi_{{\mathrm{2}}} :$] {
                \,
              }{1  \odot  \mathsf{J}  \vdash_{\mathsf{GS} }  \mathsf{J}}
            \end{array}
            \]
            Since $\mathsf{Depth}  (  \Pi_{{\mathrm{1}}}  )   +   \mathsf{Depth}  (  \Pi_{{\mathrm{2}}}  )   \, \mGLsym{=} \, 0$ we find ourselves in a base
            case. $\Pi_{{\mathrm{1}}}$ and $\Pi_{{\mathrm{2}}}$ are cut-free, so we must produce a
            cut-free proof, $\Pi$ with $1  \boxast [ { 1 }^{ \mGLmv{n} } ]   \mGLsym{=}  1$,
            which we do as follows:
            \[
            \inferrule* [flushleft,right=$\mGLdruleGSTXXUnitLName{}$,] {
              \inferrule* [flushleft,right=$\mGLdruleGSTXXUnitRName{}$,] {
                \,
              }{\emptyset  \odot  \emptyset  \vdash_{\mathsf{GS} }  \mathsf{J}}
            }{1  \odot  \mathsf{J}  \vdash_{\mathsf{GS} }  \mathsf{J}}
            \]
        \item \textbf{Lin:}
        \[
          \begin{array}{lll}
            \inferrule* [flushleft,right=$\mGLdruleGSTXXidName{}$,left=$\Pi_{{\mathrm{1}}} :$] {
              \,
            }{1  \odot  \mathsf{Lin} \, \mGLnt{A}  \vdash_{\mathsf{GS} }  \mathsf{Lin} \, \mGLnt{A}}
            & \quad &
            \inferrule* [flushleft,right=$\mGLdruleGSTXXidName{}$,left=$\Pi_{{\mathrm{2}}} :$] {
              \,
            }{1  \odot  \mathsf{Lin} \, \mGLnt{A}  \vdash_{\mathsf{GS} }  \mathsf{Lin} \, \mGLnt{A}}
          \end{array}
          \]
          Since $\mathsf{Depth}  (  \Pi_{{\mathrm{1}}}  )   +   \mathsf{Depth}  (  \Pi_{{\mathrm{2}}}  )   \, \mGLsym{=} \, 0$ we find ourselves in a base
          case.$\Pi_{{\mathrm{1}}}$ and $\Pi_{{\mathrm{2}}}$ are cut-free, so we must produce a
          cut-free proof, $\Pi$ with $1  \boxast [ { 1 }^{ \mGLmv{n} } ]   \mGLsym{=}  1$,
          which we do as follows:
          \[
          \inferrule* [flushleft,right=$\mGLdruleGSTXXLinRName{}$,] {
            \inferrule* [flushleft,right=$\mGLdruleMSTXXLinLName{}$,] {
              \inferrule* [flushleft,right=$\mGLdruleMSTXXidName{}$,] {
              \,
              }{\emptyset  \odot  \emptyset  \mGLsym{;}  \mGLnt{A}  \vdash_{\mathsf{MS} }  \mGLnt{A}}
            }{1  \odot  \mathsf{Lin} \, \mGLnt{A}  \mGLsym{;}  \emptyset  \vdash_{\mathsf{MS} }  \mGLnt{A}}
          }{1  \odot  \mathsf{Lin} \, \mGLnt{A}  \vdash_{\mathsf{GS} }  \mathsf{Lin} \, \mGLnt{A}}
          \]
      \end{enumerate}

    \item \textbf{Axiom Cases}
      \begin{enumerate}
        \item \textbf{Axiom on the left:}
        \[
          \begin{array}{lll}
            \inferrule* [flushleft,right=$\mGLdruleGSTXXidName{}$,left=$\Pi_{{\mathrm{1}}} :$] {
              \,
            }{1  \odot  \mGLnt{X}  \vdash_{\mathsf{GS} }  \mGLnt{X}}
            & \quad &
            \inferrule* [flushleft,right=,left=$\Pi_{{\mathrm{2}}} :$] {
              \pi_2
            }{(  \delta_{{\mathrm{1}}}  \mGLsym{,}  \mGLnt{r}  \mGLsym{,}  \delta_{{\mathrm{3}}}  )   \odot   ( \Delta_{{\mathrm{1}}}  \mGLsym{,}  \mGLnt{X}  \mGLsym{,}  \Delta_{{\mathrm{3}}} )   \vdash_{\mathsf{GS} }  \mGLnt{Y}}
          \end{array}
          \]
      We know:
       $\Pi_{{\mathrm{2}}}$ is a proof of $(  \delta_{{\mathrm{1}}}  \mGLsym{,}  \mGLnt{r}  \mGLsym{,}  \delta_{{\mathrm{3}}}  )   \odot   ( \Delta_{{\mathrm{1}}}  \mGLsym{,}  \mGLnt{X}  \mGLsym{,}  \Delta_{{\mathrm{3}}} )   \vdash_{\mathsf{GS} }  \mGLnt{Y}$
       and that $\mGLnt{r}  *  1  \mGLsym{=}  \mGLnt{r}$
      Thus, we construct the following proof $\Pi \, \mGLsym{=} \, \Pi_{{\mathrm{2}}}$.

        \item \textbf{Axiom on the right:}
        \[
          \begin{array}{lll}
           \inferrule* [flushleft,right=,left=$\Pi_{{\mathrm{1}}} :$] {
              \pi_1
            }{\delta_{{\mathrm{1}}}  \odot  \Delta_{{\mathrm{1}}}  \vdash_{\mathsf{GS} }  \mGLnt{X}}
            & \quad &
             \inferrule* [flushleft,right=,left=$\Pi_{{\mathrm{1}}} :$] {
              \,
            }{1  \odot  \mGLnt{X}  \vdash_{\mathsf{GS} }  \mGLnt{X}}
          \end{array}
          \]
      We know:
       $\Pi_{{\mathrm{1}}}$ is a proof of $\delta_{{\mathrm{1}}}  \odot  \Delta_{{\mathrm{1}}}  \vdash_{\mathsf{GS} }  \mGLnt{X}$
       and that $1  \boxast [ { \delta_{{\mathrm{1}}} }^{ \mGLmv{n} } ]   \mGLsym{=}  \delta_{{\mathrm{1}}}$
      Thus, we construct the following proof $\Pi \, \mGLsym{=} \, \Pi_{{\mathrm{1}}}$.
    \end{enumerate}

    \item \textbf{Principal Formula vs Principal Formula}
      \begin{enumerate}
        \item \textbf{Tensor:}
        \[
          \begin{array}{lll}
            \inferrule* [flushleft,right=$\mGLdruleGSTXXTenRName{}$,left=$\Pi_{{\mathrm{1}}} :$] {
              \inferrule* [flushleft,right=,left=$\Pi_{{\mathrm{3}}} :$] {
                \pi_3
              }{\delta_{{\mathrm{1}}}  \odot  \Delta_{{\mathrm{1}}}  \vdash_{\mathsf{GS} }  \mGLnt{X}}\\
              \inferrule* [flushleft,right=,left=$\Pi_{{\mathrm{4}}} :$] {
                \pi_4
              }{\delta_{{\mathrm{2}}}  \odot  \Delta_{{\mathrm{2}}}  \vdash_{\mathsf{GS} }  \mGLnt{Y}}
            }{(  \delta_{{\mathrm{1}}}  \mGLsym{,}  \delta_{{\mathrm{2}}}  )   \odot   ( \Delta_{{\mathrm{1}}}  \mGLsym{,}  \Delta_{{\mathrm{2}}} )   \vdash_{\mathsf{GS} }  \mGLnt{X}  \boxtimes  \mGLnt{Y}}
            & \quad &
            \inferrule* [flushleft,right=$\mGLdruleGSTXXTenLName{}$,left=$\Pi_{{\mathrm{2}}} :$] {
              \inferrule* [flushleft,right=,left=$\Pi_{{\mathrm{5}}} :$] {
                \,
              }{(  \delta_{{\mathrm{1}}}  \mGLsym{,}  \mGLnt{r}  \mGLsym{,}  \mGLnt{r}  \mGLsym{,}  \delta_{{\mathrm{4}}}  )   \odot   ( \Delta_{{\mathrm{1}}}  \mGLsym{,}  \mGLnt{X}  \mGLsym{,}  \mGLnt{Y}  \mGLsym{,}  \Delta_{{\mathrm{4}}} )   \vdash_{\mathsf{GS} }  \mGLnt{Z}}
          }{(  \delta_{{\mathrm{1}}}  \mGLsym{,}  \mGLnt{r}  \mGLsym{,}  \delta_{{\mathrm{4}}}  )   \odot   ( \Delta_{{\mathrm{1}}}  \mGLsym{,}  \mGLnt{X}  \boxtimes  \mGLnt{Y}  \mGLsym{,}  \Delta_{{\mathrm{4}}} )   \vdash_{\mathsf{GS} }  \mGLnt{Z}}
          \end{array}
          \]
    We know:
    \[
    \begin{array}{lll}
      \mathsf{CutRank} \, \mGLsym{(}  \Pi_{{\mathrm{1}}}  \mGLsym{)} \, \mGLsym{=} \, \mGLkw{Max} \, \mGLsym{(}  \mathsf{CutRank} \, \mGLsym{(}  \Pi_{{\mathrm{3}}}  \mGLsym{)}  \mGLsym{,}  \mathsf{CutRank} \, \mGLsym{(}  \Pi_{{\mathrm{4}}}  \mGLsym{)}  \mGLsym{)} \, \leq \,  \mathsf{Rank}  (  \mGLnt{X}  \boxtimes  \mGLnt{Y}  )\\
      \mathsf{CutRank} \, \mGLsym{(}  \Pi_{{\mathrm{3}}}  \mGLsym{)} \, \leq \,  \mathsf{Rank}  (  \mGLnt{X}  \boxtimes  \mGLnt{Y}  )\\
      \mathsf{CutRank} \, \mGLsym{(}  \Pi_{{\mathrm{4}}}  \mGLsym{)} \, \leq \,  \mathsf{Rank}  (  \mGLnt{X}  \boxtimes  \mGLnt{Y}  )\\
      \mathsf{CutRank} \, \mGLsym{(}  \Pi_{{\mathrm{2}}}  \mGLsym{)} \, \mGLsym{=} \, \mathsf{CutRank} \, \mGLsym{(}  \Pi_{{\mathrm{5}}}  \mGLsym{)} \, \leq \,  \mathsf{Rank}  (  \mGLnt{X}  \boxtimes  \mGLnt{Y}  )\\
      \mathsf{Rank}  (  \mGLnt{X}  )  \, \mGLsym{<} \,  \mathsf{Rank}  (  \mGLnt{X}  \boxtimes  \mGLnt{Y}  )\\
      \mathsf{Rank}  (  \mGLnt{Y}  )  \, \mGLsym{<} \,  \mathsf{Rank}  (  \mGLnt{X}  \boxtimes  \mGLnt{Y}  )\\
      \mGLnt{r}  *   (  \delta_{{\mathrm{2}}}  \mGLsym{,}  \delta_{{\mathrm{3}}}  )   \mGLsym{=}   (  \mGLnt{r}  *  \delta_{{\mathrm{2}}}  \mGLsym{,}  \mGLnt{r}  *  \delta_{{\mathrm{3}}}  )
    \end{array}
    \]
    Instead of applying the induction hypothesis,
    we can directly build the proof $\Pi$:
    \[
      \inferrule* [flushleft,right=$\mGLdruleGSTXXCutName{}$,left=$\Pi :$] {
        \inferrule* [flushleft,left=$\Pi_{{\mathrm{4}}} : $] {
            \pi_4
          }{\delta_{{\mathrm{3}}}  \odot  \Delta_{{\mathrm{3}}}  \vdash_{\mathsf{GS} }  \mGLnt{Y}}\\
          \inferrule* [flushleft,right=$\mGLdruleGSTXXCutName{}$] {
            \inferrule* [flushleft,left=$\Pi_{{\mathrm{3}}} : $] {
            \pi_3
          }{\delta_{{\mathrm{1}}}  \odot  \Delta_{{\mathrm{1}}}  \vdash_{\mathsf{GS} }  \mGLnt{X}}\\
          \inferrule* [flushleft,right=,left=$\Pi_{{\mathrm{5}}} :$] {
            \pi_5
          }{(  \delta_{{\mathrm{1}}}  \mGLsym{,}  \mGLnt{r}  \mGLsym{,}  \mGLnt{r}  \mGLsym{,}  \delta_{{\mathrm{4}}}  )   \odot   ( \Delta_{{\mathrm{1}}}  \mGLsym{,}  \mGLnt{X}  \mGLsym{,}  \mGLnt{Y}  \mGLsym{,}  \Delta_{{\mathrm{4}}} )   \vdash_{\mathsf{GS} }  \mGLnt{Z}}
          }{(  \delta_{{\mathrm{1}}}  \mGLsym{,}  \mGLnt{r}  *  \delta_{{\mathrm{2}}}  \mGLsym{,}  \mGLnt{r}  \mGLsym{,}  \delta_{{\mathrm{4}}}  )   \odot   ( \Delta_{{\mathrm{1}}}  \mGLsym{,}  \Delta_{{\mathrm{2}}}  \mGLsym{,}  \mGLnt{Y}  \mGLsym{,}  \Delta_{{\mathrm{4}}} )   \vdash_{\mathsf{GS} }  \mGLnt{Z}}
      }{(  \delta_{{\mathrm{1}}}  \mGLsym{,}  \mGLnt{r}  *  \delta_{{\mathrm{2}}}  \mGLsym{,}  \mGLnt{r}  *  \delta_{{\mathrm{3}}}  \mGLsym{,}  \delta_{{\mathrm{4}}}  )   \odot   ( \Delta_{{\mathrm{1}}}  \mGLsym{,}  \Delta_{{\mathrm{2}}}  \mGLsym{,}  \Delta_{{\mathrm{3}}}  \mGLsym{,}  \Delta_{{\mathrm{4}}} )   \vdash_{\mathsf{GS} }  \mGLnt{Z}}
      \]
      So $\mathsf{CutRank} \, \mGLsym{(}  \Pi  \mGLsym{)} \, \mGLsym{=} \, \mGLkw{Max} \, \mGLsym{(}    \mathsf{CutRank} \, \mGLsym{(}  \Pi_{{\mathrm{3}}}  \mGLsym{)}  \mGLsym{,}  \mathsf{CutRank} \, \mGLsym{(}  \Pi_{{\mathrm{4}}}  \mGLsym{)}  \mGLsym{,}  \mathsf{CutRank} \, \mGLsym{(}  \Pi_{{\mathrm{5}}}  \mGLsym{)}  \mGLsym{,}   \mathsf{Rank}  (  \mGLnt{X}  )   + 1   \mGLsym{,}   \mathsf{Rank}  (  \mGLnt{Y}  )   + 1   \mGLsym{)} \, \leq \,  \mathsf{Rank}  (  \mGLnt{X}  \boxtimes  \mGLnt{Y}  )$

        \item \textbf{Tensor Unit:}
        \[
          \inferrule* [flushleft,right=$\mGLdruleGSTXXUnitRName{}$,left=$\Pi_{{\mathrm{1}}} :$] {
             \,
          }{\emptyset  \odot  \emptyset  \vdash_{\mathsf{GS} }  \mathsf{J}}
    \]

      \[
        \inferrule* [flushleft,right=$\mGLdruleGSTXXUnitLName{}$,left=$\Pi_{{\mathrm{2}}} :$] {
          \inferrule* [flushleft,left=$\Pi_{{\mathrm{3}}} : $] {
            \pi_3
          }{(  \delta_{{\mathrm{1}}}  \mGLsym{,}  \delta_{{\mathrm{2}}}  )   \odot   ( \Delta_{{\mathrm{1}}}  \mGLsym{,}  \Delta_{{\mathrm{2}}} )   \vdash_{\mathsf{GS} }  \mGLnt{X}}
        }{(  \delta_{{\mathrm{1}}}  \mGLsym{,}  \mGLnt{r}  \mGLsym{,}  \delta_{{\mathrm{2}}}  )   \odot   ( \Delta_{{\mathrm{1}}}  \mGLsym{,}  \mathsf{J}  \mGLsym{,}  \Delta_{{\mathrm{2}}} )   \vdash_{\mathsf{GS} }  \mGLnt{X}}
        \]

      We know:
      \[
      \begin{array}{lll}
        \mathsf{CutRank} \, \mGLsym{(}  \Pi_{{\mathrm{3}}}  \mGLsym{)} \, \leq \, \mathsf{CutRank} \, \mGLsym{(}  \Pi_{{\mathrm{2}}}  \mGLsym{)} \, \leq \,  \mathsf{Rank}  (  \mathsf{J}  )\\
        \mGLnt{r}  *  \emptyset  \mGLsym{=}  \emptyset
      \end{array}
      \]
      So $\Pi_{{\mathrm{3}}}$ is a proof of
      $(  \delta_{{\mathrm{1}}}  \mGLsym{,}  \mGLnt{r}  *  \emptyset  \mGLsym{,}  \delta_{{\mathrm{2}}}  )   \odot   ( \Delta_{{\mathrm{1}}}  \mGLsym{,}  \emptyset  \mGLsym{,}  \Delta_{{\mathrm{2}}} )   \vdash_{\mathsf{GS} }  \mGLnt{X}$ with
      $\mathsf{CutRank} \, \mGLsym{(}  \Pi_{{\mathrm{3}}}  \mGLsym{)} \, \leq \,  \mathsf{Rank}  (  \mathsf{J}  )$.
      Thus, we construct the following proof $\Pi \, \mGLsym{=} \, \Pi_{{\mathrm{3}}}$
      \end{enumerate}

    \item \textbf{Secondary Conclusion}
      \begin{enumerate}
        \item \textbf{Left Introduction of Tensor:}
        \[
          \inferrule* [flushleft,right=$\mGLdruleGSTXXTenLName{}$,left=$\Pi_{{\mathrm{1}}} :$] {
            \inferrule* [flushleft,right=,left=$\Pi_{{\mathrm{3}}} :$] {
            \pi_3
          }{(  \delta_{{\mathrm{2}}}  \mGLsym{,}  \mGLnt{r}  \mGLsym{,}  \mGLnt{r}  \mGLsym{,}  \delta_{{\mathrm{3}}}  )   \odot   ( \Delta_{{\mathrm{2}}}  \mGLsym{,}  \mGLnt{X}  \mGLsym{,}  \mGLnt{Y}  \mGLsym{,}  \Delta_{{\mathrm{3}}} )   \vdash_{\mathsf{GS} }  \mGLnt{Z}}
          }{(  \delta_{{\mathrm{2}}}  \mGLsym{,}  \mGLnt{r}  \mGLsym{,}  \delta_{{\mathrm{3}}}  )   \odot   ( \Delta_{{\mathrm{2}}}  \mGLsym{,}  \mGLnt{X}  \boxtimes  \mGLnt{Y}  \mGLsym{,}  \Delta_{{\mathrm{3}}} )   \vdash_{\mathsf{GS} }  \mGLnt{Z}}
    \]
      \[
        \inferrule* [flushleft,right=,left=$\Pi_{{\mathrm{2}}} :$] {
          \pi_2
        }{(  \delta_{{\mathrm{1}}}  \mGLsym{,}  \delta  \mGLsym{,}  \delta_{{\mathrm{4}}}  )   \odot   ( \Delta_{{\mathrm{1}}}  \mGLsym{,}   \mGLnt{Z} ^{ \mGLmv{n} }   \mGLsym{,}  \Delta_{{\mathrm{4}}} )   \vdash_{\mathsf{GS} }  \mGLnt{W} }
        \]
      We know:
      \[
      \begin{array}{lll}
        \mathsf{Depth}  (  \Pi_{{\mathrm{3}}}  )   +   \mathsf{Depth}  (  \Pi_{{\mathrm{2}}}  )   \, \mGLsym{<} \,  \mathsf{Depth}  (  \Pi_{{\mathrm{1}}}  )   +   \mathsf{Depth}  (  \Pi_{{\mathrm{2}}}  )\\
        \mathsf{CutRank} \, \mGLsym{(}  \Pi_{{\mathrm{3}}}  \mGLsym{)} \, \leq \, \mathsf{CutRank} \, \mGLsym{(}  \Pi_{{\mathrm{1}}}  \mGLsym{)} \, \leq \,  \mathsf{Rank}  (  \mGLnt{Z}  )
      \end{array}
      \]

      and so applying the induction hypothesis
      to $\Pi_{{\mathrm{3}}}$ and $\Pi_{{\mathrm{2}}}$
      implies that there is a proof $\Pi'$ of
      $(   \delta_{{\mathrm{1}}}  \mGLsym{,}  \delta  \boxast [ {  (  \delta_{{\mathrm{2}}}  \mGLsym{,}  \mGLnt{r}  \mGLsym{,}  \mGLnt{r}  \mGLsym{,}  \delta_{{\mathrm{3}}}  )  }^{ \mGLmv{n} } ]   \mGLsym{,}  \delta_{{\mathrm{4}}}  )   \odot   ( \Delta_{{\mathrm{1}}}  \mGLsym{,}   ( \Delta_{{\mathrm{2}}}  \mGLsym{,}  \mGLnt{X}  \mGLsym{,}  \mGLnt{Y}  \mGLsym{,}  \Delta_{{\mathrm{3}}} )   \mGLsym{,}  \Delta_{{\mathrm{4}}} )   \vdash_{\mathsf{GS} }  \mGLnt{W}$ with
      $\mathsf{CutRank} \, \mGLsym{(}  \Pi'  \mGLsym{)} \, \leq \,  \mathsf{Rank}  (  \mGLnt{Z}  )$.
      Thus, we construct the following proof $\Pi$:
      \[
        \inferrule* [flushleft,right=$\mGLdruleGSTXXTenLName{}$,left=$\Pi :$] {
          \inferrule* [flushleft,right=,left=$\Pi' :$] {
            \pi'
          }{(  \delta_{{\mathrm{1}}}  \mGLsym{,}   (  \delta'_{{\mathrm{2}}}  \mGLsym{,}  \mGLnt{r'}  \mGLsym{,}  \mGLnt{r'}  \mGLsym{,}  \delta'_{{\mathrm{3}}}  )   \mGLsym{,}  \delta_{{\mathrm{4}}}  )   \odot   ( \Delta_{{\mathrm{1}}}  \mGLsym{,}   ( \Delta_{{\mathrm{2}}}  \mGLsym{,}  \mGLnt{X}  \mGLsym{,}  \mGLnt{Y}  \mGLsym{,}  \Delta_{{\mathrm{3}}} )   \mGLsym{,}  \Delta_{{\mathrm{4}}} )   \vdash_{\mathsf{GS} }  \mGLnt{W}}
        }{(  \delta_{{\mathrm{1}}}  \mGLsym{,}   (  \delta'_{{\mathrm{2}}}  \mGLsym{,}  \mGLnt{r'}  \mGLsym{,}  \delta'_{{\mathrm{3}}}  )   \mGLsym{,}  \delta_{{\mathrm{4}}}  )   \odot   ( \Delta_{{\mathrm{1}}}  \mGLsym{,}   ( \Delta_{{\mathrm{2}}}  \mGLsym{,}  \mGLnt{X}  \boxtimes  \mGLnt{Y}  \mGLsym{,}  \Delta_{{\mathrm{3}}} )   \mGLsym{,}  \Delta_{{\mathrm{4}}} )   \vdash_{\mathsf{GS} }  \mGLnt{W}}
        \]
        Given the above, we know:
        \[
          \begin{array}{lll}
            \mathsf{CutRank} \, \mGLsym{(}  \Pi  \mGLsym{)} \, \mGLsym{=} \, \mathsf{CutRank} \, \mGLsym{(}  \Pi'  \mGLsym{)} \, \leq \,  \mathsf{Rank}  (  \mGLnt{Z}  )\\
            \delta  \boxast [ {  (  \delta_{{\mathrm{2}}}  \mGLsym{,}  \mGLnt{r}  \mGLsym{,}  \mGLnt{r}  \mGLsym{,}  \delta_{{\mathrm{3}}}  )  }^{ \mGLmv{n} } ]   \mGLsym{=}   (  \delta'_{{\mathrm{2}}}  \mGLsym{,}  \mGLnt{r'}  \mGLsym{,}  \mGLnt{r'}  \mGLsym{,}  \delta'_{{\mathrm{3}}}  )
          \end{array}
          \]
          Implies
          \[
          \begin{array}{lll}
            (      \delta  \boxast [ { \delta_{{\mathrm{2}}} }^{ \mGLmv{n} } ]   \mGLsym{,}  \delta  \boxast [ { \mGLnt{r} }^{ \mGLmv{n} } ]   \mGLsym{,}  \delta  \boxast [ { \mGLnt{r} }^{ \mGLmv{n} } ]   \mGLsym{,}  \delta  \boxast [ { \delta_{{\mathrm{3}}} }^{ \mGLmv{n} } ]   )   \mGLsym{=}   (  \delta'_{{\mathrm{2}}}  \mGLsym{,}  \mGLnt{r'}  \mGLsym{,}  \mGLnt{r'}  \mGLsym{,}  \delta'_{{\mathrm{3}}}  )\\
            (     \delta  \boxast [ { \delta_{{\mathrm{2}}} }^{ \mGLmv{n} } ]   \mGLsym{,}  \delta  \boxast [ { \mGLnt{r} }^{ \mGLmv{n} } ]   \mGLsym{,}  \delta  \boxast [ { \delta_{{\mathrm{3}}} }^{ \mGLmv{n} } ]   )   \mGLsym{=}   (  \delta'_{{\mathrm{2}}}  \mGLsym{,}  \mGLnt{r'}  \mGLsym{,}  \delta'_{{\mathrm{3}}}  )\\
            \delta  \boxast [ {  (  \delta_{{\mathrm{2}}}  \mGLsym{,}  \mGLnt{r}  \mGLsym{,}  \delta_{{\mathrm{3}}}  )  }^{ \mGLmv{n} } ]   \mGLsym{=}   (  \delta'_{{\mathrm{2}}}  \mGLsym{,}  \mGLnt{r'}  \mGLsym{,}  \delta'_{{\mathrm{3}}}  )
          \end{array}
          \]

        \item \textbf{Left Introduction of Tensor Unit:}
        \[
          \inferrule* [flushleft,right=$\mGLdruleGSTXXUnitLName{}$,left=$\Pi_{{\mathrm{1}}} :$] {
              \inferrule* [flushleft,right=,left=$\Pi_{{\mathrm{3}}} :$] {
                \pi_3
              }{(  \delta_{{\mathrm{2}}}  \mGLsym{,}  \delta_{{\mathrm{3}}}  )   \odot   ( \Delta_{{\mathrm{2}}}  \mGLsym{,}  \Delta_{{\mathrm{3}}} )   \vdash_{\mathsf{GS} }  \mGLnt{Y}}
          }{(  \delta_{{\mathrm{2}}}  \mGLsym{,}  \mGLnt{r}  \mGLsym{,}  \delta_{{\mathrm{3}}}  )   \odot   ( \Delta_{{\mathrm{2}}}  \mGLsym{,}  \mathsf{J}  \mGLsym{,}  \Delta_{{\mathrm{3}}} )   \vdash_{\mathsf{GS} }  \mGLnt{X}}
    \]
    \[
      \inferrule* [flushleft,right=,left=$\Pi_{{\mathrm{2}}} :$] {
        \pi_2
      }{(  \delta_{{\mathrm{1}}}  \mGLsym{,}  \delta  \mGLsym{,}  \delta_{{\mathrm{4}}}  )   \odot   ( \Delta_{{\mathrm{1}}}  \mGLsym{,}   \mGLnt{X} ^{ \mGLmv{n} }   \mGLsym{,}  \Delta_{{\mathrm{4}}} )   \vdash_{\mathsf{GS} }  \mGLnt{Y}}
\]
      We know:
      \[
      \begin{array}{lll}
        \mathsf{Depth}  (  \Pi_{{\mathrm{2}}}  )   +   \mathsf{Depth}  (  \Pi_{{\mathrm{3}}}  )   \, \mGLsym{<} \,  \mathsf{Depth}  (  \Pi_{{\mathrm{1}}}  )   +   \mathsf{Depth}  (  \Pi_{{\mathrm{2}}}  )\\
        \mathsf{CutRank} \, \mGLsym{(}  \Pi_{{\mathrm{3}}}  \mGLsym{)} \, \leq \, \mathsf{CutRank} \, \mGLsym{(}  \Pi_{{\mathrm{1}}}  \mGLsym{)} \, \leq \,  \mathsf{Rank}  (  \mGLnt{X}  )
      \end{array}
      \]

      and so applying the induction hypothesis
      to $\Pi_{{\mathrm{2}}}$ and $\Pi_{{\mathrm{3}}}$
      implies that there is a proof $\Pi'$ of
      $(  \delta_{{\mathrm{1}}}  \mGLsym{,}   (  \delta'_{{\mathrm{2}}}  \mGLsym{,}  \delta'_{{\mathrm{3}}}  )   \mGLsym{,}  \delta_{{\mathrm{4}}}  )   \odot   ( \Delta_{{\mathrm{1}}}  \mGLsym{,}  \Delta_{{\mathrm{2}}}  \mGLsym{,}  \Delta_{{\mathrm{3}}}  \mGLsym{,}  \Delta_{{\mathrm{4}}} )   \vdash_{\mathsf{GS} }  \mGLnt{Y}$ with
      $\mathsf{CutRank} \, \mGLsym{(}  \Pi'  \mGLsym{)} \, \leq \,  \mathsf{Rank}  (  \mGLnt{X}  )$ and $\delta  \boxast [ {  (  \delta_{{\mathrm{2}}}  \mGLsym{,}  \delta_{{\mathrm{3}}}  )  }^{ \mGLmv{n} } ]   \mGLsym{=}   (  \delta'_{{\mathrm{2}}}  \mGLsym{,}  \delta'_{{\mathrm{3}}}  )$
      where $ | \delta_{{\mathrm{2}}} |=| \delta'_{{\mathrm{2}}} |$ and $| \delta_{{\mathrm{3}}} |=| \delta'_{{\mathrm{3}}} |$.
      Thus, we construct the following proof $\Pi$:

      \[
        \inferrule* [flushleft,right=$\mGLdruleGSTXXUnitLName{}$,left=$\Pi :$] {
          \inferrule* [flushleft,right=,left=$\Pi' :$] {
            \pi'
          }{(  \delta_{{\mathrm{1}}}  \mGLsym{,}  \delta'_{{\mathrm{2}}}  \mGLsym{,}  \delta'_{{\mathrm{3}}}  \mGLsym{,}  \delta_{{\mathrm{4}}}  )   \odot   ( \Delta_{{\mathrm{1}}}  \mGLsym{,}  \Delta_{{\mathrm{2}}}  \mGLsym{,}  \Delta_{{\mathrm{3}}}  \mGLsym{,}  \Delta_{{\mathrm{4}}} )   \vdash_{\mathsf{GS} }  \mGLnt{Y}}
        }{(  \delta_{{\mathrm{1}}}  \mGLsym{,}  \delta'_{{\mathrm{2}}}  \mGLsym{,}   (   \delta  \boxast [ { \mGLnt{r} }^{ \mGLmv{n} } ]   )   \mGLsym{,}  \delta'_{{\mathrm{3}}}  \mGLsym{,}  \delta_{{\mathrm{4}}}  )   \odot   ( \Delta_{{\mathrm{1}}}  \mGLsym{,}  \Delta_{{\mathrm{2}}}  \mGLsym{,}  \mathsf{J}  \mGLsym{,}  \Delta_{{\mathrm{3}}}  \mGLsym{,}  \Delta_{{\mathrm{4}}} )   \vdash_{\mathsf{GS} }  \mGLnt{Y}}
        \]
        Given the above, we know:
        \[
          \begin{array}{lll}
            \mathsf{CutRank} \, \mGLsym{(}  \Pi  \mGLsym{)} \, \mGLsym{=} \, \mathsf{CutRank} \, \mGLsym{(}  \Pi'  \mGLsym{)} \, \leq \,  \mathsf{Rank}  (  \mGLnt{X}  )\\
            \delta  \boxast [ {  (  \delta_{{\mathrm{2}}}  \mGLsym{,}  \delta_{{\mathrm{3}}}  )  }^{ \mGLmv{n} } ]   \mGLsym{=}   (  \delta'_{{\mathrm{2}}}  \mGLsym{,}  \delta'_{{\mathrm{3}}}  )\\
            (    \delta  \boxast [ { \delta_{{\mathrm{2}}} }^{ \mGLmv{n} } ]   \mGLsym{,}  \delta  \boxast [ { \delta_{{\mathrm{3}}} }^{ \mGLmv{n} } ]   )   \mGLsym{=}  \delta'_{{\mathrm{2}}}  \mGLsym{,}  \delta'_{{\mathrm{3}}}\\
            \delta  \boxast [ { \delta_{{\mathrm{2}}} }^{ \mGLmv{n} } ]   \mGLsym{=}  \delta'_{{\mathrm{2}}}\\
            \delta  \boxast [ { \delta_{{\mathrm{3}}} }^{ \mGLmv{n} } ]   \mGLsym{=}  \delta'_{{\mathrm{3}}}\\
            \delta  \boxast [ { \mGLnt{r} }^{ \mGLmv{n} } ]   \mGLsym{=}   \delta  \boxast [ { \mGLnt{r} }^{ \mGLmv{n} } ]\\
            (     \delta  \boxast [ { \delta_{{\mathrm{2}}} }^{ \mGLmv{n} } ]   \mGLsym{,}  \delta  \boxast [ { \mGLnt{r} }^{ \mGLmv{n} } ]   \mGLsym{,}  \delta  \boxast [ { \delta_{{\mathrm{3}}} }^{ \mGLmv{n} } ]   )   \mGLsym{=}   (  \delta'_{{\mathrm{2}}}  \mGLsym{,}   \delta  \boxast [ { \mGLnt{r} }^{ \mGLmv{n} } ]   \mGLsym{,}  \delta'_{{\mathrm{3}}}  )\\
            \delta  \boxast [ {  (  \delta_{{\mathrm{2}}}  \mGLsym{,}  \mGLnt{r}  \mGLsym{,}  \delta_{{\mathrm{3}}}  )  }^{ \mGLmv{n} } ]   \mGLsym{=}   (  \delta'_{{\mathrm{2}}}  \mGLsym{,}   \delta  \boxast [ { \mGLnt{r} }^{ \mGLmv{n} } ]   \mGLsym{,}  \delta'_{{\mathrm{3}}}  )
          \end{array}
          \]

          \item \textbf{Weakening:}
          \[
        \inferrule* [flushleft,right=$\mGLdruleGSTXXWeakName{}$,left=$\Pi_{{\mathrm{1}}} :$] {
          \inferrule* [flushleft,left=$\Pi_{{\mathrm{3}}} : $] {
            \pi_3
          }{(  \delta_{{\mathrm{2}}}  \mGLsym{,}  \delta_{{\mathrm{3}}}  )   \odot   ( \Delta_{{\mathrm{1}}}  \mGLsym{,}  \Delta_{{\mathrm{2}}} )   \vdash_{\mathsf{GS} }  \mGLnt{X}}
        }{(  \delta_{{\mathrm{2}}}  \mGLsym{,}  \mathsf{0}  \mGLsym{,}  \delta_{{\mathrm{3}}}  )   \odot   ( \Delta_{{\mathrm{1}}}  \mGLsym{,}  \mGLnt{Y}  \mGLsym{,}  \Delta_{{\mathrm{2}}} )   \vdash_{\mathsf{GS} }  \mGLnt{X}}
        \]
        \[
          \inferrule* [flushleft,right=,left=$\Pi_{{\mathrm{2}}} :$] {
            \pi_2
          }{(  \delta_{{\mathrm{1}}}  \mGLsym{,}  \delta  \mGLsym{,}  \delta_{{\mathrm{4}}}  )   \odot   ( \Delta_{{\mathrm{1}}}  \mGLsym{,}   \mGLnt{X} ^{ \mGLmv{n} }   \mGLsym{,}  \Delta_{{\mathrm{4}}} )   \vdash_{\mathsf{GS} }  \mGLnt{Z}}
    \]
    We know:
      \[
      \begin{array}{lll}
        \mathsf{Depth}  (  \Pi_{{\mathrm{3}}}  )   +   \mathsf{Depth}  (  \Pi_{{\mathrm{2}}}  )   \, \mGLsym{<} \,  \mathsf{Depth}  (  \Pi_{{\mathrm{1}}}  )   +   \mathsf{Depth}  (  \Pi_{{\mathrm{2}}}  )\\
        \mathsf{CutRank} \, \mGLsym{(}  \Pi_{{\mathrm{3}}}  \mGLsym{)} \, \leq \, \mathsf{CutRank} \, \mGLsym{(}  \Pi_{{\mathrm{1}}}  \mGLsym{)} \, \leq \,  \mathsf{Rank}  (  \mGLnt{X}  )
      \end{array}
      \]

      and so applying the induction hypothesis
      to $\Pi_{{\mathrm{3}}}$ and $\Pi_{{\mathrm{2}}}$
      implies that there is a proof $\Pi'$ of
      $\delta_{{\mathrm{1}}}  \mGLsym{,}   (  \delta'_{{\mathrm{2}}}  \mGLsym{,}  \delta'_{{\mathrm{3}}}  )   \mGLsym{,}  \delta_{{\mathrm{4}}}  \odot  \Delta_{{\mathrm{1}}}  \mGLsym{,}  \Delta_{{\mathrm{2}}}  \mGLsym{,}  \Delta_{{\mathrm{3}}}  \mGLsym{,}  \Delta_{{\mathrm{4}}}  \vdash_{\mathsf{GS} }  \mGLnt{Z}$ with
      $\mathsf{CutRank} \, \mGLsym{(}  \Pi'  \mGLsym{)} \, \leq \,  \mathsf{Rank}  (  \mGLnt{X}  )$ and $\delta  \boxast [ {  (  \delta_{{\mathrm{2}}}  \mGLsym{,}  \delta_{{\mathrm{3}}}  )  }^{ \mGLmv{n} } ]   \mGLsym{=}   (  \delta'_{{\mathrm{2}}}  \mGLsym{,}  \delta'_{{\mathrm{3}}}  )$
      where $ | \delta_{{\mathrm{2}}} |=| \delta'_{{\mathrm{2}}} |$ and $| \delta_{{\mathrm{3}}} |=| \delta'_{{\mathrm{3}}} |$.
      Thus, we construct the following proof $\Pi$:

          \[
        \inferrule* [flushleft,right=$\mGLdruleGSTXXWeakName{}$,left=$\Pi :$] {
          \inferrule* [flushleft,left=$\Pi' : $] {
            \pi'
          }{\delta_{{\mathrm{1}}}  \mGLsym{,}  \delta'_{{\mathrm{2}}}  \mGLsym{,}  \delta'_{{\mathrm{3}}}  \mGLsym{,}  \delta_{{\mathrm{4}}}  \odot  \Delta_{{\mathrm{1}}}  \mGLsym{,}  \Delta_{{\mathrm{2}}}  \mGLsym{,}  \Delta_{{\mathrm{3}}}  \mGLsym{,}  \Delta_{{\mathrm{4}}}  \vdash_{\mathsf{GS} }  \mGLnt{Z}}
        }{\delta_{{\mathrm{1}}}  \mGLsym{,}  \delta'_{{\mathrm{2}}}  \mGLsym{,}  \mathsf{0}  \mGLsym{,}  \delta'_{{\mathrm{3}}}  \mGLsym{,}  \delta_{{\mathrm{4}}}  \odot  \Delta_{{\mathrm{1}}}  \mGLsym{,}  \Delta_{{\mathrm{2}}}  \mGLsym{,}  \mGLnt{Y}  \mGLsym{,}  \Delta_{{\mathrm{3}}}  \mGLsym{,}  \Delta_{{\mathrm{4}}}  \vdash_{\mathsf{GS} }  \mGLnt{Z}}
        \]
        Given the above, we know:
        \[
          \begin{array}{lll}
            \mathsf{CutRank} \, \mGLsym{(}  \Pi  \mGLsym{)} \, \mGLsym{=} \, \mathsf{CutRank} \, \mGLsym{(}  \Pi'  \mGLsym{)} \, \leq \,  \mathsf{Rank}  (  \mGLnt{X}  )\\
            \delta  \boxast [ {  (  \delta_{{\mathrm{2}}}  \mGLsym{,}  \delta_{{\mathrm{3}}}  )  }^{ \mGLmv{n} } ]   \mGLsym{=}   (  \delta'_{{\mathrm{2}}}  \mGLsym{,}  \delta'_{{\mathrm{3}}}  )\\
            (    \delta  \boxast [ { \delta_{{\mathrm{2}}} }^{ \mGLmv{n} } ]   \mGLsym{,}  \delta  \boxast [ { \delta_{{\mathrm{3}}} }^{ \mGLmv{n} } ]   )   \mGLsym{=}  \delta'_{{\mathrm{2}}}  \mGLsym{,}  \delta'_{{\mathrm{3}}}\\
            \delta  \boxast [ { \delta_{{\mathrm{2}}} }^{ \mGLmv{n} } ]   \mGLsym{=}  \delta'_{{\mathrm{2}}}\\
            \delta  \boxast [ { \delta_{{\mathrm{3}}} }^{ \mGLmv{n} } ]   \mGLsym{=}  \delta'_{{\mathrm{3}}}\\
            \delta  \boxast [ { \mathsf{0} }^{ \mGLmv{n} } ]   \mGLsym{=}  \mathsf{0}\\
            (     \delta  \boxast [ { \delta_{{\mathrm{2}}} }^{ \mGLmv{n} } ]   \mGLsym{,}  \delta  \boxast [ { \mathsf{0} }^{ \mGLmv{n} } ]   \mGLsym{,}  \delta  \boxast [ { \delta_{{\mathrm{3}}} }^{ \mGLmv{n} } ]   )   \mGLsym{=}   (  \delta'_{{\mathrm{2}}}  \mGLsym{,}  \mathsf{0}  \mGLsym{,}  \delta'_{{\mathrm{3}}}  )\\
            \delta  \boxast [ {  (  \delta_{{\mathrm{2}}}  \mGLsym{,}  \mGLnt{r}  \mGLsym{,}  \delta_{{\mathrm{3}}}  )  }^{ \mGLmv{n} } ]   \mGLsym{=}   (  \delta'_{{\mathrm{2}}}  \mGLsym{,}  \mathsf{0}  \mGLsym{,}  \delta'_{{\mathrm{3}}}  )
          \end{array}
          \]
      \end{enumerate}

    \item \textbf{Secondary Hypothesis}
      \begin{enumerate}
        \item \textbf{Right introduction of tensor product (first case):}
        \[
          \inferrule* [flushleft,right=,left=$\Pi_{{\mathrm{1}}} :$] {
            \pi_1
          }{\delta_{{\mathrm{2}}}  \odot  \Delta_{{\mathrm{2}}}  \vdash_{\mathsf{GS} }  \mGLnt{X}}
    \]

      \[
        \inferrule* [flushleft,right=$\mGLdruleGSTXXTenRName{}$,left=$\Pi_{{\mathrm{2}}} :$] {
          \inferrule* [flushleft,left=$\Pi_{{\mathrm{3}}} : $] {
            \pi_3
          }{\delta_{{\mathrm{1}}}  \mGLsym{,}  \delta  \mGLsym{,}  \delta_{{\mathrm{3}}}  \odot  \Delta_{{\mathrm{1}}}  \mGLsym{,}   \mGLnt{X} ^{ \mGLmv{n} }   \mGLsym{,}  \Delta_{{\mathrm{3}}}  \vdash_{\mathsf{GS} }  \mGLnt{Y} }\\
          \inferrule* [flushleft,right=,left=$\Pi_{{\mathrm{4}}} :$] {
            \pi_4
          }{ \delta_{{\mathrm{4}}}  \odot  \Delta_{{\mathrm{4}}}  \vdash_{\mathsf{GS} }  \mGLnt{Z}}
        }{\delta_{{\mathrm{1}}}  \mGLsym{,}  \delta  \mGLsym{,}  \delta_{{\mathrm{3}}}  \mGLsym{,}  \delta_{{\mathrm{4}}}  \odot  \Delta_{{\mathrm{1}}}  \mGLsym{,}   \mGLnt{X} ^{ \mGLmv{n} }   \mGLsym{,}  \Delta_{{\mathrm{3}}}  \mGLsym{,}  \Delta_{{\mathrm{4}}}  \vdash_{\mathsf{GS} }  \mGLnt{Y}  \boxtimes  \mGLnt{Z}  }
        \]

      We know:
      \[
      \begin{array}{lll}
        \mathsf{Depth}  (  \Pi_{{\mathrm{1}}}  )   +   \mathsf{Depth}  (  \Pi_{{\mathrm{3}}}  )   \, \mGLsym{<} \,  \mathsf{Depth}  (  \Pi_{{\mathrm{1}}}  )   +   \mathsf{Depth}  (  \Pi_{{\mathrm{2}}}  )\\
        \mathsf{CutRank} \, \mGLsym{(}  \Pi_{{\mathrm{3}}}  \mGLsym{)} \, \leq \, \mathsf{CutRank} \, \mGLsym{(}  \Pi_{{\mathrm{2}}}  \mGLsym{)} \, \leq \,  \mathsf{Rank}  (  \mGLnt{X}  )
      \end{array}
      \]

      and so applying the induction hypothesis
      to $\Pi_{{\mathrm{1}}}$ and $\Pi_{{\mathrm{3}}}$
      implies that there is a proof $\Pi'$ of
      $\delta_{{\mathrm{1}}}  \mGLsym{,}  \delta'_{{\mathrm{2}}}  \mGLsym{,}  \delta_{{\mathrm{3}}}  \odot  \Delta_{{\mathrm{1}}}  \mGLsym{,}  \Delta_{{\mathrm{2}}}  \mGLsym{,}  \Delta_{{\mathrm{3}}}  \vdash_{\mathsf{GS} }  \mGLnt{Y}$ with
      $\mathsf{CutRank} \, \mGLsym{(}  \Pi'  \mGLsym{)} \, \leq \,  \mathsf{Rank}  (  \mGLnt{X}  )$ and $(   \delta  \boxast [ {  (  \delta_{{\mathrm{2}}}  )  }^{ \mGLmv{n} } ]   )   \mGLsym{=}  \delta'_{{\mathrm{2}}}$.
      Thus, we construct the following proof $\Pi$:

        \[
        \inferrule* [flushleft,right=$\mGLdruleGSTXXTenRName{}$,left=$\Pi :$] {
          \inferrule* [flushleft,left=$\Pi' : $] {
            \pi'
          }{\delta_{{\mathrm{1}}}  \mGLsym{,}  \delta'_{{\mathrm{2}}}  \mGLsym{,}  \delta_{{\mathrm{3}}}  \odot  \Delta_{{\mathrm{1}}}  \mGLsym{,}  \Delta_{{\mathrm{2}}}  \mGLsym{,}  \Delta_{{\mathrm{3}}}  \vdash_{\mathsf{GS} }  \mGLnt{Y}}\\
          \inferrule* [flushleft,right=,left=$\Pi_{{\mathrm{4}}} :$] {
            \pi_4
          }{\delta_{{\mathrm{4}}}  \odot  \Delta_{{\mathrm{4}}}  \vdash_{\mathsf{GS} }  \mGLnt{Z}}
        }{\delta_{{\mathrm{1}}}  \mGLsym{,}  \delta'_{{\mathrm{2}}}  \mGLsym{,}  \delta_{{\mathrm{3}}}  \mGLsym{,}  \delta_{{\mathrm{4}}}  \odot  \Delta_{{\mathrm{1}}}  \mGLsym{,}  \Delta_{{\mathrm{2}}}  \mGLsym{,}  \Delta_{{\mathrm{3}}}  \mGLsym{,}  \Delta_{{\mathrm{4}}}  \vdash_{\mathsf{GS} }  \mGLnt{Y}  \boxtimes  \mGLnt{Z}}
        \]
        Given the above, we know:
        \[
          \begin{array}{lll}
            \mathsf{CutRank} \, \mGLsym{(}  \Pi  \mGLsym{)} \, \mGLsym{=} \, \mGLkw{Max} \, \mGLsym{(}  \mathsf{CutRank} \, \mGLsym{(}  \Pi'  \mGLsym{)}  \mGLsym{,}  \mathsf{CutRank} \, \mGLsym{(}  \Pi_{{\mathrm{4}}}  \mGLsym{)}  \mGLsym{)}\\
            \mathsf{CutRank} \, \mGLsym{(}  \Pi_{{\mathrm{4}}}  \mGLsym{)} \, \leq \,  \mathsf{Rank}  (  \mGLnt{X}  )\\
            \mathsf{CutRank} \, \mGLsym{(}  \Pi'  \mGLsym{)} \, \leq \,  \mathsf{Rank}  (  \mGLnt{X}  )\\
            \mathsf{CutRank} \, \mGLsym{(}  \Pi  \mGLsym{)} \, \leq \,  \mathsf{Rank}  (  \mGLnt{X}  )\\
            \delta  \boxast [ { \delta_{{\mathrm{2}}} }^{ \mGLmv{n} } ]   \mGLsym{=}  \delta'_{{\mathrm{2}}}\\
          \end{array}
          \]

        \item \textbf{Right introduction of tensor product (second case):}
        \[
          \inferrule* [flushleft,right=,left=$\Pi_{{\mathrm{1}}} :$] {
            \pi_1
          }{\delta_{{\mathrm{3}}}  \odot  \Delta_{{\mathrm{3}}}  \vdash_{\mathsf{GS} }  \mGLnt{X}}
    \]

      \[
        \inferrule* [flushleft,right=$\mGLdruleGSTXXTenRName{}$,left=$\Pi_{{\mathrm{2}}} :$] {
          \inferrule* [flushleft,left=$\Pi_{{\mathrm{3}}} : $] {
            \pi_3
          }{ \delta_{{\mathrm{1}}}  \odot  \Delta_{{\mathrm{1}}}  \vdash_{\mathsf{GS} }  \mGLnt{Y}}\\
          \inferrule* [flushleft,right=,left=$\Pi_{{\mathrm{4}}} :$] {
            \pi_4
          }{\delta_{{\mathrm{2}}}  \mGLsym{,}  \delta  \mGLsym{,}  \delta_{{\mathrm{4}}}  \odot  \Delta_{{\mathrm{2}}}  \mGLsym{,}   \mGLnt{X} ^{ \mGLmv{n} }   \mGLsym{,}  \Delta_{{\mathrm{4}}}  \vdash_{\mathsf{GS} }  \mGLnt{Z} }
        }{\delta_{{\mathrm{1}}}  \mGLsym{,}  \delta  \mGLsym{,}  \delta_{{\mathrm{3}}}  \mGLsym{,}  \delta_{{\mathrm{4}}}  \odot  \Delta_{{\mathrm{1}}}  \mGLsym{,}   \mGLnt{X} ^{ \mGLmv{n} }   \mGLsym{,}  \Delta_{{\mathrm{3}}}  \mGLsym{,}  \Delta_{{\mathrm{4}}}  \vdash_{\mathsf{GS} }  \mGLnt{Y}  \boxtimes  \mGLnt{Z}  }
        \]
        Similar to previous case.

        \item \textbf{Left introduction of tensor product:}
        \[
          \inferrule* [flushleft,right=,left=$\Pi_{{\mathrm{1}}} :$] {
            \pi_1
          }{\delta_{{\mathrm{3}}}  \odot  \Delta_{{\mathrm{3}}}  \vdash_{\mathsf{GS} }  \mGLnt{X}}
    \]

      \[
        \inferrule* [flushleft,right=$\mGLdruleGSTXXTenLName{}$,left=$\Pi_{{\mathrm{2}}} :$] {
          \inferrule* [flushleft,right=,left=$\Pi_{{\mathrm{3}}} :$] {
            \pi_3
          }{(  \delta_{{\mathrm{1}}}  \mGLsym{,}  \mGLnt{r}  \mGLsym{,}  \mGLnt{r}  \mGLsym{,}  \delta_{{\mathrm{2}}}  \mGLsym{,}  \delta  \mGLsym{,}  \delta_{{\mathrm{4}}}  )   \odot   ( \Delta_{{\mathrm{1}}}  \mGLsym{,}  \mGLnt{Y}  \mGLsym{,}  \mGLnt{Z}  \mGLsym{,}  \Delta_{{\mathrm{2}}}  \mGLsym{,}   \mGLnt{X} ^{ \mGLmv{n} }   \mGLsym{,}  \Delta_{{\mathrm{4}}} )   \vdash_{\mathsf{GS} }  \mGLnt{Z}}
          }{(  \delta_{{\mathrm{1}}}  \mGLsym{,}  \mGLnt{r}  \mGLsym{,}  \delta_{{\mathrm{2}}}  \mGLsym{,}  \delta  \mGLsym{,}  \delta_{{\mathrm{4}}}  )   \odot   ( \Delta_{{\mathrm{1}}}  \mGLsym{,}  \mGLnt{Y}  \boxtimes  \mGLnt{Z}  \mGLsym{,}  \Delta_{{\mathrm{2}}}  \mGLsym{,}   \mGLnt{X} ^{ \mGLmv{n} }   \mGLsym{,}  \Delta_{{\mathrm{4}}} )   \vdash_{\mathsf{GS} }  \mGLnt{Z}}
        \]

      We know:
      \[
      \begin{array}{lll}
        \mathsf{Depth}  (  \Pi_{{\mathrm{1}}}  )   +   \mathsf{Depth}  (  \Pi_{{\mathrm{3}}}  )   \, \mGLsym{<} \,  \mathsf{Depth}  (  \Pi_{{\mathrm{1}}}  )   +   \mathsf{Depth}  (  \Pi_{{\mathrm{2}}}  )\\
        \mathsf{CutRank} \, \mGLsym{(}  \Pi_{{\mathrm{3}}}  \mGLsym{)} \, \leq \, \mathsf{CutRank} \, \mGLsym{(}  \Pi_{{\mathrm{2}}}  \mGLsym{)} \, \leq \,  \mathsf{Rank}  (  \mGLnt{X}  )
      \end{array}
      \]

      and so applying the induction hypothesis
      to $\Pi_{{\mathrm{1}}}$ and $\Pi_{{\mathrm{3}}}$
      implies that there is a proof $\Pi'$ of
      $(  \delta_{{\mathrm{1}}}  \mGLsym{,}  \mGLnt{r}  \mGLsym{,}  \mGLnt{r}  \mGLsym{,}  \delta_{{\mathrm{2}}}  \mGLsym{,}  \delta'_{{\mathrm{3}}}  \mGLsym{,}  \delta_{{\mathrm{4}}}  )   \odot   ( \Delta_{{\mathrm{1}}}  \mGLsym{,}  \mGLnt{Y}  \mGLsym{,}  \mGLnt{Z}  \mGLsym{,}  \Delta_{{\mathrm{2}}}  \mGLsym{,}  \Delta_{{\mathrm{3}}}  \mGLsym{,}  \Delta_{{\mathrm{4}}} )   \vdash_{\mathsf{GS} }  \mGLnt{Z}$ with
      $\mathsf{CutRank} \, \mGLsym{(}  \Pi'  \mGLsym{)} \, \leq \,  \mathsf{Rank}  (  \mGLnt{X}  )$ and $\delta  \boxast [ { \delta_{{\mathrm{3}}} }^{ \mGLmv{n} } ]   \mGLsym{=}  \delta'_{{\mathrm{3}}}$.
      Thus, we construct the following proof $\Pi$:

      \[
        \inferrule* [flushleft,right=$\mGLdruleGSTXXTenLName{}$,left=$\Pi :$] {
          \inferrule* [flushleft,right=,left=$\Pi' :$] {
            \pi'
          }{(  \delta_{{\mathrm{1}}}  \mGLsym{,}  \mGLnt{r}  \mGLsym{,}  \mGLnt{r}  \mGLsym{,}  \delta_{{\mathrm{2}}}  \mGLsym{,}  \delta'_{{\mathrm{3}}}  \mGLsym{,}  \delta_{{\mathrm{4}}}  )   \odot   ( \Delta_{{\mathrm{1}}}  \mGLsym{,}  \mGLnt{Y}  \mGLsym{,}  \mGLnt{Z}  \mGLsym{,}  \Delta_{{\mathrm{2}}}  \mGLsym{,}  \Delta_{{\mathrm{3}}}  \mGLsym{,}  \Delta_{{\mathrm{4}}} )   \vdash_{\mathsf{GS} }  \mGLnt{Z}}
        }{(  \delta_{{\mathrm{1}}}  \mGLsym{,}  \mGLnt{r}  \mGLsym{,}  \delta_{{\mathrm{2}}}  \mGLsym{,}  \delta'_{{\mathrm{3}}}  \mGLsym{,}  \delta_{{\mathrm{4}}}  )   \odot   ( \Delta_{{\mathrm{1}}}  \mGLsym{,}  \mGLnt{Y}  \boxtimes  \mGLnt{Z}  \mGLsym{,}  \Delta_{{\mathrm{2}}}  \mGLsym{,}  \Delta_{{\mathrm{3}}}  \mGLsym{,}  \Delta_{{\mathrm{4}}} )   \vdash_{\mathsf{GS} }  \mGLnt{Z}}
        \]
        Given the above, we know:
        \[
          \begin{array}{lll}
            \mathsf{CutRank} \, \mGLsym{(}  \Pi  \mGLsym{)} \, \mGLsym{=} \, \mathsf{CutRank} \, \mGLsym{(}  \Pi'  \mGLsym{)} \, \leq \,  \mathsf{Rank}  (  \mGLnt{X}  )\\
            \delta  \boxast [ { \delta_{{\mathrm{3}}} }^{ \mGLmv{n} } ]   \mGLsym{=}  \delta'_{{\mathrm{3}}}\\
          \end{array}
          \]
        \item \textbf{Left introduction of tensor product: second case}
        \[
          \inferrule* [flushleft,right=,left=$\Pi_{{\mathrm{1}}} :$] {
            \pi_1
          }{\delta_{{\mathrm{2}}}  \odot  \Delta_{{\mathrm{2}}}  \vdash_{\mathsf{GS} }  \mGLnt{X}}
    \]

      \[
        \inferrule* [flushleft,right=$\mGLdruleGSTXXTenLName{}$,left=$\Pi_{{\mathrm{2}}} :$] {
          \inferrule* [flushleft,right=,left=$\Pi_{{\mathrm{3}}} :$] {
            \pi_3
          }{(  \delta_{{\mathrm{1}}}  \mGLsym{,}  \delta  \mGLsym{,}  \delta_{{\mathrm{3}}}  \mGLsym{,}  \mGLnt{r}  \mGLsym{,}  \mGLnt{r}  \mGLsym{,}  \delta_{{\mathrm{4}}}  )   \odot   ( \Delta_{{\mathrm{1}}}  \mGLsym{,}   \mGLnt{X} ^{ \mGLmv{n} }   \mGLsym{,}  \Delta_{{\mathrm{3}}}  \mGLsym{,}  \mGLnt{Y}  \mGLsym{,}  \mGLnt{Z}  \mGLsym{,}  \Delta_{{\mathrm{4}}} )   \vdash_{\mathsf{GS} }  \mGLnt{Z}}
          }{(  \delta_{{\mathrm{1}}}  \mGLsym{,}  \mGLnt{r}  \mGLsym{,}  \delta_{{\mathrm{2}}}  \mGLsym{,}  \delta  \mGLsym{,}  \delta_{{\mathrm{4}}}  )   \odot   ( \Delta_{{\mathrm{1}}}  \mGLsym{,}  \mGLnt{Y}  \boxtimes  \mGLnt{Z}  \mGLsym{,}  \Delta_{{\mathrm{2}}}  \mGLsym{,}   \mGLnt{X} ^{ \mGLmv{n} }   \mGLsym{,}  \Delta_{{\mathrm{4}}} )   \vdash_{\mathsf{GS} }  \mGLnt{Z}}
        \]
        Similar to the previous case.

        \item \textbf{Left introduction of the unit of tensor:}
        \[
          \inferrule* [flushleft,right=,left=$\Pi_{{\mathrm{1}}} :$] {
            \pi_1
          }{\delta_{{\mathrm{2}}}  \odot  \Delta_{{\mathrm{2}}}  \vdash_{\mathsf{GS} }  \mGLnt{X}}
    \]

      \[
        \inferrule* [flushleft,right=$\mGLdruleGSTXXUnitLName{}$,left=$\Pi_{{\mathrm{2}}} :$] {
          \inferrule* [flushleft,left=$\Pi_{{\mathrm{3}}} : $] {
            \pi_3
          }{(  \delta_{{\mathrm{1}}}  \mGLsym{,}  \delta  \mGLsym{,}  \delta_{{\mathrm{3}}}  \mGLsym{,}  \delta_{{\mathrm{4}}}  )   \odot   ( \Delta_{{\mathrm{1}}}  \mGLsym{,}   \mGLnt{X} ^{ \mGLmv{n} }   \mGLsym{,}  \Delta_{{\mathrm{3}}}  \mGLsym{,}  \Delta_{{\mathrm{4}}} )   \vdash_{\mathsf{GS} }  \mGLnt{Y}}\\
        }{(  \delta_{{\mathrm{1}}}  \mGLsym{,}  \delta  \mGLsym{,}  \delta_{{\mathrm{3}}}  \mGLsym{,}  \mGLnt{r}  \mGLsym{,}  \delta_{{\mathrm{4}}}  )   \odot   ( \Delta_{{\mathrm{1}}}  \mGLsym{,}   \mGLnt{X} ^{ \mGLmv{n} }   \mGLsym{,}  \Delta_{{\mathrm{3}}}  \mGLsym{,}  \mathsf{J}  \mGLsym{,}  \Delta_{{\mathrm{4}}} )   \vdash_{\mathsf{GS} }  \mGLnt{Y}}
        \]

      We know:
      \[
      \begin{array}{lll}
        \mathsf{Depth}  (  \Pi_{{\mathrm{1}}}  )   +   \mathsf{Depth}  (  \Pi_{{\mathrm{3}}}  )   \, \mGLsym{<} \,  \mathsf{Depth}  (  \Pi_{{\mathrm{1}}}  )   +   \mathsf{Depth}  (  \Pi_{{\mathrm{2}}}  )\\
        \mathsf{CutRank} \, \mGLsym{(}  \Pi_{{\mathrm{3}}}  \mGLsym{)} \, \leq \, \mathsf{CutRank} \, \mGLsym{(}  \Pi_{{\mathrm{2}}}  \mGLsym{)} \, \leq \,  \mathsf{Rank}  (  \mGLnt{X}  )
      \end{array}
      \]

      and so applying the induction hypothesis
      to $\Pi_{{\mathrm{1}}}$ and $\Pi_{{\mathrm{3}}}$
      implies that there is a proof $\Pi'$ of
      $(  \delta_{{\mathrm{1}}}  \mGLsym{,}  \delta'_{{\mathrm{2}}}  \mGLsym{,}  \delta_{{\mathrm{3}}}  \mGLsym{,}  \delta_{{\mathrm{4}}}  )   \odot   ( \Delta_{{\mathrm{1}}}  \mGLsym{,}  \Delta_{{\mathrm{2}}}  \mGLsym{,}  \Delta_{{\mathrm{3}}}  \mGLsym{,}  \Delta_{{\mathrm{4}}} )   \vdash_{\mathsf{GS} }  \mGLnt{Y}$ with
      $\mathsf{CutRank} \, \mGLsym{(}  \Pi'  \mGLsym{)} \, \leq \,  \mathsf{Rank}  (  \mGLnt{X}  )$ and $\delta  \boxast [ { \delta_{{\mathrm{2}}} }^{ \mGLmv{n} } ]   \mGLsym{=}  \delta'_{{\mathrm{2}}}$.
      Thus, we construct the following proof $\Pi$:
      \[
        \inferrule* [flushleft,right=$\mGLdruleGSTXXUnitLName{}$,left=$\Pi :$] {
          \inferrule* [flushleft,left=$\Pi' : $] {
            \pi'
          }{(  \delta_{{\mathrm{1}}}  \mGLsym{,}  \delta'_{{\mathrm{2}}}  \mGLsym{,}  \delta_{{\mathrm{3}}}  \mGLsym{,}  \delta_{{\mathrm{4}}}  )   \odot   ( \Delta_{{\mathrm{1}}}  \mGLsym{,}  \Delta_{{\mathrm{2}}}  \mGLsym{,}  \Delta_{{\mathrm{3}}}  \mGLsym{,}  \Delta_{{\mathrm{4}}} )   \vdash_{\mathsf{GS} }  \mGLnt{Y}}\\
        }{(  \delta_{{\mathrm{1}}}  \mGLsym{,}  \delta'_{{\mathrm{2}}}  \mGLsym{,}  \delta_{{\mathrm{3}}}  \mGLsym{,}  \mGLnt{r}  \mGLsym{,}  \delta_{{\mathrm{4}}}  )   \odot   ( \Delta_{{\mathrm{1}}}  \mGLsym{,}  \Delta_{{\mathrm{2}}}  \mGLsym{,}  \Delta_{{\mathrm{3}}}  \mGLsym{,}  \mathsf{J}  \mGLsym{,}  \Delta_{{\mathrm{4}}} )   \vdash_{\mathsf{GS} }  \mGLnt{Y}}
        \]
        Given the above, we know:
        \[
          \begin{array}{lll}
            \mathsf{CutRank} \, \mGLsym{(}  \Pi  \mGLsym{)} \, \mGLsym{=} \, \mathsf{CutRank} \, \mGLsym{(}  \Pi'  \mGLsym{)} \, \leq \,  \mathsf{Rank}  (  \mGLnt{X}  )\\
            \delta  \boxast [ { \delta_{{\mathrm{2}}} }^{ \mGLmv{n} } ]   \mGLsym{=}  \delta'_{{\mathrm{2}}}\\
          \end{array}
          \]
        \item \textbf{Left introduction of the unit of tensor: second case}
        \[
          \inferrule* [flushleft,right=,left=$\Pi_{{\mathrm{1}}} :$] {
            \pi_1
          }{\delta_{{\mathrm{3}}}  \odot  \Delta_{{\mathrm{3}}}  \vdash_{\mathsf{GS} }  \mGLnt{X}}
    \]
      \[
        \inferrule* [flushleft,right=$\mGLdruleGSTXXUnitLName{}$,left=$\Pi_{{\mathrm{2}}} :$] {
          \inferrule* [flushleft,left=$\Pi_{{\mathrm{3}}} : $] {
            \pi_3
          }{(  \delta_{{\mathrm{1}}}  \mGLsym{,}  \delta_{{\mathrm{2}}}  \mGLsym{,}  \delta  \mGLsym{,}  \delta_{{\mathrm{4}}}  )   \odot   ( \Delta_{{\mathrm{1}}}  \mGLsym{,}  \Delta_{{\mathrm{2}}}  \mGLsym{,}   \mGLnt{X} ^{ \mGLmv{n} }   \mGLsym{,}  \Delta_{{\mathrm{4}}} )   \vdash_{\mathsf{GS} }  \mGLnt{Y}}\\
        }{(  \delta_{{\mathrm{1}}}  \mGLsym{,}  \mGLnt{r}  \mGLsym{,}  \delta_{{\mathrm{2}}}  \mGLsym{,}  \delta  \mGLsym{,}  \delta_{{\mathrm{4}}}  )   \odot   ( \Delta_{{\mathrm{1}}}  \mGLsym{,}  \mathsf{J}  \mGLsym{,}  \Delta_{{\mathrm{2}}}  \mGLsym{,}   \mGLnt{X} ^{ \mGLmv{n} }   \mGLsym{,}  \Delta_{{\mathrm{4}}} )   \vdash_{\mathsf{GS} }  \mGLnt{Y}}
        \]
        Similar to the previous case.
        \item \textbf{Weakening:}
        \[
          \inferrule* [flushleft,right=,left=$\Pi_{{\mathrm{1}}} :$] {
            \pi_1
          }{\delta_{{\mathrm{2}}}  \odot  \Delta_{{\mathrm{2}}}  \vdash_{\mathsf{GS} }  \mGLnt{X}}
    \]

      \[
        \inferrule* [flushleft,right=$\mGLdruleGSTXXWeakName{}$,left=$\Pi_{{\mathrm{2}}} :$] {
          \inferrule* [flushleft,left=$\Pi_{{\mathrm{3}}} : $] {
            \pi_3
          }{(  \delta_{{\mathrm{1}}}  \mGLsym{,}  \delta  \mGLsym{,}  \delta_{{\mathrm{3}}}  \mGLsym{,}  \delta_{{\mathrm{4}}}  )   \odot   ( \Delta_{{\mathrm{1}}}  \mGLsym{,}   \mGLnt{X} ^{ \mGLmv{n} }   \mGLsym{,}  \Delta_{{\mathrm{3}}}  \mGLsym{,}  \Delta_{{\mathrm{4}}} )   \vdash_{\mathsf{GS} }  \mGLnt{Y}}\\
        }{(  \delta_{{\mathrm{1}}}  \mGLsym{,}  \delta  \mGLsym{,}  \delta_{{\mathrm{3}}}  \mGLsym{,}  \mathsf{0}  \mGLsym{,}  \delta_{{\mathrm{4}}}  )   \odot   ( \Delta_{{\mathrm{1}}}  \mGLsym{,}   \mGLnt{X} ^{ \mGLmv{n} }   \mGLsym{,}  \Delta_{{\mathrm{3}}}  \mGLsym{,}  \mGLnt{Z}  \mGLsym{,}  \Delta_{{\mathrm{4}}} )   \vdash_{\mathsf{GS} }  \mGLnt{Y}}
        \]
        Similar to the unit tensor case.
        \item \textbf{Weakening: second case}
        \[
          \inferrule* [flushleft,right=,left=$\Pi_{{\mathrm{1}}} :$] {
            \pi_1
          }{\delta_{{\mathrm{3}}}  \odot  \Delta_{{\mathrm{3}}}  \vdash_{\mathsf{GS} }  \mGLnt{X}}
    \]
      \[
        \inferrule* [flushleft,right=,left=$\Pi_{{\mathrm{2}}} :$] {
          \inferrule* [flushleft,left=$\Pi_{{\mathrm{3}}} : $] {
            \pi_3
          }{(  \delta_{{\mathrm{1}}}  \mGLsym{,}  \delta_{{\mathrm{2}}}  \mGLsym{,}  \delta  \mGLsym{,}  \delta_{{\mathrm{4}}}  )   \odot   ( \Delta_{{\mathrm{1}}}  \mGLsym{,}  \Delta_{{\mathrm{2}}}  \mGLsym{,}   \mGLnt{X} ^{ \mGLmv{n} }   \mGLsym{,}  \Delta_{{\mathrm{4}}} )   \vdash_{\mathsf{GS} }  \mGLnt{Y}}\\
        }{(  \delta_{{\mathrm{1}}}  \mGLsym{,}  \mathsf{0}  \mGLsym{,}  \delta_{{\mathrm{2}}}  \mGLsym{,}  \delta  \mGLsym{,}  \delta_{{\mathrm{4}}}  )   \odot   ( \Delta_{{\mathrm{1}}}  \mGLsym{,}  \mGLnt{Z}  \mGLsym{,}  \Delta_{{\mathrm{2}}}  \mGLsym{,}   \mGLnt{X} ^{ \mGLmv{n} }   \mGLsym{,}  \Delta_{{\mathrm{4}}} )   \vdash_{\mathsf{GS} }  \mGLnt{Y}}
        \]
        Similar to the unit tensor case.
        \item \textbf{Right introduction of Lin:} %
        \[
          \begin{array}{lll}
            \inferrule* [flushleft,right=,left=$\Pi_{{\mathrm{1}}} :$] {
              \pi_1
            }{\delta_{{\mathrm{2}}}  \odot  \Delta_{{\mathrm{2}}}  \vdash_{\mathsf{GS} }  \mGLnt{X}}
            & \quad &
            \inferrule* [flushleft,right=$\mGLdruleGSTXXLinRName{}$,left=$\Pi_{{\mathrm{2}}} :$] {
              \inferrule* [flushleft,right=, left=$\Pi_{{\mathrm{3}}} :$] {
              \pi_3
              }{(  \delta_{{\mathrm{1}}}  \mGLsym{,}  \delta  \mGLsym{,}  \delta_{{\mathrm{3}}}  )   \odot   ( \Delta_{{\mathrm{1}}}  \mGLsym{,}   \mGLnt{X} ^{ \mGLmv{n} }   \mGLsym{,}  \Delta_{{\mathrm{3}}} )   \mGLsym{;}  \emptyset  \vdash_{\mathsf{MS} }  \mGLnt{A}}
            }{(  \delta_{{\mathrm{1}}}  \mGLsym{,}  \delta  \mGLsym{,}  \delta_{{\mathrm{3}}}  )   \odot   ( \Delta_{{\mathrm{1}}}  \mGLsym{,}   \mGLnt{X} ^{ \mGLmv{n} }   \mGLsym{,}  \Delta_{{\mathrm{3}}} )   \vdash_{\mathsf{GS} }  \mathsf{Lin} \, \mGLnt{A}}
          \end{array}
          \]
          We know the following:
          \[
            \begin{array}{lll}
              \mathsf{Depth}  (  \Pi_{{\mathrm{1}}}  )   +   \mathsf{Depth}  (  \Pi_{{\mathrm{3}}}  )   \, \mGLsym{<} \,  \mathsf{Depth}  (  \Pi_{{\mathrm{1}}}  )   +   \mathsf{Depth}  (  \Pi_{{\mathrm{2}}}  )\\
              \mathsf{CutRank} \, \mGLsym{(}  \Pi_{{\mathrm{3}}}  \mGLsym{)} \, \leq \, \mathsf{CutRank} \, \mGLsym{(}  \Pi_{{\mathrm{2}}}  \mGLsym{)} \, \leq \,  \mathsf{Rank}  (  \mGLnt{X}  )
            \end{array}
          \]
          Thus, we apply the induction hypothesis of Lemma~\ref{lemma:cut_reduction_for_mgl} (2)
          to $\Pi_{{\mathrm{1}}}$ and $\Pi_{{\mathrm{3}}}$ to obtain a proof $\Pi'$ of the sequent
          $(  \delta_{{\mathrm{1}}}  \mGLsym{,}  \delta'_{{\mathrm{2}}}  \mGLsym{,}  \delta_{{\mathrm{3}}}  )   \odot   ( \Delta_{{\mathrm{1}}}  \mGLsym{,}  \Delta_{{\mathrm{2}}}  \mGLsym{,}  \Delta_{{\mathrm{3}}} )   \mGLsym{;}  \emptyset  \vdash_{\mathsf{MS} }  \mGLnt{A}$ with $\mathsf{CutRank} \, \mGLsym{(}  \Pi'  \mGLsym{)} \, \leq \,  \mathsf{Rank}  (  \mGLsym{(}  \mGLnt{X}  \mGLsym{)}  )$
          and $(   \delta  \boxast [ { \delta_{{\mathrm{2}}} }^{ \mGLmv{n} } ]   )   \mGLsym{=}  \delta'_{{\mathrm{2}}}$.
          Now we define the proof $\Pi$ as follows:
          \[
            \inferrule* [flushleft,right=$\mGLdruleGSTXXLinRName{}$, left=$\Pi' :$] {
              \inferrule* [flushleft,right=, left=$\Pi :$] {
                \pi'
              }{(  \delta_{{\mathrm{1}}}  \mGLsym{,}  \delta'_{{\mathrm{2}}}  \mGLsym{,}  \delta_{{\mathrm{3}}}  )   \odot   ( \Delta_{{\mathrm{1}}}  \mGLsym{,}  \Delta_{{\mathrm{2}}}  \mGLsym{,}  \Delta_{{\mathrm{3}}} )   \mGLsym{;}  \emptyset  \vdash_{\mathsf{MS} }  \mGLnt{A}}
            }{(  \delta_{{\mathrm{1}}}  \mGLsym{,}  \delta'_{{\mathrm{2}}}  \mGLsym{,}  \delta_{{\mathrm{3}}}  )   \odot   ( \Delta_{{\mathrm{1}}}  \mGLsym{,}  \Delta_{{\mathrm{2}}}  \mGLsym{,}  \Delta_{{\mathrm{3}}} )   \vdash_{\mathsf{GS} }  \mathsf{Lin} \, \mGLnt{A}}
            \]
          with: $\mathsf{CutRank} \, \mGLsym{(}  \Pi  \mGLsym{)} \, \mGLsym{=} \, \mathsf{CutRank} \, \mGLsym{(}  \Pi'  \mGLsym{)} \, \leq \,  \mathsf{Rank}  (  \mGLnt{X}  )$ and $(   \delta  \boxast [ { \delta_{{\mathrm{2}}} }^{ \mGLmv{n} } ]   )   \mGLsym{=}  \delta'_{{\mathrm{2}}}$

      \end{enumerate}
    \item \textbf{Structural}
      \begin{enumerate}
        \item \textbf{Weakening}
        \[
          \inferrule* [flushleft,right=,left=$\Pi_{{\mathrm{1}}} :$] {
            \pi_1
          }{\delta_{{\mathrm{2}}}  \odot  \Delta_{{\mathrm{2}}}  \vdash_{\mathsf{GS} }  \mGLnt{X}}
    \]
      \[
        \inferrule* [flushleft,right=$\mGLdruleGSTXXWeakName{}$,left=$\Pi_{{\mathrm{2}}} :$] {
          \inferrule* [flushleft,left=$\Pi_{{\mathrm{3}}} : $] {
            \pi_3
          }{(  \delta_{{\mathrm{1}}}  \mGLsym{,}  \delta_{{\mathrm{3}}}  )   \odot   ( \Delta_{{\mathrm{1}}}  \mGLsym{,}  \Delta_{{\mathrm{3}}} )   \vdash_{\mathsf{GS} }  \mGLnt{Y}}\\
        }{(  \delta_{{\mathrm{1}}}  \mGLsym{,}  \mathsf{0}  \mGLsym{,}  \delta_{{\mathrm{3}}}  )   \odot   ( \Delta_{{\mathrm{1}}}  \mGLsym{,}  \mGLnt{X}  \mGLsym{,}  \Delta_{{\mathrm{3}}} )   \vdash_{\mathsf{GS} }  \mGLnt{Y}}
        \]

      We know:
      \[
      \begin{array}{lll}
        \mathsf{CutRank} \, \mGLsym{(}  \Pi_{{\mathrm{3}}}  \mGLsym{)} \, \leq \, \mathsf{CutRank} \, \mGLsym{(}  \Pi_{{\mathrm{2}}}  \mGLsym{)} \, \leq \,  \mathsf{Rank}  (  \mGLnt{X}  )
      \end{array}
      \]

      We directly construct the proof $\Pi$ of
      $\delta_{{\mathrm{1}}}  \mGLsym{,}   (   \mathsf{0}  \boxast [ { \delta_{{\mathrm{2}}} }^{ \mGLmv{n} } ]   )   \mGLsym{,}  \delta_{{\mathrm{3}}}  \odot  \Delta_{{\mathrm{1}}}  \mGLsym{,}  \Delta_{{\mathrm{2}}}  \mGLsym{,}  \Delta_{{\mathrm{3}}}  \vdash_{\mathsf{GS} }  \mGLnt{Y}$
      by $| \Delta_{{\mathrm{2}}}| = m$ applications of weakening.

      \[
        \inferrule* [flushleft,right=$\mGLdruleGSTXXWeakName{}$,left=$\Pi :$] {
          \inferrule* [flushleft,right=$\mGLdruleGSTXXWeakName{}$] {
            \inferrule* [flushleft,right=$\mGLdruleGSTXXWeakName{}$,left=$\Pi_{{\mathrm{3}}} : $] {
              \pi_3
          }{(  \delta_{{\mathrm{1}}}  \mGLsym{,}  \delta_{{\mathrm{3}}}  )   \odot   ( \Delta_{{\mathrm{1}}}  \mGLsym{,}  \Delta_{{\mathrm{3}}} )   \vdash_{\mathsf{GS} }  \mGLnt{Y}}
          }{\vdots}
        }{(  \delta_{{\mathrm{1}}}  \mGLsym{,}   (   \mathsf{0}  \boxast [ { \delta_{{\mathrm{2}}} }^{ \mGLmv{n} } ]   )   \mGLsym{,}  \delta_{{\mathrm{3}}}  )   \odot   ( \Delta_{{\mathrm{1}}}  \mGLsym{,}  \Delta_{{\mathrm{2}}}  \mGLsym{,}  \Delta_{{\mathrm{3}}} )   \vdash_{\mathsf{GS} }  \mGLnt{Y}}
        \]

    with: $\mathsf{CutRank} \, \mGLsym{(}  \Pi  \mGLsym{)} \, \mGLsym{=} \, \mathsf{CutRank} \, \mGLsym{(}  \Pi_{{\mathrm{2}}}  \mGLsym{)} \, \leq \,  \mathsf{Rank}  (  \mGLnt{X}  )$
        \item \textbf{Contraction}
        \[
          \inferrule* [flushleft,right=,left=$\Pi_{{\mathrm{1}}} :$] {
            \pi_1
          }{\delta_{{\mathrm{2}}}  \odot  \Delta_{{\mathrm{2}}}  \vdash_{\mathsf{GS} }  \mGLnt{X}}
    \]
    \[
      \inferrule* [flushleft,right=$\mGLdruleGSTXXContName{}$,left=$\Pi_{{\mathrm{2}}} :$] {
        \inferrule* [flushleft,right=,left=$\Pi_{{\mathrm{3}}} :$] {
          \pi_3
        }{(  \delta_{{\mathrm{1}}}  \mGLsym{,}  \mGLnt{r_{{\mathrm{1}}}}  \mGLsym{,}  \mGLnt{r_{{\mathrm{2}}}}  \mGLsym{,}  \delta_{{\mathrm{3}}}  )   \odot   ( \Delta_{{\mathrm{1}}}  \mGLsym{,}  \mGLnt{X}  \mGLsym{,}  \mGLnt{X}  \mGLsym{,}  \Delta_{{\mathrm{3}}} )   \vdash_{\mathsf{GS} }  \mGLnt{Y}}
      }{(  \delta_{{\mathrm{1}}}  \mGLsym{,}  \mGLnt{r_{{\mathrm{1}}}}  +  \mGLnt{r_{{\mathrm{2}}}}  \mGLsym{,}  \delta_{{\mathrm{3}}}  )   \odot   ( \Delta_{{\mathrm{1}}}  \mGLsym{,}  \mGLnt{X}  \mGLsym{,}  \Delta_{{\mathrm{3}}} )   \vdash_{\mathsf{GS} }  \mGLnt{Y}}
      \]
      We know the following:
      \[
        \begin{array}{lll}
          \mathsf{Depth}  (  \Pi_{{\mathrm{1}}}  )   +   \mathsf{Depth}  (  \Pi_{{\mathrm{3}}}  )   \, \mGLsym{<} \,  \mathsf{Depth}  (  \Pi_{{\mathrm{1}}}  )   +   \mathsf{Depth}  (  \Pi_{{\mathrm{2}}}  )\\
          \mathsf{CutRank} \, \mGLsym{(}  \Pi_{{\mathrm{3}}}  \mGLsym{)} \, \leq \, \mathsf{CutRank} \, \mGLsym{(}  \Pi_{{\mathrm{2}}}  \mGLsym{)} \, \leq \,  \mathsf{Rank}  (  \mGLnt{X}  )
        \end{array}
      \]
      Thus, we apply the induction hypothesis
      to $\Pi_{{\mathrm{1}}}$ and $\Pi_{{\mathrm{3}}}$ to obtain a proof $\Pi'$ of the sequent
      $(  \delta_{{\mathrm{1}}}  \mGLsym{,}  \delta'_{{\mathrm{2}}}  \mGLsym{,}  \delta_{{\mathrm{3}}}  )   \odot   ( \Delta_{{\mathrm{1}}}  \mGLsym{,}  \Delta_{{\mathrm{2}}}  \mGLsym{,}  \Delta_{{\mathrm{3}}} )   \vdash_{\mathsf{GS} }  \mGLnt{Y}$ with $\mathsf{CutRank} \, \mGLsym{(}  \Pi'  \mGLsym{)} \, \leq \,  \mathsf{Rank}  (  \mGLsym{(}  \mGLnt{X}  \mGLsym{)}  )$
      and $(    (  \mGLnt{r_{{\mathrm{1}}}}  \mGLsym{,}  \mGLnt{r_{{\mathrm{2}}}}  )   \boxast [ { \delta_{{\mathrm{2}}} }^{ \mGLmv{n} } ]   )   \mGLsym{=}  \delta'_{{\mathrm{2}}}$.
      Since $(  \mGLnt{r_{{\mathrm{1}}}}  \mGLsym{,}  \mGLnt{r_{{\mathrm{2}}}}  )   \boxast [ { \delta_{{\mathrm{2}}} }^{ \mGLmv{n} } ]   \mGLsym{=}   (  \mGLnt{r_{{\mathrm{1}}}}  *  \delta_{{\mathrm{2}}}  )   +   (  \mGLnt{r_{{\mathrm{2}}}}  *  \delta_{{\mathrm{2}}}  )$ by definition, and
      for any $k, 1<= k<= | \delta_2 | $
      \[
        \begin{array}{lll}
          (  \mGLnt{r_{{\mathrm{1}}}}  *  \delta_{{\mathrm{2}}}  )   +   (  \mGLnt{r_{{\mathrm{2}}}}  *  \delta_{{\mathrm{2}}}  )   \mGLsym{(}  \mGLmv{k}  \mGLsym{)}  \mGLsym{=}   (  \mGLnt{r_{{\mathrm{1}}}}  *  \delta_{{\mathrm{2}}}  )   \mGLsym{(}  \mGLmv{k}  \mGLsym{)}  +   (  \mGLnt{r_{{\mathrm{2}}}}  *  \delta_{{\mathrm{2}}}  )   \mGLsym{(}  \mGLmv{k}  \mGLsym{)}  \mGLsym{=}  \mGLnt{r_{{\mathrm{1}}}}  *  \delta_{{\mathrm{2}}}  \mGLsym{(}  \mGLmv{k}  \mGLsym{)}  +  \mGLnt{r_{{\mathrm{2}}}}  *  \delta_{{\mathrm{2}}}  \mGLsym{(}  \mGLmv{k}  \mGLsym{)}\\
          \mGLnt{r_{{\mathrm{1}}}}  *  \delta_{{\mathrm{2}}}  \mGLsym{(}  \mGLmv{k}  \mGLsym{)}  +  \mGLnt{r_{{\mathrm{2}}}}  *  \delta_{{\mathrm{2}}}  \mGLsym{(}  \mGLmv{k}  \mGLsym{)}  \mGLsym{=}   (  \mGLnt{r_{{\mathrm{1}}}}  +  \mGLnt{r_{{\mathrm{2}}}}  )   *  \delta_{{\mathrm{2}}}  \mGLsym{(}  \mGLmv{k}  \mGLsym{)}\\
          (  \mGLnt{r_{{\mathrm{1}}}}  \mGLsym{,}  \mGLnt{r_{{\mathrm{2}}}}  )   \boxast [ { \delta_{{\mathrm{2}}} }^{ \mGLmv{n} } ]   \mGLsym{=}   (  \mGLnt{r_{{\mathrm{1}}}}  +  \mGLnt{r_{{\mathrm{2}}}}  )   *  \delta_{{\mathrm{2}}} \\
          (  \mGLnt{r_{{\mathrm{1}}}}  +  \mGLnt{r_{{\mathrm{2}}}}  )   *  \delta_{{\mathrm{2}}}  \mGLsym{=}    (  \mGLnt{r_{{\mathrm{1}}}}  +  \mGLnt{r_{{\mathrm{2}}}}  )   \boxast [ { \delta_{{\mathrm{2}}} }^{ \mGLmv{n} } ]\\
          (  \mGLnt{r_{{\mathrm{1}}}}  +  \mGLnt{r_{{\mathrm{2}}}}  )   *  \delta_{{\mathrm{2}}}  \mGLsym{=}  \delta'_{{\mathrm{2}}}
        \end{array}
      \]
        So we construct the proof $\Pi \, \mGLsym{=} \, \Pi'$
        \item \textbf{Exchange}
        \[
          \inferrule* [flushleft,right=,left=$\Pi_{{\mathrm{1}}} :$] {
            \pi_1
          }{\delta_{{\mathrm{2}}}  \odot  \Delta_{{\mathrm{2}}}  \vdash_{\mathsf{GS} }  \mGLnt{X}}
    \]
    \[
      \inferrule* [flushleft,right=$\mGLdruleGSTXXExName{}$,left=$\Pi_{{\mathrm{2}}} :$] {
        \inferrule* [flushleft,right=,left=$\Pi_{{\mathrm{3}}} :$] {
          \pi_3
        }{(  \delta_{{\mathrm{1}}}  \mGLsym{,}  \delta  \mGLsym{,}  \delta_{{\mathrm{3}}}  \mGLsym{,}  \mGLnt{r_{{\mathrm{1}}}}  \mGLsym{,}  \mGLnt{r_{{\mathrm{2}}}}  \mGLsym{,}  \delta_{{\mathrm{4}}}  )   \odot   ( \Delta_{{\mathrm{1}}}  \mGLsym{,}   \mGLnt{X} ^{ \mGLmv{n} }   \mGLsym{,}  \Delta_{{\mathrm{3}}}  \mGLsym{,}  \mGLnt{Y}  \mGLsym{,}  \mGLnt{Z}  \mGLsym{,}  \Delta_{{\mathrm{4}}} )   \vdash_{\mathsf{GS} }  \mGLnt{W}}
      }{(  \delta_{{\mathrm{1}}}  \mGLsym{,}  \delta  \mGLsym{,}  \delta_{{\mathrm{3}}}  \mGLsym{,}  \mGLnt{r_{{\mathrm{2}}}}  \mGLsym{,}  \mGLnt{r_{{\mathrm{1}}}}  \mGLsym{,}  \delta_{{\mathrm{4}}}  )   \odot   ( \Delta_{{\mathrm{1}}}  \mGLsym{,}   \mGLnt{X} ^{ \mGLmv{n} }   \mGLsym{,}  \Delta_{{\mathrm{3}}}  \mGLsym{,}  \mGLnt{Z}  \mGLsym{,}  \mGLnt{Y}  \mGLsym{,}  \Delta_{{\mathrm{4}}} )   \vdash_{\mathsf{GS} }  \mGLnt{W}}
      \]

      We know:
      \[
      \begin{array}{lll}
        \mathsf{Depth}  (  \Pi_{{\mathrm{1}}}  )   +   \mathsf{Depth}  (  \Pi_{{\mathrm{3}}}  )   \, \mGLsym{<} \,  \mathsf{Depth}  (  \Pi_{{\mathrm{1}}}  )   +   \mathsf{Depth}  (  \Pi_{{\mathrm{2}}}  )\\
        \mathsf{CutRank} \, \mGLsym{(}  \Pi_{{\mathrm{3}}}  \mGLsym{)} \, \leq \, \mathsf{CutRank} \, \mGLsym{(}  \Pi_{{\mathrm{2}}}  \mGLsym{)} \, \leq \,  \mathsf{Rank}  (  \mGLnt{X}  )
      \end{array}
      \]

      and so applying the induction hypothesis
      to $\Pi_{{\mathrm{1}}}$ and $\Pi_{{\mathrm{3}}}$
      implies that there is a proof $\Pi'$ of
      $(  \delta_{{\mathrm{1}}}  \mGLsym{,}  \delta'_{{\mathrm{2}}}  \mGLsym{,}  \delta_{{\mathrm{3}}}  \mGLsym{,}  \mGLnt{r_{{\mathrm{1}}}}  \mGLsym{,}  \mGLnt{r_{{\mathrm{2}}}}  \mGLsym{,}  \delta_{{\mathrm{4}}}  )   \odot   ( \Delta_{{\mathrm{1}}}  \mGLsym{,}  \Delta_{{\mathrm{2}}}  \mGLsym{,}  \Delta_{{\mathrm{3}}}  \mGLsym{,}  \mGLnt{Y}  \mGLsym{,}  \mGLnt{Z}  \mGLsym{,}  \Delta_{{\mathrm{4}}} )   \vdash_{\mathsf{GS} }  \mGLnt{W}$ with
      $\mathsf{CutRank} \, \mGLsym{(}  \Pi'  \mGLsym{)} \, \leq \,  \mathsf{Rank}  (  \mGLnt{X}  )$ and $(   \delta  \boxast [ { \delta_{{\mathrm{2}}} }^{ \mGLmv{n} } ]   )   \mGLsym{=}  \delta'_{{\mathrm{2}}}$.
      Thus, we construct the following proof $\Pi$:

      \[
        \inferrule* [flushleft,right=$\mGLdruleGSTXXExName{}$,left=$\Pi :$] {
          \inferrule* [flushleft,right=,left=$\Pi' :$] {
            \pi'
          }{(  \delta_{{\mathrm{1}}}  \mGLsym{,}  \delta'_{{\mathrm{2}}}  \mGLsym{,}  \delta_{{\mathrm{3}}}  \mGLsym{,}  \mGLnt{r_{{\mathrm{1}}}}  \mGLsym{,}  \mGLnt{r_{{\mathrm{2}}}}  \mGLsym{,}  \delta_{{\mathrm{4}}}  )   \odot   ( \Delta_{{\mathrm{1}}}  \mGLsym{,}  \Delta_{{\mathrm{2}}}  \mGLsym{,}  \Delta_{{\mathrm{3}}}  \mGLsym{,}  \mGLnt{Y}  \mGLsym{,}  \mGLnt{Z}  \mGLsym{,}  \Delta_{{\mathrm{4}}} )   \vdash_{\mathsf{GS} }  \mGLnt{W}}
        }{(  \delta_{{\mathrm{1}}}  \mGLsym{,}  \delta'_{{\mathrm{2}}}  \mGLsym{,}  \delta_{{\mathrm{3}}}  \mGLsym{,}  \mGLnt{r_{{\mathrm{2}}}}  \mGLsym{,}  \mGLnt{r_{{\mathrm{1}}}}  \mGLsym{,}  \delta_{{\mathrm{4}}}  )   \odot   ( \Delta_{{\mathrm{1}}}  \mGLsym{,}  \Delta_{{\mathrm{2}}}  \mGLsym{,}  \Delta_{{\mathrm{3}}}  \mGLsym{,}  \mGLnt{Z}  \mGLsym{,}  \mGLnt{Y}  \mGLsym{,}  \Delta_{{\mathrm{4}}} )   \vdash_{\mathsf{GS} }  \mGLnt{W}}
        \]
        Given the above, we know:
        \[
          \begin{array}{lll}
            \mathsf{CutRank} \, \mGLsym{(}  \Pi  \mGLsym{)} \, \mGLsym{=} \, \mathsf{CutRank} \, \mGLsym{(}  \Pi'  \mGLsym{)} \, \leq \,  \mathsf{Rank}  (  \mGLnt{X}  )\\
            \delta  \boxast [ { \delta_{{\mathrm{2}}} }^{ \mGLmv{n} } ]   \mGLsym{=}  \delta'_{{\mathrm{2}}}\\
          \end{array}
          \]
        \item \textbf{Exchange (second case)}
        \[
          \inferrule* [flushleft,right=,left=$\Pi_{{\mathrm{1}}} :$] {
            \pi_1
          }{\delta_{{\mathrm{3}}}  \odot  \Delta_{{\mathrm{3}}}  \vdash_{\mathsf{GS} }  \mGLnt{X}}
    \]
    \[
      \inferrule* [flushleft,right=$\mGLdruleGSTXXExName{}$,left=$\Pi_{{\mathrm{2}}} :$] {
        \inferrule* [flushleft,right=,left=$\Pi_{{\mathrm{3}}} :$] {
          \pi_3
        }{(  \delta_{{\mathrm{1}}}  \mGLsym{,}  \delta  \mGLsym{,}  \delta_{{\mathrm{3}}}  \mGLsym{,}  \mGLnt{r_{{\mathrm{1}}}}  \mGLsym{,}  \mGLnt{r_{{\mathrm{2}}}}  \mGLsym{,}  \delta_{{\mathrm{4}}}  )   \odot   ( \Delta_{{\mathrm{1}}}  \mGLsym{,}  \Delta_{{\mathrm{2}}}  \mGLsym{,}  \mGLnt{Y}  \mGLsym{,}  \mGLnt{Z}  \mGLsym{,}   \mGLnt{X} ^{ \mGLmv{n} }   \mGLsym{,}  \Delta_{{\mathrm{4}}} )   \vdash_{\mathsf{GS} }  \mGLnt{W}}
      }{(  \delta_{{\mathrm{1}}}  \mGLsym{,}  \delta  \mGLsym{,}  \delta_{{\mathrm{3}}}  \mGLsym{,}  \mGLnt{r_{{\mathrm{2}}}}  \mGLsym{,}  \mGLnt{r_{{\mathrm{1}}}}  \mGLsym{,}  \delta_{{\mathrm{4}}}  )   \odot   ( \Delta_{{\mathrm{1}}}  \mGLsym{,}  \Delta_{{\mathrm{2}}}  \mGLsym{,}  \mGLnt{Z}  \mGLsym{,}  \mGLnt{Y}  \mGLsym{,}   \mGLnt{X} ^{ \mGLmv{n} }   \mGLsym{,}  \Delta_{{\mathrm{4}}} )   \vdash_{\mathsf{GS} }  \mGLnt{W}}
      \]
      This is similar to the above case
        \item \textbf{Exchange (third case)}
        \[
          \inferrule* [flushleft,right=,left=$\Pi_{{\mathrm{1}}} :$] {
            \pi_1
          }{\delta_{{\mathrm{2}}}  \odot  \Delta_{{\mathrm{2}}}  \vdash_{\mathsf{GS} }  \mGLnt{X}}
    \]
    \[
      \inferrule* [flushleft,right=$\mGLdruleGSTXXExName{}$,left=$\Pi_{{\mathrm{2}}} :$] {
        \inferrule* [flushleft,right=,left=$\Pi_{{\mathrm{3}}} :$] {
          \pi_3
        }{(  \delta_{{\mathrm{1}}}  \mGLsym{,}  \mGLnt{r_{{\mathrm{1}}}}  \mGLsym{,}  \mGLnt{r_{{\mathrm{2}}}}  \mGLsym{,}  \delta_{{\mathrm{3}}}  )   \odot   ( \Delta_{{\mathrm{1}}}  \mGLsym{,}  \mGLnt{X}  \mGLsym{,}  \mGLnt{Y}  \mGLsym{,}  \Delta_{{\mathrm{3}}} )   \vdash_{\mathsf{GS} }  \mGLnt{Z}}
      }{(  \delta_{{\mathrm{1}}}  \mGLsym{,}  \mGLnt{r_{{\mathrm{2}}}}  \mGLsym{,}  \mGLnt{r_{{\mathrm{1}}}}  \mGLsym{,}  \delta_{{\mathrm{3}}}  )   \odot   ( \Delta_{{\mathrm{1}}}  \mGLsym{,}  \mGLnt{Y}  \mGLsym{,}  \mGLnt{X}  \mGLsym{,}  \Delta_{{\mathrm{3}}} )   \vdash_{\mathsf{GS} }  \mGLnt{Z}}
      \]
      We know:
      \[
      \begin{array}{lll}
        \mathsf{Depth}  (  \Pi_{{\mathrm{1}}}  )   +   \mathsf{Depth}  (  \Pi_{{\mathrm{3}}}  )   \, \mGLsym{<} \,  \mathsf{Depth}  (  \Pi_{{\mathrm{1}}}  )   +   \mathsf{Depth}  (  \Pi_{{\mathrm{2}}}  )\\
        \mathsf{CutRank} \, \mGLsym{(}  \Pi_{{\mathrm{3}}}  \mGLsym{)} \, \leq \, \mathsf{CutRank} \, \mGLsym{(}  \Pi_{{\mathrm{2}}}  \mGLsym{)} \, \leq \,  \mathsf{Rank}  (  \mGLnt{X}  )
      \end{array}
      \]

      and so applying the induction hypothesis
      to $\Pi_{{\mathrm{1}}}$ and $\Pi_{{\mathrm{3}}}$
      implies that there is a proof $\Pi'$ of
      $(  \delta_{{\mathrm{1}}}  \mGLsym{,}  \delta'_{{\mathrm{2}}}  \mGLsym{,}  \mGLnt{r_{{\mathrm{2}}}}  \mGLsym{,}  \delta_{{\mathrm{3}}}  )   \odot   ( \Delta_{{\mathrm{1}}}  \mGLsym{,}  \Delta_{{\mathrm{2}}}  \mGLsym{,}  \mGLnt{Y}  \mGLsym{,}  \Delta_{{\mathrm{3}}} )   \vdash_{\mathsf{GS} }  \mGLnt{Z}$ with
      $\mathsf{CutRank} \, \mGLsym{(}  \Pi'  \mGLsym{)} \, \leq \,  \mathsf{Rank}  (  \mGLnt{X}  )$ and $(   \mGLnt{r_{{\mathrm{1}}}}  \boxast [ { \delta_{{\mathrm{2}}} }^{ \mGLmv{n} } ]   )   \mGLsym{=}  \delta'_{{\mathrm{2}}}$.
      Thus, we construct the following proof $\Pi \, \mGLsym{=} \, \Pi'$:

        \item \textbf{Approximation}
        \[
          \inferrule* [flushleft,right=,left=$\Pi_{{\mathrm{1}}} :$] {
            \pi_1
          }{\delta_{{\mathrm{2}}}  \odot  \Delta_{{\mathrm{2}}}  \vdash_{\mathsf{GS} }  \mGLnt{X}}
    \]
    \[
      \inferrule* [flushleft,right=$\mGLdruleGSTXXSubName{}$,left=$\Pi_{{\mathrm{2}}} :$] {
        \inferrule* [flushleft,right=,left=$\Pi_{{\mathrm{3}}} :$] {
          \pi_3
        }{(  \delta_{{\mathrm{1}}}  \mGLsym{,}  \delta'  \mGLsym{,}  \delta_{{\mathrm{3}}}  )   \odot   ( \Delta_{{\mathrm{1}}}  \mGLsym{,}   \mGLnt{X} ^{ \mGLmv{n} }   \mGLsym{,}  \Delta_{{\mathrm{3}}} )   \vdash_{\mathsf{GS} }  \mGLnt{Y}}\\{\delta'  \leq  \delta}
      }{(  \delta_{{\mathrm{1}}}  \mGLsym{,}  \delta  \mGLsym{,}  \delta_{{\mathrm{3}}}  )   \odot   ( \Delta_{{\mathrm{1}}}  \mGLsym{,}   \mGLnt{X} ^{ \mGLmv{n} }   \mGLsym{,}  \Delta_{{\mathrm{3}}} )   \vdash_{\mathsf{GS} }  \mGLnt{Y}}
      \]

      We know:
      \[
      \begin{array}{lll}
        \mathsf{Depth}  (  \Pi_{{\mathrm{1}}}  )   +   \mathsf{Depth}  (  \Pi_{{\mathrm{3}}}  )   \, \mGLsym{<} \,  \mathsf{Depth}  (  \Pi_{{\mathrm{1}}}  )   +   \mathsf{Depth}  (  \Pi_{{\mathrm{2}}}  )\\
        \mathsf{CutRank} \, \mGLsym{(}  \Pi_{{\mathrm{3}}}  \mGLsym{)} \, \leq \, \mathsf{CutRank} \, \mGLsym{(}  \Pi_{{\mathrm{2}}}  \mGLsym{)} \, \leq \,  \mathsf{Rank}  (  \mGLnt{X}  )
      \end{array}
      \]

      and so applying the induction hypothesis
      to $\Pi_{{\mathrm{1}}}$ and $\Pi_{{\mathrm{3}}}$
      implies that there is a proof $\Pi'$ of
      $(   \delta_{{\mathrm{1}}}  \mGLsym{,}  \delta'  \boxast [ { \delta_{{\mathrm{2}}} }^{ \mGLmv{n} } ]   \mGLsym{,}  \delta_{{\mathrm{3}}}  )   \odot   ( \Delta_{{\mathrm{1}}}  \mGLsym{,}  \Delta_{{\mathrm{2}}}  \mGLsym{,}  \Delta_{{\mathrm{3}}} )   \vdash_{\mathsf{GS} }  \mGLnt{Y}$ with
      $\mathsf{CutRank} \, \mGLsym{(}  \Pi'  \mGLsym{)} \, \leq \,  \mathsf{Rank}  (  \mGLnt{X}  )$.
      We know:
      \[
      \begin{array}{lll}
        \delta'  \leq  \delta\\
        \delta'  \boxast [ { \delta_{{\mathrm{2}}} }^{ \mGLmv{n} } ]   \leq   \delta  \boxast [ { \delta_{{\mathrm{2}}} }^{ \mGLmv{n} } ]
      \end{array}
      \]
      So we construct the proof $\Pi$:

      \[
        \inferrule* [flushleft,right=$\mGLdruleGSTXXSubName{}$,left=$\Pi :$] {
          \inferrule* [flushleft,right=,left=$\Pi' :$] {
            \pi'
          }{(   \delta_{{\mathrm{1}}}  \mGLsym{,}  \delta'  \boxast [ { \delta_{{\mathrm{2}}} }^{ \mGLmv{n} } ]   \mGLsym{,}  \delta_{{\mathrm{3}}}  )   \odot   ( \Delta_{{\mathrm{1}}}  \mGLsym{,}  \Delta_{{\mathrm{2}}}  \mGLsym{,}  \Delta_{{\mathrm{3}}} )   \vdash_{\mathsf{GS} }  \mGLnt{Y}}\\{\delta'  \boxast [ { \delta_{{\mathrm{2}}} }^{ \mGLmv{n} } ]   \mGLsym{=}   \delta  \boxast [ { \delta_{{\mathrm{2}}} }^{ \mGLmv{n} } ]}
        }{(   \delta_{{\mathrm{1}}}  \mGLsym{,}  \delta  \boxast [ { \delta_{{\mathrm{2}}} }^{ \mGLmv{n} } ]   \mGLsym{,}  \delta_{{\mathrm{3}}}  )   \odot   ( \Delta_{{\mathrm{1}}}  \mGLsym{,}   \mGLnt{X} ^{ \mGLmv{n} }   \mGLsym{,}  \Delta_{{\mathrm{3}}} )   \vdash_{\mathsf{GS} }  \mGLnt{Y}}
        \]
      \end{enumerate}
\end{enumerate}

\end{proof}

\begin{lemma}[Cut Reduction GS/MS]
    If $\Pi_{{\mathrm{1}}}$ is a proof of $\delta_{{\mathrm{2}}}  \odot  \Delta_{{\mathrm{2}}}  \vdash_{\mathsf{GS} }  \mGLnt{X}$ 
    and $\Pi_{{\mathrm{2}}}$ is a proof of $(  \delta_{{\mathrm{1}}}  \mGLsym{,}  \delta  \mGLsym{,}  \delta_{{\mathrm{3}}}  )   \odot   ( \Delta_{{\mathrm{1}}}  \mGLsym{,}   \mGLnt{X} ^{ \mGLmv{n} }   \mGLsym{,}  \Delta_{{\mathrm{3}}} )   \mGLsym{;}  \Gamma  \vdash_{\mathsf{MS} }  \mGLnt{B}$ 
    with $\mathsf{CutRank} \, \mGLsym{(}  \Pi_{{\mathrm{1}}}  \mGLsym{)} \, \leq \,  \mathsf{Rank}  (  \mGLnt{X}  )$
    and $\mathsf{CutRank} \, \mGLsym{(}  \Pi_{{\mathrm{2}}}  \mGLsym{)} \, \leq \,  \mathsf{Rank}  (  \mGLnt{X}  )$,
    then there exists a proof $\Pi$ of the sequent $(  \delta_{{\mathrm{1}}}  \mGLsym{,}   (   \delta  \boxast [ { \delta_{{\mathrm{2}}} }^{ \mGLmv{n} } ]   )   \mGLsym{,}  \delta_{{\mathrm{3}}}  )   \odot   ( \Delta_{{\mathrm{1}}}  \mGLsym{,}  \Delta_{{\mathrm{2}}}  \mGLsym{,}  \Delta_{{\mathrm{3}}} )   \mGLsym{;}  \Gamma  \vdash_{\mathsf{MS} }  \mGLnt{B}$ 
    with $\mathsf{CutRank} \, \mGLsym{(}  \Pi  \mGLsym{)} \, \leq \,  \mathsf{Rank}  (  \mGLnt{X}  )$.

  \end{lemma}


\begin{proof}
      This is by induction on $\mathsf{Depth}  (  \Pi_{{\mathrm{1}}}  )   +   \mathsf{Depth}  (  \Pi_{{\mathrm{2}}}  )$.
\begin{enumerate}
  \item \textbf{Commuting Conversions}
  \begin{enumerate}
    \item \textbf{left-side:} Suppose we have

    \[
      \inferrule* [flushleft,right=,left=$\Pi_{{\mathrm{1}}} :$] {
        \pi_1
      }{\delta_{{\mathrm{3}}}  \odot  \Delta_{{\mathrm{3}}}  \vdash_{\mathsf{GS} }  \mGLnt{X}}
     \]
     \[
      \inferrule* [flushleft,right=$\mGLdruleMSTXXCutName{}$,left=$\Pi_{{\mathrm{2}}} :$] {
        \inferrule* [flushleft,right=,left=$\Pi_{{\mathrm{3}}} :$] {
          \pi_3
        }{(  \delta_{{\mathrm{2}}}  \mGLsym{,}  \gamma_{{\mathrm{3}}}  \mGLsym{,}  \delta_{{\mathrm{4}}}  )   \odot   ( \Delta_{{\mathrm{2}}}  \mGLsym{,}   \mGLnt{X} ^{ \mGLmv{n} }   \mGLsym{,}  \Delta_{{\mathrm{4}}} )   \vdash_{\mathsf{GS} }  \mGLnt{Y}}\\
        \inferrule* [flushleft,right=,left=$\Pi_{{\mathrm{4}}} :$] {
          \pi_4
        }{(  \delta_{{\mathrm{1}}}  \mGLsym{,}  \delta  \mGLsym{,}  \delta_{{\mathrm{5}}}  )   \odot   ( \Delta_{{\mathrm{1}}}  \mGLsym{,}   \mGLnt{Y} ^{ \mGLmv{m} }   \mGLsym{,}  \Delta_{{\mathrm{5}}} )   \mGLsym{;}  \Gamma  \vdash_{\mathsf{MS} }  \mGLnt{A}}
      }{(  \delta_{{\mathrm{1}}}  \mGLsym{,}   (   \delta  \boxast [ {  (  \delta_{{\mathrm{2}}}  )  }^{ \mGLmv{m} } ]   )   \mGLsym{,}   (   \delta  \boxast [ {  (  \gamma_{{\mathrm{3}}}  )  }^{ \mGLmv{m} } ]   )   \mGLsym{,}   (   \delta  \boxast [ {  (  \delta_{{\mathrm{4}}}  )  }^{ \mGLmv{m} } ]   )   \mGLsym{,}  \delta_{{\mathrm{5}}}  )   \odot   ( \Delta_{{\mathrm{1}}}  \mGLsym{,}  \Delta_{{\mathrm{2}}}  \mGLsym{,}   \mGLnt{X} ^{ \mGLmv{n} }   \mGLsym{,}  \Delta_{{\mathrm{4}}}  \mGLsym{,}  \Delta_{{\mathrm{5}}} )   \mGLsym{;}  \Gamma  \vdash_{\mathsf{MS} }  \mGLnt{A}}
    \]

    We know:
    \[
    \begin{array}{lll}
      \mathsf{Depth}  (  \Pi_{{\mathrm{1}}}  )   +   \mathsf{Depth}  (  \Pi_{{\mathrm{3}}}  )   \, \mGLsym{<} \,  \mathsf{Depth}  (  \Pi_{{\mathrm{1}}}  )   +   \mathsf{Depth}  (  \Pi_{{\mathrm{2}}}  )\\
      \mathsf{CutRank} \, \mGLsym{(}  \Pi_{{\mathrm{3}}}  \mGLsym{)} \, \leq \, \mGLkw{Max} \, \mGLsym{(}   \mathsf{CutRank} \, \mGLsym{(}  \Pi_{{\mathrm{3}}}  \mGLsym{)}  \mGLsym{,}  \mathsf{CutRank} \, \mGLsym{(}  \Pi_{{\mathrm{4}}}  \mGLsym{)}  \mGLsym{,}   \mathsf{Rank}  (  \mGLnt{Y}  )   + 1   \mGLsym{)} \, \leq \,  \mathsf{Rank}  (  \mGLnt{X}  )
    \end{array}
    \]

    and so applying the mutual induction hypothesis
    from Lemma~\ref{lemma:cut_reduction_for_mgl} (1)
    to $\Pi_{{\mathrm{1}}}$ and $\Pi_{{\mathrm{3}}}$
    implies that there is a proof $\Pi'$ of
    $(  \delta_{{\mathrm{2}}}  \mGLsym{,}  \gamma  \mGLsym{,}  \delta_{{\mathrm{4}}}  )   \odot   ( \Delta_{{\mathrm{2}}}  \mGLsym{,}  \Delta_{{\mathrm{3}}}  \mGLsym{,}  \Delta_{{\mathrm{4}}} )   \vdash_{\mathsf{GS} }  \mGLnt{Y}$ with
    $\mathsf{CutRank} \, \mGLsym{(}  \Pi'  \mGLsym{)} \, \leq \,  \mathsf{Rank}  (  \mGLnt{X}  )$
    and $(   \gamma_{{\mathrm{3}}}  \boxast [ {  (  \delta_{{\mathrm{3}}}  )  }^{ \mGLmv{n} } ]   )   \mGLsym{=}  \gamma$.
    Thus, we construct the following proof $\Pi$:
    \[
    \inferrule* [flushleft,right=$\mGLdruleMSTXXCutName{}$,left=$\Pi :$] {
      \inferrule* [flushleft,left=$\Pi' : $] {
        \pi
      }{(  \delta_{{\mathrm{2}}}  \mGLsym{,}  \gamma  \mGLsym{,}  \delta_{{\mathrm{4}}}  )   \odot   ( \Delta_{{\mathrm{2}}}  \mGLsym{,}  \Delta_{{\mathrm{3}}}  \mGLsym{,}  \Delta_{{\mathrm{4}}} )   \vdash_{\mathsf{GS} }  \mGLnt{Y}}\\
      \inferrule* [flushleft,right=,left=$\Pi_{{\mathrm{4}}} :$] {
        \pi_4
      }{(  \delta_{{\mathrm{1}}}  \mGLsym{,}  \delta  \mGLsym{,}  \delta_{{\mathrm{5}}}  )   \odot   ( \Delta_{{\mathrm{1}}}  \mGLsym{,}   \mGLnt{Y} ^{ \mGLmv{n} }   \mGLsym{,}  \Delta_{{\mathrm{5}}} )   \mGLsym{;}  \Gamma  \vdash_{\mathsf{MS} }  \mGLnt{A}}
    }{(  \delta_{{\mathrm{1}}}  \mGLsym{,}   (   \delta  \boxast [ {  (  \delta_{{\mathrm{2}}}  )  }^{ \mGLmv{n} } ]   )   \mGLsym{,}   (   \delta  \boxast [ { \gamma }^{ \mGLmv{n} } ]   )   \mGLsym{,}   (   \delta  \boxast [ {  (  \delta_{{\mathrm{4}}}  )  }^{ \mGLmv{n} } ]   )   \mGLsym{,}  \delta_{{\mathrm{5}}}  )   \odot   ( \Delta_{{\mathrm{1}}}  \mGLsym{,}  \Delta_{{\mathrm{2}}}  \mGLsym{,}  \Delta_{{\mathrm{3}}}  \mGLsym{,}  \Delta_{{\mathrm{4}}}  \mGLsym{,}  \Delta_{{\mathrm{5}}} )   \mGLsym{;}  \Gamma  \vdash_{\mathsf{MS} }  \mGLnt{A}}
    \]
The rest of the proof is the same as the corresponding
case for Lemma~\ref{lemma:cut_reduction_for_mgl} (1)

\item \textbf{cut vs. right-side cut (left case):} Suppose we have:

  \[
  \inferrule* [flushleft,right=,left=$\Pi_{{\mathrm{1}}} :$] {
    \pi_1
  }{\delta_{{\mathrm{2}}}  \odot  \Delta_{{\mathrm{2}}}  \vdash_{\mathsf{GS} }  \mGLnt{X}}
  \]
  \[
  \inferrule* [flushleft,right=$\mGLdruleMSTXXCutName{}$,left=$\Pi_{{\mathrm{2}}} :$] {
    \inferrule* [flushleft,right=,left=$\Pi_{{\mathrm{3}}} :$] {
      \pi_3
    }{\delta_{{\mathrm{4}}}  \odot  \Delta_{{\mathrm{4}}}  \vdash_{\mathsf{GS} }  \mGLnt{Y}}\\
    \inferrule* [flushleft,right=,left=$\Pi_{{\mathrm{4}}} :$] {
      \pi_4
    }{(  \delta_{{\mathrm{1}}}  \mGLsym{,}  \delta  \mGLsym{,}  \delta_{{\mathrm{4}}}  \mGLsym{,}  \delta'  \mGLsym{,}  \delta_{{\mathrm{5}}}  )   \odot   ( \Delta_{{\mathrm{1}}}  \mGLsym{,}   \mGLnt{X} ^{ \mGLmv{n} }   \mGLsym{,}  \Delta_{{\mathrm{3}}}  \mGLsym{,}   \mGLnt{Y} ^{ \mGLmv{m} }   \mGLsym{,}  \Delta_{{\mathrm{5}}} )   \mGLsym{;}  \Gamma  \vdash_{\mathsf{MS} }  \mGLnt{A}}
  }{(  \delta_{{\mathrm{1}}}  \mGLsym{,}  \delta  \mGLsym{,}  \delta_{{\mathrm{3}}}  \mGLsym{,}   (   \delta'  \boxast [ { \delta_{{\mathrm{4}}} }^{ \mGLmv{m} } ]   )   \mGLsym{,}  \delta_{{\mathrm{5}}}  )   \odot   ( \Delta_{{\mathrm{1}}}  \mGLsym{,}   \mGLnt{X} ^{ \mGLmv{n} }   \mGLsym{,}  \Delta_{{\mathrm{3}}}  \mGLsym{,}  \Delta_{{\mathrm{4}}}  \mGLsym{,}  \Delta_{{\mathrm{5}}} )   \mGLsym{;}  \Gamma  \vdash_{\mathsf{MS} }  \mGLnt{A}}
  \]
We know:
\[
\begin{array}{lll}
\mathsf{Depth}  (  \Pi_{{\mathrm{1}}}  )   +   \mathsf{Depth}  (  \Pi_{{\mathrm{4}}}  )   \, \mGLsym{<} \,  \mathsf{Depth}  (  \Pi_{{\mathrm{1}}}  )   +   \mathsf{Depth}  (  \Pi_{{\mathrm{2}}}  )\\
\mathsf{CutRank} \, \mGLsym{(}  \Pi_{{\mathrm{4}}}  \mGLsym{)} \, \leq \, \mGLkw{Max} \, \mGLsym{(}   \mathsf{CutRank} \, \mGLsym{(}  \Pi_{{\mathrm{3}}}  \mGLsym{)}  \mGLsym{,}  \mathsf{CutRank} \, \mGLsym{(}  \Pi_{{\mathrm{4}}}  \mGLsym{)}  \mGLsym{,}   \mathsf{Rank}  (  \mGLnt{Y}  )   + 1   \mGLsym{)} \, \leq \,  \mathsf{Rank}  (  \mGLnt{X}  )
\end{array}
\]
and so applying the induction hypothesis
to $\Pi_{{\mathrm{1}}}$ and $\Pi_{{\mathrm{4}}}$
implies that there is a proof $\Pi'$ of
$(  \delta_{{\mathrm{1}}}  \mGLsym{,}  \gamma  \mGLsym{,}  \delta_{{\mathrm{4}}}  \mGLsym{,}  \delta'  \mGLsym{,}  \delta_{{\mathrm{5}}}  )   \odot   ( \Delta_{{\mathrm{1}}}  \mGLsym{,}  \Delta_{{\mathrm{2}}}  \mGLsym{,}  \Delta_{{\mathrm{4}}}  \mGLsym{,}   \mGLnt{Y} ^{ \mGLmv{m} }   \mGLsym{,}  \Delta_{{\mathrm{5}}} )   \mGLsym{;}  \Gamma  \vdash_{\mathsf{MS} }  \mGLnt{A}$ with
$\mathsf{CutRank} \, \mGLsym{(}  \Pi'  \mGLsym{)} \, \leq \,  \mathsf{Rank}  (  \mGLnt{X}  )$
and $(   \delta  \boxast [ {  (  \delta_{{\mathrm{2}}}  )  }^{ \mGLmv{n} } ]   )   \mGLsym{=}  \gamma$.
Thus, we construct the following proof $\Pi$:

\[
\inferrule* [flushleft,right=$\mGLdruleMSTXXCutName{}$,left=$\Pi :$] {
\inferrule* [flushleft,left=$\Pi_{{\mathrm{3}}} : $] {
  \pi_3
}{\delta_{{\mathrm{4}}}  \odot  \Delta_{{\mathrm{4}}}  \vdash_{\mathsf{GS} }  \mGLnt{Y}}\\
\inferrule* [flushleft,right=,left=$\Pi' :$] {
  \pi'
}{(  \delta_{{\mathrm{1}}}  \mGLsym{,}  \gamma  \mGLsym{,}  \delta_{{\mathrm{4}}}  \mGLsym{,}  \delta'  \mGLsym{,}  \delta_{{\mathrm{5}}}  )   \odot   ( \Delta_{{\mathrm{1}}}  \mGLsym{,}  \Delta_{{\mathrm{2}}}  \mGLsym{,}  \Delta_{{\mathrm{4}}}  \mGLsym{,}   \mGLnt{Y} ^{ \mGLmv{n} }   \mGLsym{,}  \Delta_{{\mathrm{5}}} )   \mGLsym{;}  \Gamma  \vdash_{\mathsf{MS} }  \mGLnt{A}}
}{(  \delta_{{\mathrm{1}}}  \mGLsym{,}  \gamma  \mGLsym{,}  \delta_{{\mathrm{3}}}  \mGLsym{,}   (   \delta'  \boxast [ { \delta_{{\mathrm{4}}} }^{ \mGLmv{m} } ]   )   \mGLsym{,}  \delta_{{\mathrm{5}}}  )   \odot   ( \Delta_{{\mathrm{1}}}  \mGLsym{,}  \Delta_{{\mathrm{2}}}  \mGLsym{,}  \Delta_{{\mathrm{3}}}  \mGLsym{,}  \Delta_{{\mathrm{4}}}  \mGLsym{,}  \Delta_{{\mathrm{5}}} )   \mGLsym{;}  \Gamma  \vdash_{\mathsf{MS} }  \mGLnt{A}}
\]
Given the above we know:
\[
\begin{array}{lll}
(   \delta  \boxast [ {  (  \delta_{{\mathrm{2}}}  )  }^{ \mGLmv{n} } ]   )   \mGLsym{=}  \gamma\\
\mathsf{CutRank} \, \mGLsym{(}  \Pi'  \mGLsym{)} \, \leq \,  \mathsf{Rank}  (  \mGLnt{X}  )\\
\mathsf{CutRank} \, \mGLsym{(}  \Pi_{{\mathrm{2}}}  \mGLsym{)} \, \mGLsym{=} \, \mGLkw{Max} \, \mGLsym{(}   \mathsf{CutRank} \, \mGLsym{(}  \Pi_{{\mathrm{3}}}  \mGLsym{)}  \mGLsym{,}  \mathsf{CutRank} \, \mGLsym{(}  \Pi_{{\mathrm{4}}}  \mGLsym{)}  \mGLsym{,}   \mathsf{Rank}  (  \mGLnt{Y}  )   + 1   \mGLsym{)} \, \leq \,  \mathsf{Rank}  (  \mGLnt{X}  )\\
\end{array}
\]
This implies:
\[
\begin{array}{lll}
\mathsf{CutRank} \, \mGLsym{(}  \Pi_{{\mathrm{3}}}  \mGLsym{)} \, \leq \,  \mathsf{Rank}  (  \mGLnt{X}  )\\
\mathsf{Rank}  (  \mGLnt{Y}  )   + 1  \, \leq \,  \mathsf{Rank}  (  \mGLnt{X}  )
\end{array}
\]
Thus, we obtain our result:
\[
\mathsf{CutRank} \, \mGLsym{(}  \Pi  \mGLsym{)} \, \mGLsym{=} \, \mGLkw{Max} \, \mGLsym{(}   \mathsf{CutRank} \, \mGLsym{(}  \Pi_{{\mathrm{3}}}  \mGLsym{)}  \mGLsym{,}  \mathsf{CutRank} \, \mGLsym{(}  \Pi'  \mGLsym{)}  \mGLsym{,}   \mathsf{Rank}  (  \mGLnt{Y}  )   + 1   \mGLsym{)} \, \leq \,  \mathsf{Rank}  (  \mGLnt{X}  )
\]

\item \textbf{cut vs. right-side cut (right case):} Suppose we have:
\[
      \inferrule* [flushleft,right=,left=$\Pi_{{\mathrm{1}}} :$] {
        \pi_1
      }{\delta_{{\mathrm{2}}}  \odot  \Delta_{{\mathrm{2}}}  \vdash_{\mathsf{GS} }  \mGLnt{X}}
\]

\[
      \inferrule* [flushleft,right=$\mGLdruleMSTXXCutName{}$,left=$\Pi_{{\mathrm{2}}} :$] {
        \inferrule* [flushleft,right=,left=$\Pi_{{\mathrm{3}}} :$] {
          \pi_3
        }{\delta_{{\mathrm{3}}}  \odot  \Delta_{{\mathrm{3}}}  \vdash_{\mathsf{GS} }  \mGLnt{Y}}\\
        \inferrule* [flushleft,right=,left=$\Pi_{{\mathrm{4}}} :$] {
          \pi_4
        }{(  \delta_{{\mathrm{1}}}  \mGLsym{,}  \mGLnt{r_{{\mathrm{1}}}}  \mGLsym{,}  \delta_{{\mathrm{4}}}  \mGLsym{,}  \mGLnt{r_{{\mathrm{2}}}}  \mGLsym{,}  \delta_{{\mathrm{5}}}  )   \odot   ( \Delta_{{\mathrm{1}}}  \mGLsym{,}   \mGLnt{Y} ^{ \mGLmv{n} }   \mGLsym{,}  \Delta_{{\mathrm{4}}}  \mGLsym{,}   \mGLnt{X} ^{ \mGLmv{n} }   \mGLsym{,}  \Delta_{{\mathrm{5}}} )   \mGLsym{;}  \Gamma  \vdash_{\mathsf{MS} }  \mGLnt{A}}
      }{(  \delta_{{\mathrm{4}}}  \mGLsym{,}  \mGLnt{r_{{\mathrm{1}}}}  \mGLsym{,}  \delta_{{\mathrm{1}}}  \mGLsym{,}  \delta_{{\mathrm{3}}}  \mGLsym{,}  \delta_{{\mathrm{5}}}  )   \odot   ( \Delta_{{\mathrm{4}}}  \mGLsym{,}   \mGLnt{X} ^{ \mGLmv{n} }   \mGLsym{,}  \Delta_{{\mathrm{1}}}  \mGLsym{,}  \Delta_{{\mathrm{3}}}  \mGLsym{,}  \Delta_{{\mathrm{5}}} )   \mGLsym{;}  \Gamma  \vdash_{\mathsf{MS} }  \mGLnt{A}}
\]
Similar to the previous case.
  \end{enumerate}

 \item \textbf{Axiom Cases}
  \begin{enumerate}
   \item \textbf{Axiom on the left:}
   \[
    \begin{array}{lll}
      \inferrule* [flushleft,right=$\mGLdruleGSTXXidName{}$,left=$\Pi_{{\mathrm{1}}} :$] {
        \,
      }{1  \odot  \mGLnt{X}  \vdash_{\mathsf{GS} }  \mGLnt{X}}
      & \quad &
      \inferrule* [flushleft,right=,left=$\Pi_{{\mathrm{2}}} :$] {
        \pi_2
      }{(  \delta_{{\mathrm{1}}}  \mGLsym{,}  \mGLnt{r}  \mGLsym{,}  \delta_{{\mathrm{3}}}  )   \odot   ( \Delta_{{\mathrm{1}}}  \mGLsym{,}  \mGLnt{X}  \mGLsym{,}  \Delta_{{\mathrm{3}}} )   \mGLsym{;}  \Gamma  \vdash_{\mathsf{MS} }  \mGLnt{A}}
    \end{array}
    \]
We know:
 $\Pi_{{\mathrm{2}}}$ is a proof of $(  \delta_{{\mathrm{1}}}  \mGLsym{,}  \mGLnt{r}  \mGLsym{,}  \delta_{{\mathrm{3}}}  )   \odot   ( \Delta_{{\mathrm{1}}}  \mGLsym{,}  \mGLnt{X}  \mGLsym{,}  \Delta_{{\mathrm{3}}} )   \mGLsym{;}  \Gamma  \vdash_{\mathsf{MS} }  \mGLnt{A}$
 and that $\mGLnt{r}  *  1  \leq  \mGLnt{r}$
Thus, we construct the following proof $\Pi \, \mGLsym{=} \, \Pi_{{\mathrm{2}}}$.
  \end{enumerate}

 \item \textbf{Principle Formula vs Principle Formula}
  \begin{enumerate}
   \item \textbf{Lin:}
   \[
    \inferrule* [flushleft,right=$\mGLdruleGSTXXLinRName{}$, left=$\Pi_{{\mathrm{1}}} :$] {
      \inferrule* [flushleft,right=, left=$\Pi_{{\mathrm{3}}} :$] {
        \pi_3
      }{\delta_{{\mathrm{2}}}  \odot  \Delta_{{\mathrm{2}}}  \mGLsym{;}  \emptyset  \vdash_{\mathsf{MS} }  \mGLnt{A}}
    }{\delta_{{\mathrm{2}}}  \odot  \Delta_{{\mathrm{2}}}  \vdash_{\mathsf{GS} }  \mathsf{Lin} \, \mGLnt{A}}
    \]
    \[
      \inferrule* [flushleft,right=$\mGLdruleMSTXXLinLName{}$, left=$\Pi_{{\mathrm{2}}} :$] {
        \inferrule* [flushleft,right=, left=$\Pi_{{\mathrm{4}}} :$] {
          \pi_4
        }{\delta_{{\mathrm{1}}}  \odot  \Delta_{{\mathrm{1}}}  \mGLsym{;}   ( \mGLnt{A}  \mGLsym{,}  \Gamma_{{\mathrm{1}}} )   \vdash_{\mathsf{MS} }  \mGLnt{B}}
      }{(  \delta_{{\mathrm{1}}}  \mGLsym{,}  1  )   \odot   ( \Delta_{{\mathrm{1}}}  \mGLsym{,}  \mathsf{Lin} \, \mGLnt{A} )   \mGLsym{;}  \Gamma_{{\mathrm{1}}}  \vdash_{\mathsf{MS} }  \mGLnt{B}}
      \]
      We know the following:
      \[
        \begin{array}{lll}
          \mathsf{Depth}  (  \Pi_{{\mathrm{1}}}  )   +   \mathsf{Depth}  (  \Pi_{{\mathrm{2}}}  ) = \mGLsym{(}    \mathsf{Depth}  (  \Pi_{{\mathrm{3}}}  )   + 1   \mGLsym{)}  +  \mGLsym{(}    \mathsf{Depth}  (  \Pi_{{\mathrm{4}}}  )   + 1   \mGLsym{)}\\
          \mathsf{CutRank} \, \mGLsym{(}  \Pi_{{\mathrm{1}}}  \mGLsym{)} \, \mGLsym{=} \, \mathsf{CutRank} \, \mGLsym{(}  \Pi_{{\mathrm{3}}}  \mGLsym{)} \, \leq \,  \mathsf{Rank}  (  \mathsf{Lin} \, \mGLnt{A}  )\\
          \mathsf{CutRank} \, \mGLsym{(}  \Pi_{{\mathrm{2}}}  \mGLsym{)} \, \mGLsym{=} \, \mathsf{CutRank} \, \mGLsym{(}  \Pi_{{\mathrm{3}}}  \mGLsym{)} \, \leq \,  \mathsf{Rank}  (  \mathsf{Lin} \, \mGLnt{A}  )
        \end{array}
      \]
      These imply that:
      \[
        \begin{array}{lll}
          \mathsf{Depth}  (  \Pi_{{\mathrm{3}}}  )   +   \mathsf{Depth}  (  \Pi_{{\mathrm{4}}}  )   \, \mGLsym{<} \,  \mathsf{Depth}  (  \Pi_{{\mathrm{1}}}  )   +   \mathsf{Depth}  (  \Pi_{{\mathrm{2}}}  )\\
          \mathsf{CutRank} \, \mGLsym{(}  \Pi_{{\mathrm{3}}}  \mGLsym{)} \, \leq \,  \mathsf{Rank}  (  \mathsf{Lin} \, \mGLnt{A}  )
          \mathsf{CutRank} \, \mGLsym{(}  \Pi_{{\mathrm{4}}}  \mGLsym{)} \, \leq \,  \mathsf{Rank}  (  \mathsf{Lin} \, \mGLnt{A}  )
        \end{array}
      \]
      Thus, we apply the induction hypothesis of Lemma~\ref{lemma:cut_reduction_for_mgl}(3)
      to $\Pi_{{\mathrm{3}}}$ and $\Pi_{{\mathrm{4}}}$ to obtain a proof $\Pi'$ of the sequent
      $(  \delta_{{\mathrm{1}}}  \mGLsym{,}  \delta_{{\mathrm{2}}}  )   \odot   ( \Delta_{{\mathrm{1}}}  \mGLsym{,}  \Delta_{{\mathrm{2}}} )   \mGLsym{;}  \Gamma_{{\mathrm{1}}}  \vdash_{\mathsf{MS} }  \mGLnt{B}$ with $\mathsf{CutRank} \, \mGLsym{(}  \Pi'  \mGLsym{)} \, \leq \,  \mathsf{Rank}  (  \mGLsym{(}  \mGLnt{A}  \mGLsym{)}  )$.
      Since $(   1  \boxast [ { \delta_{{\mathrm{2}}} }^{ \mGLmv{n} } ]   )   \mGLsym{=}   (  1  *  \delta_{{\mathrm{2}}}  )   \mGLsym{=}  \delta_{{\mathrm{2}}}$, we construct $\Pi \, \mGLsym{=} \, \Pi'$.

      \item \textbf{Graded Tensor}
      \[
            \inferrule* [flushleft,right=$\mGLdruleGSTXXTenRName{}$,left=$\Pi_{{\mathrm{1}}} :$] {
              \inferrule* [flushleft,right=,left=$\Pi_{{\mathrm{3}}} :$] {
                \pi_3
              }{\delta_{{\mathrm{1}}}  \odot  \Delta_{{\mathrm{1}}}  \vdash_{\mathsf{GS} }  \mGLnt{X}}\\
              \inferrule* [flushleft,right=,left=$\Pi_{{\mathrm{4}}} :$] {
                \pi_4
              }{\delta_{{\mathrm{2}}}  \odot  \Delta_{{\mathrm{2}}}  \vdash_{\mathsf{GS} }  \mGLnt{Y}}
            }{(  \delta_{{\mathrm{1}}}  \mGLsym{,}  \delta_{{\mathrm{2}}}  )   \odot   ( \Delta_{{\mathrm{1}}}  \mGLsym{,}  \Delta_{{\mathrm{2}}} )   \vdash_{\mathsf{GS} }  \mGLnt{X}  \boxtimes  \mGLnt{Y}}
            \]
            \[
            \inferrule* [flushleft,right=$\mGLdruleMSTXXGTenLName{}$,left=$\Pi_{{\mathrm{2}}} :$] {
              \inferrule* [flushleft,right=,left=$\Pi_{{\mathrm{5}}} :$] {
                \pi_5
              }{(  \delta_{{\mathrm{1}}}  \mGLsym{,}  \mGLnt{r}  \mGLsym{,}  \mGLnt{r}  \mGLsym{,}  \delta_{{\mathrm{4}}}  )   \odot   ( \Delta_{{\mathrm{1}}}  \mGLsym{,}  \mGLnt{X}  \mGLsym{,}  \mGLnt{Y}  \mGLsym{,}  \Delta_{{\mathrm{4}}} )   \mGLsym{;}  \Gamma  \vdash_{\mathsf{MS} }  \mGLnt{A}}
          }{(  \delta_{{\mathrm{1}}}  \mGLsym{,}  \mGLnt{r}  \mGLsym{,}  \delta_{{\mathrm{4}}}  )   \odot   ( \Delta_{{\mathrm{1}}}  \mGLsym{,}  \mGLnt{X}  \boxtimes  \mGLnt{Y}  \mGLsym{,}  \Delta_{{\mathrm{4}}} )   \mGLsym{;}  \Gamma  \vdash_{\mathsf{MS} }  \mGLnt{A}}
          \]
    We know:
    \[
    \begin{array}{lll}
      \mathsf{CutRank} \, \mGLsym{(}  \Pi_{{\mathrm{1}}}  \mGLsym{)} \, \mGLsym{=} \, \mGLkw{Max} \, \mGLsym{(}  \mathsf{CutRank} \, \mGLsym{(}  \Pi_{{\mathrm{3}}}  \mGLsym{)}  \mGLsym{,}  \mathsf{CutRank} \, \mGLsym{(}  \Pi_{{\mathrm{4}}}  \mGLsym{)}  \mGLsym{)} \, \leq \,  \mathsf{Rank}  (  \mGLnt{X}  \boxtimes  \mGLnt{Y}  )\\
      \mathsf{CutRank} \, \mGLsym{(}  \Pi_{{\mathrm{3}}}  \mGLsym{)} \, \leq \,  \mathsf{Rank}  (  \mGLnt{X}  \boxtimes  \mGLnt{Y}  )\\
      \mathsf{CutRank} \, \mGLsym{(}  \Pi_{{\mathrm{4}}}  \mGLsym{)} \, \leq \,  \mathsf{Rank}  (  \mGLnt{X}  \boxtimes  \mGLnt{Y}  )\\
      \mathsf{CutRank} \, \mGLsym{(}  \Pi_{{\mathrm{2}}}  \mGLsym{)} \, \mGLsym{=} \, \mathsf{CutRank} \, \mGLsym{(}  \Pi_{{\mathrm{5}}}  \mGLsym{)} \, \leq \,  \mathsf{Rank}  (  \mGLnt{X}  \boxtimes  \mGLnt{Y}  )\\
      \mathsf{Rank}  (  \mGLnt{X}  )  \, \mGLsym{<} \,  \mathsf{Rank}  (  \mGLnt{X}  \boxtimes  \mGLnt{Y}  )\\
      \mathsf{Rank}  (  \mGLnt{Y}  )  \, \mGLsym{<} \,  \mathsf{Rank}  (  \mGLnt{X}  \boxtimes  \mGLnt{Y}  )\\
      \mGLnt{r}  *   (  \delta_{{\mathrm{2}}}  \mGLsym{,}  \delta_{{\mathrm{3}}}  )   \mGLsym{=}   (  \mGLnt{r}  *  \delta_{{\mathrm{2}}}  \mGLsym{,}  \mGLnt{r}  *  \delta_{{\mathrm{3}}}  )
    \end{array}
    \]
    Instead of applying the induction hypothesis,
    we can directly build the proof $\Pi$:
    \begin{gather*}
      \inferrule* [flushleft,right=$\mGLdruleMSTXXGCutName{}$,left=$\Pi :$] {
        \inferrule* [flushleft,left=$\Pi_{{\mathrm{4}}} : $] {
            \pi_4
          }{\delta_{{\mathrm{3}}}  \odot  \Delta_{{\mathrm{3}}}  \vdash_{\mathsf{GS} }  \mGLnt{Y}}\\
          \inferrule* [flushleft,right=$\mGLdruleMSTXXGCutName{}$] {
            \inferrule* [flushleft,left=$\Pi_{{\mathrm{3}}} : $] {
            \pi_3
          }{\delta_{{\mathrm{1}}}  \odot  \Delta_{{\mathrm{1}}}  \vdash_{\mathsf{GS} }  \mGLnt{X}}\\
          \inferrule* [flushleft,right=,left=$\Pi_{{\mathrm{5}}} :$] {
            \pi_5
          }{(  \delta_{{\mathrm{1}}}  \mGLsym{,}  \mGLnt{r}  \mGLsym{,}  \mGLnt{r}  \mGLsym{,}  \delta_{{\mathrm{4}}}  )   \odot   ( \Delta_{{\mathrm{1}}}  \mGLsym{,}  \mGLnt{X}  \mGLsym{,}  \mGLnt{Y}  \mGLsym{,}  \Delta_{{\mathrm{4}}} )   \mGLsym{;}  \Gamma  \vdash_{\mathsf{MS} }  \mGLnt{A}}
          }{(  \delta_{{\mathrm{1}}}  \mGLsym{,}  \mGLnt{r}  *  \delta_{{\mathrm{2}}}  \mGLsym{,}  \mGLnt{r}  \mGLsym{,}  \delta_{{\mathrm{4}}}  )   \odot   ( \Delta_{{\mathrm{1}}}  \mGLsym{,}  \Delta_{{\mathrm{2}}}  \mGLsym{,}  \mGLnt{Y}  \mGLsym{,}  \Delta_{{\mathrm{4}}} )   \mGLsym{;}  \Gamma  \vdash_{\mathsf{MS} }  \mGLnt{A}}
      }{(  \delta_{{\mathrm{1}}}  \mGLsym{,}  \mGLnt{r}  *  \delta_{{\mathrm{2}}}  \mGLsym{,}  \mGLnt{r}  *  \delta_{{\mathrm{3}}}  \mGLsym{,}  \delta_{{\mathrm{4}}}  )   \odot   ( \Delta_{{\mathrm{1}}}  \mGLsym{,}  \Delta_{{\mathrm{2}}}  \mGLsym{,}  \Delta_{{\mathrm{3}}}  \mGLsym{,}  \Delta_{{\mathrm{4}}} )   \mGLsym{;}  \Gamma  \vdash_{\mathsf{MS} }  \mGLnt{A}}
      \end{gather*}
      So $\mathsf{CutRank} \, \mGLsym{(}  \Pi  \mGLsym{)} \, \mGLsym{=} \, \mGLkw{Max} \, \mGLsym{(}    \mathsf{CutRank} \, \mGLsym{(}  \Pi_{{\mathrm{3}}}  \mGLsym{)}  \mGLsym{,}  \mathsf{CutRank} \, \mGLsym{(}  \Pi_{{\mathrm{4}}}  \mGLsym{)}  \mGLsym{,}  \mathsf{CutRank} \, \mGLsym{(}  \Pi_{{\mathrm{5}}}  \mGLsym{)}  \mGLsym{,}   \mathsf{Rank}  (  \mGLnt{X}  )   + 1   \mGLsym{,}   \mathsf{Rank}  (  \mGLnt{Y}  )   + 1   \mGLsym{)} \, \leq \,  \mathsf{Rank}  (  \mGLnt{X}  \boxtimes  \mGLnt{Y}  )$

  \end{enumerate}

 \item \textbf{Secondary Conclusion}
  \begin{enumerate}
   \item \textbf{Left Introduction of Tensor:}
   \[
          \inferrule* [flushleft,right=$\mGLdruleGSTXXTenLName{}$,left=$\Pi_{{\mathrm{1}}} :$] {
            \inferrule* [flushleft,right=,left=$\Pi_{{\mathrm{3}}} :$] {
            \pi_3
          }{(  \delta_{{\mathrm{2}}}  \mGLsym{,}  \mGLnt{r}  \mGLsym{,}  \mGLnt{r}  \mGLsym{,}  \delta_{{\mathrm{3}}}  )   \odot   ( \Delta_{{\mathrm{2}}}  \mGLsym{,}  \mGLnt{X}  \mGLsym{,}  \mGLnt{Y}  \mGLsym{,}  \Delta_{{\mathrm{3}}} )   \vdash_{\mathsf{GS} }  \mGLnt{Z}}
          }{(  \delta_{{\mathrm{2}}}  \mGLsym{,}  \mGLnt{r}  \mGLsym{,}  \delta_{{\mathrm{3}}}  )   \odot   ( \Delta_{{\mathrm{2}}}  \mGLsym{,}  \mGLnt{X}  \boxtimes  \mGLnt{Y}  \mGLsym{,}  \Delta_{{\mathrm{3}}} )   \vdash_{\mathsf{GS} }  \mGLnt{Z}}
    \]
      \[
        \inferrule* [flushleft,right=,left=$\Pi_{{\mathrm{2}}} :$] {
          \pi_2
        }{(  \delta_{{\mathrm{1}}}  \mGLsym{,}  \delta  \mGLsym{,}  \delta_{{\mathrm{4}}}  )   \odot   ( \Delta_{{\mathrm{1}}}  \mGLsym{,}   \mGLnt{Z} ^{ \mGLmv{n} }   \mGLsym{,}  \Delta_{{\mathrm{4}}} )   \mGLsym{;}  \Gamma  \vdash_{\mathsf{MS} }  \mGLnt{A} }
        \]
      We know:
      \[
      \begin{array}{lll}
        \mathsf{Depth}  (  \Pi_{{\mathrm{3}}}  )   +   \mathsf{Depth}  (  \Pi_{{\mathrm{2}}}  )   \, \mGLsym{<} \,  \mathsf{Depth}  (  \Pi_{{\mathrm{1}}}  )   +   \mathsf{Depth}  (  \Pi_{{\mathrm{2}}}  )\\
        \mathsf{CutRank} \, \mGLsym{(}  \Pi_{{\mathrm{3}}}  \mGLsym{)} \, \leq \, \mathsf{CutRank} \, \mGLsym{(}  \Pi_{{\mathrm{1}}}  \mGLsym{)} \, \leq \,  \mathsf{Rank}  (  \mGLnt{Z}  )
      \end{array}
      \]

      and so applying the induction hypothesis
      to $\Pi_{{\mathrm{3}}}$ and $\Pi_{{\mathrm{2}}}$
      implies that there is a proof $\Pi'$ of
      $(  \delta_{{\mathrm{1}}}  \mGLsym{,}   (  \delta'_{{\mathrm{2}}}  \mGLsym{,}  \mGLnt{r'}  \mGLsym{,}  \mGLnt{r'}  \mGLsym{,}  \delta'_{{\mathrm{3}}}  )   \mGLsym{,}  \delta_{{\mathrm{4}}}  )   \odot   ( \Delta_{{\mathrm{1}}}  \mGLsym{,}   ( \Delta_{{\mathrm{2}}}  \mGLsym{,}  \mGLnt{X}  \mGLsym{,}  \mGLnt{Y}  \mGLsym{,}  \Delta_{{\mathrm{3}}} )   \mGLsym{,}  \Delta_{{\mathrm{4}}} )   \mGLsym{;}  \Gamma  \vdash_{\mathsf{MS} }  \mGLnt{A}$ with
      $\mathsf{CutRank} \, \mGLsym{(}  \Pi'  \mGLsym{)} \, \leq \,  \mathsf{Rank}  (  \mGLnt{Z}  )$ and $\delta  \boxast [ {  (  \delta_{{\mathrm{2}}}  \mGLsym{,}  \mGLnt{r}  \mGLsym{,}  \mGLnt{r}  \mGLsym{,}  \delta_{{\mathrm{3}}}  )  }^{ \mGLmv{n} } ]   \mGLsym{=}   (  \delta'_{{\mathrm{2}}}  \mGLsym{,}  \mGLnt{r'}  \mGLsym{,}  \mGLnt{r'}  \mGLsym{,}  \delta'_{{\mathrm{3}}}  )$
      where $ | \delta_{{\mathrm{2}}} |=| \delta'_{{\mathrm{2}}} |$ and $| \delta_{{\mathrm{3}}} |=| \delta'_{{\mathrm{3}}} |$.
      Thus, we construct the following proof $\Pi$:
      \[
        \inferrule* [flushleft,right=$\mGLdruleMSTXXGTenLName{}$,left=$\Pi :$] {
          \inferrule* [flushleft,right=,left=$\Pi' :$] {
            \pi'
          }{(  \delta_{{\mathrm{1}}}  \mGLsym{,}   (  \delta'_{{\mathrm{2}}}  \mGLsym{,}  \mGLnt{r'}  \mGLsym{,}  \mGLnt{r'}  \mGLsym{,}  \delta'_{{\mathrm{3}}}  )   \mGLsym{,}  \delta_{{\mathrm{4}}}  )   \odot   ( \Delta_{{\mathrm{1}}}  \mGLsym{,}   ( \Delta_{{\mathrm{2}}}  \mGLsym{,}  \mGLnt{X}  \mGLsym{,}  \mGLnt{Y}  \mGLsym{,}  \Delta_{{\mathrm{3}}} )   \mGLsym{,}  \Delta_{{\mathrm{4}}} )   \mGLsym{;}  \Gamma  \vdash_{\mathsf{MS} }  \mGLnt{A}}
        }{(  \delta_{{\mathrm{1}}}  \mGLsym{,}   (  \delta'_{{\mathrm{2}}}  \mGLsym{,}  \mGLnt{r'}  \mGLsym{,}  \delta'_{{\mathrm{3}}}  )   \mGLsym{,}  \delta_{{\mathrm{4}}}  )   \odot   ( \Delta_{{\mathrm{1}}}  \mGLsym{,}   ( \Delta_{{\mathrm{2}}}  \mGLsym{,}  \mGLnt{X}  \boxtimes  \mGLnt{Y}  \mGLsym{,}  \Delta_{{\mathrm{3}}} )   \mGLsym{,}  \Delta_{{\mathrm{4}}} )   \mGLsym{;}  \Gamma  \vdash_{\mathsf{MS} }  \mGLnt{A}}
        \]
        Given the above, we know:
        \[
          \begin{array}{lll}
            \mathsf{CutRank} \, \mGLsym{(}  \Pi  \mGLsym{)} \, \mGLsym{=} \, \mathsf{CutRank} \, \mGLsym{(}  \Pi'  \mGLsym{)} \, \leq \,  \mathsf{Rank}  (  \mGLnt{Z}  )\\
            \delta  \boxast [ {  (  \delta_{{\mathrm{2}}}  \mGLsym{,}  \mGLnt{r}  \mGLsym{,}  \mGLnt{r}  \mGLsym{,}  \delta_{{\mathrm{3}}}  )  }^{ \mGLmv{n} } ]   \mGLsym{=}   (  \delta'_{{\mathrm{2}}}  \mGLsym{,}  \mGLnt{r'}  \mGLsym{,}  \mGLnt{r'}  \mGLsym{,}  \delta'_{{\mathrm{3}}}  )
          \end{array}
          \]
          Implies
          \[
          \begin{array}{lll}
            (      \delta  \boxast [ { \delta_{{\mathrm{2}}} }^{ \mGLmv{n} } ]   \mGLsym{,}  \delta  \boxast [ { \mGLnt{r} }^{ \mGLmv{n} } ]   \mGLsym{,}  \delta  \boxast [ { \mGLnt{r} }^{ \mGLmv{n} } ]   \mGLsym{,}  \delta  \boxast [ { \delta_{{\mathrm{3}}} }^{ \mGLmv{n} } ]   )   \mGLsym{=}   (  \delta'_{{\mathrm{2}}}  \mGLsym{,}  \mGLnt{r'}  \mGLsym{,}  \mGLnt{r'}  \mGLsym{,}  \delta'_{{\mathrm{3}}}  )\\
            (     \delta  \boxast [ { \delta_{{\mathrm{2}}} }^{ \mGLmv{n} } ]   \mGLsym{,}  \delta  \boxast [ { \mGLnt{r} }^{ \mGLmv{n} } ]   \mGLsym{,}  \delta  \boxast [ { \delta_{{\mathrm{3}}} }^{ \mGLmv{n} } ]   )   \mGLsym{=}   (  \delta'_{{\mathrm{2}}}  \mGLsym{,}  \mGLnt{r'}  \mGLsym{,}  \delta'_{{\mathrm{3}}}  )\\
            \delta  \boxast [ {  (  \delta_{{\mathrm{2}}}  \mGLsym{,}  \mGLnt{r}  \mGLsym{,}  \delta_{{\mathrm{3}}}  )  }^{ \mGLmv{n} } ]   \mGLsym{=}   (  \delta'_{{\mathrm{2}}}  \mGLsym{,}  \mGLnt{r'}  \mGLsym{,}  \delta'_{{\mathrm{3}}}  )
          \end{array}
          \]

   \item \textbf{Left Introduction of Tensor Unit:}
   \[
          \inferrule* [flushleft,right=$\mGLdruleGSTXXUnitLName{}$,left=$\Pi_{{\mathrm{1}}} :$] {
              \inferrule* [flushleft,right=,left=$\Pi_{{\mathrm{3}}} :$] {
                \pi_3
              }{(  \delta_{{\mathrm{2}}}  \mGLsym{,}  \delta_{{\mathrm{3}}}  )   \odot   ( \Delta_{{\mathrm{2}}}  \mGLsym{,}  \Delta_{{\mathrm{3}}} )   \vdash_{\mathsf{GS} }  \mGLnt{Y}}
          }{(  \delta_{{\mathrm{2}}}  \mGLsym{,}  \mGLnt{r}  \mGLsym{,}  \delta_{{\mathrm{3}}}  )   \odot   ( \Delta_{{\mathrm{2}}}  \mGLsym{,}  \mathsf{J}  \mGLsym{,}  \Delta_{{\mathrm{3}}} )   \vdash_{\mathsf{GS} }  \mGLnt{X}}
    \]
    \[
      \inferrule* [flushleft,right=,left=$\Pi_{{\mathrm{2}}} :$] {
        \pi_2
      }{(  \delta_{{\mathrm{1}}}  \mGLsym{,}  \delta  \mGLsym{,}  \delta_{{\mathrm{4}}}  )   \odot   ( \Delta_{{\mathrm{1}}}  \mGLsym{,}   \mGLnt{X} ^{ \mGLmv{n} }   \mGLsym{,}  \Delta_{{\mathrm{4}}} )   \mGLsym{;}  \Gamma  \vdash_{\mathsf{MS} }  \mGLnt{A}}
\]
      We know:
      \[
      \begin{array}{lll}
        \mathsf{Depth}  (  \Pi_{{\mathrm{2}}}  )   +   \mathsf{Depth}  (  \Pi_{{\mathrm{3}}}  )   \, \mGLsym{<} \,  \mathsf{Depth}  (  \Pi_{{\mathrm{1}}}  )   +   \mathsf{Depth}  (  \Pi_{{\mathrm{2}}}  )\\
        \mathsf{CutRank} \, \mGLsym{(}  \Pi_{{\mathrm{3}}}  \mGLsym{)} \, \leq \, \mathsf{CutRank} \, \mGLsym{(}  \Pi_{{\mathrm{1}}}  \mGLsym{)} \, \leq \,  \mathsf{Rank}  (  \mGLnt{X}  )
      \end{array}
      \]

      and so applying the induction hypothesis
      to $\Pi_{{\mathrm{2}}}$ and $\Pi_{{\mathrm{3}}}$
      implies that there is a proof $\Pi'$ of
      $(  \delta_{{\mathrm{1}}}  \mGLsym{,}   (  \delta'_{{\mathrm{2}}}  \mGLsym{,}  \delta'_{{\mathrm{3}}}  )   \mGLsym{,}  \delta_{{\mathrm{4}}}  )   \odot   ( \Delta_{{\mathrm{1}}}  \mGLsym{,}  \Delta_{{\mathrm{2}}}  \mGLsym{,}  \Delta_{{\mathrm{3}}}  \mGLsym{,}  \Delta_{{\mathrm{4}}} )   \mGLsym{;}  \Gamma  \vdash_{\mathsf{MS} }  \mGLnt{A}$ with
      $\mathsf{CutRank} \, \mGLsym{(}  \Pi'  \mGLsym{)} \, \leq \,  \mathsf{Rank}  (  \mGLnt{X}  )$ and $\delta  \boxast [ {  (  \delta_{{\mathrm{2}}}  \mGLsym{,}  \delta_{{\mathrm{3}}}  )  }^{ \mGLmv{n} } ]   \mGLsym{=}   (  \delta'_{{\mathrm{2}}}  \mGLsym{,}  \delta'_{{\mathrm{3}}}  )$
      where $ | \delta_{{\mathrm{2}}} |=| \delta'_{{\mathrm{2}}} |$ and $| \delta_{{\mathrm{3}}} |=| \delta'_{{\mathrm{3}}} |$.
      Thus, we construct the following proof $\Pi$:

      \[
        \inferrule* [flushleft,right=$\mGLdruleMSTXXGUnitLName{}$,left=$\Pi :$] {
          \inferrule* [flushleft,right=,left=$\Pi' :$] {
            \pi'
          }{(  \delta_{{\mathrm{1}}}  \mGLsym{,}  \delta'_{{\mathrm{2}}}  \mGLsym{,}  \delta'_{{\mathrm{3}}}  \mGLsym{,}  \delta_{{\mathrm{4}}}  )   \odot   ( \Delta_{{\mathrm{1}}}  \mGLsym{,}  \Delta_{{\mathrm{2}}}  \mGLsym{,}  \Delta_{{\mathrm{3}}}  \mGLsym{,}  \Delta_{{\mathrm{4}}} )   \mGLsym{;}  \Gamma  \vdash_{\mathsf{MS} }  \mGLnt{A}}
        }{(  \delta_{{\mathrm{1}}}  \mGLsym{,}  \delta'_{{\mathrm{2}}}  \mGLsym{,}   (   \delta  \boxast [ { \mGLnt{r} }^{ \mGLmv{n} } ]   )   \mGLsym{,}  \delta'_{{\mathrm{3}}}  \mGLsym{,}  \delta_{{\mathrm{4}}}  )   \odot   ( \Delta_{{\mathrm{1}}}  \mGLsym{,}  \Delta_{{\mathrm{2}}}  \mGLsym{,}  \mathsf{J}  \mGLsym{,}  \Delta_{{\mathrm{3}}}  \mGLsym{,}  \Delta_{{\mathrm{4}}} )   \mGLsym{;}  \Gamma  \vdash_{\mathsf{MS} }  \mGLnt{A}}
        \]
        Given the above, we know:
        \[
          \begin{array}{lll}
            \mathsf{CutRank} \, \mGLsym{(}  \Pi  \mGLsym{)} \, \mGLsym{=} \, \mathsf{CutRank} \, \mGLsym{(}  \Pi'  \mGLsym{)} \, \leq \,  \mathsf{Rank}  (  \mGLnt{X}  )\\
            \delta  \boxast [ {  (  \delta_{{\mathrm{2}}}  \mGLsym{,}  \delta_{{\mathrm{3}}}  )  }^{ \mGLmv{n} } ]   \mGLsym{=}   (  \delta'_{{\mathrm{2}}}  \mGLsym{,}  \delta'_{{\mathrm{3}}}  )\\
            (    \delta  \boxast [ { \delta_{{\mathrm{2}}} }^{ \mGLmv{n} } ]   \mGLsym{,}  \delta  \boxast [ { \delta_{{\mathrm{3}}} }^{ \mGLmv{n} } ]   )   \mGLsym{=}  \delta'_{{\mathrm{2}}}  \mGLsym{,}  \delta'_{{\mathrm{3}}}\\
            \delta  \boxast [ { \delta_{{\mathrm{2}}} }^{ \mGLmv{n} } ]   \mGLsym{=}  \delta'_{{\mathrm{2}}}\\
            \delta  \boxast [ { \delta_{{\mathrm{3}}} }^{ \mGLmv{n} } ]   \mGLsym{=}  \delta'_{{\mathrm{3}}}\\
            \delta  \boxast [ { \mGLnt{r} }^{ \mGLmv{n} } ]   \mGLsym{=}   \delta  \boxast [ { \mGLnt{r} }^{ \mGLmv{n} } ]\\
            (     \delta  \boxast [ { \delta_{{\mathrm{2}}} }^{ \mGLmv{n} } ]   \mGLsym{,}  \delta  \boxast [ { \mGLnt{r} }^{ \mGLmv{n} } ]   \mGLsym{,}  \delta  \boxast [ { \delta_{{\mathrm{3}}} }^{ \mGLmv{n} } ]   )   \mGLsym{=}   (  \delta'_{{\mathrm{2}}}  \mGLsym{,}   \delta  \boxast [ { \mGLnt{r} }^{ \mGLmv{n} } ]   \mGLsym{,}  \delta'_{{\mathrm{3}}}  )\\
            \delta  \boxast [ {  (  \delta_{{\mathrm{2}}}  \mGLsym{,}  \mGLnt{r}  \mGLsym{,}  \delta_{{\mathrm{3}}}  )  }^{ \mGLmv{n} } ]   \mGLsym{=}   (  \delta'_{{\mathrm{2}}}  \mGLsym{,}   \delta  \boxast [ { \mGLnt{r} }^{ \mGLmv{n} } ]   \mGLsym{,}  \delta'_{{\mathrm{3}}}  )
          \end{array}
          \]
          \item \textbf{Weakening:}
          \[
        \inferrule* [flushleft,right=$\mGLdruleGSTXXWeakName{}$,left=$\Pi_{{\mathrm{1}}} :$] {
          \inferrule* [flushleft,left=$\Pi_{{\mathrm{3}}} : $] {
            \pi_3
          }{(  \delta_{{\mathrm{2}}}  \mGLsym{,}  \delta_{{\mathrm{3}}}  )   \odot   ( \Delta_{{\mathrm{1}}}  \mGLsym{,}  \Delta_{{\mathrm{2}}} )   \vdash_{\mathsf{GS} }  \mGLnt{X}}
        }{(  \delta_{{\mathrm{2}}}  \mGLsym{,}  \mathsf{0}  \mGLsym{,}  \delta_{{\mathrm{3}}}  )   \odot   ( \Delta_{{\mathrm{1}}}  \mGLsym{,}  \mGLnt{Y}  \mGLsym{,}  \Delta_{{\mathrm{2}}} )   \vdash_{\mathsf{GS} }  \mGLnt{X}}
        \]
        \[
          \inferrule* [flushleft,right=,left=$\Pi_{{\mathrm{2}}} :$] {
            \pi_2
          }{(  \delta_{{\mathrm{1}}}  \mGLsym{,}  \delta  \mGLsym{,}  \delta_{{\mathrm{4}}}  )   \odot   ( \Delta_{{\mathrm{1}}}  \mGLsym{,}   \mGLnt{X} ^{ \mGLmv{n} }   \mGLsym{,}  \Delta_{{\mathrm{4}}} )   \mGLsym{;}  \Gamma  \vdash_{\mathsf{MS} }  \mGLnt{A}}
    \]
    We know:
      \[
      \begin{array}{lll}
        \mathsf{Depth}  (  \Pi_{{\mathrm{3}}}  )   +   \mathsf{Depth}  (  \Pi_{{\mathrm{2}}}  )   \, \mGLsym{<} \,  \mathsf{Depth}  (  \Pi_{{\mathrm{1}}}  )   +   \mathsf{Depth}  (  \Pi_{{\mathrm{2}}}  )\\
        \mathsf{CutRank} \, \mGLsym{(}  \Pi_{{\mathrm{3}}}  \mGLsym{)} \, \leq \, \mathsf{CutRank} \, \mGLsym{(}  \Pi_{{\mathrm{1}}}  \mGLsym{)} \, \leq \,  \mathsf{Rank}  (  \mGLnt{X}  )
      \end{array}
      \]

      and so applying the induction hypothesis
      to $\Pi_{{\mathrm{3}}}$ and $\Pi_{{\mathrm{2}}}$
      implies that there is a proof $\Pi'$ of
      $\delta_{{\mathrm{1}}}  \mGLsym{,}   (  \delta'_{{\mathrm{2}}}  \mGLsym{,}  \delta'_{{\mathrm{3}}}  )   \mGLsym{,}  \delta_{{\mathrm{4}}}  \odot  \Delta_{{\mathrm{1}}}  \mGLsym{,}  \Delta_{{\mathrm{2}}}  \mGLsym{,}  \Delta_{{\mathrm{3}}}  \mGLsym{,}  \Delta_{{\mathrm{4}}}  \mGLsym{;}  \Gamma  \vdash_{\mathsf{MS} }  \mGLnt{A}$ with
      $\mathsf{CutRank} \, \mGLsym{(}  \Pi'  \mGLsym{)} \, \leq \,  \mathsf{Rank}  (  \mGLnt{X}  )$ and $\delta  \boxast [ {  (  \delta_{{\mathrm{2}}}  \mGLsym{,}  \delta_{{\mathrm{3}}}  )  }^{ \mGLmv{n} } ]   \mGLsym{=}   (  \delta'_{{\mathrm{2}}}  \mGLsym{,}  \delta'_{{\mathrm{3}}}  )$
      where $ | \delta_{{\mathrm{2}}} |=| \delta'_{{\mathrm{2}}} |$ and $| \delta_{{\mathrm{3}}} |=| \delta'_{{\mathrm{3}}} |$.
      Thus, we construct the following proof $\Pi$:

          \[
        \inferrule* [flushleft,right=$\mGLdruleMSTXXWeakName{}$,left=$\Pi :$] {
          \inferrule* [flushleft,left=$\Pi' : $] {
            \pi'
          }{\delta_{{\mathrm{1}}}  \mGLsym{,}  \delta'_{{\mathrm{2}}}  \mGLsym{,}  \delta'_{{\mathrm{3}}}  \mGLsym{,}  \delta_{{\mathrm{4}}}  \odot  \Delta_{{\mathrm{1}}}  \mGLsym{,}  \Delta_{{\mathrm{2}}}  \mGLsym{,}  \Delta_{{\mathrm{3}}}  \mGLsym{,}  \Delta_{{\mathrm{4}}}  \mGLsym{;}  \Gamma  \vdash_{\mathsf{MS} }  \mGLnt{A}}
        }{\delta_{{\mathrm{1}}}  \mGLsym{,}  \delta'_{{\mathrm{2}}}  \mGLsym{,}  \mathsf{0}  \mGLsym{,}  \delta'_{{\mathrm{3}}}  \mGLsym{,}  \delta_{{\mathrm{4}}}  \odot  \Delta_{{\mathrm{1}}}  \mGLsym{,}  \Delta_{{\mathrm{2}}}  \mGLsym{,}  \mGLnt{Y}  \mGLsym{,}  \Delta_{{\mathrm{3}}}  \mGLsym{,}  \Delta_{{\mathrm{4}}}  \mGLsym{;}  \Gamma  \vdash_{\mathsf{MS} }  \mGLnt{A}}
        \]
        Given the above, we know:
        \[
          \begin{array}{lll}
            \mathsf{CutRank} \, \mGLsym{(}  \Pi  \mGLsym{)} \, \mGLsym{=} \, \mathsf{CutRank} \, \mGLsym{(}  \Pi'  \mGLsym{)} \, \leq \,  \mathsf{Rank}  (  \mGLnt{X}  )\\
            \delta  \boxast [ {  (  \delta_{{\mathrm{2}}}  \mGLsym{,}  \delta_{{\mathrm{3}}}  )  }^{ \mGLmv{n} } ]   \mGLsym{=}   (  \delta'_{{\mathrm{2}}}  \mGLsym{,}  \delta'_{{\mathrm{3}}}  )\\
            (    \delta  \boxast [ { \delta_{{\mathrm{2}}} }^{ \mGLmv{n} } ]   \mGLsym{,}  \delta  \boxast [ { \delta_{{\mathrm{3}}} }^{ \mGLmv{n} } ]   )   \mGLsym{=}  \delta'_{{\mathrm{2}}}  \mGLsym{,}  \delta'_{{\mathrm{3}}}\\
            \delta  \boxast [ { \delta_{{\mathrm{2}}} }^{ \mGLmv{n} } ]   \mGLsym{=}  \delta'_{{\mathrm{2}}}\\
            \delta  \boxast [ { \delta_{{\mathrm{3}}} }^{ \mGLmv{n} } ]   \mGLsym{=}  \delta'_{{\mathrm{3}}}\\
            \delta  \boxast [ { \mathsf{0} }^{ \mGLmv{n} } ]   \mGLsym{=}  \mathsf{0}\\
            (     \delta  \boxast [ { \delta_{{\mathrm{2}}} }^{ \mGLmv{n} } ]   \mGLsym{,}  \delta  \boxast [ { \mathsf{0} }^{ \mGLmv{n} } ]   \mGLsym{,}  \delta  \boxast [ { \delta_{{\mathrm{3}}} }^{ \mGLmv{n} } ]   )   \mGLsym{=}   (  \delta'_{{\mathrm{2}}}  \mGLsym{,}  \mathsf{0}  \mGLsym{,}  \delta'_{{\mathrm{3}}}  )\\
            \delta  \boxast [ {  (  \delta_{{\mathrm{2}}}  \mGLsym{,}  \mathsf{0}  \mGLsym{,}  \delta_{{\mathrm{3}}}  )  }^{ \mGLmv{n} } ]   \mGLsym{=}   (  \delta'_{{\mathrm{2}}}  \mGLsym{,}  \mathsf{0}  \mGLsym{,}  \delta'_{{\mathrm{3}}}  )
          \end{array}
          \]
  \end{enumerate}

 \item \textbf{Secondary Hypothesis}
  \begin{enumerate}
   \item \textbf{Right introduction of tensor product (first case):}
   \[
    \inferrule* [flushleft,right=,left=$\Pi_{{\mathrm{1}}} :$] {
      \pi_1
    }{\delta_{{\mathrm{2}}}  \odot  \Delta_{{\mathrm{2}}}  \vdash_{\mathsf{GS} }  \mGLnt{X}}
\]

\[
  \inferrule* [flushleft,right=$\mGLdruleMSTXXTenRName{}$,left=$\Pi_{{\mathrm{2}}} :$] {
    \inferrule* [flushleft,left=$\Pi_{{\mathrm{3}}} : $] {
      \pi_3
    }{\delta_{{\mathrm{1}}}  \mGLsym{,}  \delta  \mGLsym{,}  \delta_{{\mathrm{3}}}  \odot  \Delta_{{\mathrm{1}}}  \mGLsym{,}   \mGLnt{X} ^{ \mGLmv{n} }   \mGLsym{,}  \Delta_{{\mathrm{3}}}  \mGLsym{;}  \Gamma_{{\mathrm{1}}}  \vdash_{\mathsf{MS} }  \mGLnt{A} }\\
    \inferrule* [flushleft,right=,left=$\Pi_{{\mathrm{4}}} :$] {
      \pi_4
    }{ \delta_{{\mathrm{4}}}  \odot  \Delta_{{\mathrm{4}}}  \mGLsym{;}  \Gamma_{{\mathrm{2}}}  \vdash_{\mathsf{MS} }  \mGLnt{B}}
  }{\delta_{{\mathrm{1}}}  \mGLsym{,}  \delta  \mGLsym{,}  \delta_{{\mathrm{3}}}  \mGLsym{,}  \delta_{{\mathrm{4}}}  \odot  \Delta_{{\mathrm{1}}}  \mGLsym{,}   \mGLnt{X} ^{ \mGLmv{n} }   \mGLsym{,}  \Delta_{{\mathrm{3}}}  \mGLsym{,}  \Delta_{{\mathrm{4}}}  \mGLsym{;}  \Gamma_{{\mathrm{1}}}  \mGLsym{,}  \Gamma_{{\mathrm{2}}}  \vdash_{\mathsf{MS} }  \mGLnt{A}  \otimes  \mGLnt{B}  }
  \]

We know:
\[
\begin{array}{lll}
  \mathsf{Depth}  (  \Pi_{{\mathrm{1}}}  )   +   \mathsf{Depth}  (  \Pi_{{\mathrm{3}}}  )   \, \mGLsym{<} \,  \mathsf{Depth}  (  \Pi_{{\mathrm{1}}}  )   +   \mathsf{Depth}  (  \Pi_{{\mathrm{2}}}  )\\
  \mathsf{CutRank} \, \mGLsym{(}  \Pi_{{\mathrm{3}}}  \mGLsym{)} \, \leq \, \mathsf{CutRank} \, \mGLsym{(}  \Pi_{{\mathrm{2}}}  \mGLsym{)} \, \leq \,  \mathsf{Rank}  (  \mGLnt{X}  )
\end{array}
\]

and so applying the induction hypothesis
to $\Pi_{{\mathrm{1}}}$ and $\Pi_{{\mathrm{3}}}$
implies that there is a proof $\Pi'$ of
$\delta_{{\mathrm{1}}}  \mGLsym{,}  \delta'_{{\mathrm{2}}}  \mGLsym{,}  \delta_{{\mathrm{3}}}  \odot  \Delta_{{\mathrm{1}}}  \mGLsym{,}  \Delta_{{\mathrm{2}}}  \mGLsym{,}  \Delta_{{\mathrm{3}}}  \mGLsym{;}  \Gamma_{{\mathrm{1}}}  \vdash_{\mathsf{MS} }  \mGLnt{A}$ with
$\mathsf{CutRank} \, \mGLsym{(}  \Pi'  \mGLsym{)} \, \leq \,  \mathsf{Rank}  (  \mGLnt{X}  )$ and $(   \delta  \boxast [ {  (  \delta_{{\mathrm{2}}}  )  }^{ \mGLmv{n} } ]   )   \mGLsym{=}  \delta'_{{\mathrm{2}}}$.
Thus, we construct the following proof $\Pi$:

  \[
  \inferrule* [flushleft,right=$\mGLdruleMSTXXTenRName{}$,left=$\Pi :$] {
    \inferrule* [flushleft,left=$\Pi' : $] {
      \pi'
    }{\delta_{{\mathrm{1}}}  \mGLsym{,}  \delta'_{{\mathrm{2}}}  \mGLsym{,}  \delta_{{\mathrm{3}}}  \odot  \Delta_{{\mathrm{1}}}  \mGLsym{,}  \Delta_{{\mathrm{2}}}  \mGLsym{,}  \Delta_{{\mathrm{3}}}  \mGLsym{;}  \Gamma_{{\mathrm{1}}}  \vdash_{\mathsf{MS} }  \mGLnt{A}}\\
    \inferrule* [flushleft,right=,left=$\Pi_{{\mathrm{4}}} :$] {
      \pi_4
    }{\delta_{{\mathrm{4}}}  \odot  \Delta_{{\mathrm{4}}}  \mGLsym{;}  \Gamma_{{\mathrm{2}}}  \vdash_{\mathsf{MS} }  \mGLnt{B}}
  }{\delta_{{\mathrm{1}}}  \mGLsym{,}  \delta'_{{\mathrm{2}}}  \mGLsym{,}  \delta_{{\mathrm{3}}}  \mGLsym{,}  \delta_{{\mathrm{4}}}  \odot  \Delta_{{\mathrm{1}}}  \mGLsym{,}  \Delta_{{\mathrm{2}}}  \mGLsym{,}  \Delta_{{\mathrm{3}}}  \mGLsym{,}  \Delta_{{\mathrm{4}}}  \mGLsym{;}  \Gamma_{{\mathrm{1}}}  \mGLsym{,}  \Gamma_{{\mathrm{2}}}  \vdash_{\mathsf{MS} }  \mGLnt{A}  \otimes  \mGLnt{B}}
  \]
  Given the above, we know:
  \[
    \begin{array}{lll}
      \mathsf{CutRank} \, \mGLsym{(}  \Pi  \mGLsym{)} \, \mGLsym{=} \, \mGLkw{Max} \, \mGLsym{(}  \mathsf{CutRank} \, \mGLsym{(}  \Pi'  \mGLsym{)}  \mGLsym{,}  \mathsf{CutRank} \, \mGLsym{(}  \Pi_{{\mathrm{4}}}  \mGLsym{)}  \mGLsym{)}\\
      \mathsf{CutRank} \, \mGLsym{(}  \Pi_{{\mathrm{4}}}  \mGLsym{)} \, \leq \,  \mathsf{Rank}  (  \mGLnt{X}  )\\
      \mathsf{CutRank} \, \mGLsym{(}  \Pi'  \mGLsym{)} \, \leq \,  \mathsf{Rank}  (  \mGLnt{X}  )\\
      \mathsf{CutRank} \, \mGLsym{(}  \Pi  \mGLsym{)} \, \leq \,  \mathsf{Rank}  (  \mGLnt{X}  )\\
      \delta  \boxast [ { \delta_{{\mathrm{2}}} }^{ \mGLmv{n} } ]   \mGLsym{=}  \delta'_{{\mathrm{2}}}\\
    \end{array}
    \]
   \item \textbf{Right introduction of tensor product (second case):}
   \[
    \inferrule* [flushleft,right=,left=$\Pi_{{\mathrm{1}}} :$] {
      \pi_1
    }{\delta_{{\mathrm{2}}}  \odot  \Delta_{{\mathrm{2}}}  \vdash_{\mathsf{GS} }  \mGLnt{X}}
\]

\[
  \inferrule* [flushleft,right=$\mGLdruleMSTXXTenRName{}$,left=$\Pi_{{\mathrm{2}}} :$] {
    \inferrule* [flushleft,left=$\Pi_{{\mathrm{3}}} : $] {
      \pi_3
    }{\delta_{{\mathrm{1}}}  \odot  \Delta_{{\mathrm{1}}}  \mGLsym{;}  \Gamma_{{\mathrm{1}}}  \vdash_{\mathsf{MS} }  \mGLnt{A} }\\
    \inferrule* [flushleft,right=,left=$\Pi_{{\mathrm{4}}} :$] {
      \pi_4
    }{ \delta_{{\mathrm{2}}}  \mGLsym{,}  \delta  \mGLsym{,}  \delta_{{\mathrm{4}}}  \odot  \Delta_{{\mathrm{2}}}  \mGLsym{,}   \mGLnt{X} ^{ \mGLmv{n} }   \mGLsym{,}  \Delta_{{\mathrm{4}}}  \mGLsym{;}  \Gamma_{{\mathrm{2}}}  \vdash_{\mathsf{MS} }  \mGLnt{B}}
  }{\delta_{{\mathrm{1}}}  \mGLsym{,}  \delta_{{\mathrm{2}}}  \mGLsym{,}  \delta  \mGLsym{,}  \delta_{{\mathrm{4}}}  \odot  \Delta_{{\mathrm{1}}}  \mGLsym{,}  \Delta_{{\mathrm{2}}}  \mGLsym{,}   \mGLnt{X} ^{ \mGLmv{n} }   \mGLsym{,}  \Delta_{{\mathrm{4}}}  \mGLsym{;}  \Gamma_{{\mathrm{1}}}  \mGLsym{,}  \Gamma_{{\mathrm{2}}}  \vdash_{\mathsf{MS} }  \mGLnt{A}  \otimes  \mGLnt{B}  }
  \]
  Similar to previous case.
   \item \textbf{Left introduction of linear tensor product:}
   \[
    \inferrule* [flushleft,right=,left=$\Pi_{{\mathrm{1}}} :$] {
      \pi_1
    }{\delta_{{\mathrm{2}}}  \odot  \Delta_{{\mathrm{2}}}  \vdash_{\mathsf{GS} }  \mGLnt{X}}
\]

\[
  \inferrule* [flushleft,right=$\mGLdruleMSTXXTenLName{}$,left=$\Pi_{{\mathrm{2}}} :$] {
    \inferrule* [flushleft,right=,left=$\Pi_{{\mathrm{3}}} :$] {
      \pi_3
    }{(  \delta_{{\mathrm{1}}}  \mGLsym{,}  \delta  \mGLsym{,}  \delta_{{\mathrm{3}}}  )   \odot   ( \Delta_{{\mathrm{1}}}  \mGLsym{,}   \mGLnt{X} ^{ \mGLmv{n} }   \mGLsym{,}  \Delta_{{\mathrm{3}}} )   \mGLsym{;}  \Gamma_{{\mathrm{1}}}  \mGLsym{,}  \mGLnt{A}  \mGLsym{,}  \mGLnt{B}  \mGLsym{,}  \Gamma_{{\mathrm{2}}}  \vdash_{\mathsf{MS} }  \mGLnt{C}}
    }{(  \delta_{{\mathrm{1}}}  \mGLsym{,}  \delta  \mGLsym{,}  \delta_{{\mathrm{3}}}  )   \odot   ( \Delta_{{\mathrm{1}}}  \mGLsym{,}   \mGLnt{X} ^{ \mGLmv{n} }   \mGLsym{,}  \Delta_{{\mathrm{3}}} )   \mGLsym{;}  \Gamma_{{\mathrm{1}}}  \mGLsym{,}  \mGLnt{A}  \otimes  \mGLnt{B}  \mGLsym{,}  \Gamma_{{\mathrm{2}}}  \vdash_{\mathsf{MS} }  \mGLnt{C}}
  \]

We know:
\[
\begin{array}{lll}
  \mathsf{Depth}  (  \Pi_{{\mathrm{1}}}  )   +   \mathsf{Depth}  (  \Pi_{{\mathrm{3}}}  )   \, \mGLsym{<} \,  \mathsf{Depth}  (  \Pi_{{\mathrm{1}}}  )   +   \mathsf{Depth}  (  \Pi_{{\mathrm{2}}}  )\\
  \mathsf{CutRank} \, \mGLsym{(}  \Pi_{{\mathrm{3}}}  \mGLsym{)} \, \leq \, \mathsf{CutRank} \, \mGLsym{(}  \Pi_{{\mathrm{2}}}  \mGLsym{)} \, \leq \,  \mathsf{Rank}  (  \mGLnt{X}  )
\end{array}
\]

and so applying the induction hypothesis
to $\Pi_{{\mathrm{1}}}$ and $\Pi_{{\mathrm{3}}}$
implies that there is a proof $\Pi'$ of
$(  \delta_{{\mathrm{1}}}  \mGLsym{,}  \delta'_{{\mathrm{2}}}  \mGLsym{,}  \delta_{{\mathrm{3}}}  )   \odot   ( \Delta_{{\mathrm{1}}}  \mGLsym{,}  \Delta_{{\mathrm{2}}}  \mGLsym{,}  \Delta_{{\mathrm{3}}} )   \mGLsym{;}  \Gamma_{{\mathrm{1}}}  \mGLsym{,}  \mGLnt{A}  \mGLsym{,}  \mGLnt{B}  \mGLsym{,}  \Gamma_{{\mathrm{2}}}  \vdash_{\mathsf{MS} }  \mGLnt{C}$ with
$\mathsf{CutRank} \, \mGLsym{(}  \Pi'  \mGLsym{)} \, \leq \,  \mathsf{Rank}  (  \mGLnt{X}  )$ and $\delta  \boxast [ { \delta_{{\mathrm{2}}} }^{ \mGLmv{n} } ]   \mGLsym{=}  \delta'_{{\mathrm{2}}}$.
Thus, we construct the following proof $\Pi$:

\[
  \inferrule* [flushleft,right=$\mGLdruleMSTXXTenLName{}$,left=$\Pi :$] {
    \inferrule* [flushleft,right=,left=$\Pi' :$] {
      \pi'
    }{(  \delta_{{\mathrm{1}}}  \mGLsym{,}  \delta'_{{\mathrm{2}}}  \mGLsym{,}  \delta_{{\mathrm{3}}}  )   \odot   ( \Delta_{{\mathrm{1}}}  \mGLsym{,}  \Delta_{{\mathrm{2}}}  \mGLsym{,}  \Delta_{{\mathrm{3}}} )   \mGLsym{;}  \Gamma_{{\mathrm{1}}}  \mGLsym{,}  \mGLnt{A}  \mGLsym{,}  \mGLnt{B}  \mGLsym{,}  \Gamma_{{\mathrm{2}}}  \vdash_{\mathsf{MS} }  \mGLnt{C}}
  }{(  \delta_{{\mathrm{1}}}  \mGLsym{,}  \delta'_{{\mathrm{2}}}  \mGLsym{,}  \delta_{{\mathrm{3}}}  )   \odot   ( \Delta_{{\mathrm{1}}}  \mGLsym{,}  \Delta_{{\mathrm{2}}}  \mGLsym{,}  \Delta_{{\mathrm{3}}} )   \mGLsym{;}  \Gamma_{{\mathrm{1}}}  \mGLsym{,}  \mGLnt{A}  \otimes  \mGLnt{B}  \mGLsym{,}  \Gamma_{{\mathrm{2}}}  \vdash_{\mathsf{MS} }  \mGLnt{C}}
  \]
  Given the above, we know:
  \[
    \begin{array}{lll}
      \mathsf{CutRank} \, \mGLsym{(}  \Pi  \mGLsym{)} \, \mGLsym{=} \, \mathsf{CutRank} \, \mGLsym{(}  \Pi'  \mGLsym{)} \, \leq \,  \mathsf{Rank}  (  \mGLnt{X}  )\\
      \delta  \boxast [ { \delta_{{\mathrm{2}}} }^{ \mGLmv{n} } ]   \mGLsym{=}  \delta'_{{\mathrm{2}}}\\
    \end{array}
    \]

   \item \textbf{Left introduction of the unit of linear tensor:}
   \[
          \inferrule* [flushleft,right=,left=$\Pi_{{\mathrm{1}}} :$] {
            \pi_1
          }{\delta_{{\mathrm{2}}}  \odot  \Delta_{{\mathrm{2}}}  \vdash_{\mathsf{GS} }  \mGLnt{X}}
    \]

      \[
        \inferrule* [flushleft,right=$\mGLdruleMSTXXUnitLName{}$,left=$\Pi_{{\mathrm{2}}} :$] {
          \inferrule* [flushleft,left=$\Pi_{{\mathrm{3}}} : $] {
            \pi_3
          }{(  \delta_{{\mathrm{1}}}  \mGLsym{,}  \delta  \mGLsym{,}  \delta_{{\mathrm{3}}}  )   \odot   ( \Delta_{{\mathrm{1}}}  \mGLsym{,}   \mGLnt{X} ^{ \mGLmv{n} }   \mGLsym{,}  \Delta_{{\mathrm{3}}} )   \mGLsym{;}  \Gamma_{{\mathrm{1}}}  \mGLsym{,}  \Gamma_{{\mathrm{2}}}  \vdash_{\mathsf{MS} }  \mGLnt{A}}\\
        }{(  \delta_{{\mathrm{1}}}  \mGLsym{,}  \delta  \mGLsym{,}  \delta_{{\mathrm{3}}}  )   \odot   ( \Delta_{{\mathrm{1}}}  \mGLsym{,}   \mGLnt{X} ^{ \mGLmv{n} }   \mGLsym{,}  \Delta_{{\mathrm{3}}} )   \mGLsym{;}  \Gamma_{{\mathrm{1}}}  \mGLsym{,}  \mathsf{I}  \mGLsym{,}  \Gamma_{{\mathrm{2}}}  \vdash_{\mathsf{MS} }  \mGLnt{A}}
        \]

      We know:
      \[
      \begin{array}{lll}
        \mathsf{Depth}  (  \Pi_{{\mathrm{1}}}  )   +   \mathsf{Depth}  (  \Pi_{{\mathrm{3}}}  )   \, \mGLsym{<} \,  \mathsf{Depth}  (  \Pi_{{\mathrm{1}}}  )   +   \mathsf{Depth}  (  \Pi_{{\mathrm{2}}}  )\\
        \mathsf{CutRank} \, \mGLsym{(}  \Pi_{{\mathrm{3}}}  \mGLsym{)} \, \leq \, \mathsf{CutRank} \, \mGLsym{(}  \Pi_{{\mathrm{2}}}  \mGLsym{)} \, \leq \,  \mathsf{Rank}  (  \mGLnt{X}  )
      \end{array}
      \]

      and so applying the induction hypothesis
      to $\Pi_{{\mathrm{1}}}$ and $\Pi_{{\mathrm{3}}}$
      implies that there is a proof $\Pi'$ of
      $(  \delta_{{\mathrm{1}}}  \mGLsym{,}  \delta'_{{\mathrm{2}}}  \mGLsym{,}  \delta_{{\mathrm{3}}}  )   \odot   ( \Delta_{{\mathrm{1}}}  \mGLsym{,}  \Delta_{{\mathrm{2}}}  \mGLsym{,}  \Delta_{{\mathrm{3}}} )   \mGLsym{;}  \Gamma_{{\mathrm{1}}}  \mGLsym{,}  \Gamma_{{\mathrm{2}}}  \vdash_{\mathsf{MS} }  \mGLnt{A}$ with
      $\mathsf{CutRank} \, \mGLsym{(}  \Pi'  \mGLsym{)} \, \leq \,  \mathsf{Rank}  (  \mGLnt{X}  )$ and $\delta  \boxast [ { \delta_{{\mathrm{2}}} }^{ \mGLmv{n} } ]   \mGLsym{=}  \delta'_{{\mathrm{2}}}$.
      Thus, we construct the following proof $\Pi$:
      \[
        \inferrule* [flushleft,right=$\mGLdruleMSTXXUnitLName{}$,left=$\Pi :$] {
          \inferrule* [flushleft,left=$\Pi' : $] {
            \pi'
          }{(  \delta_{{\mathrm{1}}}  \mGLsym{,}  \delta'_{{\mathrm{2}}}  \mGLsym{,}  \delta_{{\mathrm{3}}}  )   \odot   ( \Delta_{{\mathrm{1}}}  \mGLsym{,}  \Delta_{{\mathrm{2}}}  \mGLsym{,}  \Delta_{{\mathrm{3}}} )   \mGLsym{;}  \Gamma_{{\mathrm{1}}}  \mGLsym{,}  \Gamma_{{\mathrm{2}}}  \vdash_{\mathsf{MS} }  \mGLnt{A}}\\
        }{(  \delta_{{\mathrm{1}}}  \mGLsym{,}  \delta'_{{\mathrm{2}}}  \mGLsym{,}  \delta_{{\mathrm{3}}}  )   \odot   ( \Delta_{{\mathrm{1}}}  \mGLsym{,}  \Delta_{{\mathrm{2}}}  \mGLsym{,}  \Delta_{{\mathrm{3}}} )   \mGLsym{;}  \Gamma_{{\mathrm{1}}}  \mGLsym{,}  \mathsf{I}  \mGLsym{,}  \Gamma_{{\mathrm{2}}}  \vdash_{\mathsf{MS} }  \mGLnt{A}}
        \]
        Given the above, we know:
        \[
          \begin{array}{lll}
            \mathsf{CutRank} \, \mGLsym{(}  \Pi  \mGLsym{)} \, \mGLsym{=} \, \mathsf{CutRank} \, \mGLsym{(}  \Pi'  \mGLsym{)} \, \leq \,  \mathsf{Rank}  (  \mGLnt{X}  )\\
            \delta  \boxast [ { \delta_{{\mathrm{2}}} }^{ \mGLmv{n} } ]   \mGLsym{=}  \delta'_{{\mathrm{2}}}\\
          \end{array}
          \]
   \item \textbf{Left introduction of graded tensor product:}
   \[
    \inferrule* [flushleft,right=,left=$\Pi_{{\mathrm{1}}} :$] {
      \pi_1
    }{\delta_{{\mathrm{3}}}  \odot  \Delta_{{\mathrm{3}}}  \vdash_{\mathsf{GS} }  \mGLnt{X}}
\]

\[
  \inferrule* [flushleft,right=$\mGLdruleMSTXXGTenLName{}$,left=$\Pi_{{\mathrm{2}}} :$] {
    \inferrule* [flushleft,right=,left=$\Pi_{{\mathrm{3}}} :$] {
      \pi_3
    }{(  \delta_{{\mathrm{1}}}  \mGLsym{,}  \mGLnt{r}  \mGLsym{,}  \mGLnt{r}  \mGLsym{,}  \delta_{{\mathrm{2}}}  \mGLsym{,}  \delta  \mGLsym{,}  \delta_{{\mathrm{4}}}  )   \odot   ( \Delta_{{\mathrm{1}}}  \mGLsym{,}  \mGLnt{Y}  \mGLsym{,}  \mGLnt{Z}  \mGLsym{,}  \Delta_{{\mathrm{2}}}  \mGLsym{,}   \mGLnt{X} ^{ \mGLmv{n} }   \mGLsym{,}  \Delta_{{\mathrm{4}}} )   \mGLsym{;}  \Gamma  \vdash_{\mathsf{MS} }  \mGLnt{A}}
    }{(  \delta_{{\mathrm{1}}}  \mGLsym{,}  \mGLnt{r}  \mGLsym{,}  \delta_{{\mathrm{2}}}  \mGLsym{,}  \delta  \mGLsym{,}  \delta_{{\mathrm{4}}}  )   \odot   ( \Delta_{{\mathrm{1}}}  \mGLsym{,}  \mGLnt{Y}  \boxtimes  \mGLnt{Z}  \mGLsym{,}  \Delta_{{\mathrm{2}}}  \mGLsym{,}   \mGLnt{X} ^{ \mGLmv{n} }   \mGLsym{,}  \Delta_{{\mathrm{4}}} )   \mGLsym{;}  \Gamma  \vdash_{\mathsf{MS} }  \mGLnt{A}}
  \]

We know:
\[
\begin{array}{lll}
  \mathsf{Depth}  (  \Pi_{{\mathrm{1}}}  )   +   \mathsf{Depth}  (  \Pi_{{\mathrm{3}}}  )   \, \mGLsym{<} \,  \mathsf{Depth}  (  \Pi_{{\mathrm{1}}}  )   +   \mathsf{Depth}  (  \Pi_{{\mathrm{2}}}  )\\
  \mathsf{CutRank} \, \mGLsym{(}  \Pi_{{\mathrm{3}}}  \mGLsym{)} \, \leq \, \mathsf{CutRank} \, \mGLsym{(}  \Pi_{{\mathrm{2}}}  \mGLsym{)} \, \leq \,  \mathsf{Rank}  (  \mGLnt{X}  )
\end{array}
\]

and so applying the induction hypothesis
to $\Pi_{{\mathrm{1}}}$ and $\Pi_{{\mathrm{3}}}$
implies that there is a proof $\Pi'$ of
$(  \delta_{{\mathrm{1}}}  \mGLsym{,}  \mGLnt{r}  \mGLsym{,}  \mGLnt{r}  \mGLsym{,}  \delta_{{\mathrm{2}}}  \mGLsym{,}  \delta'_{{\mathrm{3}}}  \mGLsym{,}  \delta_{{\mathrm{4}}}  )   \odot   ( \Delta_{{\mathrm{1}}}  \mGLsym{,}  \mGLnt{Y}  \mGLsym{,}  \mGLnt{Z}  \mGLsym{,}  \Delta_{{\mathrm{2}}}  \mGLsym{,}  \Delta_{{\mathrm{3}}}  \mGLsym{,}  \Delta_{{\mathrm{4}}} )   \mGLsym{;}  \Gamma  \vdash_{\mathsf{MS} }  \mGLnt{A}$ with
$\mathsf{CutRank} \, \mGLsym{(}  \Pi'  \mGLsym{)} \, \leq \,  \mathsf{Rank}  (  \mGLnt{X}  )$ and $\delta  \boxast [ { \delta_{{\mathrm{3}}} }^{ \mGLmv{n} } ]   \mGLsym{=}  \delta'_{{\mathrm{3}}}$.
Thus, we construct the following proof $\Pi$:

\[
  \inferrule* [flushleft,right=$\mGLdruleMSTXXGTenLName{}$,left=$\Pi :$] {
    \inferrule* [flushleft,right=,left=$\Pi' :$] {
      \pi'
    }{(  \delta_{{\mathrm{1}}}  \mGLsym{,}  \mGLnt{r}  \mGLsym{,}  \mGLnt{r}  \mGLsym{,}  \delta_{{\mathrm{2}}}  \mGLsym{,}  \delta'_{{\mathrm{3}}}  \mGLsym{,}  \delta_{{\mathrm{4}}}  )   \odot   ( \Delta_{{\mathrm{1}}}  \mGLsym{,}  \mGLnt{Y}  \mGLsym{,}  \mGLnt{Z}  \mGLsym{,}  \Delta_{{\mathrm{2}}}  \mGLsym{,}  \Delta_{{\mathrm{3}}}  \mGLsym{,}  \Delta_{{\mathrm{4}}} )   \mGLsym{;}  \Gamma  \vdash_{\mathsf{MS} }  \mGLnt{A}}
  }{(  \delta_{{\mathrm{1}}}  \mGLsym{,}  \mGLnt{r}  \mGLsym{,}  \delta_{{\mathrm{2}}}  \mGLsym{,}  \delta'_{{\mathrm{3}}}  \mGLsym{,}  \delta_{{\mathrm{4}}}  )   \odot   ( \Delta_{{\mathrm{1}}}  \mGLsym{,}  \mGLnt{Y}  \boxtimes  \mGLnt{Z}  \mGLsym{,}  \Delta_{{\mathrm{2}}}  \mGLsym{,}  \Delta_{{\mathrm{3}}}  \mGLsym{,}  \Delta_{{\mathrm{4}}} )   \mGLsym{;}  \Gamma  \vdash_{\mathsf{MS} }  \mGLnt{A}}
  \]
  Given the above, we know:
  \[
    \begin{array}{lll}
      \mathsf{CutRank} \, \mGLsym{(}  \Pi  \mGLsym{)} \, \mGLsym{=} \, \mathsf{CutRank} \, \mGLsym{(}  \Pi'  \mGLsym{)} \, \leq \,  \mathsf{Rank}  (  \mGLnt{X}  )\\
      \delta  \boxast [ { \delta_{{\mathrm{3}}} }^{ \mGLmv{n} } ]   \mGLsym{=}  \delta'_{{\mathrm{3}}}\\
    \end{array}
    \]
    \item \textbf{Left introduction of graded tensor product: second case}
    \[
      \inferrule* [flushleft,right=,left=$\Pi_{{\mathrm{1}}} :$] {
        \pi_1
      }{\delta_{{\mathrm{2}}}  \odot  \Delta_{{\mathrm{2}}}  \vdash_{\mathsf{GS} }  \mGLnt{X}}
\]

  \[
    \inferrule* [flushleft,right=$\mGLdruleMSTXXGTenLName{}$,left=$\Pi_{{\mathrm{2}}} :$] {
      \inferrule* [flushleft,right=,left=$\Pi_{{\mathrm{3}}} :$] {
        \pi_3
      }{(  \delta_{{\mathrm{1}}}  \mGLsym{,}  \delta  \mGLsym{,}  \delta_{{\mathrm{3}}}  \mGLsym{,}  \mGLnt{r}  \mGLsym{,}  \mGLnt{r}  \mGLsym{,}  \delta_{{\mathrm{4}}}  )   \odot   ( \Delta_{{\mathrm{1}}}  \mGLsym{,}   \mGLnt{X} ^{ \mGLmv{n} }   \mGLsym{,}  \Delta_{{\mathrm{3}}}  \mGLsym{,}  \mGLnt{Y}  \mGLsym{,}  \mGLnt{Z}  \mGLsym{,}  \Delta_{{\mathrm{4}}} )   \mGLsym{;}  \Gamma  \vdash_{\mathsf{MS} }  \mGLnt{A}}
      }{(  \delta_{{\mathrm{1}}}  \mGLsym{,}  \mGLnt{r}  \mGLsym{,}  \delta_{{\mathrm{2}}}  \mGLsym{,}  \delta  \mGLsym{,}  \delta_{{\mathrm{4}}}  )   \odot   ( \Delta_{{\mathrm{1}}}  \mGLsym{,}  \mGLnt{Y}  \boxtimes  \mGLnt{Z}  \mGLsym{,}  \Delta_{{\mathrm{2}}}  \mGLsym{,}   \mGLnt{X} ^{ \mGLmv{n} }   \mGLsym{,}  \Delta_{{\mathrm{4}}} )   \mGLsym{;}  \Gamma  \vdash_{\mathsf{MS} }  \mGLnt{A}}
    \]
    Similar to the previous case.
   \item \textbf{Left introduction of the unit of graded tensor:}
   \[
          \inferrule* [flushleft,right=,left=$\Pi_{{\mathrm{1}}} :$] {
            \pi_1
          }{\delta_{{\mathrm{2}}}  \odot  \Delta_{{\mathrm{2}}}  \vdash_{\mathsf{GS} }  \mGLnt{X}}
    \]

      \[
        \inferrule* [flushleft,right=$\mGLdruleMSTXXGUnitLName{}$,left=$\Pi_{{\mathrm{2}}} :$] {
          \inferrule* [flushleft,left=$\Pi_{{\mathrm{3}}} : $] {
            \pi_3
          }{(  \delta_{{\mathrm{1}}}  \mGLsym{,}  \delta  \mGLsym{,}  \delta_{{\mathrm{3}}}  \mGLsym{,}  \delta_{{\mathrm{4}}}  )   \odot   ( \Delta_{{\mathrm{1}}}  \mGLsym{,}   \mGLnt{X} ^{ \mGLmv{n} }   \mGLsym{,}  \Delta_{{\mathrm{3}}}  \mGLsym{,}  \Delta_{{\mathrm{4}}} )   \mGLsym{;}  \Gamma  \vdash_{\mathsf{MS} }  \mGLnt{A}}\\
        }{(  \delta_{{\mathrm{1}}}  \mGLsym{,}  \delta  \mGLsym{,}  \delta_{{\mathrm{3}}}  \mGLsym{,}  \mGLnt{r}  \mGLsym{,}  \delta_{{\mathrm{4}}}  )   \odot   ( \Delta_{{\mathrm{1}}}  \mGLsym{,}   \mGLnt{X} ^{ \mGLmv{n} }   \mGLsym{,}  \Delta_{{\mathrm{3}}}  \mGLsym{,}  \mathsf{J}  \mGLsym{,}  \Delta_{{\mathrm{4}}} )   \mGLsym{;}  \Gamma  \vdash_{\mathsf{MS} }  \mGLnt{A}}
        \]

      We know:
      \[
      \begin{array}{lll}
        \mathsf{Depth}  (  \Pi_{{\mathrm{1}}}  )   +   \mathsf{Depth}  (  \Pi_{{\mathrm{3}}}  )   \, \mGLsym{<} \,  \mathsf{Depth}  (  \Pi_{{\mathrm{1}}}  )   +   \mathsf{Depth}  (  \Pi_{{\mathrm{2}}}  )\\
        \mathsf{CutRank} \, \mGLsym{(}  \Pi_{{\mathrm{3}}}  \mGLsym{)} \, \leq \, \mathsf{CutRank} \, \mGLsym{(}  \Pi_{{\mathrm{2}}}  \mGLsym{)} \, \leq \,  \mathsf{Rank}  (  \mGLnt{X}  )
      \end{array}
      \]

      and so applying the induction hypothesis
      to $\Pi_{{\mathrm{1}}}$ and $\Pi_{{\mathrm{3}}}$
      implies that there is a proof $\Pi'$ of
      $(  \delta_{{\mathrm{1}}}  \mGLsym{,}  \delta'_{{\mathrm{2}}}  \mGLsym{,}  \delta_{{\mathrm{3}}}  \mGLsym{,}  \delta_{{\mathrm{4}}}  )   \odot   ( \Delta_{{\mathrm{1}}}  \mGLsym{,}  \Delta_{{\mathrm{2}}}  \mGLsym{,}  \Delta_{{\mathrm{3}}}  \mGLsym{,}  \Delta_{{\mathrm{4}}} )   \mGLsym{;}  \Gamma  \vdash_{\mathsf{MS} }  \mGLnt{A}$ with
      $\mathsf{CutRank} \, \mGLsym{(}  \Pi'  \mGLsym{)} \, \leq \,  \mathsf{Rank}  (  \mGLnt{X}  )$ and $\delta  \boxast [ { \delta_{{\mathrm{2}}} }^{ \mGLmv{n} } ]   \mGLsym{=}  \delta'_{{\mathrm{2}}}$.
      Thus, we construct the following proof $\Pi$:
      \[
        \inferrule* [flushleft,right=$\mGLdruleMSTXXGUnitLName{}$,left=$\Pi :$] {
          \inferrule* [flushleft,left=$\Pi' : $] {
            \pi'
          }{(  \delta_{{\mathrm{1}}}  \mGLsym{,}  \delta'_{{\mathrm{2}}}  \mGLsym{,}  \delta_{{\mathrm{3}}}  \mGLsym{,}  \delta_{{\mathrm{4}}}  )   \odot   ( \Delta_{{\mathrm{1}}}  \mGLsym{,}  \Delta_{{\mathrm{2}}}  \mGLsym{,}  \Delta_{{\mathrm{3}}}  \mGLsym{,}  \Delta_{{\mathrm{4}}} )   \mGLsym{;}  \Gamma  \vdash_{\mathsf{MS} }  \mGLnt{A}}\\
        }{(  \delta_{{\mathrm{1}}}  \mGLsym{,}  \delta'_{{\mathrm{2}}}  \mGLsym{,}  \delta_{{\mathrm{3}}}  \mGLsym{,}  \mGLnt{r}  \mGLsym{,}  \delta_{{\mathrm{4}}}  )   \odot   ( \Delta_{{\mathrm{1}}}  \mGLsym{,}  \Delta_{{\mathrm{2}}}  \mGLsym{,}  \Delta_{{\mathrm{3}}}  \mGLsym{,}  \mathsf{J}  \mGLsym{,}  \Delta_{{\mathrm{4}}} )   \mGLsym{;}  \Gamma  \vdash_{\mathsf{MS} }  \mGLnt{A}}
        \]
        Given the above, we know:
        \[
          \begin{array}{lll}
            \mathsf{CutRank} \, \mGLsym{(}  \Pi  \mGLsym{)} \, \mGLsym{=} \, \mathsf{CutRank} \, \mGLsym{(}  \Pi'  \mGLsym{)} \, \leq \,  \mathsf{Rank}  (  \mGLnt{X}  )\\
            \delta  \boxast [ { \delta_{{\mathrm{2}}} }^{ \mGLmv{n} } ]   \mGLsym{=}  \delta'_{{\mathrm{2}}}\\
          \end{array}
          \]
   \item \textbf{Right introduction of linear implication:}
   \[
    \inferrule* [flushleft,right=,left=$\Pi_{{\mathrm{1}}} :$] {
      \pi_1
    }{\delta_{{\mathrm{2}}}  \odot  \Delta_{{\mathrm{2}}}  \vdash_{\mathsf{GS} }  \mGLnt{X}}
\]
\[
        \inferrule* [flushleft,right=$\mGLdruleMSTXXGImpRName{}$,left=$\Pi_{{\mathrm{2}}} :$] {
          \inferrule* [flushleft,left=$\Pi_{{\mathrm{3}}} : $] {
            \pi_3
          }{(  \delta_{{\mathrm{1}}}  \mGLsym{,}  \delta  \mGLsym{,}  \delta_{{\mathrm{3}}}  )   \odot   ( \Delta_{{\mathrm{1}}}  \mGLsym{,}   \mGLnt{X} ^{ \mGLmv{n} }   \mGLsym{,}  \Delta_{{\mathrm{3}}} )   \mGLsym{;}  \Gamma  \mGLsym{,}  \mGLnt{A}  \vdash_{\mathsf{MS} }  \mGLnt{B}}\\
        }{(  \delta_{{\mathrm{1}}}  \mGLsym{,}  \delta  \mGLsym{,}  \delta_{{\mathrm{3}}}  )   \odot   ( \Delta_{{\mathrm{1}}}  \mGLsym{,}   \mGLnt{X} ^{ \mGLmv{n} }   \mGLsym{,}  \Delta_{{\mathrm{3}}} )   \mGLsym{;}  \Gamma  \vdash_{\mathsf{MS} }  \mGLnt{A}  \multimap  \mGLnt{B}}
        \]

      We know:
      \[
      \begin{array}{lll}
        \mathsf{Depth}  (  \Pi_{{\mathrm{1}}}  )   +   \mathsf{Depth}  (  \Pi_{{\mathrm{3}}}  )   \, \mGLsym{<} \,  \mathsf{Depth}  (  \Pi_{{\mathrm{1}}}  )   +   \mathsf{Depth}  (  \Pi_{{\mathrm{2}}}  )\\
        \mathsf{CutRank} \, \mGLsym{(}  \Pi_{{\mathrm{3}}}  \mGLsym{)} \, \leq \, \mathsf{CutRank} \, \mGLsym{(}  \Pi_{{\mathrm{2}}}  \mGLsym{)} \, \leq \,  \mathsf{Rank}  (  \mGLnt{X}  )
      \end{array}
      \]

      and so applying the induction hypothesis
      to $\Pi_{{\mathrm{1}}}$ and $\Pi_{{\mathrm{3}}}$
      implies that there is a proof $\Pi'$ of
      $(  \delta_{{\mathrm{1}}}  \mGLsym{,}  \delta'_{{\mathrm{2}}}  \mGLsym{,}  \delta_{{\mathrm{3}}}  )   \odot   ( \Delta_{{\mathrm{1}}}  \mGLsym{,}  \Delta_{{\mathrm{2}}}  \mGLsym{,}  \Delta_{{\mathrm{3}}} )   \mGLsym{;}  \Gamma  \mGLsym{,}  \mGLnt{A}  \vdash_{\mathsf{MS} }  \mGLnt{B}$ with
      $\mathsf{CutRank} \, \mGLsym{(}  \Pi'  \mGLsym{)} \, \leq \,  \mathsf{Rank}  (  \mGLnt{X}  )$ and $\delta  \boxast [ { \delta_{{\mathrm{2}}} }^{ \mGLmv{n} } ]   \mGLsym{=}  \delta'_{{\mathrm{2}}}$.
      Thus, we construct the following proof $\Pi$:

      \[
        \inferrule* [flushleft,right=$\mGLdruleMSTXXGImpRName{}$,left=$\Pi_{{\mathrm{1}}} :$] {
          \inferrule* [flushleft,left=$\Pi' : $] {
            \pi'
          }{(  \delta_{{\mathrm{1}}}  \mGLsym{,}  \delta'_{{\mathrm{2}}}  \mGLsym{,}  \delta_{{\mathrm{3}}}  )   \odot   ( \Delta_{{\mathrm{1}}}  \mGLsym{,}  \Delta_{{\mathrm{2}}}  \mGLsym{,}  \Delta_{{\mathrm{3}}} )   \mGLsym{;}  \Gamma  \mGLsym{,}  \mGLnt{A}  \vdash_{\mathsf{MS} }  \mGLnt{B}}\\
        }{(  \delta_{{\mathrm{1}}}  \mGLsym{,}  \delta'_{{\mathrm{2}}}  \mGLsym{,}  \delta_{{\mathrm{3}}}  )   \odot   ( \Delta_{{\mathrm{1}}}  \mGLsym{,}  \Delta_{{\mathrm{2}}}  \mGLsym{,}  \Delta_{{\mathrm{3}}} )   \mGLsym{;}  \Gamma  \vdash_{\mathsf{MS} }  \mGLnt{A}  \multimap  \mGLnt{B}}
        \]
        Given the above we know:
        \[
        \begin{array}{lll}
          \mathsf{CutRank} \, \mGLsym{(}  \Pi  \mGLsym{)} \, \mGLsym{=} \, \mathsf{CutRank} \, \mGLsym{(}  \Pi'  \mGLsym{)} \, \leq \,  \mathsf{Rank}  (  \mGLnt{X}  )\\
          \delta  \boxast [ { \delta_{{\mathrm{2}}} }^{ \mGLmv{n} } ]   \mGLsym{=}  \delta'_{{\mathrm{2}}}
        \end{array}
        \]
   \item \textbf{Left introduction of linear implication:}
   \[
    \inferrule* [flushleft,right=,left=$\Pi_{{\mathrm{1}}} :$] {
      \pi_1
    }{\delta_{{\mathrm{2}}}  \odot  \Delta_{{\mathrm{2}}}  \vdash_{\mathsf{GS} }  \mGLnt{X}}
\]

\[
  \inferrule* [flushleft,right=$\mGLdruleMSTXXGImpLName{}$,left=$\Pi_{{\mathrm{2}}} :$] {
    \inferrule* [flushleft,left=$\Pi_{{\mathrm{3}}} : $] {
      \pi_3
    }{\delta_{{\mathrm{1}}}  \mGLsym{,}  \delta  \mGLsym{,}  \delta_{{\mathrm{3}}}  \odot  \Delta_{{\mathrm{1}}}  \mGLsym{,}   \mGLnt{X} ^{ \mGLmv{n} }   \mGLsym{,}  \Delta_{{\mathrm{3}}}  \mGLsym{;}  \Gamma_{{\mathrm{2}}}  \vdash_{\mathsf{MS} }  \mGLnt{A} }\\
    \inferrule* [flushleft,right=,left=$\Pi_{{\mathrm{4}}} :$] {
      \pi_4
    }{ \delta_{{\mathrm{4}}}  \odot  \Delta_{{\mathrm{4}}}  \mGLsym{;}  \Gamma_{{\mathrm{1}}}  \mGLsym{,}  \mGLnt{B}  \mGLsym{,}  \Gamma_{{\mathrm{3}}}  \vdash_{\mathsf{MS} }  \mGLnt{C}}
  }{\delta_{{\mathrm{1}}}  \mGLsym{,}  \delta  \mGLsym{,}  \delta_{{\mathrm{3}}}  \mGLsym{,}  \delta_{{\mathrm{4}}}  \odot  \Delta_{{\mathrm{1}}}  \mGLsym{,}   \mGLnt{X} ^{ \mGLmv{n} }   \mGLsym{,}  \Delta_{{\mathrm{3}}}  \mGLsym{,}  \Delta_{{\mathrm{4}}}  \mGLsym{;}  \Gamma_{{\mathrm{1}}}  \mGLsym{,}  \mGLnt{A}  \multimap  \mGLnt{B}  \mGLsym{,}  \Gamma_{{\mathrm{2}}}  \vdash_{\mathsf{MS} }  \mGLnt{C} }
  \]

      We know:
      \[
      \begin{array}{lll}
        \mathsf{Depth}  (  \Pi_{{\mathrm{1}}}  )   +   \mathsf{Depth}  (  \Pi_{{\mathrm{3}}}  )   \, \mGLsym{<} \,  \mathsf{Depth}  (  \Pi_{{\mathrm{1}}}  )   +   \mathsf{Depth}  (  \Pi_{{\mathrm{2}}}  )\\
        \mathsf{CutRank} \, \mGLsym{(}  \Pi_{{\mathrm{3}}}  \mGLsym{)} \, \leq \, \mathsf{CutRank} \, \mGLsym{(}  \Pi_{{\mathrm{2}}}  \mGLsym{)} \, \leq \,  \mathsf{Rank}  (  \mGLnt{X}  )
      \end{array}
      \]

      and so applying the induction hypothesis
      to $\Pi_{{\mathrm{1}}}$ and $\Pi_{{\mathrm{3}}}$
      implies that there is a proof $\Pi'$ of
      $\delta_{{\mathrm{1}}}  \mGLsym{,}  \delta'_{{\mathrm{2}}}  \mGLsym{,}  \delta_{{\mathrm{3}}}  \odot  \Delta_{{\mathrm{1}}}  \mGLsym{,}  \Delta_{{\mathrm{2}}}  \mGLsym{,}  \Delta_{{\mathrm{3}}}  \mGLsym{;}  \Gamma_{{\mathrm{2}}}  \vdash_{\mathsf{MS} }  \mGLnt{A}$ with
      $\mathsf{CutRank} \, \mGLsym{(}  \Pi'  \mGLsym{)} \, \leq \,  \mathsf{Rank}  (  \mGLnt{X}  )$ and $\delta  \boxast [ { \delta_{{\mathrm{2}}} }^{ \mGLmv{n} } ]   \mGLsym{=}  \delta'_{{\mathrm{2}}}$.
      Thus, we construct the following proof $\Pi$:

      \[
        \inferrule* [flushleft,right=$\mGLdruleMSTXXGImpLName{}$,left=$\Pi :$] {
          \inferrule* [flushleft,left=$\Pi' : $] {
            \pi'
          }{\delta_{{\mathrm{1}}}  \mGLsym{,}  \delta'_{{\mathrm{2}}}  \mGLsym{,}  \delta_{{\mathrm{3}}}  \odot  \Delta_{{\mathrm{1}}}  \mGLsym{,}  \Delta_{{\mathrm{2}}}  \mGLsym{,}  \Delta_{{\mathrm{3}}}  \mGLsym{;}  \Gamma_{{\mathrm{2}}}  \vdash_{\mathsf{MS} }  \mGLnt{A}}\\
          \inferrule* [flushleft,right=,left=$\Pi_{{\mathrm{4}}} :$] {
            \pi_4
          }{\delta_{{\mathrm{4}}}  \odot  \Delta_{{\mathrm{4}}}  \mGLsym{;}  \Gamma_{{\mathrm{1}}}  \mGLsym{,}  \mGLnt{B}  \mGLsym{,}  \Gamma_{{\mathrm{3}}}  \vdash_{\mathsf{MS} }  \mGLnt{C}}
        }{\delta_{{\mathrm{1}}}  \mGLsym{,}  \delta'_{{\mathrm{2}}}  \mGLsym{,}  \delta_{{\mathrm{3}}}  \mGLsym{,}  \delta_{{\mathrm{4}}}  \odot  \Delta_{{\mathrm{1}}}  \mGLsym{,}  \Delta_{{\mathrm{2}}}  \mGLsym{,}  \Delta_{{\mathrm{3}}}  \mGLsym{,}  \Delta_{{\mathrm{4}}}  \mGLsym{;}  \Gamma_{{\mathrm{1}}}  \mGLsym{,}  \mGLnt{A}  \multimap  \mGLnt{B}  \mGLsym{,}  \Gamma_{{\mathrm{2}}}  \vdash_{\mathsf{MS} }  \mGLnt{C}}
        \]
        Given the above we know:
        \[
        \begin{array}{lll}
          \mathsf{CutRank} \, \mGLsym{(}  \Pi  \mGLsym{)} \, \mGLsym{=} \, \mGLkw{Max} \, \mGLsym{(}  \mathsf{CutRank} \, \mGLsym{(}  \Pi'  \mGLsym{)}  \mGLsym{,}  \mathsf{CutRank} \, \mGLsym{(}  \Pi_{{\mathrm{4}}}  \mGLsym{)}  \mGLsym{)} \, \leq \,  \mathsf{Rank}  (  \mGLnt{X}  )\\
          \delta  \boxast [ { \delta_{{\mathrm{2}}} }^{ \mGLmv{n} } ]   \mGLsym{=}  \delta'_{{\mathrm{2}}}
        \end{array}
        \]

   \item \textbf{Right introduction of Grd:}
   \[
    \begin{array}{lll}
      \inferrule* [flushleft,right=,left=$\Pi_{{\mathrm{1}}} :$] {
        \pi_1
      }{\delta_{{\mathrm{2}}}  \odot  \Delta_{{\mathrm{2}}}  \vdash_{\mathsf{GS} }  \mGLnt{X}}
      & \quad &
      \inferrule* [flushleft,right=$\mGLdruleMSTXXGrdRName{}$,left=$\Pi_{{\mathrm{2}}} :$] {
        \inferrule* [flushleft,right=, left=$\Pi_{{\mathrm{3}}} :$] {
        \pi_3
        }{(  \delta_{{\mathrm{1}}}  \mGLsym{,}  \delta  \mGLsym{,}  \delta_{{\mathrm{3}}}  )   \odot   ( \Delta_{{\mathrm{1}}}  \mGLsym{,}   \mGLnt{X} ^{ \mGLmv{n} }   \mGLsym{,}  \Delta_{{\mathrm{3}}} )   \vdash_{\mathsf{GS} }  \mGLnt{Y}}
      }{(  \mGLnt{r}  *  \delta_{{\mathrm{1}}}  \mGLsym{,}  \delta  \mGLsym{,}  \delta_{{\mathrm{3}}}  )   \odot   ( \Delta_{{\mathrm{1}}}  \mGLsym{,}   \mGLnt{X} ^{ \mGLmv{n} }   \mGLsym{,}  \Delta_{{\mathrm{3}}} )   \mGLsym{;}  \emptyset  \vdash_{\mathsf{MS} }   \mathsf{Grd} _{ \mGLnt{r} }\, \mGLnt{Y}}
    \end{array}
    \]
    We know the following:
    \[
      \begin{array}{lll}
        \mathsf{Depth}  (  \Pi_{{\mathrm{1}}}  )   +   \mathsf{Depth}  (  \Pi_{{\mathrm{2}}}  ) = \mathsf{Depth}  (  \Pi_{{\mathrm{1}}}  )   +  \mGLsym{(}    \mathsf{Depth}  (  \Pi_{{\mathrm{3}}}  )   + 1   \mGLsym{)}\\
        \mathsf{CutRank} \, \mGLsym{(}  \Pi_{{\mathrm{1}}}  \mGLsym{)} \, \leq \,  \mathsf{Rank}  (  \mGLnt{X}  )\\
        \mathsf{CutRank} \, \mGLsym{(}  \Pi_{{\mathrm{2}}}  \mGLsym{)} \, \mGLsym{=} \, \mathsf{CutRank} \, \mGLsym{(}  \Pi_{{\mathrm{3}}}  \mGLsym{)} \, \leq \,  \mathsf{Rank}  (  \mGLnt{X}  )
      \end{array}
    \]
    These imply that:
    \[
      \begin{array}{lll}
        \mathsf{Depth}  (  \Pi_{{\mathrm{1}}}  )   +   \mathsf{Depth}  (  \Pi_{{\mathrm{3}}}  )   \, \mGLsym{<} \,  \mathsf{Depth}  (  \Pi_{{\mathrm{1}}}  )   +   \mathsf{Depth}  (  \Pi_{{\mathrm{2}}}  )\\
        \mathsf{CutRank} \, \mGLsym{(}  \Pi_{{\mathrm{3}}}  \mGLsym{)} \, \leq \,  \mathsf{Rank}  (  \mGLnt{X}  )
      \end{array}
    \]
    Thus, we apply the induction hypothesis of Lemma~\ref{lemma:cut_reduction_for_mgl} (1) to
    $\Pi_{{\mathrm{1}}}$ and $\Pi_{{\mathrm{3}}}$ to obtain a proof $\Pi'$ of the sequent
    $(  \delta_{{\mathrm{1}}}  \mGLsym{,}  \delta'_{{\mathrm{2}}}  \mGLsym{,}  \delta_{{\mathrm{3}}}  )   \odot   ( \Delta_{{\mathrm{1}}}  \mGLsym{,}  \Delta_{{\mathrm{2}}}  \mGLsym{,}  \Delta_{{\mathrm{3}}} )   \vdash_{\mathsf{GS} }  \mGLnt{Y}$ with $\mathsf{CutRank} \, \mGLsym{(}  \Pi'  \mGLsym{)} \, \leq \,  \mathsf{Rank}  (  \mGLsym{(}  \mGLnt{X}  \mGLsym{)}  )$
    and $(   \delta  \boxast [ { \delta_{{\mathrm{2}}} }^{ \mGLmv{n} } ]   )   \mGLsym{=}  \delta'_{{\mathrm{2}}}$.
    Now we define the proof $\Pi$ as follows:
    \[
      \inferrule* [flushleft,right=$\mGLdruleMSTXXGrdRName{}$, left=$\Pi' :$] {
        \inferrule* [flushleft,right=, left=$\Pi :$] {
          \pi'
        }{(  \delta_{{\mathrm{1}}}  \mGLsym{,}  \delta'_{{\mathrm{2}}}  \mGLsym{,}  \delta_{{\mathrm{3}}}  )   \odot   ( \Delta_{{\mathrm{1}}}  \mGLsym{,}  \Delta_{{\mathrm{2}}}  \mGLsym{,}  \Delta_{{\mathrm{3}}} )   \vdash_{\mathsf{GS} }  \mGLnt{Y}}
      }{(  \mGLnt{r}  *  \delta_{{\mathrm{1}}}  \mGLsym{,}  \delta'_{{\mathrm{2}}}  \mGLsym{,}  \delta_{{\mathrm{3}}}  )   \odot   ( \Delta_{{\mathrm{1}}}  \mGLsym{,}  \Delta_{{\mathrm{2}}}  \mGLsym{,}  \Delta_{{\mathrm{3}}} )   \mGLsym{;}  \emptyset  \vdash_{\mathsf{MS} }   \mathsf{Grd} _{ \mGLnt{r} }\, \mGLnt{Y}}
      \]
    with: $\mathsf{CutRank} \, \mGLsym{(}  \Pi  \mGLsym{)} \, \mGLsym{=} \, \mathsf{CutRank} \, \mGLsym{(}  \Pi'  \mGLsym{)} \, \leq \,  \mathsf{Rank}  (  \mGLnt{X}  )$
    \item \textbf{Weakening:}
    \[
      \inferrule* [flushleft,right=,left=$\Pi_{{\mathrm{1}}} :$] {
        \pi_1
      }{\delta_{{\mathrm{2}}}  \odot  \Delta_{{\mathrm{2}}}  \vdash_{\mathsf{GS} }  \mGLnt{X}}
\]

  \[
    \inferrule* [flushleft,right=$\mGLdruleMSTXXWeakName{}$,left=$\Pi_{{\mathrm{2}}} :$] {
      \inferrule* [flushleft,left=$\Pi_{{\mathrm{3}}} : $] {
        \pi_3
      }{(  \delta_{{\mathrm{1}}}  \mGLsym{,}  \delta  \mGLsym{,}  \delta_{{\mathrm{3}}}  \mGLsym{,}  \delta_{{\mathrm{4}}}  )   \odot   ( \Delta_{{\mathrm{1}}}  \mGLsym{,}   \mGLnt{X} ^{ \mGLmv{n} }   \mGLsym{,}  \Delta_{{\mathrm{3}}}  \mGLsym{,}  \Delta_{{\mathrm{4}}} )   \mGLsym{;}  \Gamma  \vdash_{\mathsf{MS} }  \mGLnt{A}}\\
    }{(  \delta_{{\mathrm{1}}}  \mGLsym{,}  \delta  \mGLsym{,}  \delta_{{\mathrm{3}}}  \mGLsym{,}  \mathsf{0}  \mGLsym{,}  \delta_{{\mathrm{4}}}  )   \odot   ( \Delta_{{\mathrm{1}}}  \mGLsym{,}   \mGLnt{X} ^{ \mGLmv{n} }   \mGLsym{,}  \Delta_{{\mathrm{3}}}  \mGLsym{,}  \mGLnt{Z}  \mGLsym{,}  \Delta_{{\mathrm{4}}} )   \mGLsym{;}  \Gamma  \vdash_{\mathsf{MS} }  \mGLnt{A}}
    \]
    Similar to the unit of the graded tensor case.
    \item \textbf{Weakening: second case}
    \[
      \inferrule* [flushleft,right=,left=$\Pi_{{\mathrm{1}}} :$] {
        \pi_1
      }{\delta_{{\mathrm{3}}}  \odot  \Delta_{{\mathrm{3}}}  \vdash_{\mathsf{GS} }  \mGLnt{X}}
\]
  \[
    \inferrule* [flushleft,right=$\mGLdruleMSTXXWeakName{}$,left=$\Pi_{{\mathrm{2}}} :$] {
      \inferrule* [flushleft,left=$\Pi_{{\mathrm{3}}} : $] {
        \pi_3
      }{(  \delta_{{\mathrm{1}}}  \mGLsym{,}  \delta_{{\mathrm{2}}}  \mGLsym{,}  \delta  \mGLsym{,}  \delta_{{\mathrm{4}}}  )   \odot   ( \Delta_{{\mathrm{1}}}  \mGLsym{,}  \Delta_{{\mathrm{2}}}  \mGLsym{,}   \mGLnt{X} ^{ \mGLmv{n} }   \mGLsym{,}  \Delta_{{\mathrm{4}}} )   \mGLsym{;}  \Gamma  \vdash_{\mathsf{MS} }  \mGLnt{A}}\\
    }{(  \delta_{{\mathrm{1}}}  \mGLsym{,}  \mathsf{0}  \mGLsym{,}  \delta_{{\mathrm{2}}}  \mGLsym{,}  \delta  \mGLsym{,}  \delta_{{\mathrm{4}}}  )   \odot   ( \Delta_{{\mathrm{1}}}  \mGLsym{,}  \mGLnt{Z}  \mGLsym{,}  \Delta_{{\mathrm{2}}}  \mGLsym{,}   \mGLnt{X} ^{ \mGLmv{n} }   \mGLsym{,}  \Delta_{{\mathrm{4}}} )   \mGLsym{;}  \Gamma  \vdash_{\mathsf{MS} }  \mGLnt{A}}
    \]
    Similar to the unit of the graded tensor case.
  \end{enumerate}

 \item \textbf{Structural}
  \begin{enumerate}
    \item \textbf{Weakening}
    \[
      \inferrule* [flushleft,right=,left=$\Pi_{{\mathrm{1}}} :$] {
        \pi_1
      }{\delta_{{\mathrm{2}}}  \odot  \Delta_{{\mathrm{2}}}  \vdash_{\mathsf{GS} }  \mGLnt{X}}
\]
  \[
    \inferrule* [flushleft,right=$\mGLdruleMSTXXWeakName{}$,left=$\Pi_{{\mathrm{2}}} :$] {
      \inferrule* [flushleft,left=$\Pi_{{\mathrm{3}}} : $] {
        \pi_3
      }{(  \delta_{{\mathrm{1}}}  \mGLsym{,}  \delta_{{\mathrm{3}}}  )   \odot   ( \Delta_{{\mathrm{1}}}  \mGLsym{,}  \Delta_{{\mathrm{3}}} )   \mGLsym{;}  \Gamma  \vdash_{\mathsf{MS} }  \mGLnt{A}}\\
    }{(  \delta_{{\mathrm{1}}}  \mGLsym{,}  \mathsf{0}  \mGLsym{,}  \delta_{{\mathrm{3}}}  )   \odot   ( \Delta_{{\mathrm{1}}}  \mGLsym{,}  \mGLnt{X}  \mGLsym{,}  \Delta_{{\mathrm{3}}} )   \mGLsym{;}  \Gamma  \vdash_{\mathsf{MS} }  \mGLnt{A}}
    \]
Similar to the structural weakening case for GS
    \item \textbf{Contraction}
    \[
      \inferrule* [flushleft,right=,left=$\Pi_{{\mathrm{1}}} :$] {
        \pi_1
      }{\delta_{{\mathrm{2}}}  \odot  \Delta_{{\mathrm{2}}}  \vdash_{\mathsf{GS} }  \mGLnt{X}}
\]
\[
  \inferrule* [flushleft,right=$\mGLdruleMSTXXContName{}$,left=$\Pi_{{\mathrm{2}}} :$] {
    \inferrule* [flushleft,right=,left=$\Pi_{{\mathrm{3}}} :$] {
      \pi_3
    }{(  \delta_{{\mathrm{1}}}  \mGLsym{,}  \mGLnt{r_{{\mathrm{1}}}}  \mGLsym{,}  \mGLnt{r_{{\mathrm{2}}}}  \mGLsym{,}  \delta_{{\mathrm{3}}}  )   \odot   ( \Delta_{{\mathrm{1}}}  \mGLsym{,}  \mGLnt{X}  \mGLsym{,}  \mGLnt{X}  \mGLsym{,}  \Delta_{{\mathrm{3}}} )   \mGLsym{;}  \Gamma  \vdash_{\mathsf{MS} }  \mGLnt{A}}
  }{(  \delta_{{\mathrm{1}}}  \mGLsym{,}  \mGLnt{r_{{\mathrm{1}}}}  +  \mGLnt{r_{{\mathrm{2}}}}  \mGLsym{,}  \delta_{{\mathrm{3}}}  )   \odot   ( \Delta_{{\mathrm{1}}}  \mGLsym{,}  \mGLnt{X}  \mGLsym{,}  \Delta_{{\mathrm{3}}} )   \mGLsym{;}  \Gamma  \vdash_{\mathsf{MS} }  \mGLnt{A}}
  \]
 Similar to the structural contraction case for GS
   \item \textbf{Exchange}
   \[
     \inferrule* [flushleft,right=,left=$\Pi_{{\mathrm{1}}} :$] {
       \pi_1
     }{\delta_{{\mathrm{2}}}  \odot  \Delta_{{\mathrm{2}}}  \vdash_{\mathsf{GS} }  \mGLnt{X}}
\]
\[
 \inferrule* [flushleft,right=$\mGLdruleMSTXXExName{}$,left=$\Pi_{{\mathrm{2}}} :$] {
   \inferrule* [flushleft,right=,left=$\Pi_{{\mathrm{3}}} :$] {
     \pi_3
   }{(  \delta_{{\mathrm{1}}}  \mGLsym{,}  \delta  \mGLsym{,}  \delta_{{\mathrm{3}}}  )   \odot   ( \Delta_{{\mathrm{1}}}  \mGLsym{,}   \mGLnt{X} ^{ \mGLmv{n} }   \mGLsym{,}  \Delta_{{\mathrm{3}}} )   \mGLsym{;}  \Gamma_{{\mathrm{1}}}  \mGLsym{,}  \mGLnt{A}  \mGLsym{,}  \mGLnt{B}  \mGLsym{,}  \Gamma_{{\mathrm{2}}}  \vdash_{\mathsf{MS} }  \mGLnt{A}}
 }{(  \delta_{{\mathrm{1}}}  \mGLsym{,}  \delta  \mGLsym{,}  \delta_{{\mathrm{3}}}  )   \odot   ( \Delta_{{\mathrm{1}}}  \mGLsym{,}   \mGLnt{X} ^{ \mGLmv{n} }   \mGLsym{,}  \Delta_{{\mathrm{3}}} )   \mGLsym{;}  \Gamma_{{\mathrm{1}}}  \mGLsym{,}  \mGLnt{B}  \mGLsym{,}  \mGLnt{A}  \mGLsym{,}  \Gamma_{{\mathrm{2}}}  \vdash_{\mathsf{MS} }  \mGLnt{A}}
 \]

 We know:
 \[
 \begin{array}{lll}
   \mathsf{Depth}  (  \Pi_{{\mathrm{1}}}  )   +   \mathsf{Depth}  (  \Pi_{{\mathrm{3}}}  )   \, \mGLsym{<} \,  \mathsf{Depth}  (  \Pi_{{\mathrm{1}}}  )   +   \mathsf{Depth}  (  \Pi_{{\mathrm{2}}}  )\\
   \mathsf{CutRank} \, \mGLsym{(}  \Pi_{{\mathrm{3}}}  \mGLsym{)} \, \leq \, \mathsf{CutRank} \, \mGLsym{(}  \Pi_{{\mathrm{2}}}  \mGLsym{)} \, \leq \,  \mathsf{Rank}  (  \mGLnt{X}  )
 \end{array}
 \]

 and so applying the induction hypothesis
 to $\Pi_{{\mathrm{1}}}$ and $\Pi_{{\mathrm{3}}}$
 implies that there is a proof $\Pi'$ of
 $(  \delta_{{\mathrm{1}}}  \mGLsym{,}  \delta  \mGLsym{,}  \delta_{{\mathrm{3}}}  )   \odot   ( \Delta_{{\mathrm{1}}}  \mGLsym{,}  \Delta_{{\mathrm{2}}}  \mGLsym{,}  \Delta_{{\mathrm{3}}} )   \mGLsym{;}  \Gamma_{{\mathrm{1}}}  \mGLsym{,}  \mGLnt{A}  \mGLsym{,}  \mGLnt{B}  \mGLsym{,}  \Gamma_{{\mathrm{2}}}  \vdash_{\mathsf{MS} }  \mGLnt{A}$ with
 $\mathsf{CutRank} \, \mGLsym{(}  \Pi'  \mGLsym{)} \, \leq \,  \mathsf{Rank}  (  \mGLnt{X}  )$ and $(   \delta  \boxast [ { \delta_{{\mathrm{2}}} }^{ \mGLmv{n} } ]   )   \mGLsym{=}  \delta'_{{\mathrm{2}}}$.
 Thus, we construct the following proof $\Pi$:

 \[
  \inferrule* [flushleft,right=$\mGLdruleMSTXXExName{}$,left=$\Pi :$] {
    \inferrule* [flushleft,right=,left=$\Pi' :$] {
      \pi'
    }{(  \delta_{{\mathrm{1}}}  \mGLsym{,}  \delta'_{{\mathrm{2}}}  \mGLsym{,}  \delta_{{\mathrm{3}}}  )   \odot   ( \Delta_{{\mathrm{1}}}  \mGLsym{,}  \Delta_{{\mathrm{2}}}  \mGLsym{,}  \Delta_{{\mathrm{3}}} )   \mGLsym{;}  \Gamma_{{\mathrm{1}}}  \mGLsym{,}  \mGLnt{A}  \mGLsym{,}  \mGLnt{B}  \mGLsym{,}  \Gamma_{{\mathrm{2}}}  \vdash_{\mathsf{MS} }  \mGLnt{A}}
  }{(  \delta_{{\mathrm{1}}}  \mGLsym{,}  \delta'_{{\mathrm{2}}}  \mGLsym{,}  \delta_{{\mathrm{3}}}  )   \odot   ( \Delta_{{\mathrm{1}}}  \mGLsym{,}  \Delta_{{\mathrm{2}}}  \mGLsym{,}  \Delta_{{\mathrm{3}}} )   \mGLsym{;}  \Gamma_{{\mathrm{1}}}  \mGLsym{,}  \mGLnt{B}  \mGLsym{,}  \mGLnt{A}  \mGLsym{,}  \Gamma_{{\mathrm{2}}}  \vdash_{\mathsf{MS} }  \mGLnt{A}}
  \]
   Given the above, we know:
   \[
     \begin{array}{lll}
       \mathsf{CutRank} \, \mGLsym{(}  \Pi  \mGLsym{)} \, \mGLsym{=} \, \mathsf{CutRank} \, \mGLsym{(}  \Pi'  \mGLsym{)} \, \leq \,  \mathsf{Rank}  (  \mGLnt{X}  )\\
       \delta  \boxast [ { \delta_{{\mathrm{2}}} }^{ \mGLmv{n} } ]   \mGLsym{=}  \delta'_{{\mathrm{2}}}\\
     \end{array}
     \]

   \item \textbf{Approximation}
   \[
    \inferrule* [flushleft,right=,left=$\Pi_{{\mathrm{1}}} :$] {
      \pi_1
    }{\delta_{{\mathrm{2}}}  \odot  \Delta_{{\mathrm{2}}}  \vdash_{\mathsf{GS} }  \mGLnt{X}}
\]
\[
\inferrule* [flushleft,right=$\mGLdruleMSTXXSubName{}$,left=$\Pi_{{\mathrm{2}}} :$] {
  \inferrule* [flushleft,right=,left=$\Pi_{{\mathrm{3}}} :$] {
    \pi_3
  }{(  \delta_{{\mathrm{1}}}  \mGLsym{,}  \delta'  \mGLsym{,}  \delta_{{\mathrm{3}}}  )   \odot   ( \Delta_{{\mathrm{1}}}  \mGLsym{,}   \mGLnt{X} ^{ \mGLmv{n} }   \mGLsym{,}  \Delta_{{\mathrm{3}}} )   \mGLsym{;}  \Gamma  \vdash_{\mathsf{MS} }  \mGLnt{A}}\\{\delta'  \leq  \delta}
}{(  \delta_{{\mathrm{1}}}  \mGLsym{,}  \delta  \mGLsym{,}  \delta_{{\mathrm{3}}}  )   \odot   ( \Delta_{{\mathrm{1}}}  \mGLsym{,}   \mGLnt{X} ^{ \mGLmv{n} }   \mGLsym{,}  \Delta_{{\mathrm{3}}} )   \mGLsym{;}  \Gamma  \vdash_{\mathsf{MS} }  \mGLnt{A}}
\]
Similar to the case for GS
  \end{enumerate}
\end{enumerate}

\end{proof}
\begin{lemma}[Cut Reduction MS]
    If $\Pi_{{\mathrm{1}}}$ is a proof of $\delta_{{\mathrm{2}}}  \odot  \Delta_{{\mathrm{2}}}  \mGLsym{;}  \Gamma_{{\mathrm{2}}}  \vdash_{\mathsf{MS} }  \mGLnt{A}$
    and $\Pi_{{\mathrm{2}}}$ is a proof of $\delta_{{\mathrm{1}}}  \odot  \Delta_{{\mathrm{1}}}  \mGLsym{;}   ( \Gamma_{{\mathrm{1}}}  \mGLsym{,}  \mGLnt{A}  \mGLsym{,}  \Gamma_{{\mathrm{3}}} )   \vdash_{\mathsf{MS} }  \mGLnt{B}$
    with $\mathsf{CutRank} \, \mGLsym{(}  \Pi_{{\mathrm{1}}}  \mGLsym{)} \, \leq \,  \mathsf{Rank}  (  \mGLnt{A}  )$
    and $\mathsf{CutRank} \, \mGLsym{(}  \Pi_{{\mathrm{2}}}  \mGLsym{)} \, \leq \,  \mathsf{Rank}  (  \mGLnt{A}  )$,
    then there exists 
    a proof $\Pi$ of the sequent
    $(  \delta_{{\mathrm{1}}}  \mGLsym{,}  \delta_{{\mathrm{2}}}  )   \odot   ( \Delta_{{\mathrm{1}}}  \mGLsym{,}  \Delta_{{\mathrm{2}}} )   \mGLsym{;}   ( \Gamma_{{\mathrm{1}}}  \mGLsym{,}  \Gamma_{{\mathrm{2}}}  \mGLsym{,}  \Gamma_{{\mathrm{3}}} )   \vdash_{\mathsf{MS} }  \mGLnt{B}$
    with $\mathsf{CutRank} \, \mGLsym{(}  \Pi  \mGLsym{)} \, \leq \,  \mathsf{Rank}  (  \mGLnt{A}  )$.

  \end{lemma}

\begin{proof}
      This is by induction on $\mathsf{Depth}  (  \Pi_{{\mathrm{1}}}  )   +   \mathsf{Depth}  (  \Pi_{{\mathrm{2}}}  )$.
\begin{enumerate}
  \item \textbf{Commuting Conversions}
    \begin{enumerate}
      \item \textbf{left-side:} Suppose we have

      \[
        \inferrule* [flushleft,right=,left=$\Pi_{{\mathrm{1}}} :$] {
          \pi_1
        }{\delta_{{\mathrm{3}}}  \odot  \Delta_{{\mathrm{3}}}  \mGLsym{;}  \Gamma_{{\mathrm{3}}}  \vdash_{\mathsf{MS} }  \mGLnt{A}}
       \]
       \[
        \inferrule* [flushleft,right=$\mGLdruleMSTXXCutName{}$,left=$\Pi_{{\mathrm{2}}} :$] {
          \inferrule* [flushleft,right=,left=$\Pi_{{\mathrm{3}}} :$] {
            \pi_3
          }{\delta_{{\mathrm{2}}}  \odot  \Delta_{{\mathrm{2}}}  \mGLsym{;}   ( \Gamma_{{\mathrm{2}}}  \mGLsym{,}  \mGLnt{A}  \mGLsym{,}  \Gamma_{{\mathrm{4}}} )   \vdash_{\mathsf{MS} }  \mGLnt{B}}\\
          \inferrule* [flushleft,right=,left=$\Pi_{{\mathrm{4}}} :$] {
            \pi_4
          }{\delta_{{\mathrm{1}}}  \odot  \Delta_{{\mathrm{1}}}  \mGLsym{;}   ( \Gamma_{{\mathrm{1}}}  \mGLsym{,}  \mGLnt{B}  \mGLsym{,}  \Gamma_{{\mathrm{5}}} )   \vdash_{\mathsf{MS} }  \mGLnt{C}}
        }{\delta_{{\mathrm{1}}}  \mGLsym{,}  \delta_{{\mathrm{2}}}  \odot  \Delta_{{\mathrm{1}}}  \mGLsym{,}  \Delta_{{\mathrm{2}}}  \mGLsym{;}   ( \Gamma_{{\mathrm{1}}}  \mGLsym{,}  \Gamma_{{\mathrm{2}}}  \mGLsym{,}  \mGLnt{A}  \mGLsym{,}  \Gamma_{{\mathrm{4}}}  \mGLsym{,}  \Gamma_{{\mathrm{5}}} )   \vdash_{\mathsf{MS} }  \mGLnt{C}}
      \]

      We know:
      \[
      \begin{array}{lll}
        \mathsf{Depth}  (  \Pi_{{\mathrm{1}}}  )   +   \mathsf{Depth}  (  \Pi_{{\mathrm{3}}}  )   \, \mGLsym{<} \,  \mathsf{Depth}  (  \Pi_{{\mathrm{1}}}  )   +   \mathsf{Depth}  (  \Pi_{{\mathrm{2}}}  )\\
        \mathsf{CutRank} \, \mGLsym{(}  \Pi_{{\mathrm{3}}}  \mGLsym{)} \, \leq \, \mGLkw{Max} \, \mGLsym{(}   \mathsf{CutRank} \, \mGLsym{(}  \Pi_{{\mathrm{3}}}  \mGLsym{)}  \mGLsym{,}  \mathsf{CutRank} \, \mGLsym{(}  \Pi_{{\mathrm{4}}}  \mGLsym{)}  \mGLsym{,}   \mathsf{Rank}  (  \mGLnt{B}  )   + 1   \mGLsym{)} \, \leq \,  \mathsf{Rank}  (  \mGLnt{A}  )
      \end{array}
      \]

      and so applying the induction hypothesis
      to $\Pi_{{\mathrm{1}}}$ and $\Pi_{{\mathrm{3}}}$
      implies that there is a proof $\Pi'$ of
      $\delta_{{\mathrm{2}}}  \mGLsym{,}  \delta_{{\mathrm{3}}}  \odot  \Delta_{{\mathrm{2}}}  \mGLsym{,}  \Delta_{{\mathrm{3}}}  \mGLsym{;}   ( \Gamma_{{\mathrm{2}}}  \mGLsym{,}  \Gamma_{{\mathrm{3}}}  \mGLsym{,}  \Gamma_{{\mathrm{4}}} )   \vdash_{\mathsf{MS} }  \mGLnt{B}$ with
      $\mathsf{CutRank} \, \mGLsym{(}  \Pi'  \mGLsym{)} \, \leq \,  \mathsf{Rank}  (  \mGLnt{A}  )$.
      Thus, we construct the following proof $\Pi$:
      \[
        \inferrule* [flushleft,right=$\mGLdruleMSTXXCutName{}$,left=$\Pi :$] {
          \inferrule* [flushleft,right=,left=$\Pi' :$] {
            \pi'
          }{\delta_{{\mathrm{2}}}  \mGLsym{,}  \delta_{{\mathrm{3}}}  \odot  \Delta_{{\mathrm{2}}}  \mGLsym{,}  \Delta_{{\mathrm{3}}}  \mGLsym{;}   ( \Gamma_{{\mathrm{2}}}  \mGLsym{,}  \Gamma_{{\mathrm{3}}}  \mGLsym{,}  \Gamma_{{\mathrm{4}}} )   \vdash_{\mathsf{MS} }  \mGLnt{B}}\\
          \inferrule* [flushleft,right=,left=$\Pi_{{\mathrm{4}}} :$] {
            \pi_4
          }{\delta_{{\mathrm{1}}}  \odot  \Delta_{{\mathrm{1}}}  \mGLsym{;}   ( \Gamma_{{\mathrm{1}}}  \mGLsym{,}  \mGLnt{B}  \mGLsym{,}  \Gamma_{{\mathrm{5}}} )   \vdash_{\mathsf{MS} }  \mGLnt{C}}
        }{\delta_{{\mathrm{1}}}  \mGLsym{,}  \delta_{{\mathrm{2}}}  \mGLsym{,}  \delta_{{\mathrm{3}}}  \odot  \Delta_{{\mathrm{1}}}  \mGLsym{,}  \Delta_{{\mathrm{2}}}  \mGLsym{,}  \Delta_{{\mathrm{3}}}  \mGLsym{;}   ( \Gamma_{{\mathrm{1}}}  \mGLsym{,}  \Gamma_{{\mathrm{2}}}  \mGLsym{,}  \Gamma_{{\mathrm{3}}}  \mGLsym{,}  \Gamma_{{\mathrm{4}}}  \mGLsym{,}  \Gamma_{{\mathrm{5}}} )   \vdash_{\mathsf{MS} }  \mGLnt{C}}
      \]
      Given the above we know:
      \[
      \begin{array}{lll}
        \mathsf{CutRank} \, \mGLsym{(}  \Pi'  \mGLsym{)} \, \leq \,  \mathsf{Rank}  (  \mGLnt{A}  )\\
        \mathsf{CutRank} \, \mGLsym{(}  \Pi_{{\mathrm{2}}}  \mGLsym{)} \, \mGLsym{=} \, \mGLkw{Max} \, \mGLsym{(}   \mathsf{CutRank} \, \mGLsym{(}  \Pi_{{\mathrm{3}}}  \mGLsym{)}  \mGLsym{,}  \mathsf{CutRank} \, \mGLsym{(}  \Pi_{{\mathrm{4}}}  \mGLsym{)}  \mGLsym{,}   \mathsf{Rank}  (  \mGLnt{B}  )   + 1   \mGLsym{)} \, \leq \,  \mathsf{Rank}  (  \mGLnt{A}  )\\
      \end{array}
      \]
      This implies:
      \[
      \begin{array}{lll}
        \mathsf{CutRank} \, \mGLsym{(}  \Pi_{{\mathrm{4}}}  \mGLsym{)} \, \leq \,  \mathsf{Rank}  (  \mGLnt{A}  )\\
        \mathsf{Rank}  (  \mGLnt{B}  )   + 1  \, \leq \,  \mathsf{Rank}  (  \mGLnt{A}  )
      \end{array}
      \]
      Thus, we obtain our result
      \[
        \mathsf{CutRank} \, \mGLsym{(}  \Pi  \mGLsym{)} \, \mGLsym{=} \, \mGLkw{Max} \, \mGLsym{(}   \mathsf{CutRank} \, \mGLsym{(}  \Pi'  \mGLsym{)}  \mGLsym{,}  \mathsf{CutRank} \, \mGLsym{(}  \Pi_{{\mathrm{4}}}  \mGLsym{)}  \mGLsym{,}   \mathsf{Rank}  (  \mGLnt{B}  )   + 1   \mGLsym{)} \, \leq \,  \mathsf{Rank}  (  \mGLnt{A}  )
      \]

\item \textbf{cut vs. right-side cut (left case):} Suppose we have:

\[
  \inferrule* [flushleft,right=,left=$\Pi_{{\mathrm{1}}} :$] {
    \pi_1
  }{\delta_{{\mathrm{3}}}  \odot  \Delta_{{\mathrm{3}}}  \mGLsym{;}  \Gamma_{{\mathrm{2}}}  \vdash_{\mathsf{MS} }  \mGLnt{A}}
 \]
 \[
  \inferrule* [flushleft,right=$\mGLdruleMSTXXCutName{}$,left=$\Pi_{{\mathrm{2}}} :$] {
    \inferrule* [flushleft,right=,left=$\Pi_{{\mathrm{3}}} :$] {
      \pi_3
    }{\delta_{{\mathrm{2}}}  \odot  \Delta_{{\mathrm{2}}}  \mGLsym{;}  \Gamma_{{\mathrm{4}}}  \vdash_{\mathsf{MS} }  \mGLnt{B}}\\
    \inferrule* [flushleft,right=,left=$\Pi_{{\mathrm{4}}} :$] {
      \pi_4
    }{\delta_{{\mathrm{1}}}  \odot  \Delta_{{\mathrm{1}}}  \mGLsym{;}   ( \Gamma_{{\mathrm{1}}}  \mGLsym{,}  \mGLnt{A}  \mGLsym{,}  \Gamma_{{\mathrm{3}}}  \mGLsym{,}  \mGLnt{B}  \mGLsym{,}  \Gamma_{{\mathrm{5}}} )   \vdash_{\mathsf{MS} }  \mGLnt{C}}
  }{\delta_{{\mathrm{1}}}  \mGLsym{,}  \delta_{{\mathrm{2}}}  \odot  \Delta_{{\mathrm{1}}}  \mGLsym{,}  \Delta_{{\mathrm{2}}}  \mGLsym{;}   ( \Gamma_{{\mathrm{1}}}  \mGLsym{,}  \mGLnt{A}  \mGLsym{,}  \Gamma_{{\mathrm{3}}}  \mGLsym{,}  \Gamma_{{\mathrm{4}}}  \mGLsym{,}  \Gamma_{{\mathrm{5}}} )   \vdash_{\mathsf{MS} }  \mGLnt{C}}
\]

We know:
\[
\begin{array}{lll}
  \mathsf{Depth}  (  \Pi_{{\mathrm{1}}}  )   +   \mathsf{Depth}  (  \Pi_{{\mathrm{4}}}  )   \, \mGLsym{<} \,  \mathsf{Depth}  (  \Pi_{{\mathrm{1}}}  )   +   \mathsf{Depth}  (  \Pi_{{\mathrm{2}}}  )\\
  \mathsf{CutRank} \, \mGLsym{(}  \Pi_{{\mathrm{4}}}  \mGLsym{)} \, \leq \, \mGLkw{Max} \, \mGLsym{(}   \mathsf{CutRank} \, \mGLsym{(}  \Pi_{{\mathrm{3}}}  \mGLsym{)}  \mGLsym{,}  \mathsf{CutRank} \, \mGLsym{(}  \Pi_{{\mathrm{4}}}  \mGLsym{)}  \mGLsym{,}   \mathsf{Rank}  (  \mGLnt{B}  )   + 1   \mGLsym{)} \, \leq \,  \mathsf{Rank}  (  \mGLnt{A}  )
\end{array}
\]
and so applying the induction hypothesis
to $\Pi_{{\mathrm{1}}}$ and $\Pi_{{\mathrm{4}}}$
implies that there is a proof $\Pi'$ of
$\delta_{{\mathrm{1}}}  \mGLsym{,}  \delta_{{\mathrm{3}}}  \odot  \Delta_{{\mathrm{1}}}  \mGLsym{,}  \Delta_{{\mathrm{3}}}  \mGLsym{;}   ( \Gamma_{{\mathrm{1}}}  \mGLsym{,}  \Gamma_{{\mathrm{2}}}  \mGLsym{,}  \Gamma_{{\mathrm{3}}}  \mGLsym{,}  \mGLnt{B}  \mGLsym{,}  \Gamma_{{\mathrm{5}}} )   \vdash_{\mathsf{MS} }  \mGLnt{C}$ with
$\mathsf{CutRank} \, \mGLsym{(}  \Pi'  \mGLsym{)} \, \leq \,  \mathsf{Rank}  (  \mGLnt{A}  )$
Thus, we construct the following proof $\Pi$:

\[
        \inferrule* [flushleft,right=$\mGLdruleMSTXXGExName{}$,left=$\Pi :$] {
          \inferrule* [flushleft,right=$\mGLdruleMSTXXCutName{}$,left=] {
          \inferrule* [flushleft,left=$\Pi_{{\mathrm{3}}} : $] {
            \pi_3
          }{\delta_{{\mathrm{2}}}  \odot  \Delta_{{\mathrm{2}}}  \mGLsym{;}  \Gamma_{{\mathrm{4}}}  \vdash_{\mathsf{MS} }  \mGLnt{B}}\\
          \inferrule* [flushleft,right=,left=$\Pi' :$] {
            \pi'
          }{\delta_{{\mathrm{1}}}  \mGLsym{,}  \delta_{{\mathrm{3}}}  \odot  \Delta_{{\mathrm{1}}}  \mGLsym{,}  \Delta_{{\mathrm{3}}}  \mGLsym{;}   ( \Gamma_{{\mathrm{1}}}  \mGLsym{,}  \Gamma_{{\mathrm{2}}}  \mGLsym{,}  \Gamma_{{\mathrm{3}}}  \mGLsym{,}  \mGLnt{B}  \mGLsym{,}  \Gamma_{{\mathrm{5}}} )   \vdash_{\mathsf{MS} }  \mGLnt{C}}
        }{\delta_{{\mathrm{1}}}  \mGLsym{,}  \delta_{{\mathrm{3}}}  \mGLsym{,}  \delta_{{\mathrm{2}}}  \odot  \Delta_{{\mathrm{1}}}  \mGLsym{,}  \Delta_{{\mathrm{3}}}  \mGLsym{,}  \Delta_{{\mathrm{2}}}  \mGLsym{;}   ( \Gamma_{{\mathrm{1}}}  \mGLsym{,}  \Gamma_{{\mathrm{2}}}  \mGLsym{,}  \Gamma_{{\mathrm{3}}}  \mGLsym{,}  \Gamma_{{\mathrm{4}}}  \mGLsym{,}  \Gamma_{{\mathrm{5}}} )   \vdash_{\mathsf{MS} }  \mGLnt{C} }
        }{\delta_{{\mathrm{1}}}  \mGLsym{,}  \delta_{{\mathrm{2}}}  \mGLsym{,}  \delta_{{\mathrm{3}}}  \odot  \Delta_{{\mathrm{1}}}  \mGLsym{,}  \Delta_{{\mathrm{2}}}  \mGLsym{,}  \Delta_{{\mathrm{3}}}  \mGLsym{;}   ( \Gamma_{{\mathrm{1}}}  \mGLsym{,}  \Gamma_{{\mathrm{2}}}  \mGLsym{,}  \Gamma_{{\mathrm{3}}}  \mGLsym{,}  \Gamma_{{\mathrm{4}}}  \mGLsym{,}  \Gamma_{{\mathrm{5}}} )   \vdash_{\mathsf{MS} }  \mGLnt{C}}
        \]

Given the above we know:
\[
\begin{array}{lll}
  \mathsf{CutRank} \, \mGLsym{(}  \Pi'  \mGLsym{)} \, \leq \,  \mathsf{Rank}  (  \mGLnt{A}  )\\
  \mathsf{CutRank} \, \mGLsym{(}  \Pi_{{\mathrm{2}}}  \mGLsym{)} \, \mGLsym{=} \, \mGLkw{Max} \, \mGLsym{(}   \mathsf{CutRank} \, \mGLsym{(}  \Pi_{{\mathrm{3}}}  \mGLsym{)}  \mGLsym{,}  \mathsf{CutRank} \, \mGLsym{(}  \Pi_{{\mathrm{4}}}  \mGLsym{)}  \mGLsym{,}   \mathsf{Rank}  (  \mGLnt{B}  )   + 1   \mGLsym{)} \, \leq \,  \mathsf{Rank}  (  \mGLnt{A}  )\\
\end{array}
\]
This implies:
\[
\begin{array}{lll}
  \mathsf{CutRank} \, \mGLsym{(}  \Pi_{{\mathrm{3}}}  \mGLsym{)} \, \leq \,  \mathsf{Rank}  (  \mGLnt{A}  )\\
  \mathsf{Rank}  (  \mGLnt{B}  )   + 1  \, \leq \,  \mathsf{Rank}  (  \mGLnt{A}  )
\end{array}
\]
Thus, we obtain our result:
\[
  \mathsf{CutRank} \, \mGLsym{(}  \Pi  \mGLsym{)} \, \mGLsym{=} \, \mGLkw{Max} \, \mGLsym{(}   \mathsf{CutRank} \, \mGLsym{(}  \Pi_{{\mathrm{3}}}  \mGLsym{)}  \mGLsym{,}  \mathsf{CutRank} \, \mGLsym{(}  \Pi'  \mGLsym{)}  \mGLsym{,}   \mathsf{Rank}  (  \mGLnt{B}  )   + 1   \mGLsym{)} \, \leq \,  \mathsf{Rank}  (  \mGLnt{A}  )
  \]

  \item \textbf{cut vs. right-side cut (right case):} Suppose we have:
  \[
    \inferrule* [flushleft,right=,left=$\Pi_{{\mathrm{1}}} :$] {
      \pi_1
    }{\delta_{{\mathrm{3}}}  \odot  \Delta_{{\mathrm{3}}}  \mGLsym{;}  \Gamma_{{\mathrm{2}}}  \vdash_{\mathsf{MS} }  \mGLnt{A}}
   \]
   \[
    \inferrule* [flushleft,right=$\mGLdruleMSTXXCutName{}$,left=$\Pi_{{\mathrm{2}}} :$] {
      \inferrule* [flushleft,right=,left=$\Pi_{{\mathrm{3}}} :$] {
        \pi_3
      }{\delta_{{\mathrm{2}}}  \odot  \Delta_{{\mathrm{2}}}  \mGLsym{;}  \Gamma_{{\mathrm{4}}}  \vdash_{\mathsf{MS} }  \mGLnt{B}}\\
      \inferrule* [flushleft,right=,left=$\Pi_{{\mathrm{4}}} :$] {
        \pi_4
      }{\delta_{{\mathrm{1}}}  \odot  \Delta_{{\mathrm{1}}}  \mGLsym{;}   ( \Gamma_{{\mathrm{1}}}  \mGLsym{,}  \mGLnt{B}  \mGLsym{,}  \Gamma_{{\mathrm{3}}}  \mGLsym{,}  \mGLnt{A}  \mGLsym{,}  \Gamma_{{\mathrm{5}}} )   \vdash_{\mathsf{MS} }  \mGLnt{C}}
    }{\delta_{{\mathrm{1}}}  \mGLsym{,}  \delta_{{\mathrm{2}}}  \odot  \Delta_{{\mathrm{1}}}  \mGLsym{,}  \Delta_{{\mathrm{2}}}  \mGLsym{;}   ( \Gamma_{{\mathrm{1}}}  \mGLsym{,}  \mGLnt{B}  \mGLsym{,}  \Gamma_{{\mathrm{3}}}  \mGLsym{,}  \Gamma_{{\mathrm{4}}}  \mGLsym{,}  \Gamma_{{\mathrm{5}}} )   \vdash_{\mathsf{MS} }  \mGLnt{C}}
  \]
This is similar to the previous case.

    \end{enumerate}

  \item \textbf{eta-Expansion}
  \begin{enumerate}
    \item \textbf{Tensor:}
    \[
      \begin{array}{lll}
        \inferrule* [flushleft,right=$\mGLdruleMSTXXidName{}$,left=$\Pi_{{\mathrm{1}}} :$] {
          \,
        }{\emptyset  \odot  \emptyset  \mGLsym{;}  \mGLnt{A}  \otimes  \mGLnt{B}  \vdash_{\mathsf{MS} }  \mGLnt{A}  \otimes  \mGLnt{B}}
        & \quad &
        \inferrule* [flushleft,right=$\mGLdruleMSTXXidName{}$,left=$\Pi_{{\mathrm{2}}} :$] {
          \,
        }{\emptyset  \odot  \emptyset  \mGLsym{;}  \mGLnt{A}  \otimes  \mGLnt{B}  \vdash_{\mathsf{MS} }  \mGLnt{A}  \otimes  \mGLnt{B}}
      \end{array}
      \]
      Since $\mathsf{Depth}  (  \Pi_{{\mathrm{1}}}  )   +   \mathsf{Depth}  (  \Pi_{{\mathrm{2}}}  )   \, \mGLsym{=} \, 0$ we find ourselves in a base
      case.   $\Pi_{{\mathrm{1}}}$ and $\Pi_{{\mathrm{2}}}$ are cut-free, so we must produce a
      cut-free proof, $\Pi$ with 
      which we do as follows:
      \[
      \inferrule* [flushleft,right=$\mGLdruleMSTXXTenLName{}$] {
        \inferrule* [flushleft,right=$\mGLdruleMSTXXTenRName{}$] {
          \inferrule* [flushleft,right=$\mGLdruleMSTXXidName{}$] {
            \,
          }{\emptyset  \odot  \emptyset  \mGLsym{;}  \mGLnt{A}  \vdash_{\mathsf{MS} }  \mGLnt{A}}\\
          \inferrule* [flushleft,right=$\mGLdruleMSTXXidName{}$] {
            \,
          }{\emptyset  \odot  \emptyset  \mGLsym{;}  \mGLnt{B}  \vdash_{\mathsf{MS} }  \mGLnt{B}}
        }{\emptyset  \odot  \emptyset  \mGLsym{;}  \mGLnt{A}  \mGLsym{,}  \mGLnt{B}  \vdash_{\mathsf{MS} }  \mGLnt{A}  \otimes  \mGLnt{B}}
      }{\emptyset  \odot  \emptyset  \mGLsym{;}  \mGLnt{A}  \otimes  \mGLnt{B}  \vdash_{\mathsf{MS} }  \mGLnt{A}  \otimes  \mGLnt{B}}
      \]

    \item \textbf{Tensor Unit:}
    \[
      \begin{array}{lll}
        \inferrule* [flushleft,right=$\mGLdruleMSTXXidName{}$,left=$\Pi_{{\mathrm{1}}} :$] {
          \,
        }{\emptyset  \odot  \emptyset  \mGLsym{;}  \mathsf{I}  \vdash_{\mathsf{MS} }  \mathsf{I}}
        & \quad &
        \inferrule* [flushleft,right=$\mGLdruleMSTXXidName{}$,left=$\Pi_{{\mathrm{2}}} :$] {
          \,
        }{\emptyset  \odot  \emptyset  \mGLsym{;}  \mathsf{I}  \vdash_{\mathsf{MS} }  \mathsf{I}}
      \end{array}
      \]
      Since $\mathsf{Depth}  (  \Pi_{{\mathrm{1}}}  )   +   \mathsf{Depth}  (  \Pi_{{\mathrm{2}}}  )   \, \mGLsym{=} \, 0$ we find ourselves in a base
      case. $\Pi_{{\mathrm{1}}}$ and $\Pi_{{\mathrm{2}}}$ are cut-free, so we must produce a
      cut-free proof, $\Pi$ with 
      which we do as follows:
      \[
      \inferrule* [flushleft,right=$\mGLdruleMSTXXUnitLName{}$] {
        \inferrule* [flushleft,right=$\mGLdruleMSTXXUnitRName{}$] {
          \,
        }{\emptyset  \odot  \emptyset  \mGLsym{;}  \emptyset  \vdash_{\mathsf{MS} }  \mathsf{I}}
      }{\emptyset  \odot  \emptyset  \mGLsym{;}  \mathsf{I}  \vdash_{\mathsf{MS} }  \mathsf{I}}
      \]
    \item \textbf{Grd:} 
    \[
      \begin{array}{lll}
        \inferrule* [flushleft,right=$\mGLdruleMSTXXidName{}$,left=$\Pi_{{\mathrm{1}}} :$] {
          \,
        }{\emptyset  \odot  \emptyset  \mGLsym{;}   \mathsf{Grd} _{ \mGLnt{r} }\, \mGLnt{X}   \vdash_{\mathsf{MS} }   \mathsf{Grd} _{ \mGLnt{r} }\, \mGLnt{X}}
        & \quad &
        \inferrule* [flushleft,right=$\mGLdruleMSTXXidName{}$,left=$\Pi_{{\mathrm{2}}} :$] {
          \,
        }{\emptyset  \odot  \emptyset  \mGLsym{;}   \mathsf{Grd} _{ \mGLnt{r} }\, \mGLnt{X}   \vdash_{\mathsf{MS} }   \mathsf{Grd} _{ \mGLnt{r} }\, \mGLnt{X}}
      \end{array}
      \]
      Since $\mathsf{Depth}  (  \Pi_{{\mathrm{1}}}  )   +   \mathsf{Depth}  (  \Pi_{{\mathrm{2}}}  )   \, \mGLsym{=} \, 0$ we find ourselves in a base
      case.$\Pi_{{\mathrm{1}}}$ and $\Pi_{{\mathrm{2}}}$ are cut-free, so we must produce a
      cut-free proof, $\Pi$ , which we do as follows:
      \[
      \inferrule* [flushleft,right=$\mGLdruleMSTXXGrdLName{}$] {
        \inferrule* [flushleft,right=$\mGLdruleMSTXXGrdRName{}$] {
          \inferrule* [flushleft,right=$\mGLdruleGSTXXidName{}$] {
          \,
          }{1  \odot  \mGLnt{X}  \vdash_{\mathsf{GS} }  \mGLnt{X}}
        }{\mGLnt{r}  *  1  \odot  \mGLnt{X}  \mGLsym{;}  \emptyset  \vdash_{\mathsf{MS} }   \mathsf{Grd} _{ \mGLnt{r} }\, \mGLnt{X}}
      }{\emptyset  \odot  \emptyset  \mGLsym{;}   \mathsf{Grd} _{ \mGLnt{r} }\, \mGLnt{X}   \vdash_{\mathsf{MS} }   \mathsf{Grd} _{ \mGLnt{r} }\, \mGLnt{X}}
      \]
  \end{enumerate}
 \item \textbf{Axiom Cases}
  \begin{enumerate}
   \item \textbf{Axiom on the left:}
   \[
          \begin{array}{lll}
            \inferrule* [flushleft,right=$\mGLdruleMSTXXidName{}$,left=$\Pi_{{\mathrm{1}}} :$] {
              \,
            }{\emptyset  \odot  \emptyset  \mGLsym{;}  \mGLnt{A}  \vdash_{\mathsf{MS} }  \mGLnt{A}}
            & \quad &
            \inferrule* [flushleft,right=,left=$\Pi_{{\mathrm{2}}} :$] {
              \pi_2
            }{\delta_{{\mathrm{1}}}  \odot  \Delta_{{\mathrm{1}}}  \mGLsym{;}  \Gamma_{{\mathrm{1}}}  \mGLsym{,}  \mGLnt{A}  \mGLsym{,}  \Gamma_{{\mathrm{2}}}  \vdash_{\mathsf{MS} }  \mGLnt{B}}
          \end{array}
          \]
      We know:
       $\Pi_{{\mathrm{2}}}$ is a proof of $\delta_{{\mathrm{1}}}  \mGLsym{,}  \emptyset  \odot  \Delta_{{\mathrm{1}}}  \mGLsym{,}  \emptyset  \mGLsym{;}  \Gamma_{{\mathrm{1}}}  \mGLsym{,}  \mGLnt{A}  \mGLsym{,}  \Gamma_{{\mathrm{2}}}  \vdash_{\mathsf{MS} }  \mGLnt{B}$
      Thus, we construct the following proof $\Pi \, \mGLsym{=} \, \Pi_{{\mathrm{2}}}$.
   \item \textbf{Axiom on the right:}
   \[
    \begin{array}{lll}
     \inferrule* [flushleft,right=,left=$\Pi_{{\mathrm{1}}} :$] {
        \pi_1
      }{\delta_{{\mathrm{2}}}  \odot  \Delta_{{\mathrm{2}}}  \mGLsym{;}  \Gamma_{{\mathrm{2}}}  \vdash_{\mathsf{MS} }  \mGLnt{A}}
      & \quad &
       \inferrule* [flushleft,right=$\mGLdruleMSTXXidName{}$,left=$\Pi_{{\mathrm{2}}} :$] {
        \,
      }{\emptyset  \odot  \emptyset  \mGLsym{;}  \mGLnt{A}  \vdash_{\mathsf{MS} }  \mGLnt{A}}
    \end{array}
    \]
We know:
 $\Pi_{{\mathrm{1}}}$ is a proof of $\emptyset  \mGLsym{,}  \delta_{{\mathrm{2}}}  \odot  \emptyset  \mGLsym{,}  \Delta_{{\mathrm{2}}}  \mGLsym{;}  \Gamma_{{\mathrm{2}}}  \vdash_{\mathsf{MS} }  \mGLnt{A}$
Thus, we construct the following proof $\Pi \, \mGLsym{=} \, \Pi_{{\mathrm{1}}}$.
  \end{enumerate}

 \item \textbf{Principle Formula vs Principle Formula}
  \begin{enumerate}
   \item \textbf{Linear Tensor:}
   \[
    \begin{array}{lll}
      \inferrule* [flushleft,right=$\mGLdruleMSTXXTenRName{}$,left=$\Pi_{{\mathrm{1}}} :$] {
        \inferrule* [flushleft,right=,left=$\Pi_{{\mathrm{3}}} :$] {
          \pi_3
        }{\delta_{{\mathrm{2}}}  \odot  \Delta_{{\mathrm{2}}}  \mGLsym{;}  \Gamma_{{\mathrm{2}}}  \vdash_{\mathsf{MS} }  \mGLnt{A}}\\
        \inferrule* [flushleft,right=,left=$\Pi_{{\mathrm{4}}} :$] {
          \pi_4
        }{\delta_{{\mathrm{3}}}  \odot  \Delta_{{\mathrm{3}}}  \mGLsym{;}  \Gamma_{{\mathrm{3}}}  \vdash_{\mathsf{MS} }  \mGLnt{B}}
      }{(  \delta_{{\mathrm{2}}}  \mGLsym{,}  \delta_{{\mathrm{3}}}  )   \odot   ( \Delta_{{\mathrm{2}}}  \mGLsym{,}  \Delta_{{\mathrm{3}}} )   \mGLsym{;}  \Gamma_{{\mathrm{2}}}  \mGLsym{,}  \Gamma_{{\mathrm{3}}}  \vdash_{\mathsf{MS} }  \mGLnt{A}  \otimes  \mGLnt{B}}
      & \quad &
      \inferrule* [flushleft,right=$\mGLdruleMSTXXTenLName{}$,left=$\Pi_{{\mathrm{2}}} :$] {
        \inferrule* [flushleft,right=,left=$\Pi_{{\mathrm{5}}} :$] {
          \pi_5,
        }{\delta_{{\mathrm{1}}}  \odot  \Delta_{{\mathrm{1}}}  \mGLsym{;}  \Gamma_{{\mathrm{1}}}  \mGLsym{,}  \mGLnt{A}  \mGLsym{,}  \mGLnt{B}  \mGLsym{,}  \Gamma_{{\mathrm{4}}}  \vdash_{\mathsf{MS} }  \mGLnt{C}}
    }{\delta_{{\mathrm{1}}}  \odot  \Delta_{{\mathrm{1}}}  \mGLsym{;}  \Gamma_{{\mathrm{1}}}  \mGLsym{,}  \mGLnt{A}  \otimes  \mGLnt{B}  \mGLsym{,}  \Gamma_{{\mathrm{4}}}  \vdash_{\mathsf{MS} }  \mGLnt{C}}
    \end{array}
    \]
We know:
\[
\begin{array}{lll}
\mathsf{CutRank} \, \mGLsym{(}  \Pi_{{\mathrm{1}}}  \mGLsym{)} \, \mGLsym{=} \, \mGLkw{Max} \, \mGLsym{(}  \mathsf{CutRank} \, \mGLsym{(}  \Pi_{{\mathrm{3}}}  \mGLsym{)}  \mGLsym{,}  \mathsf{CutRank} \, \mGLsym{(}  \Pi_{{\mathrm{4}}}  \mGLsym{)}  \mGLsym{)} \, \leq \,  \mathsf{Rank}  (  \mGLnt{A}  \otimes  \mGLnt{B}  )\\
\mathsf{CutRank} \, \mGLsym{(}  \Pi_{{\mathrm{3}}}  \mGLsym{)} \, \leq \,  \mathsf{Rank}  (  \mGLnt{A}  \otimes  \mGLnt{B}  )\\
\mathsf{CutRank} \, \mGLsym{(}  \Pi_{{\mathrm{4}}}  \mGLsym{)} \, \leq \,  \mathsf{Rank}  (  \mGLnt{A}  \otimes  \mGLnt{B}  )\\
\mathsf{CutRank} \, \mGLsym{(}  \Pi_{{\mathrm{2}}}  \mGLsym{)} \, \mGLsym{=} \, \mathsf{CutRank} \, \mGLsym{(}  \Pi_{{\mathrm{5}}}  \mGLsym{)} \, \leq \,  \mathsf{Rank}  (  \mGLnt{A}  \otimes  \mGLnt{B}  )\\
\mathsf{Rank}  (  \mGLnt{A}  )  \, \mGLsym{<} \,  \mathsf{Rank}  (  \mGLnt{A}  \otimes  \mGLnt{B}  )\\
\mathsf{Rank}  (  \mGLnt{B}  )  \, \mGLsym{<} \,  \mathsf{Rank}  (  \mGLnt{A}  \otimes  \mGLnt{B}  )\\
\end{array}
\]
Instead of applying the induction hypothesis,
we can directly build the proof $\Pi$:
\begin{gather*}
\inferrule* [flushleft,right=$\mGLdruleMSTXXCutName{}$,left=$\Pi :$] {
  \inferrule* [flushleft,left=$\Pi_{{\mathrm{4}}} : $] {
      \pi_4
    }{\delta_{{\mathrm{3}}}  \odot  \Delta_{{\mathrm{3}}}  \mGLsym{;}  \Gamma_{{\mathrm{3}}}  \vdash_{\mathsf{MS} }  \mGLnt{B}}\\
    \inferrule* [flushleft,right=$\mGLdruleMSTXXCutName{}$] {
      \inferrule* [flushleft,left=$\Pi_{{\mathrm{3}}} : $] {
      \pi_3
    }{\delta_{{\mathrm{2}}}  \odot  \Delta_{{\mathrm{2}}}  \mGLsym{;}  \Gamma_{{\mathrm{2}}}  \vdash_{\mathsf{MS} }  \mGLnt{A}}\\
    \inferrule* [flushleft,right=,left=$\Pi_{{\mathrm{5}}} :$] {
      \pi_5
    }{\delta_{{\mathrm{1}}}  \odot  \Delta_{{\mathrm{1}}}  \mGLsym{;}  \Gamma_{{\mathrm{1}}}  \mGLsym{,}  \mGLnt{A}  \mGLsym{,}  \mGLnt{B}  \mGLsym{,}  \Gamma_{{\mathrm{4}}}  \vdash_{\mathsf{MS} }  \mGLnt{C}}
    }{\delta_{{\mathrm{1}}}  \mGLsym{,}  \delta_{{\mathrm{2}}}  \odot  \Delta_{{\mathrm{1}}}  \mGLsym{,}  \Delta_{{\mathrm{2}}}  \mGLsym{;}  \Gamma_{{\mathrm{1}}}  \mGLsym{,}  \Gamma_{{\mathrm{2}}}  \mGLsym{,}  \mGLnt{B}  \mGLsym{,}  \Gamma_{{\mathrm{4}}}  \vdash_{\mathsf{MS} }  \mGLnt{C}}
}{\delta_{{\mathrm{1}}}  \mGLsym{,}  \delta_{{\mathrm{2}}}  \mGLsym{,}  \delta_{{\mathrm{3}}}  \odot  \Delta_{{\mathrm{1}}}  \mGLsym{,}  \Delta_{{\mathrm{2}}}  \mGLsym{,}  \Delta_{{\mathrm{3}}}  \mGLsym{;}  \Gamma_{{\mathrm{1}}}  \mGLsym{,}  \Gamma_{{\mathrm{2}}}  \mGLsym{,}  \Gamma_{{\mathrm{3}}}  \mGLsym{,}  \Gamma_{{\mathrm{4}}}  \vdash_{\mathsf{MS} }  \mGLnt{C}}
\end{gather*}
So $\mathsf{CutRank} \, \mGLsym{(}  \Pi  \mGLsym{)} \, \mGLsym{=} \, \mGLkw{Max} \, \mGLsym{(}    \mathsf{CutRank} \, \mGLsym{(}  \Pi_{{\mathrm{3}}}  \mGLsym{)}  \mGLsym{,}  \mathsf{CutRank} \, \mGLsym{(}  \Pi_{{\mathrm{4}}}  \mGLsym{)}  \mGLsym{,}  \mathsf{CutRank} \, \mGLsym{(}  \Pi_{{\mathrm{5}}}  \mGLsym{)}  \mGLsym{,}   \mathsf{Rank}  (  \mGLnt{X}  )   + 1   \mGLsym{,}   \mathsf{Rank}  (  \mGLnt{Y}  )   + 1   \mGLsym{)} \, \leq \,  \mathsf{Rank}  (  \mGLnt{X}  \boxtimes  \mGLnt{Y}  )$
   \item \textbf{Tensor Unit:}
   \[
    \inferrule* [flushleft,right=$\mGLdruleMSTXXUnitRName{}$,left=$\Pi_{{\mathrm{1}}} :$] {
       \,
    }{\emptyset  \odot  \emptyset  \mGLsym{;}  \emptyset  \vdash_{\mathsf{MS} }  \mathsf{I}}
\]

\[
  \inferrule* [flushleft,right=$\mGLdruleMSTXXUnitLName{}$,left=$\Pi_{{\mathrm{2}}} :$] {
    \inferrule* [flushleft,left=$\Pi_{{\mathrm{3}}} : $] {
      \pi_3
    }{\delta_{{\mathrm{1}}}  \odot  \Delta_{{\mathrm{1}}}  \mGLsym{;}  \Gamma_{{\mathrm{1}}}  \mGLsym{,}  \Gamma_{{\mathrm{2}}}  \vdash_{\mathsf{MS} }  \mGLnt{A}}
  }{\delta_{{\mathrm{1}}}  \odot  \Delta_{{\mathrm{1}}}  \mGLsym{;}  \Gamma_{{\mathrm{1}}}  \mGLsym{,}  \mathsf{I}  \mGLsym{,}  \Gamma_{{\mathrm{2}}}  \vdash_{\mathsf{MS} }  \mGLnt{A}}
  \]

We know:
\[
\begin{array}{lll}
  \mathsf{CutRank} \, \mGLsym{(}  \Pi_{{\mathrm{3}}}  \mGLsym{)} \, \leq \, \mathsf{CutRank} \, \mGLsym{(}  \Pi_{{\mathrm{2}}}  \mGLsym{)} \, \leq \,  \mathsf{Rank}  (  \mathsf{I}  )\\
\end{array}
\]
So $\Pi_{{\mathrm{3}}}$ is a proof of
$\delta_{{\mathrm{1}}}  \mGLsym{,}  \emptyset  \odot  \Delta_{{\mathrm{1}}}  \mGLsym{,}  \emptyset  \mGLsym{;}  \Gamma_{{\mathrm{1}}}  \mGLsym{,}  \emptyset  \mGLsym{,}  \Gamma_{{\mathrm{2}}}  \vdash_{\mathsf{MS} }  \mGLnt{A}$ with
$\mathsf{CutRank} \, \mGLsym{(}  \Pi_{{\mathrm{3}}}  \mGLsym{)} \, \leq \,  \mathsf{Rank}  (  \mathsf{I}  )$.
Thus, we construct the following proof $\Pi \, \mGLsym{=} \, \Pi_{{\mathrm{3}}}$
   \item \textbf{Grd:}
   \[
    \inferrule* [flushleft,right=$\mGLdruleMSTXXGrdRName{}$, left=$\Pi_{{\mathrm{1}}} :$] {
      \inferrule* [flushleft,right=, left=$\Pi_{{\mathrm{3}}} :$] {
        \pi_3
      }{\delta_{{\mathrm{2}}}  \odot  \Delta_{{\mathrm{2}}}  \vdash_{\mathsf{GS} }  \mGLnt{X}}
    }{\mGLnt{r}  *  \delta_{{\mathrm{2}}}  \odot  \Delta_{{\mathrm{2}}}  \mGLsym{;}  \emptyset  \vdash_{\mathsf{MS} }   \mathsf{Grd} _{ \mGLnt{r} }\, \mGLnt{X}}
    \]
    \[
      \inferrule* [flushleft,right=$\mGLdruleMSTXXGrdLName{}$, left=$\Pi_{{\mathrm{2}}} :$] {
        \inferrule* [flushleft,right=, left=$\Pi_{{\mathrm{4}}} :$] {
          \pi_4
        }{(  \delta_{{\mathrm{1}}}  \mGLsym{,}  \mGLnt{r}  )   \odot   ( \Delta_{{\mathrm{1}}}  \mGLsym{,}  \mGLnt{X} )   \mGLsym{;}  \Gamma  \vdash_{\mathsf{MS} }  \mGLnt{A}}
      }{\delta_{{\mathrm{1}}}  \odot  \Delta_{{\mathrm{1}}}  \mGLsym{;}   (  \mathsf{Grd} _{ \mGLnt{r} }\, \mGLnt{X}   \mGLsym{,}  \Gamma )   \vdash_{\mathsf{MS} }  \mGLnt{A}}
      \]
      We know the following:
      \[
        \begin{array}{lll}
          \mathsf{Depth}  (  \Pi_{{\mathrm{1}}}  )   +   \mathsf{Depth}  (  \Pi_{{\mathrm{2}}}  ) = \mGLsym{(}    \mathsf{Depth}  (  \Pi_{{\mathrm{3}}}  )   + 1   \mGLsym{)}  +  \mGLsym{(}    \mathsf{Depth}  (  \Pi_{{\mathrm{4}}}  )   + 1   \mGLsym{)}\\
          \mathsf{CutRank} \, \mGLsym{(}  \Pi_{{\mathrm{1}}}  \mGLsym{)} \, \mGLsym{=} \, \mathsf{CutRank} \, \mGLsym{(}  \Pi_{{\mathrm{3}}}  \mGLsym{)} \, \leq \,  \mathsf{Rank}  (   \mathsf{Grd} _{ \mGLnt{r} }\, \mGLnt{X}   )\\
          \mathsf{CutRank} \, \mGLsym{(}  \Pi_{{\mathrm{2}}}  \mGLsym{)} \, \mGLsym{=} \, \mathsf{CutRank} \, \mGLsym{(}  \Pi_{{\mathrm{4}}}  \mGLsym{)} \, \leq \,  \mathsf{Rank}  (   \mathsf{Grd} _{ \mGLnt{r} }\, \mGLnt{X}   )
        \end{array}
      \]
      These imply that:
      \[
        \begin{array}{lll}
          \mathsf{Depth}  (  \Pi_{{\mathrm{3}}}  )   +   \mathsf{Depth}  (  \Pi_{{\mathrm{4}}}  )   \, \mGLsym{<} \,  \mathsf{Depth}  (  \Pi_{{\mathrm{1}}}  )   +   \mathsf{Depth}  (  \Pi_{{\mathrm{2}}}  )\\
          \mathsf{CutRank} \, \mGLsym{(}  \Pi_{{\mathrm{3}}}  \mGLsym{)} \, \leq \,  \mathsf{Rank}  (   \mathsf{Grd} _{ \mGLnt{r} }\, \mGLnt{X}   )
          \mathsf{CutRank} \, \mGLsym{(}  \Pi_{{\mathrm{4}}}  \mGLsym{)} \, \leq \,  \mathsf{Rank}  (   \mathsf{Grd} _{ \mGLnt{r} }\, \mGLnt{X}   )
        \end{array}
      \]
      Thus, we apply the induction hypothesis of Lemma~\ref{lemma:cut_reduction_for_mgl} (2)
      to $\Pi_{{\mathrm{1}}}$ and $\Pi_{{\mathrm{3}}}$ to obtain a proof $\Pi'$ of the sequent
      $(  \delta_{{\mathrm{1}}}  \mGLsym{,}  \mGLnt{r}  *  \delta_{{\mathrm{2}}}  )   \odot   ( \Delta_{{\mathrm{1}}}  \mGLsym{,}  \Delta_{{\mathrm{2}}} )   \mGLsym{;}  \Gamma  \vdash_{\mathsf{MS} }  \mGLnt{A}$ with $\mathsf{CutRank} \, \mGLsym{(}  \Pi'  \mGLsym{)} \, \leq \,  \mathsf{Rank}  (   \mathsf{Grd} _{ \mGLnt{r} }\, \mGLnt{X}   )$.
      Which is exactly what we want.

      \item \textbf{Linear Implication:}
      \[
      \inferrule* [flushleft,right=$\mGLdruleMSTXXImpRName{}$,left=$\Pi_{{\mathrm{1}}} :$] {
          \inferrule* [flushleft,left=$\Pi_{{\mathrm{3}}} : $] {
            \pi_3
          }{\delta_{{\mathrm{3}}}  \odot  \Delta_{{\mathrm{3}}}  \mGLsym{;}  \Gamma_{{\mathrm{2}}}  \mGLsym{,}  \mGLnt{A}  \vdash_{\mathsf{MS} }  \mGLnt{B}}\\
        }{\delta_{{\mathrm{3}}}  \odot  \Delta_{{\mathrm{3}}}  \mGLsym{;}  \Gamma_{{\mathrm{2}}}  \vdash_{\mathsf{MS} }  \mGLnt{A}  \multimap  \mGLnt{B}}
        \]
        \[
      \inferrule* [flushleft,right=$\mGLdruleMSTXXImpLName{}$,left=$\Pi_{{\mathrm{2}}} :$] {
         \inferrule* [flushleft,left=$\Pi_{{\mathrm{4}}} : $] {
          \pi_4
        }{\delta_{{\mathrm{1}}}  \odot  \Delta_{{\mathrm{1}}}  \mGLsym{;}  \Gamma_{{\mathrm{3}}}  \vdash_{\mathsf{MS} }  \mGLnt{A} }\\
        \inferrule* [flushleft,right=$\mGLdruleMSTXXImpLName{}$,left=$\Pi_{{\mathrm{5}}} :$] {
         \pi_5
         }{ \delta_{{\mathrm{2}}}  \odot  \Delta_{{\mathrm{2}}}  \mGLsym{;}  \Gamma_{{\mathrm{1}}}  \mGLsym{,}  \mGLnt{B}  \mGLsym{,}  \Gamma_{{\mathrm{4}}}  \vdash_{\mathsf{MS} }  \mGLnt{C}}
          }{\delta_{{\mathrm{1}}}  \mGLsym{,}  \delta_{{\mathrm{2}}}  \odot  \Delta_{{\mathrm{1}}}  \mGLsym{,}  \Delta_{{\mathrm{2}}}  \mGLsym{;}  \Gamma_{{\mathrm{1}}}  \mGLsym{,}  \mGLnt{A}  \multimap  \mGLnt{B}  \mGLsym{,}  \Gamma_{{\mathrm{3}}}  \mGLsym{,}  \Gamma_{{\mathrm{4}}}  \vdash_{\mathsf{MS} }  \mGLnt{C} }
         \]
         We know:
    \[
    \begin{array}{lll}
      \mathsf{CutRank} \, \mGLsym{(}  \Pi_{{\mathrm{1}}}  \mGLsym{)} \, \mGLsym{=} \, \mGLkw{Max} \, \mGLsym{(}  \mathsf{CutRank} \, \mGLsym{(}  \Pi_{{\mathrm{3}}}  \mGLsym{)}  \mGLsym{,}  \mathsf{CutRank} \, \mGLsym{(}  \Pi_{{\mathrm{4}}}  \mGLsym{)}  \mGLsym{)} \, \leq \,  \mathsf{Rank}  (  \mGLnt{A}  \multimap  \mGLnt{B}  )\\
      \mathsf{CutRank} \, \mGLsym{(}  \Pi_{{\mathrm{3}}}  \mGLsym{)} \, \leq \,  \mathsf{Rank}  (  \mGLnt{A}  \multimap  \mGLnt{B}  )\\
      \mathsf{CutRank} \, \mGLsym{(}  \Pi_{{\mathrm{4}}}  \mGLsym{)} \, \leq \,  \mathsf{Rank}  (  \mGLnt{A}  \multimap  \mGLnt{B}  )\\
      \mathsf{CutRank} \, \mGLsym{(}  \Pi_{{\mathrm{2}}}  \mGLsym{)} \, \mGLsym{=} \, \mathsf{CutRank} \, \mGLsym{(}  \Pi_{{\mathrm{5}}}  \mGLsym{)} \, \leq \,  \mathsf{Rank}  (  \mGLnt{A}  \multimap  \mGLnt{B}  )\\
      \mathsf{Rank}  (  \mGLnt{A}  )  \, \mGLsym{<} \,  \mathsf{Rank}  (  \mGLnt{A}  \multimap  \mGLnt{B}  )\\
      \mathsf{Rank}  (  \mGLnt{B}  )  \, \mGLsym{<} \,  \mathsf{Rank}  (  \mGLnt{A}  \multimap  \mGLnt{B}  )\\
    \end{array}
    \]
    Instead of applying the induction hypothesis,
    we can directly build the proof $\Pi$:
    \begin{gather*}
    \inferrule* [flushleft,right=$\mGLdruleMSTXXGExName{}$] {
    \inferrule* [flushleft,right=$\mGLdruleMSTXXGExName{}$] {
      \inferrule* [flushleft,right=$\mGLdruleMSTXXExName{}$] {
      \inferrule* [flushleft,right=$\mGLdruleMSTXXCutName{}$] {
        \inferrule* [flushleft,right=$\mGLdruleMSTXXCutName{}$] {
          \inferrule* [flushleft,left=$\Pi_{{\mathrm{4}}} : $] {
           \pi_4
          }{\delta_{{\mathrm{1}}}  \odot  \Delta_{{\mathrm{1}}}  \mGLsym{;}  \Gamma_{{\mathrm{2}}}  \vdash_{\mathsf{MS} }  \mGLnt{A} }\\
            \inferrule* [flushleft,left=$\Pi_{{\mathrm{3}}} : $] {
            \pi_3
          }{\delta_{{\mathrm{3}}}  \odot  \Delta_{{\mathrm{3}}}  \mGLsym{;}  \Gamma_{{\mathrm{3}}}  \mGLsym{,}  \mGLnt{A}  \vdash_{\mathsf{MS} }  \mGLnt{B}}
          }{\delta_{{\mathrm{3}}}  \mGLsym{,}  \delta_{{\mathrm{1}}}  \odot  \Delta_{{\mathrm{3}}}  \mGLsym{,}  \Delta_{{\mathrm{1}}}  \mGLsym{;}  \Gamma_{{\mathrm{3}}}  \mGLsym{,}  \Gamma_{{\mathrm{2}}}  \vdash_{\mathsf{MS} }  \mGLnt{B} } \\
          \inferrule* [flushleft,left=$\Pi_{{\mathrm{4}}} : $] {
            \pi_5
          }{\delta_{{\mathrm{2}}}  \odot  \Delta_{{\mathrm{2}}}  \mGLsym{;}  \Gamma_{{\mathrm{1}}}  \mGLsym{,}  \mGLnt{B}  \mGLsym{,}  \Gamma_{{\mathrm{4}}}  \vdash_{\mathsf{MS} }  \mGLnt{C}}
      }{\delta_{{\mathrm{2}}}  \mGLsym{,}  \delta_{{\mathrm{3}}}  \mGLsym{,}  \delta_{{\mathrm{1}}}  \odot  \Delta_{{\mathrm{2}}}  \mGLsym{,}  \Delta_{{\mathrm{3}}}  \mGLsym{,}  \Delta_{{\mathrm{1}}}  \mGLsym{;}  \Gamma_{{\mathrm{1}}}  \mGLsym{,}  \Gamma_{{\mathrm{3}}}  \mGLsym{,}  \Gamma_{{\mathrm{2}}}  \mGLsym{,}  \Gamma_{{\mathrm{4}}}  \vdash_{\mathsf{MS} }  \mGLnt{C}}
    }{\delta_{{\mathrm{2}}}  \mGLsym{,}  \delta_{{\mathrm{3}}}  \mGLsym{,}  \delta_{{\mathrm{1}}}  \odot  \Delta_{{\mathrm{2}}}  \mGLsym{,}  \Delta_{{\mathrm{3}}}  \mGLsym{,}  \Delta_{{\mathrm{1}}}  \mGLsym{;}  \Gamma_{{\mathrm{1}}}  \mGLsym{,}  \Gamma_{{\mathrm{2}}}  \mGLsym{,}  \Gamma_{{\mathrm{3}}}  \mGLsym{,}  \Gamma_{{\mathrm{4}}}  \vdash_{\mathsf{MS} }  \mGLnt{C}}
    }{\delta_{{\mathrm{2}}}  \mGLsym{,}  \delta_{{\mathrm{1}}}  \mGLsym{,}  \delta_{{\mathrm{3}}}  \odot  \Delta_{{\mathrm{2}}}  \mGLsym{,}  \Delta_{{\mathrm{1}}}  \mGLsym{,}  \Delta_{{\mathrm{3}}}  \mGLsym{;}  \Gamma_{{\mathrm{1}}}  \mGLsym{,}  \Gamma_{{\mathrm{2}}}  \mGLsym{,}  \Gamma_{{\mathrm{3}}}  \mGLsym{,}  \Gamma_{{\mathrm{4}}}  \vdash_{\mathsf{MS} }  \mGLnt{C}}
    }{\delta_{{\mathrm{1}}}  \mGLsym{,}  \delta_{{\mathrm{2}}}  \mGLsym{,}  \delta_{{\mathrm{3}}}  \odot  \Delta_{{\mathrm{1}}}  \mGLsym{,}  \Delta_{{\mathrm{2}}}  \mGLsym{,}  \Delta_{{\mathrm{3}}}  \mGLsym{;}  \Gamma_{{\mathrm{1}}}  \mGLsym{,}  \Gamma_{{\mathrm{2}}}  \mGLsym{,}  \Gamma_{{\mathrm{3}}}  \mGLsym{,}  \Gamma_{{\mathrm{4}}}  \vdash_{\mathsf{MS} }  \mGLnt{C}}
      \end{gather*}
      So $\mathsf{CutRank} \, \mGLsym{(}  \Pi  \mGLsym{)} \, \mGLsym{=} \, \mGLkw{Max} \, \mGLsym{(}    \mathsf{CutRank} \, \mGLsym{(}  \Pi_{{\mathrm{3}}}  \mGLsym{)}  \mGLsym{,}  \mathsf{CutRank} \, \mGLsym{(}  \Pi_{{\mathrm{4}}}  \mGLsym{)}  \mGLsym{,}  \mathsf{CutRank} \, \mGLsym{(}  \Pi_{{\mathrm{5}}}  \mGLsym{)}  \mGLsym{,}   \mathsf{Rank}  (  \mGLnt{A}  )   + 1   \mGLsym{,}   \mathsf{Rank}  (  \mGLnt{B}  )   + 1   \mGLsym{)} \, \leq \,  \mathsf{Rank}  (  \mGLnt{A}  \otimes  \mGLnt{B}  )$

    \end{enumerate}

 \item \textbf{Secondary Conclusion}
  \begin{enumerate}
   \item \textbf{Left Introduction of Linear Tensor:}
   \[
    \inferrule* [flushleft,right=$\mGLdruleMSTXXTenLName{}$,left=$\Pi_{{\mathrm{1}}} :$] {
      \inferrule* [flushleft,right=,left=$\Pi_{{\mathrm{3}}} :$] {
        \pi_3,
      }{\delta_{{\mathrm{2}}}  \odot  \Delta_{{\mathrm{2}}}  \mGLsym{;}  \Gamma_{{\mathrm{2}}}  \mGLsym{,}  \mGLnt{A}  \mGLsym{,}  \mGLnt{B}  \mGLsym{,}  \Gamma_{{\mathrm{3}}}  \vdash_{\mathsf{MS} }  \mGLnt{C}}
  }{\delta_{{\mathrm{2}}}  \odot  \Delta_{{\mathrm{2}}}  \mGLsym{;}  \Gamma_{{\mathrm{2}}}  \mGLsym{,}  \mGLnt{A}  \otimes  \mGLnt{B}  \mGLsym{,}  \Gamma_{{\mathrm{3}}}  \vdash_{\mathsf{MS} }  \mGLnt{C}}
  \]

   \[
    \inferrule* [flushleft,right=,left=$\Pi_{{\mathrm{2}}} :$] {
      \pi_2
    }{\delta_{{\mathrm{1}}}  \odot  \Delta_{{\mathrm{2}}}  \mGLsym{;}  \Gamma_{{\mathrm{1}}}  \mGLsym{,}  \mGLnt{C}  \mGLsym{,}  \Gamma_{{\mathrm{4}}}  \vdash_{\mathsf{MS} }  \mGLnt{D}}
\]
We know:
\[
\begin{array}{lll}
  \mathsf{Depth}  (  \Pi_{{\mathrm{3}}}  )   +   \mathsf{Depth}  (  \Pi_{{\mathrm{2}}}  )   \, \mGLsym{<} \,  \mathsf{Depth}  (  \Pi_{{\mathrm{1}}}  )   +   \mathsf{Depth}  (  \Pi_{{\mathrm{2}}}  )\\
  \mathsf{CutRank} \, \mGLsym{(}  \Pi_{{\mathrm{3}}}  \mGLsym{)} \, \leq \, \mathsf{CutRank} \, \mGLsym{(}  \Pi_{{\mathrm{1}}}  \mGLsym{)} \, \leq \,  \mathsf{Rank}  (  \mGLnt{C}  )
\end{array}
\]

and so applying the induction hypothesis
to $\Pi_{{\mathrm{3}}}$ and $\Pi_{{\mathrm{2}}}$
implies that there is a proof $\Pi'$ of
$\delta_{{\mathrm{1}}}  \mGLsym{,}  \delta_{{\mathrm{2}}}  \odot  \Delta_{{\mathrm{1}}}  \mGLsym{,}  \Delta_{{\mathrm{2}}}  \mGLsym{;}  \Gamma_{{\mathrm{1}}}  \mGLsym{,}  \Gamma_{{\mathrm{2}}}  \mGLsym{,}  \mGLnt{A}  \mGLsym{,}  \mGLnt{B}  \mGLsym{,}  \Gamma_{{\mathrm{3}}}  \mGLsym{,}  \Gamma_{{\mathrm{4}}}  \vdash_{\mathsf{MS} }  \mGLnt{D}$ with
$\mathsf{CutRank} \, \mGLsym{(}  \Pi'  \mGLsym{)} \, \leq \,  \mathsf{Rank}  (  \mGLnt{C}  )$ 
Thus, we construct the following proof $\Pi$:
\[
  \inferrule* [flushleft,right=$\mGLdruleMSTXXTenLName{}$,left=$\Pi :$] {
    \inferrule* [flushleft,right=,left=$\Pi' :$] {
      \pi'
    }{\delta_{{\mathrm{1}}}  \mGLsym{,}  \delta_{{\mathrm{2}}}  \odot  \Delta_{{\mathrm{1}}}  \mGLsym{,}  \Delta_{{\mathrm{2}}}  \mGLsym{;}  \Gamma_{{\mathrm{1}}}  \mGLsym{,}  \Gamma_{{\mathrm{2}}}  \mGLsym{,}  \mGLnt{A}  \mGLsym{,}  \mGLnt{B}  \mGLsym{,}  \Gamma_{{\mathrm{3}}}  \mGLsym{,}  \Gamma_{{\mathrm{4}}}  \vdash_{\mathsf{MS} }  \mGLnt{D}}
  }{\delta_{{\mathrm{1}}}  \mGLsym{,}  \delta_{{\mathrm{2}}}  \odot  \Delta_{{\mathrm{1}}}  \mGLsym{,}  \Delta_{{\mathrm{2}}}  \mGLsym{;}  \Gamma_{{\mathrm{1}}}  \mGLsym{,}  \Gamma_{{\mathrm{2}}}  \mGLsym{,}  \mGLnt{A}  \otimes  \mGLnt{B}  \mGLsym{,}  \Gamma_{{\mathrm{3}}}  \mGLsym{,}  \Gamma_{{\mathrm{4}}}  \vdash_{\mathsf{MS} }  \mGLnt{D}}
  \]
  Given the above, we know:
  \[
    \begin{array}{lll}
      \mathsf{CutRank} \, \mGLsym{(}  \Pi  \mGLsym{)} \, \mGLsym{=} \, \mathsf{CutRank} \, \mGLsym{(}  \Pi'  \mGLsym{)} \, \leq \,  \mathsf{Rank}  (  \mGLnt{C}  )\\
    \end{array}
    \]
   \item \textbf{Left Introduction of Linear Tensor Unit:}
   \[
    \inferrule* [flushleft,right=$\mGLdruleMSTXXUnitLName{}$,left=$\Pi_{{\mathrm{1}}} :$] {
        \inferrule* [flushleft,right=,left=$\Pi_{{\mathrm{3}}} :$] {
          \pi_3
        }{\delta_{{\mathrm{2}}}  \odot  \Delta_{{\mathrm{2}}}  \mGLsym{;}  \Gamma_{{\mathrm{2}}}  \mGLsym{,}  \Gamma_{{\mathrm{3}}}  \vdash_{\mathsf{MS} }  \mGLnt{A}}
    }{\delta_{{\mathrm{2}}}  \odot  \Delta_{{\mathrm{2}}}  \mGLsym{;}  \Gamma_{{\mathrm{2}}}  \mGLsym{,}  \mathsf{I}  \mGLsym{,}  \Gamma_{{\mathrm{3}}}  \vdash_{\mathsf{MS} }  \mGLnt{A}}
\]
\[
\inferrule* [flushleft,right=,left=$\Pi_{{\mathrm{2}}} :$] {
  \pi_2
}{\delta_{{\mathrm{1}}}  \odot  \Delta_{{\mathrm{1}}}  \mGLsym{;}  \Gamma_{{\mathrm{1}}}  \mGLsym{,}  \mGLnt{A}  \mGLsym{,}  \Gamma_{{\mathrm{4}}}  \vdash_{\mathsf{MS} }  \mGLnt{B}}
\]
We know:
\[
\begin{array}{lll}
  \mathsf{Depth}  (  \Pi_{{\mathrm{2}}}  )   +   \mathsf{Depth}  (  \Pi_{{\mathrm{3}}}  )   \, \mGLsym{<} \,  \mathsf{Depth}  (  \Pi_{{\mathrm{1}}}  )   +   \mathsf{Depth}  (  \Pi_{{\mathrm{2}}}  )\\
  \mathsf{CutRank} \, \mGLsym{(}  \Pi_{{\mathrm{3}}}  \mGLsym{)} \, \leq \, \mathsf{CutRank} \, \mGLsym{(}  \Pi_{{\mathrm{1}}}  \mGLsym{)} \, \leq \,  \mathsf{Rank}  (  \mGLnt{A}  )
\end{array}
\]

and so applying the induction hypothesis
to $\Pi_{{\mathrm{2}}}$ and $\Pi_{{\mathrm{3}}}$
implies that there is a proof $\Pi'$ of
$\delta_{{\mathrm{1}}}  \mGLsym{,}  \delta_{{\mathrm{2}}}  \odot  \Delta_{{\mathrm{1}}}  \mGLsym{,}  \Delta_{{\mathrm{2}}}  \mGLsym{;}  \Gamma_{{\mathrm{1}}}  \mGLsym{,}  \Gamma_{{\mathrm{2}}}  \mGLsym{,}  \Gamma_{{\mathrm{3}}}  \mGLsym{,}  \Gamma_{{\mathrm{4}}}  \vdash_{\mathsf{MS} }  \mGLnt{B}$ with
$\mathsf{CutRank} \, \mGLsym{(}  \Pi'  \mGLsym{)} \, \leq \,  \mathsf{Rank}  (  \mGLnt{X}  )$ 
Thus, we construct the following proof $\Pi$:

\[
  \inferrule* [flushleft,right=$\mGLdruleMSTXXUnitLName{}$,left=$\Pi :$] {
    \inferrule* [flushleft,right=,left=$\Pi' :$] {
      \pi'
    }{\delta_{{\mathrm{1}}}  \mGLsym{,}  \delta_{{\mathrm{2}}}  \odot  \Delta_{{\mathrm{1}}}  \mGLsym{,}  \Delta_{{\mathrm{2}}}  \mGLsym{;}  \Gamma_{{\mathrm{1}}}  \mGLsym{,}  \Gamma_{{\mathrm{2}}}  \mGLsym{,}  \Gamma_{{\mathrm{3}}}  \mGLsym{,}  \Gamma_{{\mathrm{4}}}  \vdash_{\mathsf{MS} }  \mGLnt{B}}
  }{\delta_{{\mathrm{1}}}  \mGLsym{,}  \delta_{{\mathrm{2}}}  \odot  \Delta_{{\mathrm{1}}}  \mGLsym{,}  \Delta_{{\mathrm{2}}}  \mGLsym{;}  \Gamma_{{\mathrm{1}}}  \mGLsym{,}  \Gamma_{{\mathrm{2}}}  \mGLsym{,}  \mathsf{I}  \mGLsym{,}  \Gamma_{{\mathrm{3}}}  \mGLsym{,}  \Gamma_{{\mathrm{4}}}  \vdash_{\mathsf{MS} }  \mGLnt{B}}
  \]
  We know $\mathsf{CutRank} \, \mGLsym{(}  \Pi  \mGLsym{)} \, \mGLsym{=} \, \mathsf{CutRank} \, \mGLsym{(}  \Pi'  \mGLsym{)} \, \leq \,  \mathsf{Rank}  (  \mGLnt{A}  )$

   \item \textbf{Left introduction of graded tensor product:}
   \[
          \inferrule* [flushleft,right=$\mGLdruleMSTXXGTenLName{}$,left=$\Pi_{{\mathrm{1}}} :$] {
            \inferrule* [flushleft,right=,left=$\Pi_{{\mathrm{3}}} :$] {
            \pi_3
          }{(  \delta_{{\mathrm{2}}}  \mGLsym{,}  \mGLnt{r}  \mGLsym{,}  \mGLnt{r}  \mGLsym{,}  \delta_{{\mathrm{3}}}  )   \odot   ( \Delta_{{\mathrm{2}}}  \mGLsym{,}  \mGLnt{X}  \mGLsym{,}  \mGLnt{Y}  \mGLsym{,}  \Delta_{{\mathrm{3}}} )   \mGLsym{;}  \Gamma_{{\mathrm{2}}}  \vdash_{\mathsf{MS} }  \mGLnt{A}}
          }{(  \delta_{{\mathrm{2}}}  \mGLsym{,}  \mGLnt{r}  \mGLsym{,}  \delta_{{\mathrm{3}}}  )   \odot   ( \Delta_{{\mathrm{2}}}  \mGLsym{,}  \mGLnt{X}  \boxtimes  \mGLnt{Y}  \mGLsym{,}  \Delta_{{\mathrm{3}}} )   \mGLsym{;}  \Gamma_{{\mathrm{2}}}  \vdash_{\mathsf{MS} }  \mGLnt{A}}
    \]
      \[
        \inferrule* [flushleft,right=,left=$\Pi_{{\mathrm{2}}} :$] {
          \pi_2
        }{\delta_{{\mathrm{1}}}  \odot  \Delta_{{\mathrm{1}}}  \mGLsym{;}  \Gamma_{{\mathrm{1}}}  \mGLsym{,}  \mGLnt{A}  \mGLsym{,}  \Gamma_{{\mathrm{3}}}  \vdash_{\mathsf{MS} }  \mGLnt{B} }
        \]
      We know:
      \[
      \begin{array}{lll}
        \mathsf{Depth}  (  \Pi_{{\mathrm{3}}}  )   +   \mathsf{Depth}  (  \Pi_{{\mathrm{2}}}  )   \, \mGLsym{<} \,  \mathsf{Depth}  (  \Pi_{{\mathrm{1}}}  )   +   \mathsf{Depth}  (  \Pi_{{\mathrm{2}}}  )\\
        \mathsf{CutRank} \, \mGLsym{(}  \Pi_{{\mathrm{3}}}  \mGLsym{)} \, \leq \, \mathsf{CutRank} \, \mGLsym{(}  \Pi_{{\mathrm{1}}}  \mGLsym{)} \, \leq \,  \mathsf{Rank}  (  \mGLnt{A}  )
      \end{array}
      \]

      and so applying the induction hypothesis
      to $\Pi_{{\mathrm{3}}}$ and $\Pi_{{\mathrm{2}}}$
      implies that there is a proof $\Pi'$ of
      $(  \delta_{{\mathrm{1}}}  \mGLsym{,}  \delta_{{\mathrm{2}}}  \mGLsym{,}  \mGLnt{r}  \mGLsym{,}  \mGLnt{r}  \mGLsym{,}  \delta_{{\mathrm{3}}}  )   \odot   ( \Delta_{{\mathrm{1}}}  \mGLsym{,}  \Delta_{{\mathrm{2}}}  \mGLsym{,}  \mGLnt{X}  \mGLsym{,}  \mGLnt{Y}  \mGLsym{,}  \Delta_{{\mathrm{3}}} )   \mGLsym{;}  \Gamma_{{\mathrm{1}}}  \mGLsym{,}  \Gamma_{{\mathrm{2}}}  \mGLsym{,}  \Gamma_{{\mathrm{3}}}  \vdash_{\mathsf{MS} }  \mGLnt{B}$ with
      $\mathsf{CutRank} \, \mGLsym{(}  \Pi'  \mGLsym{)} \, \leq \,  \mathsf{Rank}  (  \mGLnt{A}  )$ 
      Thus, we construct the following proof $\Pi$:
      \[
        \inferrule* [flushleft,right=$\mGLdruleMSTXXGTenLName{}$,left=$\Pi :$] {
          \inferrule* [flushleft,right=,left=$\Pi' :$] {
            \pi'
          }{(  \delta_{{\mathrm{1}}}  \mGLsym{,}  \delta_{{\mathrm{2}}}  \mGLsym{,}  \mGLnt{r}  \mGLsym{,}  \mGLnt{r}  \mGLsym{,}  \delta_{{\mathrm{3}}}  )   \odot   ( \Delta_{{\mathrm{1}}}  \mGLsym{,}  \Delta_{{\mathrm{2}}}  \mGLsym{,}  \mGLnt{X}  \mGLsym{,}  \mGLnt{Y}  \mGLsym{,}  \Delta_{{\mathrm{3}}} )   \mGLsym{;}  \Gamma_{{\mathrm{1}}}  \mGLsym{,}  \Gamma_{{\mathrm{2}}}  \mGLsym{,}  \Gamma_{{\mathrm{3}}}  \vdash_{\mathsf{MS} }  \mGLnt{B}}
        }{(  \delta_{{\mathrm{1}}}  \mGLsym{,}  \delta_{{\mathrm{2}}}  \mGLsym{,}  \mGLnt{r}  \mGLsym{,}  \delta_{{\mathrm{3}}}  )   \odot   ( \Delta_{{\mathrm{1}}}  \mGLsym{,}  \Delta_{{\mathrm{2}}}  \mGLsym{,}  \mGLnt{X}  \boxtimes  \mGLnt{Y}  \mGLsym{,}  \Delta_{{\mathrm{3}}} )   \mGLsym{;}  \Gamma_{{\mathrm{1}}}  \mGLsym{,}  \Gamma_{{\mathrm{2}}}  \mGLsym{,}  \Gamma_{{\mathrm{3}}}  \vdash_{\mathsf{MS} }  \mGLnt{B}}
        \]
        Given the above, we know $\mathsf{CutRank} \, \mGLsym{(}  \Pi  \mGLsym{)} \, \mGLsym{=} \, \mathsf{CutRank} \, \mGLsym{(}  \Pi'  \mGLsym{)} \, \leq \,  \mathsf{Rank}  (  \mGLnt{A}  )$

    \item \textbf{Left Introduction of Graded Tensor Unit:}
    \[
          \inferrule* [flushleft,right=$\mGLdruleMSTXXGUnitLName{}$,left=$\Pi_{{\mathrm{1}}} :$] {
              \inferrule* [flushleft,right=,left=$\Pi_{{\mathrm{3}}} :$] {
                \pi_3
              }{(  \delta_{{\mathrm{2}}}  \mGLsym{,}  \delta_{{\mathrm{3}}}  )   \odot   ( \Delta_{{\mathrm{2}}}  \mGLsym{,}  \Delta_{{\mathrm{3}}} )   \mGLsym{;}  \Gamma_{{\mathrm{2}}}  \vdash_{\mathsf{MS} }  \mGLnt{A}}
          }{(  \delta_{{\mathrm{2}}}  \mGLsym{,}  \mGLnt{r}  \mGLsym{,}  \delta_{{\mathrm{3}}}  )   \odot   ( \Delta_{{\mathrm{2}}}  \mGLsym{,}  \mathsf{J}  \mGLsym{,}  \Delta_{{\mathrm{3}}} )   \mGLsym{;}  \Gamma_{{\mathrm{2}}}  \vdash_{\mathsf{MS} }  \mGLnt{A}}
    \]
    \[
      \inferrule* [flushleft,right=,left=$\Pi_{{\mathrm{2}}} :$] {
        \pi_2
      }{\delta_{{\mathrm{1}}}  \odot  \Delta_{{\mathrm{1}}}  \mGLsym{;}  \Gamma_{{\mathrm{1}}}  \mGLsym{,}  \mGLnt{A}  \mGLsym{,}  \Gamma_{{\mathrm{3}}}  \vdash_{\mathsf{MS} }  \mGLnt{B}}
\]
      We know:
      \[
      \begin{array}{lll}
        \mathsf{Depth}  (  \Pi_{{\mathrm{2}}}  )   +   \mathsf{Depth}  (  \Pi_{{\mathrm{3}}}  )   \, \mGLsym{<} \,  \mathsf{Depth}  (  \Pi_{{\mathrm{1}}}  )   +   \mathsf{Depth}  (  \Pi_{{\mathrm{2}}}  )\\
        \mathsf{CutRank} \, \mGLsym{(}  \Pi_{{\mathrm{3}}}  \mGLsym{)} \, \leq \, \mathsf{CutRank} \, \mGLsym{(}  \Pi_{{\mathrm{1}}}  \mGLsym{)} \, \leq \,  \mathsf{Rank}  (  \mGLnt{A}  )
      \end{array}
      \]

      and so applying the induction hypothesis
      to $\Pi_{{\mathrm{2}}}$ and $\Pi_{{\mathrm{3}}}$
      implies that there is a proof $\Pi'$ of
      $(  \delta_{{\mathrm{1}}}  \mGLsym{,}  \delta_{{\mathrm{2}}}  \mGLsym{,}  \delta_{{\mathrm{3}}}  )   \odot   ( \Delta_{{\mathrm{1}}}  \mGLsym{,}  \Delta_{{\mathrm{2}}}  \mGLsym{,}  \Delta_{{\mathrm{3}}} )   \mGLsym{;}  \Gamma_{{\mathrm{1}}}  \mGLsym{,}  \Gamma_{{\mathrm{2}}}  \mGLsym{,}  \Gamma_{{\mathrm{3}}}  \vdash_{\mathsf{MS} }  \mGLnt{B}$ with
      $\mathsf{CutRank} \, \mGLsym{(}  \Pi'  \mGLsym{)} \, \leq \,  \mathsf{Rank}  (  \mGLnt{A}  )$ 
      Thus, we construct the following proof $\Pi$:

      \[
        \inferrule* [flushleft,right=$\mGLdruleMSTXXGUnitLName{}$,left=$\Pi :$] {
          \inferrule* [flushleft,right=,left=$\Pi' :$] {
            \pi'
          }{(  \delta_{{\mathrm{1}}}  \mGLsym{,}  \delta_{{\mathrm{2}}}  \mGLsym{,}  \delta_{{\mathrm{3}}}  )   \odot   ( \Delta_{{\mathrm{1}}}  \mGLsym{,}  \Delta_{{\mathrm{2}}}  \mGLsym{,}  \Delta_{{\mathrm{3}}} )   \mGLsym{;}  \Gamma_{{\mathrm{1}}}  \mGLsym{,}  \Gamma_{{\mathrm{2}}}  \mGLsym{,}  \Gamma_{{\mathrm{3}}}  \vdash_{\mathsf{MS} }  \mGLnt{B}}
        }{(  \delta_{{\mathrm{1}}}  \mGLsym{,}  \delta_{{\mathrm{2}}}  \mGLsym{,}  \mGLnt{r}  \mGLsym{,}  \delta_{{\mathrm{3}}}  )   \odot   ( \Delta_{{\mathrm{1}}}  \mGLsym{,}  \Delta_{{\mathrm{2}}}  \mGLsym{,}  \mathsf{J}  \mGLsym{,}  \Delta_{{\mathrm{3}}} )   \mGLsym{;}  \Gamma_{{\mathrm{1}}}  \mGLsym{,}  \Gamma_{{\mathrm{2}}}  \mGLsym{,}  \Gamma_{{\mathrm{3}}}  \vdash_{\mathsf{MS} }  \mGLnt{B}}
        \]
        Given the above, we know  $\mathsf{CutRank} \, \mGLsym{(}  \Pi  \mGLsym{)} \, \mGLsym{=} \, \mathsf{CutRank} \, \mGLsym{(}  \Pi'  \mGLsym{)} \, \leq \,  \mathsf{Rank}  (  \mGLnt{A}  )$
    \item \textbf{Left Introduction of Linear Implication:}
    \[
      \inferrule* [flushleft,right=$\mGLdruleMSTXXImpLName{}$,left=$\Pi_{{\mathrm{1}}} :$] {
         \inferrule* [flushleft,left=$\Pi_{{\mathrm{3}}} : $] {
          \pi_3
        }{\delta_{{\mathrm{2}}}  \odot  \Delta_{{\mathrm{2}}}  \mGLsym{;}  \Gamma_{{\mathrm{3}}}  \vdash_{\mathsf{MS} }  \mGLnt{A} }\\
        \inferrule* [flushleft,right=,left=$\Pi_{{\mathrm{4}}} :$] {
         \pi_4
         }{ \delta_{{\mathrm{3}}}  \odot  \Delta_{{\mathrm{3}}}  \mGLsym{;}  \Gamma_{{\mathrm{2}}}  \mGLsym{,}  \mGLnt{B}  \mGLsym{,}  \Gamma_{{\mathrm{4}}}  \vdash_{\mathsf{MS} }  \mGLnt{C}}
          }{\delta_{{\mathrm{2}}}  \mGLsym{,}  \delta_{{\mathrm{3}}}  \odot  \Delta_{{\mathrm{2}}}  \mGLsym{,}  \Delta_{{\mathrm{3}}}  \mGLsym{;}  \Gamma_{{\mathrm{2}}}  \mGLsym{,}  \mGLnt{A}  \multimap  \mGLnt{B}  \mGLsym{,}  \Gamma_{{\mathrm{3}}}  \mGLsym{,}  \Gamma_{{\mathrm{4}}}  \vdash_{\mathsf{MS} }  \mGLnt{C} }
         \]
         \[
          \inferrule* [flushleft,right=,left=$\Pi_{{\mathrm{2}}} :$] {
            \pi_2
          }{\delta_{{\mathrm{1}}}  \odot  \Delta_{{\mathrm{1}}}  \mGLsym{;}  \Gamma_{{\mathrm{1}}}  \mGLsym{,}  \mGLnt{C}  \mGLsym{,}  \Gamma_{{\mathrm{5}}}  \vdash_{\mathsf{MS} }  \mGLnt{D}}
    \]

We know:
\[
\begin{array}{lll}
  \mathsf{Depth}  (  \Pi_{{\mathrm{4}}}  )   +   \mathsf{Depth}  (  \Pi_{{\mathrm{2}}}  )   \, \mGLsym{<} \,  \mathsf{Depth}  (  \Pi_{{\mathrm{1}}}  )   +   \mathsf{Depth}  (  \Pi_{{\mathrm{2}}}  )\\
  \mathsf{CutRank} \, \mGLsym{(}  \Pi_{{\mathrm{4}}}  \mGLsym{)} \, \leq \, \mathsf{CutRank} \, \mGLsym{(}  \Pi_{{\mathrm{1}}}  \mGLsym{)} \, \leq \,  \mathsf{Rank}  (  \mGLnt{C}  )
\end{array}
\]

and so applying the induction hypothesis
to $\Pi_{{\mathrm{4}}}$ and $\Pi_{{\mathrm{2}}}$
implies that there is a proof $\Pi'$ of
$\delta_{{\mathrm{1}}}  \mGLsym{,}  \delta_{{\mathrm{3}}}  \odot  \Delta_{{\mathrm{1}}}  \mGLsym{,}  \Delta_{{\mathrm{3}}}  \mGLsym{;}  \Gamma_{{\mathrm{1}}}  \mGLsym{,}  \Gamma_{{\mathrm{2}}}  \mGLsym{,}  \mGLnt{B}  \mGLsym{,}  \Gamma_{{\mathrm{4}}}  \mGLsym{,}  \Gamma_{{\mathrm{5}}}  \vdash_{\mathsf{MS} }  \mGLnt{D}$ with
$\mathsf{CutRank} \, \mGLsym{(}  \Pi'  \mGLsym{)} \, \leq \,  \mathsf{Rank}  (  \mGLnt{C}  )$
Thus, we construct the following proof $\Pi$:

\[
  \inferrule* [flushleft,right=$\mGLdruleMSTXXExName{}$,left=$\Pi :$] {
  \inferrule* [flushleft,right=$\mGLdruleMSTXXImpLName{}$,left=]{
    \inferrule* [flushleft,right=,left=$\Pi' :$] {
      \pi'
    }{\delta_{{\mathrm{1}}}  \mGLsym{,}  \delta_{{\mathrm{3}}}  \odot  \Delta_{{\mathrm{1}}}  \mGLsym{,}  \Delta_{{\mathrm{3}}}  \mGLsym{;}  \Gamma_{{\mathrm{1}}}  \mGLsym{,}  \Gamma_{{\mathrm{2}}}  \mGLsym{,}  \mGLnt{B}  \mGLsym{,}  \Gamma_{{\mathrm{4}}}  \mGLsym{,}  \Gamma_{{\mathrm{5}}}  \vdash_{\mathsf{MS} }  \mGLnt{D}} \\
    \inferrule* [flushleft,right=,left=$\Pi_{{\mathrm{4}}} :$] {
      \pi_4
      }{ \delta_{{\mathrm{3}}}  \odot  \Delta_{{\mathrm{3}}}  \mGLsym{;}  \Gamma_{{\mathrm{2}}}  \mGLsym{,}  \mGLnt{B}  \mGLsym{,}  \Gamma_{{\mathrm{4}}}  \vdash_{\mathsf{MS} }  \mGLnt{C}}
  }{\delta_{{\mathrm{2}}}  \mGLsym{,}  \delta_{{\mathrm{1}}}  \mGLsym{,}  \delta_{{\mathrm{3}}}  \odot  \Delta_{{\mathrm{2}}}  \mGLsym{,}  \Delta_{{\mathrm{1}}}  \mGLsym{,}  \Delta_{{\mathrm{3}}}  \mGLsym{;}  \Gamma_{{\mathrm{1}}}  \mGLsym{,}  \Gamma_{{\mathrm{2}}}  \mGLsym{,}  \mGLnt{A}  \multimap  \mGLnt{B}  \mGLsym{,}  \Gamma_{{\mathrm{3}}}  \mGLsym{,}  \Gamma_{{\mathrm{4}}}  \mGLsym{,}  \Gamma_{{\mathrm{5}}}  \vdash_{\mathsf{MS} }  \mGLnt{D}}
  }{\delta_{{\mathrm{1}}}  \mGLsym{,}  \delta_{{\mathrm{2}}}  \mGLsym{,}  \delta_{{\mathrm{3}}}  \odot  \Delta_{{\mathrm{1}}}  \mGLsym{,}  \Delta_{{\mathrm{2}}}  \mGLsym{,}  \Delta_{{\mathrm{3}}}  \mGLsym{;}  \Gamma_{{\mathrm{1}}}  \mGLsym{,}  \Gamma_{{\mathrm{2}}}  \mGLsym{,}  \mGLnt{B}  \mGLsym{,}  \Gamma_{{\mathrm{4}}}  \mGLsym{,}  \Gamma_{{\mathrm{5}}}  \vdash_{\mathsf{MS} }  \mGLnt{D}}
  \]
  We know $\mathsf{CutRank} \, \mGLsym{(}  \Pi  \mGLsym{)} \, \mGLsym{=} \, \mathsf{CutRank} \, \mGLsym{(}  \Pi'  \mGLsym{)} \, \leq \,  \mathsf{Rank}  (  \mGLnt{C}  )$

   \item \textbf{Left introduction of Lin:} 
   \[
    \inferrule* [flushleft,right=$\mGLdruleMSTXXLinLName{}$, left=$\Pi_{{\mathrm{1}}} :$] {
      \inferrule* [flushleft,right=, left=$\Pi_{{\mathrm{3}}} :$] {
        \pi_3
      }{\delta_{{\mathrm{2}}}  \odot  \Delta_{{\mathrm{2}}}  \mGLsym{;}   ( \mGLnt{A}  \mGLsym{,}  \Gamma_{{\mathrm{2}}} )   \vdash_{\mathsf{MS} }  \mGLnt{B}}
    }{(  \delta_{{\mathrm{2}}}  \mGLsym{,}  1  )   \odot   ( \Delta_{{\mathrm{2}}}  \mGLsym{,}  \mathsf{Lin} \, \mGLnt{A} )   \mGLsym{;}  \Gamma_{{\mathrm{2}}}  \vdash_{\mathsf{MS} }  \mGLnt{B}}
    \]
    \[
      \inferrule* [flushleft,right=, left=$\Pi_{{\mathrm{2}}} :$] {
        \pi_2
      }{\delta_{{\mathrm{1}}}  \odot  \Delta_{{\mathrm{1}}}  \mGLsym{;}   ( \Gamma_{{\mathrm{1}}}  \mGLsym{,}  \mGLnt{B}  \mGLsym{,}  \Gamma_{{\mathrm{3}}} )   \vdash_{\mathsf{MS} }  \mGLnt{C}}
      \]
      We know the following:
      \[
        \begin{array}{lll}
          \mathsf{Depth}  (  \Pi_{{\mathrm{1}}}  )   +   \mathsf{Depth}  (  \Pi_{{\mathrm{2}}}  ) = \mGLsym{(}    \mathsf{Depth}  (  \Pi_{{\mathrm{3}}}  )   + 1   \mGLsym{)}  +   \mathsf{Depth}  (  \Pi_{{\mathrm{2}}}  )\\
          \mathsf{CutRank} \, \mGLsym{(}  \Pi_{{\mathrm{1}}}  \mGLsym{)} \, \mGLsym{=} \, \mathsf{CutRank} \, \mGLsym{(}  \Pi_{{\mathrm{3}}}  \mGLsym{)} \, \leq \,  \mathsf{Rank}  (  \mGLnt{B}  )\\
          \mathsf{CutRank} \, \mGLsym{(}  \Pi_{{\mathrm{2}}}  \mGLsym{)} \, \leq \,  \mathsf{Rank}  (  \mGLnt{B}  )
        \end{array}
      \]
      These imply that:
      \[
        \begin{array}{lll}
          \mathsf{Depth}  (  \Pi_{{\mathrm{3}}}  )   +   \mathsf{Depth}  (  \Pi_{{\mathrm{4}}}  )   \, \mGLsym{<} \,  \mathsf{Depth}  (  \Pi_{{\mathrm{1}}}  )   +   \mathsf{Depth}  (  \Pi_{{\mathrm{2}}}  )\\
          \mathsf{CutRank} \, \mGLsym{(}  \Pi_{{\mathrm{3}}}  \mGLsym{)} \, \leq \,  \mathsf{Rank}  (  \mGLnt{B}  )
          \mathsf{CutRank} \, \mGLsym{(}  \Pi_{{\mathrm{4}}}  \mGLsym{)} \, \leq \,  \mathsf{Rank}  (  \mGLnt{B}  )
        \end{array}
      \]
      Thus, we apply the induction hypothesis of Lemma~\ref{lemma:cut_reduction_for_mgl} (3) to $\Pi_{{\mathrm{3}}}$ and $\Pi_{{\mathrm{2}}}$ to obtain a proof
      $\Pi'$ of the sequent $(  \delta_{{\mathrm{1}}}  \mGLsym{,}  \delta_{{\mathrm{2}}}  )   \odot   ( \Delta_{{\mathrm{1}}}  \mGLsym{,}  \Delta_{{\mathrm{2}}} )   \mGLsym{;}   ( \Gamma_{{\mathrm{1}}}  \mGLsym{,}  \mGLnt{A}  \mGLsym{,}  \Gamma_{{\mathrm{2}}}  \mGLsym{,}  \Gamma_{{\mathrm{3}}} )   \vdash_{\mathsf{MS} }  \mGLnt{C}$ with $\mathsf{CutRank} \, \mGLsym{(}  \Pi'  \mGLsym{)} \, \leq \,  \mathsf{Rank}  (  \mGLnt{B}  )$.
      Now we define the proof $\Pi$ as follows:
      \[
        \inferrule* [flushleft,right=$\mGLdruleMSTXXLinLName{}$, left=$\Pi :$] {
          \inferrule* [flushleft,right=$\mGLdruleMSTXXExName{}$, left=] {
            \inferrule* [flushleft,right=, left=$\Pi' :$] {
              \pi'
            }{(  \delta_{{\mathrm{1}}}  \mGLsym{,}  \delta_{{\mathrm{2}}}  )   \odot   ( \Delta_{{\mathrm{1}}}  \mGLsym{,}  \Delta_{{\mathrm{2}}} )   \mGLsym{;}   ( \Gamma_{{\mathrm{1}}}  \mGLsym{,}  \mGLnt{A}  \mGLsym{,}  \Gamma_{{\mathrm{2}}}  \mGLsym{,}  \Gamma_{{\mathrm{3}}} )   \vdash_{\mathsf{MS} }  \mGLnt{C}}
          }{(  \delta_{{\mathrm{1}}}  \mGLsym{,}  \delta_{{\mathrm{2}}}  )   \odot   ( \Delta_{{\mathrm{1}}}  \mGLsym{,}  \Delta_{{\mathrm{2}}} )   \mGLsym{;}   ( \mGLnt{A}  \mGLsym{,}  \Gamma_{{\mathrm{1}}}  \mGLsym{,}  \Gamma_{{\mathrm{2}}}  \mGLsym{,}  \Gamma_{{\mathrm{3}}} )   \vdash_{\mathsf{MS} }  \mGLnt{C}}
        }{(  \delta_{{\mathrm{1}}}  \mGLsym{,}  \delta_{{\mathrm{2}}}  \mGLsym{,}  1  )   \odot   ( \Delta_{{\mathrm{1}}}  \mGLsym{,}  \Delta_{{\mathrm{2}}}  \mGLsym{,}  \mathsf{Lin} \, \mGLnt{A} )   \mGLsym{;}   ( \Gamma_{{\mathrm{1}}}  \mGLsym{,}  \Gamma_{{\mathrm{2}}}  \mGLsym{,}  \Gamma_{{\mathrm{3}}} )   \vdash_{\mathsf{MS} }  \mGLnt{C}}
        \]
      with: $\mathsf{CutRank} \, \mGLsym{(}  \Pi  \mGLsym{)} \, \mGLsym{=} \, \mathsf{CutRank} \, \mGLsym{(}  \Pi'  \mGLsym{)} \, \leq \,  \mathsf{Rank}  (  \mGLnt{B}  )$
   \item \textbf{Left introduction of Grd:} 
   \[
    \inferrule* [flushleft,right=$\mGLdruleMSTXXGrdRName{}$, left=$\Pi_{{\mathrm{1}}} :$] {
      \inferrule* [flushleft,right=, left=$\Pi_{{\mathrm{3}}} :$] {
        \pi_3
      }{(  \delta_{{\mathrm{2}}}  \mGLsym{,}  \mGLnt{r}  )   \odot   ( \Delta_{{\mathrm{2}}}  \mGLsym{,}  \mGLnt{X} )   \mGLsym{;}  \Gamma_{{\mathrm{2}}}  \vdash_{\mathsf{MS} }  \mGLnt{B}}
    }{\delta_{{\mathrm{2}}}  \odot  \Delta_{{\mathrm{2}}}  \mGLsym{;}   (  \mathsf{Grd} _{ \mGLnt{r} }\, \mGLnt{X}   \mGLsym{,}  \Gamma_{{\mathrm{2}}} )   \vdash_{\mathsf{MS} }  \mGLnt{B}}
    \]
    \[
      \inferrule* [flushleft,right=, left=$\Pi_{{\mathrm{2}}} :$] {
        \pi_2
      }{\delta_{{\mathrm{1}}}  \odot  \Delta_{{\mathrm{1}}}  \mGLsym{;}   ( \Gamma_{{\mathrm{1}}}  \mGLsym{,}  \mGLnt{B}  \mGLsym{,}  \Gamma_{{\mathrm{3}}} )   \vdash_{\mathsf{MS} }  \mGLnt{C}}
      \]
      We know the following:
      \[
        \begin{array}{lll}
          \mathsf{Depth}  (  \Pi_{{\mathrm{1}}}  )   +   \mathsf{Depth}  (  \Pi_{{\mathrm{2}}}  ) = \mGLsym{(}    \mathsf{Depth}  (  \Pi_{{\mathrm{3}}}  )   + 1   \mGLsym{)}  +   \mathsf{Depth}  (  \Pi_{{\mathrm{2}}}  )\\
          \mathsf{CutRank} \, \mGLsym{(}  \Pi_{{\mathrm{1}}}  \mGLsym{)} \, \mGLsym{=} \, \mathsf{CutRank} \, \mGLsym{(}  \Pi_{{\mathrm{3}}}  \mGLsym{)} \, \leq \,  \mathsf{Rank}  (  \mGLnt{B}  )\\
          \mathsf{CutRank} \, \mGLsym{(}  \Pi_{{\mathrm{2}}}  \mGLsym{)} \, \leq \,  \mathsf{Rank}  (  \mGLnt{B}  )
        \end{array}
      \]
      These imply that:
      \[
        \begin{array}{lll}
          \mathsf{Depth}  (  \Pi_{{\mathrm{3}}}  )   +   \mathsf{Depth}  (  \Pi_{{\mathrm{4}}}  )   \, \mGLsym{<} \,  \mathsf{Depth}  (  \Pi_{{\mathrm{1}}}  )   +   \mathsf{Depth}  (  \Pi_{{\mathrm{2}}}  )\\
          \mathsf{CutRank} \, \mGLsym{(}  \Pi_{{\mathrm{3}}}  \mGLsym{)} \, \leq \,  \mathsf{Rank}  (  \mGLnt{B}  )
          \mathsf{CutRank} \, \mGLsym{(}  \Pi_{{\mathrm{4}}}  \mGLsym{)} \, \leq \,  \mathsf{Rank}  (  \mGLnt{B}  )
        \end{array}
      \]
      Thus, we apply the induction hypothesis of Lemma~\ref{lemma:cut_reduction_for_mgl} (3) to $\Pi_{{\mathrm{3}}}$ and $\Pi_{{\mathrm{2}}}$ to obtain a proof
      $\Pi'$ of the sequent $(  \delta_{{\mathrm{1}}}  \mGLsym{,}  \delta_{{\mathrm{2}}}  \mGLsym{,}  \mGLnt{r}  )   \odot   ( \Delta_{{\mathrm{1}}}  \mGLsym{,}  \Delta_{{\mathrm{2}}}  \mGLsym{,}  \mGLnt{X} )   \mGLsym{;}   ( \Gamma_{{\mathrm{1}}}  \mGLsym{,}  \Gamma_{{\mathrm{2}}}  \mGLsym{,}  \Gamma_{{\mathrm{3}}} )   \vdash_{\mathsf{MS} }  \mGLnt{C}$ with $\mathsf{CutRank} \, \mGLsym{(}  \Pi'  \mGLsym{)} \, \leq \,  \mathsf{Rank}  (  \mGLnt{B}  )$.
      Now we define the proof $\Pi$ as follows:
      \[
        \inferrule* [flushleft,right=$\mGLdruleMSTXXExName{}$, left=$\Pi :$] {
          \inferrule* [flushleft,right=$\mGLdruleMSTXXGrdRName{}$, left=] {
            \inferrule* [flushleft,right=, left=$\Pi' :$] {
              \pi'
            }{(  \delta_{{\mathrm{1}}}  \mGLsym{,}  \delta_{{\mathrm{2}}}  \mGLsym{,}  \mGLnt{r}  )   \odot   ( \Delta_{{\mathrm{1}}}  \mGLsym{,}  \Delta_{{\mathrm{2}}}  \mGLsym{,}  \mGLnt{X} )   \mGLsym{;}   ( \Gamma_{{\mathrm{1}}}  \mGLsym{,}  \Gamma_{{\mathrm{2}}}  \mGLsym{,}  \Gamma_{{\mathrm{3}}} )   \vdash_{\mathsf{MS} }  \mGLnt{C}}
          }{(  \delta_{{\mathrm{1}}}  \mGLsym{,}  \delta_{{\mathrm{2}}}  )   \odot   ( \Delta_{{\mathrm{1}}}  \mGLsym{,}  \Delta_{{\mathrm{2}}} )   \mGLsym{;}   (  \mathsf{Grd} _{ \mGLnt{r} }\, \mGLnt{X}   \mGLsym{,}  \Gamma_{{\mathrm{1}}}  \mGLsym{,}  \Gamma_{{\mathrm{2}}}  \mGLsym{,}  \Gamma_{{\mathrm{3}}} )   \vdash_{\mathsf{MS} }  \mGLnt{C}}
        }{(  \delta_{{\mathrm{1}}}  \mGLsym{,}  \delta_{{\mathrm{2}}}  )   \odot   ( \Delta_{{\mathrm{1}}}  \mGLsym{,}  \Delta_{{\mathrm{2}}} )   \mGLsym{;}   ( \Gamma_{{\mathrm{1}}}  \mGLsym{,}   \mathsf{Grd} _{ \mGLnt{r} }\, \mGLnt{X}   \mGLsym{,}  \Gamma_{{\mathrm{2}}}  \mGLsym{,}  \Gamma_{{\mathrm{3}}} )   \vdash_{\mathsf{MS} }  \mGLnt{C}}
        \]
      with: $\mathsf{CutRank} \, \mGLsym{(}  \Pi  \mGLsym{)} \, \mGLsym{=} \, \mathsf{CutRank} \, \mGLsym{(}  \Pi'  \mGLsym{)} \, \leq \,  \mathsf{Rank}  (  \mGLnt{B}  )$
  \end{enumerate}

 \item \textbf{Secondary Hypothesis}
  \begin{enumerate}
    \item \textbf{Right introduction of linear tensor product (first case):}
    \[
      \inferrule* [flushleft,right=,left=$\Pi_{{\mathrm{1}}} :$] {
        \pi_1
      }{\delta_{{\mathrm{3}}}  \odot  \Delta_{{\mathrm{3}}}  \mGLsym{;}  \Gamma_{{\mathrm{2}}}  \vdash_{\mathsf{MS} }  \mGLnt{A}}
\]

\[
  \inferrule* [flushleft,right=$\mGLdruleMSTXXTenRName{}$,left=$\Pi_{{\mathrm{2}}} :$] {
    \inferrule* [flushleft,right=,left=$\Pi_{{\mathrm{3}}} :$] {
      \pi_3
    }{\delta_{{\mathrm{1}}}  \odot  \Delta_{{\mathrm{1}}}  \mGLsym{;}  \Gamma_{{\mathrm{1}}}  \mGLsym{,}  \mGLnt{A}  \mGLsym{,}  \Gamma_{{\mathrm{3}}}  \vdash_{\mathsf{MS} }  \mGLnt{B}}\\
    \inferrule* [flushleft,right=,left=$\Pi_{{\mathrm{4}}} :$] {
      \pi_4
    }{\delta_{{\mathrm{2}}}  \odot  \Delta_{{\mathrm{2}}}  \mGLsym{;}  \Gamma_{{\mathrm{4}}}  \vdash_{\mathsf{MS} }  \mGLnt{C}}
  }{(  \delta_{{\mathrm{1}}}  \mGLsym{,}  \delta_{{\mathrm{2}}}  )   \odot   ( \Delta_{{\mathrm{1}}}  \mGLsym{,}  \Delta_{{\mathrm{2}}} )   \mGLsym{;}  \Gamma_{{\mathrm{2}}}  \mGLsym{,}  \mGLnt{A}  \mGLsym{,}  \Gamma_{{\mathrm{3}}}  \mGLsym{,}  \Gamma_{{\mathrm{4}}}  \vdash_{\mathsf{MS} }  \mGLnt{B}  \otimes  \mGLnt{C}}
  \]
      We know:
      \[
      \begin{array}{lll}
        \mathsf{Depth}  (  \Pi_{{\mathrm{1}}}  )   +   \mathsf{Depth}  (  \Pi_{{\mathrm{3}}}  )   \, \mGLsym{<} \,  \mathsf{Depth}  (  \Pi_{{\mathrm{1}}}  )   +   \mathsf{Depth}  (  \Pi_{{\mathrm{2}}}  )\\
        \mathsf{CutRank} \, \mGLsym{(}  \Pi_{{\mathrm{3}}}  \mGLsym{)} \, \leq \, \mathsf{CutRank} \, \mGLsym{(}  \Pi_{{\mathrm{2}}}  \mGLsym{)} \, \leq \,  \mathsf{Rank}  (  \mGLnt{A}  )
      \end{array}
      \]

      and so applying the induction hypothesis
      to $\Pi_{{\mathrm{1}}}$ and $\Pi_{{\mathrm{3}}}$
      implies that there is a proof $\Pi'$ of
      $\delta_{{\mathrm{1}}}  \mGLsym{,}  \delta_{{\mathrm{3}}}  \odot  \Delta_{{\mathrm{1}}}  \mGLsym{,}  \Delta_{{\mathrm{3}}}  \mGLsym{;}  \Gamma_{{\mathrm{1}}}  \mGLsym{,}  \Gamma_{{\mathrm{2}}}  \mGLsym{,}  \Gamma_{{\mathrm{3}}}  \vdash_{\mathsf{MS} }  \mGLnt{B}$ with
      $\mathsf{CutRank} \, \mGLsym{(}  \Pi'  \mGLsym{)} \, \leq \,  \mathsf{Rank}  (  \mGLnt{A}  )$.
      Thus, we construct the following proof $\Pi$:

      \[
        \inferrule* [flushleft,right=$\mGLdruleMSTXXGExName{}$,left=$\Pi :$] {
        \inferrule* [flushleft,right=$\mGLdruleMSTXXTenRName{}$,left=] {
          \inferrule* [flushleft,left=$\Pi' : $] {
            \pi'
          }{\delta_{{\mathrm{1}}}  \mGLsym{,}  \delta_{{\mathrm{3}}}  \odot  \Delta_{{\mathrm{1}}}  \mGLsym{,}  \Delta_{{\mathrm{3}}}  \mGLsym{;}  \Gamma_{{\mathrm{1}}}  \mGLsym{,}  \Gamma_{{\mathrm{2}}}  \mGLsym{,}  \Gamma_{{\mathrm{3}}}  \vdash_{\mathsf{MS} }  \mGLnt{B}}\\
          \inferrule* [flushleft,right=,left=$\Pi_{{\mathrm{4}}} :$] {
            \pi_4
          }{\delta_{{\mathrm{2}}}  \odot  \Delta_{{\mathrm{2}}}  \mGLsym{;}  \Gamma_{{\mathrm{4}}}  \vdash_{\mathsf{MS} }  \mGLnt{C}}
        }{\delta_{{\mathrm{1}}}  \mGLsym{,}  \delta_{{\mathrm{3}}}  \mGLsym{,}  \delta_{{\mathrm{2}}}  \odot  \Delta_{{\mathrm{1}}}  \mGLsym{,}  \Delta_{{\mathrm{3}}}  \mGLsym{,}  \Delta_{{\mathrm{2}}}  \mGLsym{;}  \Gamma_{{\mathrm{1}}}  \mGLsym{,}  \Gamma_{{\mathrm{2}}}  \mGLsym{,}  \Gamma_{{\mathrm{3}}}  \mGLsym{,}  \Gamma_{{\mathrm{4}}}  \vdash_{\mathsf{MS} }  \mGLnt{B}  \otimes  \mGLnt{C}}
        }{\delta_{{\mathrm{1}}}  \mGLsym{,}  \delta_{{\mathrm{2}}}  \mGLsym{,}  \delta_{{\mathrm{3}}}  \odot  \Delta_{{\mathrm{1}}}  \mGLsym{,}  \Delta_{{\mathrm{2}}}  \mGLsym{,}  \Delta_{{\mathrm{3}}}  \mGLsym{;}  \Gamma_{{\mathrm{1}}}  \mGLsym{,}  \Gamma_{{\mathrm{2}}}  \mGLsym{,}  \Gamma_{{\mathrm{3}}}  \mGLsym{,}  \Gamma_{{\mathrm{4}}}  \vdash_{\mathsf{MS} }  \mGLnt{B}  \otimes  \mGLnt{C}}
        \]
   \item \textbf{Right introduction of linear tensor product (second case):}
   \[
    \inferrule* [flushleft,right=,left=$\Pi_{{\mathrm{1}}} :$] {
      \pi_1
    }{\delta_{{\mathrm{3}}}  \odot  \Delta_{{\mathrm{3}}}  \mGLsym{;}  \Gamma_{{\mathrm{3}}}  \vdash_{\mathsf{MS} }  \mGLnt{A}}
\]

\[
\inferrule* [flushleft,right=$\mGLdruleMSTXXTenRName{}$,left=$\Pi_{{\mathrm{2}}} :$] {
  \inferrule* [flushleft,right=,left=$\Pi_{{\mathrm{3}}} :$] {
    \pi_3
  }{\delta_{{\mathrm{1}}}  \odot  \Delta_{{\mathrm{1}}}  \mGLsym{;}  \Gamma_{{\mathrm{1}}}  \vdash_{\mathsf{MS} }  \mGLnt{B}}\\
  \inferrule* [flushleft,right=,left=$\Pi_{{\mathrm{4}}} :$] {
    \pi_4
  }{\delta_{{\mathrm{2}}}  \odot  \Delta_{{\mathrm{2}}}  \mGLsym{;}  \Gamma_{{\mathrm{2}}}  \mGLsym{,}  \mGLnt{A}  \mGLsym{,}  \Gamma_{{\mathrm{4}}}  \vdash_{\mathsf{MS} }  \mGLnt{C}}
}{(  \delta_{{\mathrm{1}}}  \mGLsym{,}  \delta_{{\mathrm{2}}}  )   \odot   ( \Delta_{{\mathrm{1}}}  \mGLsym{,}  \Delta_{{\mathrm{2}}} )   \mGLsym{;}  \Gamma_{{\mathrm{2}}}  \mGLsym{,}  \mGLnt{A}  \mGLsym{,}  \Gamma_{{\mathrm{3}}}  \mGLsym{,}  \Gamma_{{\mathrm{4}}}  \vdash_{\mathsf{MS} }  \mGLnt{B}  \otimes  \mGLnt{C}}
\]
Similar to the previous case.
   \item \textbf{Left introduction of linear tensor product:}
   \[
    \inferrule* [flushleft,right=,left=$\Pi_{{\mathrm{1}}} :$] {
      \pi_1
    }{\delta_{{\mathrm{2}}}  \odot  \Delta_{{\mathrm{2}}}  \mGLsym{;}  \Gamma_{{\mathrm{3}}}  \vdash_{\mathsf{MS} }  \mGLnt{A}}
\]

\[
  \inferrule* [flushleft,right=$\mGLdruleMSTXXTenLName{}$,left=$\Pi_{{\mathrm{2}}} :$] {
    \inferrule* [flushleft,right=,left=$\Pi_{{\mathrm{3}}} :$] {
      \pi_3
    }{\delta_{{\mathrm{1}}}  \odot  \Delta_{{\mathrm{1}}}  \mGLsym{;}  \Gamma_{{\mathrm{1}}}  \mGLsym{,}  \mGLnt{A}  \mGLsym{,}  \Gamma_{{\mathrm{3}}}  \mGLsym{,}  \mGLnt{B}  \mGLsym{,}  \mGLnt{C}  \mGLsym{,}  \Gamma_{{\mathrm{4}}}  \vdash_{\mathsf{MS} }  \mGLnt{D}}
    }{\delta_{{\mathrm{1}}}  \odot  \Delta_{{\mathrm{1}}}  \mGLsym{;}  \Gamma_{{\mathrm{1}}}  \mGLsym{,}  \mGLnt{A}  \mGLsym{,}  \Gamma_{{\mathrm{3}}}  \mGLsym{,}  \mGLnt{B}  \otimes  \mGLnt{C}  \mGLsym{,}  \Gamma_{{\mathrm{4}}}  \vdash_{\mathsf{MS} }  \mGLnt{D}}
  \]
      We know:
      \[
      \begin{array}{lll}
        \mathsf{Depth}  (  \Pi_{{\mathrm{1}}}  )   +   \mathsf{Depth}  (  \Pi_{{\mathrm{3}}}  )   \, \mGLsym{<} \,  \mathsf{Depth}  (  \Pi_{{\mathrm{1}}}  )   +   \mathsf{Depth}  (  \Pi_{{\mathrm{2}}}  )\\
        \mathsf{CutRank} \, \mGLsym{(}  \Pi_{{\mathrm{3}}}  \mGLsym{)} \, \leq \, \mathsf{CutRank} \, \mGLsym{(}  \Pi_{{\mathrm{2}}}  \mGLsym{)} \, \leq \,  \mathsf{Rank}  (  \mGLnt{A}  )
      \end{array}
      \]

      and so applying the induction hypothesis
      to $\Pi_{{\mathrm{1}}}$ and $\Pi_{{\mathrm{3}}}$
      implies that there is a proof $\Pi'$ of
      $\delta_{{\mathrm{1}}}  \mGLsym{,}  \delta_{{\mathrm{2}}}  \odot  \Delta_{{\mathrm{1}}}  \mGLsym{,}  \Delta_{{\mathrm{2}}}  \mGLsym{;}  \Gamma_{{\mathrm{1}}}  \mGLsym{,}  \Gamma_{{\mathrm{2}}}  \mGLsym{,}  \Gamma_{{\mathrm{3}}}  \mGLsym{,}  \mGLnt{B}  \mGLsym{,}  \mGLnt{C}  \mGLsym{,}  \Gamma_{{\mathrm{4}}}  \vdash_{\mathsf{MS} }  \mGLnt{D}$ with
      $\mathsf{CutRank} \, \mGLsym{(}  \Pi'  \mGLsym{)} \, \leq \,  \mathsf{Rank}  (  \mGLnt{A}  )$.
      Thus, we construct the following proof $\Pi$:

      \[
        \inferrule* [flushleft,right=$\mGLdruleMSTXXTenLName{}$,left=$\Pi :$] {
          \inferrule* [flushleft,right=,left=$\Pi' :$] {
            \pi'
          }{\delta_{{\mathrm{1}}}  \mGLsym{,}  \delta_{{\mathrm{2}}}  \odot  \Delta_{{\mathrm{1}}}  \mGLsym{,}  \Delta_{{\mathrm{2}}}  \mGLsym{;}  \Gamma_{{\mathrm{1}}}  \mGLsym{,}  \Gamma_{{\mathrm{2}}}  \mGLsym{,}  \Gamma_{{\mathrm{3}}}  \mGLsym{,}  \mGLnt{B}  \mGLsym{,}  \mGLnt{C}  \mGLsym{,}  \Gamma_{{\mathrm{4}}}  \vdash_{\mathsf{MS} }  \mGLnt{D}}
        }{\delta_{{\mathrm{1}}}  \mGLsym{,}  \delta_{{\mathrm{2}}}  \odot  \Delta_{{\mathrm{1}}}  \mGLsym{,}  \Delta_{{\mathrm{2}}}  \mGLsym{;}  \Gamma_{{\mathrm{1}}}  \mGLsym{,}  \Gamma_{{\mathrm{2}}}  \mGLsym{,}  \Gamma_{{\mathrm{3}}}  \mGLsym{,}  \mGLnt{B}  \otimes  \mGLnt{C}  \mGLsym{,}  \Gamma_{{\mathrm{4}}}  \vdash_{\mathsf{MS} }  \mGLnt{D}}
        \]

        \item \textbf{Left introduction of linear tensor product (second case):}
        \[
         \inferrule* [flushleft,right=,left=$\Pi_{{\mathrm{1}}} :$] {
           \pi_1
         }{\delta_{{\mathrm{2}}}  \odot  \Delta_{{\mathrm{2}}}  \mGLsym{;}  \Gamma_{{\mathrm{3}}}  \vdash_{\mathsf{MS} }  \mGLnt{A}}
     \]

     \[
       \inferrule* [flushleft,right=$\mGLdruleMSTXXTenLName{}$,left=$\Pi_{{\mathrm{2}}} :$] {
         \inferrule* [flushleft,right=,left=$\Pi_{{\mathrm{3}}} :$] {
           \pi_3
         }{\delta_{{\mathrm{1}}}  \odot  \Delta_{{\mathrm{1}}}  \mGLsym{;}  \Gamma_{{\mathrm{1}}}  \mGLsym{,}  \mGLnt{B}  \mGLsym{,}  \mGLnt{C}  \mGLsym{,}  \Gamma_{{\mathrm{2}}}  \mGLsym{,}  \mGLnt{A}  \mGLsym{,}  \Gamma_{{\mathrm{4}}}  \vdash_{\mathsf{MS} }  \mGLnt{D}}
         }{\delta_{{\mathrm{1}}}  \odot  \Delta_{{\mathrm{1}}}  \mGLsym{;}  \Gamma_{{\mathrm{1}}}  \mGLsym{,}  \mGLnt{B}  \otimes  \mGLnt{C}  \mGLsym{,}  \Gamma_{{\mathrm{2}}}  \mGLsym{,}  \mGLnt{A}  \mGLsym{,}  \Gamma_{{\mathrm{4}}}  \vdash_{\mathsf{MS} }  \mGLnt{D}}
       \]
       Similar to previous case.
   \item \textbf{Left introduction of the unit of linear tensor:}
   \[
    \inferrule* [flushleft,right=,left=$\Pi_{{\mathrm{1}}} :$] {
      \pi_1
    }{\delta_{{\mathrm{2}}}  \odot  \Delta_{{\mathrm{2}}}  \mGLsym{;}  \Gamma_{{\mathrm{3}}}  \vdash_{\mathsf{MS} }  \mGLnt{A}}
\]

      \[
        \inferrule* [flushleft,right=$\mGLdruleMSTXXUnitLName{}$,left=$\Pi_{{\mathrm{2}}} :$] {
          \inferrule* [flushleft,left=$\Pi_{{\mathrm{3}}} : $] {
            \pi_3
          }{\delta_{{\mathrm{1}}}  \odot  \Delta_{{\mathrm{1}}}  \mGLsym{;}  \Gamma_{{\mathrm{1}}}  \mGLsym{,}  \Gamma_{{\mathrm{2}}}  \mGLsym{,}  \mGLnt{A}  \mGLsym{,}  \Gamma_{{\mathrm{4}}}  \vdash_{\mathsf{MS} }  \mGLnt{B}}
        }{\delta_{{\mathrm{1}}}  \odot  \Delta_{{\mathrm{1}}}  \mGLsym{;}  \Gamma_{{\mathrm{1}}}  \mGLsym{,}  \mathsf{I}  \mGLsym{,}  \Gamma_{{\mathrm{2}}}  \mGLsym{,}  \mGLnt{A}  \mGLsym{,}  \Gamma_{{\mathrm{4}}}  \vdash_{\mathsf{MS} }  \mGLnt{B}}
        \]

      We know:
      \[
      \begin{array}{lll}
        \mathsf{Depth}  (  \Pi_{{\mathrm{1}}}  )   +   \mathsf{Depth}  (  \Pi_{{\mathrm{3}}}  )   \, \mGLsym{<} \,  \mathsf{Depth}  (  \Pi_{{\mathrm{1}}}  )   +   \mathsf{Depth}  (  \Pi_{{\mathrm{2}}}  )\\
        \mathsf{CutRank} \, \mGLsym{(}  \Pi_{{\mathrm{3}}}  \mGLsym{)} \, \leq \, \mathsf{CutRank} \, \mGLsym{(}  \Pi_{{\mathrm{2}}}  \mGLsym{)} \, \leq \,  \mathsf{Rank}  (  \mGLnt{A}  )
      \end{array}
      \]

      and so applying the induction hypothesis
      to $\Pi_{{\mathrm{1}}}$ and $\Pi_{{\mathrm{3}}}$
      implies that there is a proof $\Pi'$ of
      $\delta_{{\mathrm{1}}}  \mGLsym{,}  \delta_{{\mathrm{2}}}  \odot  \Delta_{{\mathrm{1}}}  \mGLsym{,}  \Delta_{{\mathrm{2}}}  \mGLsym{;}  \Gamma_{{\mathrm{1}}}  \mGLsym{,}  \Gamma_{{\mathrm{2}}}  \mGLsym{,}  \Gamma_{{\mathrm{3}}}  \mGLsym{,}  \Gamma_{{\mathrm{4}}}  \vdash_{\mathsf{MS} }  \mGLnt{B}$ with
      $\mathsf{CutRank} \, \mGLsym{(}  \Pi'  \mGLsym{)} \, \leq \,  \mathsf{Rank}  (  \mGLnt{A}  )$.
      Thus, we construct the following proof $\Pi$:

      \[
        \inferrule* [flushleft,right=$\mGLdruleMSTXXUnitLName{}$,left=$\Pi :$] {
          \inferrule* [flushleft,left=$\Pi' : $] {
            \pi'
          }{\delta_{{\mathrm{1}}}  \mGLsym{,}  \delta_{{\mathrm{2}}}  \odot  \Delta_{{\mathrm{1}}}  \mGLsym{,}  \Delta_{{\mathrm{2}}}  \mGLsym{;}  \Gamma_{{\mathrm{1}}}  \mGLsym{,}  \Gamma_{{\mathrm{2}}}  \mGLsym{,}  \Gamma_{{\mathrm{3}}}  \mGLsym{,}  \Gamma_{{\mathrm{4}}}  \vdash_{\mathsf{MS} }  \mGLnt{B}}
        }{\delta_{{\mathrm{1}}}  \mGLsym{,}  \delta_{{\mathrm{2}}}  \odot  \Delta_{{\mathrm{1}}}  \mGLsym{,}  \Delta_{{\mathrm{2}}}  \mGLsym{;}  \Gamma_{{\mathrm{1}}}  \mGLsym{,}  \mathsf{I}  \mGLsym{,}  \Gamma_{{\mathrm{2}}}  \mGLsym{,}  \Gamma_{{\mathrm{3}}}  \mGLsym{,}  \Gamma_{{\mathrm{4}}}  \vdash_{\mathsf{MS} }  \mGLnt{B}}
        \]
        \item \textbf{Left introduction of the unit of linear tensor (second case):}
        \[
         \inferrule* [flushleft,right=,left=$\Pi_{{\mathrm{1}}} :$] {
           \pi_1
         }{\delta_{{\mathrm{2}}}  \odot  \Delta_{{\mathrm{2}}}  \mGLsym{;}  \Gamma_{{\mathrm{2}}}  \vdash_{\mathsf{MS} }  \mGLnt{A}}
     \]

           \[
             \inferrule* [flushleft,right=$\mGLdruleMSTXXUnitLName{}$,left=$\Pi_{{\mathrm{2}}} :$] {
               \inferrule* [flushleft,left=$\Pi_{{\mathrm{3}}} : $] {
                 \pi_3
               }{\delta_{{\mathrm{1}}}  \odot  \Delta_{{\mathrm{1}}}  \mGLsym{;}  \Gamma_{{\mathrm{1}}}  \mGLsym{,}  \mGLnt{A}  \mGLsym{,}  \Gamma_{{\mathrm{3}}}  \mGLsym{,}  \Gamma_{{\mathrm{4}}}  \vdash_{\mathsf{MS} }  \mGLnt{B}}
             }{\delta_{{\mathrm{1}}}  \odot  \Delta_{{\mathrm{1}}}  \mGLsym{;}  \Gamma_{{\mathrm{1}}}  \mGLsym{,}  \mGLnt{A}  \mGLsym{,}  \Gamma_{{\mathrm{3}}}  \mGLsym{,}  \mathsf{I}  \mGLsym{,}  \Gamma_{{\mathrm{4}}}  \vdash_{\mathsf{MS} }  \mGLnt{B}}
             \]
     Similar to the previous case.
      \item \textbf{Left introduction of graded tensor product:}
      \[
    \inferrule* [flushleft,right=,left=$\Pi_{{\mathrm{1}}} :$] {
      \pi_1
    }{\delta_{{\mathrm{3}}}  \odot  \Delta_{{\mathrm{3}}}  \mGLsym{;}  \Gamma_{{\mathrm{2}}}  \vdash_{\mathsf{MS} }  \mGLnt{A}}
\]

\[
  \inferrule* [flushleft,right=$\mGLdruleMSTXXGTenLName{}$,left=$\Pi_{{\mathrm{2}}} :$] {
    \inferrule* [flushleft,right=,left=$\Pi_{{\mathrm{3}}} :$] {
      \pi_3
    }{(  \delta_{{\mathrm{1}}}  \mGLsym{,}  \mGLnt{r}  \mGLsym{,}  \mGLnt{r}  \mGLsym{,}  \delta_{{\mathrm{2}}}  )   \odot   ( \Delta_{{\mathrm{1}}}  \mGLsym{,}  \mGLnt{Y}  \mGLsym{,}  \mGLnt{Z}  \mGLsym{,}  \Delta_{{\mathrm{2}}} )   \mGLsym{;}  \Gamma_{{\mathrm{1}}}  \mGLsym{,}  \mGLnt{A}  \mGLsym{,}  \Gamma_{{\mathrm{3}}}  \vdash_{\mathsf{MS} }  \mGLnt{B}}
    }{(  \delta_{{\mathrm{1}}}  \mGLsym{,}  \mGLnt{r}  \mGLsym{,}  \delta_{{\mathrm{2}}}  )   \odot   ( \Delta_{{\mathrm{1}}}  \mGLsym{,}  \mGLnt{Y}  \boxtimes  \mGLnt{Z}  \mGLsym{,}  \Delta_{{\mathrm{2}}} )   \mGLsym{;}  \Gamma_{{\mathrm{1}}}  \mGLsym{,}  \mGLnt{A}  \mGLsym{,}  \Gamma_{{\mathrm{3}}}  \vdash_{\mathsf{MS} }  \mGLnt{B}}
  \]

We know:
\[
\begin{array}{lll}
  \mathsf{Depth}  (  \Pi_{{\mathrm{1}}}  )   +   \mathsf{Depth}  (  \Pi_{{\mathrm{3}}}  )   \, \mGLsym{<} \,  \mathsf{Depth}  (  \Pi_{{\mathrm{1}}}  )   +   \mathsf{Depth}  (  \Pi_{{\mathrm{2}}}  )\\
  \mathsf{CutRank} \, \mGLsym{(}  \Pi_{{\mathrm{3}}}  \mGLsym{)} \, \leq \, \mathsf{CutRank} \, \mGLsym{(}  \Pi_{{\mathrm{2}}}  \mGLsym{)} \, \leq \,  \mathsf{Rank}  (  \mGLnt{A}  )
\end{array}
\]

and so applying the induction hypothesis
to $\Pi_{{\mathrm{1}}}$ and $\Pi_{{\mathrm{3}}}$
implies that there is a proof $\Pi'$ of
$(  \delta_{{\mathrm{1}}}  \mGLsym{,}  \mGLnt{r}  \mGLsym{,}  \mGLnt{r}  \mGLsym{,}  \delta_{{\mathrm{2}}}  \mGLsym{,}  \delta_{{\mathrm{3}}}  )   \odot   ( \Delta_{{\mathrm{1}}}  \mGLsym{,}  \mGLnt{Y}  \mGLsym{,}  \mGLnt{Z}  \mGLsym{,}  \Delta_{{\mathrm{2}}}  \mGLsym{,}  \Delta_{{\mathrm{3}}} )   \mGLsym{;}  \Gamma_{{\mathrm{1}}}  \mGLsym{,}  \Gamma_{{\mathrm{2}}}  \mGLsym{,}  \Gamma_{{\mathrm{3}}}  \vdash_{\mathsf{MS} }  \mGLnt{B}$ with
$\mathsf{CutRank} \, \mGLsym{(}  \Pi'  \mGLsym{)} \, \leq \,  \mathsf{Rank}  (  \mGLnt{X}  )$
Thus, we construct the following proof $\Pi$:

\[
  \inferrule* [flushleft,right=$\mGLdruleMSTXXGTenLName{}$,left=$\Pi :$] {
    \inferrule* [flushleft,right=,left=$\Pi' :$] {
      \pi'
    }{(  \delta_{{\mathrm{1}}}  \mGLsym{,}  \mGLnt{r}  \mGLsym{,}  \mGLnt{r}  \mGLsym{,}  \delta_{{\mathrm{2}}}  \mGLsym{,}  \delta_{{\mathrm{3}}}  )   \odot   ( \Delta_{{\mathrm{1}}}  \mGLsym{,}  \mGLnt{Y}  \mGLsym{,}  \mGLnt{Z}  \mGLsym{,}  \Delta_{{\mathrm{2}}}  \mGLsym{,}  \Delta_{{\mathrm{3}}} )   \mGLsym{;}  \Gamma_{{\mathrm{1}}}  \mGLsym{,}  \Gamma_{{\mathrm{2}}}  \mGLsym{,}  \Gamma_{{\mathrm{3}}}  \vdash_{\mathsf{MS} }  \mGLnt{B}}
  }{(  \delta_{{\mathrm{1}}}  \mGLsym{,}  \mGLnt{r}  \mGLsym{,}  \delta_{{\mathrm{2}}}  \mGLsym{,}  \delta_{{\mathrm{3}}}  )   \odot   ( \Delta_{{\mathrm{1}}}  \mGLsym{,}  \mGLnt{Y}  \boxtimes  \mGLnt{Z}  \mGLsym{,}  \Delta_{{\mathrm{2}}}  \mGLsym{,}  \Delta_{{\mathrm{3}}} )   \mGLsym{;}  \Gamma_{{\mathrm{1}}}  \mGLsym{,}  \Gamma_{{\mathrm{2}}}  \mGLsym{,}  \Gamma_{{\mathrm{3}}}  \vdash_{\mathsf{MS} }  \mGLnt{B}}
  \]
  Given the above, we know:
      $\mathsf{CutRank} \, \mGLsym{(}  \Pi  \mGLsym{)} \, \mGLsym{=} \, \mathsf{CutRank} \, \mGLsym{(}  \Pi'  \mGLsym{)} \, \leq \,  \mathsf{Rank}  (  \mGLnt{A}  )$

    \item \textbf{Left introduction of the unit of graded tensor:}
    \[
      \inferrule* [flushleft,right=,left=$\Pi_{{\mathrm{1}}} :$] {
        \pi_1
      }{\delta_{{\mathrm{3}}}  \odot  \Delta_{{\mathrm{3}}}  \mGLsym{;}  \Gamma_{{\mathrm{2}}}  \vdash_{\mathsf{MS} }  \mGLnt{A}}
\]

  \[
    \inferrule* [flushleft,right=$\mGLdruleMSTXXGUnitLName{}$,left=$\Pi_{{\mathrm{2}}} :$] {
      \inferrule* [flushleft,left=$\Pi_{{\mathrm{3}}} : $] {
        \pi_3
      }{(  \delta_{{\mathrm{1}}}  \mGLsym{,}  \delta_{{\mathrm{2}}}  )   \odot   ( \Delta_{{\mathrm{1}}}  \mGLsym{,}  \Delta_{{\mathrm{2}}} )   \mGLsym{;}  \Gamma_{{\mathrm{1}}}  \mGLsym{,}  \mGLnt{A}  \mGLsym{,}  \Gamma_{{\mathrm{3}}}  \vdash_{\mathsf{MS} }  \mGLnt{A}}\\
    }{(  \delta_{{\mathrm{1}}}  \mGLsym{,}  \mGLnt{r}  \mGLsym{,}  \delta_{{\mathrm{2}}}  )   \odot   ( \Delta_{{\mathrm{1}}}  \mGLsym{,}  \mathsf{J}  \mGLsym{,}  \Delta_{{\mathrm{2}}} )   \mGLsym{;}  \Gamma_{{\mathrm{1}}}  \mGLsym{,}  \mGLnt{A}  \mGLsym{,}  \Gamma_{{\mathrm{3}}}  \vdash_{\mathsf{MS} }  \mGLnt{A}}
    \]

      We know:
      \[
      \begin{array}{lll}
        \mathsf{Depth}  (  \Pi_{{\mathrm{1}}}  )   +   \mathsf{Depth}  (  \Pi_{{\mathrm{3}}}  )   \, \mGLsym{<} \,  \mathsf{Depth}  (  \Pi_{{\mathrm{1}}}  )   +   \mathsf{Depth}  (  \Pi_{{\mathrm{2}}}  )\\
        \mathsf{CutRank} \, \mGLsym{(}  \Pi_{{\mathrm{3}}}  \mGLsym{)} \, \leq \, \mathsf{CutRank} \, \mGLsym{(}  \Pi_{{\mathrm{2}}}  \mGLsym{)} \, \leq \,  \mathsf{Rank}  (  \mGLnt{A}  )
      \end{array}
      \]

      and so applying the induction hypothesis
      to $\Pi_{{\mathrm{1}}}$ and $\Pi_{{\mathrm{3}}}$
      implies that there is a proof $\Pi'$ of
      $(  \delta_{{\mathrm{1}}}  \mGLsym{,}  \delta_{{\mathrm{2}}}  \mGLsym{,}  \delta_{{\mathrm{3}}}  )   \odot   ( \Delta_{{\mathrm{1}}}  \mGLsym{,}  \Delta_{{\mathrm{2}}}  \mGLsym{,}  \Delta_{{\mathrm{3}}} )   \mGLsym{;}  \Gamma_{{\mathrm{1}}}  \mGLsym{,}  \Gamma_{{\mathrm{2}}}  \mGLsym{,}  \Gamma_{{\mathrm{3}}}  \vdash_{\mathsf{MS} }  \mGLnt{A}$ with
      $\mathsf{CutRank} \, \mGLsym{(}  \Pi'  \mGLsym{)} \, \leq \,  \mathsf{Rank}  (  \mGLnt{A}  )$.
      Thus, we construct the following proof $\Pi$:

      \[
        \inferrule* [flushleft,right=$\mGLdruleMSTXXGUnitLName{}$,left=$\Pi :$] {
          \inferrule* [flushleft,left=$\Pi' : $] {
            \pi'
          }{(  \delta_{{\mathrm{1}}}  \mGLsym{,}  \delta_{{\mathrm{2}}}  \mGLsym{,}  \delta_{{\mathrm{3}}}  )   \odot   ( \Delta_{{\mathrm{1}}}  \mGLsym{,}  \Delta_{{\mathrm{2}}}  \mGLsym{,}  \Delta_{{\mathrm{3}}} )   \mGLsym{;}  \Gamma_{{\mathrm{1}}}  \mGLsym{,}  \Gamma_{{\mathrm{2}}}  \mGLsym{,}  \Gamma_{{\mathrm{3}}}  \vdash_{\mathsf{MS} }  \mGLnt{A}}
        }{(  \delta_{{\mathrm{1}}}  \mGLsym{,}  \mGLnt{r}  \mGLsym{,}  \delta_{{\mathrm{2}}}  \mGLsym{,}  \delta_{{\mathrm{3}}}  )   \odot   ( \Delta_{{\mathrm{1}}}  \mGLsym{,}  \mathsf{J}  \mGLsym{,}  \Delta_{{\mathrm{2}}}  \mGLsym{,}  \Delta_{{\mathrm{3}}} )   \mGLsym{;}  \Gamma_{{\mathrm{1}}}  \mGLsym{,}  \Gamma_{{\mathrm{2}}}  \mGLsym{,}  \Gamma_{{\mathrm{3}}}  \vdash_{\mathsf{MS} }  \mGLnt{A}}
        \]
        We know $\mathsf{CutRank} \, \mGLsym{(}  \Pi  \mGLsym{)} \, \mGLsym{=} \, \mathsf{CutRank} \, \mGLsym{(}  \Pi'  \mGLsym{)} \, \leq \,  \mathsf{Rank}  (  \mGLnt{A}  )$
    \item \textbf{Right introduction of linear implication:}
    \[
    \inferrule* [flushleft,right=,left=$\Pi_{{\mathrm{1}}} :$] {
      \pi_1
    }{\delta_{{\mathrm{2}}}  \odot  \Delta_{{\mathrm{2}}}  \mGLsym{;}  \Gamma_{{\mathrm{2}}}  \vdash_{\mathsf{MS} }  \mGLnt{A}}
\]

      \[
        \inferrule* [flushleft,right=$\mGLdruleMSTXXImpRName{}$,left=$\Pi_{{\mathrm{2}}} :$] {
          \inferrule* [flushleft,left=$\Pi_{{\mathrm{3}}} : $] {
            \pi_3
          }{\delta_{{\mathrm{1}}}  \odot  \Delta_{{\mathrm{1}}}  \mGLsym{;}  \Gamma_{{\mathrm{1}}}  \mGLsym{,}  \mGLnt{A}  \mGLsym{,}  \Gamma_{{\mathrm{3}}}  \mGLsym{,}  \mGLnt{B}  \vdash_{\mathsf{MS} }  \mGLnt{C}}
        }{\delta_{{\mathrm{1}}}  \odot  \Delta_{{\mathrm{1}}}  \mGLsym{;}  \Gamma_{{\mathrm{1}}}  \mGLsym{,}  \mGLnt{A}  \mGLsym{,}  \Gamma_{{\mathrm{3}}}  \vdash_{\mathsf{MS} }  \mGLnt{B}  \multimap  \mGLnt{C}}
        \]

      We know:
      \[
      \begin{array}{lll}
        \mathsf{Depth}  (  \Pi_{{\mathrm{1}}}  )   +   \mathsf{Depth}  (  \Pi_{{\mathrm{3}}}  )   \, \mGLsym{<} \,  \mathsf{Depth}  (  \Pi_{{\mathrm{1}}}  )   +   \mathsf{Depth}  (  \Pi_{{\mathrm{2}}}  )\\
        \mathsf{CutRank} \, \mGLsym{(}  \Pi_{{\mathrm{3}}}  \mGLsym{)} \, \leq \, \mathsf{CutRank} \, \mGLsym{(}  \Pi_{{\mathrm{2}}}  \mGLsym{)} \, \leq \,  \mathsf{Rank}  (  \mGLnt{A}  )
      \end{array}
      \]

      and so applying the induction hypothesis
      to $\Pi_{{\mathrm{1}}}$ and $\Pi_{{\mathrm{3}}}$
      implies that there is a proof $\Pi'$ of
      $\delta_{{\mathrm{1}}}  \odot  \Delta_{{\mathrm{1}}}  \mGLsym{;}  \Gamma_{{\mathrm{1}}}  \mGLsym{,}  \Gamma_{{\mathrm{2}}}  \mGLsym{,}  \Gamma_{{\mathrm{3}}}  \mGLsym{,}  \mGLnt{B}  \vdash_{\mathsf{MS} }  \mGLnt{C}$ with
      $\mathsf{CutRank} \, \mGLsym{(}  \Pi'  \mGLsym{)} \, \leq \,  \mathsf{Rank}  (  \mGLnt{A}  )$.
      Thus, we construct the following proof $\Pi$:

      \[
        \inferrule* [flushleft,right=$\mGLdruleMSTXXImpRName{}$,left=$\Pi :$] {
          \inferrule* [flushleft,left=$\Pi' : $] {
            \pi'
          }{\delta_{{\mathrm{1}}}  \odot  \Delta_{{\mathrm{1}}}  \mGLsym{;}  \Gamma_{{\mathrm{1}}}  \mGLsym{,}  \Gamma_{{\mathrm{2}}}  \mGLsym{,}  \Gamma_{{\mathrm{3}}}  \mGLsym{,}  \mGLnt{B}  \vdash_{\mathsf{MS} }  \mGLnt{C}}
        }{\delta_{{\mathrm{1}}}  \odot  \Delta_{{\mathrm{1}}}  \mGLsym{;}  \Gamma_{{\mathrm{1}}}  \mGLsym{,}  \Gamma_{{\mathrm{2}}}  \mGLsym{,}  \Gamma_{{\mathrm{3}}}  \vdash_{\mathsf{MS} }  \mGLnt{B}  \multimap  \mGLnt{C}}
        \]
        We know $\mathsf{CutRank} \, \mGLsym{(}  \Pi  \mGLsym{)} \, \mGLsym{=} \, \mathsf{CutRank} \, \mGLsym{(}  \Pi'  \mGLsym{)} \, \leq \,  \mathsf{Rank}  (  \mGLnt{A}  )$

    \item \textbf{Left introduction of linear implication:}
    \[
    \inferrule* [flushleft,right=,left=$\Pi_{{\mathrm{1}}} :$] {
      \pi_1
    }{\delta_{{\mathrm{3}}}  \odot  \Delta_{{\mathrm{3}}}  \mGLsym{;}  \Gamma_{{\mathrm{3}}}  \vdash_{\mathsf{MS} }  \mGLnt{A}}
\]

      \[
        \inferrule* [flushleft,right=$\mGLdruleMSTXXImpLName{}$,left=$\Pi_{{\mathrm{2}}} :$] {
          \inferrule* [flushleft,left=$\Pi_{{\mathrm{3}}} : $] {
            \pi_3
          }{\delta_{{\mathrm{2}}}  \odot  \Delta_{{\mathrm{2}}}  \mGLsym{;}  \Gamma_{{\mathrm{2}}}  \mGLsym{,}  \mGLnt{A}  \mGLsym{,}  \Gamma_{{\mathrm{4}}}  \vdash_{\mathsf{MS} }  \mGLnt{B}}\\
          \inferrule* [flushleft,right=,left=$\Pi_{{\mathrm{4}}} :$] {
            \pi_4
          }{\delta_{{\mathrm{1}}}  \odot  \Delta_{{\mathrm{1}}}  \mGLsym{;}  \Gamma_{{\mathrm{1}}}  \mGLsym{,}  \mGLnt{C}  \mGLsym{,}  \Gamma_{{\mathrm{5}}}  \vdash_{\mathsf{MS} }  \mGLnt{D}}
        }{\delta_{{\mathrm{1}}}  \mGLsym{,}  \delta_{{\mathrm{2}}}  \odot  \Delta_{{\mathrm{1}}}  \mGLsym{,}  \Delta_{{\mathrm{2}}}  \mGLsym{;}  \Gamma_{{\mathrm{1}}}  \mGLsym{,}  \mGLnt{B}  \multimap  \mGLnt{C}  \mGLsym{,}  \Gamma_{{\mathrm{2}}}  \mGLsym{,}  \mGLnt{A}  \mGLsym{,}  \Gamma_{{\mathrm{4}}}  \mGLsym{,}  \Gamma_{{\mathrm{5}}}  \vdash_{\mathsf{MS} }  \mGLnt{D} }
        \]

      We know:
      \[
      \begin{array}{lll}
        \mathsf{Depth}  (  \Pi_{{\mathrm{1}}}  )   +   \mathsf{Depth}  (  \Pi_{{\mathrm{3}}}  )   \, \mGLsym{<} \,  \mathsf{Depth}  (  \Pi_{{\mathrm{1}}}  )   +   \mathsf{Depth}  (  \Pi_{{\mathrm{2}}}  )\\
        \mathsf{CutRank} \, \mGLsym{(}  \Pi_{{\mathrm{3}}}  \mGLsym{)} \, \leq \, \mathsf{CutRank} \, \mGLsym{(}  \Pi_{{\mathrm{2}}}  \mGLsym{)} \, \leq \,  \mathsf{Rank}  (  \mGLnt{X}  )
      \end{array}
      \]

      and so applying the induction hypothesis
      to $\Pi_{{\mathrm{1}}}$ and $\Pi_{{\mathrm{3}}}$
      implies that there is a proof $\Pi'$ of
      $\delta_{{\mathrm{2}}}  \mGLsym{,}  \delta_{{\mathrm{3}}}  \odot  \Delta_{{\mathrm{2}}}  \mGLsym{,}  \Delta_{{\mathrm{3}}}  \mGLsym{;}  \Gamma_{{\mathrm{2}}}  \mGLsym{,}  \Gamma_{{\mathrm{3}}}  \mGLsym{,}  \Gamma_{{\mathrm{4}}}  \vdash_{\mathsf{MS} }  \mGLnt{B}$ with
      $\mathsf{CutRank} \, \mGLsym{(}  \Pi'  \mGLsym{)} \, \leq \,  \mathsf{Rank}  (  \mGLnt{A}  )$.
      Thus, we construct the following proof $\Pi$:

      \[
        \inferrule* [flushleft,right=$\mGLdruleMSTXXImpLName{}$,left=$\Pi :$] {
          \inferrule* [flushleft,left=$\Pi' : $] {
            \pi'
          }{\delta_{{\mathrm{2}}}  \mGLsym{,}  \delta_{{\mathrm{3}}}  \odot  \Delta_{{\mathrm{2}}}  \mGLsym{,}  \Delta_{{\mathrm{3}}}  \mGLsym{;}  \Gamma_{{\mathrm{2}}}  \mGLsym{,}  \Gamma_{{\mathrm{3}}}  \mGLsym{,}  \Gamma_{{\mathrm{4}}}  \vdash_{\mathsf{MS} }  \mGLnt{B}}\\
          \inferrule* [flushleft,right=,left=$\Pi_{{\mathrm{4}}} :$] {
            \pi_4
          }{\delta_{{\mathrm{1}}}  \odot  \Delta_{{\mathrm{1}}}  \mGLsym{;}  \Gamma_{{\mathrm{1}}}  \mGLsym{,}  \mGLnt{C}  \mGLsym{,}  \Gamma_{{\mathrm{5}}}  \vdash_{\mathsf{MS} }  \mGLnt{D}}
        }{\delta_{{\mathrm{1}}}  \mGLsym{,}  \delta_{{\mathrm{2}}}  \mGLsym{,}  \delta_{{\mathrm{3}}}  \odot  \Delta_{{\mathrm{1}}}  \mGLsym{,}  \Delta_{{\mathrm{2}}}  \mGLsym{,}  \Delta_{{\mathrm{3}}}  \mGLsym{;}  \Gamma_{{\mathrm{1}}}  \mGLsym{,}  \mGLnt{B}  \multimap  \mGLnt{C}  \mGLsym{,}  \Gamma_{{\mathrm{2}}}  \mGLsym{,}  \Gamma_{{\mathrm{3}}}  \mGLsym{,}  \Gamma_{{\mathrm{4}}}  \mGLsym{,}  \Gamma_{{\mathrm{5}}}  \vdash_{\mathsf{MS} }  \mGLnt{D}}
        \]
        We know that $\mathsf{CutRank} \, \mGLsym{(}  \Pi  \mGLsym{)} \, \mGLsym{=} \, \mGLkw{Max} \, \mGLsym{(}  \mathsf{CutRank} \, \mGLsym{(}  \Pi'  \mGLsym{)}  \mGLsym{,}  \mathsf{CutRank} \, \mGLsym{(}  \Pi_{{\mathrm{4}}}  \mGLsym{)}  \mGLsym{)} \, \leq \,  \mathsf{Rank}  (  \mGLnt{A}  )$

   \item \textbf{Left introduction of linear implication (second case):}
        \[
        \inferrule* [flushleft,right=,left=$\Pi_{{\mathrm{1}}} :$] {
          \pi_1
        }{\delta_{{\mathrm{3}}}  \odot  \Delta_{{\mathrm{3}}}  \mGLsym{;}  \Gamma_{{\mathrm{4}}}  \vdash_{\mathsf{MS} }  \mGLnt{A}}
    \]

          \[
            \inferrule* [flushleft,right=$\mGLdruleMSTXXImpLName{}$,left=$\Pi_{{\mathrm{2}}} :$] {
              \inferrule* [flushleft,left=$\Pi_{{\mathrm{3}}} : $] {
                \pi_3
              }{\delta_{{\mathrm{2}}}  \odot  \Delta_{{\mathrm{2}}}  \mGLsym{;}  \Gamma_{{\mathrm{2}}}  \vdash_{\mathsf{MS} }  \mGLnt{B}}\\
              \inferrule* [flushleft,right=,left=$\Pi_{{\mathrm{4}}} :$] {
                \pi_4
              }{\delta_{{\mathrm{1}}}  \odot  \Delta_{{\mathrm{1}}}  \mGLsym{;}  \Gamma_{{\mathrm{1}}}  \mGLsym{,}  \mGLnt{C}  \mGLsym{,}  \Gamma_{{\mathrm{3}}}  \mGLsym{,}  \mGLnt{A}  \mGLsym{,}  \Gamma_{{\mathrm{5}}}  \vdash_{\mathsf{MS} }  \mGLnt{D}}
            }{\delta_{{\mathrm{1}}}  \mGLsym{,}  \delta_{{\mathrm{2}}}  \odot  \Delta_{{\mathrm{1}}}  \mGLsym{,}  \Delta_{{\mathrm{2}}}  \mGLsym{;}  \Gamma_{{\mathrm{1}}}  \mGLsym{,}  \mGLnt{B}  \multimap  \mGLnt{C}  \mGLsym{,}  \Gamma_{{\mathrm{2}}}  \mGLsym{,}  \Gamma_{{\mathrm{3}}}  \mGLsym{,}  \mGLnt{A}  \mGLsym{,}  \Gamma_{{\mathrm{5}}}  \vdash_{\mathsf{MS} }  \mGLnt{D} }
            \]
  We know:
      \[
      \begin{array}{lll}
        \mathsf{Depth}  (  \Pi_{{\mathrm{1}}}  )   +   \mathsf{Depth}  (  \Pi_{{\mathrm{4}}}  )   \, \mGLsym{<} \,  \mathsf{Depth}  (  \Pi_{{\mathrm{1}}}  )   +   \mathsf{Depth}  (  \Pi_{{\mathrm{2}}}  )\\
        \mathsf{CutRank} \, \mGLsym{(}  \Pi_{{\mathrm{3}}}  \mGLsym{)} \, \leq \, \mathsf{CutRank} \, \mGLsym{(}  \Pi_{{\mathrm{2}}}  \mGLsym{)} \, \leq \,  \mathsf{Rank}  (  \mGLnt{X}  )
      \end{array}
      \]

      and so applying the induction hypothesis
      to $\Pi_{{\mathrm{1}}}$ and $\Pi_{{\mathrm{4}}}$
      implies that there is a proof $\Pi'$ of
      $\delta_{{\mathrm{1}}}  \mGLsym{,}  \delta_{{\mathrm{3}}}  \odot  \Delta_{{\mathrm{1}}}  \mGLsym{,}  \Delta_{{\mathrm{3}}}  \mGLsym{;}  \Gamma_{{\mathrm{1}}}  \mGLsym{,}  \mGLnt{C}  \mGLsym{,}  \Gamma_{{\mathrm{3}}}  \mGLsym{,}  \Gamma_{{\mathrm{4}}}  \mGLsym{,}  \Gamma_{{\mathrm{5}}}  \vdash_{\mathsf{MS} }  \mGLnt{D}$ with
      $\mathsf{CutRank} \, \mGLsym{(}  \Pi'  \mGLsym{)} \, \leq \,  \mathsf{Rank}  (  \mGLnt{A}  )$.
      Thus, we construct the following proof $\Pi$:

      \[
        \inferrule* [flushleft,right=$\mGLdruleMSTXXGExName{}$,left=$\Pi :$] {
        \inferrule* [flushleft,right=$\mGLdruleMSTXXImpLName{}$,left=] {
          \inferrule* [flushleft,left=$\Pi_{{\mathrm{3}}} : $] {
            \pi_3
          }{\delta_{{\mathrm{2}}}  \odot  \Delta_{{\mathrm{2}}}  \mGLsym{;}  \Gamma_{{\mathrm{2}}}  \vdash_{\mathsf{MS} }  \mGLnt{B}}\\
          \inferrule* [flushleft,right=,left=$\Pi' :$] {
            \pi'
          }{\delta_{{\mathrm{1}}}  \mGLsym{,}  \delta_{{\mathrm{3}}}  \odot  \Delta_{{\mathrm{1}}}  \mGLsym{,}  \Delta_{{\mathrm{3}}}  \mGLsym{;}  \Gamma_{{\mathrm{1}}}  \mGLsym{,}  \mGLnt{C}  \mGLsym{,}  \Gamma_{{\mathrm{3}}}  \mGLsym{,}  \Gamma_{{\mathrm{4}}}  \mGLsym{,}  \Gamma_{{\mathrm{5}}}  \vdash_{\mathsf{MS} }  \mGLnt{D}}
        }{\delta_{{\mathrm{1}}}  \mGLsym{,}  \delta_{{\mathrm{3}}}  \mGLsym{,}  \delta_{{\mathrm{2}}}  \odot  \Delta_{{\mathrm{1}}}  \mGLsym{,}  \Delta_{{\mathrm{3}}}  \mGLsym{,}  \Delta_{{\mathrm{2}}}  \mGLsym{;}  \Gamma_{{\mathrm{1}}}  \mGLsym{,}  \mGLnt{B}  \multimap  \mGLnt{C}  \mGLsym{,}  \Gamma_{{\mathrm{2}}}  \mGLsym{,}  \Gamma_{{\mathrm{3}}}  \mGLsym{,}  \Gamma_{{\mathrm{4}}}  \mGLsym{,}  \Gamma_{{\mathrm{5}}}  \vdash_{\mathsf{MS} }  \mGLnt{D}}
        }{\delta_{{\mathrm{1}}}  \mGLsym{,}  \delta_{{\mathrm{2}}}  \mGLsym{,}  \delta_{{\mathrm{3}}}  \odot  \Delta_{{\mathrm{1}}}  \mGLsym{,}  \Delta_{{\mathrm{2}}}  \mGLsym{,}  \Delta_{{\mathrm{3}}}  \mGLsym{;}  \Gamma_{{\mathrm{1}}}  \mGLsym{,}  \mGLnt{B}  \multimap  \mGLnt{C}  \mGLsym{,}  \Gamma_{{\mathrm{2}}}  \mGLsym{,}  \Gamma_{{\mathrm{3}}}  \mGLsym{,}  \Gamma_{{\mathrm{4}}}  \mGLsym{,}  \Gamma_{{\mathrm{5}}}  \vdash_{\mathsf{MS} }  \mGLnt{D}}
        \]
        We know that $\mathsf{CutRank} \, \mGLsym{(}  \Pi  \mGLsym{)} \, \mGLsym{=} \, \mGLkw{Max} \, \mGLsym{(}  \mathsf{CutRank} \, \mGLsym{(}  \Pi'  \mGLsym{)}  \mGLsym{,}  \mathsf{CutRank} \, \mGLsym{(}  \Pi_{{\mathrm{3}}}  \mGLsym{)}  \mGLsym{)} \, \leq \,  \mathsf{Rank}  (  \mGLnt{A}  )$

   \item \textbf{Left introduction of linear implication (third case):}
        \[
        \inferrule* [flushleft,right=,left=$\Pi_{{\mathrm{1}}} :$] {
          \pi_1
        }{\delta_{{\mathrm{3}}}  \odot  \Delta_{{\mathrm{3}}}  \mGLsym{;}  \Gamma_{{\mathrm{2}}}  \vdash_{\mathsf{MS} }  \mGLnt{A}}
    \]

          \[
            \inferrule* [flushleft,right=$\mGLdruleMSTXXImpLName{}$,left=$\Pi_{{\mathrm{2}}} :$] {
              \inferrule* [flushleft,left=$\Pi_{{\mathrm{3}}} : $] {
                \pi_3
              }{\delta_{{\mathrm{2}}}  \odot  \Delta_{{\mathrm{2}}}  \mGLsym{;}  \Gamma_{{\mathrm{4}}}  \vdash_{\mathsf{MS} }  \mGLnt{B}}\\
              \inferrule* [flushleft,right=,left=$\Pi_{{\mathrm{4}}} :$] {
                \pi_4
              }{\delta_{{\mathrm{1}}}  \odot  \Delta_{{\mathrm{1}}}  \mGLsym{;}  \Gamma_{{\mathrm{1}}}  \mGLsym{,}  \mGLnt{A}  \mGLsym{,}  \Gamma_{{\mathrm{3}}}  \mGLsym{,}  \mGLnt{C}  \mGLsym{,}  \Gamma_{{\mathrm{5}}}  \vdash_{\mathsf{MS} }  \mGLnt{D}}
            }{\delta_{{\mathrm{1}}}  \mGLsym{,}  \delta_{{\mathrm{2}}}  \odot  \Delta_{{\mathrm{1}}}  \mGLsym{,}  \Delta_{{\mathrm{2}}}  \mGLsym{;}  \Gamma_{{\mathrm{1}}}  \mGLsym{,}  \mGLnt{A}  \mGLsym{,}  \Gamma_{{\mathrm{3}}}  \mGLsym{,}  \mGLnt{B}  \multimap  \mGLnt{C}  \mGLsym{,}  \Gamma_{{\mathrm{4}}}  \mGLsym{,}  \Gamma_{{\mathrm{5}}}  \vdash_{\mathsf{MS} }  \mGLnt{D} }
            \]
    Similar to the previous case.
   \item \textbf{Left introduction of Lin:} %
    \[
      \inferrule* [flushleft,right=, left=$\Pi_{{\mathrm{1}}} :$] {
        \pi_1
      }{\delta_{{\mathrm{2}}}  \odot  \Delta_{{\mathrm{2}}}  \mGLsym{;}  \Gamma_{{\mathrm{2}}}  \vdash_{\mathsf{MS} }  \mGLnt{B}}
      \]
      \[
    \inferrule* [flushleft,right=$\mGLdruleMSTXXLinLName{}$, left=$\Pi_{{\mathrm{2}}} :$] {
      \inferrule* [flushleft,right=, left=$\Pi_{{\mathrm{3}}} :$] {
        \pi_3
      }{\delta_{{\mathrm{1}}}  \odot  \Delta_{{\mathrm{1}}}  \mGLsym{;}  \mGLnt{A}  \mGLsym{,}  \Gamma_{{\mathrm{1}}}  \mGLsym{,}  \mGLnt{B}  \mGLsym{,}  \Gamma_{{\mathrm{3}}}  \vdash_{\mathsf{MS} }  \mGLnt{C}}
    }{(  \delta_{{\mathrm{1}}}  \mGLsym{,}  1  )   \odot   ( \Delta_{{\mathrm{1}}}  \mGLsym{,}  \mathsf{Lin} \, \mGLnt{A} )   \mGLsym{;}  \Gamma_{{\mathrm{1}}}  \mGLsym{,}  \mGLnt{B}  \mGLsym{,}  \Gamma_{{\mathrm{3}}}  \vdash_{\mathsf{MS} }  \mGLnt{C}}
    \]
      We know the following:
      \[
        \begin{array}{lll}
          \mathsf{Depth}  (  \Pi_{{\mathrm{3}}}  )   +   \mathsf{Depth}  (  \Pi_{{\mathrm{1}}}  )   \, \mGLsym{<} \,  \mathsf{Depth}  (  \Pi_{{\mathrm{1}}}  )   +   \mathsf{Depth}  (  \Pi_{{\mathrm{2}}}  )\\
          \mathsf{CutRank} \, \mGLsym{(}  \Pi_{{\mathrm{3}}}  \mGLsym{)} \, \leq \, \mathsf{CutRank} \, \mGLsym{(}  \Pi_{{\mathrm{2}}}  \mGLsym{)} \, \leq \,  \mathsf{Rank}  (  \mGLnt{B}  )
          \mathsf{CutRank} \, \mGLsym{(}  \Pi_{{\mathrm{1}}}  \mGLsym{)} \, \leq \,  \mathsf{Rank}  (  \mGLnt{B}  )
        \end{array}
      \]
      Thus, we apply the induction hypothesis to $\Pi_{{\mathrm{3}}}$ and $\Pi_{{\mathrm{1}}}$ to obtain
      a proof $\Pi'$ of the sequent $\delta_{{\mathrm{1}}}  \mGLsym{,}  \delta_{{\mathrm{2}}}  \odot  \Delta_{{\mathrm{1}}}  \mGLsym{,}  \Delta_{{\mathrm{2}}}  \mGLsym{;}  \mGLnt{A}  \mGLsym{,}  \Gamma_{{\mathrm{1}}}  \mGLsym{,}  \Gamma_{{\mathrm{2}}}  \mGLsym{,}  \Gamma_{{\mathrm{3}}}  \vdash_{\mathsf{MS} }  \mGLnt{C}$
      with $\mathsf{CutRank} \, \mGLsym{(}  \Pi'  \mGLsym{)} \, \leq \,  \mathsf{Rank}  (  \mGLnt{B}  )$.
      Now we define the proof $\Pi$ as follows:
      \[
        \inferrule* [flushleft,right=$\mGLdruleMSTXXGExName{}$, left=$\Pi :$] {
        \inferrule* [flushleft,right=$\mGLdruleMSTXXLinLName{}$, left=] {
          \inferrule* [flushleft,right=, left=$\Pi' :$] {
              \pi'
          }{\delta_{{\mathrm{1}}}  \mGLsym{,}  \delta_{{\mathrm{2}}}  \odot  \Delta_{{\mathrm{1}}}  \mGLsym{,}  \Delta_{{\mathrm{2}}}  \mGLsym{;}  \mGLnt{A}  \mGLsym{,}  \Gamma_{{\mathrm{1}}}  \mGLsym{,}  \Gamma_{{\mathrm{2}}}  \mGLsym{,}  \Gamma_{{\mathrm{3}}}  \vdash_{\mathsf{MS} }  \mGLnt{C}}
        }{(  \delta_{{\mathrm{1}}}  \mGLsym{,}  \delta_{{\mathrm{2}}}  \mGLsym{,}  1  )   \odot   ( \Delta_{{\mathrm{1}}}  \mGLsym{,}  \Delta_{{\mathrm{2}}}  \mGLsym{,}  \mathsf{Lin} \, \mGLnt{A} )   \mGLsym{;}  \Gamma_{{\mathrm{1}}}  \mGLsym{,}  \Gamma_{{\mathrm{2}}}  \mGLsym{,}  \Gamma_{{\mathrm{3}}}  \vdash_{\mathsf{MS} }  \mGLnt{C}}
        }{(  \delta_{{\mathrm{1}}}  \mGLsym{,}  1  \mGLsym{,}  \delta_{{\mathrm{2}}}  )   \odot   ( \Delta_{{\mathrm{1}}}  \mGLsym{,}  \mathsf{Lin} \, \mGLnt{A}  \mGLsym{,}  \Delta_{{\mathrm{2}}} )   \mGLsym{;}  \Gamma_{{\mathrm{1}}}  \mGLsym{,}  \Gamma_{{\mathrm{2}}}  \mGLsym{,}  \Gamma_{{\mathrm{3}}}  \vdash_{\mathsf{MS} }  \mGLnt{C}}
        \]
      with: $\mathsf{CutRank} \, \mGLsym{(}  \Pi  \mGLsym{)} \, \mGLsym{=} \, \mathsf{CutRank} \, \mGLsym{(}  \Pi'  \mGLsym{)} \, \leq \,  \mathsf{Rank}  (  \mGLnt{B}  )$
   \item \textbf{Left introduction of Grd:} %
   \[
    \inferrule* [flushleft,right=, left=$\Pi_{{\mathrm{1}}} :$] {
      \pi_1
    }{\delta_{{\mathrm{2}}}  \odot  \Delta_{{\mathrm{2}}}  \mGLsym{;}  \Gamma_{{\mathrm{2}}}  \vdash_{\mathsf{MS} }  \mGLnt{A}}
    \]
    \[
  \inferrule* [flushleft,right=$\mGLdruleMSTXXGrdLName{}$, left=$\Pi_{{\mathrm{2}}} :$] {
    \inferrule* [flushleft,right=, left=$\Pi_{{\mathrm{3}}} :$] {
      \pi_3
    }{\delta_{{\mathrm{1}}}  \mGLsym{,}  \mGLnt{r}  \odot  \Delta_{{\mathrm{1}}}  \mGLsym{,}  \mGLnt{X}  \mGLsym{;}   ( \Gamma_{{\mathrm{1}}}  \mGLsym{,}  \mGLnt{A}  \mGLsym{,}  \Gamma_{{\mathrm{3}}} )   \vdash_{\mathsf{MS} }  \mGLnt{B}}
  }{\delta_{{\mathrm{1}}}  \odot  \Delta_{{\mathrm{1}}}  \mGLsym{;}   \mathsf{Grd} _{ \mGLnt{r} }\, \mGLnt{X}   \mGLsym{,}  \Gamma_{{\mathrm{1}}}  \mGLsym{,}  \mGLnt{A}  \mGLsym{,}  \Gamma_{{\mathrm{3}}}  \vdash_{\mathsf{MS} }  \mGLnt{B}}
  \]
    We know the following:
    \[
      \begin{array}{lll}
        \mathsf{Depth}  (  \Pi_{{\mathrm{3}}}  )   +   \mathsf{Depth}  (  \Pi_{{\mathrm{1}}}  )   \, \mGLsym{<} \,  \mathsf{Depth}  (  \Pi_{{\mathrm{1}}}  )   +   \mathsf{Depth}  (  \Pi_{{\mathrm{2}}}  )\\
        \mathsf{CutRank} \, \mGLsym{(}  \Pi_{{\mathrm{3}}}  \mGLsym{)} \, \leq \, \mathsf{CutRank} \, \mGLsym{(}  \Pi_{{\mathrm{2}}}  \mGLsym{)} \, \leq \,  \mathsf{Rank}  (  \mGLnt{A}  )
        \mathsf{CutRank} \, \mGLsym{(}  \Pi_{{\mathrm{1}}}  \mGLsym{)} \, \leq \,  \mathsf{Rank}  (  \mGLnt{A}  )
      \end{array}
    \]
    Thus, we apply the induction hypothesis
    to $\Pi_{{\mathrm{3}}}$ and $\Pi_{{\mathrm{1}}}$ to obtain a proof
    $\Pi'$ of the sequent $\delta_{{\mathrm{1}}}  \mGLsym{,}  \mGLnt{r}  \mGLsym{,}  \delta_{{\mathrm{2}}}  \odot  \Delta_{{\mathrm{1}}}  \mGLsym{,}  \mGLnt{X}  \mGLsym{,}  \Delta_{{\mathrm{2}}}  \mGLsym{;}   ( \Gamma_{{\mathrm{1}}}  \mGLsym{,}  \Gamma_{{\mathrm{2}}}  \mGLsym{,}  \Gamma_{{\mathrm{3}}} )   \vdash_{\mathsf{MS} }  \mGLnt{B}$
    with $\mathsf{CutRank} \, \mGLsym{(}  \Pi'  \mGLsym{)} \, \leq \,  \mathsf{Rank}  (  \mGLnt{A}  )$.
    Now we define the proof $\Pi$ as follows:
    \[
      \inferrule* [flushleft,right=$\mGLdruleMSTXXGrdLName{}$, left=$\Pi :$] {
      \inferrule* [flushleft,right=$\mGLdruleMSTXXGExName{}$, left=] {
        \inferrule* [flushleft,right=, left=$\Pi' :$] {
            \pi'
        }{\delta_{{\mathrm{1}}}  \mGLsym{,}  \mGLnt{r}  \mGLsym{,}  \delta_{{\mathrm{2}}}  \odot  \Delta_{{\mathrm{1}}}  \mGLsym{,}  \mGLnt{X}  \mGLsym{,}  \Delta_{{\mathrm{2}}}  \mGLsym{;}   ( \Gamma_{{\mathrm{1}}}  \mGLsym{,}  \Gamma_{{\mathrm{2}}}  \mGLsym{,}  \Gamma_{{\mathrm{3}}} )   \vdash_{\mathsf{MS} }  \mGLnt{B}}
      }{\delta_{{\mathrm{1}}}  \mGLsym{,}  \delta_{{\mathrm{2}}}  \mGLsym{,}  \mGLnt{r}  \odot  \Delta_{{\mathrm{1}}}  \mGLsym{,}  \Delta_{{\mathrm{2}}}  \mGLsym{,}  \mGLnt{X}  \mGLsym{;}   ( \Gamma_{{\mathrm{1}}}  \mGLsym{,}  \Gamma_{{\mathrm{2}}}  \mGLsym{,}  \Gamma_{{\mathrm{3}}} )   \vdash_{\mathsf{MS} }  \mGLnt{B}}
      }{\delta_{{\mathrm{1}}}  \mGLsym{,}  \delta_{{\mathrm{2}}}  \odot  \Delta_{{\mathrm{1}}}  \mGLsym{,}  \Delta_{{\mathrm{2}}}  \mGLsym{;}   (  \mathsf{Grd} _{ \mGLnt{r} }\, \mGLnt{X}   \mGLsym{,}  \Gamma_{{\mathrm{1}}}  \mGLsym{,}  \Gamma_{{\mathrm{2}}}  \mGLsym{,}  \Gamma_{{\mathrm{3}}} )   \vdash_{\mathsf{MS} }  \mGLnt{B}}
      \]
    with: $\mathsf{CutRank} \, \mGLsym{(}  \Pi  \mGLsym{)} \, \mGLsym{=} \, \mathsf{CutRank} \, \mGLsym{(}  \Pi'  \mGLsym{)} \, \leq \,  \mathsf{Rank}  (  \mGLnt{A}  )$
  \end{enumerate}

 \item \textbf{Structural}
  \begin{enumerate}
   \item \textbf{Weakening}
   \[
    \inferrule* [flushleft,right=,left=$\Pi_{{\mathrm{1}}} :$] {
      \pi_1
    }{\delta_{{\mathrm{3}}}  \odot  \Delta_{{\mathrm{3}}}  \mGLsym{;}  \Gamma_{{\mathrm{2}}}  \vdash_{\mathsf{MS} }  \mGLnt{A}}
\]

      \[
        \inferrule* [flushleft,right=$\mGLdruleMSTXXWeakName{}$,left=$\Pi_{{\mathrm{2}}} :$] {
          \inferrule* [flushleft,left=$\Pi_{{\mathrm{3}}} : $] {
            \pi_3
          }{\delta_{{\mathrm{1}}}  \mGLsym{,}  \delta_{{\mathrm{2}}}  \odot  \Delta_{{\mathrm{1}}}  \mGLsym{,}  \Delta_{{\mathrm{2}}}  \mGLsym{;}  \Gamma_{{\mathrm{1}}}  \mGLsym{,}  \mGLnt{A}  \mGLsym{,}  \Gamma_{{\mathrm{3}}}  \vdash_{\mathsf{MS} }  \mGLnt{B}}
        }{\delta_{{\mathrm{1}}}  \mGLsym{,}  \mathsf{0}  \mGLsym{,}  \delta_{{\mathrm{2}}}  \odot  \Delta_{{\mathrm{1}}}  \mGLsym{,}  \mGLnt{X}  \mGLsym{,}  \Delta_{{\mathrm{2}}}  \mGLsym{;}  \Gamma_{{\mathrm{1}}}  \mGLsym{,}  \mGLnt{A}  \mGLsym{,}  \Gamma_{{\mathrm{3}}}  \vdash_{\mathsf{MS} }  \mGLnt{B}}
        \]

      We know:
      \[
      \begin{array}{lll}
        \mathsf{Depth}  (  \Pi_{{\mathrm{1}}}  )   +   \mathsf{Depth}  (  \Pi_{{\mathrm{3}}}  )   \, \mGLsym{<} \,  \mathsf{Depth}  (  \Pi_{{\mathrm{1}}}  )   +   \mathsf{Depth}  (  \Pi_{{\mathrm{2}}}  )\\
        \mathsf{CutRank} \, \mGLsym{(}  \Pi_{{\mathrm{3}}}  \mGLsym{)} \, \leq \, \mathsf{CutRank} \, \mGLsym{(}  \Pi_{{\mathrm{2}}}  \mGLsym{)} \, \leq \,  \mathsf{Rank}  (  \mGLnt{X}  )
      \end{array}
      \]

      and so applying the induction hypothesis
      to $\Pi_{{\mathrm{1}}}$ and $\Pi_{{\mathrm{3}}}$
      implies that there is a proof $\Pi'$ of
      $\delta_{{\mathrm{1}}}  \mGLsym{,}  \delta_{{\mathrm{2}}}  \mGLsym{,}  \delta_{{\mathrm{3}}}  \odot  \Delta_{{\mathrm{1}}}  \mGLsym{,}  \Delta_{{\mathrm{2}}}  \mGLsym{,}  \Delta_{{\mathrm{3}}}  \mGLsym{;}  \Gamma_{{\mathrm{1}}}  \mGLsym{,}  \Gamma_{{\mathrm{2}}}  \mGLsym{,}  \Gamma_{{\mathrm{3}}}  \vdash_{\mathsf{MS} }  \mGLnt{B}$ with
      $\mathsf{CutRank} \, \mGLsym{(}  \Pi'  \mGLsym{)} \, \leq \,  \mathsf{Rank}  (  \mGLnt{A}  )$.
      Thus, we construct the following proof $\Pi$:

      \[
        \inferrule* [flushleft,right=$\mGLdruleMSTXXWeakName{}$,left=$\Pi :$] {
          \inferrule* [flushleft,left=$\Pi' : $] {
            \pi'
          }{\delta_{{\mathrm{1}}}  \mGLsym{,}  \delta_{{\mathrm{2}}}  \mGLsym{,}  \delta_{{\mathrm{3}}}  \odot  \Delta_{{\mathrm{1}}}  \mGLsym{,}  \Delta_{{\mathrm{2}}}  \mGLsym{,}  \Delta_{{\mathrm{3}}}  \mGLsym{;}  \Gamma_{{\mathrm{1}}}  \mGLsym{,}  \Gamma_{{\mathrm{2}}}  \mGLsym{,}  \Gamma_{{\mathrm{3}}}  \vdash_{\mathsf{MS} }  \mGLnt{B}}
        }{\delta_{{\mathrm{1}}}  \mGLsym{,}  \mathsf{0}  \mGLsym{,}  \delta_{{\mathrm{2}}}  \mGLsym{,}  \delta_{{\mathrm{3}}}  \odot  \Delta_{{\mathrm{1}}}  \mGLsym{,}  \mGLnt{X}  \mGLsym{,}  \Delta_{{\mathrm{2}}}  \mGLsym{,}  \Delta_{{\mathrm{3}}}  \mGLsym{;}  \Gamma_{{\mathrm{1}}}  \mGLsym{,}  \Gamma_{{\mathrm{2}}}  \mGLsym{,}  \Gamma_{{\mathrm{3}}}  \vdash_{\mathsf{MS} }  \mGLnt{B}}
        \]
        with: $\mathsf{CutRank} \, \mGLsym{(}  \Pi  \mGLsym{)} \, \mGLsym{=} \, \mathsf{CutRank} \, \mGLsym{(}  \Pi'  \mGLsym{)} \, \leq \,  \mathsf{Rank}  (  \mGLnt{A}  )$
   \item \textbf{Contraction}
   \[
    \inferrule* [flushleft,right=,left=$\Pi_{{\mathrm{1}}} :$] {
      \pi_1
    }{\delta_{{\mathrm{3}}}  \odot  \Delta_{{\mathrm{3}}}  \mGLsym{;}  \Gamma_{{\mathrm{2}}}  \vdash_{\mathsf{MS} }  \mGLnt{A}}
\]

      \[
        \inferrule* [flushleft,right=$\mGLdruleMSTXXContName{}$,left=$\Pi_{{\mathrm{2}}} :$] {
          \inferrule* [flushleft,left=$\Pi_{{\mathrm{3}}} : $] {
            \pi_3
          }{\delta_{{\mathrm{1}}}  \mGLsym{,}  \mGLnt{r_{{\mathrm{1}}}}  \mGLsym{,}  \mGLnt{r_{{\mathrm{2}}}}  \mGLsym{,}  \delta_{{\mathrm{2}}}  \odot  \Delta_{{\mathrm{1}}}  \mGLsym{,}  \mGLnt{X}  \mGLsym{,}  \mGLnt{X}  \mGLsym{,}  \Delta_{{\mathrm{2}}}  \mGLsym{;}  \Gamma_{{\mathrm{1}}}  \mGLsym{,}  \mGLnt{A}  \mGLsym{,}  \Gamma_{{\mathrm{3}}}  \vdash_{\mathsf{MS} }  \mGLnt{B}}
        }{ \delta_{{\mathrm{1}}}  \mGLsym{,}  \mGLnt{r_{{\mathrm{1}}}}  +  \mGLnt{r_{{\mathrm{2}}}}  \mGLsym{,}  \delta_{{\mathrm{2}}}  \odot  \Delta_{{\mathrm{1}}}  \mGLsym{,}  \mGLnt{X}  \mGLsym{,}  \Delta_{{\mathrm{2}}}  \mGLsym{;}  \Gamma_{{\mathrm{1}}}  \mGLsym{,}  \mGLnt{A}  \mGLsym{,}  \Gamma_{{\mathrm{3}}}  \vdash_{\mathsf{MS} }  \mGLnt{B}}
        \]

      We know:
      \[
      \begin{array}{lll}
        \mathsf{Depth}  (  \Pi_{{\mathrm{1}}}  )   +   \mathsf{Depth}  (  \Pi_{{\mathrm{3}}}  )   \, \mGLsym{<} \,  \mathsf{Depth}  (  \Pi_{{\mathrm{1}}}  )   +   \mathsf{Depth}  (  \Pi_{{\mathrm{2}}}  )\\
        \mathsf{CutRank} \, \mGLsym{(}  \Pi_{{\mathrm{3}}}  \mGLsym{)} \, \leq \, \mathsf{CutRank} \, \mGLsym{(}  \Pi_{{\mathrm{2}}}  \mGLsym{)} \, \leq \,  \mathsf{Rank}  (  \mGLnt{A}  )
      \end{array}
      \]

      and so applying the induction hypothesis
      to $\Pi_{{\mathrm{1}}}$ and $\Pi_{{\mathrm{3}}}$
      implies that there is a proof $\Pi'$ of
      $\delta_{{\mathrm{1}}}  \mGLsym{,}  \mGLnt{r_{{\mathrm{1}}}}  \mGLsym{,}  \mGLnt{r_{{\mathrm{2}}}}  \mGLsym{,}  \delta_{{\mathrm{2}}}  \odot  \Delta_{{\mathrm{1}}}  \mGLsym{,}  \mGLnt{X}  \mGLsym{,}  \mGLnt{X}  \mGLsym{,}  \Delta_{{\mathrm{2}}}  \mGLsym{;}  \Gamma_{{\mathrm{1}}}  \mGLsym{,}  \Gamma_{{\mathrm{2}}}  \mGLsym{,}  \Gamma_{{\mathrm{3}}}  \vdash_{\mathsf{MS} }  \mGLnt{B}$ with
      $\mathsf{CutRank} \, \mGLsym{(}  \Pi'  \mGLsym{)} \, \leq \,  \mathsf{Rank}  (  \mGLnt{X}  )$.
      Thus, we construct the following proof $\Pi$:

      \[
        \inferrule* [flushleft,right=$\mGLdruleMSTXXContName{}$,left=$\Pi :$] {
          \inferrule* [flushleft,left=$\Pi' : $] {
            \pi_3'
          }{\delta_{{\mathrm{1}}}  \mGLsym{,}  \mGLnt{r_{{\mathrm{1}}}}  \mGLsym{,}  \mGLnt{r_{{\mathrm{2}}}}  \mGLsym{,}  \delta_{{\mathrm{2}}}  \odot  \Delta_{{\mathrm{1}}}  \mGLsym{,}  \mGLnt{X}  \mGLsym{,}  \mGLnt{X}  \mGLsym{,}  \Delta_{{\mathrm{2}}}  \mGLsym{;}  \Gamma_{{\mathrm{1}}}  \mGLsym{,}  \Gamma_{{\mathrm{2}}}  \mGLsym{,}  \Gamma_{{\mathrm{3}}}  \vdash_{\mathsf{MS} }  \mGLnt{B}}
        }{\delta_{{\mathrm{1}}}  \mGLsym{,}  \mGLnt{r_{{\mathrm{1}}}}  +  \mGLnt{r_{{\mathrm{2}}}}  \mGLsym{,}  \delta_{{\mathrm{2}}}  \odot  \Delta_{{\mathrm{1}}}  \mGLsym{,}  \mGLnt{X}  \mGLsym{,}  \Delta_{{\mathrm{2}}}  \mGLsym{;}  \Gamma_{{\mathrm{1}}}  \mGLsym{,}  \Gamma_{{\mathrm{2}}}  \mGLsym{,}  \Gamma_{{\mathrm{3}}}  \vdash_{\mathsf{MS} }  \mGLnt{B}}
        \]
        with: $\mathsf{CutRank} \, \mGLsym{(}  \Pi  \mGLsym{)} \, \mGLsym{=} \, \mathsf{CutRank} \, \mGLsym{(}  \Pi'  \mGLsym{)} \, \leq \,  \mathsf{Rank}  (  \mGLnt{A}  )$
   \item \textbf{Graded exchange}
   \[
    \inferrule* [flushleft,right=,left=$\Pi_{{\mathrm{1}}} :$] {
      \pi_1
    }{\delta_{{\mathrm{2}}}  \odot  \Delta  \mGLsym{;}  \Gamma_{{\mathrm{2}}}  \vdash_{\mathsf{MS} }  \mGLnt{A}}
\]

      \[
        \inferrule* [flushleft,right=$\mGLdruleMSTXXGExName{}$,left=$\Pi_{{\mathrm{2}}} :$] {
          \inferrule* [flushleft,left=$\Pi_{{\mathrm{3}}} : $] {
            \pi_3
          }{\delta_{{\mathrm{1}}}  \mGLsym{,}  \mGLnt{r_{{\mathrm{2}}}}  \mGLsym{,}  \mGLnt{r_{{\mathrm{1}}}}  \mGLsym{,}  \delta_{{\mathrm{2}}}  \odot  \Delta_{{\mathrm{1}}}  \mGLsym{,}  \mGLnt{X}  \mGLsym{,}  \mGLnt{Y}  \mGLsym{,}  \Delta_{{\mathrm{2}}}  \mGLsym{;}  \Gamma_{{\mathrm{1}}}  \mGLsym{,}  \mGLnt{A}  \mGLsym{,}  \Gamma_{{\mathrm{3}}}  \vdash_{\mathsf{MS} }  \mGLnt{B}}
        }{ \delta_{{\mathrm{1}}}  \mGLsym{,}  \mGLnt{r_{{\mathrm{1}}}}  \mGLsym{,}  \mGLnt{r_{{\mathrm{2}}}}  \mGLsym{,}  \delta_{{\mathrm{2}}}  \odot  \Delta_{{\mathrm{1}}}  \mGLsym{,}  \mGLnt{Y}  \mGLsym{,}  \mGLnt{X}  \mGLsym{,}  \Delta_{{\mathrm{2}}}  \mGLsym{;}  \Gamma_{{\mathrm{1}}}  \mGLsym{,}  \mGLnt{A}  \mGLsym{,}  \Gamma_{{\mathrm{3}}}  \vdash_{\mathsf{MS} }  \mGLnt{B}}
        \]

      We know:
      \[
      \begin{array}{lll}
        \mathsf{Depth}  (  \Pi_{{\mathrm{1}}}  )   +   \mathsf{Depth}  (  \Pi_{{\mathrm{3}}}  )   \, \mGLsym{<} \,  \mathsf{Depth}  (  \Pi_{{\mathrm{1}}}  )   +   \mathsf{Depth}  (  \Pi_{{\mathrm{2}}}  )\\
        \mathsf{CutRank} \, \mGLsym{(}  \Pi_{{\mathrm{3}}}  \mGLsym{)} \, \leq \, \mathsf{CutRank} \, \mGLsym{(}  \Pi_{{\mathrm{2}}}  \mGLsym{)} \, \leq \,  \mathsf{Rank}  (  \mGLnt{A}  )
      \end{array}
      \]

      and so applying the induction hypothesis
      to $\Pi_{{\mathrm{1}}}$ and $\Pi_{{\mathrm{3}}}$
      implies that there is a proof $\Pi'$ of
      $\delta_{{\mathrm{1}}}  \mGLsym{,}  \mGLnt{r_{{\mathrm{2}}}}  \mGLsym{,}  \mGLnt{r_{{\mathrm{1}}}}  \mGLsym{,}  \delta_{{\mathrm{2}}}  \odot  \Delta_{{\mathrm{1}}}  \mGLsym{,}  \mGLnt{X}  \mGLsym{,}  \mGLnt{Y}  \mGLsym{,}  \Delta_{{\mathrm{2}}}  \mGLsym{;}  \Gamma_{{\mathrm{1}}}  \mGLsym{,}  \Gamma_{{\mathrm{2}}}  \mGLsym{,}  \Gamma_{{\mathrm{3}}}  \vdash_{\mathsf{MS} }  \mGLnt{B}$ with
      $\mathsf{CutRank} \, \mGLsym{(}  \Pi'  \mGLsym{)} \, \leq \,  \mathsf{Rank}  (  \mGLnt{X}  )$.
      Thus, we construct the following proof $\Pi$:

      \[
        \inferrule* [flushleft,right=$\mGLdruleMSTXXGExName{}$,left=$\Pi :$] {
          \inferrule* [flushleft,left=$\Pi' : $] {
            \pi_3'
          }{\delta_{{\mathrm{1}}}  \mGLsym{,}  \mGLnt{r_{{\mathrm{2}}}}  \mGLsym{,}  \mGLnt{r_{{\mathrm{1}}}}  \mGLsym{,}  \delta_{{\mathrm{2}}}  \odot  \Delta_{{\mathrm{1}}}  \mGLsym{,}  \mGLnt{X}  \mGLsym{,}  \mGLnt{Y}  \mGLsym{,}  \Delta_{{\mathrm{2}}}  \mGLsym{;}  \Gamma_{{\mathrm{1}}}  \mGLsym{,}  \Gamma_{{\mathrm{2}}}  \mGLsym{,}  \Gamma_{{\mathrm{3}}}  \vdash_{\mathsf{MS} }  \mGLnt{B}}
        }{\delta_{{\mathrm{1}}}  \mGLsym{,}  \mGLnt{r_{{\mathrm{1}}}}  \mGLsym{,}  \mGLnt{r_{{\mathrm{2}}}}  \mGLsym{,}  \delta_{{\mathrm{2}}}  \odot  \Delta_{{\mathrm{1}}}  \mGLsym{,}  \mGLnt{Y}  \mGLsym{,}  \mGLnt{X}  \mGLsym{,}  \Delta_{{\mathrm{2}}}  \mGLsym{;}  \Gamma_{{\mathrm{1}}}  \mGLsym{,}  \Gamma_{{\mathrm{2}}}  \mGLsym{,}  \Gamma_{{\mathrm{3}}}  \vdash_{\mathsf{MS} }  \mGLnt{B}}
        \]
        with: $\mathsf{CutRank} \, \mGLsym{(}  \Pi  \mGLsym{)} \, \mGLsym{=} \, \mathsf{CutRank} \, \mGLsym{(}  \Pi'  \mGLsym{)} \, \leq \,  \mathsf{Rank}  (  \mGLnt{A}  )$
        \item \textbf{Linear exchange}
        \[
          \inferrule* [flushleft,right=,left=$\Pi_{{\mathrm{1}}} :$] {
            \pi_1
          }{\delta_{{\mathrm{2}}}  \odot  \Delta_{{\mathrm{2}}}  \mGLsym{;}  \Gamma  \vdash_{\mathsf{MS} }  \mGLnt{A}}
      \]

            \[
              \inferrule* [flushleft,right=$\mGLdruleMSTXXExName{}$,left=$\Pi_{{\mathrm{2}}} :$] {
                \inferrule* [flushleft,left=$\Pi_{{\mathrm{3}}} : $] {
                  \pi_3
                }{\delta_{{\mathrm{1}}}  \odot  \Delta_{{\mathrm{1}}}  \mGLsym{;}  \Gamma_{{\mathrm{1}}}  \mGLsym{,}  \mGLnt{A}  \mGLsym{,}  \mGLnt{B}  \mGLsym{,}  \Gamma_{{\mathrm{3}}}  \vdash_{\mathsf{MS} }  \mGLnt{C}}
              }{ \delta_{{\mathrm{1}}}  \odot  \Delta_{{\mathrm{1}}}  \mGLsym{;}  \Gamma_{{\mathrm{1}}}  \mGLsym{,}  \mGLnt{B}  \mGLsym{,}  \mGLnt{A}  \mGLsym{,}  \Gamma_{{\mathrm{3}}}  \vdash_{\mathsf{MS} }  \mGLnt{C}}
              \]
         We know:
      \[
      \begin{array}{lll}
        \mathsf{Depth}  (  \Pi_{{\mathrm{1}}}  )   +   \mathsf{Depth}  (  \Pi_{{\mathrm{3}}}  )   \, \mGLsym{<} \,  \mathsf{Depth}  (  \Pi_{{\mathrm{1}}}  )   +   \mathsf{Depth}  (  \Pi_{{\mathrm{2}}}  )\\
        \mathsf{CutRank} \, \mGLsym{(}  \Pi_{{\mathrm{3}}}  \mGLsym{)} \, \leq \, \mathsf{CutRank} \, \mGLsym{(}  \Pi_{{\mathrm{2}}}  \mGLsym{)} \, \leq \,  \mathsf{Rank}  (  \mGLnt{A}  )
      \end{array}
      \]

      and so applying the induction hypothesis
      to $\Pi_{{\mathrm{1}}}$ and $\Pi_{{\mathrm{3}}}$
      implies that there is a proof $\Pi'$ of
      $\delta_{{\mathrm{1}}}  \odot  \Delta_{{\mathrm{1}}}  \mGLsym{;}  \Gamma_{{\mathrm{1}}}  \mGLsym{,}  \Gamma_{{\mathrm{2}}}  \mGLsym{,}  \mGLnt{B}  \mGLsym{,}  \Gamma_{{\mathrm{3}}}  \vdash_{\mathsf{MS} }  \mGLnt{C}$ with
      $\mathsf{CutRank} \, \mGLsym{(}  \Pi'  \mGLsym{)} \, \leq \,  \mathsf{Rank}  (  \mGLnt{X}  )$.
      Thus, we construct the following proof $\Pi$:
      \[
              \inferrule* [flushleft,right=$\mGLdruleMSTXXExName{}$,left=$\Pi :$] {
                \inferrule* [flushleft,left=$\Pi' : $] {
                  \pi'
                }{\delta_{{\mathrm{1}}}  \odot  \Delta_{{\mathrm{1}}}  \mGLsym{;}  \Gamma_{{\mathrm{1}}}  \mGLsym{,}  \Gamma_{{\mathrm{2}}}  \mGLsym{,}  \mGLnt{B}  \mGLsym{,}  \Gamma_{{\mathrm{3}}}  \vdash_{\mathsf{MS} }  \mGLnt{C}}
              }{ \delta_{{\mathrm{1}}}  \odot  \Delta_{{\mathrm{1}}}  \mGLsym{;}  \Gamma_{{\mathrm{1}}}  \mGLsym{,}  \mGLnt{B}  \mGLsym{,}  \Gamma_{{\mathrm{2}}}  \mGLsym{,}  \Gamma_{{\mathrm{3}}}  \vdash_{\mathsf{MS} }  \mGLnt{C}}
              \]
              with: $\mathsf{CutRank} \, \mGLsym{(}  \Pi  \mGLsym{)} \, \mGLsym{=} \, \mathsf{CutRank} \, \mGLsym{(}  \Pi'  \mGLsym{)} \, \leq \,  \mathsf{Rank}  (  \mGLnt{A}  )$

        \item \textbf{Linear exchange (second case)}
        \[
          \inferrule* [flushleft,right=,left=$\Pi_{{\mathrm{1}}} :$] {
            \pi_1
          }{\delta_{{\mathrm{2}}}  \odot  \Delta_{{\mathrm{2}}}  \mGLsym{;}  \Gamma  \vdash_{\mathsf{MS} }  \mGLnt{A}}
      \]

            \[
              \inferrule* [flushleft,right=$\mGLdruleMSTXXExName{}$,left=$\Pi_{{\mathrm{2}}} :$] {
                \inferrule* [flushleft,left=$\Pi_{{\mathrm{3}}} : $] {
                  \pi_3
                }{\delta_{{\mathrm{1}}}  \odot  \Delta_{{\mathrm{1}}}  \mGLsym{;}  \Gamma_{{\mathrm{1}}}  \mGLsym{,}  \mGLnt{A}  \mGLsym{,}  \Gamma_{{\mathrm{3}}}  \mGLsym{,}  \mGLnt{C}  \mGLsym{,}  \mGLnt{B}  \mGLsym{,}  \Gamma_{{\mathrm{4}}}  \vdash_{\mathsf{MS} }  \mGLnt{D}}
              }{ \delta_{{\mathrm{1}}}  \odot  \Delta_{{\mathrm{1}}}  \mGLsym{;}  \Gamma_{{\mathrm{1}}}  \mGLsym{,}  \mGLnt{A}  \mGLsym{,}  \Gamma_{{\mathrm{3}}}  \mGLsym{,}  \mGLnt{B}  \mGLsym{,}  \mGLnt{C}  \mGLsym{,}  \Gamma_{{\mathrm{4}}}  \vdash_{\mathsf{MS} }  \mGLnt{B}}
              \]
 We know:
      \[
      \begin{array}{lll}
        \mathsf{Depth}  (  \Pi_{{\mathrm{1}}}  )   +   \mathsf{Depth}  (  \Pi_{{\mathrm{3}}}  )   \, \mGLsym{<} \,  \mathsf{Depth}  (  \Pi_{{\mathrm{1}}}  )   +   \mathsf{Depth}  (  \Pi_{{\mathrm{2}}}  )\\
        \mathsf{CutRank} \, \mGLsym{(}  \Pi_{{\mathrm{3}}}  \mGLsym{)} \, \leq \, \mathsf{CutRank} \, \mGLsym{(}  \Pi_{{\mathrm{2}}}  \mGLsym{)} \, \leq \,  \mathsf{Rank}  (  \mGLnt{A}  )
      \end{array}
      \]

      and so applying the induction hypothesis
      to $\Pi_{{\mathrm{1}}}$ and $\Pi_{{\mathrm{3}}}$
      implies that there is a proof $\Pi'$ of
      $\delta_{{\mathrm{1}}}  \odot  \Delta_{{\mathrm{1}}}  \mGLsym{;}  \Gamma_{{\mathrm{1}}}  \mGLsym{,}  \Gamma_{{\mathrm{2}}}  \mGLsym{,}  \Gamma_{{\mathrm{3}}}  \mGLsym{,}  \mGLnt{C}  \mGLsym{,}  \mGLnt{B}  \mGLsym{,}  \Gamma_{{\mathrm{4}}}  \vdash_{\mathsf{MS} }  \mGLnt{D}$ with
      $\mathsf{CutRank} \, \mGLsym{(}  \Pi'  \mGLsym{)} \, \leq \,  \mathsf{Rank}  (  \mGLnt{X}  )$.
      Thus, we construct the following proof $\Pi$:
      \[
              \inferrule* [flushleft,right=$\mGLdruleMSTXXExName{}$,left=$\Pi :$] {
                \inferrule* [flushleft,left=$\Pi' : $] {
                  \pi'
                }{\delta_{{\mathrm{1}}}  \odot  \Delta_{{\mathrm{1}}}  \mGLsym{;}  \Gamma_{{\mathrm{1}}}  \mGLsym{,}  \Gamma_{{\mathrm{2}}}  \mGLsym{,}  \Gamma_{{\mathrm{3}}}  \mGLsym{,}  \mGLnt{C}  \mGLsym{,}  \mGLnt{B}  \mGLsym{,}  \Gamma_{{\mathrm{4}}}  \vdash_{\mathsf{MS} }  \mGLnt{D}}
              }{ \delta_{{\mathrm{1}}}  \odot  \Delta_{{\mathrm{1}}}  \mGLsym{;}  \Gamma_{{\mathrm{1}}}  \mGLsym{,}  \Gamma_{{\mathrm{2}}}  \mGLsym{,}  \Gamma_{{\mathrm{3}}}  \mGLsym{,}  \mGLnt{B}  \mGLsym{,}  \mGLnt{C}  \mGLsym{,}  \Gamma_{{\mathrm{4}}}  \vdash_{\mathsf{MS} }  \mGLnt{D}}
              \]
              with: $\mathsf{CutRank} \, \mGLsym{(}  \Pi  \mGLsym{)} \, \mGLsym{=} \, \mathsf{CutRank} \, \mGLsym{(}  \Pi'  \mGLsym{)} \, \leq \,  \mathsf{Rank}  (  \mGLnt{A}  )$

        \item \textbf{Linear exchange (third case )}
        \[
          \inferrule* [flushleft,right=,left=$\Pi_{{\mathrm{1}}} :$] {
            \pi_1
          }{\delta_{{\mathrm{2}}}  \odot  \Delta_{{\mathrm{2}}}  \mGLsym{;}  \Gamma  \vdash_{\mathsf{MS} }  \mGLnt{A}}
      \]

            \[
              \inferrule* [flushleft,right=$\mGLdruleMSTXXExName{}$,left=$\Pi_{{\mathrm{2}}} :$] {
                \inferrule* [flushleft,left=$\Pi_{{\mathrm{3}}} : $] {
                  \pi_3
                }{\delta_{{\mathrm{1}}}  \odot  \Delta_{{\mathrm{1}}}  \mGLsym{;}  \Gamma_{{\mathrm{1}}}  \mGLsym{,}  \mGLnt{C}  \mGLsym{,}  \mGLnt{B}  \mGLsym{,}  \Gamma_{{\mathrm{2}}}  \mGLsym{,}  \mGLnt{A}  \mGLsym{,}  \Gamma_{{\mathrm{4}}}  \vdash_{\mathsf{MS} }  \mGLnt{D}}
              }{ \delta_{{\mathrm{1}}}  \odot  \Delta_{{\mathrm{1}}}  \mGLsym{;}  \Gamma_{{\mathrm{1}}}  \mGLsym{,}  \mGLnt{B}  \mGLsym{,}  \mGLnt{C}  \mGLsym{,}  \Gamma_{{\mathrm{2}}}  \mGLsym{,}  \mGLnt{A}  \mGLsym{,}  \Gamma_{{\mathrm{4}}}  \vdash_{\mathsf{MS} }  \mGLnt{B}}
              \]
      Similar to the previous case.
  \end{enumerate}
\end{enumerate}

\end{proof}

A notable example:

\begin{gather*}
  \begin{array}{lll}
    \inferrule* [flushleft,right=$\mGLdruleGSTXXLinRName{}$,left=$\Pi_{{\mathrm{1}}} :$] {
      \inferrule* [flushleft,right=$\mGLdruleMSTXXGrdRName{}$, left=] {
        \inferrule* [flushleft,right=$\mGLdruleGSTXXLinRName{}$, left=] {
          \inferrule* [flushleft,right=$\mGLdruleMSTXXLinLName{}$, left=] {
          \inferrule* [flushleft,right=$\mGLdruleMSTXXidName{}$, left=] {
          \
          }{\emptyset  \odot  \emptyset  \mGLsym{;}  \mGLnt{A}  \vdash_{\mathsf{MS} }  \mGLnt{A}}
          }{1  \odot  \mathsf{Lin} \, \mGLnt{A}  \mGLsym{;}  \emptyset  \vdash_{\mathsf{MS} }  \mGLnt{A}}
        }{1  \odot  \mathsf{Lin} \, \mGLnt{A}  \vdash_{\mathsf{GS} }  \mathsf{Lin} \, \mGLnt{A}}
    }{1  *  1  \odot  \mathsf{Lin} \, \mGLnt{A}  \mGLsym{;}  \emptyset  \vdash_{\mathsf{MS} }   \mathsf{Grd} _{ 1 }\, \mathsf{Lin} \, \mGLnt{A}}
    }{1  \odot  \mathsf{Lin} \, \mGLnt{A}  \vdash_{\mathsf{GS} }  \mathsf{Lin} \,  \mathsf{Grd} _{ 1 }\, \mathsf{Lin} \, \mGLnt{A}}
    & \quad &
    \inferrule* [flushleft,right=$\mGLdruleGSTXXLinRName{}$,left=$\Pi_{{\mathrm{0}}} :$] {
     \inferrule* [flushleft,right=$\mGLdruleMSTXXLinLName{}$, left=$\Pi_{{\mathrm{2}}} :$] {
       \inferrule* [flushleft,right=$\mGLdruleMSTXXGrdLName{}$, left=] {
        \inferrule* [flushleft,right=$\mGLdruleMSTXXLinLName{}$, left=] {
         \inferrule* [flushleft,right=$\mGLdruleMSTXXidName{}$, left=] {
          \
          }{\emptyset  \odot  \emptyset  \mGLsym{;}  \mGLnt{A}  \vdash_{\mathsf{MS} }  \mGLnt{A}}
          }{1  \odot  \mathsf{Lin} \, \mGLnt{A}  \mGLsym{;}  \emptyset  \vdash_{\mathsf{MS} }  \mGLnt{A}}
        }{\emptyset  \odot  \emptyset  \mGLsym{;}   \mathsf{Grd} _{ 1 }\, \mathsf{Lin} \, \mGLnt{A}   \vdash_{\mathsf{MS} }  \mGLnt{A}}
      }{1  \odot  \mathsf{Lin} \,  \mathsf{Grd} _{ 1 }\, \mathsf{Lin} \, \mGLnt{A}   \mGLsym{;}  \emptyset  \vdash_{\mathsf{MS} }  \mGLnt{A}}
    }{1  \odot  \mathsf{Lin} \,  \mathsf{Grd} _{ 1 }\, \mathsf{Lin} \, \mGLnt{A}   \vdash_{\mathsf{GS} }  \mathsf{Lin} \, \mGLnt{A}}
  \end{array}
  \end{gather*}

$ \Pi_{{\mathrm{1}}}$ and $\Pi_{{\mathrm{0}}}$ would be handled by a secondary hypothesis case for the
right introduction of $\mathsf{Lin}$. The inductive step of the reduction would be on
$ \Pi_{{\mathrm{1}}}$ and $\Pi_{{\mathrm{2}}}$. The right and left rules for the modal operators line up,
resulting in three principal versus principal cases. Finally the reduction bottoms
out at an axiom case. Resulting in the following proof:

\[
  \inferrule* [flushleft,right=$\mGLdruleGSTXXLinRName{}$, left=$\Pi' :$] {
    \inferrule* [flushleft,right=$\mGLdruleMSTXXLinLName{}$, left=] {
      \inferrule* [flushleft,right=$\mGLdruleMSTXXidName{}$, left=] {
      \
    }{\emptyset  \odot  \emptyset  \mGLsym{;}  \mGLnt{A}  \vdash_{\mathsf{MS} }  \mGLnt{A}}
    }{1  \odot  \mathsf{Lin} \, \mGLnt{A}  \mGLsym{;}  \emptyset  \vdash_{\mathsf{MS} }  \mGLnt{A}}
  }{1  \odot  \mathsf{Lin} \, \mGLnt{A}  \vdash_{\mathsf{GS} }  \mathsf{Lin} \, \mGLnt{A}}
  \]

%% file: decreasing-mGL-ottput.tex
\decreasingMGL*

Decreasing order is done termless for simplicity.
\begin{proof}
  By induction on $\mathsf{Depth}  (  \Pi  )$. For convenience we will only consider $\Pi$ a proof of $\delta  \odot  \Delta  \vdash_{\mathsf{GS} }  \mGLnt{Y}$. 
  If the last inference in $\Pi$ is not a cut, then we simply apply
  the induction hypothesis. Thus, suppose the last inference in $\Pi$ is a cut on a formula $X$. If
  $\mathsf{CutRank} \, \mGLsym{(}  \Pi  \mGLsym{)} \, \mGLsym{>} \,  \mathsf{Rank}  (  \mGLnt{X}  )   + 1$, then $X$ is not the main cut formula in $\Pi$, and hence, we can just
  apply the induction hypothesis to the premises of cut. This leaves the case where $\mathsf{CutRank} \, \mGLsym{(}  \Pi  \mGLsym{)} \, \mGLsym{=} \,  \mathsf{Rank}  (  \mGLnt{X}  )   + 1 $, and so $X$ is the main cut formula in $\Pi$. This implies $\Pi$ is of the following form: \\

\[
      \inferrule* [flushleft,right=GS-MCut,left=$\Pi :$] {        
        \inferrule* [flushleft,left=$\Pi_{{\mathrm{1}}} : $] {
          \pi_1
        }{\delta_{{\mathrm{2}}}  \odot  \Delta_{{\mathrm{2}}}  \vdash_{\mathsf{GS} }  \mGLnt{X}}\\ 
        \inferrule* [flushleft,right=,left=$\Pi_{{\mathrm{2}}} :$] {
          \pi_2
        }{(  \delta_{{\mathrm{1}}}  \mGLsym{,}  \delta  \mGLsym{,}  \delta_{{\mathrm{3}}}  )   \odot   ( \Delta_{{\mathrm{1}}}  \mGLsym{,}   \mGLnt{X} ^{ \mGLmv{n} }   \mGLsym{,}  \Delta_{{\mathrm{3}}} )   \vdash_{\mathsf{GS} }  \mGLnt{Y}} 
      }{(  \delta_{{\mathrm{1}}}  \mGLsym{,}  \delta'_{{\mathrm{2}}}  \mGLsym{,}  \delta_{{\mathrm{3}}}  )   \odot   ( \Delta_{{\mathrm{1}}}  \mGLsym{,}  \Delta_{{\mathrm{2}}}  \mGLsym{,}  \Delta_{{\mathrm{3}}} )   \vdash_{\mathsf{GS} }  \mGLnt{Y}} 
      \]

  We assumed above that $\mathsf{CutRank} \, \mGLsym{(}  \Pi  \mGLsym{)} \, \mGLsym{=} \,  \mathsf{Rank}  (  \mGLnt{X}  )   + 1$, and hence, these imply:
  \[
  \begin{array}{lll}
   \mathsf{Depth}  (  \Pi  )  \, \mGLsym{=} \,  \mathsf{Depth}  (  \Pi_{{\mathrm{1}}}  )   +   \mathsf{Depth}  (  \Pi_{{\mathrm{2}}}  )    + 1 \\
   \mathsf{Depth}  (  \Pi_{{\mathrm{1}}}  )  \, \mGLsym{<} \,  \mathsf{Depth}  (  \Pi  ) \\
   \mathsf{Depth}  (  \Pi_{{\mathrm{2}}}  )  \, \mGLsym{<} \,  \mathsf{Depth}  (  \Pi  ) \\
   \mathsf{CutRank} \, \mGLsym{(}  \Pi  \mGLsym{)}  = Max ( \mathsf{CutRank} \, \mGLsym{(}  \Pi_{{\mathrm{1}}}  \mGLsym{)} , \mathsf{CutRank} \, \mGLsym{(}  \Pi_{{\mathrm{2}}}  \mGLsym{)} , \mathsf{Rank}  (  \mGLnt{X}  )   + 1) \\
   \mathsf{CutRank} \, \mGLsym{(}  \Pi_{{\mathrm{1}}}  \mGLsym{)} \, \leq \,  \mathsf{Rank}  (  \mGLnt{X}  )   + 1 \\
   \mathsf{CutRank} \, \mGLsym{(}  \Pi_{{\mathrm{2}}}  \mGLsym{)} \, \leq \,  \mathsf{Rank}  (  \mGLnt{X}  )   + 1 
\end{array}
\]

\noindent
  Now without loss of generality we assume neither $\Pi_{{\mathrm{1}}}$ nor $\Pi_{{\mathrm{2}}}$ are axioms which follow 
  similarly. By the induction hypothesis we may conclude that there are proofs $\Pi'_{{\mathrm{1}}}$
  and $\Pi'_{{\mathrm{2}}}$ of the same sequents of $\Pi_{{\mathrm{1}}}$ and $\Pi_{{\mathrm{2}}}$, but with 
  $\mathsf{CutRank} \, \mGLsym{(}  \Pi'_{{\mathrm{1}}}  \mGLsym{)} \, \leq \,  \mathsf{Rank}  (  \mGLnt{X}  ) $ and $\mathsf{CutRank} \, \mGLsym{(}  \Pi'_{{\mathrm{2}}}  \mGLsym{)} \, \leq \,  \mathsf{Rank}  (  \mGLnt{X}  ) $. 
  Therefore, by applying the cut reduction lemma to $\Pi'_{{\mathrm{1}}}$ and $\Pi'_{{\mathrm{2}}}$ we obtain a proof $\Pi'$ 
  of the sequent $(  \delta_{{\mathrm{1}}}  \mGLsym{,}  \delta'_{{\mathrm{2}}}  \mGLsym{,}  \delta_{{\mathrm{3}}}  )   \odot   ( \Delta_{{\mathrm{1}}}  \mGLsym{,}  \Delta_{{\mathrm{2}}}  \mGLsym{,}  \Delta_{{\mathrm{3}}} )   \vdash_{\mathsf{GS} }  \mGLnt{Y}$ with $\mathsf{CutRank} \, \mGLsym{(}  \Pi'  \mGLsym{)} \, \leq \,  \mathsf{Rank}  (  \mGLnt{X}  )  \, \mGLsym{<} \,  \mathsf{Rank}  (  \mGLnt{X}  )   + 1  \, \mGLsym{=} \, \mathsf{CutRank} \, \mGLsym{(}  \Pi  \mGLsym{)} $.
  
\end{proof}

%% file: subformula-ottput.tex
Subformula definitions and property proof are done termless for simplicity. 
\begin{definition}
    We will define the function $Sf$ from formulas to sets of formulas as follows 
    \[
      \begin{array}{lll}
       Sf(\phi) & = & \{\phi \} \text{when } \phi \text{ is atomic} \\
       Sf(\mGLnt{X}  \boxtimes  \mGLnt{Y}) & = & \bigcup \{Sf(X), Sf(Y), \{\mGLnt{X}  \boxtimes  \mGLnt{Y} \}\} \\
       Sf(\mGLnt{A}  \otimes  \mGLnt{B}) & = & \bigcup \{Sf(A), Sf(B), \{\mGLnt{A}  \otimes  \mGLnt{B}\}\} \\
       Sf(\mGLnt{A}  \multimap  \mGLnt{B}) & = & \bigcup \{Sf(A), Sf(B), \{\mGLnt{A}  \multimap  \mGLnt{B}\}\} \\
       Sf(\mathsf{Lin} \, \mGLnt{A}) & = & Sf(A) \cup \{\mathsf{Lin} \, \mGLnt{A}\} \\
       Sf(\mathsf{Grd} _{ \mGLnt{r} }\, \mGLnt{X}) & = & Sf(X) \cup \{\mathsf{Grd} _{ \mGLnt{r} }\, \mGLnt{X}\}
        \end{array}
    \]
\end{definition}
Note: if $\phi \in Sf(\psi)$ then $Sf(\phi) \subseteq Sf(\psi)$
\begin{definition}
    We will define $\mathcal{F}(S)$, where $S$ is a well formed 
    judgement, as the set of formulas that appear in the judgement  $S$. 
    For a judgement $\delta  \odot  \Delta  \vdash_{\mathsf{GS} }  \mGLnt{X}$, $\mathcal{F}(\delta  \odot  \Delta  \vdash_{\mathsf{GS} }  \mGLnt{X})= \{X\} \cup \Delta $ 
    where $\Delta$ is viewed as a set. \\
    Similarly, we will define $\mathcal{F}[\Pi]$, where $\Pi$ is a proof of a judgement  
    $S$ as $\mathcal{F}[\Pi] = \mathcal{F}(S) \cup ( \bigcup \mathcal{F}[\Pi_{i}] )$ where 
    $\Pi_i$'s are the premises of $\Pi$
\end{definition}
\subformula*
\noindent
With the above definitions we can now restate the subformula property formally as follows: \\
\noindent
For any judgement $\delta  \odot  \Delta  \vdash_{\mathsf{GS} }  \mGLnt{X}$ and cut-free proof $\Pi$ of $\delta  \odot  \Delta  \vdash_{\mathsf{GS} }  \mGLnt{X}$, \\
  $$\mathcal{F}[\Pi] \subseteq \bigcup_{\phi \in \mathcal{F}(\delta  \odot  \Delta  \vdash_{\mathsf{GS} }  \mGLnt{X})} Sf(\phi)$$
  \noindent
For any judgement $\delta  \odot  \Delta  \mGLsym{;}  \Gamma  \vdash_{\mathsf{MS} }  \mGLnt{A}$ and cut-free proof $\Pi$ of $\delta  \odot  \Delta  \mGLsym{;}  \Gamma  \vdash_{\mathsf{MS} }  \mGLnt{A}$, \\
  $$\mathcal{F}[\Pi] \subseteq \bigcup_{\phi \in \mathcal{F}(\delta  \odot  \Delta  \mGLsym{;}  \Gamma  \vdash_{\mathsf{MS} }  \mGLnt{A})} Sf(\phi)$$

Proof is by mutual induction on derivations.
\begin{proof}
    The axiom cases vacuously true. Mutual induction is used in the right 
    introductions for the Grd and Lin constructors. All the inductive cases 
    follow a similar pattern. The proof for GS-TenL is provided as an example. 
    \[
     \inferrule* [flushleft,right=, left=$\Pi_{{\mathrm{1}}} :$] {
      \inferrule* [flushleft,right=, left=$\Pi_{{\mathrm{2}}} :$] {
       \pi_2
       }{(  \delta_{{\mathrm{1}}}  \mGLsym{,}  \mGLnt{r}  \mGLsym{,}  \mGLnt{r}  \mGLsym{,}  \delta_{{\mathrm{2}}}  )   \odot   ( \Delta_{{\mathrm{1}}}  \mGLsym{,}  \mGLnt{X}  \mGLsym{,}  \mGLnt{Y}  \mGLsym{,}  \Delta_{{\mathrm{2}}} )   \vdash_{\mathsf{GS} }  \mGLnt{Z}}
     }{(  \delta_{{\mathrm{1}}}  \mGLsym{,}  \mGLnt{r}  \mGLsym{,}  \delta_{{\mathrm{2}}}  )   \odot   ( \Delta_{{\mathrm{1}}}  \mGLsym{,}  \mGLnt{X}  \boxtimes  \mGLnt{Y}  \mGLsym{,}  \Delta_{{\mathrm{2}}} )   \vdash_{\mathsf{GS} }  \mGLnt{Z}}
    \]
    By assumption we know 
    $$\mathcal{F}[ \Pi_{{\mathrm{2}}} ] \subseteq 
    \bigcup  \{ Sf(\phi) \mid \phi \in \mathcal{F}((  \delta_{{\mathrm{1}}}  \mGLsym{,}  \mGLnt{r}  \mGLsym{,}  \mGLnt{r}  \mGLsym{,}  \delta_{{\mathrm{2}}}  )   \odot   ( \Delta_{{\mathrm{1}}}  \mGLsym{,}  \mGLnt{X}  \mGLsym{,}  \mGLnt{Y}  \mGLsym{,}  \Delta_{{\mathrm{2}}} )   \vdash_{\mathsf{GS} }  \mGLnt{Z})  \} $$
    We also know 
    \[
      \begin{array}{lll}
      \bigcup  \{ Sf(\phi) \mid \phi \in \mathcal{F}((  \delta_{{\mathrm{1}}}  \mGLsym{,}  \mGLnt{r}  \mGLsym{,}  \delta_{{\mathrm{2}}}  )   \odot   ( \Delta_{{\mathrm{1}}}  \mGLsym{,}  \mGLnt{X}  \boxtimes  \mGLnt{Y}  \mGLsym{,}  \Delta_{{\mathrm{2}}} )   \vdash_{\mathsf{GS} }  \mGLnt{Z}) \}\\
       = \left( \bigcup  \{ Sf(\phi) \mid \phi \in \{ \delta_1, r, \delta_2 , \Delta_1,  \Delta_2,  Z)  \} \right)
     \cup ( \bigcup \{ Sf(X) , Sf(Y), \mGLnt{X}  \boxtimes  \mGLnt{Y} \}) \\
     = \left( \bigcup  \{ Sf(\phi) \mid \phi \in \mathcal{F}((  \delta_{{\mathrm{1}}}  \mGLsym{,}  \mGLnt{r}  \mGLsym{,}  \mGLnt{r}  \mGLsym{,}  \delta_{{\mathrm{2}}}  )   \odot   ( \Delta_{{\mathrm{1}}}  \mGLsym{,}  \mGLnt{X}  \mGLsym{,}  \mGLnt{Y}  \mGLsym{,}  \Delta_{{\mathrm{2}}} )   \vdash_{\mathsf{GS} }  \mGLnt{Z})  \} \right) 
     \cup \{ \mGLnt{X}  \boxtimes  \mGLnt{Y} \}
    
    \end{array}
    \]
    So 
    \[ 
      \mathcal{F}[ \Pi_{{\mathrm{2}}} ] \subseteq   \bigcup  \{ Sf(\phi)  \mid \phi \in \mathcal{F}((  \delta_{{\mathrm{1}}}  \mGLsym{,}  \mGLnt{r}  \mGLsym{,}  \delta_{{\mathrm{2}}}  )   \odot   ( \Delta_{{\mathrm{1}}}  \mGLsym{,}  \mGLnt{X}  \boxtimes  \mGLnt{Y}  \mGLsym{,}  \Delta_{{\mathrm{2}}} )   \vdash_{\mathsf{GS} }  \mGLnt{Z}) \}
      \]
    from the definition of $\mathcal{F}[ \Pi ] $ 
    \[
      \mathcal{F}[ \Pi_{{\mathrm{1}}} ] = \mathcal{F}[ \Pi_{{\mathrm{2}}} ] \cup \mathcal{F}( (  \delta_{{\mathrm{1}}}  \mGLsym{,}  \mGLnt{r}  \mGLsym{,}  \delta_{{\mathrm{2}}}  )   \odot   ( \Delta_{{\mathrm{1}}}  \mGLsym{,}  \mGLnt{X}  \boxtimes  \mGLnt{Y}  \mGLsym{,}  \Delta_{{\mathrm{2}}} )   \vdash_{\mathsf{GS} }  \mGLnt{Z} )
      \]
    from the definition of $ Sf(\phi) $ 
    \[
      \mathcal{F}( (  \delta_{{\mathrm{1}}}  \mGLsym{,}  \mGLnt{r}  \mGLsym{,}  \delta_{{\mathrm{2}}}  )   \odot   ( \Delta_{{\mathrm{1}}}  \mGLsym{,}  \mGLnt{X}  \boxtimes  \mGLnt{Y}  \mGLsym{,}  \Delta_{{\mathrm{2}}} )   \vdash_{\mathsf{GS} }  \mGLnt{Z} )  \subseteq 
    \bigcup \{ Sf(\phi) \mid \phi \in \mathcal{F}((  \delta_{{\mathrm{1}}}  \mGLsym{,}  \mGLnt{r}  \mGLsym{,}  \delta_{{\mathrm{2}}}  )   \odot   ( \Delta_{{\mathrm{1}}}  \mGLsym{,}  \mGLnt{X}  \boxtimes  \mGLnt{Y}  \mGLsym{,}  \Delta_{{\mathrm{2}}} )   \vdash_{\mathsf{GS} }  \mGLnt{Z})\} 
    \]
    So we know
    \[
      \mathcal{F}[ \Pi_{{\mathrm{2}}} ] \cup \mathcal{F}( (  \delta_{{\mathrm{1}}}  \mGLsym{,}  \mGLnt{r}  \mGLsym{,}  \delta_{{\mathrm{2}}}  )   \odot   ( \Delta_{{\mathrm{1}}}  \mGLsym{,}  \mGLnt{X}  \boxtimes  \mGLnt{Y}  \mGLsym{,}  \Delta_{{\mathrm{2}}} )   \vdash_{\mathsf{GS} }  \mGLnt{Z} ) \subseteq 
    \bigcup \{ Sf(\phi)  \mid \phi \in \mathcal{F}((  \delta_{{\mathrm{1}}}  \mGLsym{,}  \mGLnt{r}  \mGLsym{,}  \delta_{{\mathrm{2}}}  )   \odot   ( \Delta_{{\mathrm{1}}}  \mGLsym{,}  \mGLnt{X}  \boxtimes  \mGLnt{Y}  \mGLsym{,}  \Delta_{{\mathrm{2}}} )   \vdash_{\mathsf{GS} }  \mGLnt{Z})\} 
    \]
    That is, 
    \[
      \mathcal{F}[ \Pi_{{\mathrm{1}}} ] \subseteq \bigcup \{Sf(\phi) \mid \phi \in \mathcal{F}((  \delta_{{\mathrm{1}}}  \mGLsym{,}  \mGLnt{r}  \mGLsym{,}  \delta_{{\mathrm{2}}}  )   \odot   ( \Delta_{{\mathrm{1}}}  \mGLsym{,}  \mGLnt{X}  \boxtimes  \mGLnt{Y}  \mGLsym{,}  \Delta_{{\mathrm{2}}} )   \vdash_{\mathsf{GS} }  \mGLnt{Z})\} 
      \]
\end{proof} 

%% file: mGL-interpretation-ottput.tex
\begin{definition}[Interpretation of \mGLL{} Logic.]
  \label{def:full-mgl-intepretation}
  %
  %
  Given a \mGLL{} model with symmetric monoidal
  adjunction $\cat{C} : \func{Mny} \dashv \func{Lin} :\cat{M}$
  and strict exponential action
  $\odot : \op{\mathcal{R}} \times \mathcal{C} \mto \mathcal{C}$,
  we interpret $\mGL{}$ proofs via two mutually
  defined interpretations $\interp{-}^{\GS}$ and $\interp{-}^{\MS{}}$
  on types and proofs (derivations):
  \begin{itemize}
  \item For every type $X$ there is an object $\interp{ \mGLnt{X} }^{\GS} \in \cat{C}$ and
    for every type $A$ there is an object $\interp{ \mGLnt{A} }^{\MS} \in \cat{M}$, mutually defined:
    \begin{align*}
      \begin{array}{rl}
      \interp{ \mathsf{J} }^{\GS} & = \mathsf{J} \\
      \interp{ \mGLnt{X}  \boxtimes  \mGLnt{Y} }^{\GS} & = \interp{ \mGLnt{X} }^{\GS} \boxtimes \interp{ \mGLnt{Y} }^{\GS} \\
      \interp{ \mathsf{Lin} \, \mGLnt{A} }^{\GS} & = \Lin \interp{ \mGLnt{A} }^\MS \\
      \quad & \quad
      \end{array}
      \quad & \quad
      \begin{array}{rl}
      \interp{ \mathsf{I} }^\MS & = I \\
      \interp{ \mGLnt{A}  \otimes  \mGLnt{B} }^\MS & = \interp{ \mGLnt{A} }^\MS \otimes \interp{ \mGLnt{B} }^\MS \\
      \interp{ \mGLnt{A}  \multimap  \mGLnt{B} }^\MS & = \interp{ \mGLnt{A} }^\MS \multimap \interp{ \mGLnt{B} }^\MS \\
      \interp{  \mathsf{Grd} _{ \mGLnt{r} }\, \mGLnt{X}  }^\MS & = \Mny (r \odot \interp{ \mGLnt{X} }^{\GS} )
      \end{array}
    \end{align*}

    \item $\GS{}$ contexts are interpetered as objects composed of tensors, e.g.,
        for the denotation of a derivation $\Pi$ of
        $(  \mGLnt{r_{{\mathrm{1}}}}  \mGLsym{,} \, ... \, \mGLsym{,}  \mGLnt{r_{\mGLmv{n}}}  )   \odot   ( \mGLmv{x_{{\mathrm{1}}}}  \mGLsym{:}  \mGLnt{X_{{\mathrm{1}}}}  \mGLsym{,} \, ... \, \mGLsym{,}  \mGLmv{x_{\mGLmv{n}}}  \mGLsym{:}  \mGLnt{X_{\mGLmv{n}}} )   \vdash_{\mathsf{GS} }  \mGLnt{t_{{\mathrm{1}}}}  \mGLsym{:}  \mGLnt{X}$
        its source is given by:
        \begin{align*}
          \emptyset^{\GS} & = \mathsf{J} \\
          \interp{ \mGLmv{x_{{\mathrm{1}}}}  \mGLsym{:}  \mGLnt{X_{{\mathrm{1}}}}  \mGLsym{,} \, ... \, \mGLsym{,}  \mGLmv{x_{\mGLmv{n}}}  \mGLsym{:}  \mGLnt{X_{\mGLmv{n}}} }_{  (  \mGLnt{r_{{\mathrm{1}}}}  \mGLsym{,} \, ... \, \mGLsym{,}  \mGLnt{r_{\mGLmv{n}}}  )  }^{\GS} & =
          (\mGLnt{r_{{\mathrm{1}}}}  \odot   \interp{ \mGLnt{X_{{\mathrm{1}}}} }^{\mathsf{GS} }) \boxtimes \ldots \boxtimes (\mGLnt{r_{\mGLmv{n}}}  \odot   \interp{ \mGLnt{X_{\mGLmv{n}}} }^{\mathsf{GS} })
        \end{align*}

    \item $\MS{}$ contexts are interpetered as object composed of tensors, e.g.,
    for the denotation of a derivation of
    $(  \mGLnt{r_{{\mathrm{1}}}}  \mGLsym{,} \, ... \, \mGLsym{,}  \mGLnt{r_{\mGLmv{n}}}  )   \odot   ( \mGLmv{x_{{\mathrm{1}}}}  \mGLsym{:}  \mGLnt{X_{{\mathrm{1}}}}  \mGLsym{,} \, ... \, \mGLsym{,}  \mGLmv{x_{\mGLmv{n}}}  \mGLsym{:}  \mGLnt{X_{\mGLmv{n}}} )   \mGLsym{;}  \mGLmv{y_{{\mathrm{1}}}}  \mGLsym{:}  \mGLnt{A_{{\mathrm{1}}}}  \mGLsym{,} \, ... \, \mGLsym{,}  \mGLmv{y_{\mGLmv{m}}}  \mGLsym{:}  \mGLnt{A_{\mGLmv{m}}}  \vdash_{\mathsf{MS} }  \mGLnt{l}  \mGLsym{:}  \mGLnt{B}$
    its source is given by:
    \[
    \begin{array}{lll}
      \interp{  ( \mGLmv{x_{{\mathrm{1}}}}  \mGLsym{:}  \mGLnt{X_{{\mathrm{1}}}}  \mGLsym{,} \, ... \, \mGLsym{,}  \mGLmv{x_{\mGLmv{n}}}  \mGLsym{:}  \mGLnt{X_{\mGLmv{n}}} )  ;  ( \mGLmv{y_{{\mathrm{1}}}}  \mGLsym{:}  \mGLnt{A_{{\mathrm{1}}}}  \mGLsym{,} \, ... \, \mGLsym{,}  \mGLmv{y_{\mGLmv{m}}}  \mGLsym{:}  \mGLnt{A_{\mGLmv{m}}} )  }_{  (  \mGLnt{r_{{\mathrm{1}}}}  \mGLsym{,} \, ... \, \mGLsym{,}  \mGLnt{r_{\mGLmv{n}}}  )  }^{\MS}\\
      \,\,\,\,\,= \func{Mny} \mGLsym{(}  \mGLnt{r_{{\mathrm{1}}}}  \odot   \interp{ \mGLnt{X_{{\mathrm{1}}}} }^{\mathsf{GS} }   \mGLsym{)} \otimes \ldots \otimes \func{Mny} \mGLsym{(}  \mGLnt{r_{\mGLmv{n}}}  \odot   \interp{ \mGLnt{X_{\mGLmv{n}}} }^{\mathsf{GS} }   \mGLsym{)} \otimes \interp{ \mGLnt{A_{{\mathrm{1}}}} }^{\MS} \otimes \ldots \otimes \interp{ \mGLnt{A_{\mGLmv{m}}} }^{\MS}
    \end{array}
    \]

   \item
    The interpretation of proofs is given as follows by induction
    on the structure of proofs, considering each proof rule in turn (see below).

    \textbf{Notational conventions}
    We assume below
    that the premise proofs of each rule are called $\Pi$
    if there is only one, or are numbered $\Pi_{{\mathrm{1}}}$, $\Pi_{{\mathrm{2}}}$, etc.
    if there are more than one.

    For brevity we often ommit the interpretation brackets $\interp{-}$.

    For $\delta = r_1, \ldots, r_n$
    and $\Delta = X_1, \ldots, X_n$
    $\delta {\odot} \Delta = (r_1 {\odot} X_1)  \boxtimes \ldots \boxtimes (r_n {\odot} X_n)$.

    To apply an operation $\mathit{op}$ to a $n$-ary product, we write $\overline{\mathit{op}}$.

    \textbf{Derived operations}
    As per Proposition 1 of Benton~\cite{Benton:1994}, an adjunction over monoidal categories,
    whose functors are lax monoidal, induces a strong monoidal structure. We use this on the
    left, where:
    \begin{align*}
      n_{\boxtimes}^\Mny & : \Mny (A \boxtimes B) \\
      & \xrightarrow{\Mny (\eta \boxtimes \eta)}
      \Mny (\Lin (\Mny A) \boxtimes \Lin (\Mny B)) \\
      & \xrightarrow{\Mny (m^\Lin)}
      \Mny (\Lin (\Mny A \otimes \Mny B)) \\
      & \xrightarrow{\varepsilon}
      \Mny A \otimes \Mny B
    \end{align*}
    and
    \begin{align*}
      n_{\mathsf{J}}^\Mny : \Mny \mathsf{J} \xrightarrow{\Mny (m^\Lin_{\mathsf{I}})}
      \Mny \Lin \mathsf{J} \xrightarrow{\varepsilon} \mathsf{I}
     \end{align*}

    For a (lax/strong) monoidal functor, we can derive a $n$-ary version of its monoidal
    operation in a standard way.
    For example, for $\mathsf{m}^\Mny : \Mny A \otimes \Mny B
    \rightarrow \Mny (A \boxtimes B)$, we write its $n$-ary version as
    \begin{align*}
      \overline{m^\Mny} : \Mny A_1 \boxtimes \ldots \boxtimes \Mny A_n \to
                          \Mny (A_1 \otimes \ldots \otimes A_n)
    \end{align*}
    where for $0$-arity we let %
    \begin{align*}
      m^\Mny_0 & : \mathsf{J} \rightarrow \Mny \mathsf{I} \\
      m^\Mny_0 & = m_{\mathsf{J}}
    \end{align*}
    and for $n+1$ then (where we write $\overline{\Delta}_n$ for the $n$-times product of
    objects in $\Delta$ to which we can lift functors pointwise):
    \begin{align*}
      m^\Mny_{n+1} & : \Mny A \boxtimes \overline{\Mny \Delta}_n \rightarrow \Mny (A \otimes \overline{\Delta}_n) \\
      m^\Mny_{n+1} & = m^\Mny \circ (id \boxtimes m^\Mny_{n}) \\
    \end{align*}
    We will work up to associativity (which is witnessed by an isomorphism).

    \textbf{Inductive definition - GS}

    \begin{itemize}
    \item
      %
    \item (\mGLdruleGSTXXidName)
      \begin{align*}
        \ruleinterp{\mGLdruleGSTXXid{}} =
         1  \odot   \interp{ \mGLnt{X} }^{\mathsf{GS} } \morph{\mathsf{id}} \interp{ \mGLnt{X} }
      \end{align*}
      the soundness of which follows by the following property:
      \begin{align*}
        \begin{array}{rrl}
             & 1  \odot   \interp{ \mGLnt{X} }^{\mathsf{GS} } \\
      \equiv & \interp{ \mGLnt{X} }^\GS & \{\textit{strict action}\}
      \end{array}
      \end{align*}
      %

      %
      %
    \item (\mGLdruleGSTXXUnitRName)
      \begin{align*}
        \ruleinterp{\mGLdruleGSTXXUnitR{}}
      \end{align*}
      We need a morphism:
      \begin{align*}
             & \interpGS{\emptyset} \mto^{\interpGS{\mGLmv{j}}} \interpGS{\mathsf{J}} \\
      \equiv \;\; & \mathsf{J} \mto^{id} \mathsf{J}
      \end{align*}

      %
      %
    \item (\mGLdruleGSTXXUnitLName)
      \begin{align*}
        \ruleinterp{\mGLdruleGSTXXUnitL{}}
      \end{align*}
      By induction we have the following morphism:
      \[
        \delta_{{\mathrm{1}}} \odot \interpGS{\Delta_{{\mathrm{1}}}} \boxtimes \delta_{{\mathrm{2}}} \odot \interpGS{\Delta_{{\mathrm{2}}}} \mto^{\interpGS{\Pi}} \interpGS{\mGLnt{X}}
      \]
      Then the final interpretation is as follows:
      \[
      \begin{array}{rll}
               & (\delta_{{\mathrm{1}}} \odot \interpGS{\Delta_{{\mathrm{1}}}}) \boxtimes (\mGLnt{r} \odot \mathsf{J}) \boxtimes (\delta_{{\mathrm{2}}} \odot \interpGS{\Delta_{{\mathrm{2}}}}) \\
        \mto^{id \boxtimes \, n_{\mathsf{J}, r} \, \boxtimes id}  & (\delta_{{\mathrm{1}}} \odot \interpGS{\Delta_{{\mathrm{1}}}}) \boxtimes \mathsf{J} \boxtimes (\delta_{{\mathrm{2}}} \odot \interpGS{\Delta_{{\mathrm{2}}}})\\
        \mto^{\lambda \boxtimes id} & (\delta_{{\mathrm{1}}} \odot \interpGS{\Delta_{{\mathrm{1}}}}) \boxtimes (\delta_{{\mathrm{2}}} \odot \interpGS{\Delta_{{\mathrm{2}}}}) \\
        \mto^{\interpGS{\Pi}} & \interpGS{\mGLnt{X}}
      \end{array}
      \]

      %
      %
    \item (\mGLdruleGSTXXTenRName)
      \begin{align*}
        & \ruleinterp{\mGLdruleGSTXXTenR{}} = \\
        & (  \delta_{{\mathrm{1}}}   \odot  \interp{  \Delta_{{\mathrm{1}}}  }  )   \boxtimes   (  \delta_{{\mathrm{2}}}   \odot  \interp{  \Delta_{{\mathrm{2}}}  }  )
          \morph{\interp{ \Pi_{{\mathrm{1}}} }^{\GST} \boxtimes \interp{ \Pi_{{\mathrm{2}}} }^{\GST}}
            \interp{ \mGLnt{X} }   \boxtimes   \interp{ \mGLnt{Y} }
      \end{align*}

      %
      %
    \item (\mGLdruleGSTXXTenLName)
      \begin{align*}
        & \ruleinterp{\mGLdruleGSTXXTenL{}} =
      \end{align*}
      \begin{align*}
        & ((  \delta_{{\mathrm{1}}}   \odot  \interp{  \Delta_{{\mathrm{1}}}  }  )   \boxtimes   ( \mGLnt{r}  \odot  \mGLsym{(}   \interp{ \mGLnt{X} }^{\mathsf{GS} }   \boxtimes   \interp{ \mGLnt{Y} }^{\mathsf{GS} }   \mGLsym{)} )   \boxtimes   (  \delta_{{\mathrm{2}}}   \odot  \interp{  \Delta_{{\mathrm{2}}}  }  )) \\
        \morph{id \boxtimes n_{\boxtimes, r, \interp{X}, \interp{Y}} \boxtimes id} &
        (  \delta_{{\mathrm{1}}}   \odot  \interp{  \Delta_{{\mathrm{1}}}  }  )   \boxtimes   ( \mGLnt{r}  \odot   \interp{ \mGLnt{X} }^{\mathsf{GS} }  )   \boxtimes   ( \mGLnt{r}  \odot   \interp{ \mGLnt{Y} }^{\mathsf{GS} }  )   \boxtimes   (  \delta_{{\mathrm{2}}}   \odot  \interp{  \Delta_{{\mathrm{2}}}  }  ) \\
        \morph{\interp{ \Pi }^{\GST}} &
        \interp{ \mGLnt{Z} }
      \end{align*}

      %
      %
    \item (\mGLdruleGSTXXLinRName)
      \begin{align*}
        & \ruleinterp{\mGLdruleGSTXXLinR{}} =
      \end{align*}
      First, by induction we have the interpretation of the premise as follows:
      \[
        \func{Mny} \mGLsym{(}  \mGLnt{r_{{\mathrm{1}}}}  \odot   \interp{ \mGLnt{X_{{\mathrm{1}}}} }^{\mathsf{GS} }   \mGLsym{)} \otimes \ldots \otimes \func{Mny} \mGLsym{(}  \mGLnt{r_{\mGLmv{n}}}  \odot   \interp{ \mGLnt{X_{\mGLmv{n}}} }^{\mathsf{GS} }   \mGLsym{)}
        \mto^{\interp{\Pi}^\MS} \interp{ \mGLnt{B} }
      \]
      where $\delta  \odot  \Delta = (  \mGLnt{r_{{\mathrm{1}}}}  \mGLsym{,} \, ... \, \mGLsym{,}  \mGLnt{r_{\mGLmv{n}}}  )   \odot   ( \mGLnt{X_{{\mathrm{1}}}}  \mGLsym{,} \, ... \, \mGLsym{,}  \mGLnt{X_{\mGLmv{n}}} )$

      Now we have the following interpretation of the conclusion:
      \[
      \begin{array}{rll}
        & (\mGLsym{(}  \mGLnt{r_{{\mathrm{1}}}}  \odot   \interp{ \mGLnt{X_{{\mathrm{1}}}} }^{\mathsf{GS} }   \mGLsym{)}) \boxtimes \ldots \boxtimes (\mGLsym{(}  \mGLnt{r_{\mGLmv{n}}}  \odot   \interp{ \mGLnt{X_{\mGLmv{n}}} }^{\mathsf{GS} }   \mGLsym{)})\\
        \morph{\eta_{r_1} \boxtimes \ldots \boxtimes \eta_{r_n}} &
        \Lin (\func{Mny} \mGLsym{(}  \mGLnt{r_{{\mathrm{1}}}}  \odot   \interp{ \mGLnt{X_{{\mathrm{1}}}} }^{\mathsf{GS} }   \mGLsym{)}) \boxtimes \ldots \boxtimes \Lin(\func{Mny}(\mGLnt{r_{\mGLmv{n}}} \odot \interp{ \mGLnt{X_{\mGLmv{n}}} }))\\
        \mto^{\overline{\func{m}_{\Lin}}} &
        \Lin (\func{Mny}(\mGLnt{r_{{\mathrm{1}}}} \odot \interp{ \mGLnt{X_{{\mathrm{1}}}} }) \otimes \ldots \otimes \func{Mny}(\mGLnt{r_{\mGLmv{n}}} \odot \interp{ \mGLnt{X_{\mGLmv{n}}} }))\\
        \mto^{\Lin \interp{\Pi}^\MS} &
        \Lin \interp{ \mGLnt{B} }
      \end{array}
      \]
      %

      %
      %
    \item (\mGLdruleGSTXXCutName)
      \begin{align*}
        \ruleinterp{\mGLdruleGSTXXCut{}}
      \end{align*}
      By induction we have the following morphisms:
      \[
      \begin{array}{lll}
        \delta_2 \odot \Delta_2 \mto^{\interp{\Pi_1}} \mGLnt{X}\\
        \delta_1 \odot \Delta_1 \boxtimes r \odot \mGLnt{X} \boxtimes \delta_3 \odot \Delta_3 \mto^{\interp{\Pi_2}} \mGLnt{X}\\
      \end{array}
      \]
      First, we promote $\interp{\Pi_1}$:
      \[
      \begin{array}{lll}
          & r \odot (\delta_2 \odot \Delta_2) \mto^{r \odot \interp{\Pi_1}} r \odot \mGLnt{X}\\
      \stackrel{\delta_{\Delta_2,r,\delta_2}}{=} & (r * \delta_2) \odot \Delta_2 \mto^{r \odot \interp{\Pi_1}} r \odot \mGLnt{X}
      \end{array}
      \]
      Then we compose with $\interp{\Pi_2}$:
      \[
      \begin{array}{rll}
        & \delta_1 \odot \Delta_1 \boxtimes (r * \delta_2) \odot \Delta_2 \boxtimes \delta_3 \odot \Delta_3\\
        \mto^{\id \boxtimes (r \odot \interp{\Pi_1}) \boxtimes \id} &
        \delta_1 \odot \Delta_1 \boxtimes r \odot \mGLnt{X} \boxtimes \delta_3 \odot \Delta_3\\
        \mto^{\interp{\Pi_2}} &
        \mGLnt{Y}
      \end{array}
      \]

      %
      %
    \item (\mGLdruleGSTXXWeakName)
      \begin{align*}
        & \ruleinterp{\mGLdruleGSTXXWeak{}}
      \end{align*}
      By the induction hypothesis we have:
      \[
      \delta_1 \odot \Delta_1 \boxtimes \delta_2 \odot \Delta_2 \mto^{\interp{\Pi}} \mGLnt{Y}
      \]
      Then the intepretation is:
      \[
      \begin{array}{rll}
        & \delta_1 \odot \Delta_1 \boxtimes 0 \odot \interp{\mGLnt{X}}^\GS \boxtimes \delta_2 \odot \Delta_2\\
      \mto^{\id \boxtimes (\func{weak}_{\interp{X}}) \boxtimes \id}
      & \delta_1 \odot \Delta_1 \boxtimes \mathsf{J} \boxtimes \delta_2 \odot \Delta_2\\
      \mto^{\lambda \boxtimes id}
      & \delta_1 \odot \Delta_1 \boxtimes \delta_2 \odot \Delta_2\\
      \mto^{\interp{\Pi}}
      & \mGLnt{Y}\\
      \end{array}
      \]

      %
      %
    \item (\mGLdruleGSTXXContName)
      \begin{align*}
        & \ruleinterp{\mGLdruleGSTXXCont{}}
      \end{align*}
      By induction we have a morphism:
      \[
      \delta_{{\mathrm{1}}} \odot \Delta_{{\mathrm{1}}} \boxtimes \mGLnt{r_{{\mathrm{1}}}} \odot \mGLnt{X} \boxtimes \mGLnt{r_{{\mathrm{2}}}} \odot \mGLnt{X} \boxtimes \delta_{{\mathrm{2}}} \odot \Delta_{{\mathrm{2}}} \mto^{\interp{\Pi}} \mGLnt{Y}
      \]
      The final interpretation is therefore as follows:
      \[
      \begin{array}{rll}
                          & \delta_{{\mathrm{1}}} \odot \Delta_{{\mathrm{1}}} \boxtimes \mGLsym{(}  \mGLnt{r_{{\mathrm{1}}}}  +  \mGLnt{r_{{\mathrm{2}}}}  \mGLsym{)} \odot \mGLnt{X} \boxtimes \delta_{{\mathrm{2}}} \odot \Delta_{{\mathrm{2}}}\\
        \mto^{\id \boxtimes (\mathsf{contr}_{r1,r2,X}) \boxtimes \id}
                          & \delta_{{\mathrm{1}}} \odot \Delta_{{\mathrm{1}}} \boxtimes (\mGLnt{r_{{\mathrm{1}}}} \odot \mGLnt{X} \boxtimes \mGLnt{r_{{\mathrm{2}}}} \odot \mGLnt{X}) \boxtimes \delta_{{\mathrm{2}}} \odot \Delta_{{\mathrm{2}}}\\
        \mto^=            & \delta_{{\mathrm{1}}} \odot \Delta_{{\mathrm{1}}} \boxtimes \mGLnt{r_{{\mathrm{1}}}} \odot \mGLnt{X} \boxtimes \mGLnt{r_{{\mathrm{2}}}} \odot \mGLnt{X} \boxtimes \delta_{{\mathrm{2}}} \odot \Delta_{{\mathrm{2}}}\\
        \mto^{\interp{\Pi}} & \mGLnt{Y}

      \end{array}
      \]

      %
      %
    \item (\mGLdruleGSTXXExName)
      \begin{align*}
        & \ruleinterp{\mGLdruleGSTXXEx{}}
      \end{align*}
      This case easily follows from symmetry of $\boxtimes$ the tensor product in $\catUnder{R}{C}$.

      %
      %
    \item (\mGLdruleGSTXXSubName)
      \begin{align*}
        & \ruleinterp{\mGLdruleGSTXXSub{}}
      \end{align*}
      This interpretation is constructed via functoriality of $\odot$ in its first argument:
      \begin{align*}
        \delta_{{\mathrm{2}}} \odot \interp{\Delta}
        \xrightarrow{\interp{\delta_{{\mathrm{1}}}  \leq  \delta_{{\mathrm{2}}}} \odot \interp{\Delta}}
        \delta_{{\mathrm{1}}} \odot \interp{\Delta}
        \xrightarrow{\interp{\Pi_{{\mathrm{1}}}}}
        \interp{ \mGLnt{X} }
      \end{align*}
      Note the syntactic sugar on rules for (sub) given at the start
      of Section~\ref{sec:full-eq-theory} allows a refinement of the
      inequality to focus on just one grade in the context. That is,
      given an inequality $\mGLnt{r}  \leq  \mGLnt{s}$ then we write:
      \begin{align*}
        \inferrule*[right=$\mGLdruleGSTXXSubName{}$]
                   { \delta'_{{\mathrm{1}}}  \mGLsym{,}  \mGLnt{s}  \mGLsym{,}  \delta'_{{\mathrm{2}}}  \odot  \Delta'_{{\mathrm{1}}}  \mGLsym{,}  \mGLnt{X'}  \mGLsym{,}  \Delta'_{{\mathrm{2}}}  \vdash_{\mathsf{GS} }  \mGLnt{t}  \mGLsym{:}  \mGLnt{Y} \quad \mGLnt{r}  \leq  \mGLnt{s}}
                   { \delta'_{{\mathrm{1}}}  \mGLsym{,}  \mGLnt{r}  \mGLsym{,}  \delta'_{{\mathrm{2}}}  \odot  \Delta'_{{\mathrm{1}}}  \mGLsym{,}  \mGLnt{X'}  \mGLsym{,}  \Delta'_{{\mathrm{2}}}  \vdash_{\mathsf{GS} }  \mGLnt{t}  \mGLsym{:}  \mGLnt{Y} }
      \end{align*}
      In this case, we can specialise the semantics to:
      \begin{gather*}
        (\delta'_{{\mathrm{1}}} \odot \interp{\Delta'_{{\mathrm{1}}}})
        \boxtimes
        (s \odot \interp{\mGLnt{X'}})
        \boxtimes
        (\delta'_{{\mathrm{2}}} \odot \interp{\Delta'_{{\mathrm{2}}}})
        \xrightarrow{id \boxtimes (\interp{\mGLnt{s}  \leq  \mGLnt{r}} \odot \interp{X'}) \boxtimes id}
        (\delta'_{{\mathrm{1}}} \odot \interp{\Delta'_{{\mathrm{1}}}})
        \boxtimes
        (r \odot \interp{\mGLnt{X'}})
        \boxtimes
        (\delta'_{{\mathrm{2}}} \odot \interp{\Delta'_{{\mathrm{2}}}})
        \xrightarrow{\interp{\Pi_{{\mathrm{1}}}}}
        \interp{ \mGLnt{Y} }
      \end{gather*}

    \end{itemize}

  \item \textbf{Inductive definition - MS}

  The interpretation of proofs is given as follows by induction
    on the structure of proofs, considering each proof rule in turn.
    We assume below
    that the premise proofs of each rule are called $\Pi$
    if there is only one, or are numbered $\Pi_{{\mathrm{1}}}$, $\Pi_{{\mathrm{2}}}$, etc.
    if there are more than one.

    \begin{itemize}
      %
      %
    \item (\mGLdruleMSTXXidName{})
      \begin{align*}
        & \ruleinterp{\mGLdruleMSTXXid{}}
      \end{align*}
      The final interpretation is as follows:
      \[
      \interpMS{\mGLnt{A}} \mto^{\mathsf{id}} \interpMS{\mGLnt{A}}
      \]

      %
      %
    \item (\mGLdruleMSTXXSubName{})
      \begin{align*}
        & \ruleinterp{\mGLdruleMSTXXSub{}}
      \end{align*}
      The assumption $\delta_{{\mathrm{1}}}  \leq  \delta_{{\mathrm{2}}}$ corresponds to a list of morphism
      in $\mathcal{R}$.  Thus, this case follows from functorality in the
      first argument of $\func{Mny}$.

      %
      %
    \item (\mGLdruleMSTXXUnitRName{})
      \begin{align*}
        & \ruleinterp{\mGLdruleMSTXXUnitR{}}
      \end{align*}
      We need a morphism:
      \[
      I \otimes I \mto^{\interpMS{\mGLmv{i}}} I
      \]
      Thus we simply choose $\lambda_I$ of the monoidal category $\mathcal{M}$.
      %
      %
    \item (\mGLdruleMSTXXUnitLName{})
      \begin{align*}
        & \ruleinterp{\mGLdruleMSTXXUnitL{}}
      \end{align*}
      By induction we have the following morphism:
      \[
      \Mny (\delta \odot \interpGS{\Delta}) \otimes \interpMS{\Gamma_{{\mathrm{1}}}} \otimes \interpMS{\Gamma_{{\mathrm{2}}}} \mto^{\interpMS{\Pi}} \interpMS{\mGLnt{A}}
      \]
      Then the final interpretation is as follows:
      \[
      \begin{array}{rll}
               & \Mny(\delta \odot \interpGS{\Delta}) \otimes \interpMS{\Gamma_{{\mathrm{1}}}} \otimes I \otimes \interpMS{\Gamma_{{\mathrm{2}}}}\\
        \mto^= & \Mny(\delta \odot \interpGS{\Delta}) \otimes (\interpMS{\Gamma_{{\mathrm{1}}}} \otimes I) \otimes \interpMS{\Gamma_{{\mathrm{2}}}}\\
        \mto^{\id \otimes \lambda \otimes \id} & \Mny(\delta \odot \interpGS{\Delta}) \otimes \interpMS{\Gamma_{{\mathrm{1}}}} \otimes \interpMS{\Gamma_{{\mathrm{2}}}}\\
        \mto^{\interpMS{\Pi}} & \interpMS{\mGLnt{A}}
      \end{array}
      \]
      %
      %
    \item (\mGLdruleMSTXXImpRName{})
      \begin{align*}
        & \ruleinterp{\mGLdruleMSTXXImpR{}}
      \end{align*}
      We know by assumption that $\cat{M}$ is closed, and thus, there is a natural isomoprhism:
      \[
      \mathsf{Hom}_{\cat{C}}(C \otimes A,B) \mto^{\mathsf{curry}} \mathsf{Hom}_{\cat{C}}(C,\mGLnt{A}  \multimap  \mGLnt{B})
      \]
      Now by induction we have the following morphism:
      \[
      \Mny(\delta \odot \interpGS{\Delta}) \otimes \interpMS{\Gamma} \otimes \interpMS{\mGLnt{A}} \mto^{\interpMS{\Pi}} \interpMS{\mGLnt{B}}
      \]
      The final interpretation is as follows:
      \[
      \Mny(\delta \odot \interpGS{\Delta}) \otimes \interpMS{\Gamma} \mto^{\mathsf{curry}(\interpMS{\Pi})} (\interpMS{\mGLnt{A}} \multimap \interpMS{\mGLnt{B}})
      \]

      %
      %
    \item (\mGLdruleMSTXXImpLName{})
      \begin{align*}
        & \ruleinterp{\mGLdruleMSTXXImpL{}}
      \end{align*}
      Recall that due to $\cat{M}$ being closed we have the following natural isomorphism:
      \[
      \mathsf{Hom}_{\cat{C}}(C \otimes A,B) \mto^{\mathsf{curry}} \mathsf{Hom}_{\cat{C}}(C,\mGLnt{A}  \multimap  \mGLnt{B})
      \]
      Now replacing $C$ with $\mGLnt{A}  \multimap  \mGLnt{B}$ we obtain the following:
      \[
      \mathsf{Hom}_{\cat{C}}((\mGLnt{A}  \multimap  \mGLnt{B}) \otimes A,B) \mto^{\mathsf{curry}} \mathsf{Hom}_{\cat{C}}(\mGLnt{A}  \multimap  \mGLnt{B},\mGLnt{A}  \multimap  \mGLnt{B})
      \]
      Thus, $\mathsf{curry}^{-1}(\id) : \mGLsym{(}  \mGLnt{A}  \multimap  \mGLnt{B}  \mGLsym{)}  \otimes  \mGLnt{A} \mto B$.

      \ \\ \noindent
      Next by the induction hypothesis we have the following morphisms:
      \[
      \begin{array}{lll}
        \Mny(\delta_{{\mathrm{2}}} \odot \interpGS{\Delta_{{\mathrm{2}}}}) \otimes \interpMS{\Gamma_{{\mathrm{2}}}} \mto^{\interpMS{\Pi_2}} \interpMS{\mGLnt{A}}\\
        \Mny(\delta_{{\mathrm{1}}} \odot \interpGS{\Delta_{{\mathrm{1}}}}) \otimes \interpMS{\Gamma_{{\mathrm{1}}}} \otimes \interpMS{\mGLnt{B}} \otimes \interpMS{\Gamma_{{\mathrm{3}}}} \mto^{\interpMS{\Pi_1}} \interpMS{\mGLnt{C}}
      \end{array}
      \]
      First, we can compose $\interpMS{\Pi_2}$ and $\mathsf{curry}^{-1}(\id)$:
      \[
      \begin{array}{rll}
        & \Mny(\delta_{{\mathrm{2}}} \odot \interpGS{\Delta_{{\mathrm{2}}}}) \otimes (\interpMS{\mGLnt{A}} \multimap \interpMS{\mGLnt{B}}) \otimes \interpMS{\Gamma_{{\mathrm{2}}}} \\
        \mto^= & (\interpMS{\mGLnt{A}} \multimap \interpMS{\mGLnt{B}}) \otimes \Mny(\delta_{{\mathrm{2}}} \odot \interpGS{\Delta_{{\mathrm{2}}}}) \otimes \interpMS{\Gamma_{{\mathrm{2}}}} \\
        \mto^{id \otimes \interp{\Pi_2}} & (\interpMS{\mGLnt{A}} \multimap \interpMS{\mGLnt{B}}) \otimes \interpMS{\mGLnt{A}}\\
        \mto^{\mathsf{curry}^{-1}(\id)} & \interpMS{\mGLnt{B}}\\
      \end{array}
      \]
      We will denote the previous composition by $\interpMS{\Pi} = (\mathsf{curry}^{-1}(\id)) \circ (id \otimes \interp{\Pi_2})$ in the final interpretation:
      \[
      \begin{array}{rll}
        & \Mny(\delta_{{\mathrm{1}}} \odot \interpGS{\Delta_{{\mathrm{1}}}}) \otimes \Mny(\delta_{{\mathrm{2}}} \odot \interpGS{\Delta_{{\mathrm{2}}}}) \otimes \interpMS{\Gamma_{{\mathrm{1}}}} \otimes (\interpMS{\mGLnt{A}} \multimap \interpMS{\mGLnt{B}}) \otimes \interpMS{\Gamma_{{\mathrm{2}}}} \otimes \interpMS{\Gamma_{{\mathrm{3}}}}\\
        \mto^= & \Mny(\delta_{{\mathrm{1}}} \odot \interpGS{\Delta_{{\mathrm{1}}}}) \otimes \interpMS{\Gamma_{{\mathrm{1}}}} \otimes \Mny(\delta_{{\mathrm{2}}} \odot \interpGS{\Delta_{{\mathrm{2}}}}) \otimes (\interpMS{\mGLnt{A}} \multimap \interpMS{\mGLnt{B}}) \otimes \interpMS{\Gamma_{{\mathrm{2}}}} \otimes \interpMS{\Gamma_{{\mathrm{3}}}}\\
        \mto^{\id \otimes \interpMS{\Pi} \otimes \id} & \Mny(\delta_{{\mathrm{1}}} \odot \interpGS{\Delta_{{\mathrm{1}}}}) \otimes \interpMS{\Gamma_{{\mathrm{1}}}} \otimes \interpMS{\mGLnt{B}} \otimes \interpMS{\Gamma_{{\mathrm{3}}}}\\
        \mto^{\interpMS{\Pi_1}} & \interpMS{\mGLnt{C}}
      \end{array}
      \]

            %
      %
    \item (\mGLdruleMSTXXTenLName{})
      \begin{align*}
        & \ruleinterp{\mGLdruleMSTXXTenL{}}
      \end{align*}
      This case follows directly from the induction hypothesis.

      %
      %
    \item (\mGLdruleMSTXXTenRName{})
      \begin{align*}
        & \ruleinterp{\mGLdruleMSTXXTenR{}}
      \end{align*}
      By the induction hypothesis we have the following morphisms:
      \[
      \begin{array}{lll}
        \Mny(\delta_{{\mathrm{1}}} \odot\interpGS{\Delta_{{\mathrm{1}}}}) \otimes \interpMS{\Gamma_{{\mathrm{1}}}} \mto^{\interpMS{\Pi_1}} \interpMS{\mGLnt{A}}\\
        \Mny(\delta_{{\mathrm{2}}} \odot\interpGS{\Delta_{{\mathrm{2}}}}) \otimes \interpMS{\Gamma_{{\mathrm{2}}}} \mto^{\interpMS{\Pi_2}} \interpMS{\mGLnt{B}}\\
      \end{array}
      \]
      The final interpretation is as follows:
      \[
      \begin{array}{rll}
        & \Mny(\delta_{{\mathrm{1}}} \odot\interpGS{\Delta_{{\mathrm{1}}}}) \otimes \Mny(\delta_{{\mathrm{2}}} \odot\interpGS{\Delta_{{\mathrm{2}}}}) \otimes \interpMS{\Gamma_{{\mathrm{1}}}} \otimes \interpMS{\Gamma_{{\mathrm{2}}}}\\
        \mto^= & \Mny(\delta_{{\mathrm{1}}} \odot\interpGS{\Delta_{{\mathrm{1}}}}) \otimes \interpMS{\Gamma_{{\mathrm{1}}}} \otimes \Mny(\delta_{{\mathrm{2}}} \odot\interpGS{\Delta_{{\mathrm{2}}}}) \otimes \interpMS{\Gamma_{{\mathrm{2}}}}\\
        \mto^{\interpMS{\Pi_1} \otimes \interpMS{\Pi_2}} & \interpMS{\mGLnt{A}} \otimes \interpMS{\mGLnt{B}}
      \end{array}
      \]

      %
      %
    \item (\mGLdruleMSTXXGUnitLName{})
      \begin{align*}
        & \ruleinterp{\mGLdruleMSTXXGUnitL{}} \\
        &
      \end{align*}
      By the induction hypothesis we have:
      \[
      \Mny(\delta_{{\mathrm{1}}} \odot \interpGS{\Delta_{{\mathrm{1}}}}) \otimes \Mny(\delta_{{\mathrm{1}}} \odot \interpGS{\Delta_{{\mathrm{1}}}}) \otimes \interpMS{\Gamma} \mto^{\interpMS{\Pi}} \interpMS{\mGLnt{A}}
      \]
      The interpretation is then:
      \[
      \begin{array}{rll}
        & \Mny(\delta_{{\mathrm{1}}} \odot \interpGS{\Delta_{{\mathrm{1}}}}) \otimes \Mny(r \odot J) \otimes \Mny(\delta_{{\mathrm{2}}} \odot \interpGS{\Delta_{{\mathrm{2}}}}) \otimes \interpMS{\Gamma}\\
        \mto^{\id \otimes \Mny (n_{J,r}) \otimes \id} & \Mny(\delta_{{\mathrm{1}}} \odot \interpGS{\Delta_{{\mathrm{1}}}}) \otimes \Mny(J) \otimes \Mny(\delta_{{\mathrm{2}}} \odot \interpGS{\Delta_{{\mathrm{2}}}}) \otimes \interpMS{\Gamma}\\
        \mto^{\id \otimes n^{\Mny}_J \otimes \id} & \Mny(\delta_{{\mathrm{1}}} \odot \interpGS{\Delta_{{\mathrm{1}}}}) \otimes I \otimes \Mny(\delta_{{\mathrm{2}}} \odot \interpGS{\Delta_{{\mathrm{2}}}}) \otimes \interpMS{\Gamma}\\
        \mto^{=} & \Mny(\delta_{{\mathrm{1}}} \odot \interpGS{\Delta_{{\mathrm{1}}}}) \otimes \Mny(\delta_{{\mathrm{2}}} \odot \interpGS{\Delta_{{\mathrm{2}}}}) \otimes \interpMS{\Gamma}\\
        \mto^{\interpMS{\Pi}} & \interpMS{\mGLnt{A}}
      \end{array}
      \]

      %
      %
    \item (\mGLdruleMSTXXGTenLName{})
      \begin{align*}
        & \ruleinterp{\mGLdruleMSTXXGTenL{}}
      \end{align*}
      By the induction hypothesis we have the following:
      \[
      \Mny(\delta_{{\mathrm{1}}} \odot \interpGS{\Delta_{{\mathrm{1}}}}) \otimes \Mny(\mGLnt{r} \odot \interpGS{\mGLnt{X}}) \otimes \Mny(\mGLnt{r} \odot \interpGS{\mGLnt{Y}}) \otimes \Mny(\delta_{{\mathrm{2}}} \odot \interpGS{\Delta_{{\mathrm{2}}}}) \otimes \interpMS{\Gamma} \mto^{\interpMS{\Pi}} \interpMS{\mGLnt{A}}
      \]
      The final interpretation is then:
      \begin{gather*}
      \begin{array}{rll}\
        & \Mny(\delta_{{\mathrm{1}}} \odot \interpGS{\Delta_{{\mathrm{1}}}}) \otimes \Mny(\mGLnt{r} \odot (\interpGS{\mGLnt{X}} \boxtimes \interpGS{\mGLnt{Y}})) \otimes \Mny(\delta_{{\mathrm{2}}} \odot \interpGS{\Delta_{{\mathrm{2}}}}) \otimes \interpMS{\Gamma}\\
        \mto^{\id \otimes \Mny (n^\odot_{\boxtimes,r}) \otimes \id}& \Mny(\delta_{{\mathrm{1}}} \odot \interpGS{\Delta_{{\mathrm{1}}}}) \otimes \Mny((\mGLnt{r} \odot \interpGS{\mGLnt{X}}) \boxtimes (\mGLnt{r} \odot \interpGS{\mGLnt{Y}})) \otimes \Mny(\delta_{{\mathrm{2}}} \odot \interpGS{\Delta_{{\mathrm{2}}}}) \otimes \interpMS{\Gamma}\\
        \mto^{\id \otimes n^\Mny_{\boxtimes} \otimes \id}& \Mny(\delta_{{\mathrm{1}}} \odot \interpGS{\Delta_{{\mathrm{1}}}}) \otimes \Mny(\mGLnt{r} \odot \interpGS{\mGLnt{X}}) \otimes \Mny(\mGLnt{r} \odot \interpGS{\mGLnt{Y}}) \otimes \Mny(\delta_{{\mathrm{2}}} \odot \interpGS{\Delta_{{\mathrm{2}}}}) \otimes \interpMS{\Gamma}\\
        \mto^{\interpMS{\Pi}} & \interpMS{\mGLnt{A}}
      \end{array}
      \end{gather*}

      %
      %
    \item (\mGLdruleMSTXXGrdRName{})
      \begin{align*}
        & \ruleinterp{\mGLdruleMSTXXGrdR{}}
      \end{align*}
      Suppose $\delta  \odot  \Delta$ = $(\mGLnt{r_{{\mathrm{1}}}}  \mGLsym{,} \, ... \, \mGLsym{,}  \mGLnt{r_{\mGLmv{n}}}) \odot (\mGLmv{x_{{\mathrm{1}}}}  \mGLsym{:}  \mGLnt{X_{{\mathrm{1}}}}  \mGLsym{,} \, ... \, \mGLsym{,}  \mGLmv{x_{\mGLmv{n}}}  \mGLsym{:}  \mGLnt{X_{\mGLmv{n}}})$.
      Then by the inductive hypothesis we have:
      \[
      (\mGLnt{r_{{\mathrm{1}}}}  \odot   \interp{ \mGLnt{X_{{\mathrm{1}}}} }^{\mathsf{GS} }) \boxtimes \ldots \boxtimes (\mGLnt{r_{\mGLmv{n}}}  \odot   \interp{ \mGLnt{X_{\mGLmv{n}}} }^{\mathsf{GS} })
      \mto^{\interp{\Pi}} \interp{\mGLnt{X}}^\GS
      \]
      The final interpretation is as follows:
      \begin{gather*}
      \begin{array}{rll}
        & \func{Mny} \mGLsym{(}  \mGLsym{(}  \mGLnt{r}  *  \mGLnt{r_{{\mathrm{1}}}}  \mGLsym{)}  \odot   \interp{ \mGLnt{X_{{\mathrm{1}}}} }^{\mathsf{GS} }   \mGLsym{)} \otimes \cdots
          \otimes \func{Mny} \mGLsym{(}  \mGLsym{(}  \mGLnt{r}  *  \mGLnt{r_{\mGLmv{n}}}  \mGLsym{)}  \odot   \interp{ \mGLnt{X_{\mGLmv{n}}} }^{\mathsf{GS} }   \mGLsym{)} \\
      \xrightarrow{\overline{m^{\mathsf{Mny}}}} & \Mny(((\mGLnt{r}  *  \mGLnt{r_{{\mathrm{1}}}}) \odot \interp{\mGLnt{X_{{\mathrm{1}}}}}^\GS \boxtimes \cdots \boxtimes ((\mGLnt{r}  *  \mGLnt{r_{\mGLmv{n}}}) \odot \interp{\mGLnt{X_{\mGLmv{n}}}}^\GS)) \\
          =^{\mathsf{Mny}(\delta^\odot \boxtimes \ldots \boxtimes \delta^\odot)} &
          \Mny((\mGLnt{r} \odot (\mGLnt{r_{{\mathrm{1}}}} \odot \interp{\mGLnt{X_{{\mathrm{1}}}}}^\GS)) \boxtimes \cdots \boxtimes (\mGLnt{r} \odot (\mGLnt{r_{\mGLmv{n}}} \odot \interp{\mGLnt{X_{\mGLmv{n}}}}^\GS))) \\
\xrightarrow{\Mny (\overline{m_{\boxtimes,r,X,Y}})}  &
\Mny(\mGLnt{r} \odot ((\mGLnt{r_{{\mathrm{1}}}} \odot \interp{\mGLnt{X_{{\mathrm{1}}}}}^\GS) \boxtimes \cdots \boxtimes (\mGLnt{r_{\mGLmv{n}}} \odot \interp{\mGLnt{X_{\mGLmv{n}}}}^\GS)))  \\
          \mto^{\Mny(r \odot \interp{\Pi})}
        & \Mny(r \odot \interp{\mGLnt{X}})^\GS
      \end{array}
      \end{gather*}

      %
      %
    \item (\mGLdruleMSTXXLinLName{})
      \begin{align*}
        & \ruleinterp{\mGLdruleMSTXXLinL{}}
      \end{align*}
      By induction we have the following morphism:
      \[
      \Mny(\delta \odot \interpGS{\Delta}) \otimes \interpMS{\mGLnt{A}} \otimes \interpMS{\Gamma} \mto^{\interpMS{\Pi}} \interpMS{B}
      \]
      Thus, the final interpretation is as follows:
      \[
      \begin{array}{rlll}
        & \Mny(\delta \odot \interpGS{\Delta}) \otimes \Mny(1 \odot \Lin\interpMS{\mGLnt{A}}) \otimes \interpMS{\Gamma}\\
= &  \Mny(\delta \odot \interpGS{\Delta}) \otimes \Mny(\Lin\interpMS{\mGLnt{A}}) \otimes \interpMS{\Gamma}\\
        \mto^{\id \otimes\ \varepsilon\ \otimes \id} & \Mny(\delta \odot \interpGS{\Delta}) \otimes \interpMS{\mGLnt{A}} \otimes \interpMS{\Gamma}\\
        \mto^{\interpMS{\Pi}} & \interpMS{B}
      \end{array}
      \]

      %
      %
    \item (\mGLdruleMSTXXGrdLName{})
      \begin{align*}
        & \ruleinterp{\mGLdruleMSTXXGrdL{}}
      \end{align*}
      By induction we have the following morphism:
      \[
      \Mny(\delta \odot \interpGS{\Delta}) \otimes \Mny(r \odot \interpGS{X}) \otimes \interpMS{\Gamma} \mto^{\interpMS{\Pi}} \interpMS{C}
      \]
      Notice that this is exactly the interpretation of the conclusion.

      %
      %
    \item (\mGLdruleMSTXXCutName{})
      \begin{align*}
        & \ruleinterp{\mGLdruleMSTXXCut{}}
      \end{align*}
      This case follows from simply composing the two morphisms obtained from the induction hypothesis.

      %
      %
    \item (\mGLdruleMSTXXGCutName{})
      \begin{align*}
        & \ruleinterp{\mGLdruleMSTXXGCut{}} = \\
        &
      \end{align*}
      By induction on the premises we have:
      \begin{gather*}
        \begin{align*}
        & \delta_{{\mathrm{2}}} \odot \interpGS{\Delta_{{\mathrm{2}}}} \xrightarrow{\interpGS{\Pi_1}} \interpGS{\mGLnt{X}} \\
        & \Mny(\delta_{{\mathrm{1}}} \odot \interpGS{\Delta_{{\mathrm{1}}}}) \otimes
          \Mny(\mGLnt{r} \odot \interpGS{\mGLnt{X}}) \otimes
          \Mny(\delta_{{\mathrm{3}}} \odot \interpGS{\Delta_{{\mathrm{3}}}}) \otimes \interpMS{\Gamma} \mto^{\interpMS{\Pi_2}} \interpMS{B}
          \end{align*}
      \end{gather*}
      Let $\delta_{{\mathrm{2}}} = r_1 , \ldots, r_n$
      and let $\Delta_{{\mathrm{2}}} = \mGLnt{X_{{\mathrm{1}}}}, \ldots, \mGLnt{X_{\mGLmv{n}}}$.

      Then we construct the interpretation:
      \begin{gather*}
        \begin{align*}
       & \Mny(\delta_{{\mathrm{1}}} \odot \interpGS{\Delta_{{\mathrm{1}}}}) \otimes
          \Mny(\mGLnt{r}  *  \mGLnt{r_{{\mathrm{0}}}} \odot \interpGS{\mGLnt{X_{{\mathrm{1}}}}}) \otimes
          \ldots
          \otimes
          \Mny(\mGLnt{r}  *  \mGLnt{r_{\mGLmv{n}}} \odot \interpGS{\mGLnt{X_{\mGLmv{n}}}}) \otimes
          \Mny(\delta_{{\mathrm{2}}} \odot \interpGS{\Delta_{{\mathrm{3}}}}) \otimes \interpMS{\Gamma}  \\
\stackrel{id \otimes \Mny(\delta) \otimes \ldots \otimes \Mny(\delta) \otimes id}{=}  \;\;&
        \Mny(\delta_{{\mathrm{1}}} \odot \interpGS{\Delta_{{\mathrm{1}}}}) \otimes
          \Mny(\mGLnt{r} \odot (\mGLnt{r_{{\mathrm{0}}}} \odot \interpGS{\mGLnt{X_{{\mathrm{1}}}}})) \otimes
          \ldots
          \otimes
          \Mny(\mGLnt{r} \odot (\mGLnt{r_{\mGLmv{n}}} \odot \interpGS{\mGLnt{X_{\mGLmv{n}}}})) \otimes
          \Mny(\delta_{{\mathrm{2}}} \odot \interpGS{\Delta_{{\mathrm{3}}}}) \otimes \interpMS{\Gamma}  \\
          \xrightarrow{id \otimes \overline{\mathsf{m}^\Mny} \otimes id} \;\; &
          \Mny(\delta_{{\mathrm{1}}} \odot \interpGS{\Delta_{{\mathrm{1}}}}) \otimes
          \Mny((\mGLnt{r} \odot (\mGLnt{r_{{\mathrm{0}}}} \odot \interpGS{\mGLnt{X_{{\mathrm{1}}}}})) \boxtimes
          \ldots
          \boxtimes
          (\mGLnt{r} \odot (\mGLnt{r_{\mGLmv{n}}} \odot \interpGS{\mGLnt{X_{\mGLmv{n}}}}))) \otimes
          \Mny(\delta_{{\mathrm{2}}} \odot \interpGS{\Delta_{{\mathrm{3}}}}) \otimes \interpMS{\Gamma}
          \\
\xymatrix{\ar@{=}[r]^{id \otimes \overline{\mathsf{m}_{\boxtimes, r}} \otimes id} & } \;\; &
          \Mny(\delta_{{\mathrm{1}}} \odot \interpGS{\Delta_{{\mathrm{1}}}}) \otimes
          \Mny(\mGLnt{r} \odot ((\mGLnt{r_{{\mathrm{0}}}} \odot \interpGS{\mGLnt{X_{{\mathrm{1}}}}}) \boxtimes
          \ldots
          \boxtimes (\mGLnt{r_{\mGLmv{n}}} \odot \interpGS{\mGLnt{X_{\mGLmv{n}}}}))) \otimes
          \Mny(\delta_{{\mathrm{2}}} \odot \interpGS{\Delta_{{\mathrm{3}}}}) \otimes \interpMS{\Gamma}
          \\
\xrightarrow{id \otimes \Mny(id_r \odot \interpGS{\Pi_1}) \otimes id}
          \;\;&
         \Mny(\delta_{{\mathrm{1}}} \odot \interpGS{\Delta_{{\mathrm{1}}}}) \otimes
          \Mny(\mGLnt{r} \odot \interpGS{\mGLnt{X}}) \otimes
          \Mny(\delta_{{\mathrm{2}}} \odot \interpGS{\Delta_{{\mathrm{3}}}}) \otimes \interpMS{\Gamma} \\
\xrightarrow{\interpMS{\Pi_2}} &  \interpMS{B}
       \end{align*}
      \end{gather*}

      This case follows by applying $\Mny($ to the first
      morphism obtained from the induction hypothesis, and then
      composing it with the second morphism obtained from the
      induction hypothesis.

      %
      %
    \item (\mGLdruleMSTXXWeakName{})
      \begin{align*}
        & \ruleinterp{\mGLdruleMSTXXWeak{}}
      \end{align*}
      By induction we have the following morphism:
      \[
      \Mny(\delta_{{\mathrm{1}}} \odot \interpGS{\Delta_{{\mathrm{1}}}}) \otimes \Mny(\delta_{{\mathrm{2}}} \odot \interpGS{\Delta_{{\mathrm{2}}}}) \otimes \interpMS{\Gamma} \mto^{\interpMS{\Pi}} \interpMS{B}
      \]
      The final interpretation is as follows:
      \[
      \begin{array}{rll}
        & \Mny(\delta_{{\mathrm{1}}} \odot \interpGS{\Delta_{{\mathrm{1}}}}) \otimes \Mny(0 \odot \interpGS{X}) \otimes \Mny(\delta_{{\mathrm{2}}} \odot \interpGS{\Delta_{{\mathrm{2}}}}) \otimes \interpMS{\Gamma}\\
        \mto^{\id \otimes \Mny(\mathsf{weak}) \otimes \id} & \Mny(\delta_{{\mathrm{1}}} \odot \interpGS{\Delta_{{\mathrm{1}}}}) \otimes \Mny \mathsf{J} \otimes \Mny(\delta_{{\mathrm{2}}} \odot \interpGS{\Delta_{{\mathrm{2}}}}) \otimes \interpMS{\Gamma}\\
        \mto^{n^\Mny_I} & \Mny(\delta_{{\mathrm{1}}} \odot \interpGS{\Delta_{{\mathrm{1}}}}) \otimes \mathsf{I} \otimes \Mny(\delta_{{\mathrm{2}}} \odot \interpGS{\Delta_{{\mathrm{2}}}}) \otimes \interpMS{\Gamma}\\
        = & \Mny(\delta_{{\mathrm{1}}} \odot \interpGS{\Delta_{{\mathrm{1}}}}) \otimes \Mny(\delta_{{\mathrm{2}}} \odot \interpGS{\Delta_{{\mathrm{2}}}}) \otimes \interpMS{\Gamma}\\
        \mto^{\interpMS{\Pi}} & \interpMS{B}
      \end{array}
      \]

      %
      %
    \item (\mGLdruleMSTXXContName{})
      \begin{align*}
        & \ruleinterp{\mGLdruleMSTXXCont{}}
      \end{align*}
      By induction we have the following morphism:
      \[
      \Mny(\delta_{{\mathrm{1}}} \odot \interpGS{\Delta_{{\mathrm{1}}}}) \otimes \Mny(\mGLnt{r_{{\mathrm{1}}}} \odot \interpGS{\mGLnt{X}}) \otimes \Mny(\mGLnt{r_{{\mathrm{2}}}} \odot \interpGS{\mGLnt{X}}) \otimes \Mny(\delta_{{\mathrm{2}}} \odot \interpGS{\Delta_{{\mathrm{2}}}}) \otimes \interpMS{\Gamma} \mto^{\interpMS{\Pi}} \interpMS{B}
      \]
      The final interpretation is as follows:
      \begin{gather*}
      \begin{array}{rll}
        & \Mny(\delta_{{\mathrm{1}}} \odot \interpGS{\Delta_{{\mathrm{1}}}}) \otimes \Mny((\mGLnt{r_{{\mathrm{1}}}}  +  \mGLnt{r_{{\mathrm{2}}}}) \odot \interpGS{\mGLnt{X}}) \otimes \Mny(\delta_{{\mathrm{2}}} \odot \interpGS{\Delta_{{\mathrm{2}}}}) \otimes \interpMS{\Gamma}\\
        \mto^{\id \otimes \Mny(\mathsf{contr}) \otimes \id \otimes \id} & \Mny(\delta_{{\mathrm{1}}} \odot \interpGS{\Delta_{{\mathrm{1}}}}) \otimes \Mny((\mGLnt{r_{{\mathrm{1}}}} \odot \interpGS{\mGLnt{X}}) \boxtimes (\mGLnt{r_{{\mathrm{2}}}} \odot \interpGS{\mGLnt{X}})) \otimes \Mny(\delta_{{\mathrm{2}}} \odot \interpGS{\Delta_{{\mathrm{2}}}}) \otimes \interpMS{\Gamma}\\
        \mto^{\id \otimes n^\Mny_{\boxtimes, r} \otimes \id \otimes \id} & \Mny(\delta_{{\mathrm{1}}} \odot \interpGS{\Delta_{{\mathrm{1}}}}) \otimes \Mny(\mGLnt{r_{{\mathrm{1}}}} \odot \interpGS{\mGLnt{X}}) \otimes \Mny(\mGLnt{r_{{\mathrm{2}}}} \odot \interpGS{\mGLnt{X}}) \otimes \Mny(\delta_{{\mathrm{2}}} \odot \interpGS{\Delta_{{\mathrm{2}}}}) \otimes \interpMS{\Gamma}\\
        \mto^{\interpMS{\Pi}} & \interpMS{\mGLnt{B}}
      \end{array}
      \end{gather*}

      %
      %
    \item (\mGLdruleMSTXXExName{})
      \begin{align*}
        & \ruleinterp{\mGLdruleMSTXXEx{}}
      \end{align*}
      This case follows from symmetry of $\otimes$.

      %
      %
    \item (\mGLdruleMSTXXGExName{})
      \begin{align*}
        & \ruleinterp{\mGLdruleMSTXXGEx{}}
      \end{align*}
      This case follows from symmetry of $\otimes$.

    \end{itemize}

\end{itemize}
\end{definition}

%% file: mGL-soundness-theorem-ottput.tex
\subsection*{Soundness proof}
\mGLLSoundTheorem{}

Note that, by the definition of $\equiv$ which includes the source- and target of
cut-reductions, then we first proof the following sub-lemma before we prove
the remaining rules of the equational theory not induced by cut-reduction
(Section~\ref{subsec:proof-soundness-eq-theory-rest}):

\begin{lemma}
  Suppose $\cat{C} : \Mny \dashv
  \func{Lin} : \cat{M}$ with $\odot : \cat{R}^{op} \times \cat{C} \rightarrow \cat{C}$
  a mixed graded/linear model.
  Then:
  \begin{enumerate}
  \item If using the cut-reduction strategy a derivation
  $\Pi_{{\mathrm{1}}}$ of
  $\delta_{{\mathrm{1}}}  \odot  \Delta  \vdash_{\mathsf{GS} }  \mGLnt{t_{{\mathrm{1}}}}  \mGLsym{:}  \mGLnt{X}$

    reduces to
    the derivation $\Pi_{{\mathrm{2}}}$ of
  $\delta_{{\mathrm{2}}}  \odot  \Delta  \vdash_{\mathsf{GS} }  \mGLnt{t_{{\mathrm{2}}}}  \mGLsym{:}  \mGLnt{X}$,
  then
  ${\interp{ \Pi_{{\mathrm{2}}} }} : \interp{ \Delta }_{ \delta_{{\mathrm{1}}} } \mto \interp{ \mGLnt{X} }
  = {\interp{ \Pi_{{\mathrm{2}}} }} : \interp{ \Delta }_{ \delta_{{\mathrm{2}}} } \mto \interp{ \mGLnt{X} }$
  with $\delta_{{\mathrm{1}}}  \mGLsym{=}  \delta_{{\mathrm{2}}}$.

\item If using the cut-reduction strategy a derivation
  $\Pi_{{\mathrm{1}}}$ of
  $\delta_{{\mathrm{1}}}  \odot  \Delta  \mGLsym{;}  \Gamma  \vdash_{\mathsf{MS} }  \mGLnt{l_{{\mathrm{1}}}}  \mGLsym{:}  \mGLnt{A}$
  reduces to the derivation $\Pi_{{\mathrm{2}}}$ of
  $\delta_{{\mathrm{2}}}  \odot  \Delta  \mGLsym{;}  \Gamma  \vdash_{\mathsf{MS} }  \mGLnt{l_{{\mathrm{2}}}}  \mGLsym{:}  \mGLnt{A}$, then
  ${\interp{ \Pi_{{\mathrm{1}}} }} : \interp{ \Delta ; \Gamma }_{ \delta_{{\mathrm{1}}} } \mto \interp{ \mGLnt{A} }
  = \interp{ \Pi_{{\mathrm{2}}} } : \interp{ \Delta ; \Gamma }_{ \delta_{{\mathrm{2}}} } \mto \interp{ \mGLnt{A} }$
  with $\delta_{{\mathrm{1}}}  \mGLsym{=}  \delta_{{\mathrm{2}}}$.
  \end{enumerate}
\end{lemma}

\begin{proof}
  This is a proof by mutual induction on the form of the cut-reduction
  strategy.
  \begin{enumerate}
  \item The derivation:
    \begin{gather*}
    \inferrule* [flushleft,right=$\mGLdruleGSTXXCutName{}$] {
      \delta  \odot  \Delta  \vdash_{\mathsf{GS} }  \mGLnt{t_{{\mathrm{1}}}}  \mGLsym{:}  \mGLnt{X}\\
      \inferrule* [flushleft,right=$\mGLdruleGSTXXCutName{}$] {
        (  \delta_{{\mathrm{2}}}  \mGLsym{,}  \mGLnt{r_{{\mathrm{1}}}}  \mGLsym{,}  \delta_{{\mathrm{3}}}  )   \odot   ( \Delta_{{\mathrm{2}}}  \mGLsym{,}  \mGLnt{X}  \mGLsym{,}  \Delta_{{\mathrm{3}}} )   \vdash_{\mathsf{GS} }  \mGLnt{t_{{\mathrm{2}}}}  \mGLsym{:}  \mGLnt{Y}\\
        (  \delta_{{\mathrm{1}}}  \mGLsym{,}  \mGLnt{r_{{\mathrm{2}}}}  \mGLsym{,}  \delta_{{\mathrm{4}}}  )   \odot   ( \Delta_{{\mathrm{1}}}  \mGLsym{,}  \mGLnt{Y}  \mGLsym{,}  \Delta_{{\mathrm{4}}} )   \vdash_{\mathsf{GS} }  \mGLnt{t}  \mGLsym{:}  \mGLnt{Z}
      }{(  \delta_{{\mathrm{1}}}  \mGLsym{,}   \mGLnt{r_{{\mathrm{2}}}}   *   \delta_{{\mathrm{2}}}   \mGLsym{,}   \mGLnt{r_{{\mathrm{2}}}}   *   \mGLnt{r_{{\mathrm{1}}}}   \mGLsym{,}   \mGLnt{r_{{\mathrm{2}}}}   *   \delta_{{\mathrm{3}}}   \mGLsym{,}  \delta_{{\mathrm{4}}}  )   \odot   ( \Delta_{{\mathrm{1}}}  \mGLsym{,}  \Delta_{{\mathrm{2}}}  \mGLsym{,}  \mGLnt{X}  \mGLsym{,}  \Delta_{{\mathrm{3}}}  \mGLsym{,}  \Delta_{{\mathrm{4}}} )   \vdash_{\mathsf{GS} }  \mGLsym{[}  \mGLnt{t_{{\mathrm{2}}}}  \mGLsym{/}  \mGLmv{y}  \mGLsym{]}  \mGLnt{t}  \mGLsym{:}  \mGLnt{Z}}
    }{(  \delta_{{\mathrm{1}}}  \mGLsym{,}   \mGLnt{r_{{\mathrm{2}}}}   *   \delta_{{\mathrm{2}}}   \mGLsym{,}    \mGLnt{r_{{\mathrm{2}}}}   *   \mGLnt{r_{{\mathrm{1}}}}    *   \delta   \mGLsym{,}   \mGLnt{r_{{\mathrm{2}}}}   *   \delta_{{\mathrm{3}}}   \mGLsym{,}  \delta_{{\mathrm{4}}}  )   \odot   ( \Delta_{{\mathrm{1}}}  \mGLsym{,}  \Delta_{{\mathrm{2}}}  \mGLsym{,}  \Delta  \mGLsym{,}  \Delta_{{\mathrm{3}}}  \mGLsym{,}  \Delta_{{\mathrm{4}}} )   \vdash_{\mathsf{GS} }  \mGLsym{[}  \mGLnt{t_{{\mathrm{1}}}}  \mGLsym{/}  \mGLmv{x}  \mGLsym{]}  \mGLsym{[}  \mGLnt{t_{{\mathrm{2}}}}  \mGLsym{/}  \mGLmv{y}  \mGLsym{]}  \mGLnt{t}  \mGLsym{:}  \mGLnt{Z}}
    \end{gather*}
    reduces to the derivation:
    \begin{gather*}
    \inferrule* [flushleft,right=$\mGLdruleGSTXXCutName{}$] {
      \inferrule* [flushleft,right=$\mGLdruleGSTXXCutName{}$] {
        \delta  \odot  \Delta  \vdash_{\mathsf{GS} }  \mGLnt{t_{{\mathrm{1}}}}  \mGLsym{:}  \mGLnt{X}\\
        (  \delta_{{\mathrm{2}}}  \mGLsym{,}  \mGLnt{r_{{\mathrm{1}}}}  \mGLsym{,}  \delta_{{\mathrm{3}}}  )   \odot   ( \Delta_{{\mathrm{2}}}  \mGLsym{,}  \mGLnt{X}  \mGLsym{,}  \Delta_{{\mathrm{3}}} )   \vdash_{\mathsf{GS} }  \mGLnt{t_{{\mathrm{2}}}}  \mGLsym{:}  \mGLnt{Y}
      }{(  \delta_{{\mathrm{2}}}  \mGLsym{,}   \mGLnt{r_{{\mathrm{1}}}}   *   \delta   \mGLsym{,}  \delta_{{\mathrm{3}}}  )   \odot   ( \Delta_{{\mathrm{2}}}  \mGLsym{,}  \Delta  \mGLsym{,}  \Delta_{{\mathrm{3}}} )   \vdash_{\mathsf{GS} }  \mGLsym{[}  \mGLnt{t_{{\mathrm{1}}}}  \mGLsym{/}  \mGLmv{x}  \mGLsym{]}  \mGLnt{t_{{\mathrm{2}}}}  \mGLsym{:}  \mGLnt{Y}}\\
      (  \delta_{{\mathrm{1}}}  \mGLsym{,}  \mGLnt{r_{{\mathrm{2}}}}  \mGLsym{,}  \delta_{{\mathrm{4}}}  )   \odot   ( \Delta_{{\mathrm{1}}}  \mGLsym{,}  \mGLnt{Y}  \mGLsym{,}  \Delta_{{\mathrm{4}}} )   \vdash_{\mathsf{GS} }  \mGLnt{t}  \mGLsym{:}  \mGLnt{Z}
    }{(  \delta_{{\mathrm{1}}}  \mGLsym{,}   \mGLnt{r_{{\mathrm{2}}}}   *   \delta_{{\mathrm{2}}}   \mGLsym{,}   \mGLnt{r_{{\mathrm{2}}}}   *    \mGLnt{r_{{\mathrm{1}}}}   *   \delta    \mGLsym{,}   \mGLnt{r_{{\mathrm{2}}}}   *   \delta_{{\mathrm{3}}}   \mGLsym{,}  \delta_{{\mathrm{4}}}  )   \odot   ( \Delta_{{\mathrm{1}}}  \mGLsym{,}  \Delta_{{\mathrm{2}}}  \mGLsym{,}  \Delta  \mGLsym{,}  \Delta_{{\mathrm{3}}}  \mGLsym{,}  \Delta_{{\mathrm{4}}} )   \vdash_{\mathsf{GS} }  \mGLsym{[}  \mGLsym{[}  \mGLnt{t_{{\mathrm{1}}}}  \mGLsym{/}  \mGLmv{x}  \mGLsym{]}  \mGLnt{t_{{\mathrm{2}}}}  \mGLsym{/}  \mGLmv{y}  \mGLsym{]}  \mGLnt{t}  \mGLsym{:}  \mGLnt{Z}}
    \end{gather*}
    This case follows from associativity of composition of morphisms:
    \begin{align*}
    (lhs) \;\; & (\interp{\Pi_3} \circ (\id_{\delta_{{\mathrm{1}}}   \odot  \interp{  \Delta_{{\mathrm{1}}}  }} \boxtimes (\mGLnt{r_{{\mathrm{2}}}} \odot \interp{\Pi_2}) \boxtimes \id_{\delta_{{\mathrm{4}}}   \odot  \interp{  \Delta_{{\mathrm{4}}}  }})) \\
      & \qquad \circ (\id_{(  \delta_{{\mathrm{1}}}  \mGLsym{,}  \mGLnt{r_{{\mathrm{2}}}}  *  \delta_{{\mathrm{2}}}  )    \odot  \interp{  \Delta_{{\mathrm{1}}}  \mGLsym{,}  \Delta_{{\mathrm{2}}}  }} \boxtimes (\mGLnt{r_{{\mathrm{2}}}}  *  \mGLnt{r_{{\mathrm{1}}}}) \odot \interp{\Pi_1} \boxtimes \id_{(  \mGLnt{r_{{\mathrm{2}}}}  *  \delta_{{\mathrm{3}}}  \mGLsym{,}  \delta_{{\mathrm{4}}}  )    \odot  \interp{  \Delta_{{\mathrm{3}}}  \mGLsym{,}  \Delta_{{\mathrm{4}}}  }})  \\
    (rhs) \;\; & \interp{\Pi_3} \circ (\id_{\delta_{{\mathrm{1}}}   \odot  \interp{  \Delta_{{\mathrm{1}}}  }} \boxtimes (\mGLnt{r_{{\mathrm{2}}}} \odot (\interp{\Pi_2} \circ (\id_{\delta_{{\mathrm{2}}}   \odot  \interp{  \Delta_{{\mathrm{2}}}  }} \boxtimes (\mGLnt{r_{{\mathrm{1}}}} \odot \interp{\Pi_1}) \boxtimes \id_{\delta_{{\mathrm{3}}}   \odot  \interp{  \Delta_{{\mathrm{3}}}  }}))) \boxtimes \id_{\delta_{{\mathrm{4}}}   \odot  \interp{  \Delta_{{\mathrm{4}}}  }})
    \end{align*}

      \item The derivation:
    \begin{gather*}
    \inferrule* [flushleft,right=$\mGLdruleGSTXXCutName{}$] {
      \delta  \odot  \Delta  \vdash_{\mathsf{GS} }  \mGLnt{t_{{\mathrm{1}}}}  \mGLsym{:}  \mGLnt{X}\\
      \inferrule* [flushleft,right=$\mGLdruleGSTXXCutName{}$] {
        \delta'  \odot  \Delta'  \vdash_{\mathsf{GS} }  \mGLnt{t_{{\mathrm{2}}}}  \mGLsym{:}  \mGLnt{Y}\\
        (  \delta_{{\mathrm{1}}}  \mGLsym{,}  \mGLnt{r_{{\mathrm{1}}}}  \mGLsym{,}  \delta_{{\mathrm{2}}}  \mGLsym{,}  \mGLnt{r_{{\mathrm{2}}}}  \mGLsym{,}  \delta_{{\mathrm{3}}}  )   \odot   ( \Delta_{{\mathrm{1}}}  \mGLsym{,}  \mGLmv{x}  \mGLsym{:}  \mGLnt{X}  \mGLsym{,}  \Delta_{{\mathrm{2}}}  \mGLsym{,}  \mGLmv{y}  \mGLsym{:}  \mGLnt{Y}  \mGLsym{,}  \Delta_{{\mathrm{3}}} )   \vdash_{\mathsf{GS} }  \mGLnt{t}  \mGLsym{:}  \mGLnt{Z}
      }{(  \delta_{{\mathrm{1}}}  \mGLsym{,}  \mGLnt{r_{{\mathrm{1}}}}  \mGLsym{,}  \delta_{{\mathrm{2}}}  \mGLsym{,}   \mGLnt{r_{{\mathrm{2}}}}   *   \delta'   \mGLsym{,}  \delta_{{\mathrm{3}}}  )   \odot   ( \Delta_{{\mathrm{1}}}  \mGLsym{,}  \mGLmv{x}  \mGLsym{:}  \mGLnt{X}  \mGLsym{,}  \Delta_{{\mathrm{2}}}  \mGLsym{,}  \Delta'  \mGLsym{,}  \Delta_{{\mathrm{3}}} )   \vdash_{\mathsf{GS} }  \mGLsym{[}  \mGLnt{t_{{\mathrm{2}}}}  \mGLsym{/}  \mGLmv{y}  \mGLsym{]}  \mGLnt{t}  \mGLsym{:}  \mGLnt{Z}}
    }{(  \delta_{{\mathrm{1}}}  \mGLsym{,}   \mGLnt{r_{{\mathrm{1}}}}   *   \delta   \mGLsym{,}  \delta_{{\mathrm{2}}}  \mGLsym{,}   \mGLnt{r_{{\mathrm{2}}}}   *   \delta'   \mGLsym{,}  \delta_{{\mathrm{3}}}  )   \odot   ( \Delta_{{\mathrm{1}}}  \mGLsym{,}  \Delta  \mGLsym{,}  \Delta_{{\mathrm{2}}}  \mGLsym{,}  \Delta'  \mGLsym{,}  \Delta_{{\mathrm{3}}} )   \vdash_{\mathsf{GS} }  \mGLsym{[}  \mGLnt{t_{{\mathrm{1}}}}  \mGLsym{/}  \mGLmv{x}  \mGLsym{]}  \mGLsym{[}  \mGLnt{t_{{\mathrm{2}}}}  \mGLsym{/}  \mGLmv{y}  \mGLsym{]}  \mGLnt{t}  \mGLsym{:}  \mGLnt{Z}}
    \end{gather*}
    reduces to the derivation:
    \begin{gather*}
    \inferrule* [flushleft,right=$\mGLdruleGSTXXCutName{}$] {
      \delta'  \odot  \Delta'  \vdash_{\mathsf{GS} }  \mGLnt{t_{{\mathrm{2}}}}  \mGLsym{:}  \mGLnt{Y}\\
      \inferrule* [flushleft,right=$\mGLdruleGSTXXCutName{}$] {
        \delta  \odot  \Delta  \vdash_{\mathsf{GS} }  \mGLnt{t_{{\mathrm{1}}}}  \mGLsym{:}  \mGLnt{X}\\
        (  \delta_{{\mathrm{1}}}  \mGLsym{,}  \mGLnt{r_{{\mathrm{1}}}}  \mGLsym{,}  \delta_{{\mathrm{2}}}  \mGLsym{,}  \mGLnt{r_{{\mathrm{2}}}}  \mGLsym{,}  \delta_{{\mathrm{3}}}  )   \odot   ( \Delta_{{\mathrm{1}}}  \mGLsym{,}  \mGLmv{x}  \mGLsym{:}  \mGLnt{X}  \mGLsym{,}  \Delta_{{\mathrm{2}}}  \mGLsym{,}  \mGLmv{y}  \mGLsym{:}  \mGLnt{Y}  \mGLsym{,}  \Delta_{{\mathrm{3}}} )   \vdash_{\mathsf{GS} }  \mGLnt{t}  \mGLsym{:}  \mGLnt{Z}
      }{(  \delta_{{\mathrm{1}}}  \mGLsym{,}  \mGLnt{r_{{\mathrm{1}}}}  \mGLsym{,}  \delta_{{\mathrm{2}}}  \mGLsym{,}   \mGLnt{r_{{\mathrm{2}}}}   *   \delta'   \mGLsym{,}  \delta_{{\mathrm{3}}}  )   \odot   ( \Delta_{{\mathrm{1}}}  \mGLsym{,}  \mGLmv{x}  \mGLsym{:}  \mGLnt{X}  \mGLsym{,}  \Delta_{{\mathrm{2}}}  \mGLsym{,}  \Delta'  \mGLsym{,}  \Delta_{{\mathrm{3}}} )   \vdash_{\mathsf{GS} }  \mGLsym{[}  \mGLnt{t_{{\mathrm{1}}}}  \mGLsym{/}  \mGLmv{x}  \mGLsym{]}  \mGLnt{t}  \mGLsym{:}  \mGLnt{Z}}
    }{(  \delta_{{\mathrm{1}}}  \mGLsym{,}   \mGLnt{r_{{\mathrm{1}}}}   *   \delta   \mGLsym{,}  \delta_{{\mathrm{2}}}  \mGLsym{,}   \mGLnt{r_{{\mathrm{2}}}}   *   \delta'   \mGLsym{,}  \delta_{{\mathrm{3}}}  )   \odot   ( \Delta_{{\mathrm{1}}}  \mGLsym{,}  \Delta  \mGLsym{,}  \Delta_{{\mathrm{2}}}  \mGLsym{,}  \Delta'  \mGLsym{,}  \Delta_{{\mathrm{3}}} )   \vdash_{\mathsf{GS} }  \mGLsym{[}  \mGLnt{t_{{\mathrm{2}}}}  \mGLsym{/}  \mGLmv{y}  \mGLsym{]}  \mGLsym{[}  \mGLnt{t_{{\mathrm{1}}}}  \mGLsym{/}  \mGLmv{x}  \mGLsym{]}  \mGLnt{t}  \mGLsym{:}  \mGLnt{Z}}
    \end{gather*}
    This case follows from associativity of composition of morphisms (similarly to the above).

  \item The derivation:
    \[
    \inferrule* [flushleft,right=$\mGLdruleGSTXXidName{}$] {
      \,
    }{1  \odot  \mGLmv{x}  \mGLsym{:}  \mGLsym{(}  \mGLnt{X}  \boxtimes  \mGLnt{Y}  \mGLsym{)}  \vdash_{\mathsf{GS} }  \mGLmv{x}  \mGLsym{:}  \mGLnt{X}  \boxtimes  \mGLnt{Y}}
    \]
    expands to the derivation:
    \[
    \inferrule* [flushleft,right=$\mGLdruleGSTXXTenLName{}$] {
      \inferrule* [flushleft,right=$\mGLdruleGSTXXTenRName{}$] {
        \inferrule* [flushleft,right=$\mGLdruleGSTXXidName{}$] {}{1  \odot  \mGLmv{y}  \mGLsym{:}  \mGLnt{X}  \vdash_{\mathsf{GS} }  \mGLmv{y}  \mGLsym{:}  \mGLnt{X}}\\
        \inferrule* [flushleft,right=$\mGLdruleGSTXXidName{}$] {}{1  \odot  \mGLmv{z}  \mGLsym{:}  \mGLnt{Y}  \vdash_{\mathsf{GS} }  \mGLmv{z}  \mGLsym{:}  \mGLnt{Y}}
      }{(  1  \mGLsym{,}  1  )   \odot   ( \mGLmv{y}  \mGLsym{:}  \mGLnt{X}  \mGLsym{,}  \mGLmv{z}  \mGLsym{:}  \mGLnt{Y} )   \vdash_{\mathsf{GS} }  \mGLsym{(}  \mGLmv{y}  \mGLsym{,}  \mGLmv{z}  \mGLsym{)}  \mGLsym{:}  \mGLnt{X}  \boxtimes  \mGLnt{Y}}
    }{(  1  )   \odot   ( \mGLmv{x}  \mGLsym{:}  \mGLnt{X}  \boxtimes  \mGLnt{Y} )   \vdash_{\mathsf{GS} }   \mathsf{let} \,( \mGLmv{y} , \mGLmv{z} ) =  \mGLmv{x} \, \mathsf{in} \, \mGLsym{(}  \mGLmv{y}  \mGLsym{,}  \mGLmv{z}  \mGLsym{)}   \mGLsym{:}  \mGLnt{X}  \boxtimes  \mGLnt{Y}}
    \]
    This case follows from the fact that $\odot$ is a bifunctor
    and that $n_{\boxtimes,r,X,Y}$ is strict, such that:
    \begin{align*}
    \begin{array}{rll}
    (rhs)  & & \;\; (id_{\interpGS{X}} \boxtimes id_{\interpGS{Y}}) \circ n_{\boxtimes,1,\interpGS{X},\interpGS{Y}} \\
           & = & \;\; (id_{\interpGS{X}} \boxtimes id_{\interpGS{Y}}) \\
     (lhs) & = & \;\; id_{{\interpGS{X}} \boxtimes \interpGS{Y}} \\
     \end{array}
       \end{align*}

      \item The derivation:
    \[
    \inferrule* [flushleft,right=$\mGLdruleGSTXXidName{}$] {
      \,
    }{1  \odot  \mGLmv{x}  \mGLsym{:}  \mathsf{J}  \vdash_{\mathsf{GS} }  \mGLmv{x}  \mGLsym{:}  \mathsf{J}}
    \]
    reduces to the derivation:
    \[
    \inferrule* [flushleft,right=$\mGLdruleGSTXXUnitLName{}$] {
      \inferrule* [flushleft,right=$\mGLdruleGSTXXUnitRName{}$] {
        \,
      }{\emptyset  \odot  \emptyset  \vdash_{\mathsf{GS} }  \mathsf{j}  \mGLsym{:}  \mathsf{J}}
    }{1  \odot  \mGLmv{x}  \mGLsym{:}  \mathsf{J}  \vdash_{\mathsf{GS} }  \mathsf{let} \, \mathsf{j} \, \mGLsym{=}  \mGLmv{x} \, \mathsf{in} \, \mathsf{j}  \mGLsym{:}  \mathsf{J}}
    \]
    The interpretations are then:
    \begin{align*}
    (lhs) & = id_{\mathsf{J}} \\
    (rhs) & = id_{\mathsf{J}} \circ \circ n_{\mathsf{J},1}
    \end{align*}
    which are equivalent by the strict monoidality of $n$.

  \item The derivation:
    \begin{gather*}
    \inferrule* [flushleft,right=$\mGLdruleGSTXXCutName{}$] {
      \inferrule* [flushleft,right=$\mGLdruleGSTXXidName{}$] {
        \,
      }{1  \odot  \mGLmv{x}  \mGLsym{:}  \mGLnt{X}  \vdash_{\mathsf{GS} }  \mGLmv{x}  \mGLsym{:}  \mGLnt{X}}\\
      (  \delta_{{\mathrm{1}}}  \mGLsym{,}  \mGLnt{r}  \mGLsym{,}  \delta_{{\mathrm{2}}}  )   \odot   ( \Delta_{{\mathrm{1}}}  \mGLsym{,}  \mGLmv{y}  \mGLsym{:}  \mGLnt{X}  \mGLsym{,}  \Delta_{{\mathrm{2}}} )   \vdash_{\mathsf{GS} }  \mGLnt{t}  \mGLsym{:}  \mGLnt{Z}
    }{(  \delta_{{\mathrm{1}}}  \mGLsym{,}   \mGLnt{r}   *   1   \mGLsym{,}  \delta_{{\mathrm{2}}}  )   \odot   ( \Delta_{{\mathrm{1}}}  \mGLsym{,}  \mGLmv{x}  \mGLsym{:}  \mGLnt{X}  \mGLsym{,}  \Delta_{{\mathrm{2}}} )   \vdash_{\mathsf{GS} }  \mGLsym{[}  \mGLmv{x}  \mGLsym{/}  \mGLmv{y}  \mGLsym{]}  \mGLnt{t}  \mGLsym{:}  \mGLnt{Z}}
    \end{gather*}
    reduces to:
    \[
      (  \delta_{{\mathrm{1}}}  \mGLsym{,}  \mGLnt{r}  \mGLsym{,}  \delta_{{\mathrm{2}}}  )   \odot   ( \Delta_{{\mathrm{1}}}  \mGLsym{,}  \mGLmv{y}  \mGLsym{:}  \mGLnt{X}  \mGLsym{,}  \Delta_{{\mathrm{2}}} )   \vdash_{\mathsf{GS} }  \mGLnt{t}  \mGLsym{:}  \mGLnt{Z}
    \]
    This case follows from the identity law of composition of morphisms.

  \item The derivation:
    \begin{gather*}
    \inferrule* [flushleft,right=$\mGLdruleGSTXXCutName{}$] {
      \delta  \odot  \Delta  \vdash_{\mathsf{GS} }  \mGLnt{t}  \mGLsym{:}  \mGLnt{X}\\
      \inferrule* [flushleft,right=$\mGLdruleGSTXXidName{}$] {
        \,
      }{1  \odot  \mGLmv{x}  \mGLsym{:}  \mGLnt{X}  \vdash_{\mathsf{GS} }  \mGLmv{x}  \mGLsym{:}  \mGLnt{X}}
    }{1   *   \delta   \odot  \Delta  \vdash_{\mathsf{GS} }  \mGLsym{[}  \mGLnt{t}  \mGLsym{/}  \mGLmv{x}  \mGLsym{]}  \mGLmv{x}  \mGLsym{:}  \mGLnt{Z}}
    \end{gather*}
    reduces to:
    \[
      \delta  \odot  \Delta  \vdash_{\mathsf{GS} }  \mGLnt{t}  \mGLsym{:}  \mGLnt{X}
      \]
      This case follows from the identity law of composition of morphisms.

  \item The derivation:
    \begin{gather*}
    \inferrule* [flushleft,right=$\mGLdruleGSTXXCutName{}$] {
      \delta  \odot  \Delta  \vdash_{\mathsf{GS} }  \mGLnt{t_{{\mathrm{1}}}}  \mGLsym{:}  \mGLnt{X}\\
      \inferrule* [flushleft,right=$\mGLdruleGSTXXExName{}$] {
        (  \delta_{{\mathrm{1}}}  \mGLsym{,}  \mGLnt{r_{{\mathrm{1}}}}  \mGLsym{,}  \mGLnt{r_{{\mathrm{2}}}}  \mGLsym{,}  \delta_{{\mathrm{2}}}  )   \odot   ( \Delta_{{\mathrm{1}}}  \mGLsym{,}  \mGLmv{x}  \mGLsym{:}  \mGLnt{X}  \mGLsym{,}  \mGLmv{y}  \mGLsym{:}  \mGLnt{Y}  \mGLsym{,}  \Delta_{{\mathrm{2}}} )   \vdash_{\mathsf{GS} }  \mGLnt{t_{{\mathrm{2}}}}  \mGLsym{:}  \mGLnt{Z}
      }{(  \delta_{{\mathrm{1}}}  \mGLsym{,}  \mGLnt{r_{{\mathrm{2}}}}  \mGLsym{,}  \mGLnt{r_{{\mathrm{1}}}}  \mGLsym{,}  \delta_{{\mathrm{2}}}  )   \odot   ( \Delta_{{\mathrm{1}}}  \mGLsym{,}  \mGLmv{y}  \mGLsym{:}  \mGLnt{Y}  \mGLsym{,}  \mGLmv{x}  \mGLsym{:}  \mGLnt{X}  \mGLsym{,}  \Delta_{{\mathrm{2}}} )   \vdash_{\mathsf{GS} }  \mGLnt{t_{{\mathrm{2}}}}  \mGLsym{:}  \mGLnt{Z}}
    }{(  \delta_{{\mathrm{1}}}  \mGLsym{,}  \mGLnt{r_{{\mathrm{2}}}}  \mGLsym{,}   \mGLnt{r_{{\mathrm{1}}}}   *   \delta   \mGLsym{,}  \delta_{{\mathrm{2}}}  )   \odot   ( \Delta_{{\mathrm{1}}}  \mGLsym{,}  \mGLmv{y}  \mGLsym{:}  \mGLnt{Y}  \mGLsym{,}  \Delta  \mGLsym{,}  \Delta_{{\mathrm{2}}} )   \vdash_{\mathsf{GS} }  \mGLsym{[}  \mGLnt{t_{{\mathrm{1}}}}  \mGLsym{/}  \mGLmv{x}  \mGLsym{]}  \mGLnt{t_{{\mathrm{2}}}}  \mGLsym{:}  \mGLnt{Z}}
    \end{gather*}
    reduces to:
    \begin{gather*}
    \inferrule* [flushleft,right=$\mGLdruleGSTXXExName{}*$] {
      \inferrule* [flushleft,right=$\mGLdruleGSTXXCutName{}$] {
        \delta  \odot  \Delta  \vdash_{\mathsf{GS} }  \mGLnt{t_{{\mathrm{1}}}}  \mGLsym{:}  \mGLnt{X}\\
        (  \delta_{{\mathrm{1}}}  \mGLsym{,}  \mGLnt{r_{{\mathrm{1}}}}  \mGLsym{,}  \mGLnt{r_{{\mathrm{2}}}}  \mGLsym{,}  \delta_{{\mathrm{2}}}  )   \odot   ( \Delta_{{\mathrm{1}}}  \mGLsym{,}  \mGLmv{x}  \mGLsym{:}  \mGLnt{X}  \mGLsym{,}  \mGLmv{y}  \mGLsym{:}  \mGLnt{Y}  \mGLsym{,}  \Delta_{{\mathrm{2}}} )   \vdash_{\mathsf{GS} }  \mGLnt{t_{{\mathrm{2}}}}  \mGLsym{:}  \mGLnt{Z}
      }{(  \delta_{{\mathrm{1}}}  \mGLsym{,}   \mGLnt{r_{{\mathrm{1}}}}   *   \delta   \mGLsym{,}  \mGLnt{r_{{\mathrm{2}}}}  \mGLsym{,}  \delta_{{\mathrm{2}}}  )   \odot   ( \Delta_{{\mathrm{1}}}  \mGLsym{,}  \Delta  \mGLsym{,}  \mGLmv{y}  \mGLsym{:}  \mGLnt{Y}  \mGLsym{,}  \Delta_{{\mathrm{2}}} )   \vdash_{\mathsf{GS} }  \mGLsym{[}  \mGLnt{t_{{\mathrm{1}}}}  \mGLsym{/}  \mGLmv{x}  \mGLsym{]}  \mGLnt{t_{{\mathrm{2}}}}  \mGLsym{:}  \mGLnt{Z}}
    }{(  \delta_{{\mathrm{1}}}  \mGLsym{,}  \mGLnt{r_{{\mathrm{2}}}}  \mGLsym{,}   \mGLnt{r_{{\mathrm{1}}}}   *   \delta   \mGLsym{,}  \delta_{{\mathrm{2}}}  )   \odot   ( \Delta_{{\mathrm{1}}}  \mGLsym{,}  \mGLmv{y}  \mGLsym{:}  \mGLnt{Y}  \mGLsym{,}  \Delta  \mGLsym{,}  \Delta_{{\mathrm{2}}} )   \vdash_{\mathsf{GS} }  \mGLsym{[}  \mGLnt{t_{{\mathrm{1}}}}  \mGLsym{/}  \mGLmv{x}  \mGLsym{]}  \mGLnt{t_{{\mathrm{2}}}}  \mGLsym{:}  \mGLnt{Z}}
    \end{gather*}
    This case follows from naturality of symmetry, and the fact that
    exchange is interpreted as the symmetry in the symmetric
    monoidal category $\cat{C}$.

  \item The derivation:
    \begin{gather*}
    \inferrule* [flushleft,right=$\mGLdruleGSTXXCutName{}$] {
      \delta  \odot  \Delta  \vdash_{\mathsf{GS} }  \mGLnt{t_{{\mathrm{1}}}}  \mGLsym{:}  \mGLnt{Y}\\
      \inferrule* [flushleft,right=$\mGLdruleGSTXXExName{}$] {
        (  \delta_{{\mathrm{1}}}  \mGLsym{,}  \mGLnt{r_{{\mathrm{1}}}}  \mGLsym{,}  \mGLnt{r_{{\mathrm{2}}}}  \mGLsym{,}  \delta_{{\mathrm{2}}}  )   \odot   ( \Delta_{{\mathrm{1}}}  \mGLsym{,}  \mGLmv{x}  \mGLsym{:}  \mGLnt{X}  \mGLsym{,}  \mGLmv{y}  \mGLsym{:}  \mGLnt{Y}  \mGLsym{,}  \Delta_{{\mathrm{2}}} )   \vdash_{\mathsf{GS} }  \mGLnt{t_{{\mathrm{2}}}}  \mGLsym{:}  \mGLnt{Z}
      }{(  \delta_{{\mathrm{1}}}  \mGLsym{,}  \mGLnt{r_{{\mathrm{2}}}}  \mGLsym{,}  \mGLnt{r_{{\mathrm{1}}}}  \mGLsym{,}  \delta_{{\mathrm{2}}}  )   \odot   ( \Delta_{{\mathrm{1}}}  \mGLsym{,}  \mGLmv{y}  \mGLsym{:}  \mGLnt{Y}  \mGLsym{,}  \mGLmv{x}  \mGLsym{:}  \mGLnt{X}  \mGLsym{,}  \Delta_{{\mathrm{2}}} )   \vdash_{\mathsf{GS} }  \mGLnt{t_{{\mathrm{2}}}}  \mGLsym{:}  \mGLnt{Z}}
    }{(  \delta_{{\mathrm{1}}}  \mGLsym{,}   \mGLnt{r_{{\mathrm{2}}}}   *   \delta   \mGLsym{,}  \mGLnt{r_{{\mathrm{1}}}}  \mGLsym{,}  \delta_{{\mathrm{2}}}  )   \odot   ( \Delta_{{\mathrm{1}}}  \mGLsym{,}  \Delta  \mGLsym{,}  \mGLmv{x}  \mGLsym{:}  \mGLnt{X}  \mGLsym{,}  \Delta_{{\mathrm{2}}} )   \vdash_{\mathsf{GS} }  \mGLsym{[}  \mGLnt{t_{{\mathrm{1}}}}  \mGLsym{/}  \mGLmv{x}  \mGLsym{]}  \mGLnt{t_{{\mathrm{2}}}}  \mGLsym{:}  \mGLnt{Z}}
    \end{gather*}
    reduces to:
    \begin{gather*}
    \inferrule* [flushleft,right=$\mGLdruleGSTXXExName{}*$] {
      \inferrule* [flushleft,right=$\mGLdruleGSTXXCutName{}$] {
        \delta  \odot  \Delta  \vdash_{\mathsf{GS} }  \mGLnt{t_{{\mathrm{1}}}}  \mGLsym{:}  \mGLnt{Y}\\
        (  \delta_{{\mathrm{1}}}  \mGLsym{,}  \mGLnt{r_{{\mathrm{1}}}}  \mGLsym{,}  \mGLnt{r_{{\mathrm{2}}}}  \mGLsym{,}  \delta_{{\mathrm{2}}}  )   \odot   ( \Delta_{{\mathrm{1}}}  \mGLsym{,}  \mGLmv{x}  \mGLsym{:}  \mGLnt{X}  \mGLsym{,}  \mGLmv{y}  \mGLsym{:}  \mGLnt{Y}  \mGLsym{,}  \Delta_{{\mathrm{2}}} )   \vdash_{\mathsf{GS} }  \mGLnt{t_{{\mathrm{2}}}}  \mGLsym{:}  \mGLnt{Z}
      }{(  \delta_{{\mathrm{1}}}  \mGLsym{,}  \mGLnt{r_{{\mathrm{1}}}}  \mGLsym{,}   \mGLnt{r_{{\mathrm{2}}}}   *   \delta   \mGLsym{,}  \delta_{{\mathrm{2}}}  )   \odot   ( \Delta_{{\mathrm{1}}}  \mGLsym{,}  \mGLmv{x}  \mGLsym{:}  \mGLnt{X}  \mGLsym{,}  \Delta  \mGLsym{,}  \Delta_{{\mathrm{2}}} )   \vdash_{\mathsf{GS} }  \mGLsym{[}  \mGLnt{t_{{\mathrm{1}}}}  \mGLsym{/}  \mGLmv{x}  \mGLsym{]}  \mGLnt{t_{{\mathrm{2}}}}  \mGLsym{:}  \mGLnt{Z}}
    }{(  \delta_{{\mathrm{1}}}  \mGLsym{,}   \mGLnt{r_{{\mathrm{2}}}}   *   \delta   \mGLsym{,}  \mGLnt{r_{{\mathrm{1}}}}  \mGLsym{,}  \delta_{{\mathrm{2}}}  )   \odot   ( \Delta_{{\mathrm{1}}}  \mGLsym{,}  \Delta  \mGLsym{,}  \mGLmv{x}  \mGLsym{:}  \mGLnt{X}  \mGLsym{,}  \Delta_{{\mathrm{2}}} )   \vdash_{\mathsf{GS} }  \mGLsym{[}  \mGLnt{t_{{\mathrm{1}}}}  \mGLsym{/}  \mGLmv{x}  \mGLsym{]}  \mGLnt{t_{{\mathrm{2}}}}  \mGLsym{:}  \mGLnt{Z}}
    \end{gather*}
    Similarly to the above,
    this case follows from naturality of symmetry, and the fact that
    exchange is interpreted as the symmetry in the symmetric
    monoidal category $\cat{C}$.

  \item The derivation:
    \begin{gather*}
    \inferrule* [flushleft,right=$\mGLdruleGSTXXCutName{}$] {
      \inferrule* [flushleft,right=$\mGLdruleGSTXXTenRName{}$] {
        (  \delta_{{\mathrm{2}}}  )   \odot   ( \Delta_{{\mathrm{2}}} )   \vdash_{\mathsf{GS} }  \mGLnt{t_{{\mathrm{1}}}}  \mGLsym{:}  \mGLnt{X}\\
        (  \delta_{{\mathrm{3}}}  )   \odot   ( \Delta_{{\mathrm{3}}} )   \vdash_{\mathsf{GS} }  \mGLnt{t_{{\mathrm{2}}}}  \mGLsym{:}  \mGLnt{Y}
      }{(  \delta_{{\mathrm{2}}}  \mGLsym{,}  \delta_{{\mathrm{3}}}  )   \odot   ( \Delta_{{\mathrm{2}}}  \mGLsym{,}  \Delta_{{\mathrm{3}}} )   \vdash_{\mathsf{GS} }  \mGLsym{(}  \mGLnt{t_{{\mathrm{1}}}}  \mGLsym{,}  \mGLnt{t_{{\mathrm{2}}}}  \mGLsym{)}  \mGLsym{:}  \mGLnt{X}  \boxtimes  \mGLnt{Y}}\\
      \inferrule* [flushleft,right=$\mGLdruleGSTXXTenLName{}$] {
        (  \delta_{{\mathrm{1}}}  \mGLsym{,}  \mGLnt{r}  \mGLsym{,}  \mGLnt{r}  \mGLsym{,}  \delta_{{\mathrm{4}}}  )   \odot   ( \Delta_{{\mathrm{1}}}  \mGLsym{,}  \mGLmv{x}  \mGLsym{:}  \mGLnt{X}  \mGLsym{,}  \mGLmv{y}  \mGLsym{:}  \mGLnt{Y}  \mGLsym{,}  \Delta_{{\mathrm{4}}} )   \vdash_{\mathsf{GS} }  \mGLnt{t_{{\mathrm{3}}}}  \mGLsym{:}  \mGLnt{Z}
      }{(  \delta_{{\mathrm{1}}}  \mGLsym{,}  \mGLnt{r}  \mGLsym{,}  \delta_{{\mathrm{4}}}  )   \odot   ( \Delta_{{\mathrm{1}}}  \mGLsym{,}  \mGLmv{z}  \mGLsym{:}  \mGLnt{X}  \boxtimes  \mGLnt{Y}  \mGLsym{,}  \Delta_{{\mathrm{4}}} )   \vdash_{\mathsf{GS} }   \mathsf{let} \,( \mGLmv{x} , \mGLmv{y} ) =  \mGLmv{z} \, \mathsf{in} \, \mGLnt{t_{{\mathrm{3}}}}   \mGLsym{:}  \mGLnt{Z}}
    }{(  \delta_{{\mathrm{1}}}  \mGLsym{,}   \mGLnt{r}   *    (  \delta_{{\mathrm{2}}}  \mGLsym{,}  \delta_{{\mathrm{3}}}  )    \mGLsym{,}  \delta_{{\mathrm{4}}}  )   \odot   ( \Delta_{{\mathrm{1}}}  \mGLsym{,}  \Delta_{{\mathrm{2}}}  \mGLsym{,}  \Delta_{{\mathrm{3}}}  \mGLsym{,}  \Delta_{{\mathrm{4}}} )   \vdash_{\mathsf{GS} }   \mathsf{let} \,( \mGLmv{x} , \mGLmv{y} ) =  \mGLsym{(}  \mGLnt{t_{{\mathrm{1}}}}  \mGLsym{,}  \mGLnt{t_{{\mathrm{2}}}}  \mGLsym{)} \, \mathsf{in} \, \mGLnt{t_{{\mathrm{3}}}}   \mGLsym{:}  \mGLnt{Z}}
    \end{gather*}
    reduces to:
    \begin{gather*}
    \inferrule* [flushleft,right=$\mGLdruleGSTXXCutName{}$] {
      (  \delta_{{\mathrm{2}}}  )   \odot   ( \Delta_{{\mathrm{2}}} )   \vdash_{\mathsf{GS} }  \mGLnt{t_{{\mathrm{1}}}}  \mGLsym{:}  \mGLnt{X}\\
      \inferrule* [flushleft,right=$\mGLdruleGSTXXCutName{}$] {
        (  \delta_{{\mathrm{3}}}  )   \odot   ( \Delta_{{\mathrm{3}}} )   \vdash_{\mathsf{GS} }  \mGLnt{t_{{\mathrm{2}}}}  \mGLsym{:}  \mGLnt{Y}\\
        (  \delta_{{\mathrm{1}}}  \mGLsym{,}  \mGLnt{r}  \mGLsym{,}  \mGLnt{r}  \mGLsym{,}  \delta_{{\mathrm{4}}}  )   \odot   ( \Delta_{{\mathrm{1}}}  \mGLsym{,}  \mGLmv{x}  \mGLsym{:}  \mGLnt{X}  \mGLsym{,}  \mGLmv{y}  \mGLsym{:}  \mGLnt{Y}  \mGLsym{,}  \Delta_{{\mathrm{4}}} )   \vdash_{\mathsf{GS} }  \mGLnt{t_{{\mathrm{3}}}}  \mGLsym{:}  \mGLnt{Z}
      }{(  \delta_{{\mathrm{1}}}  \mGLsym{,}  \mGLnt{r}  \mGLsym{,}   \mGLnt{r}   *   \delta_{{\mathrm{3}}}   \mGLsym{,}  \delta_{{\mathrm{4}}}  )   \odot   ( \Delta_{{\mathrm{1}}}  \mGLsym{,}  \mGLmv{x}  \mGLsym{:}  \mGLnt{X}  \mGLsym{,}  \Delta_{{\mathrm{3}}}  \mGLsym{,}  \Delta_{{\mathrm{4}}} )   \vdash_{\mathsf{GS} }  \mGLsym{[}  \mGLnt{t_{{\mathrm{2}}}}  \mGLsym{/}  \mGLmv{x}  \mGLsym{]}  \mGLnt{t_{{\mathrm{3}}}}  \mGLsym{:}  \mGLnt{Z}}
    }{(  \delta_{{\mathrm{1}}}  \mGLsym{,}   \mGLnt{r}   *   \delta_{{\mathrm{2}}}   \mGLsym{,}   \mGLnt{r}   *   \delta_{{\mathrm{3}}}   \mGLsym{,}  \delta_{{\mathrm{4}}}  )   \odot   ( \Delta_{{\mathrm{1}}}  \mGLsym{,}  \Delta_{{\mathrm{2}}}  \mGLsym{,}  \Delta_{{\mathrm{3}}}  \mGLsym{,}  \Delta_{{\mathrm{4}}} )   \vdash_{\mathsf{GS} }  \mGLsym{[}  \mGLnt{t_{{\mathrm{1}}}}  \mGLsym{/}  \mGLmv{x}  \mGLsym{]}  \mGLsym{[}  \mGLnt{t_{{\mathrm{2}}}}  \mGLsym{/}  \mGLmv{x}  \mGLsym{]}  \mGLnt{t_{{\mathrm{3}}}}  \mGLsym{:}  \mGLnt{Z}}
    \end{gather*}
    This case follows from strictness of $n$, associativity, and functoriality:
    \begin{align*}
    (lhs) & (\interpGS{\Pi_3} \circ (id \boxtimes n_{\boxtimes,r,X,Y} \boxtimes id))
         \circ (id \boxtimes (r \odot (\interpGS{\Pi_1} \boxtimes \interpGS{\Pi_2}))
                   \boxtimes id) \\
    (rhs) & (\interp{\Pi_3} \circ (\id_{(  \delta_{{\mathrm{1}}}  \mGLsym{,}  \mGLnt{r}  )    \odot  \interp{  \mGLnt{X}  \mGLsym{,}  \Delta_{{\mathrm{1}}}  }} \boxtimes (\mGLnt{r} \odot \interp{\Pi_2}) \boxtimes \id_{\delta_{{\mathrm{4}}}   \odot  \interp{  \Delta_{{\mathrm{4}}}  }})) \\
    & \qquad \quad \circ (\id_{(  \delta_{{\mathrm{1}}}  \mGLsym{,}  \mGLnt{r}  *  \delta_{{\mathrm{2}}}  )    \odot  \interp{  \Delta_{{\mathrm{1}}}  \mGLsym{,}  \Delta_{{\mathrm{2}}}  }} \boxtimes \mGLnt{r} \odot \interp{\Pi_1} \boxtimes \id_{(  \mGLnt{r}  *  \delta_{{\mathrm{3}}}  \mGLsym{,}  \delta_{{\mathrm{4}}}  )    \odot  \interp{  \Delta_{{\mathrm{3}}}  \mGLsym{,}  \Delta_{{\mathrm{4}}}  }})
    \end{align*}
  \item The derivation:
    \begin{gather*}
    \inferrule* [flushleft,right=$\mGLdruleGSTXXCutName{}$] {
      \inferrule* [flushleft,right=$\mGLdruleGSTXXUnitRName{}$] {
        \,
      }{\emptyset  \odot  \emptyset  \vdash_{\mathsf{GS} }  \mathsf{j}  \mGLsym{:}  \mathsf{J}}\\
      \inferrule* [flushleft,right=$\mGLdruleGSTXXTenLName{}$] {
        (  \delta_{{\mathrm{1}}}  \mGLsym{,}  \delta_{{\mathrm{2}}}  )   \odot   ( \Delta_{{\mathrm{1}}}  \mGLsym{,}  \Delta_{{\mathrm{2}}} )   \vdash_{\mathsf{GS} }  \mGLnt{t}  \mGLsym{:}  \mGLnt{Z}
      }{(  \delta_{{\mathrm{1}}}  \mGLsym{,}  1  \mGLsym{,}  \delta_{{\mathrm{4}}}  )   \odot   ( \Delta_{{\mathrm{1}}}  \mGLsym{,}  \mGLmv{z}  \mGLsym{:}  \mathsf{J}  \mGLsym{,}  \Delta_{{\mathrm{4}}} )   \vdash_{\mathsf{GS} }  \mathsf{let} \, \mathsf{j} \, \mGLsym{=}  \mGLmv{z} \, \mathsf{in} \, \mGLnt{t}  \mGLsym{:}  \mGLnt{Z}}
    }{(  \delta_{{\mathrm{1}}}  \mGLsym{,}  \delta_{{\mathrm{4}}}  )   \odot   ( \Delta_{{\mathrm{1}}}  \mGLsym{,}  \Delta_{{\mathrm{4}}} )   \vdash_{\mathsf{GS} }  \mathsf{let} \, \mathsf{j} \, \mGLsym{=}  \mathsf{j} \, \mathsf{in} \, \mGLnt{t}  \mGLsym{:}  \mGLnt{Z}}
    \end{gather*}
    reduces to:
    \begin{gather*}
    (  \delta_{{\mathrm{1}}}  \mGLsym{,}  \delta_{{\mathrm{2}}}  )   \odot   ( \Delta_{{\mathrm{1}}}  \mGLsym{,}  \Delta_{{\mathrm{2}}} )   \vdash_{\mathsf{GS} }  \mGLnt{t}  \mGLsym{:}  \mGLnt{Z}
    \end{gather*}
    This case follows from the fact that $\mathsf{J}$ is the unit of tensor
    product in $(\cat{C}, \mathsf{J}, \boxtimes)$.

  \item The derivation:
    \begin{gather*}
    \inferrule* [flushleft,right=$\mGLdruleGSTXXCutName{}$] {
      \delta_{{\mathrm{2}}}  \odot  \Delta_{{\mathrm{2}}}  \vdash_{\mathsf{GS} }  \mGLnt{t_{{\mathrm{1}}}}  \mGLsym{:}  \mGLnt{X}\\
      \inferrule* [flushleft,right=$\mGLdruleGSTXXWeakName{}$] {
        (  \delta_{{\mathrm{1}}}  \mGLsym{,}  \delta_{{\mathrm{3}}}  )   \odot   ( \Delta_{{\mathrm{1}}}  \mGLsym{,}  \Delta_{{\mathrm{3}}} )   \vdash_{\mathsf{GS} }  \mGLnt{t_{{\mathrm{2}}}}  \mGLsym{:}  \mGLnt{Z}
      }{(  \delta_{{\mathrm{1}}}  \mGLsym{,}  \mathsf{0}  \mGLsym{,}  \delta_{{\mathrm{3}}}  )   \odot   ( \Delta_{{\mathrm{1}}}  \mGLsym{,}  \mGLmv{x}  \mGLsym{:}  \mGLnt{X}  \mGLsym{,}  \Delta_{{\mathrm{3}}} )   \vdash_{\mathsf{GS} }  \mGLnt{t_{{\mathrm{2}}}}  \mGLsym{:}  \mGLnt{Z}}
    }{(  \delta_{{\mathrm{1}}}  \mGLsym{,}   \mathsf{0}   *   \delta_{{\mathrm{2}}}   \mGLsym{,}  \delta_{{\mathrm{3}}}  )   \odot   ( \Delta_{{\mathrm{1}}}  \mGLsym{,}  \Delta_{{\mathrm{2}}}  \mGLsym{,}  \Delta_{{\mathrm{3}}} )   \vdash_{\mathsf{GS} }  \mGLsym{[}  \mGLnt{t_{{\mathrm{1}}}}  \mGLsym{/}  \mGLmv{x}  \mGLsym{]}  \mGLnt{t_{{\mathrm{2}}}}  \mGLsym{:}  \mGLnt{Z}}
    \end{gather*}
    reduces to:
    \begin{gather*}
      \inferrule* [flushleft,right=$\mGLdruleGSTXXWeakName{}$]
      {(  \delta_{{\mathrm{1}}}  \mGLsym{,}  \delta_{{\mathrm{3}}}  )   \odot   ( \Delta_{{\mathrm{1}}}  \mGLsym{,}  \Delta_{{\mathrm{2}}}  \mGLsym{,}  \Delta_{{\mathrm{3}}} )   \vdash_{\mathsf{GS} }  \mGLnt{t_{{\mathrm{2}}}}  \mGLsym{:}  \mGLnt{Z}}
      {(  \delta_{{\mathrm{1}}}  \mGLsym{,}   \mathsf{0}   *   \delta_{{\mathrm{2}}}   \mGLsym{,}  \delta_{{\mathrm{3}}}  )   \odot   ( \Delta_{{\mathrm{1}}}  \mGLsym{,}  \Delta_{{\mathrm{2}}}  \mGLsym{,}  \Delta_{{\mathrm{3}}} )   \vdash_{\mathsf{GS} }  \mGLnt{t_{{\mathrm{2}}}}  \mGLsym{:}  \mGLnt{Z}}
    \end{gather*}
    with morphisms:
    \begin{align*}
    (lhs) & \;\; \interpGS{\Pi_2} \circ (\lambda \boxtimes id) \circ (id_{\delta_1 \odot \Delta_{{\mathrm{1}}}} \boxtimes \mathsf{weak}_{\interpGS{X}} \boxtimes id_{\delta_3 \odot \Delta_{{\mathrm{3}}}}) \\
    & \qquad \circ (id_{\delta_1 \odot \Delta_{{\mathrm{1}}}} \boxtimes (0 \odot \interpGS{\Pi_1}) \boxtimes id_{\delta_3 \odot \Delta_{{\mathrm{3}}}}) \\
    (rhs) & \;\; (\interpGS{\Pi_2} \circ (\lambda \boxtimes id) \circ (id_{\delta_1 \odot \Delta_{{\mathrm{1}}}} \boxtimes \mathsf{weak}_{\interpGS{\mGLnt{D_{{\mathrm{2}}}}}} \boxtimes id_{\delta_3 \odot \Delta_{{\mathrm{3}}}})
    \end{align*}
    This case holds by the core properties of weakening
    $0 \odot X \mto^{\mathsf{weak}} J$ (Fig.~\ref{fig:equations-strict-action}).

  \item The derivation:
    \begin{gather*}
    \inferrule* [flushleft,right=$\mGLdruleGSTXXCutName{}$] {
      \delta_{{\mathrm{2}}}  \odot  \Delta_{{\mathrm{2}}}  \vdash_{\mathsf{GS} }  \mGLnt{t_{{\mathrm{1}}}}  \mGLsym{:}  \mGLnt{X}\\
      \inferrule* [flushleft,right=$\mGLdruleGSTXXContName{}$] {
        (  \delta_{{\mathrm{1}}}  \mGLsym{,}  \mGLnt{r_{{\mathrm{1}}}}  \mGLsym{,}  \mGLnt{r_{{\mathrm{2}}}}  \mGLsym{,}  \delta_{{\mathrm{3}}}  )   \odot   ( \Delta_{{\mathrm{1}}}  \mGLsym{,}  \mGLmv{x}  \mGLsym{:}  \mGLnt{X}  \mGLsym{,}  \mGLmv{y}  \mGLsym{:}  \mGLnt{X}  \mGLsym{,}  \Delta_{{\mathrm{3}}} )   \vdash_{\mathsf{GS} }  \mGLnt{t_{{\mathrm{2}}}}  \mGLsym{:}  \mGLnt{Z}
      }{(  \delta_{{\mathrm{1}}}  \mGLsym{,}  \mGLnt{r_{{\mathrm{1}}}}  +  \mGLnt{r_{{\mathrm{2}}}}  \mGLsym{,}  \delta_{{\mathrm{3}}}  )   \odot   ( \Delta_{{\mathrm{1}}}  \mGLsym{,}  \mGLmv{x}  \mGLsym{:}  \mGLnt{X}  \mGLsym{,}  \Delta_{{\mathrm{3}}} )   \vdash_{\mathsf{GS} }  \mGLsym{[}  \mGLmv{x}  \mGLsym{/}  \mGLmv{y}  \mGLsym{]}  \mGLnt{t_{{\mathrm{2}}}}  \mGLsym{:}  \mGLnt{Z}}
    }{(  \delta_{{\mathrm{1}}}  \mGLsym{,}   \mGLnt{r_{{\mathrm{1}}}}  +  \mGLnt{r_{{\mathrm{2}}}}   *   \delta_{{\mathrm{2}}}   \mGLsym{,}  \delta_{{\mathrm{3}}}  )   \odot   ( \Delta_{{\mathrm{1}}}  \mGLsym{,}  \Delta_{{\mathrm{2}}}  \mGLsym{,}  \Delta_{{\mathrm{3}}} )   \vdash_{\mathsf{GS} }  \mGLsym{[}  \mGLnt{t_{{\mathrm{1}}}}  \mGLsym{/}  \mGLmv{x}  \mGLsym{]}  \mGLsym{[}  \mGLmv{x}  \mGLsym{/}  \mGLmv{y}  \mGLsym{]}  \mGLnt{t_{{\mathrm{2}}}}  \mGLsym{:}  \mGLnt{Z}}
    \end{gather*}
    reduces to:
    \begin{gather*}
    \inferrule* [flushleft,right=$\mGLdruleGSTXXContName{}$] {
      \inferrule* [flushleft,right=$\mGLdruleGSTXXCutName{}$] {
        \inferrule* [flushleft,right=$\mGLdruleGSTXXCutName{}$] {
          \delta_{{\mathrm{2}}}  \odot  \Delta_{{\mathrm{2}}}  \vdash_{\mathsf{GS} }  \mGLnt{t_{{\mathrm{1}}}}  \mGLsym{:}  \mGLnt{X} \\ 
          (  \delta_{{\mathrm{1}}}  \mGLsym{,}  \mGLnt{r_{{\mathrm{1}}}}  \mGLsym{,}  \mGLnt{r_{{\mathrm{2}}}}  \mGLsym{,}  \delta_{{\mathrm{3}}}  )   \odot   ( \Delta_{{\mathrm{1}}}  \mGLsym{,}  \mGLmv{x}  \mGLsym{:}  \mGLnt{X}  \mGLsym{,}  \mGLmv{y}  \mGLsym{:}  \mGLnt{X}  \mGLsym{,}  \Delta_{{\mathrm{3}}} )   \vdash_{\mathsf{GS} }  \mGLnt{t_{{\mathrm{2}}}}  \mGLsym{:}  \mGLnt{Z}
        }{(  \delta_{{\mathrm{1}}}  \mGLsym{,}  \mGLnt{r_{{\mathrm{1}}}}  \mGLsym{,}   \mGLnt{r_{{\mathrm{2}}}}   *   \delta_{{\mathrm{2}}}   \mGLsym{,}  \delta_{{\mathrm{3}}}  )   \odot   ( \Delta_{{\mathrm{1}}}  \mGLsym{,}  \mGLmv{x}  \mGLsym{:}  \mGLnt{X}  \mGLsym{,}  \Delta_{{\mathrm{2}}}  \mGLsym{,}  \Delta_{{\mathrm{3}}} )   \vdash_{\mathsf{GS} }  \mGLsym{[}  \mGLnt{t_{{\mathrm{1}}}}  \mGLsym{/}  \mGLmv{y}  \mGLsym{]}  \mGLnt{t_{{\mathrm{2}}}}  \mGLsym{:}  \mGLnt{Z}}
      }{(  \delta_{{\mathrm{1}}}  \mGLsym{,}   \mGLnt{r_{{\mathrm{1}}}}   *   \delta_{{\mathrm{2}}}   \mGLsym{,}   \mGLnt{r_{{\mathrm{2}}}}   *   \delta_{{\mathrm{2}}}   \mGLsym{,}  \delta_{{\mathrm{3}}}  )   \odot   ( \Delta_{{\mathrm{1}}}  \mGLsym{,}  \Delta_{{\mathrm{2}}}  \mGLsym{,}  \Delta_{{\mathrm{2}}}  \mGLsym{,}  \Delta_{{\mathrm{3}}} )   \vdash_{\mathsf{GS} }  \mGLsym{[}  \mGLnt{t_{{\mathrm{1}}}}  \mGLsym{/}  \mGLmv{x}  \mGLsym{]}  \mGLsym{[}  \mGLnt{t_{{\mathrm{1}}}}  \mGLsym{/}  \mGLmv{y}  \mGLsym{]}  \mGLnt{t_{{\mathrm{2}}}}  \mGLsym{:}  \mGLnt{Z}}
    }{(  \delta_{{\mathrm{1}}}  \mGLsym{,}   \mGLnt{r_{{\mathrm{1}}}}   *   \delta_{{\mathrm{2}}}   +   \mGLnt{r_{{\mathrm{2}}}}   *   \delta_{{\mathrm{2}}}   \mGLsym{,}  \delta_{{\mathrm{3}}}  )   \odot   ( \Delta_{{\mathrm{1}}}  \mGLsym{,}  \Delta_{{\mathrm{2}}}  \mGLsym{,}  \Delta_{{\mathrm{3}}} )   \vdash_{\mathsf{GS} }  \mGLsym{[}  \Delta_{{\mathrm{2}}}  \mGLsym{/}  \Delta_{{\mathrm{2}}}  \mGLsym{]}  \mGLsym{[}  \mGLnt{t_{{\mathrm{1}}}}  \mGLsym{/}  \mGLmv{x}  \mGLsym{]}  \mGLsym{[}  \mGLnt{t_{{\mathrm{1}}}}  \mGLsym{/}  \mGLmv{y}  \mGLsym{]}  \mGLnt{t_{{\mathrm{2}}}}  \mGLsym{:}  \mGLnt{Z}}
    \end{gather*}
    Similar to the previous case, this case holds by the commutation of
    $\delta$ and
    $(\mGLnt{r_{{\mathrm{1}}}}  +  \mGLnt{r_{{\mathrm{2}}}}) \odot X \mto^{\mathsf{contr}} \mGLnt{r_{{\mathrm{1}}}} \odot X \boxtimes \mGLnt{r_{{\mathrm{2}}}} \odot X$
    (Fig.~\ref{fig:equations-strict-action}).

  \item The derivation:
    \begin{gather*}
    \inferrule* [flushleft,right=$\mGLdruleGSTXXCutName{}$] {
      \inferrule* [flushleft,right=$\mGLdruleGSTXXExName{}$] {
        (  \delta_{{\mathrm{2}}}  \mGLsym{,}  \mGLnt{r_{{\mathrm{1}}}}  \mGLsym{,}  \mGLnt{r_{{\mathrm{2}}}}  \mGLsym{,}  \delta_{{\mathrm{3}}}  )   \odot   ( \Delta_{{\mathrm{2}}}  \mGLsym{,}  \mGLmv{x}  \mGLsym{:}  \mGLnt{X}  \mGLsym{,}  \mGLmv{y}  \mGLsym{:}  \mGLnt{Y}  \mGLsym{,}  \Delta_{{\mathrm{3}}} )   \vdash_{\mathsf{GS} }  \mGLnt{t_{{\mathrm{1}}}}  \mGLsym{:}  \mGLnt{Z}
      }{(  \delta_{{\mathrm{2}}}  \mGLsym{,}  \mGLnt{r_{{\mathrm{2}}}}  \mGLsym{,}  \mGLnt{r_{{\mathrm{1}}}}  \mGLsym{,}  \delta_{{\mathrm{3}}}  )   \odot   ( \Delta_{{\mathrm{2}}}  \mGLsym{,}  \mGLmv{y}  \mGLsym{:}  \mGLnt{Y}  \mGLsym{,}  \mGLmv{x}  \mGLsym{:}  \mGLnt{X}  \mGLsym{,}  \Delta_{{\mathrm{3}}} )   \vdash_{\mathsf{GS} }  \mGLnt{t_{{\mathrm{1}}}}  \mGLsym{:}  \mGLnt{Z}}\\
      (  \delta_{{\mathrm{1}}}  \mGLsym{,}  \mGLnt{s}  \mGLsym{,}  \delta_{{\mathrm{4}}}  )   \odot   ( \Delta_{{\mathrm{1}}}  \mGLsym{,}  \mGLmv{z}  \mGLsym{:}  \mGLnt{Z}  \mGLsym{,}  \Delta_{{\mathrm{4}}} )   \vdash_{\mathsf{GS} }  \mGLnt{t_{{\mathrm{2}}}}  \mGLsym{:}  \mGLnt{W}
    }{(  \delta_{{\mathrm{1}}}  \mGLsym{,}   \mGLnt{s}   *    (  \delta_{{\mathrm{2}}}  \mGLsym{,}  \mGLnt{r_{{\mathrm{2}}}}  \mGLsym{,}  \mGLnt{r_{{\mathrm{1}}}}  \mGLsym{,}  \delta_{{\mathrm{3}}}  )    \mGLsym{,}  \delta_{{\mathrm{4}}}  )   \odot   ( \Delta_{{\mathrm{1}}}  \mGLsym{,}  \Delta_{{\mathrm{2}}}  \mGLsym{,}  \mGLmv{y}  \mGLsym{:}  \mGLnt{Y}  \mGLsym{,}  \mGLmv{x}  \mGLsym{:}  \mGLnt{X}  \mGLsym{,}  \Delta_{{\mathrm{3}}}  \mGLsym{,}  \Delta_{{\mathrm{4}}} )   \vdash_{\mathsf{GS} }  \mGLsym{[}  \mGLnt{t_{{\mathrm{1}}}}  \mGLsym{/}  \mGLmv{x}  \mGLsym{]}  \mGLnt{t_{{\mathrm{2}}}}  \mGLsym{:}  \mGLnt{W}}
    \end{gather*}
    reduces to:
    \begin{gather*}
    \inferrule* [flushleft,right=$\mGLdruleGSTXXExName{}$] {
      \inferrule* [flushleft,right=$\mGLdruleGSTXXCutName{}$] {
        (  \delta_{{\mathrm{2}}}  \mGLsym{,}  \mGLnt{r_{{\mathrm{1}}}}  \mGLsym{,}  \mGLnt{r_{{\mathrm{2}}}}  \mGLsym{,}  \delta_{{\mathrm{3}}}  )   \odot   ( \Delta_{{\mathrm{2}}}  \mGLsym{,}  \mGLmv{x}  \mGLsym{:}  \mGLnt{X}  \mGLsym{,}  \mGLmv{y}  \mGLsym{:}  \mGLnt{Y}  \mGLsym{,}  \Delta_{{\mathrm{3}}} )   \vdash_{\mathsf{GS} }  \mGLnt{t_{{\mathrm{1}}}}  \mGLsym{:}  \mGLnt{Z}\\
        (  \delta_{{\mathrm{1}}}  \mGLsym{,}  \mGLnt{s}  \mGLsym{,}  \delta_{{\mathrm{4}}}  )   \odot   ( \Delta_{{\mathrm{1}}}  \mGLsym{,}  \mGLmv{z}  \mGLsym{:}  \mGLnt{Z}  \mGLsym{,}  \Delta_{{\mathrm{4}}} )   \vdash_{\mathsf{GS} }  \mGLnt{t_{{\mathrm{2}}}}  \mGLsym{:}  \mGLnt{W}
      }{(  \delta_{{\mathrm{1}}}  \mGLsym{,}   \mGLnt{s}   *    (  \delta_{{\mathrm{2}}}  \mGLsym{,}  \mGLnt{r_{{\mathrm{1}}}}  \mGLsym{,}  \mGLnt{r_{{\mathrm{2}}}}  \mGLsym{,}  \delta_{{\mathrm{3}}}  )    \mGLsym{,}  \delta_{{\mathrm{4}}}  )   \odot   ( \Delta_{{\mathrm{1}}}  \mGLsym{,}  \Delta_{{\mathrm{2}}}  \mGLsym{,}  \mGLmv{x}  \mGLsym{:}  \mGLnt{X}  \mGLsym{,}  \mGLmv{y}  \mGLsym{:}  \mGLnt{Y}  \mGLsym{,}  \Delta_{{\mathrm{3}}}  \mGLsym{,}  \Delta_{{\mathrm{4}}} )   \vdash_{\mathsf{GS} }  \mGLsym{[}  \mGLnt{t_{{\mathrm{1}}}}  \mGLsym{/}  \mGLmv{z}  \mGLsym{]}  \mGLnt{t_{{\mathrm{2}}}}  \mGLsym{:}  \mGLnt{W}}\\
    }{(  \delta_{{\mathrm{1}}}  \mGLsym{,}   \mGLnt{s}   *    (  \delta_{{\mathrm{2}}}  \mGLsym{,}  \mGLnt{r_{{\mathrm{2}}}}  \mGLsym{,}  \mGLnt{r_{{\mathrm{1}}}}  \mGLsym{,}  \delta_{{\mathrm{3}}}  )    \mGLsym{,}  \delta_{{\mathrm{4}}}  )   \odot   ( \Delta_{{\mathrm{1}}}  \mGLsym{,}  \Delta_{{\mathrm{2}}}  \mGLsym{,}  \mGLmv{y}  \mGLsym{:}  \mGLnt{Y}  \mGLsym{,}  \mGLmv{x}  \mGLsym{:}  \mGLnt{X}  \mGLsym{,}  \Delta_{{\mathrm{3}}}  \mGLsym{,}  \Delta_{{\mathrm{4}}} )   \vdash_{\mathsf{GS} }  \mGLsym{[}  \mGLnt{t_{{\mathrm{1}}}}  \mGLsym{/}  \mGLmv{z}  \mGLsym{]}  \mGLnt{t_{{\mathrm{2}}}}  \mGLsym{:}  \mGLnt{W}}
    \end{gather*}
    This case follows from naturality of symmetry, and the fact that
    exchange is interpreted as the symmetry in $\cat{C}$ which is symmetric monoidal.

  \item The derivation:
    \begin{gather*}
    \inferrule* [flushleft,right=$\mGLdruleGSTXXCutName{}$] {
      \inferrule* [flushleft,right=$\mGLdruleGSTXXTenLName{}$] {
        (  \delta_{{\mathrm{2}}}  \mGLsym{,}  \mGLnt{r}  \mGLsym{,}  \mGLnt{r}  \mGLsym{,}  \delta_{{\mathrm{3}}}  )   \odot   ( \Delta_{{\mathrm{2}}}  \mGLsym{,}  \mGLmv{x}  \mGLsym{:}  \mGLnt{X}  \mGLsym{,}  \mGLmv{y}  \mGLsym{:}  \mGLnt{Y}  \mGLsym{,}  \Delta_{{\mathrm{3}}} )   \vdash_{\mathsf{GS} }  \mGLnt{t_{{\mathrm{1}}}}  \mGLsym{:}  \mGLnt{Z}
      }{(  \delta_{{\mathrm{2}}}  \mGLsym{,}  \mGLnt{r}  \mGLsym{,}  \delta_{{\mathrm{3}}}  )   \odot   ( \Delta_{{\mathrm{2}}}  \mGLsym{,}  \mGLmv{w}  \mGLsym{:}  \mGLnt{X}  \boxtimes  \mGLnt{Y}  \mGLsym{,}  \Delta_{{\mathrm{3}}} )   \vdash_{\mathsf{GS} }   \mathsf{let} \,( \mGLmv{x} , \mGLmv{y} ) =  \mGLmv{w} \, \mathsf{in} \, \mGLnt{t_{{\mathrm{1}}}}   \mGLsym{:}  \mGLnt{Z}}\\
      (  \delta_{{\mathrm{1}}}  \mGLsym{,}  \mGLnt{s}  \mGLsym{,}  \delta_{{\mathrm{4}}}  )   \odot   ( \Delta_{{\mathrm{1}}}  \mGLsym{,}  \mGLmv{z}  \mGLsym{:}  \mGLnt{Z}  \mGLsym{,}  \Delta_{{\mathrm{4}}} )   \vdash_{\mathsf{GS} }  \mGLnt{t_{{\mathrm{2}}}}  \mGLsym{:}  \mGLnt{W}
    }{(  \delta_{{\mathrm{1}}}  \mGLsym{,}   \mGLnt{s}   *    (  \delta_{{\mathrm{2}}}  \mGLsym{,}  \mGLnt{r}  \mGLsym{,}  \delta_{{\mathrm{3}}}  )    \mGLsym{,}  \delta_{{\mathrm{4}}}  )   \odot   ( \Delta_{{\mathrm{1}}}  \mGLsym{,}  \Delta_{{\mathrm{2}}}  \mGLsym{,}  \mGLmv{w}  \mGLsym{:}  \mGLnt{X}  \boxtimes  \mGLnt{Y}  \mGLsym{,}  \Delta_{{\mathrm{3}}}  \mGLsym{,}  \Delta_{{\mathrm{4}}} )   \vdash_{\mathsf{GS} }  \mGLsym{[}   \mathsf{let} \,( \mGLmv{x} , \mGLmv{y} ) =  \mGLmv{w} \, \mathsf{in} \, \mGLnt{t_{{\mathrm{1}}}}   \mGLsym{/}  \mGLmv{x}  \mGLsym{]}  \mGLnt{t_{{\mathrm{2}}}}  \mGLsym{:}  \mGLnt{W}}
    \end{gather*}
    reduces to:
    \begin{gather*}
    \inferrule* [flushleft,right=$\mGLdruleGSTXXExName{}$] {
      \inferrule* [flushleft,right=$\mGLdruleGSTXXCutName{}$] {
        (  \delta_{{\mathrm{2}}}  \mGLsym{,}  \mGLnt{r}  \mGLsym{,}  \mGLnt{r}  \mGLsym{,}  \delta_{{\mathrm{3}}}  )   \odot   ( \Delta_{{\mathrm{2}}}  \mGLsym{,}  \mGLmv{x}  \mGLsym{:}  \mGLnt{X}  \mGLsym{,}  \mGLmv{y}  \mGLsym{:}  \mGLnt{Y}  \mGLsym{,}  \Delta_{{\mathrm{3}}} )   \vdash_{\mathsf{GS} }  \mGLnt{t_{{\mathrm{1}}}}  \mGLsym{:}  \mGLnt{Z}\\
        (  \delta_{{\mathrm{1}}}  \mGLsym{,}  \mGLnt{s}  \mGLsym{,}  \delta_{{\mathrm{4}}}  )   \odot   ( \Delta_{{\mathrm{1}}}  \mGLsym{,}  \mGLmv{z}  \mGLsym{:}  \mGLnt{Z}  \mGLsym{,}  \Delta_{{\mathrm{4}}} )   \vdash_{\mathsf{GS} }  \mGLnt{t_{{\mathrm{2}}}}  \mGLsym{:}  \mGLnt{W}
      }{(  \delta_{{\mathrm{1}}}  \mGLsym{,}   \mGLnt{s}   *    (  \delta_{{\mathrm{2}}}  \mGLsym{,}  \mGLnt{r}  \mGLsym{,}  \mGLnt{r}  \mGLsym{,}  \delta_{{\mathrm{3}}}  )    \mGLsym{,}  \delta_{{\mathrm{4}}}  )   \odot   ( \Delta_{{\mathrm{1}}}  \mGLsym{,}  \Delta_{{\mathrm{2}}}  \mGLsym{,}  \mGLmv{x}  \mGLsym{:}  \mGLnt{X}  \mGLsym{,}  \mGLmv{y}  \mGLsym{:}  \mGLnt{Y}  \mGLsym{,}  \Delta_{{\mathrm{3}}}  \mGLsym{,}  \Delta_{{\mathrm{4}}} )   \vdash_{\mathsf{GS} }  \mGLsym{[}  \mGLnt{t_{{\mathrm{1}}}}  \mGLsym{/}  \mGLmv{z}  \mGLsym{]}  \mGLnt{t_{{\mathrm{2}}}}  \mGLsym{:}  \mGLnt{W}}\\
    }{(  \delta_{{\mathrm{1}}}  \mGLsym{,}   \mGLnt{s}   *    (  \delta_{{\mathrm{2}}}  \mGLsym{,}  \mGLnt{r}  \mGLsym{,}  \delta_{{\mathrm{3}}}  )    \mGLsym{,}  \delta_{{\mathrm{4}}}  )   \odot   ( \Delta_{{\mathrm{1}}}  \mGLsym{,}  \Delta_{{\mathrm{2}}}  \mGLsym{,}  \mGLmv{w}  \mGLsym{:}  \mGLnt{X}  \boxtimes  \mGLnt{Y}  \mGLsym{,}  \Delta_{{\mathrm{3}}}  \mGLsym{,}  \Delta_{{\mathrm{4}}} )   \vdash_{\mathsf{GS} }   \mathsf{let} \,( \mGLmv{x} , \mGLmv{y} ) =  \mGLmv{w} \, \mathsf{in} \, \mGLsym{[}  \mGLnt{t_{{\mathrm{1}}}}  \mGLsym{/}  \mGLmv{z}  \mGLsym{]}  \mGLnt{t_{{\mathrm{2}}}}   \mGLsym{:}  \mGLnt{W}}
    \end{gather*}
    Both have identical interpretations up to symmetry as per Definition~\ref{def:full-mgl-intepretation} and by the strictness of $n_{\boxtimes}$.

  \item The derivation:
    \begin{gather*}
    \inferrule* [flushleft,right=$\mGLdruleGSTXXCutName{}$] {
      \inferrule* [flushleft,right=$\mGLdruleGSTXXTenLName{}$] {
        (  \delta_{{\mathrm{2}}}  \mGLsym{,}  \delta_{{\mathrm{3}}}  )   \odot   ( \Delta_{{\mathrm{2}}}  \mGLsym{,}  \Delta_{{\mathrm{3}}} )   \vdash_{\mathsf{GS} }  \mGLnt{t_{{\mathrm{1}}}}  \mGLsym{:}  \mGLnt{Z}
      }{(  \delta_{{\mathrm{2}}}  \mGLsym{,}  1  \mGLsym{,}  \delta_{{\mathrm{3}}}  )   \odot   ( \Delta_{{\mathrm{2}}}  \mGLsym{,}  \mGLmv{w}  \mGLsym{:}  \mathsf{J}  \mGLsym{,}  \Delta_{{\mathrm{3}}} )   \vdash_{\mathsf{GS} }  \mathsf{let} \, \mathsf{j} \, \mGLsym{=}  \mGLmv{w} \, \mathsf{in} \, \mGLnt{t_{{\mathrm{1}}}}  \mGLsym{:}  \mGLnt{Z}}\\
      (  \delta_{{\mathrm{1}}}  \mGLsym{,}  \mGLnt{s}  \mGLsym{,}  \delta_{{\mathrm{4}}}  )   \odot   ( \Delta_{{\mathrm{1}}}  \mGLsym{,}  \mGLmv{z}  \mGLsym{:}  \mGLnt{Z}  \mGLsym{,}  \Delta_{{\mathrm{4}}} )   \vdash_{\mathsf{GS} }  \mGLnt{t_{{\mathrm{2}}}}  \mGLsym{:}  \mGLnt{W}
    }{(  \delta_{{\mathrm{1}}}  \mGLsym{,}   \mGLnt{s}   *    (  \delta_{{\mathrm{2}}}  \mGLsym{,}  1  \mGLsym{,}  \delta_{{\mathrm{3}}}  )    \mGLsym{,}  \delta_{{\mathrm{4}}}  )   \odot   ( \Delta_{{\mathrm{1}}}  \mGLsym{,}  \Delta_{{\mathrm{2}}}  \mGLsym{,}  \mGLmv{w}  \mGLsym{:}  \mathsf{J}  \mGLsym{,}  \Delta_{{\mathrm{3}}}  \mGLsym{,}  \Delta_{{\mathrm{4}}} )   \vdash_{\mathsf{GS} }  \mGLsym{[}  \mathsf{let} \, \mathsf{j} \, \mGLsym{=}  \mGLmv{w} \, \mathsf{in} \, \mGLnt{t_{{\mathrm{1}}}}  \mGLsym{/}  \mGLmv{x}  \mGLsym{]}  \mGLnt{t_{{\mathrm{2}}}}  \mGLsym{:}  \mGLnt{W}}
    \end{gather*}
    reduces to:
    \begin{gather*}
    \inferrule* [flushleft,right=$\mGLdruleGSTXXExName{}$] {
      \inferrule* [flushleft,right=$\mGLdruleGSTXXCutName{}$] {
        (  \delta_{{\mathrm{2}}}  \mGLsym{,}  \delta_{{\mathrm{3}}}  )   \odot   ( \Delta_{{\mathrm{2}}}  \mGLsym{,}  \Delta_{{\mathrm{3}}} )   \vdash_{\mathsf{GS} }  \mGLnt{t_{{\mathrm{1}}}}  \mGLsym{:}  \mGLnt{Z}\\
        (  \delta_{{\mathrm{1}}}  \mGLsym{,}  \mGLnt{s}  \mGLsym{,}  \delta_{{\mathrm{4}}}  )   \odot   ( \Delta_{{\mathrm{1}}}  \mGLsym{,}  \mGLmv{z}  \mGLsym{:}  \mGLnt{Z}  \mGLsym{,}  \Delta_{{\mathrm{4}}} )   \vdash_{\mathsf{GS} }  \mGLnt{t_{{\mathrm{2}}}}  \mGLsym{:}  \mGLnt{W}
      }{(  \delta_{{\mathrm{1}}}  \mGLsym{,}   \mGLnt{s}   *    (  \delta_{{\mathrm{2}}}  \mGLsym{,}  \delta_{{\mathrm{3}}}  )    \mGLsym{,}  \delta_{{\mathrm{4}}}  )   \odot   ( \Delta_{{\mathrm{1}}}  \mGLsym{,}  \Delta_{{\mathrm{2}}}  \mGLsym{,}  \Delta_{{\mathrm{3}}}  \mGLsym{,}  \Delta_{{\mathrm{4}}} )   \vdash_{\mathsf{GS} }  \mGLsym{[}  \mGLnt{t_{{\mathrm{1}}}}  \mGLsym{/}  \mGLmv{z}  \mGLsym{]}  \mGLnt{t_{{\mathrm{2}}}}  \mGLsym{:}  \mGLnt{W}}\\
    }{(  \delta_{{\mathrm{1}}}  \mGLsym{,}   \mGLnt{s}   *    (  \delta_{{\mathrm{2}}}  \mGLsym{,}  1  \mGLsym{,}  \delta_{{\mathrm{3}}}  )    \mGLsym{,}  \delta_{{\mathrm{4}}}  )   \odot   ( \Delta_{{\mathrm{1}}}  \mGLsym{,}  \Delta_{{\mathrm{2}}}  \mGLsym{,}  \mGLmv{w}  \mGLsym{:}  \mathsf{J}  \mGLsym{,}  \Delta_{{\mathrm{3}}}  \mGLsym{,}  \Delta_{{\mathrm{4}}} )   \vdash_{\mathsf{GS} }  \mathsf{let} \, \mathsf{j} \, \mGLsym{=}  \mGLmv{z} \, \mathsf{in} \, \mGLsym{[}  \mGLnt{t_{{\mathrm{1}}}}  \mGLsym{/}  \mGLmv{z}  \mGLsym{]}  \mGLnt{t_{{\mathrm{2}}}}  \mGLsym{:}  \mGLnt{W}}
    \end{gather*}
    Both have identical interpretations up to symmetry as per Definition~\ref{def:full-mgl-intepretation}.

      \item The derivation:
    \begin{gather*}
    \inferrule* [flushleft,right=$\mGLdruleGSTXXCutName{}$] {
      \delta_{{\mathrm{2}}}  \odot  \Delta_{{\mathrm{2}}}  \vdash_{\mathsf{GS} }  \mGLnt{t_{{\mathrm{1}}}}  \mGLsym{:}  \mGLnt{W}\\
      \inferrule* [flushleft,right=$\mGLdruleGSTXXTenRName{}$] {
        (  \delta_{{\mathrm{1}}}  \mGLsym{,}  \mGLnt{s}  \mGLsym{,}  \delta_{{\mathrm{3}}}  )   \odot   ( \Delta_{{\mathrm{1}}}  \mGLsym{,}  \mGLmv{w}  \mGLsym{:}  \mGLnt{W}  \mGLsym{,}  \Delta_{{\mathrm{3}}} )   \vdash_{\mathsf{GS} }  \mGLnt{t_{{\mathrm{2}}}}  \mGLsym{:}  \mGLnt{X}\\
        \delta_{{\mathrm{4}}}  \odot  \Delta_{{\mathrm{4}}}  \vdash_{\mathsf{GS} }  \mGLnt{t_{{\mathrm{3}}}}  \mGLsym{:}  \mGLnt{Y}
      }{(  \delta_{{\mathrm{1}}}  \mGLsym{,}  \mGLnt{s}  \mGLsym{,}  \delta_{{\mathrm{3}}}  \mGLsym{,}  \delta_{{\mathrm{4}}}  )   \odot   ( \Delta_{{\mathrm{1}}}  \mGLsym{,}  \mGLmv{w}  \mGLsym{:}  \mGLnt{W}  \mGLsym{,}  \Delta_{{\mathrm{3}}}  \mGLsym{,}  \Delta_{{\mathrm{4}}} )   \vdash_{\mathsf{GS} }  \mGLsym{(}  \mGLnt{t_{{\mathrm{2}}}}  \mGLsym{,}  \mGLnt{t_{{\mathrm{3}}}}  \mGLsym{)}  \mGLsym{:}  \mGLnt{X}  \boxtimes  \mGLnt{Y}}\\
    }{(  \delta_{{\mathrm{1}}}  \mGLsym{,}   \mGLnt{s}   *   \delta_{{\mathrm{2}}}   \mGLsym{,}  \delta_{{\mathrm{3}}}  \mGLsym{,}  \delta_{{\mathrm{4}}}  )   \odot   ( \Delta_{{\mathrm{1}}}  \mGLsym{,}  \Delta_{{\mathrm{2}}}  \mGLsym{,}  \Delta_{{\mathrm{3}}}  \mGLsym{,}  \Delta_{{\mathrm{4}}} )   \vdash_{\mathsf{GS} }  \mGLsym{[}  \mGLnt{t_{{\mathrm{1}}}}  \mGLsym{/}  \mGLmv{w}  \mGLsym{]}  \mGLsym{(}  \mGLnt{t_{{\mathrm{2}}}}  \mGLsym{,}  \mGLnt{t_{{\mathrm{3}}}}  \mGLsym{)}  \mGLsym{:}  \mGLnt{X}  \boxtimes  \mGLnt{Y}}
    \end{gather*}
    reduces to:
    \begin{gather*}
    \inferrule* [flushleft,right=$\mGLdruleGSTXXTenRName{}$] {
      \inferrule* [flushleft,right=$\mGLdruleGSTXXCutName{}$] {
        \delta_{{\mathrm{2}}}  \odot  \Delta_{{\mathrm{2}}}  \vdash_{\mathsf{GS} }  \mGLnt{t_{{\mathrm{1}}}}  \mGLsym{:}  \mGLnt{W}\\
        (  \delta_{{\mathrm{1}}}  \mGLsym{,}  \mGLnt{s}  \mGLsym{,}  \delta_{{\mathrm{3}}}  )   \odot   ( \Delta_{{\mathrm{1}}}  \mGLsym{,}  \mGLmv{w}  \mGLsym{:}  \mGLnt{W}  \mGLsym{,}  \Delta_{{\mathrm{3}}} )   \vdash_{\mathsf{GS} }  \mGLnt{t_{{\mathrm{2}}}}  \mGLsym{:}  \mGLnt{X}
      }{(  \delta_{{\mathrm{1}}}  \mGLsym{,}   \mGLnt{s}   *   \delta_{{\mathrm{2}}}   \mGLsym{,}  \delta_{{\mathrm{3}}}  )   \odot   ( \Delta_{{\mathrm{1}}}  \mGLsym{,}  \Delta_{{\mathrm{2}}}  \mGLsym{,}  \Delta_{{\mathrm{3}}} )   \vdash_{\mathsf{GS} }  \mGLsym{[}  \mGLnt{t_{{\mathrm{1}}}}  \mGLsym{/}  \mGLmv{w}  \mGLsym{]}  \mGLnt{t_{{\mathrm{2}}}}  \mGLsym{:}  \mGLnt{X}}\\
      \delta_{{\mathrm{4}}}  \odot  \Delta_{{\mathrm{4}}}  \vdash_{\mathsf{GS} }  \mGLnt{t_{{\mathrm{3}}}}  \mGLsym{:}  \mGLnt{Y}
    }{(  \delta_{{\mathrm{1}}}  \mGLsym{,}   \mGLnt{s}   *   \delta_{{\mathrm{2}}}   \mGLsym{,}  \delta_{{\mathrm{3}}}  \mGLsym{,}  \delta_{{\mathrm{4}}}  )   \odot   ( \Delta_{{\mathrm{1}}}  \mGLsym{,}  \Delta_{{\mathrm{2}}}  \mGLsym{,}  \Delta_{{\mathrm{3}}}  \mGLsym{,}  \Delta_{{\mathrm{4}}} )   \vdash_{\mathsf{GS} }  \mGLsym{(}  \mGLsym{[}  \mGLnt{t_{{\mathrm{1}}}}  \mGLsym{/}  \mGLmv{w}  \mGLsym{]}  \mGLnt{t_{{\mathrm{2}}}}  \mGLsym{,}  \mGLnt{t_{{\mathrm{3}}}}  \mGLsym{)}  \mGLsym{:}  \mGLnt{X}  \boxtimes  \mGLnt{Y}}
    \end{gather*}
    This case holds by functorality of $\boxtimes$. The case where $\mGLmv{w}  \mGLsym{:}  \mGLnt{W}$ is free in $\mGLnt{t_{{\mathrm{3}}}}$ is similar.

    Finally, there are a few additional cases for the structural
    rules, but all of them are equivalent to naturality of their
    respective natural transformations, and hence, we omit them here.

  \item The final cases are the commuting conversions related to the
    rule:
    \begin{center}
      \begin{math}
        \mGLdruleGSTXXLinR{}
      \end{math}
    \end{center}
    However, are all similar to the previous cases which hold by
    straightforwardly applying the induction hypotheses.
  \end{enumerate}

  We now consider part two.  Here we only give the beta-reduction
  cases for the modal operators, because they are the most
  interesting.

  Consider the following reduction:
  \begin{center}
    \begin{gather*}
      \inferrule* [flushleft,right=,fraction={===}] {
        \inferrule* [flushleft,right={ $\mGLdruleMSTXXGCutName{}$},fraction={---}] {
          \inferrule* [flushleft,right={ $\mGLdruleGSTXXLinRName{}$}] {
            \inferrule* [flushleft,right=] {
              \Pi_2\\\\
              \vdots
            }{\delta_{{\mathrm{2}}}  \odot  \Delta_{{\mathrm{2}}}  \mGLsym{;}  \emptyset  \vdash_{\mathsf{MS} }  \mGLnt{l_{{\mathrm{1}}}}  \mGLsym{:}  \mGLnt{A}}
          }{\delta_{{\mathrm{2}}}  \odot  \Delta_{{\mathrm{2}}}  \vdash_{\mathsf{GS} }  \mathsf{Lin} \, \mGLnt{l_{{\mathrm{1}}}}  \mGLsym{:}  \mathsf{Lin} \, \mGLnt{A}}\\
          \inferrule* [flushleft,right={{ $\mGLdruleMSTXXLinLName{}$}}] {
            \inferrule* [flushleft,right=] {
              \Pi_1\\\\
              \vdots
            }{\delta_{{\mathrm{1}}}  \odot  \Delta_{{\mathrm{1}}}  \mGLsym{;}   ( \mGLmv{x}  \mGLsym{:}  \mGLnt{A}  \mGLsym{,}  \Gamma )   \vdash_{\mathsf{MS} }  \mGLnt{l_{{\mathrm{2}}}}  \mGLsym{:}  \mGLnt{B}}
          }{(  \delta_{{\mathrm{1}}}  \mGLsym{,}  1  )   \odot   ( \Delta_{{\mathrm{1}}}  \mGLsym{,}  \mGLmv{z}  \mGLsym{:}  \mathsf{Lin} \, \mGLnt{A} )   \mGLsym{;}  \Gamma  \vdash_{\mathsf{MS} }  \mGLsym{[}   \mathsf{Unlin} \, \mGLmv{z}   \mGLsym{/}  \mGLmv{x}  \mGLsym{]}  \mGLnt{l_{{\mathrm{2}}}}  \mGLsym{:}  \mGLnt{B}}\\
        }{(  \delta_{{\mathrm{1}}}  \mGLsym{,}   1   *   \delta_{{\mathrm{2}}}   )   \odot   ( \Delta_{{\mathrm{1}}}  \mGLsym{,}  \Delta_{{\mathrm{2}}} )   \mGLsym{;}  \Gamma  \vdash_{\mathsf{MS} }  \mGLsym{[}  \mathsf{Lin} \, \mGLnt{l_{{\mathrm{1}}}}  \mGLsym{/}  \mGLmv{z}  \mGLsym{]}  \mGLsym{[}   \mathsf{Unlin} \, \mGLmv{z}   \mGLsym{/}  \mGLmv{x}  \mGLsym{]}  \mGLnt{l_{{\mathrm{2}}}}  \mGLsym{:}  \mGLnt{B}}
      }{(  \delta_{{\mathrm{1}}}  \mGLsym{,}  \delta_{{\mathrm{2}}}  )   \odot   ( \Delta_{{\mathrm{1}}}  \mGLsym{,}  \Delta_{{\mathrm{2}}} )   \mGLsym{;}  \Gamma  \vdash_{\mathsf{MS} }  \mGLsym{[}   \mathsf{Unlin} \, \mGLsym{(}  \mathsf{Lin} \, \mGLnt{l_{{\mathrm{1}}}}  \mGLsym{)}   \mGLsym{/}  \mGLmv{x}  \mGLsym{]}  \mGLnt{l_{{\mathrm{2}}}}  \mGLsym{:}  \mGLnt{B}}
    \end{gather*}
  \end{center}
  This proof reduces to the following:
  \begin{center}
    \begin{gather*}
        \inferrule* [flushleft,right={ $\mGLdruleMSTXXCutName{}$},fraction={---}] {
            \inferrule* [flushleft,right=] {
              \Pi_2\\\\
              \vdots
            }{\delta_{{\mathrm{2}}}  \odot  \Delta_{{\mathrm{2}}}  \mGLsym{;}  \emptyset  \vdash_{\mathsf{MS} }  \mGLnt{l_{{\mathrm{1}}}}  \mGLsym{:}  \mGLnt{A}}
            \inferrule* [flushleft,right=] {
              \Pi_1\\\\
              \vdots
            }{\delta_{{\mathrm{1}}}  \odot  \Delta_{{\mathrm{1}}}  \mGLsym{;}   ( \mGLmv{x}  \mGLsym{:}  \mGLnt{A}  \mGLsym{,}  \Gamma )   \vdash_{\mathsf{MS} }  \mGLnt{l_{{\mathrm{2}}}}  \mGLsym{:}  \mGLnt{B}}
        }{(  \delta_{{\mathrm{1}}}  \mGLsym{,}  \delta_{{\mathrm{2}}}  )   \odot   ( \Delta_{{\mathrm{1}}}  \mGLsym{,}  \Delta_{{\mathrm{2}}} )   \mGLsym{;}  \Gamma  \vdash_{\mathsf{MS} }  \mGLsym{[}  \mGLnt{l_{{\mathrm{1}}}}  \mGLsym{/}  \mGLmv{x}  \mGLsym{]}  \mGLnt{l_{{\mathrm{2}}}}  \mGLsym{:}  \mGLnt{B}}
    \end{gather*}
  \end{center}
  The previous reduction corresponds to the following equation in the interpretation (Definition~\ref{def:full-mgl-intepretation}):
  \begin{gather*}
    \begin{align*}
    & \interpMS{\Pi_2} \circ (id \otimes \varepsilon \otimes id) \circ
      (id \otimes \Mny(id_1 \odot (\Lin\interpMS{\Pi_1} \circ \mathsf{m}_{\mathsf{Lin}} \circ \overline{\eta})) \otimes id_{\interp{\Gamma}}) \circ (id \otimes \overline{\mathsf{m}^\Mny} \otimes id) \\
\{1 \odot A \equiv A\}
\equiv \;\;&
    \interpMS{\Pi_2} \circ (id \otimes \varepsilon \otimes id) \circ
    (id \otimes \Mny(\Lin\interpMS{\Pi_1} \circ \mathsf{m}_{\mathsf{Lin}} \circ \overline{\eta}) \otimes id_{\interp{\Gamma}}) \circ (id \otimes \overline{\mathsf{m}^\Mny} \otimes id) \\
\{\textit{$\varepsilon$ naturality}\}
\equiv \;\; &
    \interpMS{\Pi_2} \circ (id \otimes \interpMS{\Pi_1} \otimes id) \circ
    (id \otimes (\varepsilon \circ \Mny(\mathsf{m}_{\mathsf{Lin}} \circ \overline{\eta}) \otimes) id_{\interp{\Gamma}}) \circ (id \otimes \overline{\mathsf{m}^\Mny} \otimes id) \\
\{\textit{strong monoidal $\Mny$}\}
\equiv \;\; &
    \interpMS{\Pi_2} \circ (id \otimes \interpMS{\Pi_1} \otimes id) \circ
    (id \otimes (\varepsilon \circ \Mny(\mathsf{m}_{\mathsf{Lin}} \circ \overline{\eta}) \circ \overline{m^\Mny} \circ \overline{n^\Mny}) \otimes) id_{\interp{\Gamma}}) \circ (id \otimes \overline{\mathsf{m}^\Mny} \otimes id) \\
\{\textit{$m^\Mny$ naturality}\}
\equiv \;\; &
    \interpMS{\Pi_2} \circ (id \otimes \interpMS{\Pi_1} \otimes id) \circ
    (id \otimes (\varepsilon \circ \Mny(\mathsf{m}_{\mathsf{Lin}}) \circ \overline{m^\Mny} \circ \overline{\Mny (\eta)} \circ \overline{n^\Mny}) \otimes) id_{\interp{\Gamma}}) \circ (id \otimes \overline{\mathsf{m}^\Mny} \otimes id) \\
\{\textit{monoidal adjunction}\}
\equiv \;\; &
    \interpMS{\Pi_2} \circ (id \otimes \interpMS{\Pi_1} \otimes id) \circ
    (id \otimes (\overline{\varepsilon_{\Mny}} \circ \overline{\Mny (\eta)} \circ \overline{n^\Mny}) \otimes) id_{\interp{\Gamma}}) \circ (id \otimes \overline{\mathsf{m}^\Mny} \otimes id) \\
\{\textit{adjunction}\}
\equiv \;\; &
    \interpMS{\Pi_2} \circ (id \otimes \interpMS{\Pi_1} \otimes id) \circ
    (id \otimes \overline{n^\Mny} \otimes id_{\interp{\Gamma}}) \circ (id \otimes \overline{\mathsf{m}^\Mny} \otimes id) \\
\{\textit{strong monoidal $\Mny$}\}
\equiv \;\; &
    \interpMS{\Pi_2} \circ (id \otimes \interpMS{\Pi_1} \otimes id)
    \end{align*}
  \end{gather*}

    Next we consider the second beta-reduction for the modal operators.  The following proof:
  \begin{center}
    \begin{gather*}
    \begin{align*}
      \inferrule* [flushleft,right={ $\mGLdruleMSTXXCutName{}$}] {
        \inferrule* [flushleft,right={ $\mGLdruleMSTXXGrdRName{}$}] {
          \inferrule* [flushleft,right=] {
            \Pi_2\\\\
          \vdots
        }{\delta_{{\mathrm{2}}}  \odot   ( \Delta_{{\mathrm{2}}} )   \vdash_{\mathsf{GS} }  \mGLnt{t}  \mGLsym{:}  \mGLnt{X}}
        }{\mGLnt{r}   *   \delta_{{\mathrm{2}}}   \odot  \Delta_{{\mathrm{2}}}  \mGLsym{;}  \emptyset  \vdash_{\mathsf{MS} }  \mathsf{Grd} \, \mGLnt{r} \, \mGLnt{t}  \mGLsym{:}   \mathsf{Grd} _{ \mGLnt{r} }\, \mGLnt{X}}\\
        \inferrule* [flushleft,right={$\mGLdruleMSTXXGrdLName{}$}] {
          \inferrule* [flushleft,right=] {
          \Pi_1\\\\
          \vdots
        }{(  \delta_{{\mathrm{1}}}  \mGLsym{,}  \mGLnt{r}  )   \odot   ( \Delta_{{\mathrm{1}}}  \mGLsym{,}  \mGLmv{x}  \mGLsym{:}  \mGLnt{X} )   \mGLsym{;}  \Gamma  \vdash_{\mathsf{MS} }  \mGLnt{l}  \mGLsym{:}  \mGLnt{B}}
        }{\delta_{{\mathrm{1}}}  \odot  \Delta_{{\mathrm{1}}}  \mGLsym{;}   ( \mGLmv{z}  \mGLsym{:}   \mathsf{Grd} _{ \mGLnt{r} }\, \mGLnt{X}   \mGLsym{,}  \Gamma )   \vdash_{\mathsf{MS} }  \mathsf{let} \, \mathsf{Grd} \, \mGLnt{r} \, \mGLmv{x}  \mGLsym{=}  \mGLmv{z} \, \mathsf{in} \, \mGLnt{l}  \mGLsym{:}  \mGLnt{B}}
      }{(  \delta_{{\mathrm{1}}}  \mGLsym{,}   \mGLnt{r}   *   \delta_{{\mathrm{2}}}   )   \odot   ( \Delta_{{\mathrm{1}}}  \mGLsym{,}  \Delta_{{\mathrm{2}}} )   \mGLsym{;}  \Gamma  \vdash_{\mathsf{MS} }  \mathsf{let} \, \mathsf{Grd} \, \mGLnt{r} \, \mGLmv{x}  \mGLsym{=}  \mathsf{Grd} \, \mGLnt{r} \, \mGLnt{t} \, \mathsf{in} \, \mGLnt{l}  \mGLsym{:}  \mGLnt{B}}
    \end{align*}
    \end{gather*}
  \end{center}
  reduces to the following proof:
  \begin{center}
    \begin{gather*}
    \begin{align*}
      \inferrule* [flushleft,right=$\mGLdruleMSTXXGCutName{}$] {
        \inferrule* [flushleft,right=] {
            \Pi_2\\\\
          \vdots
        }{\delta_{{\mathrm{2}}}  \odot  \Delta_{{\mathrm{2}}}  \vdash_{\mathsf{GS} }  \mGLnt{t}  \mGLsym{:}  \mGLnt{X}}\\
        \inferrule* [flushleft,right=] {
          \Pi_1\\\\
          \vdots
        }{(  \delta_{{\mathrm{1}}}  \mGLsym{,}  \mGLnt{r}  )   \odot   ( \Delta_{{\mathrm{1}}}  \mGLsym{,}  \mGLmv{x}  \mGLsym{:}  \mGLnt{X} )   \mGLsym{;}  \Gamma  \vdash_{\mathsf{MS} }  \mGLnt{l}  \mGLsym{:}  \mGLnt{B}}
      }{(  \delta_{{\mathrm{1}}}  \mGLsym{,}   \mGLnt{r}   *   \delta_{{\mathrm{2}}}   )   \odot   ( \Delta_{{\mathrm{1}}}  \mGLsym{,}  \Delta_{{\mathrm{2}}} )   \mGLsym{;}  \Gamma  \vdash_{\mathsf{MS} }  \mGLsym{[}  \mGLnt{t}  \mGLsym{/}  \mGLmv{x}  \mGLsym{]}  \mGLnt{l}  \mGLsym{:}  \mGLnt{B}}
    \end{align*}
    \end{gather*}
  \end{center}
  We then have the following interpretation of the lhs:
  \begin{align*}
    & \interp{\Pi_1} \circ (id \otimes (\Mny(r \odot \interp{\Pi_2}) \circ \Mny (\overline{m_{\boxtimes,r,X,Y}})
    \circ \overline{m^\Mny}) \otimes id_{\interp{\Gamma}})
  \end{align*}
  which is exactly the interpretation of the rhs, due to the strict action
  properties and the strictness of the exponential action.
\end{proof}

%% file: mGL-soundness-theorem-eq-theory-ottput.tex
\subsubsection{Soundness of the (rest of the) equational theory}
\label{subsec:proof-soundness-eq-theory-rest}

By induction of the definition of the equational theory $\equiv$ (Definition~\ref{sec:full-eq-theory}).

\paragraph{$\mathsf{GS}$ system}

\begin{itemize}
\item (\textsc{contr-sym})
  \begin{align*}
    \tag{lhs}
       & \delta_{{\mathrm{1}}} \odot \Delta_{{\mathrm{1}}} \boxtimes \mGLsym{(}  \mGLnt{r}  +  \mGLnt{s}  \mGLsym{)} \odot \mGLnt{X} \boxtimes \delta_{{\mathrm{2}}} \odot \Delta_{{\mathrm{2}}}
        \xrightarrow{\interp{\Pi} \circ (id \boxtimes (\mathsf{contr}_{r,s,X}) \boxtimes id)}
                    \interp{ \mGLnt{Y} } \\
\equiv &
    \tag{rhs}
       \delta_{{\mathrm{1}}} \odot \Delta_{{\mathrm{1}}} \boxtimes \mGLsym{(}  \mGLnt{s}  +  \mGLnt{r}  \mGLsym{)} \odot \mGLnt{X} \boxtimes \delta_{{\mathrm{2}}} \odot \Delta_{{\mathrm{2}}}
        \xrightarrow{\interp{\Pi} \circ (id \boxtimes c_{X,X} \boxtimes id) \circ (id \boxtimes \mathsf{contr_{s,r,X}} \boxtimes id)}
        \interp{ \mGLnt{Y} } \\
  \end{align*}
  by the symmetry of the colax monoidal part of the exponential
  action (Definition~\ref{def:strict-action});

\item (\textsc{contr-unitL})
  \begin{align*}
    \tag{lhs}
   &  (\delta_{{\mathrm{1}}} \odot \Delta_{{\mathrm{1}}}) \boxtimes (\mGLnt{r} \odot \interp{X})
        \boxtimes (\delta_{{\mathrm{2}}} \odot \Delta_{{\mathrm{2}}}) \\
   & \xrightarrow{\interp{\Pi} \circ (\lambda \boxtimes id) \circ (id \boxtimes \func{weak}_{X} \boxtimes id) \circ (id \boxtimes \mathsf{contr}_{r,0,X} \boxtimes id)} \interp{ \mGLnt{Y} } \\
    \tag{rhs} \equiv & (\delta_{{\mathrm{1}}} \odot \Delta_{{\mathrm{1}}}) \boxtimes (\mGLnt{r} \odot \interp{X})
         \boxtimes (\delta_{{\mathrm{2}}} \odot \Delta_{{\mathrm{2}}})
         \xrightarrow{\interp{\Pi} }
         \interp{ \mGLnt{Y} }
  \end{align*}
  By the unitality of the colax monoidal action.

\item (\textsc{contr-unitR})
Similar to the above by unitality

\item (\textsc{sub-comm-conv})
 Let $f = \interp{\delta  \leq  \delta'}$, then:
 \begin{align*}
 \tag{lhs}
 & \mathcal{G}[ \delta' \odot \Delta ]
  \xrightarrow{\interp{\Phi}(\interp{\Pi}) \circ \mathcal{G}(f)}
   \interp{X'} \\
 \tag{rhs}
= & \mathcal{G}[ \delta' \odot \Delta ]
  \xrightarrow{\interp{\Phi}(\interp{\Pi \circ \mathcal{G}(f)})}
  \interp{X'}
 \end{align*}
 by naturality properties of the underlying transformations.

\item (\textsc{contr-assoc}) By the associativity of $\func{contr}$.

\item (\textsc{ex-ex}) By the involutive property of the commutativity
natural isomorphism $c$ for $\boxtimes$.

\item (\textsc{sub-refl}) Since $\mathcal{R}$ is a thin category
then any interpretation of $\delta_{{\mathrm{1}}}  \leq  \delta_{{\mathrm{1}}}$ is the identity morphisms,
and thus this equation is preserved
 in the semantics by the functor identity property of $\func{Grd}$
 in the first argument.

\item (\textsc{sub-trans}) by the functoriality of $\func{Grd}$ in the
first argument.

\item (\textsc{contr-mono}) follows by naturality of $\func{contr}$
with respect to the input functor $F\ r\ s\ X = (r + s) \odot X$
and the output functor $G\ r\ s\ X = (r \odot X) \boxtimes (s \odot X)$.

\item (\textsc{sub-unitl})
\begin{align*}
\tag{lhs}
& (  \delta_{{\mathrm{1}}}   \odot  \interp{  \Delta_{{\mathrm{1}}}  }  )   \boxtimes   ( \mGLnt{s}  \odot  \mathsf{J} )   \boxtimes   (  \delta_{{\mathrm{2}}}   \odot  \interp{  \Delta_{{\mathrm{2}}}  }  )
\xrightarrow{\interp{\Pi} \circ (\lambda \boxtimes id) \circ (id \boxtimes n_{\mathsf{J},r} \boxtimes id) \circ (id \boxtimes (\interp{\mGLnt{r}  \leq  \mGLnt{s}} \odot \mathsf{J}) \boxtimes id)}
\interp{X} \\
\tag{rhs}
\equiv \; &
(  \delta_{{\mathrm{1}}}   \odot  \interp{  \Delta_{{\mathrm{1}}}  }  )   \boxtimes   ( \mGLnt{s}  \odot  \mathsf{J} )   \boxtimes   (  \delta_{{\mathrm{2}}}   \odot  \interp{  \Delta_{{\mathrm{2}}}  }  )
\xrightarrow{\interp{\Pi} \circ (\lambda \boxtimes id) \circ (id \boxtimes n_{\mathsf{J},s} \boxtimes id)}
\interp{X}
\end{align*}
where $n_{\mathsf{J},r} \circ (\interp{\mGLnt{r}  \leq  \mGLnt{s}} \odot \mathsf{J}) \equiv n_{\mathsf{J},s}$ follows by the terminality of $\mathsf{J}$ / or alternatively via naturality of $n_{\mathsf{J}}$ in
the first parameter.

\item (\textsc{sub-tensorL}) follows by naturality of $n_{\boxtimes}$ in the first parameter (grade parameter).

\item (\textsc{mult-mono})

\hspace{-2em}\begin{align*}
& \delta_{{\mathrm{1}}}   \odot  \interp{  \Delta_{{\mathrm{1}}}  }    \boxtimes     (  \mGLnt{r'}  *  \delta'_{{\mathrm{2}}}  )    \odot  \interp{  \Delta_{{\mathrm{2}}}  }    \boxtimes    \delta_{{\mathrm{3}}}   \odot  \interp{  \Delta_{{\mathrm{3}}}  } \\
& \xrightarrow{id \boxtimes (r' \odot (\interp{\delta_{{\mathrm{2}}}  \leq  \delta'_{{\mathrm{2}}}} \odot \interp{ \Pi_{{\mathrm{1}}} })) \boxtimes id}
\delta_{{\mathrm{1}}}   \odot  \interp{  \Delta_{{\mathrm{1}}}  }    \boxtimes   \mGLnt{r'}  \odot  \mGLnt{X}   \boxtimes    \delta_{{\mathrm{3}}}   \odot  \interp{  \Delta_{{\mathrm{3}}}  } \\
& \xrightarrow{id \boxtimes (\interp{\mGLnt{r}  \leq  \mGLnt{r'}} \odot X) \boxtimes id}
\delta_{{\mathrm{1}}}   \odot  \interp{  \Delta_{{\mathrm{1}}}  }    \boxtimes   \mGLnt{r}  \odot  \mGLnt{X}   \boxtimes    \delta_{{\mathrm{3}}}   \odot  \interp{  \Delta_{{\mathrm{3}}}  } \\
& \xrightarrow{\interp{\Pi_{{\mathrm{2}}}} } Y \\
\\
\equiv &
\delta_{{\mathrm{1}}}   \odot  \interp{  \Delta_{{\mathrm{1}}}  }    \boxtimes     (  \mGLnt{r'}  *  \delta'_{{\mathrm{2}}}  )    \odot  \interp{  \Delta_{{\mathrm{2}}}  }    \boxtimes    \delta_{{\mathrm{3}}}   \odot  \interp{  \Delta_{{\mathrm{3}}}  } \\
& \xrightarrow{id \boxtimes (\interp{\mGLnt{r}  *  \delta_{{\mathrm{2}}}  \leq  \mGLnt{r'}  *  \delta'_{{\mathrm{2}}}} \odot \Delta_{{\mathrm{2}}})
             \boxtimes id}
\delta_{{\mathrm{1}}}   \odot  \interp{  \Delta_{{\mathrm{1}}}  }    \boxtimes     (  \mGLnt{r}  *  \delta_{{\mathrm{2}}}  )    \odot  \interp{  \Delta_{{\mathrm{2}}}  }    \boxtimes    \delta_{{\mathrm{3}}}   \odot  \interp{  \Delta_{{\mathrm{3}}}  } \\
& \xrightarrow{id \boxtimes (r \odot \interp{\Pi_{{\mathrm{1}}}}) \boxtimes id}
\delta_{{\mathrm{1}}}   \odot  \interp{  \Delta_{{\mathrm{1}}}  }    \boxtimes   \mGLnt{r}  \odot  \mGLnt{X}   \boxtimes    \delta_{{\mathrm{3}}}   \odot  \interp{  \Delta_{{\mathrm{3}}}  } \\
& \xrightarrow{\interp{\Pi_{{\mathrm{2}}} }} Y
\end{align*}
by functoriality of $\odot$ in both arguments, and its definition, such that:
\begin{align*}
& (\interp{\mGLnt{r}  \leq  \mGLnt{r'}} \odot id_X) \circ (id_{r'} \odot (\interp{\delta_{{\mathrm{2}}}  \leq  \delta'_{{\mathrm{2}}}} \odot \interp{ \Pi_{{\mathrm{1}}} })) \\
\equiv \; & \interp{\mGLnt{r}  \leq  \mGLnt{r'}} \odot (\interp{\delta_{{\mathrm{2}}}  \leq  \delta'_{{\mathrm{2}}}} \odot \interp{ \Pi_{{\mathrm{1}}} }))
\\
\equiv \; & \interp{\mGLnt{r}  *  \delta_{{\mathrm{2}}}  \leq  \mGLnt{r'}  *  \delta'_{{\mathrm{2}}}} \odot \interp{\Pi_{{\mathrm{1}}}}
\\
\equiv \; & (id_r \odot \interp{\Pi_{{\mathrm{1}}}}) \circ (\interp{\mGLnt{r}  *  \delta_{{\mathrm{2}}}  \leq  \mGLnt{r'}  *  \delta'_{{\mathrm{2}}}} \odot id_{\Delta_{{\mathrm{2}}}})
\end{align*}

\end{itemize}

\paragraph{$\MS{}$ system}

\begin{itemize}
\item (\textsc{sub-comm-conv})

 Let $f = \interp{\delta  \leq  \delta'}$, then:
 \begin{align*}
 \tag{lhs}
 & \mathcal{G}[ \delta' \odot \Delta ] \otimes \interp{ \Gamma }
  \xrightarrow{\interp{\Phi}(\interp{\Pi}) \circ \mathcal{G}(f)}
   \interp{X'} \\
 \tag{rhs}
= & \mathcal{G}[ \delta' \odot \Delta ] \otimes \interp{ \Gamma}
  \xrightarrow{\interp{\Phi}(\interp{\Pi \circ \mathcal{G}(f)})}
  \interp{X'}
 \end{align*}
 by naturality properties of the underlying transformations.

\item (\textsc{contr-unitL}) The lifting of (\textsc{contr-unitL}) in $\GS{}$ system.

\item (\textsc{contr-unitR}) The lifting of (\textsc{contr-unitL}) in $\GS{}$ system.

\item (\textsc{contr-sym}) The lifting of (\textsc{contr-unitL}) in $\GS{}$ system.

\item (\textsc{contr-assoc}) The lifting of (\textsc{contr-unitL}) in $\GS{}$ system.

\item (\textsc{gex-gex}) Isomorphism property of symmetry in $\mathcal{M}$.

\item (\textsc{ex-ex}) Isomorphism property of symmetry in $\mathcal{C}$.

\item (\textsc{sub-refl}) The lifting of (\textsc{sub-refl}) in $\GS{}$ system.

\item (\textsc{sub-trans}) The lifting of (\textsc{sub-trans}) in $\GS{}$ system.

\item (\textsc{contr-mono})  The lifting of (\textsc{contr-mono}) in $\GS{}$ system.

\item (\textsc{sub-unitL}) The lifting of (\textsc{sub-unitL}) in $\GS{}$ system.

\item (\textsc{sub-tensorL}) The lifting of (\textsc{sub-unitL}) in $\GS{}$ system.

\item (\textsc{mult-mono}) The lifting of (\textsc{mult-mono}) in $\GS{}$ system.

\end{itemize}

%% file: mGL-completeness-theorem-ottput.tex
We define two categories, one for each fragment of $\mGL{}$,
along with an action on the category corresponding to the graded fragment,
and an adjunction between them.
Then define an equivelence relation on arrows based on the equational theory.
Using the equivelence we define a generic model such that the interpretation of
two derivations are equal exactly when those derivations are equivelent.
Completeness follows straightforwardly from this construction.

\begin{definition}[Syntactic categories ] Given a preordered semring over $\mathcal{R}$, construct the symmetric monoidal closed category freely from the syntax of graded fragment of $\mGL{}$:
    $$\text{Obj}(\mathbb{G}) ::=  \langle r, \mathsf{J}  \rangle \mid \langle r, X \rangle  \boxtimes \langle s, Y \rangle \mid \langle r, \mathsf{Lin} \, \mGLnt{A} \rangle $$
    $$ \text{Hom}_{\mathbb{G}}( \langle r, X \rangle  , \langle 1, Y \rangle )= \{ t \mid \mGLnt{r}  \odot  \mGLmv{x}  \mGLsym{:}  \mGLnt{X}  \vdash_{\mathsf{GS} }  \mGLnt{t}  \mGLsym{:}  \mGLnt{Y} \} $$

    Similarly, construct the symmetric monoidal closed category freely from the syntax of mixed fragment of $\mGL{}$:
    $$\text{Obj}(\mathbb{L}) ::= \mathsf{I} \mid \mGLnt{A}  \otimes  \mGLnt{B} \mid \mGLnt{A}  \multimap  \mGLnt{B} \mid \mathsf{Grd} _{ \mGLnt{r} }\, \mGLnt{X} $$
    $$\text{Hom}_{\mathbb{L}}(A, B)= \{ l \mid \emptyset  \odot  \emptyset  \mGLsym{;}  \emptyset  \vdash_{\mathsf{MS} }  \mGLnt{l}  \mGLsym{:}  \mGLnt{A}  \multimap  \mGLnt{B} \} $$
    With $ r \in \mathcal{R}$.
\end{definition}
\begin{definition}
Let $ \odot : \mathcal{R} \times \mathbb{G} \mto \mathbb{G} $ be
\[    \begin{array}{rcl}
    r \odot \langle s, \mGLnt{X}  \rangle & = & \langle r * s, \mGLnt{X}  \rangle \\
    r \odot (\langle s_1, X \rangle \boxtimes \langle s_2 , Y \rangle ) & = & r \odot (\langle s_1, X \rangle) \boxtimes r \odot(\langle s_2, Y \rangle )\\

\end{array}
\]

such that for any $ r \in \mathcal{R}$ and $ \langle s , X \rangle \in \text{Obj}(\mathbb{G}) $ there are natural transformations
\[
    \begin{array}{rll}
      0 \odot \langle s , X \rangle & \mto^{\mathsf{weak}_{\langle s , X \rangle}} & \langle 0 , J \rangle \\
      (r_1 + r_2) \odot \langle s , X \rangle & \mto^{\mathsf{contr}_{r_1,r_2,X}} & (r_1 \odot \langle s , X \rangle) \boxtimes (r_2 \odot \langle s , X \rangle)
    \end{array}
    \]

\end{definition}

\begin{definition}[ $ \mathbb{G} : G \dashv L :  \mathbb{L} $ ]
\[    \begin{array}{rcl}
        L : \mathbb{G} \rightarrow \mathbb{L}  & & \\
        L(\mathsf{I}) & = & \langle 1, \mathsf{Lin} \, \mathsf{I} \rangle \\
        L(\mGLnt{A}  \otimes  \mGLnt{B} ) & = & L(A) \boxtimes L(B)  \\
        L(\mathsf{Grd} _{ \mGLnt{r} }\, \mGLnt{X} ) & = & \langle r, X \rangle
    \end{array}
\]

\[    \begin{array}{rcl}
    G : \mathbb{L} \rightarrow \mathbb{G}  & & \\
    G(\langle r, \mathsf{J}  \rangle) & = & \mathsf{Grd} _{ \mGLnt{r} }\, \mathsf{J}\\
    G(\langle r, X \rangle \boxtimes \langle s , Y \rangle ) & = & G(\langle r, X \rangle) \otimes G(\langle s, Y \rangle )\\
    G(\langle r, \mathsf{Lin} \, \mGLnt{A} \rangle) & = & \mathsf{Grd} _{ \mGLnt{r} }\, \mGLsym{(}  \mathsf{Lin} \, \mGLnt{A}  \mGLsym{)}
\end{array}
\]

\[
    \begin{array}{lll}
      \begin{array}{lll}
          & \eta : \langle r, X \rangle \mto \func{L}(\func{G}(r,X))\\
          &  \eta = \func{id}_\mathbb{G}
      \end{array}\\\\
      \begin{array}{lll}
          & \epsilon : \func{Grd} \langle 1,\func{Lin}(A) \rangle \mto A\\
          & \epsilon(\langle \mathsf{Grd} _{ 1 }\, \mGLsym{(}  \mathsf{Lin} \, \mGLnt{A}  \mGLsym{)} \rangle) = A
      \end{array}
    \end{array}
\]
\end{definition}

\begin{lemma}
    $(\mathbb{G} : G \dashv L :  \mathbb{L}, \odot ) $
    is an $\mGL{}$ model.
\end{lemma}
\begin{proof}
  $ \odot $ is an exponential action and $ \mathbb{G} : G \dashv L :  \mathbb{L} $ is a symmetric monoidal adjunction, so this follows from the definition of an $\mGL{}$: model.
\end{proof}
 \begin{lemma}
    $\mathsf{Gr}_\odot(\mathcal{R}, \mathbb{G}) = \mathbb{G}$
\end{lemma}
\begin{proof}
    For any $ r$ , $ r \odot \langle 0 , J \rangle = \langle 0 , J \rangle $  So pick $\mathsf{J} $ it be $ \langle 0 , J \rangle $
\end{proof}

\begin{definition}[ generic model ] Given a preordered semring over $\mathcal{R}$, construct the symmetric monoidal category from the syntax of graded fragment of $\mGL{}$.
    $$\text{Obj}(\mathbb{G}_\equiv) ::=  \langle r, \mathsf{J}  \rangle \mid \langle r, X \rangle  \boxtimes \langle s, Y \rangle \mid \langle r, \mathsf{Lin} \, \mGLnt{A} \rangle $$
    $$ \text{Hom}_{\mathbb{G}_\equiv}( \langle r, X \rangle  , \langle 1, Y \rangle )= \{ t \mid \mGLnt{r}  \odot  \mGLmv{x}  \mGLsym{:}  \mGLnt{X}  \vdash_{\mathsf{GS} }  \mGLnt{t}  \mGLsym{:}  \mGLnt{Y} \}_{/ \equiv} $$
    And
   $$\text{Obj}(\mathbb{L}_\equiv) ::= \mathsf{I} \mid \mGLnt{A}  \otimes  \mGLnt{B} \mid \mGLnt{A}  \multimap  \mGLnt{B} \mid \mathsf{Grd} _{ \mGLnt{r} }\, \mGLnt{X} $$
    $$\text{Hom}_{\mathbb{L}_\equiv}(A, B) = \{ l \mid \emptyset  \odot  \emptyset  \mGLsym{;}  \emptyset  \vdash_{\mathsf{MS} }  \mGLnt{l}  \mGLsym{:}  \mGLnt{A}  \multimap  \mGLnt{B} \}_{/ \equiv}  $$

    With $ r \in R$ and $ \equiv $ is the the equivelence relation from the equatonal theory moved across the interpretation.
\end{definition}

Note that moving the equivelence relation across the interpretation assumes coherence, that is that derivations of the same judgement are equivelent, to maintain reflexivity. A statement or proof of coherence is omitted, but it follows from the construction of the equational theory.

\begin{lemma}
    $ (\mathbb{G}_\equiv : G_\equiv \dashv L_\equiv :  \mathbb{L}_\equiv, \odot ) $ is an $\mGL{}$ model
\end{lemma}
\begin{proof}
    $ \odot $ is an exponential action and $(\mathbb{G}_\equiv : G_\equiv \dashv L_\equiv :  \mathbb{L}_\equiv, \odot )$ is a symmetric monoidal adjunction, so this follows from the definition of an $\mGL{}$: model
\end{proof}
\mGLLCompletenessTheorem
\begin{proof}
    If $\interp{ \Pi_{{\mathrm{1}}} } = \interp{ \Pi_{{\mathrm{2}}} }$ in all mixed graded/linear models, then
    $\interp{ \Pi_{{\mathrm{1}}} } = \interp{ \Pi_{{\mathrm{2}}} }$ in
    $ (\mathbb{G}_\equiv : G_\equiv \dashv L_\equiv :  \mathbb{L}_\equiv, \odot ) $ so
    $\Pi_{{\mathrm{1}}} \equiv \Pi_{{\mathrm{2}}}$.
\end{proof}

%% file: substitution-proof-ottput.tex
\substitutionLemma{}
\begin{proof}
  The most interesting cases of this proof are in part i.  Thus, we
  prove part i in full, and omit the parts ii and iii, because they
  are similar.  Throughout this proof we use the following definition:
  $(\mGLnt{s_{{\mathrm{1}}}}  \mGLsym{,} \, ... \, \mGLsym{,}  \mGLnt{s_{\mGLmv{n}}},\mGLmv{x_{{\mathrm{1}}}}  \mGLsym{:}  \mGLnt{X}  \mGLsym{,} \, ... \, \mGLsym{,}  \mGLmv{x_{\mGLmv{n}}}  \mGLsym{:}  \mGLnt{X}) \in (\delta,\Delta)$ iff there exists $\delta'$,
  $\delta''$, $\Delta'$, and $\Delta''$ such that $\delta  \mGLsym{=}   (  \delta'  \mGLsym{,}  \mGLnt{s_{{\mathrm{1}}}}  \mGLsym{,} \, ... \, \mGLsym{,}  \mGLnt{s_{\mGLmv{n}}}  \mGLsym{,}  \delta''  )$,
  $\Delta  \mGLsym{=}   ( \Delta'  \mGLsym{,}  \mGLmv{x_{{\mathrm{1}}}}  \mGLsym{:}  \mGLnt{X}  \mGLsym{,} \, ... \, \mGLsym{,}  \mGLmv{x_{\mGLmv{n}}}  \mGLsym{:}  \mGLnt{X}  \mGLsym{,}  \Delta'' )$, $| \delta' | = | \Delta' |$, $| \delta'' |
  = | \Delta'' |$.

  This is a proof by induction on the second assumed typing derivation.

  \begin{itemize}
  \item[] Case.\ \\
    \begin{center}
      \begin{math}
        $$\mprset{flushleft}
        \inferrule* [right=$\mGLdruleGTXXUnitIName{}$] {
          \
        }{\emptyset  \odot  \emptyset  \vdash_{\mathsf{GT} }  \mathsf{j}  \mGLsym{:}  \mathsf{J}}
      \end{math}
    \end{center}
    The result holds trivially.

  \item[] Case.\ \\
    \begin{center}
      \begin{math}
        $$\mprset{flushleft}
        \inferrule* [right=$\mGLdruleGTXXUnitEName{}$] {
          \delta'_{{\mathrm{2}}}  \odot  \Delta'_{{\mathrm{2}}}  \vdash_{\mathsf{GT} }  \mGLnt{t'_{{\mathrm{2}}}}  \mGLsym{:}  \mathsf{J}
          \\\\
              (  \delta'_{{\mathrm{1}}}  \mGLsym{,}  \mGLnt{s}  \mGLsym{,}  \delta'_{{\mathrm{3}}}  )   \odot   ( \Delta'_{{\mathrm{1}}}  \mGLsym{,}  \mGLmv{y}  \mGLsym{:}  \mathsf{J}  \mGLsym{,}  \Delta'_{{\mathrm{3}}} )   \vdash_{\mathsf{GT} }  \mGLnt{t''_{{\mathrm{2}}}}  \mGLsym{:}  \mGLnt{Y}
        }{(  \delta'_{{\mathrm{1}}}  \mGLsym{,}  \mGLnt{s}  *  \delta'_{{\mathrm{2}}}  \mGLsym{,}  \delta'_{{\mathrm{3}}}  )   \odot   ( \Delta'_{{\mathrm{1}}}  \mGLsym{,}  \Delta'_{{\mathrm{2}}}  \mGLsym{,}  \Delta'_{{\mathrm{3}}} )   \vdash_{\mathsf{GT} }  \mathsf{let} \, \mathsf{j} \, \mGLsym{(}  \mGLmv{y}  \mGLsym{)}  \mGLsym{=}  \mGLnt{t'_{{\mathrm{2}}}} \, \mathsf{in} \, \mGLnt{t''_{{\mathrm{2}}}}  \mGLsym{:}  \mGLnt{Y}}
      \end{math}
    \end{center}
    In this case we have the following assignments:
    \[
    \begin{array}{lll}
      (  \delta_{{\mathrm{1}}}  \mGLsym{,}  \mGLnt{r_{{\mathrm{1}}}}  \mGLsym{,} \, ... \, \mGLsym{,}  \mGLnt{r_{\mGLmv{n}}}  \mGLsym{,}  \delta_{{\mathrm{3}}}  ) = (  \delta'_{{\mathrm{1}}}  \mGLsym{,}  \mGLnt{s}  *  \delta'_{{\mathrm{2}}}  \mGLsym{,}  \delta'_{{\mathrm{3}}}  )\\
      ( \Delta_{{\mathrm{1}}}  \mGLsym{,}  \mGLmv{x_{{\mathrm{1}}}}  \mGLsym{:}  \mGLnt{X}  \mGLsym{,} \, ... \, \mGLsym{,}  \mGLmv{x_{\mGLmv{n}}}  \mGLsym{:}  \mGLnt{X}  \mGLsym{,}  \Delta_{{\mathrm{3}}} ) = ( \Delta'_{{\mathrm{1}}}  \mGLsym{,}  \Delta'_{{\mathrm{2}}}  \mGLsym{,}  \Delta'_{{\mathrm{3}}} )\\
      \mGLnt{t_{{\mathrm{2}}}} = \mGLsym{(}  \mathsf{let} \, \mathsf{j} \, \mGLsym{(}  \mGLmv{y}  \mGLsym{)}  \mGLsym{=}  \mGLnt{t'_{{\mathrm{2}}}} \, \mathsf{in} \, \mGLnt{t''_{{\mathrm{2}}}}  \mGLsym{)}
    \end{array}
    \]
    It suffices to case split over which vector and context $(\mGLnt{r_{{\mathrm{1}}}}  \mGLsym{,} \, ... \, \mGLsym{,}  \mGLnt{r_{\mGLmv{n}}},
    \mGLmv{x_{{\mathrm{1}}}}  \mGLsym{:}  \mGLnt{X}  \mGLsym{,} \, ... \, \mGLsym{,}  \mGLmv{x_{\mGLmv{n}}}  \mGLsym{:}  \mGLnt{X})$ fall into. Thus, we have the following cases:
    \begin{enumerate}
    \item Suppose $\mGLsym{(}  \mGLnt{r_{{\mathrm{1}}}}  \mGLsym{,} \, ... \, \mGLsym{,}  \mGLnt{r_{\mGLmv{n}}}  \mGLsym{,}  \mGLmv{x_{{\mathrm{1}}}}  \mGLsym{:}  \mGLnt{X}  \mGLsym{,} \, ... \, \mGLsym{,}  \mGLmv{x_{\mGLmv{n}}}  \mGLsym{:}  \mGLnt{X}  \mGLsym{)} \, \in \, \mGLsym{(}  \delta'_{{\mathrm{1}}}  \mGLsym{,}  \Delta'_{{\mathrm{1}}}  \mGLsym{)}$.  Hence, we have the
      following situation:
      \begin{center}
        \begin{math}
          $$\mprset{flushleft}
          \inferrule* [right=$\mGLdruleGTXXUnitEName{}$] {
            \delta'_{{\mathrm{2}}}  \odot  \Delta'_{{\mathrm{2}}}  \vdash_{\mathsf{GT} }  \mGLnt{t'_{{\mathrm{2}}}}  \mGLsym{:}  \mathsf{J}\\\\
            (  \delta''_{{\mathrm{1}}}  \mGLsym{,}  \mGLnt{r_{{\mathrm{1}}}}  \mGLsym{,} \, ... \, \mGLsym{,}  \mGLnt{r_{\mGLmv{n}}}  \mGLsym{,}  \delta'''_{{\mathrm{1}}}  \mGLsym{,}  \mGLnt{s}  \mGLsym{,}  \delta'_{{\mathrm{3}}}  )   \odot   ( \Delta''_{{\mathrm{1}}}  \mGLsym{,}  \mGLmv{x_{{\mathrm{1}}}}  \mGLsym{:}  \mGLnt{X}  \mGLsym{,} \, ... \, \mGLsym{,}  \mGLmv{x_{\mGLmv{n}}}  \mGLsym{:}  \mGLnt{X}  \mGLsym{,}  \Delta''_{{\mathrm{1}}}  \mGLsym{,}  \mGLmv{y}  \mGLsym{:}  \mathsf{J}  \mGLsym{,}  \Delta'_{{\mathrm{3}}} )   \vdash_{\mathsf{GT} }  \mGLnt{t''_{{\mathrm{2}}}}  \mGLsym{:}  \mGLnt{Y}
          }{(  \delta''_{{\mathrm{1}}}  \mGLsym{,}  \mGLnt{r_{{\mathrm{1}}}}  \mGLsym{,} \, ... \, \mGLsym{,}  \mGLnt{r_{\mGLmv{n}}}  \mGLsym{,}  \delta'''_{{\mathrm{1}}}  \mGLsym{,}  \mGLnt{s}  *  \delta'_{{\mathrm{2}}}  \mGLsym{,}  \delta'_{{\mathrm{3}}}  )   \odot   ( \Delta''_{{\mathrm{1}}}  \mGLsym{,}  \mGLmv{x_{{\mathrm{1}}}}  \mGLsym{:}  \mGLnt{X}  \mGLsym{,} \, ... \, \mGLsym{,}  \mGLmv{x_{\mGLmv{n}}}  \mGLsym{:}  \mGLnt{X}  \mGLsym{,}  \Delta''_{{\mathrm{1}}}  \mGLsym{,}  \Delta'_{{\mathrm{2}}}  \mGLsym{,}  \Delta'_{{\mathrm{3}}} )   \vdash_{\mathsf{GT} }  \mathsf{let} \, \mathsf{j} \, \mGLsym{(}  \mGLmv{y}  \mGLsym{)}  \mGLsym{=}  \mGLnt{t'_{{\mathrm{2}}}} \, \mathsf{in} \, \mGLnt{t''_{{\mathrm{2}}}}  \mGLsym{:}  \mGLnt{Y}}
        \end{math}
      \end{center}
      Now by the inductive hypothesis we have
      $(  \delta''_{{\mathrm{1}}}  \mGLsym{,}  \mGLsym{(}  \mGLnt{r_{{\mathrm{1}}}}  + \, \cdots \, +  \mGLnt{r_{\mGLmv{n}}}  \mGLsym{)}  *  \delta_{{\mathrm{2}}}  \mGLsym{,}  \delta'''_{{\mathrm{1}}}  \mGLsym{,}  \mGLnt{s}  \mGLsym{,}  \delta'_{{\mathrm{3}}}  )   \odot   ( \Delta''_{{\mathrm{1}}}  \mGLsym{,}  \Delta_{{\mathrm{2}}}  \mGLsym{,}  \Delta''_{{\mathrm{1}}}  \mGLsym{,}  \mGLmv{y}  \mGLsym{:}  \mathsf{J}  \mGLsym{,}  \Delta'_{{\mathrm{3}}} )   \vdash_{\mathsf{GT} }  \mGLsym{[}   \mGLnt{t_{{\mathrm{1}}}} ,\ldots, \mGLnt{t_{{\mathrm{1}}}}   \mGLsym{/}  \mGLmv{x_{{\mathrm{1}}}}  \mGLsym{,} \, ... \, \mGLsym{,}  \mGLmv{x_{\mGLmv{n}}}  \mGLsym{]}  \mGLnt{t''_{{\mathrm{2}}}}  \mGLsym{:}  \mGLnt{Y}$.

      Next we simply re-apply the rule obtaining:
      \begin{gather*}
        \begin{align*}
          \mprset{flushleft}
          \inferrule* [right=$\mGLdruleGTXXUnitEName{}$] {
            \delta'_{{\mathrm{2}}}  \odot  \Delta'_{{\mathrm{2}}}  \vdash_{\mathsf{GT} }  \mGLnt{t'_{{\mathrm{2}}}}  \mGLsym{:}  \mathsf{J}\\\\
            (  \delta''_{{\mathrm{1}}}  \mGLsym{,}  \mGLsym{(}  \mGLnt{r_{{\mathrm{1}}}}  + \, \cdots \, +  \mGLnt{r_{\mGLmv{n}}}  \mGLsym{)}  *  \delta_{{\mathrm{2}}}  \mGLsym{,}  \delta'''_{{\mathrm{1}}}  \mGLsym{,}  \mGLnt{s}  \mGLsym{,}  \delta'_{{\mathrm{3}}}  )   \odot   ( \Delta''_{{\mathrm{1}}}  \mGLsym{,}  \Delta_{{\mathrm{2}}}  \mGLsym{,}  \Delta''_{{\mathrm{1}}}  \mGLsym{,}  \mGLmv{y}  \mGLsym{:}  \mathsf{J}  \mGLsym{,}  \Delta'_{{\mathrm{3}}} )   \vdash_{\mathsf{GT} }  \mGLsym{[}   \mGLnt{t_{{\mathrm{1}}}} ,\ldots, \mGLnt{t_{{\mathrm{1}}}}   \mGLsym{/}  \mGLmv{x_{{\mathrm{1}}}}  \mGLsym{,} \, ... \, \mGLsym{,}  \mGLmv{x_{\mGLmv{n}}}  \mGLsym{]}  \mGLnt{t''_{{\mathrm{2}}}}  \mGLsym{:}  \mGLnt{Y}
          }{(  \delta''_{{\mathrm{1}}}  \mGLsym{,}  \mGLsym{(}  \mGLnt{r_{{\mathrm{1}}}}  + \, \cdots \, +  \mGLnt{r_{\mGLmv{n}}}  \mGLsym{)}  *  \delta_{{\mathrm{2}}}  \mGLsym{,}  \delta'''_{{\mathrm{1}}}  \mGLsym{,}  \mGLnt{s}  *  \delta'_{{\mathrm{2}}}  \mGLsym{,}  \delta'_{{\mathrm{3}}}  )   \odot   ( \Delta''_{{\mathrm{1}}}  \mGLsym{,}  \Delta_{{\mathrm{2}}}  \mGLsym{,}  \Delta''_{{\mathrm{1}}}  \mGLsym{,}  \Delta'_{{\mathrm{2}}}  \mGLsym{,}  \Delta'_{{\mathrm{3}}} )   \vdash_{\mathsf{GT} }  \mathsf{let} \, \mathsf{j} \, \mGLsym{(}  \mGLmv{y}  \mGLsym{)}  \mGLsym{=}  \mGLnt{t'_{{\mathrm{2}}}} \, \mathsf{in} \, \mGLsym{[}   \mGLnt{t_{{\mathrm{1}}}} ,\ldots, \mGLnt{t_{{\mathrm{1}}}}   \mGLsym{/}  \mGLmv{x_{{\mathrm{1}}}}  \mGLsym{,} \, ... \, \mGLsym{,}  \mGLmv{x_{\mGLmv{n}}}  \mGLsym{]}  \mGLnt{t''_{{\mathrm{2}}}}  \mGLsym{:}  \mGLnt{Y}}
        \end{align*}
      \end{gather*}
      We know that $\mGLmv{x_{{\mathrm{1}}}}  \mGLsym{,} \, ... \, \mGLsym{,}  \mGLmv{x_{\mGLmv{n}}}$ are all free in $\mGLnt{t'_{{\mathrm{2}}}}$ which implies that
      $\mathsf{let} \, \mathsf{j} \, \mGLsym{(}  \mGLmv{y}  \mGLsym{)}  \mGLsym{=}  \mGLnt{t'_{{\mathrm{2}}}} \, \mathsf{in} \, \mGLsym{[}   \mGLnt{t_{{\mathrm{1}}}} ,\ldots, \mGLnt{t_{{\mathrm{1}}}}   \mGLsym{/}  \mGLmv{x_{{\mathrm{1}}}}  \mGLsym{,} \, ... \, \mGLsym{,}  \mGLmv{x_{\mGLmv{n}}}  \mGLsym{]}  \mGLnt{t''_{{\mathrm{2}}}} = \mGLsym{[}   \mGLnt{t_{{\mathrm{1}}}} ,\ldots, \mGLnt{t_{{\mathrm{1}}}}   \mGLsym{/}  \mGLmv{x_{{\mathrm{1}}}}  \mGLsym{,} \, ... \, \mGLsym{,}  \mGLmv{x_{\mGLmv{n}}}  \mGLsym{]}  \mGLsym{(}  \mathsf{let} \, \mathsf{j} \, \mGLsym{(}  \mGLmv{y}  \mGLsym{)}  \mGLsym{=}  \mGLnt{t'_{{\mathrm{2}}}} \, \mathsf{in} \, \mGLnt{t''_{{\mathrm{2}}}}  \mGLsym{)}$, and thus we
      have derived the goal judgment.

    \item Suppose $\mGLsym{(}  \mGLnt{r_{{\mathrm{1}}}}  \mGLsym{,} \, ... \, \mGLsym{,}  \mGLnt{r_{\mGLmv{n}}}  \mGLsym{,}  \mGLmv{x_{{\mathrm{1}}}}  \mGLsym{:}  \mGLnt{X}  \mGLsym{,} \, ... \, \mGLsym{,}  \mGLmv{x_{\mGLmv{n}}}  \mGLsym{:}  \mGLnt{X}  \mGLsym{)} \, \in \, \mGLsym{(}  \delta'_{{\mathrm{2}}}  \mGLsym{,}  \Delta'_{{\mathrm{2}}}  \mGLsym{)}$.  Similar to the previous case.
    \item Suppose $\mGLsym{(}  \mGLnt{r_{{\mathrm{1}}}}  \mGLsym{,} \, ... \, \mGLsym{,}  \mGLnt{r_{\mGLmv{n}}}  \mGLsym{,}  \mGLmv{x_{{\mathrm{1}}}}  \mGLsym{:}  \mGLnt{X}  \mGLsym{,} \, ... \, \mGLsym{,}  \mGLmv{x_{\mGLmv{n}}}  \mGLsym{:}  \mGLnt{X}  \mGLsym{)} \, \in \, \mGLsym{(}  \delta'_{{\mathrm{3}}}  \mGLsym{,}  \Delta'_{{\mathrm{3}}}  \mGLsym{)}$.  Similar to the previous case.
    \end{enumerate}

  \item[] Case.\ \\
    \begin{center}
      \begin{math}
        $$\mprset{flushleft}
        \inferrule* [right=$\mGLdruleGTXXTenIName{}$] {
          \delta'_{{\mathrm{1}}}  \odot  \Delta'_{{\mathrm{1}}}  \vdash_{\mathsf{GT} }  \mGLnt{t'_{{\mathrm{1}}}}  \mGLsym{:}  \mGLnt{Y_{{\mathrm{1}}}}
          \\\\
          \delta'_{{\mathrm{2}}}  \odot  \Delta'_{{\mathrm{2}}}  \vdash_{\mathsf{GT} }  \mGLnt{t'_{{\mathrm{2}}}}  \mGLsym{:}  \mGLnt{Y_{{\mathrm{2}}}}
        }{(  \delta'_{{\mathrm{1}}}  \mGLsym{,}  \delta'_{{\mathrm{2}}}  )   \odot   ( \Delta'_{{\mathrm{1}}}  \mGLsym{,}  \Delta'_{{\mathrm{2}}} )   \vdash_{\mathsf{GT} }  \mGLsym{(}  \mGLnt{t'_{{\mathrm{1}}}}  \mGLsym{,}  \mGLnt{t'_{{\mathrm{2}}}}  \mGLsym{)}  \mGLsym{:}  \mGLnt{Y_{{\mathrm{1}}}}  \boxtimes  \mGLnt{Y_{{\mathrm{2}}}}}
      \end{math}
    \end{center}
    This case follows straightforwardly from the application of the
    induction hypothesis to premises of the rule and reapplying the
    rule.

  \item[] Case.\ \\
    \begin{center}
      \begin{math}
        $$\mprset{flushleft}
        \inferrule* [right=$\mGLdruleGTXXSubName{}$] {
          \delta'_{{\mathrm{1}}}  \odot  \Delta'_{{\mathrm{1}}}  \vdash_{\mathsf{GT} }  \mGLnt{t'}  \mGLsym{:}  \mGLnt{Y}  \quad  \delta'_{{\mathrm{1}}}  \leq  \delta'_{{\mathrm{2}}}
        }{\delta'_{{\mathrm{2}}}  \odot  \Delta'_{{\mathrm{1}}}  \vdash_{\mathsf{GT} }  \mGLnt{t'}  \mGLsym{:}  \mGLnt{Y}}
      \end{math}
    \end{center}
    with $\delta'_{{\mathrm{1}}} = (  \delta_{{\mathrm{1}}}  \mGLsym{,}  \mGLnt{r_{{\mathrm{1}}}}  \mGLsym{,} \, ... \, \mGLsym{,}  \mGLnt{r_{\mGLmv{n}}}  \mGLsym{,}  \delta_{{\mathrm{3}}}  )$
    and $\Delta'_{{\mathrm{1}}} = ( \Delta_{{\mathrm{1}}}  \mGLsym{,}  \mGLmv{x_{{\mathrm{1}}}}  \mGLsym{:}  \mGLnt{X}  \mGLsym{,} \, ... \, \mGLsym{,}  \mGLmv{x_{\mGLmv{n}}}  \mGLsym{:}  \mGLnt{X}  \mGLsym{,}  \Delta_{{\mathrm{3}}} )$
    and $\delta'_{{\mathrm{2}}} = (  \gamma_{{\mathrm{1}}}  \mGLsym{,}  \mGLnt{r'_{{\mathrm{1}}}}  \mGLsym{,} \, ... \, \mGLsym{,}  \mGLnt{r'_{\mGLmv{n}}}  \mGLsym{,}  \gamma_{{\mathrm{3}}}  )$.

    By induction we have that:
    \begin{align*}
    (  \delta_{{\mathrm{1}}}  \mGLsym{,}  \mGLsym{(}  \mGLnt{r_{{\mathrm{1}}}}  + \, \cdots \, +  \mGLnt{r_{\mGLmv{n}}}  \mGLsym{)}  *  \delta_{{\mathrm{2}}}  \mGLsym{,}  \delta_{{\mathrm{3}}}  )   \odot   ( \Delta_{{\mathrm{1}}}  \mGLsym{,}  \Delta_{{\mathrm{2}}}  \mGLsym{,}  \Delta_{{\mathrm{3}}} )   \vdash_{\mathsf{GT} }  \mGLsym{[}   \mGLnt{t_{{\mathrm{1}}}} ,\ldots, \mGLnt{t_{{\mathrm{1}}}}   \mGLsym{/}  \mGLmv{x_{{\mathrm{1}}}}  \mGLsym{,} \, ... \, \mGLsym{,}  \mGLmv{x_{\mGLmv{n}}}  \mGLsym{]}  \mGLnt{t'}  \mGLsym{:}  \mGLnt{Y}
    \end{align*}

    Since $\mGLnt{r_{{\mathrm{1}}}}  \leq  \mGLnt{r'_{{\mathrm{1}}}} \wedge \ldots
    \mGLnt{r_{\mGLmv{n}}}  \leq  \mGLnt{r'_{\mGLmv{n}}}$ from the structure of $\delta'_{{\mathrm{1}}}  \leq  \delta'_{{\mathrm{2}}}$
    then by reflexivity and monotonicity of $+$ and monotonicity of $*$:
    $$(  \mGLnt{r_{{\mathrm{1}}}}  + \, \cdots \, +  \mGLnt{r_{\mGLmv{n}}}  )   *  \delta_{{\mathrm{2}}}  \leq   (  \mGLnt{r'_{{\mathrm{1}}}}  + \, \cdots \, +  \mGLnt{r'_{\mGLmv{n}}}  )   *  \delta_{{\mathrm{2}}}$$
    Thus, we can reapply the sub rule on the inductive hypothesis:
    \begin{center}
      \begin{math}
        $$\mprset{flushleft}
        \inferrule* [right=$\mGLdruleGTXXSubName{}$] {
          \delta'_{{\mathrm{1}}}  \odot  \Delta'_{{\mathrm{1}}}  \vdash_{\mathsf{GT} }  \mGLnt{t'}  \mGLsym{:}  \mGLnt{Y}  \quad   (  \delta_{{\mathrm{1}}}  \mGLsym{,}  \mGLsym{(}  \mGLnt{r_{{\mathrm{1}}}}  + \, \cdots \, +  \mGLnt{r_{\mGLmv{n}}}  \mGLsym{)}  *  \delta_{{\mathrm{2}}}  \mGLsym{,}  \delta_{{\mathrm{3}}}  )   \leq   (  \gamma_{{\mathrm{1}}}  \mGLsym{,}  \mGLsym{(}  \mGLnt{r'_{{\mathrm{1}}}}  + \, \cdots \, +  \mGLnt{r'_{\mGLmv{n}}}  \mGLsym{)}  *  \delta_{{\mathrm{2}}}  \mGLsym{,}  \gamma_{{\mathrm{3}}}  )
        }{(  \gamma_{{\mathrm{1}}}  \mGLsym{,}  \mGLsym{(}  \mGLnt{r'_{{\mathrm{1}}}}  + \, \cdots \, +  \mGLnt{r'_{\mGLmv{n}}}  \mGLsym{)}  *  \delta_{{\mathrm{2}}}  \mGLsym{,}  \gamma_{{\mathrm{3}}}  )   \odot  \Delta'_{{\mathrm{1}}}  \vdash_{\mathsf{GT} }  \mGLnt{t'}  \mGLsym{:}  \mGLnt{Y}}
      \end{math}
    \end{center}
    satisfying the goal.

  \item[] Case.\ \\
    \begin{center}
      \begin{math}
        $$\mprset{flushleft}
        \inferrule* [right=$\mGLdruleGTXXIdName{}$] {
          \
        }{1  \odot  \mGLmv{x}  \mGLsym{:}  \mGLnt{X}  \vdash_{\mathsf{GT} }  \mGLmv{x}  \mGLsym{:}  \mGLnt{X}}
      \end{math}
    \end{center}
    This case holds by assumption and the fact that $1  *  \delta_{{\mathrm{2}}} = \delta_{{\mathrm{2}}}$.

  \item[] Case.\ \\
    \begin{center}
      \begin{math}
        $$\mprset{flushleft}
        \inferrule* [right=$\mGLdruleGTXXLinIName{}$] {
          \delta'  \odot  \Delta'  \mGLsym{;}  \emptyset  \vdash_{\mathsf{MT} }  \mGLnt{l}  \mGLsym{:}  \mGLnt{B}
        }{\delta'  \odot  \Delta'  \vdash_{\mathsf{GT} }  \mathsf{Lin} \, \mGLnt{l}  \mGLsym{:}  \mathsf{Lin} \, \mGLnt{B}}
      \end{math}
    \end{center}
    This case straightforwardly follows by applying the induction
    hypothesis and then reapplying the rule.

  \item[] Case.\ \\
    \begin{center}
      \begin{math}
        $$\mprset{flushleft}
        \inferrule* [right=$\mGLdruleGTXXTenEName{}$] {
          \delta'_{{\mathrm{2}}}  \odot  \Delta'_{{\mathrm{2}}}  \vdash_{\mathsf{GT} }  \mGLnt{t'_{{\mathrm{1}}}}  \mGLsym{:}  \mGLnt{Y_{{\mathrm{1}}}}  \boxtimes  \mGLnt{Y_{{\mathrm{2}}}}\\\\
          (  \delta'_{{\mathrm{1}}}  \mGLsym{,}  \mGLnt{s}  \mGLsym{,}  \mGLnt{s}  \mGLsym{,}  \delta'_{{\mathrm{3}}}  )   \odot   ( \Delta'_{{\mathrm{1}}}  \mGLsym{,}  \mGLmv{y_{{\mathrm{1}}}}  \mGLsym{:}  \mGLnt{Y_{{\mathrm{1}}}}  \mGLsym{,}  \mGLmv{y_{{\mathrm{2}}}}  \mGLsym{:}  \mGLnt{Y_{{\mathrm{2}}}}  \mGLsym{,}  \Delta'_{{\mathrm{3}}} )   \vdash_{\mathsf{GT} }  \mGLnt{t'_{{\mathrm{2}}}}  \mGLsym{:}  \mGLnt{Z}
        }{(  \delta'_{{\mathrm{1}}}  \mGLsym{,}   \mGLnt{s}  *  \delta'_{{\mathrm{2}}}   \mGLsym{,}  \delta'_{{\mathrm{3}}}  )   \odot   ( \Delta'_{{\mathrm{1}}}  \mGLsym{,}  \Delta'_{{\mathrm{2}}}  \mGLsym{,}  \Delta'_{{\mathrm{3}}} )   \vdash_{\mathsf{GT} }   \mathsf{let} \,( \mGLmv{y_{{\mathrm{1}}}} , \mGLmv{y_{{\mathrm{2}}}} ) =  \mGLnt{t'_{{\mathrm{1}}}} \, \mathsf{in} \, \mGLnt{t'_{{\mathrm{2}}}}   \mGLsym{:}  \mGLnt{Z}}
      \end{math}
    \end{center}
    This case is similar to the case for \mGLdruleGTXXUnitEName{} above.

  \item[] Case.\ \\
    \begin{center}
      \begin{math}
        $$\mprset{flushleft}
        \inferrule* [right=$\mGLdruleGTXXWeakName{}$] {
          (  \delta'_{{\mathrm{1}}}  \mGLsym{,}  \delta'_{{\mathrm{2}}}  )   \odot   ( \Delta'_{{\mathrm{1}}}  \mGLsym{,}  \Delta'_{{\mathrm{2}}} )   \vdash_{\mathsf{GT} }  \mGLnt{t}  \mGLsym{:}  \mGLnt{Z}
        }{(  \delta'_{{\mathrm{1}}}  \mGLsym{,}  \mathsf{0}  \mGLsym{,}  \delta'_{{\mathrm{2}}}  )   \odot   ( \Delta'_{{\mathrm{1}}}  \mGLsym{,}  \mGLmv{y}  \mGLsym{:}  \mGLnt{Y}  \mGLsym{,}  \Delta'_{{\mathrm{2}}} )   \vdash_{\mathsf{GT} }  \mGLnt{t}  \mGLsym{:}  \mGLnt{Z}}
      \end{math}
    \end{center}
    If one of $\mGLmv{x_{{\mathrm{1}}}}  \mGLsym{,} \, ... \, \mGLsym{,}  \mGLmv{x_{\mGLmv{n}}}$ are $y$, then this case holds trivially, because the $x$ would be
    free in $\mGLnt{t}$, but in the case where $y$ is not in $\mGLmv{x_{{\mathrm{1}}}}  \mGLsym{,} \, ... \, \mGLsym{,}  \mGLmv{x_{\mGLmv{n}}}$, then this case
    holds by applying the induction hypothesis and reapplying the
    rule.

  \item[] Case.\ \\
    \begin{center}
      \begin{math}
        $$\mprset{flushleft}
        \inferrule* [right=$\mGLdruleGTXXContName{}$] {
          (  \delta'_{{\mathrm{1}}}  \mGLsym{,}  \mGLnt{s_{{\mathrm{1}}}}  \mGLsym{,}  \mGLnt{s_{{\mathrm{2}}}}  \mGLsym{,}  \delta'_{{\mathrm{2}}}  )   \odot   ( \Delta'_{{\mathrm{1}}}  \mGLsym{,}  \mGLmv{y_{{\mathrm{1}}}}  \mGLsym{:}  \mGLnt{Y}  \mGLsym{,}  \mGLmv{y_{{\mathrm{2}}}}  \mGLsym{:}  \mGLnt{Y}  \mGLsym{,}  \Delta'_{{\mathrm{2}}} )   \vdash_{\mathsf{GT} }  \mGLnt{t}  \mGLsym{:}  \mGLnt{Z}
        }{(  \delta'_{{\mathrm{1}}}  \mGLsym{,}  \mGLnt{s_{{\mathrm{1}}}}  +  \mGLnt{s_{{\mathrm{2}}}}  \mGLsym{,}  \delta'_{{\mathrm{2}}}  )   \odot   ( \Delta'_{{\mathrm{1}}}  \mGLsym{,}  \mGLmv{y_{{\mathrm{1}}}}  \mGLsym{:}  \mGLnt{Y}  \mGLsym{,}  \Delta'_{{\mathrm{2}}} )   \vdash_{\mathsf{GT} }  \mGLsym{[}  \mGLmv{y_{{\mathrm{1}}}}  \mGLsym{/}  \mGLmv{y_{{\mathrm{2}}}}  \mGLsym{]}  \mGLnt{t}  \mGLsym{:}  \mGLnt{Z}}
      \end{math}
    \end{center}
    This case is similar to the previous case.  If $\mGLmv{y_{{\mathrm{1}}}}$ and/or $\mGLmv{y_{{\mathrm{2}}}}$ are in $\mGLmv{x_{{\mathrm{1}}}}  \mGLsym{,} \, ... \, \mGLsym{,}  \mGLmv{x_{\mGLmv{n}}}$, then we
    apply the induction hypothesis replacing both $\mGLmv{y_{{\mathrm{1}}}}$ and
    $\mGLmv{y_{{\mathrm{2}}}}$.  If $\mGLmv{y_{{\mathrm{1}}}}$ nor $\mGLmv{y_{{\mathrm{2}}}}$ are in $\mGLmv{x_{{\mathrm{1}}}}  \mGLsym{,} \, ... \, \mGLsym{,}  \mGLmv{x_{\mGLmv{n}}}$, then we simply apply the induction
    hypothesis, and reapply the rule.

  \item[] Case.\ \\
    \begin{center}
      \begin{math}
        $$\mprset{flushleft}
        \inferrule* [right=$\mGLdruleGTXXExName{}$] {
          (  \delta'_{{\mathrm{1}}}  \mGLsym{,}  \mGLnt{s_{{\mathrm{1}}}}  \mGLsym{,}  \mGLnt{s_{{\mathrm{2}}}}  \mGLsym{,}  \delta'_{{\mathrm{2}}}  )   \odot   ( \Delta'_{{\mathrm{1}}}  \mGLsym{,}  \mGLmv{y_{{\mathrm{1}}}}  \mGLsym{:}  \mGLnt{Y_{{\mathrm{1}}}}  \mGLsym{,}  \mGLmv{y_{{\mathrm{2}}}}  \mGLsym{:}  \mGLnt{Y_{{\mathrm{2}}}}  \mGLsym{,}  \Delta'_{{\mathrm{2}}} )   \vdash_{\mathsf{GT} }  \mGLnt{t}  \mGLsym{:}  \mGLnt{Z}
        }{(  \delta'_{{\mathrm{1}}}  \mGLsym{,}  \mGLnt{s_{{\mathrm{2}}}}  \mGLsym{,}  \mGLnt{s_{{\mathrm{1}}}}  \mGLsym{,}  \delta'_{{\mathrm{2}}}  )   \odot   ( \Delta'_{{\mathrm{1}}}  \mGLsym{,}  \mGLmv{y_{{\mathrm{1}}}}  \mGLsym{:}  \mGLnt{Y_{{\mathrm{1}}}}  \mGLsym{,}  \mGLmv{y_{{\mathrm{2}}}}  \mGLsym{:}  \mGLnt{Y_{{\mathrm{2}}}}  \mGLsym{,}  \Delta'_{{\mathrm{2}}} )   \vdash_{\mathsf{GT} }  \mGLnt{t}  \mGLsym{:}  \mGLnt{Z}}
      \end{math}
    \end{center}
    This case is similar to the previous two cases.

  \end{itemize}

\end{proof}

%% file: mst-implies-mt-ottput.tex
\begin{restatable}[$\vdash_{\mathsf{MS} }$ implies $\vdash_{\mathsf{MT} }$]{lemma}{MSTimpliesMT}
  \label{lemma:mst-implies-mt}
  The following hold by mutual induction:
  \begin{enumerate}
  \item If $\delta  \odot  \Delta  \vdash_{\mathsf{GS} }  \mGLnt{t}  \mGLsym{:}  \mGLnt{X}$, then $\delta  \odot  \Delta  \vdash_{\mathsf{GT} }  \mGLnt{t}  \mGLsym{:}  \mGLnt{X}$.
  \item If $\delta  \odot  \Delta  \mGLsym{;}  \Gamma  \vdash_{\mathsf{MS} }  \mGLnt{l}  \mGLsym{:}  \mGLnt{A}$, then $\delta  \odot  \Delta  \mGLsym{;}  \Gamma  \vdash_{\mathsf{MT} }  \mGLnt{l}  \mGLsym{:}  \mGLnt{A}$.
  \end{enumerate}
\end{restatable}
\begin{proof}
  This is a proof by mutual induction on the assumed typing
  derivation.  We consider part one first, but only supply the case
  for tensor left, because all other cases are either a result of
  substitution for typing or are trivial.

  \begin{enumerate}
  \item[] Case. \,\\
    \begin{center}
      \begin{math}
        $$\mprset{flushleft}
        \inferrule* [right=$\mGLdruleGSTXXTenLName{}$] {
          (  \delta_{{\mathrm{1}}}  \mGLsym{,}  \mGLnt{r}  \mGLsym{,}  \mGLnt{r}  \mGLsym{,}  \delta_{{\mathrm{2}}}  )   \odot   ( \Delta_{{\mathrm{1}}}  \mGLsym{,}  \mGLmv{x}  \mGLsym{:}  \mGLnt{Y_{{\mathrm{1}}}}  \mGLsym{,}  \mGLmv{y}  \mGLsym{:}  \mGLnt{Y_{{\mathrm{2}}}}  \mGLsym{,}  \Delta_{{\mathrm{2}}} )   \vdash_{\mathsf{GS} }  \mGLnt{t'}  \mGLsym{:}  \mGLnt{X}
        }{(  \delta_{{\mathrm{1}}}  \mGLsym{,}  \mGLnt{r}  \mGLsym{,}  \delta_{{\mathrm{2}}}  )   \odot   ( \Delta_{{\mathrm{1}}}  \mGLsym{,}  \mGLmv{z}  \mGLsym{:}  \mGLnt{Y_{{\mathrm{1}}}}  \boxtimes  \mGLnt{Y_{{\mathrm{2}}}}  \mGLsym{,}  \Delta_{{\mathrm{2}}} )   \vdash_{\mathsf{GS} }   \mathsf{let} \,( \mGLmv{x} , \mGLmv{y} ) =  \mGLmv{z} \, \mathsf{in} \, \mGLnt{t'}   \mGLsym{:}  \mGLnt{X}} 
      \end{math}
    \end{center}
    This rule can be derived in $\vdash_{\mathsf{GT} }$ using the tensor
    elimination rule and substitution.  First, by the induction
    hypothesis we know that:
    \[      
      (  \delta_{{\mathrm{1}}}  \mGLsym{,}  \mGLnt{r}  \mGLsym{,}  \mGLnt{r}  \mGLsym{,}  \delta_{{\mathrm{2}}}  )   \odot   ( \Delta_{{\mathrm{1}}}  \mGLsym{,}  \mGLmv{x}  \mGLsym{:}  \mGLnt{Y_{{\mathrm{1}}}}  \mGLsym{,}  \mGLmv{y}  \mGLsym{:}  \mGLnt{Y_{{\mathrm{2}}}}  \mGLsym{,}  \Delta_{{\mathrm{2}}} )   \vdash_{\mathsf{GT} }  \mGLnt{t'}  \mGLsym{:}  \mGLnt{X}
    \]
    holds.  Now we have the following derivation:
    \[
    \inferrule* [flushleft,right=$\mGLdruleGTXXTenEName{}$] {
      \inferrule* [flushleft,right=$\mGLdruleGTXXIdName{}$] {
        \,
      }{1  \odot  \mGLmv{z}  \mGLsym{:}  \mGLnt{Y_{{\mathrm{1}}}}  \boxtimes  \mGLnt{Y_{{\mathrm{2}}}}  \vdash_{\mathsf{GT} }  \mGLmv{z}  \mGLsym{:}  \mGLnt{Y_{{\mathrm{1}}}}  \boxtimes  \mGLnt{Y_{{\mathrm{2}}}}}\\
       (  \delta_{{\mathrm{1}}}  \mGLsym{,}  \mGLnt{r}  \mGLsym{,}  \mGLnt{r}  \mGLsym{,}  \delta_{{\mathrm{2}}}  )   \odot   ( \Delta_{{\mathrm{1}}}  \mGLsym{,}  \mGLmv{x}  \mGLsym{:}  \mGLnt{Y_{{\mathrm{1}}}}  \mGLsym{,}  \mGLmv{y}  \mGLsym{:}  \mGLnt{Y_{{\mathrm{2}}}}  \mGLsym{,}  \Delta_{{\mathrm{2}}} )   \vdash_{\mathsf{GT} }  \mGLnt{t'}  \mGLsym{:}  \mGLnt{X}
    }{(  \delta_{{\mathrm{1}}}  \mGLsym{,}  \mGLnt{r}  *  1  \mGLsym{,}  \delta_{{\mathrm{2}}}  )   \odot   ( \Delta_{{\mathrm{1}}}  \mGLsym{,}  \mGLmv{z}  \mGLsym{:}  \mGLnt{Y_{{\mathrm{1}}}}  \boxtimes  \mGLnt{Y_{{\mathrm{2}}}}  \mGLsym{,}  \Delta_{{\mathrm{2}}} )   \vdash_{\mathsf{GS} }   \mathsf{let} \,( \mGLmv{x} , \mGLmv{y} ) =  \mGLmv{z} \, \mathsf{in} \, \mGLnt{t'}   \mGLsym{:}  \mGLnt{X}}
    \]
    However, $\mGLnt{r}  *  1  \mGLsym{=}  \mGLnt{r}$, so let $\delta' = (  \delta_{{\mathrm{1}}}  \mGLsym{,}  \mGLnt{r}  *  1  \mGLsym{,}  \delta_{{\mathrm{2}}}  )$, and we obtain our desired derivation.
    Finally, we can see that the derived term matches our original
    term.  Thus, we obtain our result.
  \end{enumerate}
  We now move onto part two.  We consider the following cases:
  \begin{enumerate}
  \item[] Case.\ \\ 
    \begin{center}
      \begin{math}
        $$\mprset{flushleft}
        \inferrule* [right=$\mGLdruleMSTXXImpLName{}$] {
          \delta_{{\mathrm{2}}}  \odot  \Delta_{{\mathrm{2}}}  \mGLsym{;}  \Gamma_{{\mathrm{2}}}  \vdash_{\mathsf{MS} }  \mGLnt{l''}  \mGLsym{:}  \mGLnt{B_{{\mathrm{1}}}}\\\\
          \delta_{{\mathrm{1}}}  \odot  \Delta_{{\mathrm{1}}}  \mGLsym{;}   ( \Gamma_{{\mathrm{1}}}  \mGLsym{,}  \mGLmv{x}  \mGLsym{:}  \mGLnt{B_{{\mathrm{2}}}}  \mGLsym{,}  \Gamma_{{\mathrm{3}}} )   \vdash_{\mathsf{MS} }  \mGLnt{l'}  \mGLsym{:}  \mGLnt{A}
        }{(  \delta_{{\mathrm{1}}}  \mGLsym{,}  \delta_{{\mathrm{2}}}  )   \odot   ( \Delta_{{\mathrm{1}}}  \mGLsym{,}  \Delta_{{\mathrm{2}}} )   \mGLsym{;}   ( \Gamma_{{\mathrm{1}}}  \mGLsym{,}  \mGLmv{z}  \mGLsym{:}  \mGLnt{B_{{\mathrm{1}}}}  \multimap  \mGLnt{B_{{\mathrm{2}}}}  \mGLsym{,}  \Gamma_{{\mathrm{2}}}  \mGLsym{,}  \Gamma_{{\mathrm{3}}} )   \vdash_{\mathsf{MS} }  \mGLsym{[}  \mGLmv{z} \, \mGLnt{l''}  \mGLsym{/}  \mGLmv{x}  \mGLsym{]}  \mGLnt{l'}  \mGLsym{:}  \mGLnt{A}}
      \end{math}
    \end{center}
    By the induction hypothesis we can see that the following hold:
    \[
    \begin{array}{lll}
      \delta_{{\mathrm{2}}}  \odot  \Delta_{{\mathrm{2}}}  \mGLsym{;}  \Gamma_{{\mathrm{2}}}  \vdash_{\mathsf{MT} }  \mGLnt{l''}  \mGLsym{:}  \mGLnt{B_{{\mathrm{1}}}}\\
      \delta_{{\mathrm{1}}}  \odot  \Delta_{{\mathrm{1}}}  \mGLsym{;}   ( \Gamma_{{\mathrm{1}}}  \mGLsym{,}  \mGLmv{x}  \mGLsym{:}  \mGLnt{B_{{\mathrm{2}}}}  \mGLsym{,}  \Gamma_{{\mathrm{3}}} )   \vdash_{\mathsf{MT} }  \mGLnt{l'}  \mGLsym{:}  \mGLnt{A}
    \end{array}
    \]
    Next we have the following derivation:
    \[
    \inferrule* [flushleft,right=$\mGLdruleMTXXImpEName{}$] {
      \inferrule* [flushleft,right=$\mGLdruleMTXXIdName{}$] {
        \,
      }{\emptyset  \odot  \emptyset  \mGLsym{;}  \mGLmv{z}  \mGLsym{:}  \mGLnt{B_{{\mathrm{1}}}}  \multimap  \mGLnt{B_{{\mathrm{2}}}}  \vdash_{\mathsf{MT} }  \mGLmv{z}  \mGLsym{:}  \mGLnt{B_{{\mathrm{1}}}}  \multimap  \mGLnt{B_{{\mathrm{2}}}}}\\
      \delta_{{\mathrm{2}}}  \odot  \Delta_{{\mathrm{2}}}  \mGLsym{;}  \Gamma_{{\mathrm{2}}}  \vdash_{\mathsf{MT} }  \mGLnt{l''}  \mGLsym{:}  \mGLnt{B_{{\mathrm{1}}}}
    }{\delta_{{\mathrm{2}}}  \odot  \Delta_{{\mathrm{2}}}  \mGLsym{;}  \mGLmv{z}  \mGLsym{:}  \mGLnt{B_{{\mathrm{1}}}}  \multimap  \mGLnt{B_{{\mathrm{2}}}}  \mGLsym{,}  \Gamma_{{\mathrm{2}}}  \vdash_{\mathsf{MT} }  \mGLmv{z} \, \mGLnt{l''}  \mGLsym{:}  \mGLnt{B_{{\mathrm{2}}}}}
    \]
    Now using substitution for typing (Lemma~\ref{lemma:subsitution_for_gt_mt}) with $\delta_{{\mathrm{1}}}  \odot  \Delta_{{\mathrm{1}}}  \mGLsym{;}   ( \Gamma_{{\mathrm{1}}}  \mGLsym{,}  \mGLmv{x}  \mGLsym{:}  \mGLnt{B_{{\mathrm{2}}}}  \mGLsym{,}  \Gamma_{{\mathrm{3}}} )   \vdash_{\mathsf{MT} }  \mGLnt{l'}  \mGLsym{:}  \mGLnt{A}$ and
    the previous derivation we have
    $(  \delta_{{\mathrm{1}}}  \mGLsym{,}  \delta_{{\mathrm{2}}}  )   \odot   ( \Delta_{{\mathrm{1}}}  \mGLsym{,}  \Delta_{{\mathrm{2}}} )   \mGLsym{;}   ( \Gamma_{{\mathrm{1}}}  \mGLsym{,}  \mGLmv{z}  \mGLsym{:}  \mGLnt{B_{{\mathrm{1}}}}  \multimap  \mGLnt{B_{{\mathrm{2}}}}  \mGLsym{,}  \Gamma_{{\mathrm{2}}}  \mGLsym{,}  \Gamma_{{\mathrm{3}}} )   \vdash_{\mathsf{MS} }  \mGLsym{[}  \mGLmv{z} \, \mGLnt{l''}  \mGLsym{/}  \mGLmv{x}  \mGLsym{]}  \mGLnt{l'}  \mGLsym{:}  \mGLnt{A}$. So let $\delta' = (  \delta_{{\mathrm{1}}}  \mGLsym{,}  \delta_{{\mathrm{2}}}  )$.
    
  \item[] Case.\ \\ 
    \begin{center}
      \begin{math}
        $$\mprset{flushleft}
        \inferrule* [right=$\mGLdruleMSTXXLinLName{}$] {
          \delta_{{\mathrm{1}}}  \odot  \Delta_{{\mathrm{1}}}  \mGLsym{;}   ( \mGLmv{x}  \mGLsym{:}  \mGLnt{B}  \mGLsym{,}  \Gamma )   \vdash_{\mathsf{MS} }  \mGLnt{l'}  \mGLsym{:}  \mGLnt{A}
        }{(  \delta_{{\mathrm{1}}}  \mGLsym{,}  1  )   \odot   ( \Delta_{{\mathrm{1}}}  \mGLsym{,}  \mGLmv{z}  \mGLsym{:}  \mathsf{Lin} \, \mGLnt{B} )   \mGLsym{;}  \Gamma  \vdash_{\mathsf{MS} }  \mGLsym{[}   \mathsf{Unlin} \, \mGLmv{z}   \mGLsym{/}  \mGLmv{x}  \mGLsym{]}  \mGLnt{l'}  \mGLsym{:}  \mGLnt{A}}
      \end{math}
    \end{center}
    We know the following holds by the induction hypothesis:
    \[
      \delta_{{\mathrm{1}}}  \odot  \Delta_{{\mathrm{1}}}  \mGLsym{;}   ( \mGLmv{x}  \mGLsym{:}  \mGLnt{B}  \mGLsym{,}  \Gamma )   \vdash_{\mathsf{MT} }  \mGLnt{l'}  \mGLsym{:}  \mGLnt{A}
    \]
    Next we have the following derivation:
    \[
    \inferrule* [flushleft,right=$\mGLdruleMTXXLinEName{}$] {
      \inferrule* [flushleft,right=$\mGLdruleMTXXIdName{}$] {
        \,
      }{1  \odot  \mGLmv{z}  \mGLsym{:}  \mathsf{Lin} \, \mGLnt{B}  \vdash_{\mathsf{GT} }  \mGLmv{z}  \mGLsym{:}  \mathsf{Lin} \, \mGLnt{B}}
    }{1  \odot  \mGLmv{z}  \mGLsym{:}  \mathsf{Lin} \, \mGLnt{B}  \mGLsym{;}  \emptyset  \vdash_{\mathsf{MT} }   \mathsf{Unlin} \, \mGLmv{z}   \mGLsym{:}  \mGLnt{B}}
    \]
    Now using substitution for typing (Lemma~\ref{lemma:subsitution_for_gt_mt}) with $\delta_{{\mathrm{1}}}  \odot  \Delta_{{\mathrm{1}}}  \mGLsym{;}   ( \mGLmv{x}  \mGLsym{:}  \mGLnt{B}  \mGLsym{,}  \Gamma )   \vdash_{\mathsf{MS} }  \mGLnt{l'}  \mGLsym{:}  \mGLnt{A}$ and
    the previous derivation we have:
    \[
      \delta'  \odot   ( \Delta_{{\mathrm{1}}}  \mGLsym{,}  \mGLmv{z}  \mGLsym{:}  \mathsf{Lin} \, \mGLnt{B} )   \mGLsym{;}  \Gamma  \vdash_{\mathsf{MS} }  \mGLsym{[}   \mathsf{Unlin} \, \mGLmv{z}   \mGLsym{/}  \mGLmv{x}  \mGLsym{]}  \mGLnt{l'}  \mGLsym{:}  \mGLnt{A}
    \]
    for some $\delta'$ with $(  \delta_{{\mathrm{1}}}  \mGLsym{,}  1  )   \leq  \delta'$.  Thus, we obtain our result.
  \item[] Case.\ \\ 
    \begin{center}
      \begin{math}
        $$\mprset{flushleft}
        \inferrule* [right=$\mGLdruleMSTXXGrdLName{}$] {
          (  \delta  \mGLsym{,}  \mGLnt{r}  )   \odot   ( \Delta  \mGLsym{,}  \mGLmv{x}  \mGLsym{:}  \mGLnt{X} )   \mGLsym{;}  \Gamma'  \vdash_{\mathsf{MS} }  \mGLnt{l'}  \mGLsym{:}  \mGLnt{A}
        }{\delta  \odot  \Delta  \mGLsym{;}   ( \mGLmv{z}  \mGLsym{:}   \mathsf{Grd} _{ \mGLnt{r} }\, \mGLnt{X}   \mGLsym{,}  \Gamma' )   \vdash_{\mathsf{MS} }  \mathsf{let} \, \mathsf{Grd} \, \mGLnt{r} \, \mGLmv{x}  \mGLsym{=}  \mGLmv{z} \, \mathsf{in} \, \mGLnt{l'}  \mGLsym{:}  \mGLnt{A}}
      \end{math}
    \end{center}
    We know the following holds by induction hypothesis:
    \[
      (  \delta  \mGLsym{,}  \mGLnt{r}  )   \odot   ( \Delta  \mGLsym{,}  \mGLmv{x}  \mGLsym{:}  \mGLnt{X} )   \mGLsym{;}  \Gamma'  \vdash_{\mathsf{MT} }  \mGLnt{l'}  \mGLsym{:}  \mGLnt{A}
    \]
    Next we have the following derivation:
    \[
    \inferrule* [flushleft,right=$\mGLdruleMTXXGrdEName{}$] {
      \inferrule* [flushleft,right=$\mGLdruleMTXXIdName{}$] {
        \,
      }{\emptyset  \odot  \emptyset  \mGLsym{;}  \mGLmv{z}  \mGLsym{:}   \mathsf{Grd} _{ \mGLnt{r} }\, \mGLnt{X}   \vdash_{\mathsf{MT} }  \mGLmv{z}  \mGLsym{:}   \mathsf{Grd} _{ \mGLnt{r} }\, \mGLnt{X}}
      \\
      (  \delta  \mGLsym{,}  \mGLnt{r}  )   \odot   ( \Delta  \mGLsym{,}  \mGLmv{x}  \mGLsym{:}  \mGLnt{X} )   \mGLsym{;}  \Gamma'  \vdash_{\mathsf{MT} }  \mGLnt{l'}  \mGLsym{:}  \mGLnt{A}
    }{\delta  \odot  \Delta  \mGLsym{;}   ( \mGLmv{z}  \mGLsym{:}   \mathsf{Grd} _{ \mGLnt{r} }\, \mGLnt{X}   \mGLsym{,}  \Gamma' )   \vdash_{\mathsf{MS} }  \mathsf{let} \, \mathsf{Grd} \, \mGLnt{r} \, \mGLmv{x}  \mGLsym{=}  \mGLmv{z} \, \mathsf{in} \, \mGLnt{l'}  \mGLsym{:}  \mGLnt{A}}
    \]
    Here choose $\delta' = \delta$, and as we can see we obtain our result.
  \end{enumerate}
\end{proof}

%% file: mt-implies-mst-ottput.tex
\begin{restatable}[$\vdash_{\mathsf{MT} }$ implies $\vdash_{\mathsf{MS} }$]{lemma}{MTimpliesMST}
  \label{lemma:mt-implies-mst}
  The following hold by mutual induction:
  \begin{enumerate}
  \item If $\delta  \odot  \Delta  \vdash_{\mathsf{GT} }  \mGLnt{t}  \mGLsym{:}  \mGLnt{X}$, then $\delta  \odot  \Delta  \vdash_{\mathsf{GS} }  \mGLnt{t}  \mGLsym{:}  \mGLnt{X}$.
  \item If $\delta  \odot  \Delta  \mGLsym{;}  \Gamma  \vdash_{\mathsf{MT} }  \mGLnt{l}  \mGLsym{:}  \mGLnt{A}$, then $\delta  \odot  \Delta  \mGLsym{;}  \Gamma  \vdash_{\mathsf{MS} }  \mGLnt{l}  \mGLsym{:}  \mGLnt{A}$.
  \end{enumerate}
\end{restatable}

\begin{proof}
  This is a proof by mutual induction on the assumed typing
  derivation.  We only give the most interesting cases, and the others
  are either similar or hold trivially. First, we consider part one.

  \begin{enumerate}
  \item[] Case.\ \\ 
    \begin{center}
      \begin{math}
        $$\mprset{flushleft}
        \inferrule* [right=$\mGLdruleGTXXTenEName{}$] {
          \delta_{{\mathrm{2}}}  \odot  \Delta_{{\mathrm{2}}}  \vdash_{\mathsf{GT} }  \mGLnt{t'}  \mGLsym{:}  \mGLnt{Y_{{\mathrm{1}}}}  \boxtimes  \mGLnt{Y_{{\mathrm{2}}}}\\\\
          (  \delta_{{\mathrm{1}}}  \mGLsym{,}  \mGLnt{r}  \mGLsym{,}  \mGLnt{r}  \mGLsym{,}  \delta_{{\mathrm{3}}}  )   \odot   ( \Delta_{{\mathrm{1}}}  \mGLsym{,}  \mGLmv{x}  \mGLsym{:}  \mGLnt{Y_{{\mathrm{1}}}}  \mGLsym{,}  \mGLmv{y}  \mGLsym{:}  \mGLnt{Y_{{\mathrm{2}}}}  \mGLsym{,}  \Delta_{{\mathrm{3}}} )   \vdash_{\mathsf{GT} }  \mGLnt{t''}  \mGLsym{:}  \mGLnt{X}
        }{(  \delta_{{\mathrm{1}}}  \mGLsym{,}   \mGLnt{r}  *  \delta_{{\mathrm{2}}}   \mGLsym{,}  \delta_{{\mathrm{3}}}  )   \odot   ( \Delta_{{\mathrm{1}}}  \mGLsym{,}  \Delta_{{\mathrm{2}}}  \mGLsym{,}  \Delta_{{\mathrm{3}}} )   \vdash_{\mathsf{GT} }   \mathsf{let} \,( \mGLmv{x} , \mGLmv{y} ) =  \mGLnt{t'} \, \mathsf{in} \, \mGLnt{t''}   \mGLsym{:}  \mGLnt{X}}
      \end{math}
    \end{center}
    We know the following hold by induction hypothesis:
    \[
    \begin{array}{lll}
      \delta_{{\mathrm{2}}}  \odot  \Delta_{{\mathrm{2}}}  \vdash_{\mathsf{GS} }  \mGLnt{t'}  \mGLsym{:}  \mGLnt{Y_{{\mathrm{1}}}}  \boxtimes  \mGLnt{Y_{{\mathrm{2}}}}\\
      (  \delta_{{\mathrm{1}}}  \mGLsym{,}  \mGLnt{r}  \mGLsym{,}  \mGLnt{r}  \mGLsym{,}  \delta_{{\mathrm{3}}}  )   \odot   ( \Delta_{{\mathrm{1}}}  \mGLsym{,}  \mGLmv{x}  \mGLsym{:}  \mGLnt{Y_{{\mathrm{1}}}}  \mGLsym{,}  \mGLmv{y}  \mGLsym{:}  \mGLnt{Y_{{\mathrm{2}}}}  \mGLsym{,}  \Delta_{{\mathrm{3}}} )   \vdash_{\mathsf{GS} }  \mGLnt{t''}  \mGLsym{:}  \mGLnt{X}
    \end{array}
    \]
    Now we have the following derivation:
    \[
    \inferrule* [flushleft,right=$\mGLdruleGSTXXCutName{}$] {
      \delta_{{\mathrm{2}}}  \odot  \Delta_{{\mathrm{2}}}  \vdash_{\mathsf{GT} }  \mGLnt{t'}  \mGLsym{:}  \mGLnt{Y_{{\mathrm{1}}}}  \boxtimes  \mGLnt{Y_{{\mathrm{2}}}}\\
      \inferrule* [flushleft,right=$\mGLdruleGSTXXTenLName{}$] {
      (  \delta_{{\mathrm{1}}}  \mGLsym{,}  \mGLnt{r}  \mGLsym{,}  \mGLnt{r}  \mGLsym{,}  \delta_{{\mathrm{3}}}  )   \odot   ( \Delta_{{\mathrm{1}}}  \mGLsym{,}  \mGLmv{x}  \mGLsym{:}  \mGLnt{Y_{{\mathrm{1}}}}  \mGLsym{,}  \mGLmv{y}  \mGLsym{:}  \mGLnt{Y_{{\mathrm{2}}}}  \mGLsym{,}  \Delta_{{\mathrm{3}}} )   \vdash_{\mathsf{GS} }  \mGLnt{t''}  \mGLsym{:}  \mGLnt{X}
      }{(  \delta_{{\mathrm{1}}}  \mGLsym{,}  \mGLnt{r}  \mGLsym{,}  \delta_{{\mathrm{3}}}  )   \odot   ( \Delta_{{\mathrm{1}}}  \mGLsym{,}  \mGLmv{z}  \mGLsym{:}  \mGLnt{Y_{{\mathrm{1}}}}  \boxtimes  \mGLnt{Y_{{\mathrm{2}}}}  \mGLsym{,}  \Delta_{{\mathrm{3}}} )   \vdash_{\mathsf{GS} }   \mathsf{let} \,( \mGLmv{x} , \mGLmv{y} ) =  \mGLmv{z} \, \mathsf{in} \, \mGLnt{t''}   \mGLsym{:}  \mGLnt{X}}
    }{(  \delta_{{\mathrm{1}}}  \mGLsym{,}   \mGLnt{r}  *  \delta_{{\mathrm{2}}}   \mGLsym{,}  \delta_{{\mathrm{3}}}  )   \odot   ( \Delta_{{\mathrm{1}}}  \mGLsym{,}  \Delta_{{\mathrm{2}}}  \mGLsym{,}  \Delta_{{\mathrm{3}}} )   \vdash_{\mathsf{GT} }   \mathsf{let} \,( \mGLmv{x} , \mGLmv{y} ) =  \mGLnt{t'} \, \mathsf{in} \, \mGLnt{t''}   \mGLsym{:}  \mGLnt{X}}
    \]    
    And as we can see we obtain our result.
  \end{enumerate}
  We now consider part two.
  \begin{enumerate}
  \item[] Case.\ \\ 
    \begin{center}
      \begin{math}
        $$\mprset{flushleft}
        \inferrule* [right=$\mGLdruleMTXXImpEName{}$] {
          \delta_{{\mathrm{2}}}  \odot  \Delta_{{\mathrm{2}}}  \mGLsym{;}  \Gamma_{{\mathrm{2}}}  \vdash_{\mathsf{MT} }  \mGLnt{l_{{\mathrm{2}}}}  \mGLsym{:}  \mGLnt{B}\\\\
          \delta_{{\mathrm{1}}}  \odot  \Delta_{{\mathrm{1}}}  \mGLsym{;}  \Gamma_{{\mathrm{1}}}  \vdash_{\mathsf{MT} }  \mGLnt{l_{{\mathrm{1}}}}  \mGLsym{:}  \mGLnt{B}  \multimap  \mGLnt{A}
        }{(  \delta_{{\mathrm{1}}}  \mGLsym{,}  \delta_{{\mathrm{2}}}  )   \odot   ( \Delta_{{\mathrm{1}}}  \mGLsym{,}  \Delta_{{\mathrm{2}}} )   \mGLsym{;}   ( \Gamma_{{\mathrm{1}}}  \mGLsym{,}  \Gamma_{{\mathrm{2}}} )   \vdash_{\mathsf{MT} }  \mGLnt{l_{{\mathrm{1}}}} \, \mGLnt{l_{{\mathrm{2}}}}  \mGLsym{:}  \mGLnt{A}}
      \end{math}
    \end{center}
    By the induction hypothesis we know the following hold:
    \[
      \begin{array}{lll}
        \delta_{{\mathrm{2}}}  \odot  \Delta_{{\mathrm{2}}}  \mGLsym{;}  \Gamma_{{\mathrm{2}}}  \vdash_{\mathsf{MS} }  \mGLnt{l_{{\mathrm{2}}}}  \mGLsym{:}  \mGLnt{B}\\
        \delta_{{\mathrm{1}}}  \odot  \Delta_{{\mathrm{1}}}  \mGLsym{;}  \Gamma_{{\mathrm{1}}}  \vdash_{\mathsf{MS} }  \mGLnt{l_{{\mathrm{1}}}}  \mGLsym{:}  \mGLnt{B}  \multimap  \mGLnt{A}
      \end{array}
    \]
    Now we have the following derivation:
    \begin{gather*}
    \inferrule* [flushleft,right=$\mGLdruleMSTXXCutName{}$] {
      \delta_{{\mathrm{1}}}  \odot  \Delta_{{\mathrm{1}}}  \mGLsym{;}  \Gamma_{{\mathrm{1}}}  \vdash_{\mathsf{MS} }  \mGLnt{l_{{\mathrm{1}}}}  \mGLsym{:}  \mGLnt{B}  \multimap  \mGLnt{A}\\
      \inferrule* [flushleft,right=$\mGLdruleMSTXXImpLName{}$] {       
        \delta_{{\mathrm{2}}}  \odot  \Delta_{{\mathrm{2}}}  \mGLsym{;}  \Gamma_{{\mathrm{2}}}  \vdash_{\mathsf{MS} }  \mGLnt{l_{{\mathrm{2}}}}  \mGLsym{:}  \mGLnt{B}\\
        \inferrule* [flushleft,right=$\mGLdruleMSTXXidName{}$] {
          \,
        }{\emptyset  \odot  \emptyset  \mGLsym{;}  \mGLmv{x}  \mGLsym{:}  \mGLnt{B}  \vdash_{\mathsf{MS} }  \mGLmv{x}  \mGLsym{:}  \mGLnt{B}}
      }{\delta_{{\mathrm{2}}}  \odot  \Delta_{{\mathrm{2}}}  \mGLsym{;}   ( \mGLmv{z}  \mGLsym{:}  \mGLnt{A}  \multimap  \mGLnt{B}  \mGLsym{,}  \Gamma_{{\mathrm{2}}} )   \vdash_{\mathsf{MS} }  \mGLmv{z} \, \mGLnt{l_{{\mathrm{2}}}}  \mGLsym{:}  \mGLnt{B}}
    }{(  \delta_{{\mathrm{1}}}  \mGLsym{,}  \delta_{{\mathrm{2}}}  )   \odot   ( \Delta_{{\mathrm{1}}}  \mGLsym{,}  \Delta_{{\mathrm{2}}} )   \mGLsym{;}   ( \Gamma_{{\mathrm{1}}}  \mGLsym{,}  \Gamma_{{\mathrm{2}}} )   \vdash_{\mathsf{MS} }  \mGLnt{l_{{\mathrm{1}}}} \, \mGLnt{l_{{\mathrm{2}}}}  \mGLsym{:}  \mGLnt{A}}
    \end{gather*}
    And as we can see we obtain our result.

  \item[] Case.\ \\ 
    \begin{center}
      \begin{math}
        $$\mprset{flushleft}
        \inferrule* [right=$\mGLdruleMTXXLinEName{}$] {
          \delta  \odot  \Delta  \vdash_{\mathsf{GT} }  \mGLnt{t}  \mGLsym{:}  \mathsf{Lin} \, \mGLnt{A}
        }{\delta  \odot  \Delta  \mGLsym{;}  \emptyset  \vdash_{\mathsf{MT} }   \mathsf{Unlin} \, \mGLnt{t}   \mGLsym{:}  \mGLnt{A}}
      \end{math}
    \end{center}
    By the induction hypothesis we know the following hold:
    \[
      \delta  \odot  \Delta  \vdash_{\mathsf{GS} }  \mGLnt{t}  \mGLsym{:}  \mathsf{Lin} \, \mGLnt{A}
    \]
    Now we have the following derivation:
    \[
    \inferrule* [flushleft,right=$\mGLdruleMSTXXGCutName{}$] {
      \delta  \odot  \Delta  \vdash_{\mathsf{GS} }  \mGLnt{t}  \mGLsym{:}  \mathsf{Lin} \, \mGLnt{A}\\
      \inferrule* [flushleft,right=$\mGLdruleMSTXXLinLName{}$] {
        \inferrule* [flushleft,right=$\mGLdruleMSTXXidName{}$] {
        \,
      }{\emptyset  \odot  \emptyset  \mGLsym{;}  \mGLmv{x}  \mGLsym{:}  \mGLnt{A}  \vdash_{\mathsf{MS} }  \mGLmv{x}  \mGLsym{:}  \mGLnt{A}}\\      
      }{1  \odot  \mGLmv{z}  \mGLsym{:}  \mathsf{Lin} \, \mGLnt{A}  \mGLsym{;}  \emptyset  \vdash_{\mathsf{MS} }   \mathsf{Unlin} \, \mGLmv{z}   \mGLsym{:}  \mGLnt{A}}
    }{1  *  \delta  \odot  \Delta  \mGLsym{;}  \emptyset  \vdash_{\mathsf{MS} }   \mathsf{Unlin} \, \mGLnt{t}   \mGLsym{:}  \mGLnt{A}}
    \]
    And we obtain our result, because $1  *  \delta  \mGLsym{=}  \delta$.

  \item[] Case.\ \\ 
    \begin{center}
      \begin{math}
        $$\mprset{flushleft}
        \inferrule* [right=$\mGLdruleMTXXGrdEName{}$] {
          \delta_{{\mathrm{2}}}  \odot  \Delta_{{\mathrm{2}}}  \mGLsym{;}  \Gamma_{{\mathrm{2}}}  \vdash_{\mathsf{MT} }  \mGLnt{l_{{\mathrm{1}}}}  \mGLsym{:}   \mathsf{Grd} _{ \mGLnt{r} }\, \mGLnt{X}\\\\
          (  \delta_{{\mathrm{1}}}  \mGLsym{,}  \mGLnt{r}  \mGLsym{,}  \delta_{{\mathrm{3}}}  )   \odot   ( \Delta_{{\mathrm{1}}}  \mGLsym{,}  \mGLmv{x}  \mGLsym{:}  \mGLnt{X}  \mGLsym{,}  \Delta_{{\mathrm{3}}} )   \mGLsym{;}  \Gamma_{{\mathrm{1}}}  \vdash_{\mathsf{MT} }  \mGLnt{l_{{\mathrm{2}}}}  \mGLsym{:}  \mGLnt{A}
        }{(  \delta_{{\mathrm{1}}}  \mGLsym{,}  \delta_{{\mathrm{2}}}  \mGLsym{,}  \delta_{{\mathrm{3}}}  )   \odot   ( \Delta_{{\mathrm{1}}}  \mGLsym{,}  \Delta_{{\mathrm{2}}}  \mGLsym{,}  \Delta_{{\mathrm{3}}} )   \mGLsym{;}   ( \Gamma_{{\mathrm{1}}}  \mGLsym{,}  \Gamma_{{\mathrm{2}}} )   \vdash_{\mathsf{MT} }  \mathsf{let} \, \mathsf{Grd} \, \mGLnt{r} \, \mGLmv{x}  \mGLsym{=}  \mGLnt{l_{{\mathrm{1}}}} \, \mathsf{in} \, \mGLnt{l_{{\mathrm{2}}}}  \mGLsym{:}  \mGLnt{A}}
      \end{math}
    \end{center}
    We know by the induction hypothesis that the following hold:
    \[
    \begin{array}{lll}
      \delta_{{\mathrm{2}}}  \odot  \Delta_{{\mathrm{2}}}  \mGLsym{;}  \Gamma_{{\mathrm{2}}}  \vdash_{\mathsf{MS} }  \mGLnt{l_{{\mathrm{1}}}}  \mGLsym{:}   \mathsf{Grd} _{ \mGLnt{r} }\, \mGLnt{X}\\
      (  \delta_{{\mathrm{1}}}  \mGLsym{,}  \mGLnt{r}  \mGLsym{,}  \delta_{{\mathrm{2}}}  )   \odot   ( \Delta_{{\mathrm{1}}}  \mGLsym{,}  \mGLmv{x}  \mGLsym{:}  \mGLnt{X}  \mGLsym{,}  \Delta_{{\mathrm{3}}} )   \mGLsym{;}  \Gamma_{{\mathrm{1}}}  \vdash_{\mathsf{MS} }  \mGLnt{l_{{\mathrm{2}}}}  \mGLsym{:}  \mGLnt{A}
    \end{array}
    \]
    Now we have the following derivation:
    \begin{gather*}
    \inferrule* [flushleft,right=$\mGLdruleMSTXXCutName{}$] {
      \delta_{{\mathrm{2}}}  \odot  \Delta_{{\mathrm{2}}}  \mGLsym{;}  \Gamma_{{\mathrm{2}}}  \vdash_{\mathsf{MS} }  \mGLnt{l_{{\mathrm{1}}}}  \mGLsym{:}   \mathsf{Grd} _{ \mGLnt{r} }\, \mGLnt{X}\\
      \inferrule* [flushleft,right=$\mGLdruleMSTXXGrdLName{}$] {
        (  \delta_{{\mathrm{1}}}  \mGLsym{,}  \mGLnt{r}  \mGLsym{,}  \delta_{{\mathrm{3}}}  )   \odot   ( \Delta_{{\mathrm{1}}}  \mGLsym{,}  \mGLmv{x}  \mGLsym{:}  \mGLnt{X}  \mGLsym{,}  \Delta_{{\mathrm{3}}} )   \mGLsym{;}  \Gamma_{{\mathrm{1}}}  \vdash_{\mathsf{MS} }  \mGLnt{l_{{\mathrm{2}}}}  \mGLsym{:}  \mGLnt{A}
      }{(  \delta_{{\mathrm{1}}}  \mGLsym{,}  \delta_{{\mathrm{3}}}  )   \odot   ( \Delta_{{\mathrm{1}}}  \mGLsym{,}  \Delta_{{\mathrm{3}}} )   \mGLsym{;}  \Gamma_{{\mathrm{1}}}  \mGLsym{,}  \mGLmv{z}  \mGLsym{:}   \mathsf{Grd} _{ \mGLnt{r} }\, \mGLnt{X}   \vdash_{\mathsf{MS} }  \mathsf{let} \, \mathsf{Grd} \, \mGLnt{r} \, \mGLmv{x}  \mGLsym{=}  \mGLmv{z} \, \mathsf{in} \, \mGLnt{l_{{\mathrm{2}}}}  \mGLsym{:}  \mGLnt{A}}
    }{(  \delta_{{\mathrm{1}}}  \mGLsym{,}  \delta_{{\mathrm{2}}}  \mGLsym{,}  \delta_{{\mathrm{3}}}  )   \odot   ( \Delta_{{\mathrm{1}}}  \mGLsym{,}  \Delta_{{\mathrm{2}}}  \mGLsym{,}  \Delta_{{\mathrm{3}}} )   \mGLsym{;}  \Gamma_{{\mathrm{1}}}  \mGLsym{,}  \Gamma_{{\mathrm{2}}}  \vdash_{\mathsf{MS} }  \mathsf{let} \, \mathsf{Grd} \, \mGLnt{r} \, \mGLmv{x}  \mGLsym{=}  \mGLnt{l_{{\mathrm{1}}}} \, \mathsf{in} \, \mGLnt{l_{{\mathrm{2}}}}  \mGLsym{:}  \mGLnt{A}}
    \end{gather*}
    In the above we leave the use of exchange implicit.  As we can see we have obtained our result.
  \end{enumerate}
\end{proof}